\pdfoutput=1
\PassOptionsToPackage{dvipsnames}{xcolor} 
\documentclass[acmsmall, screen]{acmart}

\AtBeginDocument{\providecommand\BibTeX{{\normalfont B\kern-0.5em{\scshape i\kern-0.25em b}\kern-0.8em\TeX}}}

\newcommand{\shortVersion}{false}

\setcopyright{none}
\renewcommand\footnotetextcopyrightpermission[1]{}

\clearpage{}\usepackage[utf8]{inputenc}

\usepackage{amsfonts}
\usepackage{amsmath}
\usepackage{halloweenmath}

\usepackage{color}
\usepackage[disable]{todonotes}
\usepackage{amsthm}
\usepackage{paralist}

\usepackage{makecell}
\usepackage{subcaption} \usepackage{pifont}\usepackage{mathtools}
\usepackage{mathpartir}
\usepackage{natbib} \usepackage{enumitem}
\usepackage{relsize}
\usepackage{listings}
\usepackage{listingsutf8}
\usepackage{caption}
\usepackage{thmtools, thm-restate} 
\usepackage{booktabs}
\usepackage{balance}

\usepackage[nameinlink, noabbrev]{cleveref}

\declaretheorem{theorem}
\declaretheorem[]{lemma}

\declaretheorem[]{definition}
\declaretheorem[]{corollary}

\theoremstyle{definition}

\theoremstyle{plain}
\newtheorem{assumption}{Assumption}

\usepackage{pgfplots}
\usepackage{tikz}
\usepackage[norndcorners,customcolors]{hf-tikz}
\usetikzlibrary{positioning,shadows,arrows,calc,backgrounds,fit,shapes,snakes,shapes.multipart,decorations.pathreplacing,shapes.misc,patterns, intersections, pgfplots.fillbetween}
\usepackage[edges]{forest}
\usepackage{xspace}
\usepackage{theoremref} 

\usepackage{stmaryrd} \usepackage{wrapfig}
\usepackage{xargs}
\usepackage{bm}
\usepackage{soulutf8}
\usepackage{algorithm}
\usepackage{algorithmicx}
\usepackage{algpseudocode}
\algblock{Input}{EndInput}
\algnotext{EndInput}
\algblock{Output}{EndOutput}
\algnotext{EndOutput}
\newcommand{\Desc}[2]{\State \makebox[2em][l]{#1}#2}

\definecolor{pred}{RGB}{255,160,160}
\definecolor{pgrn}{RGB}{150,230,150}
\definecolor{pblu}{RGB}{180,180,255}
\definecolor{porg}{RGB}{255,180,100}
\setul{}{1pt}

\usepackage[normalem]{ulem}
\makeatletter
  \newcommand\reduline{\bgroup\markoverwith{\textcolor{pred}{\rule[-0.5ex]{2pt}{0.4pt}}}\ULon}
  \UL@protected\def\blueuwave{\leavevmode \bgroup 
    \ifdim \ULdepth=\maxdimen \ULdepth 3.5\p@
    \else \advance\ULdepth2\p@ 
    \fi \markoverwith{\lower\ULdepth\hbox{\textcolor{pblu}{\sixly \char58}}}\ULon}
  \UL@protected\def\bluedotuline{\leavevmode \bgroup 
    \UL@setULdepth
    \ifx\UL@on\UL@onin \advance\ULdepth2\p@\fi
    \markoverwith{\begingroup
       \lower\ULdepth\hbox{\kern.06em \textcolor{pblu}{.}\kern.04em}\endgroup}\ULon}
  \UL@protected\def\greendashuline{\leavevmode \bgroup 
    \UL@setULdepth
    \ifx\UL@on\UL@onin \advance\ULdepth2\p@\fi
    \markoverwith{\kern.13em
    \vtop{\color{pgrn}\kern\ULdepth \hrule width .3em}\kern.13em}\ULon}
\makeatother

\makeatletter
\def\uwave{\bgroup \markoverwith{\lower6\p@\hbox{\sixly \textcolor{blue}{\char58}}}\ULon}
\font\sixly=lasy6 \makeatother

\usepackage[inference]{semantic}
\newcommand{\totcountermeasures}[0]{9\xspace}

\newcommand{\Bv}{\mathbf{\textcolor{RoyalBlue}{B}}}
\newcommand{\Sv}{\texttt{\textcolor{ForestGreen}{S}}}
\newcommand{\SLSv}{\texttt{\textcolor{Orange}{SLS}}}
\newcommand{\Jv}{\textit{\textcolor{BlueGreen}{J}}}

\newcommand{\Rv}{\mathsf{\textcolor{BrickRed}{R}}}

\newcommand{\Div}{\mf{\textcolor{OrangeRed}{D}}}
\newcommand{\MS}{\mf{\textcolor{Plum}{MS}}}
\newcommand{\CT}{\mf{\textcolor{Aquamarine}{CT}}}

\newcommand{\trgB}[1]{{\mf{\col{RoyalBlue}{#1}}}}
\newcommand{\trgS}[1]{{\mf{\col{ForestGreen}{#1}}}}
\newcommand{\trgJ}[1]{{\mf{\col{BlueGreen}{#1}}}}

\newcommand{\trgR}[1]{{\mf{\col{BrickRed}{#1}}}}

\newcommand{\trgSLS}[1]{{\mf{\col{Orange}{#1}}}}

\newcommand{\Bvr}{\mathbf{\textcolor{RoyalBlue}{B}}}
\newcommand{\Svr}{\texttt{\textcolor{ForestGreen}{S}}}
\newcommand{\SLSvr}{\texttt{\textcolor{Orange}{SLS}}}
\newcommand{\Jvr}{\texttt{\textcolor{BlueGreen}{J}}} 
\newcommand{\Rvr}{\mathsf{\textcolor{BrickRed}{R}}}

\newcommand{\greencheck}{{\color{ForestGreen}\checkmark}}
\newcommand{\xmark}{\ding{55}}\newcommand{\redxmark}{{\color{Maroon}\xmark}}

\newcommand{\specj}{\textsc{Spectre v2}\xspace}
\newcommand{\specr}{\textsc{Spectre v5}\xspace} \newcommand{\specs}{\textsc{Spectre v4}\xspace}  

\newcommand{\muasm}{$\mu$\textsc{Asm}\xspace}

\newcommand{\obsKywd}[1]{\textcolor{RoyalBlue}{\mathtt{#1}}}
\newcommand{\loadObsKywd}{\obsKywd{load}}
\newcommand{\storeObsKywd}{\obsKywd{store}}
\newcommand{\callObsKywd}{\obsKywd{call}}
\newcommand{\retObsKywd}{\obsKywd{ret}}
\newcommand{\pcObsKywd}{\obsKywd{pc}}

\newcommand{\opObs}[2]{#1\ \obsKywd{op}\ #2}
\newcommand{\spKywd}{\kywd{sp}}
\newcommand{\spR}{\spKywd}

\newcommand{\loadObs}[1]{\loadObsKywd\ #1}
\newcommand{\storeObs}[1]{\storeObsKywd\ #1}
\newcommand{\pcObs}[1]{\pcObsKywd\ #1}
\newcommand{\retObs}[1]{\retObsKywd\ #1}
\newcommand{\callObs}[1]{\callObsKywd\ #1}

\newcommand{\startObsKywd}[1]{\obsKywd{start}}
\newcommand{\rollbackObsKywd}[1]{\obsKywd{rlb}} \newcommand{\startObs}[1]{\startObsKywd{}\ #1}

\newcommand{\rollbackObsB}[0]{\rollbackObsKywd{}_{\Bv}\ }

\newcommand{\pcObsID}[1]{\pcObsKywd_{\id}\ #1}
\newcommand{\rollbackObsID}[0]{\rollbackObsKywd{}_{\id}}

\newcommand{\rollbackObsS}[0]{\rollbackObsKywd{}_{\Sv}\ }

\newcommand{\rollbackObsR}[0]{\rollbackObsKywd{}_{\Rv}\ }

\newcommand{\rollbackObsJ}[0]{\rollbackObsKywd{}_{\Jv}\ }

\newcommand{\rollbackObsSLS}[0]{\rollbackObsKywd{}_{\SLSv}\ }

\newcommand{\Nat}{\mathbb{N}}

\newcommand{\Reg}{\mathit{Regs}}
\newcommand{\Val}{\mathit{Vals}}

\newcommand{\Assgn}{\mathit{Assgn}}
\newcommand{\RSB}{\mathit{Rsb}}

\newcommand{\kywd}[1]{\mathbf{#1}}
\newcommand{\skipKywd}{\kywd{skip}}
\newcommand{\storeKywd}{\kywd{store}}
\newcommand{\loadKywd}{\kywd{load}}
\newcommand{\jmpKywd}{\kywd{jmp}}
\newcommand{\jzKywd}{\kywd{beqz}}
\newcommand{\pcKywd}{\kywd{pc}}
\newcommand{\callKywd}{\kywd{call}}
\newcommand{\retKywd}{\kywd{ret}}

\newcommand{\pload}[2]{\loadKywd\ #1, #2}
\newcommand{\ploadprv}[2]{\loadKywd_{prv}\ #1, #2}

\newcommand{\pstoreprv}[2]{\storeKywd_{prv}\ #1, #2}
\newcommand{\pcall}[1]{\callKywd\ #1}
\newcommand{\pret}{\retKywd }
\newcommand{\pcondassign}[3]{#1 \xleftarrow{#3?} #2}
\newcommand{\passign}[2]{#1 \leftarrow #2}
\newcommand{\vassign}[3]{#1 \leftarrow #2\ op\ #3}
\newcommand{\pjz}[2]{\jzKywd\ #1, #2}

\newcommand{\pjmp}[1]{\jmpKywd\ #1}

\newcommand{\pc}{\pcKywd}

\newcommand{\loadC}{\kywd{load}}
\newcommand{\storeC}{\kywd{store}}
\newcommand{\retC}{\kywd{ret}}
\newcommand{\jmpC}{\kywd{jmp}}
\newcommand{\callC}{\kywd{call}}
\newcommand{\jzC}{\kywd{beqz}}

\newcommand{\Rsb}{\mathbb{R}}

\newcommand{\tup}[1]{\langle #1 \rangle}

\newcommand{\specarrowB}[1]{\xRightarrow{\mathit{#1}}\mathrel{\vphantom{\to}}_{\Bv}}
\newcommand{\specarrowS}[1]{\xRightarrow{\mathit{#1}}\mathrel{\vphantom{\to}}_{\Sv}}
\newcommand{\specarrowR}[1]{\xRightarrow{\mathit{#1}}\mathrel{\vphantom{\to}}_{\Rv}}

\newcommand{\instr}[1]{p(\sigma(\pc)) = #1}

\newcommand{\instrneq}[1]{p(\sigma(\pc)) \ne #1}

\newcommand{\exprEval}[3]{#1 \triangleright{} #2 \bigreds #3}

\newcommand{\specProject}[1]{#1\mathord{\upharpoonright_{se}}}

\newcommand{\nspecProject}[1]{#1\mathord{\upharpoonright_{ns}}}

\lstdefinestyle{MUASMstyle}
{
	frame = tb,
    belowskip=.4\baselineskip,
    aboveskip=.4\baselineskip,
  	showstringspaces = false,
  	breaklines = true,
  	breakatwhitespace = true,
  	tabsize = 3,
  	numbers = left,
    stepnumber = 1,
    numberstyle = \tiny\color{gray},
    alsoletter={.\$\%},
    basicstyle={\ttfamily\color{black}},
    keywordstyle={\ttfamily\color{Blue3}},
    keywordstyle=[2]{\ttfamily\color{Green4}},
    keywordstyle=[3]{\ttfamily\color{orange}},
    keywordstyle=[4]{\ttfamily\color{violet}},
    keywordstyle=[5]{\ttfamily\color{red}}, 
    keywords = {store, load, beqz, \leftarrow, ret, call, spbarr, jmp, skip, fences},
    morekeywords = {},
    morekeywords = [2]{A, B},
    morekeywords = [3]{\%rax, \%rbx, \%rcx, \%eax, \%ecx, \%rbp, \%cl, \%rdi, \%edx, \%esi, \%rsi, \%rdx, \%rsp, \%al, \%rip, \%ax, \%bx, \%cx,  \%bp,  \%di, \%dx, \%si, \%sp, \%ip, \%r8b},
    morekeywords = [4]{k, y, size, temp, pub, secr_addr, x, eax, edi, eip, secret, z, w, size_A, edx, public, sp, rsp},
    morekeywords = [5]{END, L1, L2, Main, Misp, Test, speculate, call_start, call_leak, return_new,
        manipulate_return_address, Manip_Stack, Speculate
    }
}

\lstdefinestyle{PrologStyle}
{
	frame = tb,
    belowskip=.4\baselineskip,
    aboveskip=.4\baselineskip,
  	showstringspaces = false,
  	breaklines = true,
  	breakatwhitespace = true,
  	tabsize = 3,
  	numbers = left,
    stepnumber = 1,
    numberstyle = \tiny\color{gray},
    alsoletter={.\$\%},
    basicstyle={\ttfamily\color{black}},
    keywordstyle={\ttfamily\color{Blue3}},
    keywordstyle=[2]{\ttfamily\color{Green4}},
    keywordstyle=[3]{\ttfamily\color{orange}},
    keywordstyle=[4]{\ttfamily\color{violet}},
    keywordstyle=[5]{\ttfamily\color{red}}, 
    keywords = {store, load, beqz, \leftarrow, ret, call, spbarr, jmp,
        xrunv4_1, xc41, enabled, decr, window, update0, spec4, incpc, trace_rawpc, trace, is
    },
    morekeywords = {},
    morekeywords = [2]{ B},
    morekeywords = [3]{\%rax, \%rbx, \%rcx, \%eax, \%ecx, \%rbp, \%cl, \%rdi, \%edx, \%esi, \%rsi, \%rdx, \%rsp, \%al, \%rip, \%ax, \%bx, \%cx,  \%bp,  \%di, \%dx, \%si, \%sp, \%ip, \%r8b},
    morekeywords = [4]{k, y, size, temp, pub, secr_addr, x, eax, edi, eip, secret, z, w},
    morekeywords = [5]{~}
}

\lstdefinestyle{PrologStyle2}
{
	frame = tb,
    belowskip=.4\baselineskip,
    aboveskip=.4\baselineskip,
  	showstringspaces = false,
  	breaklines = true,
  	breakatwhitespace = true,
  	tabsize = 3,
  	numbers = left,
    stepnumber = 1,
    numberstyle = \tiny\color{gray},
    alsoletter={.\$\%},
    basicstyle={\ttfamily\color{black}},
    keywordstyle={\ttfamily\color{Blue3}},
    keywordstyle=[2]{\ttfamily\color{Green4}},
    keywordstyle=[3]{\ttfamily\color{orange}},
    keywordstyle=[4]{\ttfamily\color{violet}},
    keywordstyle=[5]{\ttfamily\color{red}}, 
    keywords = {store, load, beqz, \leftarrow, ret, call, spbarr, jmp,
        xrunv4_1, xc41, enabled, decr, window, update0, spec4, incpc, trace_rawpc, trace, is, f, g, h
    },
    morekeywords = {},
    morekeywords = [2]{A, B},
    morekeywords = [3]{\%rax, \%rbx, \%rcx, \%eax, \%ecx, \%rbp, \%cl, \%rdi, \%edx, \%esi, \%rsi, \%rdx, \%rsp, \%al, \%rip, \%ax, \%bx, \%cx,  \%bp,  \%di, \%dx, \%si, \%sp, \%ip, \%r8b},
    morekeywords = [4]{k, y, size, temp, pub, secr_addr, x, eax, edi, eip, secret, z, w},
    morekeywords = [5]{bp, ~, a, b, c}
}

\newcommand{\clang}{\textsc{Clang}\xspace}

\newcommand{\gcc}{\textsc{Gcc}\xspace}

\AtBeginDocument{\mathchardef\mathcomma\mathcode`\,
  \mathcode`\,="8000 
}
{\catcode`,=\active
  \gdef,{\mathcomma\discretionary{}{}{}}
}

\theoremstyle{definition}

\theoremstyle{plain}

\usepackage{letltxmacro}
\newcommand*{\SavedLstInline}{}
\LetLtxMacro\SavedLstInline\lstinline
\DeclareRobustCommand*{\lstinline}{\ifmmode
    \let\SavedBGroup\bgroup
    \def\bgroup{\let\bgroup\SavedBGroup
      \hbox\bgroup
    }\fi
  \SavedLstInline
}

\definecolor{mygreen}{rgb}{0,0.8,0.6}
\lstdefinestyle{Cstyle2}
{
	frame = tb,
  belowskip=.4\baselineskip,
  aboveskip=.4\baselineskip,
  	showstringspaces = false,
  	breaklines = true,
  	breakatwhitespace = true,
  	tabsize = 3,
  	numbers = left,
    stepnumber = 1,
    numberstyle = \tiny\color{gray},
    language = {[ANSI]C},
    basicstyle=\footnotesize,
    emph={uint32_t, uint8_t},
    emphstyle={\color{blue}},
    alsoletter={.\$},
    basicstyle={\ttfamily\color{black}},
    stringstyle={\ttfamily\color{orange}}
}

\lstdefinestyle{Cstyle}
{
	frame = tb,
  belowskip=.4\baselineskip,
  aboveskip=.4\baselineskip,
  	showstringspaces = false,
  	breaklines = true,
  	breakatwhitespace = true,
  	tabsize = 3,
  	numbers = left,
    stepnumber = 1,
    numberstyle = \tiny\color{gray},
    language = {[ANSI]C},
    alsoletter={.\$},
    basicstyle={\ttfamily\color{black}},
    keywordstyle={\ttfamily\color{Blue3}},
    keywordstyle=[2]{\ttfamily\color{Green4}},
    keywordstyle=[3]{\ttfamily\color{orange}},
    keywordstyle=[4]{\ttfamily\color{violet}},
    otherkeywords = {temp,y,A,B,k,size},
    morekeywords = [2]{A,B},
    morekeywords = [3]{},
    morekeywords = [4]{x,y,p,size,k,temp, secret, public, size_A},
}

\lstdefinestyle{ASMstyle}
{
	frame = tb,
  belowskip=.4\baselineskip,
  aboveskip=.4\baselineskip,
  	showstringspaces = false,
              	breaklines = true,
  	breakatwhitespace = true,
  	tabsize = 3,
  	numbers = left,
    stepnumber = 1,
    numberstyle = \tiny\color{gray},
    alsoletter={.\$\%},
    basicstyle={\ttfamily\color{black}},
    stringstyle={\ttfamily\color{gray}},
    keywordstyle={\ttfamily\color{Blue3}},
    keywordstyle=[2]{\ttfamily\color{Green4}},
    keywordstyle=[3]{\ttfamily\color{Mahogany}},
    keywordstyle=[4]{\ttfamily\color{violet}},
    keywordstyle=[5]{\ttfamily\color{red}}, 
    keywords = {cmova, cmovae, cmovb, cmove, cmovne, cmovbe, mov, xor, or,  jbe, jne, add, jmp,  cmp, shl, and, sar, lfence, mfence, sfence, jae, clt,  and, lea, movl, subl, decl, movb, leave, ret, andl, pushl, call, pause},
    morekeywords = {},
    morekeywords = [2]{A, B},
    morekeywords = [3]{\%rax, \%rbx, \%rcx, \%eax, \%ecx, \%rbp, \%cl, \%rdi, \%edx, \%esi, \%rsi, \%rdx, \%rsp, \%al, \%rip, \%ax, \%bx, \%cx,  \%bp,  \%di, \%dx, \%si, \%sp, \%ip, \%r8b, \%ebp, \%esp},
    morekeywords = [4]{k, y, size, temp},
    morekeywords = [5]{END, .L1, .L2, .LFB12, speculate, return_new}
}

\definecolor{Blue3}{HTML}{0000CD}
\definecolor{Green4}{HTML}{008B00}
\definecolor{Red3}{HTML}{CD0000}
\definecolor{orange}{rgb}{0.8, 0.47, 0.196}

\newcommand\remove[1]{}

 \newcommand\xleadsto[1]{\if\relax\detokenize{#1}\relax
   \rightsquigarrow
   \else
   \mathrel{\begin{tikzpicture}[baseline={(current bounding box.south)}
       ]
       \node[,inner sep=.44ex
       ,align=center
       ] (tmp) {$\scriptstyle #1$};
       \path[,draw,<-
       ,decorate,decoration={,zigzag
         ,amplitude=0.7pt
         ,segment length=1.2mm,pre length=3.5pt
       }
       ] 
       (tmp.south east) -- (tmp.south west);
     \end{tikzpicture}
   }
   \fi
 }

\newcommand{\Obs}{\mathit{Obs}}

\newcommand{\lbl}{\ell}

\newcommand{\barrierKywd}{\kywd{spbarr}}

\newcommand{\rollbackObs}[1]{\rollbackObsKywd{}\ #1}

\newcommand{\pskip}{\skipKywd{}}

\newcommand{\pstore}[2]{\storeKywd\ #1, #2}

\newcommand{\pbarrier}{\barrierKywd}
\newcommand{\select}[2]{#1(#2)}

\newcommand{\id}{\mathit{id}}
\newcommand{\ctr}{\mathit{ctr}}

\newcommand{\lbbar}{\{\kern-0.5ex|}
\newcommand{\rbbar}{|\kern-0.5ex\}}
\newcommand{\lanbar}{\langle \kern-0.5ex|}
\newcommand{\ranbar}{|\kern-0.5ex\rangle}

\newcommand{\abs}[1]{\lvert #1\rvert}

\newcommand{\minWindow}[1]{\mathit{minWndw}(#1)}

\newcommand{\phiStack}{\overline{\Phi}\xspace}

\newcommand{\phiStackS}{\overline{\Phi}_{\Sv}\xspace}

\newcommand{\phiStackB}{\overline{\Phi}_{B}\xspace}

\newcommand{\phiStackR}{\overline{\Phi}_{\Rv}\xspace}

\newcommand{\bnfdef}{\ensuremath{\mathrel{::=}}}
\newcommand{\tauStack}{\overline{\tau}\xspace}

\renewcommand{\specarrowB}[1]{\xrightswishingghost{\mathit{#1}}\mathrel{\vphantom{\to}}_{\Bv}} 
\renewcommand{\specarrowS}[1]{\xrightswishingghost{\mathit{#1}}\mathrel{\vphantom{\to}}_{\Sv}} 
\newcommand{\specarrowJ}[1]{\xrightswishingghost{\mathit{#1}}\mathrel{\vphantom{\to}}_{\Jv}} 
 
\renewcommand{\specarrowR}[1]{\xrightswishingghost{\mathit{#1}}\mathrel{\vphantom{\to}}_{\Rv}} 
\newcommand{\specarrowSLS}[1]{\xrightswishingghost{\mathit{#1}}\mathrel{\vphantom{\to}}_{\SLSv}}

\newcommand{\specarrowBS}[1]{\xrightswishingghost{\mathit{#1}}\mathrel{\vphantom{\to}}_{B+S}}

\newcommand{\mi}[1]{\ensuremath{\mathit{#1}}}

\newcommand{\mtt}[1]{\ensuremath{\mathtt{#1}}}
\newcommand{\mf}[1]{\ensuremath{\mathbf{#1}}}

\newcommand{\mb}[1]{\ensuremath{\mathbb{#1}}}

\newcommand{\beh}[0]{\mi{Beh}}
\newcommand{\behNs}[1]{\beh_{NS}(#1)}
\newcommand{\Trace}[2]{#1 ~\mathghost_{\mi{NS}}~ #2}

\newcommand{\behB}[1]{\beh^{\omega}_{\Bv}(#1)}

\newcommand{\behS}[1]{\beh^{\omega}_{\Sv}(#1)}

\newcommand{\behJ}[1]{\beh^{\omega}_{\Jv}(#1)}

\newcommand{\behSLS}[1]{\beh^{\omega}_{\SLSv}(#1)}

\newcommand{\behR}[1]{\beh^{\omega}_{\Rv}(#1)}

\newcommand{\amTracevB}[2]{#1 ~\mathghost^{\omega}_{\Bv}~ #2}

\newcommand{\amTracevS}[2]{#1 ~\mathghost^{\omega}_{\Sv}~ #2}

\newcommand{\amTracevR}[2]{#1 ~\mathghost^{\omega}_{\Rv}~ #2}

\newcommand{\amTracevJ}[2]{#1 ~\mathghost^{\omega}_{\Jv}~ #2}

\newcommand{\amTracevSLS}[2]{#1 ~\mathghost^{\omega}_{\SLSv}~ #2}

\newcommand{\PhiR}{\Phi_{\Rv}}
\newcommand{\SigmaR}{\Sigma_{\Rv}}

\newcommand{\SigmaRDS}[1]{{\SigmaR}_{#1}} \newcommand{\SigmaRDSP}[1]{{\SigmaR'}_{#1}}

\newcommand{\SigmaB}{\Sigma_{\Bv}}
\newcommand{\PhiB}{\Phi_{\Bv}}

\newcommand{\SigmaS}{\trgS{\Sigma_{\Sv}}}
\newcommand{\PhiS}{\Phi_{\Sv}}

\newcommand{\SigmaSDS}[1]{{\SigmaS}_{#1}} \newcommand{\SigmaSDSP}[1]{{\SigmaS'}_{#1}}

\newcommand{\SigmaJ}{\Sigma_{\Jv}}
\newcommand{\PhiJ}{\Phi_{\Jv}}
\newcommand{\phiStackJ}{\overline{\Phi}_{\Jv}\xspace}
\newcommand{\SigmaJDS}[1]{{\SigmaJ}_{#1}} \newcommand{\SigmaJDSP}[1]{{\SigmaJ'}_{#1}}

\newcommand{\SigmaSLS}{\trgSLS{\Sigma_{\SLSv}}}
\newcommand{\PhiSLS}{\Phi_{\SLSv}}
\newcommand{\phiStackSLS}{\overline{\Phi}_{\SLSv}\xspace}

\newcommand{\SigmaSLSDS}[1]{{\SigmaSLS}_{#1}} \newcommand{\SigmaSLSDSP}[1]{{\SigmaSLS'}_{#1}}

\newcommand{\SigmaBS}{\Sigma_{\Bv + \Sv}}

\newcommand{\SigmaBR}{\Sigma_{B + R}}

\newcommand{\SigmaBSR}{\Sigma_{145}}

\newcommand{\bigspecarrowB}[1]{\prescript{}{}{\Downarrow}^{#1}_{\Bv}~}
\newcommand{\bigspecarrowS}[1]{\trgS{\prescript{}{}{\Downarrow}^{#1}}_{\Sv}~}
\newcommand{\bigspecarrowJ}[1]{\prescript{}{}{\Downarrow}^{#1}_{\Jv}~}

\newcommand{\bigspecarrowR}[1]{\prescript{}{}{\Downarrow}^{#1}_{\Rv}~}
\newcommand{\bigspecarrowSLS}[1]{\prescript{}{}{\Downarrow}^{#1}_{\SLSv}~}

\newcommand{\bigspecarrowBR}[1]{\prescript{}{}{\Downarrow}^{#1}_{15}~}
\newcommand{\bigspecarrowBS}[1]{\prescript{}{}{\Downarrow}^{#1}_{\Bv + \Sv}~}
\newcommand{\bigspecarrowBSR}[1]{\prescript{}{}{\Downarrow}^{#1}_{145}~}

\newcommand{\nsbigarrow}[1]{\Downarrow_{#1}}

 \newcommand{\nsarrow}[1]{\xrightarrow{#1}}

\newcommand{\initFunc}[1]{\src{\Omega_{0}}#1}
\newcommand{\initFuncB}[1]{\SigmaB^{\mi{init}}#1}
\newcommand{\initFuncJ}[1]{\SigmaJ^{\mi{init}}#1}

\newcommand{\initFuncSLS}[1]{\SigmaSLS^{\mi{init}}#1}
\newcommand{\initFuncS}[1]{\SigmaS^{\mi{init}}#1}
\newcommand{\initFuncR}[1]{\SigmaR^{\mi{init}}#1}

\newcommand{\finType}[1]{\vdash #1 \colon \mi{fin}}
\newcommand{\finTypef}[1]{\vdash #1 \colon \mi{fin_{f}}}

\newcommand{\empTr}{\varepsilon}

\newcommand{\sni}[0]{\text{SNI}\xspace}

\newcommand{\rsni}[0]{\text{RSNI}\xspace}

\newcommand{\rsstext}[0]{\text{RSS}\xspace}

\newcommand{\rss}[0]{\text{RSS}\xspace}
\newcommand{\rssdef}[0]{\criteria{\rsstext}{rss}}

\newcommand{\Sigmax}{\Sigma_{x}}
\newcommand{\Sigmay}{\Sigma_{y}}
\newcommand{\Sigmaxy}{\Sigma_{xy}}

\newcommand{\Phix}{\Phi_{x}}
\newcommand{\Phiy}{\Phi_{y}}
\newcommand{\Phixy}{\Phi_{xy}}

\newcommand{\phiStackxy}{\overline{\Phi}_{xy}\xspace}
\newcommand{\phiStackx}{\overline{\Phi}_{x}\xspace}
\newcommand{\phiStacky}{\overline{\Phi}_{y}\xspace}

\newcommand{\Obsx}{\Obs_{x}}
\newcommand{\Obsy}{\Obs_{y}}
\newcommand{\Obsxy}{\Obs_{xy}}

\newcommand{\joinxy}{\sqcup_{xy}}

\newcommand{\specProjectxy}[1]{#1\mathord{\upharpoonright_{xy}}}
\newcommand{\specProjectxyx}[1]{#1\mathord{\upharpoonright_{xy}^{x}}}
\newcommand{\specProjectxyy}[1]{#1\mathord{\upharpoonright_{xy}^{y}}}
\newcommand{\specProjectComp}[2]{#1\mathord{\upharpoonright_{#2}}}

\newcommand{\behxy}[1]{\beh^{\omega}_{xy}#1}
\newcommand{\behx}[1]{\beh^{\omega}_x(#1)}
\newcommand{\behy}[1]{\beh^{\omega}_y(#1)}

\newcommand{\startObsx}[1]{\startObsKywd{}_{x}\ #1}
\newcommand{\rollbackObsx}[1]{\rollbackObsKywd{}_{x}\ #1}

\newcommand{\startObsy}[1]{\startObsKywd{}_{y}\ #1}
\newcommand{\rollbackObsy}[1]{\rollbackObsKywd{}_{y}\ #1}

\newcommand{\specarrowxy}[1]{\xrightswishingghost{\mathit{#1}}\mathrel{\vphantom{\to}}_{xy}}
\newcommand{\specarrowx}[1]{\xrightswishingghost{\mathit{#1}}\mathrel{\vphantom{\to}}_{x}}
\newcommand{\specarrowy}[1]{\xrightswishingghost{\mathit{#1}}\mathrel{\vphantom{\to}}_{y}}

\newcommand{\bigspecarrowxy}[1]{\prescript{}{}{\Downarrow}^{#1}_{xy}~}
\newcommand{\bigspecarrowx}[1]{\prescript{}{}{\Downarrow}^{#1}_{x}~}

\newcommand{\Zxy}{Z_{xy}}

\newcommand{\specarrowxZ}[1]{\xrightswishingghost{\mathit{#1}}\mathrel{\vphantom{\to}}_{x}^{\textcolor{Blue3}{\specProjectxyx{\Zxy}}}}
\newcommand{\specarrowyZ}[1]{\xrightswishingghost{\mathit{#1}}\mathrel{\vphantom{\to}}_{y}^{\textcolor{Blue3}{\specProjectxyy{\Zxy}}}}

\newcommand{\amTracevxy}[2]{#1 ~\mathghost^{\omega}_{xy}~ #2}

\newcommand{\sem}[1]{\mathghost_{#1}}

\newcommand{\semb}{\sem{\Bv}} 
\newcommand{\semj}{\sem{\Jv}} 
\newcommand{\sems}{\sem{\Sv}} 
\newcommand{\semr}{\sem{\Rv}} 
\newcommand{\semsls}{\sem{\SLSv}}

\newcommand{\semd}{\sem{\Div}}
\newcommand{\semms}{\sem{\MS}}
\newcommand{\semct}{\sem{\CT}}

\newcommand{\sigmav}{\sigma_{v}}
\newcommand{\sigmat}{\sigma_{\taint}}

\newcommand{\av}{a_{v}}
\newcommand{\mv}{M_{v}}
\newcommand{\Bva}{B_{v}}

\newcommand{\at}{a_{t}}
\newcommand{\mta}{M_{t}}
\newcommand{\Bt}{B_{t}}

\DeclareMathOperator{\findfun}{find\_fun}
\newcommand{\ffun}[1]{\findfun(#1)} \DeclareMathOperator{\getprog}{get\_program}
\newcommand{\gprog}[2]{\getprog_{#1}(#2)}

\newcommand{\tauStackT}{\overline{\tau^{\taint}}\xspace}

\newcommand{\PhiSv}{\Phi_{\Sv}^{v}}
\newcommand{\PhiSt}{\Phi_{\Sv}^{t}}
\newcommand{\phiStackSv}{\overline{\Phi^{v}_{\Sv}}\xspace}
\newcommand{\phiStackSt}{\overline{\Phi^{t}_{\Sv}}\xspace}

\newcommand{\PhiJv}{\Phi_{\Jv}^{v}}
\newcommand{\PhiJt}{\Phi_{\Jv}^{t}}
\newcommand{\phiStackJv}{\overline{\Phi^{v}_{\Jv}}\xspace}
\newcommand{\phiStackJt}{\overline{\Phi^{t}_{\Jv}}\xspace}

\newcommand{\PhiSLSv}{\Phi_{\SLSv}^{v}}
\newcommand{\PhiSLSt}{\Phi_{\SLSv}^{t}}
\newcommand{\phiStackSLSv}{\overline{\Phi^{v}_{\SLSv}}\xspace}
\newcommand{\phiStackSLSt}{\overline{\Phi^{t}_{\SLSv}}\xspace}

\newcommand{\PhiRv}{\Phi_{\Rv}^{v}}
\newcommand{\PhiRt}{\Phi_{\Rv}^{t}}
\newcommand{\phiStackRv}{\overline{\Phi^{v}_{\Rv}}\xspace}
\newcommand{\phiStackRt}{\overline{\Phi^{t}_{\Rv}}\xspace}

\newcommand{\PhiBv}{\Phi_{\Bv}^{v}}
\newcommand{\PhiBt}{\Phi_{\Bv}^{t}}
\newcommand{\phiStackBv}{\overline{\Phi^{v}_{\Bv}}\xspace}

\newcommand{\Phixyv}{\Phixy^{v}}
\newcommand{\Phixyt}{\Phixy^{t}}
\newcommand{\phiStackxyv}{\overline{\Phixy^{v}}\xspace}
\newcommand{\phiStackxyt}{\overline{\Phixy^{t}}\xspace}

\newcommand{\Omegav}{\Omega_{v}}
\newcommand{\Omegat}{\Omega_{t}}

\newcommand{\instrOv}[1]{p(\Omegav(\pc)) = #1}
\newcommand{\instrneqOv}[1]{p(\Omegav(\pc)) \ne #1}

\newcommand{\instrOt}[1]{p(\Omegat(\pc)) = #1}
\newcommand{\instrneqOt}[1]{p(\Omegat(\pc)) \ne #1}

\newcommand{\trgxy}[1]{{\mf{\col{\neutcol }{#1}}}}

\newcommand{\labels}[1]{labels(#1)} 

\newcommand{\low}{\kywd{L}}

\newcommand{\high}{\kywd{H}}

\newcommand{\MemL}[1]{M^{\low#1}}
\newcommand{\MemH}[1]{M^{\high#1}}

\newcommand{\SInit}[1]{\ensuremath{{\Omega_0}\left({#1}\right)}\xspace}
\newcommand{\SInits}[1]{\ensuremath{\src{\Omega_0}\left(\src{#1}\right)}\xspace}

\newcommand{\SR}[0]{\src{L}\xspace}

\def\hrelssaux#1{\vcenter{\hbox{\ooalign{\hfil
       \raise6pt \hbox{\tiny{\com{H}}}\hfil\cr\hfil
       $\approx$}}}}
\def\hrelssb{\mathrel{\mathpalette\hrelssaux{}}}

\DeclareMathOperator\hrel{\ensuremath{\com{\hrelssb}}}

\newcommand{\hreltext}[0]{\hrel}

\newcommand{\hreldef}[0]{\criteria{\hreltext}{htracerel}}

\def\vrelssaux#1{\vcenter{\hbox{\ooalign{\hfil
       \raise6pt \hbox{\tiny{\com{V}}}\hfil\cr\hfil
       $\approx$}}}}
\def\vrelssb{\mathrel{\mathpalette\vrelssaux{}}}

\DeclareMathOperator\vrel{\ensuremath{\com{\vrelssb}}}

\newcommand{\vreltext}[0]{\vrel}
\newcommand{\vrelref}[0]{\ensuremath{\Cref{cr:vtracerel}}\xspace}

\newcommand{\vreldef}[0]{\criteria{\vreltext}{vtracerel}}

\def\arelssaux#1{\vcenter{\hbox{\ooalign{\hfil
       \raise6pt \hbox{\tiny{\com{A}}}\hfil\cr\hfil
       $\approx$}}}}
\def\arelssb{\mathrel{\mathpalette\arelssaux{}}}

\DeclareMathOperator\arel{\ensuremath{\com{\arelssb}}}

\newcommand{\areltext}[0]{\arel}
\newcommand{\arelref}[0]{\ensuremath{\Cref{cr:atracerel}}\xspace}

\newcommand{\areldef}[0]{\criteria{\areltext}{atracerel}}

\def\srelssaux#1{\vcenter{\hbox{\ooalign{\hfil
       \raise6pt \hbox{\tiny{\com{S}}}\hfil\cr\hfil
       $\approx$}}}}
\def\srelssb{\mathrel{\mathpalette\srelssaux{}}}

\DeclareMathOperator\srel{\ensuremath{\com{\srelssb}}}

\newcommand{\sreltext}[0]{\srel}
\newcommand{\srelref}[0]{\ensuremath{\Cref{cr:stracerel}}\xspace}

\newcommand{\sreldef}[0]{\criteria{\sreltext}{stracerel}}

\def\rrelssaux#1{\vcenter{\hbox{\ooalign{\hfil
       \raise6pt \hbox{\tiny{\com{R}}}\hfil\cr\hfil
       $\approx$}}}}
\def\rrelssb{\mathrel{\mathpalette\rrelssaux{}}}

\DeclareMathOperator\rrel{\ensuremath{\com{\rrelssb}}}

\newcommand{\rreltext}[0]{\rrel}
\newcommand{\rrelref}[0]{\ensuremath{\Cref{cr:rtracerel}}\xspace}

\newcommand{\rreldef}[0]{\criteria{\rreltext}{rtracerel}}

\def\crelssaux#1{\vcenter{\hbox{\ooalign{\hfil
       \raise6pt \hbox{\tiny{\com{C}}}\hfil\cr\hfil
       $\approx$}}}}
\def\crelssb{\mathrel{\mathpalette\crelssaux{}}}

\DeclareMathOperator\crel{\ensuremath{\com{\crelssb}}}

\newcommand{\creltext}[0]{\crel}

\newcommand{\creldef}[0]{\criteria{\creltext}{ctracerel}}

\def\brelssaux#1{\vcenter{\hbox{\ooalign{\hfil
       \raise6pt \hbox{\tiny{\com{B}}}\hfil\cr\hfil
       $\approx$}}}}
\def\brelssb{\mathrel{\mathpalette\brelssaux{}}}

\DeclareMathOperator\brel{\ensuremath{\com{\brelssb}}}

\newcommand{\wfc}[1]{\ensuremath{\vdash#1 : \ensuremath{\mi{WFC}}}}
\newcommand{\wfcn}[2]{\ensuremath{\vdash#1 : \ensuremath{\mi{WFC}_{#2}}}}
 
\newcommand{\relsa}{\approx~}

\newcommand{\SigmaSt}[1]{\trgS{\Sigma^{#1}_{\Sv}}} \newcommand{\SigmaJt}[1]{\trgJ{\Sigma^{#1}_{\Jv}}}
\newcommand{\SigmaRt}[1]{\trgR{\Sigma^{#1}_{\Rv}}}

\newcommand{\SigmaSLSt}[1]{\trgSLS{\Sigma^{#1}_{\SLSv}}} 

\newcommand{\Sigmaxyt}[1]{\Sigma^{#1}_{xy}}

\newcommand{\Sigmaxt}[1]{\Sigma^{#1}_{x}}
\newcommand{\Sigmayt}[1]{\Sigma^{#1}_{y}}

\newcommand{\relsax}{\approx_{x}~}
\newcommand{\relsay}{\approx_{y}~}
\newcommand{\relsaxy}{\approx_{xy}~}

\newcommand{\notfinType}[1]{\nvdash #1 \colon \mi{fin}}

\newcommand{\labelset}{L}
\newcommand{\labelmap}{\mathcal{L}}

\DeclareMathOperator{\inc}{inc}
\DeclareMathOperator{\succes}{succ}

\newcommand{\OmegaS}{\trgS{\Omega}}
\newcommand{\OmegaSLS}{\trgSLS{\Omega}}
\newcommand{\OmegaR}{\trgR{\Omega}}
\newcommand{\lub}{\sqcup}

\newcommand{\relslh}[0]{\Bumpeq} \newcommand{\rslh}{r_{slh}}
\newcommand{\rslhC}{r_{slhC}} \newcommand{\rscr}{r_{scratch}}

\def\shrelssaux#1{\vcenter{\hbox{\ooalign{\hfil
       \raise6pt \hbox{\tiny{\com{H}}}\hfil\cr\hfil
       $\relslh$}}}}
\def\shrelssb{\mathrel{\mathpalette\shrelssaux{}}}

\DeclareMathOperator\shrel{\ensuremath{\com{\shrelssb}}}

\def\ssrelssaux#1{\vcenter{\hbox{\ooalign{\hfil
       \raise6pt \hbox{\tiny{\com{S}}}\hfil\cr\hfil
       $\relslh$}}}}
\def\ssrelssb{\mathrel{\mathpalette\ssrelssaux{}}}

\DeclareMathOperator\ssrel{\ensuremath{\com{\ssrelssb}}}

\newcommand{\ssreltext}[0]{\ssrel}
\newcommand{\ssrelref}[0]{\ensuremath{\Cref{cr:sstracerel}}\xspace}

\newcommand{\ssreldef}[0]{\criteria{\ssreltext}{sstracerel}}

\newcommand{\binop}[0]{\otimes}

\newcommand{\branchlabels}[0]{\mathcal{B}_{l}}

\newcommand{\SI}[3]{#3\vdash #2 : SI} 

\newcommand{\IND}[2]{#1 \vdash #2 : I}
\def\botrule{\vspace{0mm}\hrule\vspace{2mm}}

\newcommand{\neutcol}[0]{black}
\newcommand{\stlccol}[0]{RoyalBlue}
\newcommand{\ulccol}[0]{RedOrange}

\newcommand{\commoncol}[0]{black}    

\def\relssa{\approx}  

\newcounter{proofr}
\crefname{proofr}{}{}
\newcommand{\proofref}[2]{\def\theproofref{\detokenize{#1}}\refstepcounter{proofr}\label{prf:#2}}

\newcommand{\showproof}[1]{Proof \Cref{prf:#1}}

\DeclareMathOperator\rels{\ensuremath{\com{\relssa}}}

\newcommand{\reltext}[0]{\relssa}
\newcommand{\relref}[0]{\ensuremath{\Cref{cr:tracerel}}\xspace}
\newcommand{\tracerel}[0]{\relref}
\newcommand{\reldef}[0]{\criteria{\reltext}{tracerel}}

\newcommand{\op}[0]{\ensuremath{\oplus}}
\newcommand{\bop}[0]{\ensuremath{\otimes}}

\newcommand{\bl}[1]{\col{\neutcol }{#1}}

\newcommand{\col}[2]{\ensuremath{{\color{#1}{#2}}}}
\newcommand{\com}[1]{\mi{\col{\commoncol }{#1}}}
\newcommand{\comb}[1]{\ensuremath{\bm{\col{\neutcol }{#1}}}}
\newcommand{\OB}[1]{\ensuremath{\overline{#1}}}

\newcommand{\relmiddle}[1]{\mathrel{}\middle#1\mathrel{}}
\newcommand{\myset}[2]{\ensuremath{\left\{#1 ~\relmiddle|~ #2\right\}}}
\DeclareMathOperator\glb{\ensuremath{\sqcap}}

\newcommand{\come}[0]{\com{\emptyset}\xspace}
\newcommand{\src}[1]{{\ensuremath{\textsf{{\col{\stlccol }{#1}}}}}} \newcommand{\trg}[1]{{\mf{\col{\ulccol }{#1}}}}

\newcommand{\srce}[0]{\src{\emptyset}\xspace}
\newcommand{\trgSe}[0]{\trgS{\emptyset}\xspace}
\newcommand{\trgBe}[0]{\trgB{\emptyset}\xspace}
\newcommand{\safeta}[0]{\ensuremath{S}}
\newcommand{\unta}[0]{\ensuremath{U}}

\newcommand{\safe}[1]{\vdash #1 : \mathit{safe}}
\newcommand{\unsafe}[1]{\vdash #1 : \mathit{unsafe}}

\newcommand{\contextletter}[0]{A}
\newcommand{\ctx}[1]{\ensuremath{\contextletter}} \newcommand{\ctxs}[1]{\src{\ctx{\contextletter}#1}\xspace}
\newcommand{\ctxB}[1]{\trgB{\ctx{\contextletter}#1}\xspace}\newcommand{\ctxS}[1]{\trgS{\ctx{\contextletter}#1}\xspace}\newcommand{\ctxR}[1]{\trgR{\ctx{\contextletter}#1}\xspace}\newcommand{\ctxSLS}[1]{\trgSLS{\ctx{\contextletter}#1}\xspace}\newcommand{\ctxc}[1]{\com{\ctx{\contextletter}#1}\xspace}\newcommand{\hole}[1]{\ensuremath{\left[#1\right]}}

\newcommand{\bigred}[0]{\ensuremath{\ \downarrow\ }}
\DeclareMathOperator\bigreds{\bigred}

\newcommand{\taint}[0]{t}

\newcommand{\taintpc}[0]{\taint_{pc}}
\newcommand{\tracesymbol}[0]{\tau}
\newcommand{\actionsymbol}[0]{\tau}

\newcommand{\terminationsymbol}[0]{\lightning}
\newcommand{\ts}[0]{\terminationsymbol}
\newcommand{\proc}[2]{\ensuremath{(#1)_{#2} }}

\newcommand{\taintS}{\trgS{\sigma}}
\newcommand{\taintB}{\trgB{\sigma}}

\newcommand{\taintSLS}{\trgSLS{\sigma}}

\newcommand{\aca}[1]{\actionsymbol{#1}}
\newcommand{\acac}[1]{\com{\actionsymbol{#1}}}
\newcommand{\acas}[1]{\src{\actionsymbol{#1}}}

\newcommand{\tra}[1]{\OB{\tracesymbol{#1}}}
\newcommand{\tras}[1]{\src{\tra{#1}}}
\newcommand{\trac}[1]{\com{\tra{#1}}}

\newcommand{\behav}[1]{{Beh}{\left(#1\right)}}
\newcommand{\behavs}[1]{\src{\behav{#1}}}

\newcommand{\behavc}[1]{\com{\behav{#1}}}

\newcommand{\funname}[1]{\mtt{#1}}
\newcommand{\fun}[2]{\ensuremath{{\bl{\funname{#1}\left(#2\right)}}}\xspace}
\newcommand{\dom}[1]{\fun{dom}{#1}}

\newcommand{\asm}[1]{\mtt{#1}}
\newcommand{\labelfont}[1]{\ensuremath{\asm{#1}}}
\newcommand{\clh}[1]{\ensuremath{\labelfont{call}~ #1{?}}} 	\newcommand{\cbh}[1]{\ensuremath{\labelfont{call}~ #1{!}}} 	

\newcommand{\rth}[0]{\ensuremath{\labelfont{ret}{!}}}	 		\newcommand{\rbh}[0]{\ensuremath{\labelfont{ret}{?}}}

\def\dotminus{\mathbin{\ooalign{\hss\raise1ex\hbox{.}\hss\cr
  \mathsurround=0pt$-$}}}

\newcommand{\compskel}[5]{\ensuremath{\bl{\prescript{#1}{#2}{\left\llbracket \src{#3} \right\rrbracket^{#4}_{#5}}}}}
\newcommand{\comp}[1]{\compskel{}{}{\bl{#1}}{}{}}

\newcommand{\rplSym}{rpl} 

\newcommand{\complfenceB}[1]{\compskel{}{}{#1}{\com{f}}{\Bv}} \newcommand{\compuslhB}[1]{\compskel{}{}{#1}{\com{USLH}}{\Bv}}
\newcommand{\compsslhB}[1]{\compskel{}{}{#1}{\com{SSLH}}{\Bv}}
\newcommand{\compretpJ}[1]{\compskel{}{}{#1}{\com{\rplSym}}{\Jv}}  \newcommand{\compretpJF}[1]{\compskel{}{}{#1}{\com{\rplSym}f}{\Jv}}
\newcommand{\compretpR}[1]{\compskel{}{}{#1}{\com{\rplSym}}{\Rv}}  \newcommand{\complfenceS}[1]{\compskel{}{}{#1}{\com{f}}{\Sv}} \newcommand{\complfenceSLS}[1]{\compskel{}{}{#1}{\com{f}}{\SLSv}} \newcommand{\complfenceR}[1]{\compskel{}{}{#1}{\com{f}}{\Rv}}

\newcommand{\isdef}[0]{\ensuremath{\mathrel{\overset{\makebox[0pt]{\mbox{\normalfont\tiny\sffamily def}}}{=}}}}
\newcommand{\loweq}{\ensuremath{=_{L}}}

\newcommand{\formatCompilers}[1]{\mf{\mi{#1}}\xspace}
\newcounter{criteria}
\crefname{criteria}{}{}
\newcommand{\criteria}[2]{\def\thecriteria{\detokenize{#1}}\refstepcounter{criteria}\label{cr:#2}#1}

\newcommand{\rdsscomp}[0]{\formatCompilers{RSSC}}
\newcommand{\rdss}[0]{\Cref{cr:rdss}\xspace}

\newcommand{\rdssdef}[0]{\criteria{\rdsscomp}{rdss}}

\newcommand{\rdsspcomp}[0]{\formatCompilers{RSSP}}
\newcommand{\rdssp}[0]{\Cref{cr:rdssp}\xspace}
\newcommand{\rssp}[0]{\mi{RSSP}}
\newcommand{\rdsspdef}[0]{\criteria{\rdsspcomp}{rdssp}}

\newcommand{\rsnip}[0]{\mi{RSNIP}}

\newcommand{\sstext}[0]{\text{SS}}
\newcommand{\ssref}[0]{\Cref{cr:ss}\xspace}
\renewcommand{\ss}[0]{\ssref}
\newcommand{\ssdef}[0]{\criteria{\sstext}{ss}}

\newcommand{\backtrskel}[3]{\ensuremath{\bl{\left\langle\!\left\langle {#1} \right\rangle\!\right\rangle^{#2}_{#3}}}}

\newcommand{\backtrfencec}[1]{\backtrskel{#1}{}{\com{c}}}

\newcommand{\contractColor}{Black}
\newcommand{\contra}[1]{{\ensuremath{\textsf{{\col{\contractColor}{#1}}}}}}
\newcommand{\Mcontra}[1]{{\ensuremath{\textit{{\col{Green}{#1}}}}}}
\newcommand{\contract}[2]{\contra{\mathghost^{#1}_{#2}}} \newcommand{\Mcontract}[2]{\Mcontra{\mathfrak{C}^{#1}_{#2}}} \newcommand{\contractSpec}[2]{\contra{\mathghost^{#1}_{#2}}} \newenvironment{insight}
    {\begin{center}
    \begin{tabular}{|p{0.9\textwidth}|}
    \hline\\
    }
    { 
    \\\\\hline
    \end{tabular} 
    \end{center}
    }

\newcommand{\BREAK}[0]{
\botrule
\begin{center}$\spadesuit$\end{center}
\botrule}

\newcommand{\techReportAppendix}[1]{\IfEqCase{\shortVersion}{{false}{Appendix \Cref{#1}}{true}{\cite{techReport}}}
}

\newcommand{\techReportAppendices}[2]{\IfEqCase{\shortVersion}{{false}{Appendices \Cref{#1}--\Cref{#2}}{true}{\cite{techReport}}}
}

\newcommand{\onlyTechReport}[1]{\IfEqCase{\shortVersion}{{false}{#1}}
}

\newcommand{\onlyShortVersion}[1]{\IfEqCase{\shortVersion}{{true}{#1}}
}

\newcommand{\IfShortFirstOption}[2]{
    \IfEqCase{\shortVersion}{{true}{#1}{false}{#2}
    }
}

\pgfdeclarelayer{background}
\pgfdeclarelayer{veryback}
\pgfdeclarelayer{veryback2}
\pgfdeclarelayer{veryback3}
\pgfdeclarelayer{back2}
\pgfdeclarelayer{foreground}
\pgfsetlayers{veryback3,veryback2,veryback,background,back2,main,foreground}

\newcommand{\SITab}{\textcolor{PineGreen}{SI}}
\newcommand{\ITab}{\textcolor{BurntOrange}{I}} \newcommand{\N}{\textcolor{Mahogany}{\xmark}} \newcommand{\specInstr}{\mi{instr}_{\mathghost}}

\newcommand{\clgen}[1]{\ensuremath{\labelfont{call}~#1}}
\newcommand{\rtgen}{\ensuremath{\labelfont{ret}}}

\newcommand{\storel}[1]{\ensuremath{\labelfont{store}(#1)}}
\newcommand{\loadl}[1]{\ensuremath{\labelfont{load}(#1)}}
\newcommand{\pcl}[1]{\ensuremath{\labelfont{pc}(#1)}}
\newcommand{\startl}[1]{\ensuremath{\labelfont{start}_{#1}}}
\newcommand{\rollbl}[1]{\ensuremath{\labelfont{rlb}_{#1}}}

\newcommand{\trappedS}[1]{#1 : \text{trappedSpec}}

\newcommand{\pmodret}[1]{\kywd{modret}\ #1}
\newcommand{\ppopret}{\kywd{popret}}

\newcommand{\modretC}{modret}
\newcommand{\popretC}{popret}
\newcommand{\insertedInstrs}[1]{\mathit{injInst}(#1)}
\newcommand{\speculationInstrs}[1]{\mathit{specInst}(#1)}

\newcommand{\initFuncxy}[1]{\Sigmaxy^{\mi{init}}#1}
\newcommand{\initFuncx}[1]{\Sigmax^{\mi{init}}#1}
\newcommand{\initFuncy}[1]{\Sigmay^{\mi{init}}#1}

\newcommand{\safeN}[1]{\text{safeN}_{#1}} \newcommand{\trappedC}[1]{\mi{trapped}}

\newcommand{\crssp}{\textit{CRSSP}}

\newcommand{\slhInv}{Inv_{slh}}
\newcommand{\nesting}[0]{Nest}

\newcommand{\SigmaBSLS}{\Sigma_{\Bv + \SLSv}}
\newcommand{\SigmaBSSLS}{\Sigma_{\Bv + \Sv + \SLSv}}

\newcommand{\bigspecarrowBSLS}[1]{\prescript{}{}{\Downarrow}^{#1}_{\Bv + \SLSv}~}
\newcommand{\bigspecarrowBSSLS}[1]{\prescript{}{}{\Downarrow}^{#1}_{\Bv + \Sv + \SLSv}~}
\newcommand{\region}[1]{Region~#1}

\newcommand{\coqed}[0]{\CoqSymbol\xspace}
\newcommand{\BareCoqSymbol}{\includegraphics[height=0.9em]{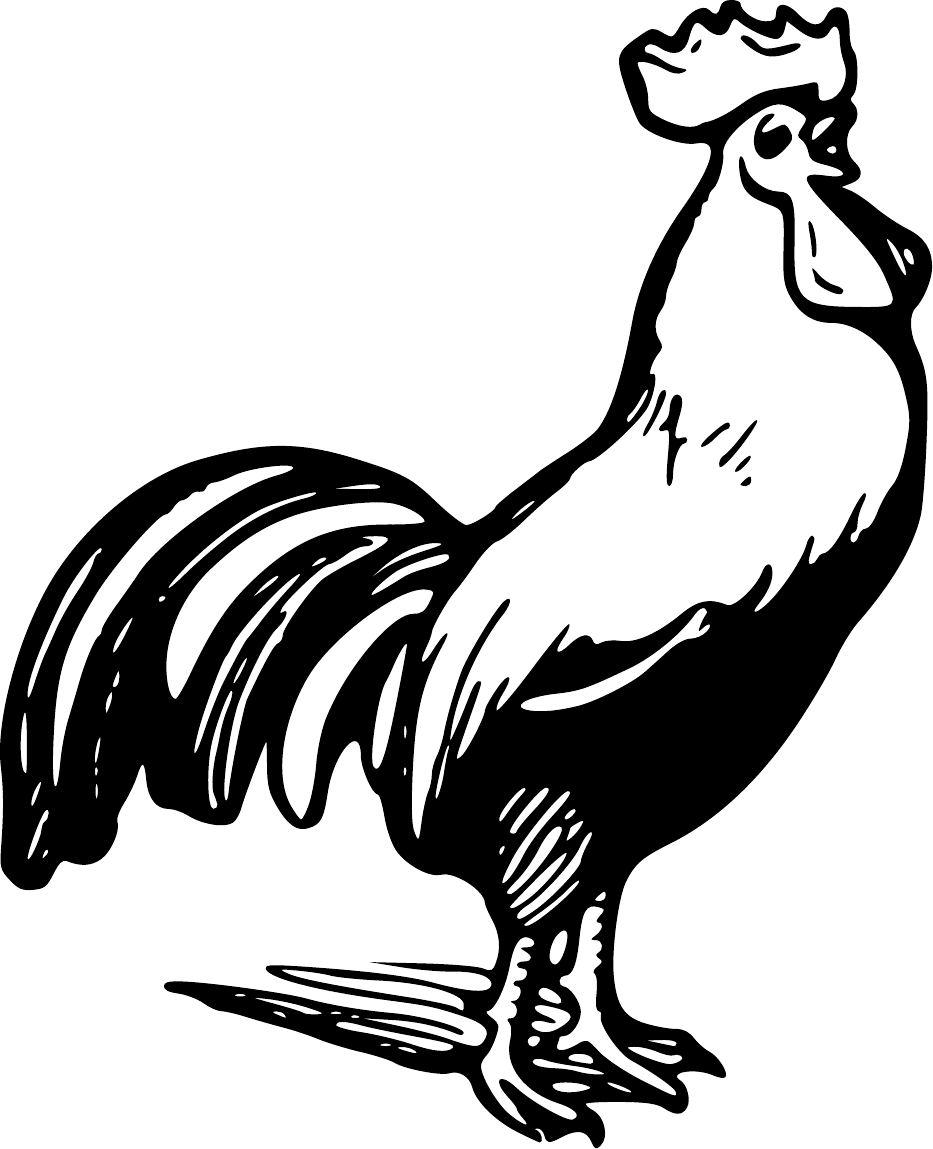}}
\newcommand{\CoqSymbol}{\raisebox{-.2ex}{\BareCoqSymbol\,}}

 \newcounter{typerule}
\crefname{typerule}{rule}{rules}

\newcommand{\typeruleInt}[5]{\def\thetyperule{#1}\refstepcounter{typerule}\label{tr:#4}\ensuremath{\begin{array}{c}#5 \inference{#2}{#3}\end{array}} 
}
\newcommand{\typerule}[4]{\typeruleInt{#1}{#2}{#3}{#4}{\textsf{\scriptsize ({#1})} \\      }
}

\newcommand{\mytoprule}[1]{\vspace{1mm}\noindent\hrulefill\ \raisebox{-0.5ex}{\fbox{\ensuremath{#1}}} \hrulefill\hrulefill\hrulefill\vspace{0.5mm}}

\makeatletter
\xdef\@thefnmark{\@empty}

\newcommand{\Thmref}[1]{\Cref{#1}~(\nameref{#1})}
\makeatother

\clearpage{}
\usepackage{alphalph,etoolbox}
\appto\appendix{
}

\setlist[description]{style=nextline}
\newcommand{\MP}[1]{}
\newcommand{\MPin}[1]{}

\begin{document}

\begin{abstract}
  Mainstream compilers implement different  countermeasures to prevent specific classes of speculative execution attacks.
  Unfortunately, these countermeasures either lack formal guarantees or come with proofs restricted to speculative semantics capturing only a subset of the speculation mechanisms supported by modern CPUs, thereby limiting their practical applicability.
  Ideally, these security proofs should target a speculative semantics capturing the effects of \emph{all} speculation mechanisms implemented in modern CPUs.
  However, this is impractical and requires new secure compilation proofs to support additional speculation mechanisms.

  In this paper, we address this problem by proposing a novel secure compilation framework that allows \emph{lifting} the security guarantees provided by Spectre countermeasures from weaker speculative semantics (ignoring some speculation mechanisms) to stronger ones (accounting for the omitted mechanisms) \emph{without} requiring new secure compilation proofs.
  Using our lifting framework, we performed the most comprehensive security analysis of Spectre countermeasures implemented in mainstream compilers to date.
  Our analysis spans \totcountermeasures different countermeasures against 5 classes of Spectre attacks, which we proved secure against a speculative semantics accounting for five different speculation mechanisms.
  Our analysis highlights that  fence-based and retpoline-based countermeasures can be securely lifted to the strongest speculative semantics under study.
  In contrast, countermeasures based on speculative load hardening   cannot be securely lifted to semantics supporting indirect jump speculation.

\begin{center}\small\it
This paper uses  syntax highlighting accessible to both colourblind and black \& white readers.
  \\
For a better experience, please print or view this in colour~\cite{patrignani2020should}.
\end{center}   
\vspace{-5pt}
\end{abstract}
\title{Do You Even Lift? Strengthening Compiler Security Guarantees Against Spectre Attacks - Technical Report}

\author{Xaver Fabian}
\affiliation{\institution{Cispa Helmholtz Center for Information Security}
  \city{Saarbr\"ucken}
  \country{Germany}
}
\affiliation{
  \institution{University of Trento}
  \city{Trento}
  \country{Italy}
}

\author{Marco Patrignani}
\affiliation{\institution{University of Trento}
  \city{Trento}
  \country{Italy}
}

\author{Marco Guarnieri}
\affiliation{\institution{IMDEA Software Institute}
  \city{Madrid}
  \country{Spain}}

\author{Michael Backes}
\affiliation{\institution{Cispa Helmholtz Center for Information Security}
  \city{Saarbr\"ucken}
  \country{Germany}
}

\maketitle

\section{Introduction}

Spectre~\cite{spectre} and other speculative execution attacks exploit the fact that modern CPUs speculate over the outcome of different instructions---branches~\cite{spectre}, indirect jumps~\cite{spectre}, stores and loads~\cite{S_specv4}, and returns~\cite{spectreRsb}---to bypass software-level security checks and leak sensitive information.

To mitigate these attacks, mainstream compilers like \gcc and \clang{} implement countermeasures in the form of secure compilation passes~\cite{retpoline, transientfail, sok_countermeasures}.
These passes modify a given program to prevent specific classes of speculative leaks.
Unfortunately, the majority of these countermeasures lack formal security guarantees.
Even countermeasures that come with formal security guarantees, however,  are proved secure against  models (called {speculative semantics}) that \emph{only} capture the specific speculative leaks each countermeasure is designed to prevent.
For instance, some Spectre-PHT\footnote{Spectre-PHT~\cite{spectre} refers to a class of speculative execution attacks exploiting speculation over branch instructions. Here, PHT stands for ``Pattern History Table'', one of the microarchitectural mechanism responsible for branch speculation.} countermeasures have recently been proved secure~\cite{S_sec_comp} against a speculative semantics that \emph{only} models speculation over branch instructions.\looseness=-1

Modern CPUs, however,  employ a variety of speculation mechanisms which need to be accounted for when reasoning about speculative leaks.
This limits the practical applicability of existing security proofs that focus on restricted classes of speculative leaks.
For instance, the security proofs from~\citet{S_sec_comp} ignore some speculation mechanisms (e.g., speculation over memory disambiguation or indirect jumps) implemented in all mainstream CPUs, which might compromise the proved guarantees.
Unfortunately, extending security proofs to support new speculation mechanisms is far from trivial since programs that are seemingly secure when considering each speculation mechanism in isolation might still leak due to their interactions~\cite{spec_comb}. 
This has direct impact on existing security proofs:
as we show in \Cref{sec:countermeasures}, the security guarantees of \emph{speculative load hardening}, a countermeasure implemented in the \textsc{Clang} compiler, (and the corresponding proof~\cite{S_sec_comp}) break when extending the underlying speculative semantics to support speculation over indirect jumps.

Thus, establishing the security guarantees of any countermeasure ideally requires proving the security of that countermeasure against attacker models capturing the effects of \emph{all} speculation mechanisms implemented in modern CPUs.
This approach, however, is impractical:
 (1) it requires developing new secure compilation proofs against \emph{stronger} models (i.e., models accounting for additional speculation mechanisms)  for those countermeasures that have already been proved secure against weaker models,
 and
 (2) it requires additional secure compilation proofs whenever a new speculation mechanism is discovered from reverse engineering of existing CPUs.

In this paper, we address this problem by developing a formal framework that allows us to precisely characterize when the security guarantees provided by Spectre countermeasures can be \emph{lifted} from weaker models  (ignoring some speculation mechanism) to stronger ones (accounting for the omitted mechanisms). 
This lifting allows us to account for further speculative mechanisms \emph{without} requiring new secure compilation proofs.
Using our lifting framework, we performed a comprehensive security analysis of the Spectre countermeasures implemented in mainstream compilers, which we proved secure against a speculative semantics accounting for all known speculation mechanisms for which formal models exist.
Concretely, we make the following contributions:
\begin{itemize} 

\item We formalise two novel speculative semantics capturing speculation over indirect jumps~\cite{spectre} (denoted as $\semj$) and straight-line speculation~\cite{sls-whitepaper} (denoted  as $\semsls$).
We present these novel semantics alongside the formalisation of the language model we use (\Cref{sec:bg-formal}).

\item We develop a new framework for reasoning about the security of compiler-level countermeasures against  leaks induced by multiple speculation mechanisms  (\Cref{sec:security-notions}). This framework integrates the core ideas from the composition framework from~\citet{spec_comb} and from the secure compilation framework from~\citet{S_sec_comp} to allow reasoning about secure compilers against multiple speculative semantics.
We equip our framework with a precise notion of \emph{leakage ordering} (inspired by the notion of hardware-software contracts~\cite{contracts}) that precisely relates the information exposed by different speculative semantics, where semantics supporting more speculation mechanisms are stronger w.r.t. our ordering as they leak more.
The integration of all these concepts required developing new insights tailored for secure compilation, e.g., novel well-formedness conditions for composition, preservation of compiler security from composed semantics to components and from stronger semantics to weaker ones.

\item We precisely characterize under which conditions the security guarantees provided by a compiler can be \emph{lifted} from a base speculative semantics $\contract{}{x}$ to a stronger semantics $\contract{}{x+y}$, i.e., one that extends the base semantics  $\contract{}{x}$ to model the effects of additional speculation mechanisms captured by the semantics $\contract{}{y}$ (\Cref{sec:ext-framework}).
Our lifting theorem (\Cref{cor:lift-comp-pres-paper}) states that the  guarantees provided by a secure compiler for the base semantics $\contract{}{x}$ can be lifted to the extended semantics $\contract{}{x+y}$ whenever three core properties are satisfied: Security in Origin (i.e., the compiler is secure w.r.t. the base semantics  $\contract{}{x}$), Independence in Extension (i.e., the compiler does not introduce further leaks under the extension semantics $\contract{}{y}$), and Safe Nesting (i.e., there are no new leaks due to speculations arising only from the combination of semantics  $\contract{}{x}$ and $\contract{}{y}$).
Finally, to simplify proving Independence and Safe Nesting, we propose two sufficient conditions, Syntactic Independence and Trapped Speculation, that provide the same guarantees with simpler proofs but under stricter constraints.

\item Using our framework, we perform a comprehensive security analysis of Spectre countermeasures in mainstream compilers (\Cref{sec:countermeasures}).
Our analysis spans \totcountermeasures countermeasures against 5 different Spectre attacks: Spectre-PHT~\cite{spectre}, Spectre-BTB~\cite{spectre}, Spectre-STL~\cite{S_specv4}, Spectre-RSB~\cite{ret2spec, spectreRsb}, and Spectre-SLS~\cite{sls-whitepaper}.
To the best of our knowledge, this is the most extensive formal analysis of compiler countermeasures against speculative attacks to date (prior studies~\cite{S_sec_comp} are limited to Spectre-PHT countermeasures).
As part of this analysis, we precisely characterize the security guarantees of all countermeasures w.r.t. a combined speculative semantics accounting for five different speculation mechanisms (all those for which formal models exist).
We remark that our lifting theorem (\Cref{cor:lift-comp-pres-paper}) is instrumental in making our security analysis feasible since we use it to lift each countermeasure's security guarantees to all possible combined semantics \emph{without} requiring new secure compilation proofs, which significantly reduces the amount of secure compilation proofs needed.

Our security analysis highlights that:
\begin{itemize}
    \item Countermeasures that block or trap speculation, i.e., fence-based~\cite{v4_fence,Intel-compiler} and retpoline-based~\cite{retpoline} approaches, are the most secure---their security guarantees can be lifted to the stronger speculative semantics we can model.
    Furthermore, lifting their guarantees is ``easy'': it can be done with minimal proof effort since  Syntactic Independence and Trapped Speculation can be used to simplify lifting proofs.

    \item Countermeasures that mask insecure values during speculation i.e., Speculative Load Hardening~\cite{slh} (SLH) and its variants~\cite{S_sec_comp,uslh}, require careful handling of the interactions between masking code and different speculation mechanisms.
    For instance, we show that the guarantees of SLH cannot be lifted to models supporting speculation over indirect jumps, because the speculation flag (which tracks whether mispredictions are happening) is not tracked correctly between jumps.
    Even when lifting is possible, lifting SLH guarantees is more difficult than for the other countermeasures we analyzed since it requires full proofs of Independence and Safe Nesting (the simpler sufficient conditions are not applicable).\looseness=-1

\end{itemize}

\item We mechanise the core results regarding our lifting framework (not those associated with our security analysis) in the Coq proof assistant and indicate those theorems with \coqed.

\end{itemize} \smallskip
The paper concludes with a discussion of the presented result (\Cref{sec:limitations}), related work (\Cref{sec:related-work}), and conclusions (\Cref{sec:conclusion}).
For simplicity, we only discuss key aspects of our formal models here.
Full details, proofs and the Coq development can be found in the companion report.

\section{Language Formalisation: \texorpdfstring{\muasm{}}{uASM}, Speculative Semantics, and Their Combinations}\label{sec:bg-formal}

This section presents \muasm{}, an assembly-like language~\cite{spectector} that we extend with a notion of components in order to identify the unit of compilation~\cite{S_sec_comp}.
We use \muasm{} as a basis for our formal framework and secure compilers.

First, we introduce the attacker model we consider (\Cref{ssec:atk-model}).
Next, we present \muasm{}'s syntax and non-speculative semantics first (\Cref{ssec:muasm}), followed by the different speculative semantics  (\Cref{ssec:spec-semantics}).
We then show how to combine different speculative semantics (\Cref{ssec:bg-combinations}) to account for multiple sources of speculative leaks.
Finally, we describe how different semantics can be compared in terms of leaked information (\Cref{ssec:leakage-ordering}).

\subsection{Attacker Model}\label{ssec:atk-model}

We adopt a commonly-used attacker model \cite{spectector, contracts, ST_binsec, ST_constantTime_Spec,guanciale2020inspectre, ST_jasmin2, S_sec_comp, ST_blade, spec_comb}: a passive attacker observing the execution of a program through events $\tau$.
These events model timing leaks through cache and control flow while abstracting away low-level  microarchitectural details. 
{
\begin{gather*}
\begin{aligned}
    \mi{Events}~\com{\lambda} \bnfdef&\ \com{\epsilon} \mid \com{\alpha?} \mid \com{\alpha!} \mid \com{\delta} \mid \com{\ts}
    &\qquad
    \mi{Actions}~\com{\alpha} \bnfdef&\ (\clgen{f}) \mid (\rtgen{})
\end{aligned}
\\
\begin{aligned}
    \mi{\mu arch.\ Acts.}~\com{\delta} \bnfdef
        &\ 
        \com{\storel{n}} \mid \com{\loadl{n}} \mid \com{\pcl{l}} \mid \com{\startl{x}} \mid \com{\rollbl{x}} 
\end{aligned}
\end{gather*}
}

Events $\lambda$ are either the empty event $\com{\epsilon}$, an action $\com{\alpha?}$  or $\com{\alpha!}$ where $\com{?}$ denotes events \emph{from} the component \emph{to} the attacker and $\com{!}$ denotes events in the other direction, a microarchitectural action $\com{\delta}$, or the designated event  $\com{\ts}$ denoting termination.

Action $\com{\clh{f}}$ represents a call to a function $f$ in the component, while $\com{\cbh{f}}$ represents a call(back) to the attacker. 
In contrast, action $\com{\rth}$ represents a return to the attacker and $\com{\rbh}$ a return(back) to the component.

The microarchitectural actions $\com{\storel{n}}$ and $ \com{\loadl{n}}$ track addresses of store and loads, thereby capturing leaks through the data cache.
Moreover, $\com{\pcl{l}}$ tracks the program counter during execution respectively, thereby capturing leaks through the instruction cache. 
Finally, the $\com{\startl{x}}$ and $\com{\rollbl{x}}$ microarchitectural actions respectively denote the start and rollback of a speculative transaction~\cite{spectector}, i.e., a set of speculatively executed instructions.
Since we consider multiple speculative semantics (and their combinations), $\com{\startl{x}}$ and $\com{\rollbl{x}}$ actions are labelled with an identifier $x$ denoting from which semantics the transaction originated from.

Traces $\tauStack$ are sequences of events $\lambda$.
A trace  $\tauStack$ is \emph{terminating} if it ends in $\com{\ts}$.
Given a trace $\tauStack$, its \textit{non-speculative projection} $\nspecProject{\tauStack}$ \cite{spectector} consists of all observations associated with non-speculatively executed instructions and it is computed by removing all sub-sequences enclosed between  $\com{\startl{x}}$ and $\com{\rollbl{x}}$ for any $x$.
To reason about combined semantics~\cite{spec_comb}, we also need projections $\specProjectComp{\tauStack}{x}$ that removes from the trace the contributions of a specific semantics with identifier $x$.
Specifically, $\specProjectComp{\tauStack}{x}$ denotes the trace obtained by removing from $\tauStack$ all sub-sequences enclosed $\com{\startl{x}}$ and $\com{\rollbl{x}}$ for a given $x$.\looseness=-1

\subsection{Syntax and Semantics of \texorpdfstring{\muasm{}}{uASM}}\label{ssec:muasm}

\muasm{}'s syntax is presented below; we indicate the sequence of elements $e_1, \cdots, e_2$ as $\OB{e}$ and $\OB{e} \cdot e$ denotes a stack with top element $e$ and rest of stack $\OB{e}$.

\begin{gather*}
\begin{aligned}
	\mi{Programs}~ \com{W}, \com{P} \bnfdef&\ \com{M , \OB{F} , \OB{I}}
	&
	\mi{Attackers}~\ctxc{} \bnfdef&\ \com{M, \OB{F}\hole{\cdot}} 
	&
	\mi{Imports}~\com{I} \bnfdef&\ \com{f}
	\\
	\mi{Functions}~\com{F} \bnfdef&\ \emptyset \mid \com{F}; f \mapsto (l_{start}, c)
	&
	\mi{Code}~\com{c} \bnfdef&\  n :i \mid c_1;c_2
        &
	\mi{Codebase}~\com{C} \bnfdef&\ \com{\OB{F} , \OB{I}}
        \\
        \mi{Values}~ v \in&\ \Val = \mathbb{N} \cup \{\bot \}
 &
	\mi{Expressions}~  e \bnfdef&\  v \mid x \mid \ominus e \mid e_1 \otimes e_2
        \end{aligned}
	\\
	\begin{aligned}
	\mi{Instructions}~ i \bnfdef&\ \pskip \mid \passign{x}{e} \mid \pload{x}{e} \mid \pstore{x}{e} \mid                                                \pjmp{e} \mid  \pjz{x}{l} \mid \pcondassign{x}{e'}{e} \mid  \\
                                                  &\pbarrier  \mid \pcall{f} \mid \pret \mid  \ploadprv{x}{e} \mid \pstoreprv{x}{e}
	\end{aligned}
\end{gather*}

\muasm{} has a notion of components, i.e., partial programs $P$, and of attackers $A$.
Components $P$ define their memory $M$ (defined later), a list of functions $\OB{F}$, and a list of imports $\OB{I}$, which are all the functions the component expects to be defined by an attacker. 
An attacker $A$ only defines its memory and its functions. 
We indicate a program code base, i.e., its functions and imports, as $C$.
A component $P$ and attacker $A$ can be linked to obtain a whole program $W \equiv \ctxc{}\hole{P}$.
 
Functions consist of a start label $l_\mathit{start}$ indicating the position of the code $c$ of that function. Each function ends with a $\pret$ instruction.
Code $c$ is a sequence of mappings from natural-number labels to instructions $i$, where instructions $i$ include skipping, (conditional) register assignments, (private) loads, (private) stores, indirect jumps, conditional branches, conditional assignments, speculation barriers, calls, and returns.\footnote{Technically, instruction labels are drawn from a set $\labelset$ of abstract labels  mapped to natural numbers before execution.}
Instructions can refer to expressions $e$, constructed by combining registers $x$ (described below) and values $v$ with unary and binary operators.
Values come from the set $\Val$ and can be natural numbers, labels, or $\bot$.

\paragraph{Non-Speculative States}
\muasm{}'s semantics is defined in terms of program states, defined below:
{
\begin{gather*}
\begin{aligned}
    \text{Configurations } \sigma \bnfdef
        &\ \tup{p, m, r} 
    &\qquad 
    \mi{Frames}~\com{B} \bnfdef
        &\ \emptyset \mid \com{\OB{n}}; B
    &\qquad
    \text{Prog. States } \Omega \bnfdef
        &\ C;\OB{B}; \sigma
\end{aligned}
    \\
\begin{aligned}
    \mi{Register File}~\com{r} \bnfdef
        &\ \emptyset \mid \com{r}; x \mapsto v
            &
            \mi{Registers}~ x \in
                &\ \Reg
    &
    \mi{Memory}~\com{M} \bnfdef
        &\ \emptyset \mid \com{M}; n \mapsto v
            \text{ where }~ n \in
                \mathbb{Z}
\end{aligned}
\end{gather*}
}
Program states $ C;\OB{B}; \sigma$ consist of a codebase $C$, a return frame $B$, and a configuration $\sigma$.
$C$ is used to look up functions, while $B$ stores the return addresses of called functions.
$B$ consists of a stack of stacks of natural numbers $n$. 
A new empty stack $\OB{n}$ is created whenever a context switch between component and attacker happens. 
In this way, neither component nor attacker can manipulate the return stack of each other.
Configurations $\sigma$ consist of the program $p$, the memory $m$, and the register file $a$.
The code of the program $p$ is defined as the union of the code of all the functions and is a partial function mapping labels $l$ to instructions $i$.

Memories $M$ map memory addresses $n \in \mathbb{Z}$ to values $v$.
The memory is split into a public (when domain $n \geq 0$) and a private (if $n < 0$) part.
Attackers $A$ can only define and access the private memory while programs $P$ define the private memory and can access both private and public memory.\looseness=-1 

Register files $r$ map registers $x$ in $\Reg$ to their values $v$.
Note that the set $\Reg$ includes the designated registers $\pc$ and $\spR$, modelling the program counter and the stack pointer respectively.

\paragraph{Non-Speculative Operational Semantics}
\muasm is equipped with a big-step operational semantics~\cite{spectector} for expressions and a small-step operational semantics for instructions generating events.
The former has judgement $\exprEval{a}{e}{v}$ and means: ``expression $e$ reduces to value $v$ under register file $a$.''
The latter has judgement $\Omega \nsarrow{\tau} \Omega'$ meaning: ``state $\Omega$ reduces in one step to $\Omega'$ emitting label $\tau$.''

Below is a selection of these rules.
\Cref{tr:load-paper} executes a load instruction from address $n$, emitting the related microarchitectural action $\loadObs{n}$ exposing the accessed address $n$.
\Cref{tr:call-paper} executes a call to function $f'$ (whose address $n'$ is looked up in the functions table $\mathcal{F}$), it updates the registers file $\av' = \av[\pc \mapsto n']$, and it pushes a new frame with the return address $\av(\pc) + 1$ on the frames stack. 
Auxiliary function $\ffun{n} = f $ is used to look up a function $f$ starting from an address $n$ in memory, while $\src{C}.\mtt{intfs}\vdash\src{f',f} : \src{in}$ is an (omitted) judgement to determine the decorator for the call (and return) action: $?$ or $!$ depending on whether the call or return comes from the attacker to the component or vice-versa.
\Cref{tr:ret} resumes the computation with the $\pc$ set to the address ($l$) it pops from the frames stack.

\begin{center}
  \mytoprule{\sigma \nsarrow{\tau} \sigma}

  \typerule{Load}
  {
  \select{p}{\av(\pc)} = \pload{x}{e} & x \neq \pc & \exprEval{\av}{e}{n}
  }
  {
  \src{C;\OB{\Bva}; \tup{p,\mv,\av}}  \nsarrow{\loadObs{n}} \src{C;\OB{\Bva}; \tup{\mv, \av[\pc \mapsto \av(\pc)+1, x \mapsto \mv(n)]}}
  }{load-paper}

  \typerule{Call}
  {
  \select{p}{\av(\pc)} = \pcall{f'} & \mathcal{F}(f') = n' 
  & 
  f' \in C.\mtt{funs}
  &
  \av' = \av[\pc \mapsto n'] 
  & 
  \av(\pc) = n 
  \\ 
  \ffun{n} = f  
  &
  \src{C}.\mtt{intfs}\vdash\src{f',f} : \src{in} }
  {
  \src{C;\OB{\Bva}; \tup{p,\mv,\av}}  \nsarrow{\pcall{f'}?} \src{C;  (\OB{\Bva \cdot \av(\pc) + 1} ;\emptyset) ;\tup{p, \mv, \av'}}
  }{call-paper}
  \typerule{Ret}
  {
  \select{p}{\av(\pc)} = \pret 
  & 
  \av(\pc) = n 
  & 
  \av' = \av[\pc \mapsto l]  
  \\
  \ffun{n} = f & \ffun{l} = f' 
  & 
  \src{C}.\mtt{intfs}\vdash\src{f',f} : \src{in}
  }
  {
  \src{C; (\OB{\Bva} \cdot l; \OB{n}) \tup{p,\mv,\av}}  \nsarrow{\pret?} \src{C;\OB{\Bva}; \tup{p,\mv,\av'}} 
  }{ret}
\end{center}

Lastly, the semantics must capture the execution of whole programs.
Whole programs are the result of linking attackers and components and must have no undefined function imports.
Whole programs have a (straightforward and therefore omitted) big-step semantics $\nsbigarrow{}$ that concatenates single steps into multiple ones and single labels into traces.
The judgement $\com{\Omega \nsbigarrow{\tauStack} \Omega'}$ is read: "state $\com{\Omega}$ emits trace $\tauStack$ and becomes $\Omega'$."
The behaviour of a whole program $W$, written $\behNs{W}$, is the set of terminating traces it produces.

\paragraph{Source Programs}
The semantics described so far has no speculation.
We use it as semantics for the source programs of all our compilers.
We indicate the language of such source programs as \src{L}.

\subsection{Speculative Semantics}\label{ssec:spec-semantics}
The target languages of the compilers we consider all have different speculative semantics modeling the effects of speculatively executed instructions.
The speculative semantics we define are summarised in \Cref{tab:semantics}, where we list the instruction triggering speculation for each semantics.
\begin{wraptable}{r}{9cm}
    \centering
    \begin{tabular}{lcl}
    \toprule
    Semantics & Spec. Source ($\specInstr$)  & Effect \\
    \midrule
    $\semb$ \cite{spectector}    &  $\jzC{}$& branch misprediction\\
    $\sems$ \cite{spec_comb}    &  $\storeC{}$ & store bypass  \\
    $\semj$  (new)   &  (indirect) $\jmpC{}$ & different jump target \\
    $\semr$ \cite{spec_comb}    &  $\callC$, $\retC$ & return misprediction \\ $\semsls$  (new) &  $\retC$ & return bypass\\
     \bottomrule
    \end{tabular}
    \caption{Speculative semantics with the instructions they speculate on and their effects on execution.}
    \label{tab:semantics}
\end{wraptable}
We consider five different speculative semantics capturing branch speculation ($\semb$), store-bypass speculation ($\sems$), indirect jump speculation ($\semj$), speculation using a return-stack buffer ($\semr$), and straight-line speculation over return instructions ($\semsls$).
While $\semb$, $\sems$, and $\semr$ come from prior work~\cite{spectector,spec_comb}, the $\semj$ and $\semsls$ semantics are novel.
All these semantics follow the always-mispredict approach~\cite{spectector}.
At every instruction related to speculation, the semantics first speculatively executes the \textit{wrong} path for a bounded number of steps (called the speculation window) and then continues with the correct one.
The effects of the speculatively executed instructions are visible on the trace as actions enclosed between $\ensuremath{\labelfont{start}}$ and $\ensuremath{\labelfont{rlb}}$ events.

All the speculative semantics we consider follow a similar structure, which we recap now.
Formally, the speculative state $\Sigmax$ is a stack of speculative instances $\Phix$ where reductions happen only on top of the stack.
Each instance $\Phix$ contains the program state $\Omega$ and the remaining speculation window $n$ describing the number of instructions that can still be executed speculatively (or $\bot$ when no speculation is happening). 
Depending on the semantics, the instance $\Phix$ may track additional data, e.g., the return-stack buffer in $\semr$. 
Below, we leave this additional information abstract and we indicate it by $\cdots$;  we refer to the companion report for full definitions of speculative states.
Throughout the paper, we fix the maximal speculation window, i.e., the maximum number of speculative instructions, to a global constant $\omega$.
{
\begin{align*}
 \mi{Speculative\ States}~ \Sigmax \bnfdef&\ \phiStackx 
 &\qquad
 \mi{Speculative\ Instance}~ \Phix \bnfdef&\ \tup{\Omega, n,...}
\end{align*}
}
Each semantics has an instruction that starts speculation (the central column of \Cref{tab:semantics}):  whenever those instructions are executed, the semantics first pushes the mispredicted state and then the correct state onto the state $\Sigmax$.

The (small-step) judgement for all speculative semantics is of the form $\Sigmax \specarrowx{\tau} \Sigmax'$ and it describes how the speculative state is updated when executing instructions.
For all speculative semantics, the behaviour $\behx{W}$ of a whole program $W$ is the trace $\tauStack$ generated by the big-step judgment $\bigspecarrowx{}$, which executes the program $W$ starting from its initial state until termination and collect all produced actions.

Below, we overview the small-step judgments  $\Sigmax \specarrowx{\tau} \Sigmax'$ for the speculative semantics we study.
We start from the two new speculative semantics $\semsls$ (\Cref{sec:sem-sls}) and $\semj$ (\Cref{sec:sem-j}) and later present the semantics $\semb$ (\Cref{sec:sem-b}), $\sems$ (\Cref{sec:sem-s}), and $\semr$ (\Cref{sec:sem-r}) from prior work~\cite{spectector,spec_comb}.
We first describe all rules for $\semsls$, which we use to explain the structure of all our speculative semantics.
In contrast, for $\semj$, $\semb$, $\sems$, and $\semr$, we only report the most significant rule for each semantics, i.e., the rule that triggers the specific form of speculation listed in \Cref{tab:semantics}.

\subsubsection{Modeling Straight-Line Speculation}\label{sec:sem-sls}
Straight-line speculation (SLS)~\cite{sls-whitepaper, sls-whitepaper2} is a speculation mechanism implemented in some CPUs where return instructions are speculatively bypassed and the execution continues speculatively (after ignoring the return) for a fixed number of steps.
The $\semsls$ semantics models the effect of straight-line speculation using the small-step rules below.
\begin{center}
     \mytoprule{\trgSLS{\PhiSLS} \specarrowSLS{\tau} \trgSLS{\phiStackSLS'}}

    \typerule{$\SLSvr$:AM-Ret-Spec}
    {
        \trgSLS{\instr{\pret}} 
        &  
        \trgSLS{\Omega} \nsarrow{\tau} \trgSLS{\Omega'} 
        & 
        \trgSLS{\Omega} = \trgSLS{\OB{F}; \OB{I}; \OB{B}; \sigma}
        &
        \ffun{\trgSLS{\Omega(\pc)}} = \trgSLS{f} 
        & 
        \trgSLS{f} \notin \trgSLS{\OB{I}} 
        \\
        \trgSLS{\Omega''} = \trgSLS{\OB{F}; \OB{I}; \OB{B}; \sigma''}
        &
        \trgSLS{\sigma''} = \trgSLS{\sigma[\pc \mapsto \Omega(\pc) + 1]} 
        & 
        \trgSLS{j} = \trgSLS{min(\omega, n)} 
    }
    {\trgSLS{{\Omega, n + 1}} \specarrowSLS{\tau} \trgSLS{\tup{\Omega', n} \cdot \tup{\Omega'', j}}
    }{sls-spec-paper}
    \typerule{$\SLSvr$:AM-Ret-Spec-att}
    {\trgSLS{\instr{\pret}} 
    &  
    \trgSLS{\Omega} \nsarrow{\tau} \trgSLS{\Omega'} 
    \\
    \ffun{\trgSLS{\Omega(\pc)}} = \trgSLS{f} 
    &
    \trgSLS{f} \in \trgSLS{\OB{I}} 
    }
    {\trgSLS{\tup{\Omega, n + 1}} \specarrowSLS{\tau} \trgSLS{\tup{\Omega', n}}
    }{sls-spec-att-paper}
    \typerule {$\SLSvr$:AM-NoSpec-action}
    {
        \trgSLS{\instrneq{\pret, \pbarrier, Z}} 
        \\
        \trgSLS{\Omega} \nsarrow{\tau} \trgSLS{\Omega'}
    }
    {\trgSLS{\tup{\Omega, n + 1}} \specarrowSLS{\tau} \trgSLS{\tup{\Omega', n}}
     }{sls-nospec-act-paper}
    \typerule{$\SLSvr$:AM-Rollback}
    { \trgSLS{n = 0}\ \text{or}\ \finType{\trgSLS{\Omega}} }
    {
    \trgSLS{\tup{\Omega, n}} \specarrowSLS{\rollbackObsSLS} \trgSLS{\varepsilon}
    }{sls-rollback-paper}
     \typerule {$\SLSvr$:AM-NoSpec-epsilon} 
    {
        \trgSLS{\instrneq{\pret, \pbarrier, Z}} 
        \\
        \trgSLS{\Omega} \nsarrow{\epsilon} \trgSLS{\Omega'}
    }
    {\trgSLS{\tup{\Omega, n + 1}} \specarrowSLS{\epsilon} \trgSLS{\tup{\Omega', n}}
     }{sls-nospec-eps-paper}
     \typerule {$\SLSvr$:AM-barr}
    {
        \trgSLS{\instr{\pbarrier}} 
        \\
        \trgSLS{\Omega'} = \trgSLS{\Omega[\pc \mapsto \pc + 1]}
    }
    {\trgSLS{\tup{\Omega, \bot}} \specarrowSLS{\epsilon} \trgSLS{\tup{\Omega', \bot}}
     }{sls-barr-paper}
     \typerule {$\SLSvr$:AM-barr-spec}
    {
        \trgSLS{\instr{\pbarrier}} 
        \\
        \trgSLS{\Omega'} = \trgSLS{\Omega[\pc \mapsto \pc + 1]}
    }
    {\trgSLS{\tup{\Omega, n + 1}} \specarrowSLS{\epsilon} \trgSLS{\tup{\Omega', 0}}
     }{sls-barr-spec-paper}
\end{center}

Speculation is started by  $\retC$ instructions.
Whenever the semantics executes $\retC$, the return is speculatively bypassed  and execution speculatively continues after the return (captured by \Cref{tr:sls-spec-paper}).
The rule pushes on the stack of speculative states a new speculative instance $\trgSLS{\tup{\Omega'', j}}$, from which  execution will continue.
Note that speculation only starts when we are inside the component. 
For this, the rule checks that the function $f$ being executed is not in the imports, i.e., it is not attacker-defined: this ensures that labels are only produced when non-attacker code is executed.
Otherwise, execution continues normally (\Cref{tr:sls-spec-att-paper}).

Executing instructions that do not trigger speculation updates the program state according to the non-speculative semantics, reduces the speculation window by 1, and produces actions when needed (\Cref{tr:sls-nospec-act-paper} and \Cref{tr:sls-nospec-eps-paper}).
These rules are only triggered when the instruction is not a return, a fence or it is not in a set of instructions \trgSLS{Z}, which is used to track the speculation-related instructions across combinations of the semantics~\cite{spec_comb}.
When the speculation window in the top instance hits 0, the speculative state is rolled back and discarded (\Cref{tr:sls-rollback-paper}) and execution continues from the speculative instance now on top.

Finally, speculation barriers terminate speculation, which we model by setting the current speculation window to $0$ (\Cref{tr:sls-barr-spec-paper}).
In non-speculative executions (i.e., when the speculation window is $\bot$), speculation barriers are handled as $\pskip$ (\Cref{tr:sls-barr-paper}).

\subsubsection{Modeling Jump Speculation}\label{sec:sem-j}
Jump speculation allows jump instructions to speculate the address where they are jumping to~\cite{spectre}.
To model this, \Cref{tr:v2-spec-paper} creates a set of speculative instances, with one speculative instance for each of the possible jump targets in the program, by updating the program counter $\pc$ to that jump target.
\begin{center}
    \mytoprule{\trgJ{\PhiJ} \specarrowJ{\tau} \trgJ{\phiStackJ'}}
    
     \typerule{$\Jvr$:AM-Jmp-Spec}
    {
    \trgJ{\instr{\pjmp{x}} } 
    & 
    \trgJ{x} \in \trgJ{\Reg} 
    & 
    \trgJ{\Omega} \nsarrow{\tau} \trgJ{\Omega'}  
    &  
    \trgJ{\Omega} = \trgJ{\OB{F}; \OB{I}; \OB{B}; \sigma}
    \\
    \ffun{\trgJ{\Omega'(\pc)}} = \trgJ{f'} 
    & 
    \ffun{\trgJ{\sigma(\pc)}} = \trgJ{f} 
    & 
    \trgJ{f} \notin \trgJ{\OB{I}} 
    &  
    \trgJ{{\OB{I}}}\vdash\trgJ{f,f'}:\src{internal} 
    \\
    \trgJ{j} = \trgJ{min(\omega, n)} 
    & 
    \trgJ{\OB{\SigmaJ''}} = \bigcup_{l \in p} \trgJ{\tup{\Omega'', j}} \text{ where } \trgJ{\Omega''} = \trgJ{\OB{F}; \OB{I}; \OB{B}; \sigma[\pc \mapsto l]}
    }{
        \trgJ{\tup{\Omega, n + 1}} \specarrowJ{\tau} \trgJ{\tup{\Omega', n} \cdot \OB{\SigmaJ''}}
    }{v2-spec-paper}
\end{center}

\subsubsection{Modeling branch speculation}\label{sec:sem-b}
CPUs speculate over the outcome of branch instructions~\cite{spectre}, which might result in speculatively executing the wrong branch.
To model this, we rely on the $\semb$ semantics from~\cite{spectector} where \Cref{tr:v1-spec-paper} speculates on branching instructions by pushing on top of the stack of speculative states the state that is opposite of the evaluated condition.

\begin{center}
  \mytoprule{\trgB{\PhiB} \specarrowB{\tau} \trgB{\phiStackB'}}

  \typerule{$\Bvr$:AM-Branch-Spec}
    {
    \trgB{\instr{\pjz{x}{l}}} 
    &  
    \trgB{\Omega} \nsarrow{\tau} \trgB{\Omega'} 
    & 
    \trgB{\Omega} = \trgB{\OB{F}; \OB{I}; \OB{B}; \sigma}
    &
    \ffun{\trgB{\sigma(\pc)}} = \trgB{f} 
    & 
    \trgB{f} \notin \trgB{\OB{I}} 
    \\
    \trgB{\Omega''} = \trgB{\OB{F}; \OB{I}; \OB{B}; \sigma''}
    &
     \trgB{\sigma''} = \trgB{\sigma[\pc \mapsto l']} & \trgB{j = min(\omega, n)} 
    \\
     \text{if } \trgB{\sigma'(\pc)} = \trgB{l} \text{ then } \trgB{l'} = \trgB{\sigma(\pc) + 1} 
     \text{ otherwise } 
     \trgB{l'} = \trgB{l} 
    }
    {
        \trgB{\tup{\Omega, n + 1}} \specarrowB{\tau} \trgB{\tup{\Omega', n} \cdot \tup{\Omega'', j}}
    }{v1-spec-paper}
\end{center}

\subsubsection{Modeling store-bypass speculation}\label{sec:sem-s}
Modern CPUs write stores to main memory asynchronously to reduce delays caused by the memory subsystem.
This may result in speculatively  bypassing store instructions and fetching stale information from the CPU's load-store queue~\cite{S_specv4}.
To model this, we rely on the $\sems$ semantics from~\cite{spec_comb} where 
\Cref{tr:v4-skip-paper} speculates on stores by adding a state at the top of the stack where a store has been skipped.

\begin{center}
 \mytoprule{\trgS{\PhiS} \specarrowS{\tau} \trgS{\phiStackS'}}

    \typerule{$\Svr$:AM-Store-Spec}
    {
    \trgS{\instr{\pstore{x}{e}}} 
    &  
    \trgS{\Omega} \nsarrow{\tau} \trgS{\Omega'} 
    & 
    \trgS{\Omega} = \trgS{\OB{F}; \OB{I}; \OB{B}; \sigma}
    &
    \ffun{\trgS{\sigma(\pc)}} = \trgS{f} 
    & 
    \trgS{f} \notin \trgS{\OB{I}} 
    \\ 
    \trgS{\Omega''} = \trgS{\OB{F}; \OB{I}; \OB{B}; \sigma''}
    &
     \trgS{\sigma''} = \trgS{\sigma[\pc \mapsto \Omega(\pc) + 1]} 
     & 
     \trgS{j} = \trgS{min(\omega, n)} 
    }
    {\trgS{\tup{\Omega, n + 1}} \specarrowS{\tau} \trgS{\tup{\Omega', n} \cdot \tup{\Omega'', j}}
    }{v4-skip-paper}
\end{center}

\subsubsection{Modeling return speculation}\label{sec:sem-r}
CPUs also speculate on the outcome of return instructions~\cite{spectreRsb}. 
For this, they rely on a microarchitectural data structure called the return-stack-buffer (RSB).
The speculative semantics $\semr$, taken from~\cite{spec_comb}, models this kind of speculation by extending the speculative instances with a return stack buffer \trgR{\Rsb}, which is a list of return locations.
\Cref{tr:v5-spec-paper}, which is the one starting speculation, works by speculatively returning to the `wrong' location \trgR{l} at the top of the RSB whenever \trgR{l} differs from the  expected  return address \trgR{B(0)} upon encountering a return instruction.
Note also that the semantics pushes a return address to the RSB whenever a \textbf{call} instruction is executed (not shown in the rules below).

\begin{center}

   \mytoprule{\trgR{\PhiR} \specarrowR{\tau} \trgR{\phiStackR'}}
   
    \typerule{$\Rvr$:AM-Ret-Spec}
    {
    \trgR{\instr{\pret}} 
    & 
    \trgR{\Omega} \nsarrow{\tau} \trgR{\Omega'}  
    &  
    \trgR{\Omega} = \trgR{\OB{F}; \OB{I}; \OB{B}; \sigma}
    &
    \trgR{\Rsb} = \trgR{\Rsb' \cdot l}  
    &  
    \trgR{l} \neq \trgR{B(0)}  
    &  
    \trgR{\Omega''} = \trgR{\OB{F}; \OB{I}; \OB{B'}; \sigma''}
    \\
    \ffun{\trgR{\sigma(\pc)}} = \trgR{f} 
    &
    \ffun{\trgR{\sigma'(\pc)}} = \trgR{f'} 
    & 
    \trgR{f, f'} \notin \trgR{\OB{I}} 
    &
     \trgR{\sigma''} = \trgR{\sigma[\pc \mapsto l]} 
     & \trgR{j} = \trgR{min(\omega, n)}
    }
    {\trgR{\tup{\Omega, \Rsb, n + 1}} \specarrowR{\tau} \trgR{\tup{\Omega', \Rsb', n} \cdot \tup{\Omega'', \Rsb', j}}
    }{v5-spec-paper}
\end{center}

\subsection{Combining Speculative Semantics}\label{ssec:bg-combinations}
To reason about leaks resulting from multiple speculation sources, we rely on the combination framework from~\citet{spec_comb}. 
This framework allows combining multiple speculative semantics (for different speculation sources) into a combined semantics that allows reasoning about both kinds of speculation.
For instance, combining $\contract{}{\Bv}$ and $\contract{}{\Sv}$ yields the composed semantics $\contract{}{\Bv + \Sv}$ that speculates on both $\jzC$ and $\storeC$ instructions. 
We remark that a composed semantics is ``stronger than its parts'', that is, it may explore speculative actions that only arise from the interaction of its component semantics.
As shown in~\cite{spec_comb}, there are to programs that contain leaks only under the composed semantics, even though the programs are leak-free when considering the base semantics in isolation.\looseness=-1

The core component of this combination framework is the notion of \emph{well-formed composition}, which needs to be tailored to the specific properties that combinations need to preserve.
Here we report two properties that well-formed combinations need to satisfy; we will introduce a third, novel, well-formedness condition in \Cref{sec:security-notions} to deal with preservation of the security properties.

\begin{definition}[Well-Formed Composition -- Part 1]

The well-formed composition of two semantics $\contract{}{x}$ and $\contract{}{y}$, denoted with $\wfc{\contract{}{x+y}}$, is defined according to these points (as well as the point later defined in \Cref{def:wfc2}.
\begin{asparaitem}

    \item \textbf{Confluence~\cite{spec_comb}:} If $\Sigmaxy \specarrowxy{\tau'} \Sigmaxy'$ and $\Sigmaxy \specarrowxy{\tau} \Sigmaxy''$, then  $\Sigmaxy' = \Sigmaxy''$ and $\tau = \tau'$.
    \item \textbf{Projection Preservation~\cite{spec_comb}:}
    $\behx{P} = \specProjectComp{\behxy{(P)}}{y}$ and $\behy{P} = \specProjectComp{\behxy{(P)}}{x}$.
\end{asparaitem}
\end{definition}

Confluence ensures the determinism of the combined semantics (where $\Sigmaxy$ is the operational state for the combined semantics), whereas Projection Preservation ensures that the speculative behaviour of the source semantics can be recovered from the behavior of the composed semantics.

\subsection{Leakage Ordering}\label{ssec:leakage-ordering}
Each speculative semantics in \Cref{tab:semantics}, as well as their compositions, capture different ``attacker models'', where the attacker can observe the effects (visible on the traces) of the speculative instructions modelled by the semantics.
To reason about the strength of these attacker models, we follow~\cite{contracts} and introduce a partial order in terms of leakage between the different semantics.
In particular, we say that semantics $\contract{}{1}$ is weaker than another semantics $\contract{}{2}$, written $\contract{}{1} \sqsubseteq \contract{}{2}$ iff $\contract{}{2}$ leaks more than $\contract{}{1}$, i.e., if any two initial configurations that result in different traces for $\contract{}{1}$ also result in different traces for $\contract{}{2}$.\footnote{For readers familiar with \cite{contracts}, we flipped the relation $\sqsubseteq$ from the original paper.}

\newcommand{\lineCol}{gray}
\begin{wrapfigure}{r}{0.7\textwidth}
    \centering
    \begin{tikzpicture}[line width=0.1 mm, arrows={[black]}]
          \node (BJSR) at (-1, 4) {$\contractSpec{}{\Bv + \Jv + \Sv + \Rv}$};
          \node (BJSSLS) at (1, 4) {$\contractSpec{}{\Bv + \Jv + \Sv + \SLSv}$};
          \node (BJS) at (-3.6, 2) {$\contractSpec{}{\Bv + \Jv + \Sv}$};
          \node (BJR) at (-2.4, 2) {$\contractSpec{}{\Bv + \Jv + \Rv}$};
          \node (BSR) at (-1.2, 2) {$\contractSpec{}{\Bv + \Sv + \Rv}$};
          \node (JSR) at (0, 2) {$\contractSpec{}{\Jv + \Sv + \Rv}$};
          \node (BJSLS) at (1.4, 2) {$\contractSpec{}{\Bv + \Jv + \SLSv}$};
          \node (BSSLS) at (2.8, 2) {$\contractSpec{}{\Bv + \Sv + \SLSv}$};
          \node (JSSLS) at (4, 2) {$\contractSpec{}{\Jv + \Sv + \SLSv}$};
          \node (BJ) at (-4,0) {$\contractSpec{}{\Bv + \Jv}$};
          \node (BS) at (-3,0) {$\contractSpec{}{\Bv + \Sv}$};
          \node (BR) at (-2,0) {$\contractSpec{}{\Bv + \Rv}$};
          \node (JS) at (-1,0) {$\contractSpec{}{\Jv + \Sv}$};
          \node (JR) at (0,0) {$\contractSpec{}{\Jv + \Rv}$};
          \node (SR) at (1,0) {$\contractSpec{}{\Sv + \Rv}$};
          \node (BSLS) at (2,0) {$\contractSpec{}{\Bv + \SLSv}$};
          \node (JSLS) at (3,0) {$\contractSpec{}{\Jv + \SLSv}$};
          \node (SSLS) at (4,0) {$\contractSpec{}{\Sv + \SLSv}$};
          \node (NS) at (0, -3) {$\contract{}{NS}$};
          \node (B) at (-2,-2) {$\contractSpec{}{\Bv}$};
          \node (J) at (-1,-2) {$\contractSpec{}{\Jv}$};
          \node (S) at (0,-2) {$\contractSpec{}{\Sv}$};
          \node (R) at (1,-2) {$\contractSpec{}{\Rv}$};
          \node  (SLS) at (2,-2) {$\contractSpec{}{\SLSv}$};
          \draw[\lineCol]  (B.north) -- (BJ.south);
          \draw[\lineCol]  (J.north) -- (JS.south);
          \draw[\lineCol]  (J.north) -- (BJ.south);
          \draw[\lineCol]  (S.north) -- (JS.south);
          \draw[\lineCol]  (S.north) -- (SR.south) -- (R.north);
          \draw[\lineCol]  (B.north) -- (BS.south) -- (S.north) -- (SSLS.south) -- (SLS.north) -- (JSLS.south) -- (J.north);
          \draw[\lineCol]  (B.north) -- (BR.south) -- (R.north) -- (JR.south) -- (J.north);
          \draw[\lineCol]  (B.north) -- (BSLS.south) -- (SLS.north);
          \draw[\lineCol] (BJ.north) -- (BJS.south);
          \draw[\lineCol] (BJ.north) -- (BJR.south);
          \draw[\lineCol]  (BJ.north) -- (BJSLS.south);
          \draw[\lineCol]  (BS.north) -- (BSR.south);
          \draw[\lineCol]  (BS.north) -- (BSSLS.south);
          \draw[\lineCol]  (BS.north) -- (BJS.south);
          \draw[\lineCol]  (BR.north) -- (BSR.south);
          \draw[\lineCol]  (BR.north) -- (BJR.south);
          \draw[\lineCol]  (BSLS.north) -- (BSSLS.south);
          \draw[\lineCol]  (BSLS.north) -- (BJSLS.south);
          \draw[\lineCol]  (JS.north) -- (BJS.south);\textbf{}
          \draw[\lineCol]  (JS.north) -- (JSR.south);
          \draw[\lineCol]  (JS.north) -- (JSSLS.south);
          \draw[\lineCol]  (JR.north) -- (BJR.south);
          \draw[\lineCol]  (JR.north) -- (JSR.south);
          \draw[\lineCol]  (JSLS.north) -- (BJSLS.south);
          \draw[\lineCol]  (JSLS.north) -- (JSSLS.south);
          \draw[\lineCol]  (SR.north) -- (BSR.south);
          \draw[\lineCol]  (SR.north) -- (JSR.south);
          \draw[\lineCol]  (SSLS.north) -- (BSSLS.south);
          \draw[\lineCol]  (SSLS.north) -- (JSSLS.south);
          \draw[\lineCol]  (BSR.north) -- (BJSR.south);
          \draw[\lineCol]  (BSSLS.north) -- (BJSSLS.south);
          \draw[\lineCol]  (BJS.north) -- (BJSR.south);
          \draw[\lineCol]  (BJR.north) -- (BJSR.south);
          \draw[\lineCol]  (BJSLS.north) -- (BJSSLS.south);
          \draw[\lineCol]  (JSR.north) -- (BJSR.south);
          \draw[\lineCol]  (JSSLS.north) -- (BJSSLS.south);
          \draw[\lineCol]  (NS.north) -- (B.south);
          \draw[\lineCol]  (NS.north) -- (J.south);
          \draw[\lineCol]  (NS.north) -- (S.south);
          \draw[\lineCol]  (NS.north) -- (R.south);
          \draw[\lineCol]  (NS.north) -- (SLS.south);
    \end{tikzpicture}
    \vspace{-10pt}
    \caption{Ordering of \muasm{} semantics. A semantics higher in the order is stronger, i.e., exposes more information.}
    \label{fig:instantiate}
\end{wrapfigure}

\Cref{fig:instantiate} depicts all the semantics studied in this paper as well as their combinations ordered according to the amount of leaked information, where  there is an edge from semantics $\contract{}{1}$ to $\contract{}{2}$ whenever $\contract{}{1} \sqsubseteq \contract{}{2}$.
In the figure,  the weakest semantics is $\contract{}{NS}$, since it only exposes information about non-speculative instructions. 
In contrast, $\contract{}{\Bv + \Sv}$ exposes speculation on both $\jzC$ and $\storeC$ instructions and is thus stronger than its components $\contract{}{\Bv}$ and $\contract{}{\Sv}$.
Note that there no single strongest semantics in \Cref{fig:instantiate} due to limitations in the combination framework we use~\cite{spec_comb}, which does not allow for combining  semantics that speculate on the same instruction like  $\contract{}{\Rv}$ and $\contract{}{\SLSv}$.

 \section{Security notions}\label{sec:security-notions}
In this section, we present the security notions used in our framework.
Since our goal is studying the security of compiler countermeasures against different classes of speculative leaks, we extend the secure compilation framework from \citet{S_sec_comp} to work with \muasm{} and with all speculative semantics from \Cref{sec:bg-formal} and their combinations. 

All our security notions definitions rely on a notion of robustness, typical for secure compilation work~\cite{journey_sec_comp}, which we explain first (\Cref{sec:robust}).
Next, we introduce two security notions for programs (\Cref{ssec:security-programs}):  Robust Speculative Non-Interference (\rsni) and Robust Speculative Safety (\rss).
We continue by presenting the secure compilation criteria (\Cref{ssec:sec_compiler_recap}).
We conclude by introducing Relation Preservation, the novel, remaining property for well-formed compositions (\Cref{ssec:security-combinations}), which together with Confluence and Projection Preservation from \Cref{ssec:bg-combinations} precisely characterize the core properties of composed semantics.

\subsection{Robustness}\label{sec:robust}

All our security definitions are \emph{robust}~\cite{journey_sec_comp,autysec,rsmove,tydisa,sec-typ-prot,dg-rs,davidcaps}, i.e., they quantify over \emph{every} possible valid attacker.
Concretely, \muasm{} defines partial programs $P$ that are linked to the attacker-controlled context $A$. 
Thus, the attacker is also code executed together with the partial program $P$.
In this work, we say that a component \emph{satisfies a property robustly} iff it satisfies the property for all possible valid attackers, where an attacker is valid, written $\vdash \ctxc{} : \com{atk}$, if it does not define a private memory and does not contain instructions that read and write to the private memory.

\subsection{Security Notions for Whole Programs}\label{ssec:security-programs}
In this section we extend the notions of \rsni (\Cref{ssec:rsni}) and \rss (\Cref{ssec:rss}) to \muasm{} and all semantics from \Cref{sec:bg-formal}.

\subsubsection{Robust Speculative Non-Interference (RSNI)}\label{ssec:rsni}
RSNI is the application of Speculative Non-Interference to the robust setting.
Speculative Non-Interference (SNI) is a class of security properties \cite{spectector, contracts} that compares the information leaked by instructions executed speculatively and non-speculatively.
Intuitively, a program satisfies SNI iff it does not leak more information under the speculative semantics than under the non-speculative semantics.
Thus, SNI semantically characterizes security against leaks introduced by speculatively executed instructions.

RSNI is parametric in (1) a policy denoting the sensitive information and (2) in the speculative semantics $\contract{}{x}$, which models the speculative behaviour of programs.
The policy describes which parts of the program state are public. 
In our case, only the private part of the memory $M$ is sensitive. 
Thus, two programs $P$ and $P'$ are low-equivalent, written $P \loweq P'$, if they only differ in their private memory.
RSNI compares the leakage between non-speculative traces and speculative traces. 
In a nutshell, a program $P$ satisfies RSNI (\Cref{def:rsni-paper}) for a speculative semantics $\contract{}{x}$ if {for any low-equivalent program $P'$} 
that {generates the same non-speculative trace}, the two programs {generate the same speculative traces as well} robustly.

\begin{definition}[Robust Speculative Non-Interference~\cite{S_sec_comp} (RSNI)]\label{def:rsni-paper}
	\begin{align*}
		\contract{}{x} \vdash \com{P} : \rsni \isdef \
            \forall \ctxc{}, W'. &\text{ if }  \vdash \ctxc{} : \com{atk} 
                \text{ and }  { \com{\ctxc{}\hole{P}} \loweq W' }
                \text{ and }  
                {\nspecProject{ \behx{ \SInit{ \com{\ctxc{}\hole{P}} } } } = \nspecProject{ \behx{ \SInit{W'} }  } } \\
            & \text{ then } {\behx{ \SInit{\com{\ctxc{}\hole{P}} } }  = \behx{ \SInit{W'} } }
	\end{align*}
\end{definition}

We remark that all programs satisfy \rsni for the non-speculative semantics $\contract{}{NS}$ because there is no speculation and, thus, $\nspecProject{\trac{}} = \trac{}$ for all its traces \cite[Theorem 3.4]{S_sec_comp}.

\subsubsection{Robust Speculative Safety}\label{ssec:rss}
SNI is a hyperproperty and requires reasoning about pairs of traces.
To simplify secure compilation proofs, we follow~\cite{S_sec_comp} and over-approximate RSNI using robust speculative safety (RSS), a safety property which only requires reasoning about single traces.

Just like RSNI, RSS is the application of Speculative Safety to the robust setting.
Speculative Safety uses taint tracking, tainting values as "safe" (denoted by $\safeta$) if the value can be speculatively leaked without violating RSNI (e.g., the public memory is safe), or "unsafe" (denoted by $\unta$) otherwise. 
Furthermore, taints are propagated during computation.
We instantiate taint tracking for all the speculative semantics in \Cref{ssec:spec-semantics}; we refer the interested reader to \cite[Section 3.2]{S_sec_comp} for the taint tracking rules since these rules are virtually unchanged.

RSS ensures that programs $P$ robustly generate only safe ($\safeta$) actions in their traces.

\begin{definition}[Robust Speculative Safety \cite{S_sec_comp} (\rss)]\label{def:rdss-contract-paper}
\begin{align*}
	\contract{}{x} \vdash\com{P} : \rss \isdef\
		\forall \ctxc{}, \acac{}, \lambda^\sigma \ldotp &\text{ if } \vdash\ctxc{}:\com{atk} \text{ and } \acac{} \in \behx{ \SInit{ \com{\ctxc{}\hole{P}} } } \text{ and } \lambda^\sigma \in \tau  \text{, then } \com{\sigma}\equiv\com{\safeta}
\end{align*}	
\end{definition}

Again, \rss trivial holds for the non-speculative semantics $\contract{}{NS}$ \cite[Theorem 3.9]{S_sec_comp} because there is no speculation.

\Cref{thm:rss-overapprox-paper}, which we proved for all speculative semantics defined in \Cref{tab:semantics}, precisely connects \rsni and \rss by showing that \rss over-approximates \rsni.
\begin{theorem}[\rss Overapproximates \rsni]\label{thm:rss-overapprox-paper}
For all semantics $\contractSpec{}{x}$ in \Cref{tab:semantics}, if $\contractSpec{}{x} \vdash P : \rss $
then $\contractSpec{}{x} \vdash P : \rsni $.
\end{theorem}

\subsection{Secure Compilation Criteria}\label{ssec:sec_compiler_recap}
We now present \textit{robust speculative safety preservation} (\rssp) and \textit{robust speculative non-interference preservation} (\rsnip), two criteria defined in  \citet{S_sec_comp} for reasoning about compiler guarantees against speculative leaks, which we make parametric in the underlying speculative semantics.

A compiler preserves $\rss$ for a given semantics $\contract{}{x}$ if given a source component that is $\rss$ under the non-speculative semantics, the compiled counterpart is also $\rss$ under  $\contract{}{x}$.
\begin{definition}[Robust speculative safety preservation \cite{S_sec_comp} (\rssp)]\label{def:compiler-contract-security-paper}
    \[
    \contract{}{x} \vdash  \comp{\cdot} : \rssp \isdef \forall \src{P} \in \src{L}. \text{ if }
    \contract{}{NS} \vdash \src{P} : \rss \text{ then } \contract{}{x} \vdash \comp{\src{P}} : \rss
    \]
\end{definition}

Similarly, a compiler preserves $\rsni$ for a given semantics  $\contract{}{x}$ if given a source component that is $\rsni$ under the non-speculative semantics, the compiled counterpart is also $\rsni$ under  $\contract{}{x}$.
\begin{definition}[Robust speculative non-interference preservation \cite{S_sec_comp} (\rsnip)]\label{def:compiler-contract-security-sni-paper}
    \[
    \contract{}{x} \vdash  \comp{\cdot} : \rsnip \isdef \forall \src{P} \in \src{L}. \text{ if }
    \contract{}{NS} \vdash \src{P} : \rsni \text{ then } \contract{}{x} \vdash \comp{\src{P}} : \rsni
    \]
\end{definition}

We conclude  by stating \Cref{cor:weak-contract-paper}, which present two new results.
It states that (1) whenever a compiler preserves the security for a stronger semantics (i.e., one that exposes more information), then it also preserves security for weaker ones, and dually that (2) whenever a compiler does not preserve security for a weaker semantics, then it does not preserve security for stronger ones.

\begin{theorem}[Leakage ordering and \rsnip, \coqed]\label{cor:weak-contract-paper}
The following statements hold for any $\contract{}{1}, \contract{}{2}$:\looseness=-1
\begin{itemize}
    \item If $\contract{}{2} \vdash  \comp{\cdot} : \rsnip$ and $\contract{}{1} \sqsubseteq \contract{}{2}$
    then $\contract{}{1} \vdash \comp{\cdot} : \rsnip$.
    \item  If $\contract{}{1} \nvdash  \comp{\cdot} : \rsnip$ and $\contract{}{1} \sqsubseteq \contract{}{2}$
    then $\contract{}{2} \nvdash \comp{\cdot} : \rsnip$.
\end{itemize}
\end{theorem}

\subsection{Well-Formed Compositions and Compilers}\label{ssec:security-combinations}

We now introduce Relation Preservation, the last property---in addition to Confluence and Projection Preservation (see~\Cref{ssec:bg-combinations})---for well-formed compositions of speculative semantics.
We remark that while Confluence and Projection Preservation come directly from~\cite{spec_comb}, Relation Preservation is new and tailored to ensure, together with the other two properties, that RSS overapproximates RSNI for any well-formed composition. 

\begin{definition}[Well-Formed Composition -- Part 2]\label{def:wfc2}
\hfill
\begin{description}
    \item[Relation Preservation]
    If\  $\Sigmaxy \relsa \Sigmaxyt{'}$ and $\Sigmaxy \bigspecarrowxy{\tauStack} \Sigmaxyt{\dagger}$
    and $\safe{\tauStack}$ then $\Sigmaxyt{'} \bigspecarrowxy{\tauStack} \Sigmaxyt{''}$ and $\Sigmaxyt{\dagger} \relsa \Sigmaxyt{''}$.
\end{description}
\end{definition}

To explain Relation Preservation we need to mention two technical details: the state relation $\relsa{}$ and the judgement $\safe{\tauStack}$.
Judgement $\safe{\tauStack}$ means that all actions on the trace are tainted as safe ($\safeta$).
Intuitively, the relation $\relsa$  relates two states iff their registers and memories locations have the same taint and all elements tainted $\safeta$ have the same values in the two states.
We note that we can derive Relation Preservation in a general manner whenever the source semantics enjoy Relation Preservation as well (like our semantics $\Bv$, $\Jv$, $\Sv$, $\Rv$, and $\SLSv$).

The main result of this section is \Cref{theorem:ss-sni-compositions}, which states that for any well-formed composition, $\rss$ overapproximates $\rsni$.
We remark that we prove \Cref{theorem:ss-sni-compositions} once and for all by exploiting (1) the well-formedness properties and (2) the fact that all compositions in \Cref{fig:instantiate} are well-formed, rather than having to prove the implication for each of the composed semantics.

\begin{theorem}[RSS Overapproximates RSNI for Compositions]\label{theorem:ss-sni-compositions}
   If $\wfc{\contract{}{x+y}}$ and $\contract{}{x+y}\vdash P : \rss $, 
   then $\contract{}{x+y}\vdash P : \rsni$.
\end{theorem}

\Cref{cor:comb-contract} relates the security of a compiler for a well-formed composition with the security of its composing semantics.
In particular, if a compiler is \rssp{} for a well-formed composition $\contract{}{x+y}$, than it is also \rssp{} for the composing semantics $\contract{}{x}$ and $\contract{}{y}$.
Dually, if a compiler is not $\rssp$ for a component, then it is not $\rssp$ for any composition.

\begin{corollary}[RSSP and compositions, \coqed]\label{cor:comb-contract}
The following statements hold for any well-formed  $\contract{}{x+y}$:
\begin{itemize}
    \item If $\contract{}{x+y} \vdash  \comp{\cdot} : \rssp$ then $\contract{}{x} \vdash  \comp{\cdot} : \rssp$ and $\contract{}{y} \vdash  \comp{\cdot} : \rssp$.
    \item 
    If $\contract{}{x} \nvdash  \comp{\cdot} : \rssp$ or $\contract{}{y} \nvdash  \comp{\cdot} : \rssp$ then
    $\contract{}{x+y} \nvdash  \comp{\cdot} : \rssp$.
\end{itemize}
\end{corollary}

We remark that an analogous of \Cref{cor:comb-contract} holds  for \rsnip{}.

 \section{Lifting Compiler Guarantees}\label{sec:ext-framework}

Compiler countermeasures against speculative leaks are often developed and proven secure against a specific speculative semantics.
For instance, countermeasures against Spectre-PHT have been proven secure against the  $\contract{}{\Bv}$ semantics~\cite{S_sec_comp} modelling speculation over branch instructions.
CPUs, however, may employ other speculation mechanisms, whose details might even be unknown when the countermeasure is designed, beyond those originally targeted by the countermeasure. 

Ensuring the security of a countermeasure hence requires continuously validating their guarantees (e.g., through proofs) against stronger and stronger semantics, as soon as new speculation mechanisms are discovered and modeled. 
For instance, in the context of Spectre-PHT attacks, countermeasures need to be proved secure against stronger semantics like $\contractSpec{}{\Bv + \Jv + \Sv + \Rv}$ rather than $\contract{}{\Bv}$ as done in~\cite{S_sec_comp}.
To reduce the burden of re-proving security whenever a new speculation mechanism is discovered, we need ways of lifting security guarantees from weaker to stronger semantics, which supports more speculations.

In this section, we address this issue by precisely characterizing under which conditions the scope of a countermeasure can be securely extended to other speculation mechanisms. 
More precisely, we study whenever the security guarantees provided by a compiler targeting a semantics $\contract{}{x}$ can be lifted to a stronger semantics $\contract{}{y}$, i.e., $\contract{}{x} \sqsubseteq \contract{}{y}$.
First, we introduce the preconditions for our lifting theorem, i.e., the notions of Independence and Safe Nesting, as well as our main result: the Lifted Compiler Preservation Theorem (\Cref{cor:lift-comp-pres-paper}) characterizing when security guarantees can be lifted to stronger semantics (\Cref{ssec:wft}). 
Then, we introduce Syntactic Independence (\Cref{ssec:si}) and Trapped Speculation (\Cref{ssec:trapped}), a set of sufficient conditions for Independence and Safe Nesting respectively. As we show in \Cref{sec:countermeasures}, these preconditions can be used in many cases to significantly simplify proofs of Independence and Safe Nesting in practice.

\subsection{Lifting Theorem}\label{ssec:wft}
In this section, we precisely characterize the sufficient conditions for  lifting the security guarantees provided by a compiler targeting a semantics $\contract{}{x}$ to a stronger semantics $\contract{}{x+y}$.
We start by providing a high-level intuition about how two component semantics $\contract{}{x}$ and $\contract{}{y}$ can interact when composed as $\contract{}{x+y}$ (\Cref{ssec:wft-interplay}).
Next, we formalize the notions of Independence (\Cref{ssec:wft-ind}) and Safe Nesting (\Cref{ssec:wft-nesting}), which precisely characterize when a compiler's guarantees can be lifted from $\contract{}{x}$ to $\contract{}{x+y}$.
Then, we introduce Conditional Robust Speculative Safety Preservation (CRSSP), a new secure compilation criterion ensuring that $\rss$ is preserved in the composed semantics $\contract{}{x+y}$ by the compiler for $\contract{}{x}$ only for those programs that already satisfy \rss for $\contract{}{y}$.
We conclude by stating and explaining our lifting theorem (\Cref{ssec:wft-theorem}) which is the main result of this section.

\subsubsection{Interplay of Semantics}\label{ssec:wft-interplay}
\begin{figure}[!ht]
\tikzset{
    between/.style args={#1 and #2}{
         at = ($(#1)!0.5!(#2)$)
    }
}
\begin{tikzpicture}
    \node (h1) {$\scriptstyle\startl{x}$};
    \node[right = 0.6 cm of h1] (m1) {$\scriptstyle\tauStack$};
    \node[right = 0.6 cm of m1] (h2) {$\scriptstyle\startl{y}$};
    \node[right = 0.6 cm of h2] (m5) {$\scriptstyle\tauStack'$};
    
    \node[below = 1 cm of h1] (b1){$\scriptstyle\startl{y}$};
    \node[right = 0.6 cm of b1] (m2) {$\scriptstyle\tauStack$};
    \node[below = 1 cm of h2] (b2){$\scriptstyle\startl{x}$};
    \node[below = 1 cm of m5] (m6) {$\scriptstyle\tauStack'$};

    \node[right = 0.6 cm of m5] (h3) {$\scriptstyle\rollbl{y}$};
    \node[right = 0.6 cm of h3] (m3) {$\scriptstyle\tauStack''$};
    \node[right = 0.6 cm of m3] (h4) {$\scriptstyle\rollbl{x}$};

    \node[below = 1 cm of h3] (b3){$\scriptstyle\rollbl{x}$};
    \node[right = 0.6 cm of b3] (m4) {$\scriptstyle\tauStack''$};
    \node[below = 1 cm of h4] (b4){$\scriptstyle\rollbl{y}$};

    \node[between = h1 and m1](t){$\cdot$};
    \node[between = b1 and m2](t){$\cdot$};
    \node[between = m1 and h2](t1){$\cdot$};
    \node[between = m2 and b2](t2){$\cdot$};
    
    \node[between = h2 and m5](t){$\cdot$};
    \node[between = b2 and m6](t){$\cdot$};
    \node[between = m5 and h3](t){$\cdot$};
    \node[between = m6 and b3](t){$\cdot$};

    \node[between = h3 and m3](t3){$\cdot$};
    \node[between = b3 and m4](t4){$\cdot$};
    \node[between = m3 and h4](t){$\cdot$};
    \node[between = m4 and b4](t){$\cdot$};

    \draw[decorate,decoration={brace,mirror, amplitude=5pt},below=10pt]
        (h2.south west) -- (h3.south east) node[font=\footnotesize, draw, midway, below=7pt] {\region{3}};
        
    \draw[decorate,decoration={brace, amplitude=5pt},above=10pt]
        (b2.north west) -- (b3.north east)  {};
    
    \node[fit=(h1) (m1) (t1), inner sep=5pt, fill=red, opacity=0.1, align=center, rounded corners] (background) {};
    \node[fit=(h2) (m5) (h3), inner sep=5pt, fill=green, opacity=0.1, align=center, rounded corners] (background) {};
    \node[fit=(t3)  (h4), inner sep=5pt, fill=red, opacity=0.1, align=center, rounded corners] (background) {};

    \node[fit=(b1) (m2) (t2), inner sep=5pt, fill=blue, opacity=0.1, align=center, rounded corners] (background) {};
    \node[fit=(b2) (m6) (b3), inner sep=5pt, fill=green, opacity=0.1, align=center, rounded corners] (background) {};
    \node[fit=(t4) (b4), inner sep=5pt, fill=blue, opacity=0.1, align=center, rounded corners] (background) {};

    \draw[draw = red, dotted] ([yshift=.3em]m1.north) to ([yshift=0.6cm] m1.north) to node[draw = black, above= 4pt,font=\footnotesize, solid](sd1){\region{1}}  ([yshift=0.6cm]m3.north) to ([yshift=.3em]m3.north);

     \draw[draw = blue, dotted ] ([yshift=-.3em]m2.south) to ([yshift=-0.6cm] m2.south) to node[draw = black, opacity = 1, below = 4pt,font=\footnotesize, solid](sd2){\region{2}}  ([yshift=-0.6cm]m4.south) to ([yshift=-.3em]m4.south);
\end{tikzpicture}
\vspace{-1em}
\caption[]{Interplay of semantics $\contract{}{x}$ and $\contract{}{y}$ when executing a program under the combined semantics $\contract{}{x+y}$.}
\label{figure:lifting-explanation}
\end{figure}

To understand the challenges involved in lifting security guarantees from a weaker semantics $\contract{}{x}$ to a stronger semantics $\contract{}{x+y}$, one needs to consider the interactions of $\contract{}{x}$ and $\contract{}{y}$.
\Cref{figure:lifting-explanation} depicts two portions of traces produced when executing a program under the composed semantics $\contract{}{x+y}$.
The first trace (top) starts with a speculative transaction from semantics $\contract{}{x}$ (highlighted in red and starting with action $\startl{x}$).
Inside the red speculative transaction, there is a \emph{nested} speculative transaction from semantics $\contract{}{y}$ (highlighted in green, starting with action $\startl{y}$ and ending with action  $\rollbl{y}$).
After the termination of the nested transaction, the outer red transaction continues until the end of its speculative window (action $\rollbl{x}$).
Dually, the second trace (bottom) starts with a speculative transaction from $\contract{}{y}$ followed by a nested transaction from $\contract{}{x}$.\looseness=-1
 
Thus, there are three different regions that might results in leaks:
\begin{asparaitem}
\item \textbf{\region{1}}: the speculative transaction started by $\contract{}{x}$ (highlighted in red),
\item \textbf{\region{2}}: the speculative transaction started by $\contract{}{y}$ (highlighted in blue), and 
\item \textbf{\region{3}}: the nested transactions (highlighted in green).
\end{asparaitem}
Note that while regions 1 and 2 are generated by a single semantics, the nested speculative transactions in \region{3} only arise in the combined semantics $\contract{}{x+y}$. 
While leaks in region 1 are fixed by proving the security of a compiler for the speculative semantics $\contract{}{x}$, i.e., $\contract{}{x} \vdash \comp{\cdot} : \rssp$, we need additional conditions to ensure the absence of leaks in Regions 2 and 3.
Thus, we introduce the notions of Independence with respect to a speculative semantics $\contract{}{y}$ and Safe Nesting which ensure the absence of leaks in Regions 2 and 3 respectively.

\subsubsection{Independence}\label{ssec:wft-ind}

Intuitively, when trying to lift our compiler security from the speculative semantics $\contract{}{x}$ to a stronger semantics $\contract{}{x+y}$, we need to ensure that the compiler does not introduce new leaks for the extension semantics $\contract{}{y}$.
The Independence property precisely characterizes this aspect:
A compiler $\comp{\cdot}$ for the origin semantics $\contract{}{x}$ is called independent for the extension semantics $\contract{}{y}$ iff
the compiler does not introduce further leaks under the extension semantics $\contract{}{y}$.

\begin{definition}[Independence in Extension]\label{def:ind-paper}
    \[
        \IND{\contract{}{y}}{\comp{\cdot}} \isdef \forall P\ldotp 
        \text{ if } \contract{}{y} \vdash P : \rss
        \text{ then }\contract{}{y} \vdash \comp{P} : \rss
    \]
\end{definition}

Note that Independence differs from \rssp{} (\Cref{def:compiler-contract-security-paper}) in that the former uses  $\contract{}{y}$ in the pre- and post-condition, whereas the latter employs the $\contract{}{NS}$ in its pre-condition.

As stated in \Cref{lem:self-ind}, a compiler that is \rssp{} for $\contract{}{x}$ is also Independent w.r.t. this semantics.
\begin{corollary}[Self Independence, \coqed]\label{lem:self-ind}
    If $\contract{}{x} \vdash \comp{\cdot} : \rssp$, then 
     $\IND{\contract{}{x}}{\comp{\cdot}}$.
\end{corollary}

\subsubsection{Safe Nested Speculation}\label{ssec:wft-nesting}
Given a combined semantics $\contract{}{x+y}$, a program $P$ has Safe Nested Speculations (denoted with $\contract{}{x+y} \vdash P : \safeN{}$) if all actions inside nested speculative transactions are safe.
Safe nesting, therefore, ensures that there are no unsafe interactions between the composing semantics $\contract{}{x}$ and $\contract{}{y}$.

\begin{definition}[Safe Nested Speculation]\label{def:safe-nest-paper}
    \begin{align*}
    \contract{}{x+y} \vdash P : \safeN{} \isdef&\
        \forall \tauStack \in \behxy{(P)}, 
        \text{ if }
        \\
        &\
        \startl{a} \cdot \tauStack' \cdot \rollbl{a} \text{ is a subtrace of $\tauStack$}
        \text{ and }
        \startl{b} \cdot \tauStack'' \cdot \rollbl{b} \text{ is a subtrace of $\tauStack'$} 
        \\
        &\
        \text{ then }
        \safe{\tauStack''} 
        \text{ where }
        a \in \{x,y\}
        \text{ and } 
        b \in \{x,y\} \setminus \{a\}
    \end{align*}
\end{definition}
Finally, we say that a compiler satisfies the Safe Nested Speculation property, written $\contract{}{x+y} \vdash \comp{\cdot} : \safeN{}$, iff all its compiled programs satisfy \Cref{def:safe-nest-paper}, i.e., $\forall P.\ \contract{}{x+y} \vdash \comp{P} : \safeN{}$.

\subsubsection{Conditional Robust Speculative Safety Preservation}\label{ssec:wft-crssp}
Often compilers implementing Spectre countermeasures are developed to prevent leaks introduced by a \emph{specific} speculation mechanism.
Hence, when lifting their security guarantees to a semantics accounting for additional speculation mechanisms, compiled programs might still contain some leaks that the compiler was not designed to prevent in the first place, i.e., leaks caused \emph{only} by the additional mechanisms.   
However, \rssp{} is too strict as a criterion here, since it cannot distinguish between the different speculation mechanisms that might cause the leak.
To account for this, we propose a new secure compilation criterion called Conditional Robust Speculative Safety Preservation ($\crssp$, \Cref{def:crssp}).
As the name indicates, $\crssp$ is a variant of $\rssp$ that restricts $\rss$ preservation \emph{only} to those programs that do not contain leaks caused only by the additional mechanisms. 

\begin{definition}[Conditional Robust Speculative Safety Preservation (\crssp)]\label{def:crssp}
  $$\contract{}{x}, \contract{}{y} \vdash  \comp{\cdot} : \crssp \isdef \forall \src{P} \in \src{L}. \text{ if }
    \contract{}{NS} \vdash \src{P} : \rss \text{ and } \contract{}{y} \vdash \src{P} : \rss \text{, then } \contract{}{x+y} \vdash \comp{\src{P}} : \rss.$$
\end{definition}

\subsubsection{Lifting Theorem}\label{ssec:wft-theorem}
We are now ready to introduce the main result of this section, that is, our lifting theorem (\Cref{cor:lift-comp-pres-paper}).
Intuitively,  we can lift the security guarantees of a compiler $\comp{\cdot}$ targeting semantics $\contract{}{x}$ to a stronger WFC semantics $\contract{}{x+y}$ for a program $P$ provided that the compiler $\comp{\cdot}$ (1) is RSSP for $\contract{}{x}$, (2) fulfills Independence for the extension semantics $\contract{}{y}$, and (3) fulfills Safe  Nesting for the combined semantics $\contract{}{x+y}$.

\begin{theorem}[Lifted Compiler Preservation, \coqed]\label{cor:lift-comp-pres-paper}
    If 
    $\contract{}{x} \vdash  \comp{\cdot} : \rssp$
    and
    $\IND{\contract{}{y}}{\comp{\cdot}}$
    and
    $\contract{}{x+y} \vdash \comp{\cdot} : \safeN{}$
    and
    $\wfc{\contract{}{x+y}}$ 
    , then
    $\contract{}{x}, \contract{}{y} \vdash  \comp{\cdot} : \crssp$.
\end{theorem}

We remark that our lifted guarantees hold only for programs $P$ that are initially secure w.r.t. the extension semantics  $\contract{}{y}$.
That is, the compiler enjoys $\contract{}{x}, \contract{}{y} \vdash  \comp{\cdot} : \crssp$ not the stronger property $\contract{}{x+y} \vdash  \comp{\cdot} : \rssp$.
The reason is that while compiler $\comp{\cdot}$ does not introduce further leaks under the extension semantics $\contract{}{y}$ (due to Independence), it might not prevent $\contract{}{y}$-leaks already present in the source program.

For the traces depicted in \Cref{figure:lifting-explanation}, \Cref{cor:lift-comp-pres-paper} ensures \rss{} (for any program satisfying \rss{} for the extension semantics $\contract{}{y}$) under the composed semantics as follows.
For \region{1}, \rss follows from $\contract{}{x} \vdash \comp{\cdot} : \rssp$.
For \region{2}, \rss follows from the program being originally \rss under $\contract{}{y}$ and from $\comp{\cdot}$ fulfilling Independence for extension $\contract{}{y}$.
Finally, for \region{3}, \rss follows from the compiled program having Safe Nesting.

\Cref{cor:lift-comp-pres-paper} allows us to lift the security guarantees of our secure compilers to stronger semantics, without worrying about unexpected leaks introduced by other speculation mechanisms (captured by the extension semantics) and, crucially, \emph{without} requiring new secure compilation proofs.

Next, we give sufficient  conditions for  Independence (\Cref{ssec:si}) and Safe Nesting (\Cref{ssec:trapped}).\looseness=-1

\subsection{Syntactic Independence: Independence for Free}\label{ssec:si}

To simplify the task of proving Independence, we now introduce \emph{Syntactic Independence} (SI) (\Cref{def:si}), a syntactic sufficient condition for Independence. As the name suggests, SI can be checked by \textit{syntactic} inspection of the compiler $\comp{\cdot}$ and the extension semantics $\contract{}{y}$.

Before formalizing SI, we introduce some notation.
Given a compiler $\comp{\cdot}$, we denote by $\insertedInstrs{\comp{\cdot}}$ the set of instructions that the compiler inserts during compilation.
For instance, for a simple compiler  $\complfenceB{\cdot}$ that inserts $\kywd{spbarr}$ instructions after branch instruction to prevent speculation~\cite{S_sec_comp}, $\insertedInstrs{\complfenceB{\cdot}}$ is the set  $\{ \kywd{spbarr} \}$.
Given a semantics $\contract{}{}$, we denote by $\speculationInstrs{\contract{}{}}$ the set of instructions that trigger speculation in $\contract{}{}$.
For instance, for semantics $\contract{}{\SLSv}$, which models straight-line speculation over return instructions, the set $\speculationInstrs{\contract{}{\SLSv}}$ is $\{ \pret \}$.

We are now ready to formalize Syntactic Independence (\Cref{def:si}).
In a nutshell, a compiler \comp{\cdot} is syntactically independent for a semantics \contract{}{} (denoted with $\SI{}{\comp{\cdot}}{\contract{}{}}$) if  the compiler does not insert (1) any instructions that trigger speculation under $\contract{}{}$, and (2) any instructions that produce data-dependent actions or modify the program state (except for the program counter $\textbf{pc}$ and the stack pointer $\textbf{sp}$).
The first requirement ensures that the compiler does not introduce new (potentially unsafe) speculative transactions, whereas the second requirement ensures that the compiler does not introduce unsafe actions into existing safe speculative transactions under $\contract{}{}$.

\begin{definition}[Syntactic Independence]\label{def:si}
    \[
    \SI{}{\comp{\cdot}}{\contract{}{}} \isdef \insertedInstrs{\comp{\cdot}} \cap \speculationInstrs{\contract{}{}} = \emptyset \text{ and } \insertedInstrs{\comp{\cdot}} \cap \{\jzC{}, \jmpC, \storeC, \loadC, \leftarrow  \} = \emptyset
    \]
\end{definition}

For instance, the compiler  $\complfenceB{\cdot}$ mentioned above is SI w.r.t. $\contract{}{\SLSv}$ since $ \insertedInstrs{\complfenceB{\cdot}} \cap \speculationInstrs{\contract{}{\SLSv}} = \{ \kywd{spbarr} \} \cap \{ \pret \} = \emptyset$ and  
 $\{ \kywd{spbarr} \} \cap \{\jzC{}, \jmpC, \storeC, \loadC, \leftarrow  \} = \emptyset$.

\Cref{cor:syn-ind-ind-paper} connects Syntactic Independence and Independence. 

\begin{lemma}[SI Implies Independence]\label{cor:syn-ind-ind-paper}
If
$\SI{\contract{}{x}}{\comp{\cdot}}{\contract{}{y}} $, then 
$ \IND{\contract{}{y}}{\comp{\cdot}}$.
\end{lemma}

We remark that checking SI is significantly simpler than manually proving Independence.
As we show in \Cref{ssec:ind-compilers}, SI plays a critical role in reducing the amount of Independence proofs necessary to carry out our security analysis.
Despite its restrictiveness, SI is applicable to two classes of compiler countermeasures---\texttt{lfence}-based countermeasures~\cite{Intel-compiler} and return-trampoline countermeasures~\cite{retpoline}---that work by stalling speculative execution.
For more complex countermeasures, e.g., speculative load hardening~\cite{slh}, that aim at preventing speculative leaks (rather than preventing speculation altogether), SI is not applicable since the compiler might instrument the program with additional instructions modifying the program state.
In this case, in our security analysis we fall-back to standard Independence proofs.

\subsection{Trapped Speculation: Fulfilling Safe Nesting}\label{ssec:trapped}
Showing that the compiled program fulfils the Safe Nesting condition is challenging because it requires reasoning about both the compiler as well as the interactions between component semantics.
To help with this, we now introduce a sufficient condition on compilers that ensures that all compiled programs enjoy the Safe Nesting property.

\begin{definition}[Trapped Speculation of Compiler]\label{def:trapped-spec-paper}
\begin{align*}
    \contract{}{x} \vdash \trappedS{\comp{\cdot}} \isdef &\
    \forall P, 
    \tauStack \in \behx{\comp{P}}, 
    \tau \in \specProject{\tauStack}\ldotp 
    \exists \ctr\ldotp \tau = \rollbackObsx{\ctr} \text{ or } \tau = \startObsx{\ctr}
\end{align*}
\end{definition}

In a nutshell, a compiler satisfies \Cref{def:trapped-spec-paper} iff it \emph{traps speculation}, which we model by requiring that the only speculative actions produced by compiled programs are either  $\startObsx{\ctr}$ (i.e., begin of speculation) or $\rollbackObsx{\ctr}$ (i.e., end of speculation).
This, thus, implies that there are no unsafe actions between the start of a speculation transaction and its rollback as required by Safe Nesting.
For example, a compiler inserting fences into the program stops speculation immediately and fulfils our definition of Trapped Speculation.
Similarly, a compiler inserting a so-called return trampoline~\cite{retpoline} (which traps speculation in a loop) also fulfills Trapped Speculation.

\Cref{def:trapped-spec-paper} relies on the speculative projection function $\specProject{}$, which removes all non-speculative observations from the trace and is defined as the inverse of the non-speculative projection $\nspecProject{}$.

\Cref{lem:trapped-imp-safe-nest-paper} connects Trapped Speculation (\Cref{def:trapped-spec-paper}) with the Safe Nesting condition $\safeN{}$. \begin{lemma}[Trapped Speculation Implies Safe Nesting, \coqed]\label{lem:trapped-imp-safe-nest-paper}
If
    $\contract{}{x} \vdash \trappedS{\comp{\cdot}} $ 
    then $\contract{}{x+y} \vdash \comp{\cdot} : \safeN{}$.
\end{lemma}

As we show in \Cref{ssec:trapped-compilers},  \Cref{def:trapped-spec-paper} significantly reduces the proof burden. 
In particular, rather than having to reason about the combined semantics $\contract{}{x+y}$ when showing Safe Nesting, we can just reason about the compiler $\comp{\cdot}$ for semantics $\contract{}{x}$.

\smallskip
With this formal setup, we now move on to our security analysis, which demonstrates how to attain \crssp\ for a number of countermeasures, by relying on the notions of Independence, Syntactic Independence, Safe Nesting, and Trapped Speculation.

 \section{Countermeasures Analysis}\label{sec:countermeasures}

In this section, we present a comprehensive analysis of the security guarantees provided by Spectre countermeasures implemented in major compilers. Our analysis covers 9 countermeasures (summarized in \Cref{ssec:compilers}) covering 5 Spectre attacks: Spectre-PHT~\cite{spectre}, Spectre-BTB~\cite{spectre}, Spectre-RSB~\cite{ret2spec, spectreRsb}, Spectre-STL~\cite{S_specv4}, and Spectre-SLS~\cite{sls-whitepaper}.
Using our secure compilation framework, we precisely characterize the security guarantees provided by these countermeasures against the five speculative semantics from \Cref{sec:bg-formal} and their combinations.

\Cref{cor:lift-comp-pres-paper} plays a key role in this analysis since we use it to lift security guarantees from simpler semantics to their combinations, thereby significantly reducing the number of secure compilation proofs that need to be carried out.
We remark that to apply the lifting theorem, which allows us to lift the security guarantees of the compiler to stronger semantics, we need to show (1) Security in Source, (2) Independence in Extension, and (3) Safe Nesting, which is what we do next.

First, we tackle Security in Source and prove that the compilers are \rssp{} w.r.t. their base speculative semantics (\Cref{ssec:sec-compilers}).
Then we focus on Independence, and show Syntactic Independence for four of our compilers and fall-back to full Independence proofs for the remaining five (\Cref{ssec:ind-compilers}).
Next, we analyze Safe Nesting and show that Trapped Speculation (\Cref{def:trapped-spec-paper}) applies in all cases except two (\Cref{ssec:trapped-compilers}).
Finally, we combine these results by evaluating the strongest security guarantees that can be achieved for these compilers using \Cref{cor:lift-comp-pres-paper} (\Cref{ssec:lift-compilers}).

\subsection{The Compilers}\label{ssec:compilers}

\Cref{tab:comp-countermeasures-new} summarizes the Spectre countermeasures that we analyze.
Next, we describe them in more detail.
We refer to our technical report for full definitions of our compilers.

\begin{table}\centering
    \renewcommand{\arraystretch}{1.3} \addtolength{\tabcolsep}{-0.2em} \begin{tabular}{clcc}
    \toprule
      Name & Symbol & Base Semantics & Source  \\
      \midrule
      \makecell{Fences for Returns Straight-Line} & \complfenceSLS{\cdot} & $\contract{}{\SLSv}$ & GCC/\clang{}   \\
      \makecell{Retpoline for Jumps}   & \compretpJ{\cdot} & $\contract{}{\Jv}$ &  GCC/\clang{}/\cite{retpoline} \\
      \makecell{Retpoline with fence for Jumps} & \compretpJF{\cdot} & $\contract{}{\Jv}$ & \gcc{}/\cite{retpoline} \\
      \makecell{Retpoline for Returns}& \compretpR{\cdot} & $\contract{}{\Rv}$ & \gcc{}/\cite{ret2spec} \\
      \makecell{Fences for Returns}  &  \complfenceR{\cdot} & $\contract{}{\Rv}$ & \cite{ret2spec} \\
      \makecell{Fences for Stores}  & \complfenceS{\cdot} & $\contract{}{\Sv}$ & \cite{v4_fence}  \\
      \makecell{Ultimate SLH for Branches}& \compuslhB{\cdot} & $\contract{}{\Bv}$ &  \cite{uslh} (extends \clang{}'s SLH) \\ 
      \makecell{Strong SLH for Branches}  & \compsslhB{\cdot} & $\contract{}{\Bv}$ & \cite{S_sec_comp} (extends \clang{}'s SLH)\\
      \makecell{Fences for Branches}  & \complfenceB{\cdot} & $\contract{}{\Bv}$ & ICC/\clang{}  
      \\
      \bottomrule
    \end{tabular}
    \caption{Analyzed compiler countermeasures.}
    \label{tab:comp-countermeasures-new}
\end{table}

\paragraph*{Fences Against Straight-Line Speculation ($\complfenceSLS{\cdot}$)}\label{ssec:sls-fences}
Modern CPUs can speculatively bypass $\retC$ instructions \cite{sls-whitepaper, sls-whitepaper2}. 
Compilers like \gcc{} and \clang{} (with option \textsf{-mharden-sls=all}) prevent this by injecting a speculation barrier after every $\retC$ instruction. Since the $\retC$ instruction is an unconditional change in control flow, the barrier\footnote{For x86, the int3 single-byte instruction is used to reduce the binary size.} will not be executed architecturally but only when straight-line speculation is happening.
We model this countermeasure in the $\complfenceSLS{\cdot}$ compiler that inserts a barrier after every $\retC$ instruction.

\paragraph*{Retpoline for Indirect Jumps ($\compretpJ{\cdot}, \compretpJF{\cdot}$)}\label{ssec:v2-retpoline}
Spectre-BTB attacks~\cite{spectre} exploit speculation over indirect jumps.
The retpoline countermeasure~\cite{retpoline} replaces all indirect jumps in the code with a return trampoline, i.e., with  a construct that traps the speculation in an infinite loop. 
Retpoline is available in all general compilers like \clang{}
(\textsf{-mretpoline}) and \gcc{} (with option \textsf{-mindirect-branch}) and is widely deployed because current developed hardware mitigations are not enough to protect against  indirect jump speculation~\cite{barberis2022branch}.
We consider two models of the retpoline countermeasure.
The $\compretpJ{\cdot}$ compiler replaces every indirect jump instruction $\pjmp{e}$ with a return trampoline.
Additionally, we consider the $\compretpJF{\cdot}$ compiler, which inserts an additional fence after $\retC$ instructions in trampolines (to prevent straight-line speculation).
This compiler corresponds to activating both flags \textsf{-mindirect-branch} and \textsf{mharden-sls=all} in \gcc{}.

\paragraph*{Retpoline for Returns ($\compretpR{\cdot}$)}\label{v5-retpoline}
A variant of the retpoline countermeasure has been proposed to prevent Spectre-RSB attacks~\cite{ret2spec}.
This countermeasure (implemented in $\gcc$ with option \textsf{-mfunction-return}) replaces each  $\retC$ instruction with a return trampoline; trapping misprediction caused by $\retC$ instructions.
We modeled this countermeasure in the $\compretpR{\cdot}$ compiler.

\paragraph*{Fences for Returns ($\complfenceR{\cdot}$)}\label{ssec:v5-fences}
\citet{ret2spec} propose to add an \textbf{lfence} instruction after every \textbf{call} instruction.
This ensures that mis-speculations over $\retC$ instructions involving the Return Stack Buffer will always land on one of the injected speculation barriers, thereby preventing speculative leaks.
We model this countermeasure in the $ \complfenceR{\cdot}$ compiler, which replaces every $\textbf{call}\ f$ instruction with $\textbf{call}\ f; \textbf{spbarr}$.

\paragraph*{Fences for Stores ($\complfenceS{\cdot}$)}\label{ssec:v4-fences}
To prevent speculation over store-to-load bypasses~\cite{v4_fence} (also known as Spectre-STL), Intel suggested to insert \textbf{lfence} instruction after every store, thereby ensuring that all stores are committed to main memory and preventing speculation.
However, no mainstream compiler implements this countermeasure due to the high performance overhead.
We model this countermeasure in the $ \complfenceS{\cdot}$ compiler, which replaces (1) every $\textbf{store}\ x, e$ instruction with $\textbf{store}\ x,e; \textbf{spbarr}$, and (2) every $\textbf{store}_{\mathit{prv}}\ x, e$ instruction with $\textbf{store}_{\mathit{prv}}\ x,e; \textbf{spbarr}$.

\paragraph{(Strong) Speculative Load Hardening (SSLH, $\compsslhB{\cdot}$)}
Modern CPUs speculate over the outcome of branch instructions~\cite{spectre}.
\clang{} (with option \textsf{-mspeculative-load-hardening}) protects against these speculative leaks by (1) using a speculation flag that tracks whenever misprediction is currently happening or not, and (2) using the flag to conditionally mask loads and stores to prevent the leaks \cite{slh}.
\citet{S_sec_comp} investigated the security of SLH and showed it insecure w.r.t. $\contract{}{\Bv}$. 
They proposed a improved version called strong-SLH and prove it secure w.r.t. to $\contract{}{\Bv}$ semantics.
We evaluate their compiler\footnote{\label{fn:shared}Technically, they targeted a While language while we have an assembly-like language \muasm{}. However, the translation is straightforward.} $\compsslhB{\cdot}$ in our framework to see if we can lift the security guarantees to stronger semantics.

\paragraph{Ultimate Speculative Load Hardening (USLH, $\compuslhB{\cdot}$)}
\citet{uslh} showed that variable-latency arithmetic instructions can leak secret information under speculation, and this is not prevented by speculative-load hardening~\cite{slh} or by its strong-variant $\compsslhB{\cdot}$~\cite{S_sec_comp}.
To prevent these speculative leaks, they propose the ``ultimate speculative load hardening'' compiler, which extends strong-SLH by additionally masking inputs to variable-latency arithmetic instructions.

We modelled the core aspects of ultimate USLH in the $\compuslhB{\cdot}$ compiler.
For this, we extended \muasm{} to support a dedicated instruction $\vassign{x}{y}{z}$ denoting variable-latency computations.
Furthermore, we extended \muasm{} events to include a new observation $\opObs{y}{z}$ that is emitted by the new instruction $\vassign{x}{y}{z}$. 
\begin{align*}
\mi{Instructions}~ i \bnfdef&\ \cdots \mid  \vassign{x}{y}{z}
&
\mi{\mu arch.\ Acts.}~\delta \bnfdef&\ \cdots \mid \com{\opObs{x}{y}}
\end{align*}
These extensions augment the constant time observer (used in our model in \Cref{ssec:atk-model} as well as in the strong-SLH formalization in \cite{S_sec_comp}) to capture leaks related with variable-latency instructions. 
We denote the $\contract{}{\Bv}$ semantics extended with the new observer as $\contract{ct + vl}{\Bv}$, and we have the following leakage ordering: $\contract{ct}{\Bv}\sqsubseteq \contract{ct + vl}{\Bv}$.

\paragraph*{Fences for Branches ($\complfenceB{\cdot}$)}
Another approach to prevent leaks due to branch misprediction is injecting \texttt{lfence} instructions after branch instructions. Compilers like Intel ICC (with flag \textsf{-mconditional-branch=all-fix}) and \clang{} (with flag \textsf{-x86-speculative-load-hardening-lfence}) implement this countermeasure.
This countermeasure was already modelled (and proved secure for $\contractSpec{}{\Bv}$) in \cite{S_sec_comp} as $\complfenceB{\cdot}$\textsuperscript{\ref{fn:shared}},
which replaces every $\jzC{}\ x, l$ with $\jzC{}\ x, l ; \textbf{spbarr} $ and we want to investigate the applicability of the lifting theorem to $\complfenceB{\cdot}$'s security guarantees as well.

\subsection{Security of the Compilers}\label{ssec:sec-compilers}

Here, we report the results of the security analysis of each compiler with respect to their base speculative semantics, as indicated in \Cref{tab:comp-countermeasures-new}.
\Cref{thm:security-base-semantics} states that each compiler is \rssp{} w.r.t. its  base speculative semantics.
Even though we present all results altogether, we remark that each point in \Cref{thm:security-base-semantics} corresponds to an independent secure compilation proof.

\begin{theorem}[Compiler Security]\label{thm:security-base-semantics}
The following statements hold:
\begin{asparaitem}
\item \textit{($\SLSv$: Fence is secure)} $\contractSpec{}{\SLSv} \vdash \complfenceSLS{\cdot} : \rssp$
\item \textit{($\Rv$: Retpoline is secure)}  $\contractSpec{}{\Rv} \vdash \compretpR{\cdot} : \rssp$
\item \textit{($\Rv$: Fence is secure)}  $\contractSpec{}{\Rv} \vdash \complfenceR{\cdot} : \rssp$
\item \textit{($\Jv$: Retpoline is secure)}  $\contractSpec{}{\Jv} \vdash \compretpJ{\cdot} : \rssp$
\item \textit{($\Jv$: Retpoline with fence is secure)} $\contractSpec{}{\Jv} \vdash \compretpJF{\cdot} : \rssp$
\item \textit{($\Sv$: Fence is secure)} 
    $\contractSpec{}{\Sv} \vdash \complfenceS{\cdot} : \rssp$
\item \textit{($\Bv$: USLH is secure)}
     $\contractSpec{ct + vl}{\Bv} \vdash \compuslhB{\cdot} : \rssp$
\item \textit{($\Bv$: SSLH is secure \cite{S_sec_comp})} 
    $\contractSpec{}{\Bv} \vdash \compsslhB{\cdot} : \rssp$
\item \textit{($\Bv$: Fence is secure \cite{S_sec_comp})}
    $\contractSpec{}{\Bv} \vdash \complfenceB{\cdot} : \rssp$
\end{asparaitem}
\end{theorem}

\subsection{Independence of the Compilers}\label{ssec:ind-compilers}

We now investigate whether our compilers satisfy the Independence in Extension condition.
Due to space constraints, we do not report on the  Independence and Syntactic Independence for all our compilers and combinations and refer the interested reader to our companion report.
Instead, \Cref{thm:ind-res-paper} reports the strongest possible semantics for which we can prove (Syntactic) Independence for each compiler.

\begin{theorem}[Compiler Independence]\label{thm:ind-res-paper}
The following statements hold:
\begin{asparaitem}
\item \textit{($\SLSv$: Fence Independence)} $\SI{\contract{}{\SLSv}}{\complfenceSLS{\cdot}}{\contract{}{\Bv+ \Jv+ \Sv + \SLSv}}$ 
\item \textit{($\Rv$: Fence Independence)} $\SI{\contract{}{\Rv}}{\complfenceR{\cdot}}{\contract{}{\Bv+ \Jv +\Sv + \Rv}}$ 
\item \textit{($\Sv$: Fence Independence)} $\SI{\contract{}{\Sv}}{\complfenceS{\cdot}}{\contract{}{\Bv + \Jv + \Sv + \Rv}}\ \text{and}\ \SI{\contract{}{\Sv}}{\complfenceS{\cdot}}{\contract{}{\Bv + \Jv + \Sv + \SLSv}}$
\item \textit{($\Bv$: Fence Independence)}  $\SI{\contract{}{\Bv}}{\complfenceB{\cdot}}{\contract{}{\Bv + \Jv + \Sv + \Rv}}\ \text{and}\ \SI{\contract{}{\Bv}}{\complfenceB{\cdot}}{\contract{}{\Bv + \Jv + \Sv + \SLSv}}$
\item \textit{($\Bv$: USLH Independence)} $\IND{\contract{}{\Bv + \Jv + \Sv + \Rv}}{\compuslhB{\cdot}} \ \text{and}\ \IND{\contract{}{\Bv + \Jv + \Sv + \SLSv}}{\compuslhB{\cdot}}$
\item \textit{($\Bv$: SSLH Independence)}  $\IND{\contract{}{\Bv + \Jv + \Sv + \Rv}}{\compsslhB{\cdot}} \ \text{and}\ \IND{\contract{}{\Bv + \Jv + \Sv + \SLSv}}{\compsslhB{\cdot}}$
\item \textit{($\Rv$: Retpoline Independence)}  $\IND{\contract{}{\Bv + \Jv + \Sv + \Rv}}{\compretpR{\cdot}}$ 
\item \textit{($\Jv$: Retpoline Independence)} $\IND{\contract{}{\Bv + \Jv + \Sv + \Rv}}{\compretpJ{\cdot}}$ 
\item \textit{($\Jv$: Retpoline with Fence Independence)} $\IND{\contract{}{\Bv + \Jv + \Sv + \Rv}}{\compretpJF{\cdot}} \ \text{and}\ \IND{\contract{}{\Bv + \Jv + \Sv + \SLSv}}{\compretpJF{\cdot}}$
\end{asparaitem}
\end{theorem}

As stated in \Cref{thm:ind-res-paper}, for compilers $\complfenceSLS{\cdot}$, $\complfenceR{\cdot}$, $\complfenceS{\cdot}$, and $\complfenceB{\cdot}$ , we directly proved that Syntactic Independence holds even for the strongest possible combined semantics.
Given that Syntactic Independence implies Independence (\Cref{cor:syn-ind-ind-paper}), this allows us to derive Independence results for all these compilers through simple syntactic checks.

For the $\compretpJ{\cdot}$, $\compretpJF{\cdot}$ and $\compretpR{\cdot}$ compilers, instead, we have to fall back to a full Independence proof for the strongest semantics.
The reason is that these compilers add $\retC$ instructions to the code which could interact with $\Rv$ or $\SLSv$, i.e., $\insertedInstrs{\comp{\cdot}} \cap \speculationInstrs{\contract{}{\Rv}} = \{\retC \}$, which violates SI.
However, for weaker combinations not including $\Rv$ or $\SLSv$,  SI applies.

We also remark that Independence does not hold for the retpoline compiler $\compretpJ{\cdot}$ and the straight-line speculation semantics $\contract{}{\SLSv}$, i.e.,  $\contract{}{\SLSv} \nvdash \compretpJ{\cdot} : I$.
The reason is that the $\retC$ instructions injected by $\compretpJ{\cdot}$ as part of return trampolines can be speculatively bypassed under $\contract{}{\SLSv}$.
The additional speculation barrier injected by the strengthened compiler $\compretpJF{\cdot}$ fixes this issue and allows recovering Independence  w.r.t. $\SLSv$ as well, i.e., $\contract{}{\SLSv} \vdash \compretpJF{\cdot} : I$.

Finally, for the SLH compilers $\compuslhB{\cdot}$ and $\compsslhB{\cdot}$, we again had to perform full Independence proofs since these compilers inject instructions, which violate SI requirements, for tracking the speculation flag and for masking  $\storeC$ and $\loadC$ instructions.

\subsection{Safe Nesting of the Compilers}\label{ssec:trapped-compilers}

The last condition to fulfill for lifting security guarantees is Safe Nesting.
Rather than directly proving Safe Nesting, we study which compilers trap speculation according to \Cref{def:trapped-spec-paper}.
\Cref{thm:safe-nesting-results} precisely characterize which compilers enjoy this property.
Combining \Cref{thm:safe-nesting-results}  with the fact that trapped speculation implies Safe Nesting (\Cref{lem:trapped-imp-safe-nest-paper}), gives us a precise characterization of which compilers enjoy the Safe Nesting property.

\begin{theorem}[Compiler Safe Nesting]\label{thm:safe-nesting-results}
The following statements hold:
\begin{asparaitem}
\item \textit{($\SLSv$: Fence traps speculation)} $\contract{}{\SLSv} \vdash \trappedS{\complfenceSLS{\cdot}}$
\item \textit{($\Rv$: Retpoline traps speculation)}  $\contract{}{\Rv} \vdash \trappedS{\compretpR{\cdot}}$
\item \textit{($\Rv$: Fence traps speculation)}  $\contract{}{\Rv} \vdash \trappedS{\complfenceR{\cdot}}$
\item \textit{($\Jv$: Retpoline traps speculation)}  $\contract{}{\Jv} \vdash \trappedS{\compretpJ{\cdot}}$
\item \textit{($\Jv$: Retpoline with fence traps speculation)} $\contract{}{\Jv} \vdash \trappedS{\compretpJF{\cdot}}$
\item \textit{($\Sv$: Fence traps speculation)} 
    $\contract{}{\Sv} \vdash \trappedS{\complfenceS{\cdot}}$
\item \textit{($\Bv$: Fence traps speculation)} 
    $\contract{}{\Bv} \vdash \trappedS{\complfenceB{\cdot}}$
\end{asparaitem}
\end{theorem}

For all our compilers that rely on inserting $\textbf{spbarr}$ as a countermeasure, proving Trapped Speculation was easy since the speculation barriers immediately stop any speculative transaction.
In contrast, for the compilers that inject retpolines, Trapped Speculation follows from the fact that the return-trampoline traps speculation in a loop that does not produce visible events.

Trapped Speculation, however, does not hold for SLH-based countermeasures because these countermeasures do not block speculative execution but rather prevent leaks during speculation.
For $\compsslhB{\cdot}$ and $\compuslhB{\cdot}$, therefore, we need to  directly prove Safe Nesting for each combined semantics, which requires reasoning about all interactions of component semantics. 
\Cref{thm:safe-nesting-results-slh} reports the strongest semantics for which safe nesting holds. 

\begin{theorem}[Compiler Safe Nesting]\label{thm:safe-nesting-results-slh}
The following statements hold:
\begin{asparaitem}   
\item \textit{($\Bv$: SSLH Safe Nesting)}  $\contract{}{\Bv+\Sv+\Rv} \vdash \compsslhB{\cdot} : \safeN{}$ \text{ and } $\contract{}{\Bv+\Sv+\SLSv} \vdash \compsslhB{\cdot} : \safeN{}$
\item \textit{($\Bv$: USLH Safe Nesting)}  $\contract{\mathit{ct}+\mathit{vl}}{\Bv+\Sv+\Rv} \vdash \compuslhB{\cdot} : \safeN{}$ \text{ and } $\contract{\mathit{ct}+\mathit{vl}}{\Bv+\Sv+\SLSv} \vdash \compuslhB{\cdot} : \safeN{}$
\end{asparaitem}
\end{theorem}
Intuitively, Safe Nesting holds for the cases in \Cref{thm:safe-nesting-results-slh} because the SLH compilers preserve the invariant that the speculation flag is correctly set even when considering speculative transactions caused by semantics like $\sems$ or $\semr$.

In contrast,  Safe Nesting does not hold for combinations including $\semj$.
This follows from the fact that SSLH/USLH compilers do not correctly propagate the value of the speculation flag during indirect jumps.
As we state in~\Cref{thm:slh-insecure-j}, this also leads to both $\compsslhB{\cdot}$ and $\compuslhB{\cdot}$ being insecure w.r.t. any combination of semantics containing $\semj$ and $\semb$, i.e., any speculative semantics supporting branch and indirect jump speculation.

\begin{theorem}[SLH Insecurity w.r.t $\Jvr$]\label{thm:slh-insecure-j}For any speculative semantics $\contract{}{}$ such that $\contract{}{\Bv + \Jv} \sqsubseteq \contract{}{}$,  $\contract{}{}\nvdash \compsslhB{\cdot} : \rsnip$  and 
    $\contract{}{}\nvdash \compuslhB{\cdot} : \rsnip$.
\end{theorem}

We prove \Cref{thm:slh-insecure-j} by  finding a source program (which we present in the companion report for space constraints), that after compilation with $\compsslhB{\cdot}$/$\compuslhB{\cdot}$ still leaks due to indirect jump speculation. The gist of the insecure program is that while the target of the indirect jump is masked by $\compsslhB{\cdot}$/$\compuslhB{\cdot}$ (to prevent leaks), this does not prevent speculation of indirect jumps that can be used to bypass the instructions tracking the speculation flag inserted by the compiler.
Note that this issue could be fixed by relying on hardware support for control-flow-integrity like Intel-CET \cite{intel_cet} or ARM-BTI \cite{arm-cet}, which restrict the targets of indirect jumps to a fixed set of addresses  \emph{even under speculation}.
This should be sufficient to ensure that attackers cannot speculatively bypass instructions setting the speculation flags and it should allow us   derive Safe Nesting for the SLH compilers w.r.t. combinations including $\semj$.
We leave investigating this for future work.

\subsection{Lifting Security Guarantees}\label{ssec:lift-compilers}

We conclude our security analysis by using the results from Sections~\ref{ssec:sec-compilers}--\ref{ssec:trapped-compilers} together with \Cref{cor:lift-comp-pres-paper} to study how far we can lift the security guarantees provided by each compiler.
This allows us to precisely characterize the security  of these countermeasures even under stronger semantics.

\Cref{thm:lifted-guarantees-paper} summarizes the strongest lifted security guarantees one can derive for each of the studied compilers using our lifting theorem (\Cref{cor:lift-comp-pres-paper}).
After an explanation of the result, we present it visually in \Cref{fig:instantiate-slsf}.

\begin{theorem}[Lifted security guarantees]\label{thm:lifted-guarantees-paper}
The following statements hold:
\begin{itemize}
\item ($\Sv$: Fence)  
        $\contract{}{\Sv}, \contract{}{\Bv + \Jv + \Rv} \vdash \complfenceS{\cdot} : \crssp
            \text{ and } 
            \contract{}{\Sv}, \contract{}{\Bv+ \Jv + \SLSv} \vdash \complfenceS{\cdot} : \crssp$
\item ($\SLSv$: Fence) 
         $\contract{}{\SLSv}, \contract{}{\Bv +\Jv +\Sv} \vdash \complfenceSLS{\cdot} : \crssp$
\item ($\Rv$: Fence)     $\contract{}{\Rv} , \contract{}{\Bv + \Jv + \Sv} \vdash \complfenceR{\cdot} : \crssp$
\item ($\Jv$: Retpoline) $\contract{}{\Jv}, \contract{}{\Bv+ \Sv+  \Rv} \vdash \compretpJ{\cdot} : \crssp$
\item ($\Jv$: Retpoline with fences) $\contract{}{\Jv},\contract{}{\Bv + \Sv +\Rv} \vdash \compretpJF{\cdot} : \crssp 
             \text{ and } \contract{}{\Jv}, \contract{}{\Bv +\Sv + \SLSv} \vdash \compretpJF{\cdot} : \crssp$
\item ($\Rv$: Retpoline) $\contract{}{\Rv}, \contract{}{\Bv+ \Jv +\Sv} \vdash \compretpR{\cdot} : \crssp$
\item ($\Bv$: Fence) $\contract{}{\Bv},\contract{}{\Jv + \Sv + \Rv} \vdash \complfenceB{\cdot} : \crssp
             \text{ and } \contract{}{\Bv}, \contract{}{\Jv + \Sv +\SLSv} \vdash \complfenceB{\cdot} : \crssp$
\item ($\Bv$: USLH) $\contract{\mathit{ct} + \mathit{vl}}{\Bv},\contract{\mathit{ct} + \mathit{vl}}{\Sv +\Rv} \vdash \compuslhB{\cdot} : \crssp          \text{ and } \contract{\mathit{ct} + \mathit{vl}}{\Bv}, \contract{\mathit{ct} + \mathit{vl}}{\Sv + \SLSv} \vdash \compuslhB{\cdot} : \crssp$
\item ($\Bv$: SSLH) $\contract{}{\Bv},\contract{}{\Sv +\Rv} \vdash \compsslhB{\cdot} : \crssp 
             \text{ and } \contract{}{\Bv}, \contract{}{\Sv + \SLSv} \vdash \compsslhB{\cdot} : \crssp$

\end{itemize}
\end{theorem}

For all our compilers except the SLH ones, we can lift the security guarantees up to $\contract{}{\Bv + \Jv + \Sv + \Rv}$ or $\contract{}{\Bv+ \Jv + \Sv +\SLSv}$, the two strongest combined speculative semantics (as shown in \Cref{fig:instantiate}), i.e.,  the strongest attacker models considered in our paper.

For the SLH compilers $\compuslhB{\cdot}$ and $\compsslhB{\cdot}$, we are able to lift the security guarantees only up to $\contract{}{\Bv + \Sv + \Rv}$ and $\contract{}{\Bv + \Sv + \SLSv}$.
Lifting the guarantees of $\compuslhB{\cdot}$ and $\compsslhB{\cdot}$ to stronger semantics including $\semj$ is not possible because Safe Nesting is not fulfilled (cf. \Cref{ssec:trapped-compilers}).

Finally, for $\complfenceSLS{\cdot}$, $\complfenceR{\cdot}$ and $\compretpR{\cdot}$, we cannot lift security to both $\contract{}{\Bv + \Jv + \Sv + \Rv}$ and $\contract{}{\Bv+ \Jv + \Sv +\SLSv}$  because we cannot compose $\contract{}{\Rv}$ and $\contract{}{\SLSv}$ due to limitations of the combination framework (cf.~\Cref{ssec:bg-combinations}).

\begin{figure}[!ht]
    \centering
    \begin{tikzpicture}[line width=0.1 mm, arrows={[black]}]
          \node (BJSR) at (-1, 4) {{$\contractSpec{}{\Bv + \Jv + \Sv + \Rv}$}};
          \node (BJSSLS) at (1, 4) {$\contractSpec{}{\Bv + \Jv + \Sv + \SLSv}$};
          \node (BJS) at (-3.6, 2) {$\contractSpec{}{\Bv + \Jv + \Sv}$};
          \node (BJR) at (-2.4, 2) {$\contractSpec{}{\Bv + \Jv + \Rv}$};
          \node (BSR) at (-1.2, 2) {$\contractSpec{}{\Bv + \Sv + \Rv}$};
          \node (JSR) at (0, 2) {$\contractSpec{}{\Jv + \Sv + \Rv}$};
          \node (BJSLS) at (1.4, 2) {$\contractSpec{}{\Bv + \Jv + \SLSv}$};
          \node (BSSLS) at (2.8, 2) {$\contractSpec{}{\Bv + \Sv + \SLSv}$};
          \node (JSSLS) at (4, 2) {$\contractSpec{}{\Jv + \Sv + \SLSv}$};
          \node (BJ) at (-4,0) {$\contractSpec{}{\Bv + \Jv}$};
          \node (BS) at (-3,0) {$\contractSpec{}{\Bv + \Sv}$};
          \node (BR) at (-2,0) {$\contractSpec{}{\Bv + \Rv}$};
          \node (JS) at (-1,0) {$\contractSpec{}{\Jv + \Sv}$};
          \node (JR) at (0,0) {$\contractSpec{}{\Jv + \Rv}$};
          \node (SR) at (1,0) {$\contractSpec{}{\Sv + \Rv}$};
          \node (BSLS) at (2,0) {$\contractSpec{}{\Bv + \SLSv}$};
          \node (JSLS) at (3,0) {$\contractSpec{}{\Jv + \SLSv}$};
          \node (SSLS) at (4,0) {$\contractSpec{}{\Sv + \SLSv}$};
          \node (NS) at (0, -3) {$\contract{}{NS}$};
          \node (B) at (-2,-2) {$\contractSpec{}{\Bv}$};
          \node (J) at (-1,-2) {$\contractSpec{}{\Jv}$};
          \node (S) at (0,-2) {$\contractSpec{}{\Sv}$};
          \node (R) at (1,-2) {$\contractSpec{}{\Rv}$};
          \node[draw, rounded corners] (SLS) at (2,-2) {$\contractSpec{}{\SLSv}$};
          \draw[\lineCol]  (B.north) -- (BJ.south);
          \draw[\lineCol]  (J.north) -- (JS.south);
          \draw[\lineCol]  (J.north) -- (BJ.south);
          \draw[\lineCol]  (S.north) -- (JS.south);
          \draw[\lineCol]  (S.north) -- (SR.south) -- (R.north);
          \draw[\lineCol]  (B.north) -- (BS.south) -- (S.north) -- (SSLS.south) -- (SLS.north) -- (JSLS.south) -- (J.north);
          \draw[\lineCol]  (B.north) -- (BR.south) -- (R.north) -- (JR.south) -- (J.north);
          \draw[\lineCol]  (B.north) -- (BSLS.south) -- (SLS.north);
          \draw[\lineCol] (BJ.north) -- (BJS.south);
          \draw[\lineCol] (BJ.north) -- (BJR.south);
          \draw[\lineCol]  (BJ.north) -- (BJSLS.south);
          \draw[\lineCol]  (BS.north) -- (BSR.south);
          \draw[\lineCol]  (BS.north) -- (BSSLS.south);
          \draw[\lineCol]  (BS.north) -- (BJS.south);
          \draw[\lineCol]  (BR.north) -- (BSR.south);
          \draw[\lineCol]  (BR.north) -- (BJR.south);
          \draw[\lineCol]  (BSLS.north) -- (BSSLS.south);
          \draw[\lineCol]  (BSLS.north) -- (BJSLS.south);
          \draw[\lineCol]  (JS.north) -- (BJS.south);\textbf{}
          \draw[\lineCol]  (JS.north) -- (JSR.south);
          \draw[\lineCol]  (JS.north) -- (JSSLS.south);
          \draw[\lineCol]  (JR.north) -- (BJR.south);
          \draw[\lineCol]  (JR.north) -- (JSR.south);
          \draw[\lineCol]  (JSLS.north) -- (BJSLS.south);
          \draw[\lineCol]  (JSLS.north) -- (JSSLS.south);
          \draw[\lineCol]  (SR.north) -- (BSR.south);
          \draw[\lineCol]  (SR.north) -- (JSR.south);
          \draw[\lineCol]  (SSLS.north) -- (BSSLS.south);
          \draw[\lineCol]  (SSLS.north) -- (JSSLS.south);
          \draw[\lineCol]  (BSR.north) -- (BJSR.south);
          \draw[\lineCol]  (BSSLS.north) -- (BJSSLS.south);
          \draw[\lineCol]  (BJS.north) -- (BJSR.south);
          \draw[\lineCol]  (BJR.north) -- (BJSR.south);
          \draw[\lineCol]  (BJSLS.north) -- (BJSSLS.south);
          \draw[\lineCol]  (JSR.north) -- (BJSR.south);
          \draw[\lineCol]  (JSSLS.north) -- (BJSSLS.south);
          \draw[\lineCol]  (NS.north) -- (B.south);
          \draw[\lineCol]  (NS.north) -- (J.south);
          \draw[\lineCol]  (NS.north) -- (S.south);
          \draw[\lineCol]  (NS.north) -- (R.south);
          \draw[\lineCol]  (NS.north) -- (SLS.south);

          \fill[blue, opacity=0.2, rounded corners] (SLS.south west) -- (BSLS.south west) -- (BSLS.north west) -- (BJSLS.south west) -- (BJSLS.north west) -- (BJSSLS.south west) -- (BJSSLS.north west) -- (BJSSLS.north east) -- (JSSLS.north east) -- (SSLS.north east) -- (SSLS.south east) -- (SLS.north east)  -- (SLS.south east) -- (SLS.south west) -- (BSLS.south west);

          \node (no) at (-7, 4) {Compilers}; 
          \node (CFR-CRPLJ-CRPLR) [below = .1 of no] {$\complfenceR{\cdot}$, $\compretpR{\cdot}$, $\compretpJ{\cdot}$}; \node (CFS-CRPLFJ-CFB) [below = .8 of no] {$\complfenceB{\cdot}$, $\complfenceS{\cdot}$, $\compretpJF{\cdot}$}; \node (CFSLS) [below = 1.5 of no] {\qquad\qquad\quad$\complfenceSLS{\cdot}$}; \node (CUSLHB-CSLHB) [below = .8 of CFSLS] {$\compuslhB{\cdot}$, $\compsslhB{\cdot}$}; 

            \draw[->, dashed]  (CFSLS.east) -- (BJSSLS.-150);
            \draw[->, dashed]  (CFS-CRPLFJ-CFB.east) -- (BJSSLS.-170);
            \draw[->, dashed]  (CFS-CRPLFJ-CFB.east) -- (BJSR.-150);
            \draw[->, dashed]  (CFR-CRPLJ-CRPLR.east) -- (BJSR.-170);
            \draw[->, dashed]  (CUSLHB-CSLHB.east) -- (BSSLS.-150);
            \draw[->, dashed]  (CUSLHB-CSLHB.east) -- (BSR.-150);
          
    \end{tikzpicture}
    \caption[]{
    On the left is a visualization of \Cref{thm:lifted-guarantees-paper}, where the list of compilers is connected with a dashed line to the strongest semantic their security is lifted to.
    On the right, the blue area of the leakage ordering represents where  \Cref{cor:lift-comp-pres-paper} is applicable for $\complfenceSLS{\cdot}$.
    We can lift the security guarantees of $\complfenceSLS{\cdot}$  from the base semantics $\contract{}{\SLSv}$ to all composed semantics in the highlighted area.
    }
    \label{fig:instantiate-slsf}
\end{figure}
We remark that our lifting theorem allowed us to derive  strong compiler guarantees, i.e., \crssp{} w.r.t.  $\contract{}{\Bv + \Jv + \Sv + \Rv}$ or $\contract{}{\Bv+ \Jv + \Sv +\SLSv}$, \emph{without} requiring new secure compilation proofs, thereby significantly reducing the proving effort.
Note that carrying out secure compilation proofs can be complex because each proof requires setting up multiple cross-language relations (for states, values, actions, etc.)~\cite{akram-csf} as well as defining an invariant that holds for speculation (both in securely-compiled code and in attacker code)~\cite{S_sec_comp}.

As a concrete example, we use \Cref{fig:instantiate-slsf} to visualize (1) the effects of \Cref{thm:lifted-guarantees-paper} and (2), this lifting for the fence compiler $\complfenceSLS{\cdot}$. 
With respect to (1), the list of compilers on the left is connected to the strongest semantics to which they can be lifted.
With respect to (2), the blue shaded area depicts the semantics to which we can lift the security guarantees using \Cref{cor:lift-comp-pres-paper} together with Independence and Safe Nesting.
Our lifting approach allowed us to derive secure compilation results for the 7 semantics in the shaded area with only 1 secure compilation proof (for the base semantics $\contract{}{\SLSv}$), 1 simple proof of Trapped Speculation, and other 6 simple proofs of Syntactic Independence (by syntactic inspection).
Without our approach, proving the same results would have required 7 fully-fledged secure compilation proofs.

 \section{Discussion}\label{sec:limitations}

\paragraph*{Scope of the Security Analysis}
Lifting the results of our security analysis to real-world CPUs and compilers is only possible to the extent that our models faithfully represent the target systems.

In terms of CPUs and speculative leaks, any information flow in the target CPU that is not captured by our speculative semantics (and their combinations) might invalidate our security proofs in practice.
In particular, all speculative semantics from~\Cref{tab:semantics} consider a commonly-used attacker model~\cite{spectector, contracts, ST_binsec, ST_constantTime_Spec,guanciale2020inspectre, ST_jasmin2, S_sec_comp, ST_blade, spec_comb} that captures a cache-based attacker by exposing control-flow and memory accesses along non-speculative and speculative execution paths.
Any leaks not reflected in control-flow and memory accesses might, therefore, be missed by our models.
Regarding speculation, for the two new semantics modeling speculation over indirect jumps ($\contractSpec{}{\Jv}$) and straight-line speculation ($\contractSpec{}{\SLSv}$), the main simplification is in the modeling of speculative indirect jumps where we over-approximate the effect of prediction structures (e.g., BTB) by allowing any speculative target.
Note, however, that we only support valid instructions in the compiled program as targets of speculative jumps; speculatively jumping in the middle of program instructions is not captured by $\contractSpec{}{\Jv}$.
For the other semantics ($\contractSpec{}{\Bv}$, $\contractSpec{}{\Sv}$, and $\contractSpec{}{\Rv}$), we directly employ state-of-the-art models from prior work, whose limitations are discussed in~\cite{spectector,spec_comb}.

In terms of compilers, any divergence between our models in \Cref{ssec:compilers} and their actual implementations might, again, invalidate our results.
One important simplification of our SLH compilers $\compuslhB{\cdot}$ and $\compsslhB{\cdot}$ is that  the speculation flag is \emph{always} stored in a dedicated register.
In contrast, the actual SLH implementation in \textsc{Clang}~\cite{slh} uses a general purpose register for storing the speculation flag which, in some cases, might be spilled to memory.
This might result in unexpected leaks in case speculation over store-bypasses ($\contractSpec{}{\Sv}$) might result in loading a stale value for the speculation flag from memory.

\paragraph*{Using the Framework}
At this point, the reader may wonder how one can use this framework when the next version of Spectre comes out.
For example, recent work~\cite{revizor2} has discovered that CPUs speculate also on division operations (since this does not lead to actual attacks, we ignored this speculative semantics in our security analysis).
Let us indicate a semantics capturing leaks resulting from that speculation over divisions with $\semd$.
The ordering of \Cref{fig:instantiate} would have many new elements, including a top element $\contractSpec{}{\Bv + \Jv + \Sv + \Rv + \Div}$.
Since none of the considered compilers introduce a division operation, it would be sufficient to prove Syntactic Independence for them, reuse all the theorems from \Cref{ssec:sec-compilers,ssec:trapped-compilers} and then apply \Cref{cor:lift-comp-pres-paper} in order to obtain that those compilers are $\crssp$ for the new speculative semantics $\contractSpec{}{\Bv + \Jv + \Sv + \Rv + \Div}$.

\paragraph*{Beyond Speculative Leaks}
The framework presented in this paper focusses on attacks that arise from speculative executions, but it can be used to reason about the composition of other security properties such as memory safety ($\semms$)~\cite{ms} and cryptographic constant time ($\semct$, equivalent to our $\contract{}{NS}$ semantics)~\cite{kocher1996timing}.
In the case of $\semms$ and $\semct$, however, existing semantics that express these leaks can be composed in a much simpler way than speculative semantics, without any nesting.
Thus, proving that e.g., countermeasures for $\contractSpec{}{\MS}$ are $\crssp$ for $\contractSpec{}{\MS + \CT}$, only requires reasoning about Independence, because Safe Nesting trivially holds.
We leave reasoning about these results (and composing the resulting semantics with the presented ones) as future work.

\paragraph*{From Single to Multiple Compilers}
We believe this framework can be used beyond reasoning about a single compiler, and it can be employed in order to reason about the application of several Spectre countermeasures.
In fact, we \emph{speculate} that we can use existing results on composing secure compilers~\cite{csc22} to prove that if a compiler $\comp{\cdot}_1$ is $\crssp$ with respect to a semantics, and another compiler $\comp{\cdot}_2$ is $\crssp$ with respect to the same semantics, then the two can be composed ($\comp{\comp{\cdot}_2}_1$) and the result is $\crssp$ with respect to the same semantics.
The presented work can then be used to (1) first lift \emph{single} compiler countermeasures to their strongest semantics and then (2) compose those countermeasures to obtain that the composition is also $\crssp$ for the strongest semantics.
We leave investigating the theory of secure compilation applied to speculative semantics, as well as its application to the results of this paper via points (1) and (2) for future work.

 \section{Related Work}\label{sec:related-work}

\paragraph{Speculative Execution Attacks}

After Spectre~\cite{spectre} has been disclosed to the public in 2018, researchers have identified many other speculative execution attacks~\cite{spectreRsb,ret2spec,S_smotherSpectre,S_trans_troj,barberis2022branch, wikner_phantom_2023, psf}.
We refer the reader to~\citet{transientfail} for a survey of existing attacks.

\paragraph{Speculative Semantics}

There are many  semantics capturing the effects of speculatively executed instructions~\cite{ST_blade,guanciale2020inspectre, spectector, S_sec_comp, cats, serberus, ST_constantTime_Spec, spec_comb, ST_abs_sem, R_sem_model}.
These semantics differ in the level of microarchitectural details that are modelled (e.g., from program-level models~\cite{spectector} to those closer to simplified CPU designs~\cite{guanciale2020inspectre}) and the languages that are used (e.g., from models targeting While languages~\cite{S_sec_comp} to those targeting assembly-style languages~\cite{spectector}); see~\cite{sok:spectre_defense} for a survey of speculative semantics.\looseness=-1

As indicated in \Cref{tab:semantics}, our branch speculation semantics $\semb$ is from~\cite{spectector}, whereas our store-bypass speculation $\sems$ and return misprediction $\semr$ semantics are from~\cite{spec_comb}.
These semantics and our new semantics  $\semj$ and  $\semsls$ all  follow the \emph{always-mispredict} strategy~\cite{spectector}, which explore mispredicted paths for a fixed number of steps before continuing the architectural execution.

\paragraph{Security Properties for speculative leaks}
Researchers have proposed many program-level properties for security against speculative leaks, which can be classified in three main groups~\cite{sok:spectre_defense}:
\begin{asparaenum}
\item Non-interference definitions ensure the security of speculative \emph{and} non-speculative instructions. 
E.g., speculative constant-time~\cite{ST_constantTime_Spec} extends constant-time to transient instructions as well.\looseness=-1 

\item Relative non-interference definitions~\cite{spectector,ST_tpod,guanciale2020inspectre,ST_specusym, spec_dec} ensure that transient instructions do not leak more information than non-transient ones.
E.g., speculative non-interference~\cite{spectector}, which we inherit from the secure compilation framework we build on~\cite{S_sec_comp}, restricts the information leaked by speculatively executed instructions (without constraining what can be leaked non-speculatively).

\item Definitions that formalise security as a safety property~\cite{cats,S_sec_comp}, which may over-approximate definitions from the groups above.\looseness=-1
\end{asparaenum}

\paragraph{Secure compilation for speculative leaks}
Our secure compilation framework extends the work by~\citet{S_sec_comp} (which is restricted to branch speculation $\semb$) with support for new speculative semantics and their combinations~\cite{spec_comb}.
In particular, the \crssp{} secure compilation criterion from \Cref{def:crssp} is an extension of the \rssp{} criterion from~\cite{S_sec_comp}.

In~\Cref{sec:countermeasures}, we analyzed Spectre countermeasures implemented in mainstream compilers (or variants of them) or suggested by hardware vendors.
Next, we review further countermeasures.

\citet{ST_jasmin} and \citet{spec_ty_v1}  extend Jasmin \cite{ST_jasmin2, ST_jasmin_imp} to protect constant-time programs against leaks induced by branch speculation. 
In contrast, Blade~\cite{ST_blade} is a countermeasure against Spectre-PHT targeting Wasm \cite{wasm} which uses a flow-sensitive security-type system to minimize the amount of protect statements (either fences or SLH) needed to secure programs. 
Differently from our work, which targets speculative non-interference, these works target speculative constant-time.

Swivel \cite{swivel} is a compiler hardening pass for Wasm that protects against multiple Spectre attacks (Spectre-PHT, Spectre-BTB, and Spectre-RSB). 
However, it lacks a formal model and security proof.

In concurrent work, \citet{serberus} proposed Serberus, a set of compiler passes that---in combination with hardware support (e.g., Intel CET-IBT and a shadow stack for return addresses)---offer protection against Spectre-PHT, Spectre-BTB,  Spectre-RSB, Spectre-STL and predictive store forwarding~\cite{psf}.
For any whole program satisfying static constant time (a stricter variant of constant-time), Serberus ensures that its compiled counterpart is speculative constant-time.
Differently from our security analysis, where we lift secure compilation guarantees from weaker semantics to stronger combined semantics, Serberus' security proof directly targets a speculative semantics incorporating all supported speculation mechanisms.
Their proof already targets a semantics at the ``top'' of the leakage order and does not need lifting (until the discovery of new speculation mechanisms).\looseness=-1

\citet{switchpoline} proposes switchpoline,  an alternative to retpoline, to protect ARM cores from Spectre-BTB.
It transforms indirect calls into direct calls and uses a switch statement to select the correct call target. 
Interestingly, the authors argue about the importance of compatible countermeasures and ensure that switchpoline is fully compatible with other Spectre countermeasures, which is in line with our Independence property (\Cref{def:ind-paper}).

 \section{Conclusion}\label{sec:conclusion}

This paper presented a secure compilation framework for reasoning about the security against leaks introduced by different speculation mechanisms modeled as (combinations of) speculative semantics.
In particular, we developed a \emph{lifting theorem} that allows us to lift  a compiler's security guarantees from a weaker base speculative semantics to a stronger extended speculative semantics that accounts for more speculation mechanisms.
Additionally, we precisely characterized the  security guarantees provided by 9 Spectre-countermeasures implemented in mainstream compilers against 23 different speculative semantics covering combinations of 5 different speculation mechanisms.
Our lifting theorem was instrumental in allowing us to precisely characterize each countermeasure's guarantees against all combined semantics \emph{without} requiring additional secure compilation proofs (beyond the proof of security against each compiler's base semantics).

\bibliographystyle{ACM-Reference-Format}
\balance
\bibliography{References/references, References/spectre, References/spectreTools}


\begin{thebibliography}{56}


\ifx \showCODEN    \undefined \def \showCODEN     #1{\unskip}     \fi
\ifx \showDOI      \undefined \def \showDOI       #1{#1}\fi
\ifx \showISBNx    \undefined \def \showISBNx     #1{\unskip}     \fi
\ifx \showISBNxiii \undefined \def \showISBNxiii  #1{\unskip}     \fi
\ifx \showISSN     \undefined \def \showISSN      #1{\unskip}     \fi
\ifx \showLCCN     \undefined \def \showLCCN      #1{\unskip}     \fi
\ifx \shownote     \undefined \def \shownote      #1{#1}          \fi
\ifx \showarticletitle \undefined \def \showarticletitle #1{#1}   \fi
\ifx \showURL      \undefined \def \showURL       {\relax}        \fi
\providecommand\bibfield[2]{#2}
\providecommand\bibinfo[2]{#2}
\providecommand\natexlab[1]{#1}
\providecommand\showeprint[2][]{arXiv:#2}

\bibitem[Abadi(1999)]%
        {sec-typ-prot}
\bibfield{author}{\bibinfo{person}{Mart\'{\i}n Abadi}.} \bibinfo{year}{1999}\natexlab{}.
\newblock \showarticletitle{Secrecy by Typing in Security Protocols}.
\newblock \bibinfo{journal}{\emph{J. ACM}} (\bibinfo{year}{1999}).
\newblock


\bibitem[Abate et~al\mbox{.}(2019)]%
        {journey_sec_comp}
\bibfield{author}{\bibinfo{person}{Carmine Abate}, \bibinfo{person}{Roberto Blanco}, \bibinfo{person}{Deepak Garg}, \bibinfo{person}{Catalin Hritcu}, \bibinfo{person}{Marco Patrignani}, {and} \bibinfo{person}{Jérémy Thibault}.} \bibinfo{year}{2019}\natexlab{}.
\newblock \showarticletitle{Journey Beyond Full Abstraction: Exploring Robust Property Preservation for Secure Compilation}. In \bibinfo{booktitle}{\emph{Proceedings of the 32nd IEEE Computer Security Foundations Symposium}} \emph{(\bibinfo{series}{CSF '19})}. \bibinfo{publisher}{IEEE}.
\newblock


\bibitem[Almeida et~al\mbox{.}(2017)]%
        {ST_jasmin2}
\bibfield{author}{\bibinfo{person}{Jos\'{e}~Bacelar Almeida}, \bibinfo{person}{Manuel Barbosa}, \bibinfo{person}{Gilles Barthe}, \bibinfo{person}{Arthur Blot}, \bibinfo{person}{Benjamin Gr\'{e}goire}, \bibinfo{person}{Vincent Laporte}, \bibinfo{person}{Tiago Oliveira}, \bibinfo{person}{Hugo Pacheco}, \bibinfo{person}{Benedikt Schmidt}, {and} \bibinfo{person}{Pierre-Yves Strub}.} \bibinfo{year}{2017}\natexlab{}.
\newblock \showarticletitle{Jasmin: High-Assurance and High-Speed Cryptography}. In \bibinfo{booktitle}{\emph{Proceedings of the 24th ACM SIGSAC Conference on Computer and Communications Security}} \emph{(\bibinfo{series}{CCS '17})}. \bibinfo{publisher}{ACM}.
\newblock


\bibitem[Almeida et~al\mbox{.}(2020)]%
        {ST_jasmin_imp}
\bibfield{author}{\bibinfo{person}{Jos{\'e}~Bacelar Almeida}, \bibinfo{person}{Manuel Barbosa}, \bibinfo{person}{Gilles Barthe}, \bibinfo{person}{Benjamin Gr{\'e}goire}, \bibinfo{person}{Adrien Koutsos}, \bibinfo{person}{Vincent Laporte}, \bibinfo{person}{Tiago Oliveira}, {and} \bibinfo{person}{Pierre-Yves Strub}.} \bibinfo{year}{2020}\natexlab{}.
\newblock \showarticletitle{The last mile: High-assurance and high-speed cryptographic implementations}. In \bibinfo{booktitle}{\emph{2020 IEEE Symposium on Security and Privacy (SP)}} \emph{(\bibinfo{series}{S\&P '20})}. IEEE.
\newblock


\bibitem[AMD(2020)]%
        {sls-whitepaper2}
\bibfield{author}{\bibinfo{person}{AMD}.} \bibinfo{year}{2020}\natexlab{}.
\newblock \bibinfo{title}{Whitepaper Straight-line Speculation}.
\newblock \bibinfo{howpublished}{\url{https://www.amd.com/system/files/documents/technical-guidance-for-mitigating-branch-type-confusion.pdf}}.
\newblock


\bibitem[AMD(2021)]%
        {psf}
\bibfield{author}{\bibinfo{person}{AMD}.} \bibinfo{year}{2021}\natexlab{}.
\newblock \bibinfo{title}{Security analysis of AMD predictive store forwarding}.
\newblock \bibinfo{howpublished}{\url{https://www.amd.com/system/files/documents/security-analysis-predictive-store-forwarding.pdf}}.
\newblock
\newblock
\shownote{Accessed: 2024-03-11}.


\bibitem[ARM(2020)]%
        {sls-whitepaper}
\bibfield{author}{\bibinfo{person}{ARM}.} \bibinfo{year}{2020}\natexlab{}.
\newblock \bibinfo{title}{Whitepaper Straight-line Speculation}.
\newblock \bibinfo{howpublished}{\url{https://developer.arm.com/documentation/102825/0100/}}.
\newblock


\bibitem[ARM(2021)]%
        {arm-cet}
\bibfield{author}{\bibinfo{person}{ARM}.} \bibinfo{year}{2021}\natexlab{}.
\newblock \bibinfo{title}{Arm Armv9-A A64 Instruction Set Architecture}.
\newblock
\newblock
\urldef\tempurl%
\url{https://developer.arm.com/documentation/100076/0100/A64-Instruction-Set-Reference/A64-General-Instructions/BTI?}
\showURL{%
\tempurl}


\bibitem[Azevedo~de Amorim et~al\mbox{.}(2018)]%
        {ms}
\bibfield{author}{\bibinfo{person}{Arthur Azevedo~de Amorim}, \bibinfo{person}{C{\u{a}}t{\u{a}}lin Hri{\c{t}}cu}, {and} \bibinfo{person}{Benjamin~C. Pierce}.} \bibinfo{year}{2018}\natexlab{}.
\newblock \showarticletitle{The Meaning of Memory Safety}. In \bibinfo{booktitle}{\emph{Principles of Security and Trust}} \emph{(\bibinfo{series}{POST '18})}, \bibfield{editor}{\bibinfo{person}{Lujo Bauer} {and} \bibinfo{person}{Ralf K{\"u}sters}} (Eds.). \bibinfo{publisher}{Springer}.
\newblock


\bibitem[Barberis et~al\mbox{.}(2022)]%
        {barberis2022branch}
\bibfield{author}{\bibinfo{person}{Enrico Barberis}, \bibinfo{person}{Pietro Frigo}, \bibinfo{person}{Marius Muench}, \bibinfo{person}{Herbert Bos}, {and} \bibinfo{person}{Cristiano Giuffrida}.} \bibinfo{year}{2022}\natexlab{}.
\newblock \showarticletitle{Branch History Injection: On the Effectiveness of Hardware Mitigations Against {Cross-Privilege} Spectre-v2 Attacks}. In \bibinfo{booktitle}{\emph{Proceedings of the 31st USENIX Security Symposium}} \emph{(\bibinfo{series}{USENIX Security '22})}. \bibinfo{publisher}{USENIX Association}.
\newblock


\bibitem[Barthe et~al\mbox{.}(2021)]%
        {ST_jasmin}
\bibfield{author}{\bibinfo{person}{Gilles Barthe}, \bibinfo{person}{Sunjay Cauligi}, \bibinfo{person}{Benjamin Gr{\'e}goire}, \bibinfo{person}{Adrien Koutsos}, \bibinfo{person}{Kevin Liao}, \bibinfo{person}{Tiago Oliveira}, \bibinfo{person}{Swarn Priya}, \bibinfo{person}{Tamara Rezk}, {and} \bibinfo{person}{Peter Schwabe}.} \bibinfo{year}{2021}\natexlab{}.
\newblock \showarticletitle{High-Assurance Cryptography in the Spectre Era}. In \bibinfo{booktitle}{\emph{Proceedings of the 42nd IEEE Symposium on Security and Privacy}} \emph{(\bibinfo{series}{S\&P '21})}. \bibinfo{publisher}{IEEE}.
\newblock


\bibitem[Bauer et~al\mbox{.}(2024)]%
        {switchpoline}
\bibfield{author}{\bibinfo{person}{Markus Bauer}, \bibinfo{person}{Lorenz Hetterich}, \bibinfo{person}{Michael Schwarz}, {and} \bibinfo{person}{Christian Rossow}.} \bibinfo{year}{2024}\natexlab{}.
\newblock \showarticletitle{{Switchpoline: A Software Mitigation for Spectre-BTB and Spectre-BHB on ARMv8}}. In \bibinfo{booktitle}{\emph{AsiaCCS}}. \bibinfo{pages}{to appear}.
\newblock


\bibitem[Bhattacharyya et~al\mbox{.}(2019)]%
        {S_smotherSpectre}
\bibfield{author}{\bibinfo{person}{Atri Bhattacharyya}, \bibinfo{person}{Alexandra Sandulescu}, \bibinfo{person}{Matthias Neugschwandtner}, \bibinfo{person}{Alessandro Sorniotti}, \bibinfo{person}{Babak Falsafi}, \bibinfo{person}{Mathias Payer}, {and} \bibinfo{person}{Anil Kurmus}.} \bibinfo{year}{2019}\natexlab{}.
\newblock \showarticletitle{SMoTherSpectre: Exploiting Speculative Execution through Port Contention}. In \bibinfo{booktitle}{\emph{Proceedings of the 26th ACM SIGSAC Conference on Computer and Communications Security}} \emph{(\bibinfo{series}{CCS '19})}. \bibinfo{publisher}{ACM}.
\newblock


\bibitem[Canella et~al\mbox{.}(2019)]%
        {transientfail}
\bibfield{author}{\bibinfo{person}{Claudio Canella}, \bibinfo{person}{Jo~Van Bulck}, \bibinfo{person}{Michael Schwarz}, \bibinfo{person}{Moritz Lipp}, \bibinfo{person}{Benjamin von Berg}, \bibinfo{person}{Philipp Ortner}, \bibinfo{person}{Frank Piessens}, \bibinfo{person}{Dmitry Evtyushkin}, {and} \bibinfo{person}{Daniel Gruss}.} \bibinfo{year}{2019}\natexlab{}.
\newblock \showarticletitle{A Systematic Evaluation of Transient Execution Attacks and Defenses}. In \bibinfo{booktitle}{\emph{Proceedings of the 28th {USENIX} Security Symposium}} \emph{(\bibinfo{series}{{USENIX} Security '19})}. \bibinfo{publisher}{{USENIX} Association}.
\newblock


\bibitem[Canella et~al\mbox{.}(2020)]%
        {sok_countermeasures}
\bibfield{author}{\bibinfo{person}{Claudio Canella}, \bibinfo{person}{Sai~Manoj Pudukotai~Dinakarrao}, \bibinfo{person}{Daniel Gruss}, {and} \bibinfo{person}{Khaled~N. Khasawneh}.} \bibinfo{year}{2020}\natexlab{}.
\newblock \showarticletitle{Evolution of Defenses against Transient-Execution Attacks}. In \bibinfo{booktitle}{\emph{Proceedings of the 2020 on Great Lakes Symposium on VLSI}} \emph{(\bibinfo{series}{GLSVLSI '20})}. \bibinfo{publisher}{ACM}.
\newblock


\bibitem[Carruth(2018)]%
        {slh}
\bibfield{author}{\bibinfo{person}{Chandler Carruth}.} \bibinfo{year}{2018}\natexlab{}.
\newblock \bibinfo{title}{Speculative Load Hardening}.
\newblock
\newblock
\urldef\tempurl%
\url{https://llvm.org/docs/SpeculativeLoadHardening.html}
\showURL{%
\tempurl}


\bibitem[Cauligi et~al\mbox{.}(2020)]%
        {ST_constantTime_Spec}
\bibfield{author}{\bibinfo{person}{Sunjay Cauligi}, \bibinfo{person}{Craig Disselkoen}, \bibinfo{person}{Klaus~v. Gleissenthall}, \bibinfo{person}{Dean Tullsen}, \bibinfo{person}{Deian Stefan}, \bibinfo{person}{Tamara Rezk}, {and} \bibinfo{person}{Gilles Barthe}.} \bibinfo{year}{2020}\natexlab{}.
\newblock \showarticletitle{Constant-Time Foundations for the New Spectre Era}. In \bibinfo{booktitle}{\emph{Proceedings of the 41st ACM SIGPLAN Conference on Programming Language Design and Implementation}} \emph{(\bibinfo{series}{PLDI '20})}. \bibinfo{publisher}{ACM}.
\newblock


\bibitem[Cauligi et~al\mbox{.}(2022)]%
        {sok:spectre_defense}
\bibfield{author}{\bibinfo{person}{Sunjay Cauligi}, \bibinfo{person}{Craig Disselkoen}, \bibinfo{person}{Daniel Moghimi}, \bibinfo{person}{Gilles Barthe}, {and} \bibinfo{person}{Deian Stefan}.} \bibinfo{year}{2022}\natexlab{}.
\newblock \showarticletitle{SoK: Practical Foundations for Software Spectre Defenses}. In \bibinfo{booktitle}{\emph{Proceedings of the 43rd IEEE Symposium on Security and Privacy}} \emph{(\bibinfo{series}{S\&P '22})}. \bibinfo{publisher}{IEEE}.
\newblock


\bibitem[Cheang et~al\mbox{.}(2019)]%
        {ST_tpod}
\bibfield{author}{\bibinfo{person}{Kevin Cheang}, \bibinfo{person}{Cameron Rasmussen}, \bibinfo{person}{Sanjit Seshia}, {and} \bibinfo{person}{Pramod Subramanyan}.} \bibinfo{year}{2019}\natexlab{}.
\newblock \showarticletitle{A Formal Approach to Secure Speculation}. In \bibinfo{booktitle}{\emph{Proceedings of the 32nd IEEE Computer Security Foundations Symposium}} \emph{(\bibinfo{series}{CSF '19})}. \bibinfo{publisher}{IEEE}.
\newblock


\bibitem[Colvin and Winter(2019)]%
        {ST_abs_sem}
\bibfield{author}{\bibinfo{person}{Robert~J. Colvin} {and} \bibinfo{person}{Kirsten Winter}.} \bibinfo{year}{2019}\natexlab{}.
\newblock \showarticletitle{An Abstract Semantics of Speculative Execution for Reasoning About Security Vulnerabilities}. In \bibinfo{booktitle}{\emph{Proceedings of the 19th Refinement Workshop}} \emph{(\bibinfo{series}{Refine '19})}. \bibinfo{publisher}{Springer}.
\newblock


\bibitem[Daniel et~al\mbox{.}(2021)]%
        {ST_binsec}
\bibfield{author}{\bibinfo{person}{Lesly-Ann Daniel}, \bibinfo{person}{S{\'e}bastien Bardin}, {and} \bibinfo{person}{Tamara Rezk}.} \bibinfo{year}{2021}\natexlab{}.
\newblock \showarticletitle{Hunting the Haunter — Efficient relational symbolic execution for Spectre with {Haunted} {RelSE}}. In \bibinfo{booktitle}{\emph{Proceedings of the 28th Annual Network and Distributed System Security Symposium}} \emph{(\bibinfo{series}{NDSS '21})}. \bibinfo{publisher}{The Internet Society}.
\newblock


\bibitem[Disselkoen et~al\mbox{.}(2019)]%
        {R_sem_model}
\bibfield{author}{\bibinfo{person}{Craig Disselkoen}, \bibinfo{person}{Radha Jagadeesan}, \bibinfo{person}{Alan Jeffrey}, {and} \bibinfo{person}{James Riely}.} \bibinfo{year}{2019}\natexlab{}.
\newblock \showarticletitle{The Code That Never Ran: Modeling Attacks on Speculative Evaluation}. In \bibinfo{booktitle}{\emph{Proceedings of the 40th IEEE Symposium on Security and Privacy}} \emph{(\bibinfo{series}{S\&P '19})}. \bibinfo{publisher}{IEEE}.
\newblock


\bibitem[El-Korashy et~al\mbox{.}(2022)]%
        {akram-csf}
\bibfield{author}{\bibinfo{person}{Akram El-Korashy}, \bibinfo{person}{Roberto Blanco}, \bibinfo{person}{Jeremy Thibault}, \bibinfo{person}{Adrien Durier}, \bibinfo{person}{Deepak Garg}, {and} \bibinfo{person}{Catalin Hritcu}.} \bibinfo{year}{2022}\natexlab{}.
\newblock \showarticletitle{SecurePtrs: Proving Secure Compilation with Data-Flow Back-Translation and Turn-Taking Simulation}. In \bibinfo{booktitle}{\emph{Proceedings of the 35th IEEE Computer Security Foundations Symposium}} \emph{(\bibinfo{series}{CSF '22})}. \bibinfo{publisher}{IEEE}.
\newblock


\bibitem[Fabian et~al\mbox{.}(2022)]%
        {spec_comb}
\bibfield{author}{\bibinfo{person}{Xaver Fabian}, \bibinfo{person}{Marco Guarnieri}, {and} \bibinfo{person}{Marco Patrignani}.} \bibinfo{year}{2022}\natexlab{}.
\newblock \showarticletitle{Automatic Detection of Speculative Execution Combinations}. In \bibinfo{booktitle}{\emph{Proceedings of the 29th ACM SIGSAC Conference on Computer and Communications Security}} \emph{(\bibinfo{series}{CCS '22})}. \bibinfo{publisher}{ACM}.
\newblock


\bibitem[Fournet et~al\mbox{.}(2007)]%
        {tydisa}
\bibfield{author}{\bibinfo{person}{C{\'e}dric Fournet}, \bibinfo{person}{Andrew~D. Gordon}, {and} \bibinfo{person}{Sergio Maffeis}.} \bibinfo{year}{2007}\natexlab{}.
\newblock \showarticletitle{A Type Discipline for Authorization Policies}.
\newblock \bibinfo{journal}{\emph{ACM Trans. Program. Lang. Syst.}}  \bibinfo{volume}{29} (\bibinfo{year}{2007}).
\newblock


\bibitem[Gordon and Jeffrey(2003)]%
        {autysec}
\bibfield{author}{\bibinfo{person}{Andrew~D. Gordon} {and} \bibinfo{person}{Alan Jeffrey}.} \bibinfo{year}{2003}\natexlab{}.
\newblock \showarticletitle{Authenticity by Typing for Security Protocols}.
\newblock \bibinfo{journal}{\emph{J. Comput. Secur.}}  \bibinfo{volume}{11} (\bibinfo{year}{2003}).
\newblock


\bibitem[Guanciale et~al\mbox{.}(2020)]%
        {guanciale2020inspectre}
\bibfield{author}{\bibinfo{person}{Roberto Guanciale}, \bibinfo{person}{Musard Balliu}, {and} \bibinfo{person}{Mads Dam}.} \bibinfo{year}{2020}\natexlab{}.
\newblock \showarticletitle{Inspectre: Breaking and fixing microarchitectural vulnerabilities by formal analysis}. In \bibinfo{booktitle}{\emph{Proceedings of the 2020 ACM SIGSAC Conference on Computer and Communications Security}}. \bibinfo{pages}{1853--1869}.
\newblock


\bibitem[Guarnieri et~al\mbox{.}(2021)]%
        {contracts}
\bibfield{author}{\bibinfo{person}{Marco Guarnieri}, \bibinfo{person}{Boris K{\"o}pf}, \bibinfo{person}{Jan Reineke}, {and} \bibinfo{person}{Pepe Vila}.} \bibinfo{year}{2021}\natexlab{}.
\newblock \showarticletitle{Hardware-software contracts for secure speculation}. In \bibinfo{booktitle}{\emph{Proceedings of the 42nd IEEE Symposium on Security and Privacy}} \emph{(\bibinfo{series}{S\&P '21})}. \bibinfo{publisher}{IEEE}.
\newblock


\bibitem[{Guarnieri} et~al\mbox{.}(2020)]%
        {spectector}
\bibfield{author}{\bibinfo{person}{Marco {Guarnieri}}, \bibinfo{person}{Boris {Köpf}}, \bibinfo{person}{Jos\'e~F. {Morales}}, \bibinfo{person}{Jan {Reineke}}, {and} \bibinfo{person}{Andr\'es {S\'anchez}}.} \bibinfo{year}{2020}\natexlab{}.
\newblock \showarticletitle{Spectector: Principled Detection of Speculative Information Flows}. In \bibinfo{booktitle}{\emph{Proceedings of the 41st IEEE Symposium on Security and Privacy}} \emph{(\bibinfo{series}{S\&P '20})}.
\newblock


\bibitem[Guo et~al\mbox{.}(2020)]%
        {ST_specusym}
\bibfield{author}{\bibinfo{person}{Shengjian Guo}, \bibinfo{person}{Yueqi Chen}, \bibinfo{person}{Peng Li}, \bibinfo{person}{Yueqiang Cheng}, \bibinfo{person}{Huibo Wang}, \bibinfo{person}{Meng Wu}, {and} \bibinfo{person}{Zhiqiang Zuo}.} \bibinfo{year}{2020}\natexlab{}.
\newblock \showarticletitle{SpecuSym: Speculative Symbolic Execution for Cache Timing Leak Detection}. In \bibinfo{booktitle}{\emph{Proceedings of the 42nd ACM/IEEE International Conference on Software Engineering}} \emph{(\bibinfo{series}{ICSE '20})}. \bibinfo{publisher}{ACM}.
\newblock


\bibitem[Haas et~al\mbox{.}(2017)]%
        {wasm}
\bibfield{author}{\bibinfo{person}{Andreas Haas}, \bibinfo{person}{Andreas Rossberg}, \bibinfo{person}{Derek~L. Schuff}, \bibinfo{person}{Ben~L. Titzer}, \bibinfo{person}{Michael Holman}, \bibinfo{person}{Dan Gohman}, \bibinfo{person}{Luke Wagner}, \bibinfo{person}{Alon Zakai}, {and} \bibinfo{person}{JF Bastien}.} \bibinfo{year}{2017}\natexlab{}.
\newblock \showarticletitle{Bringing the web up to speed with WebAssembly}. In \bibinfo{booktitle}{\emph{Proceedings of the 38th ACM SIGPLAN Conference on Programming Language Design and Implementation}} \emph{(\bibinfo{series}{PLDI '17})}. \bibinfo{publisher}{ACM}.
\newblock


\bibitem[Horn(2018)]%
        {S_specv4}
\bibfield{author}{\bibinfo{person}{J. Horn}.} \bibinfo{year}{2018}\natexlab{}.
\newblock \bibinfo{title}{Speculative execution, variant 4: Speculative store bypass}.
\newblock \bibinfo{howpublished}{\url{https://bugs.chromium.org/p/project-zero/issues/detail?id=1528}}.
\newblock
\newblock
\shownote{Accessed: 2021-04-11}.


\bibitem[Intel(2018a)]%
        {retpoline}
\bibfield{author}{\bibinfo{person}{Intel}.} \bibinfo{year}{2018}\natexlab{a}.
\newblock \bibinfo{title}{Retpoline: A Branch Target Injection Mitigation}.
\newblock
\newblock
\urldef\tempurl%
\url{https://www.intel.com/content/dam/develop/external/us/en/documents/retpoline-a-branch-target-injection-mitigation.pdf}
\showURL{%
\tempurl}


\bibitem[Intel(2018b)]%
        {v4_fence}
\bibfield{author}{\bibinfo{person}{Intel}.} \bibinfo{year}{2018}\natexlab{b}.
\newblock \bibinfo{title}{Speculative Store Bypass}.
\newblock
\newblock
\urldef\tempurl%
\url{https://www.intel.com/content/www/us/en/developer/articles/technical/software-security-guidance/advisory-guidance/speculative-store-bypass.html}
\showURL{%
\tempurl}


\bibitem[Intel(2018c)]%
        {Intel-compiler}
\bibfield{author}{\bibinfo{person}{Intel}.} \bibinfo{year}{2018}\natexlab{c}.
\newblock \bibinfo{title}{Using Intel Compilers to Mitigate Speculative Execution Side-Channel Issues}.
\newblock
\newblock
\urldef\tempurl%
\url{https://software.intel.com/en-us/articles/using-intel-compilers-to-mitigate-speculative-execution-side-channel-issues}
\showURL{%
\tempurl}


\bibitem[Kocher et~al\mbox{.}(2019)]%
        {spectre}
\bibfield{author}{\bibinfo{person}{Paul Kocher}, \bibinfo{person}{Jann Horn}, \bibinfo{person}{Anders Fogh}, \bibinfo{person}{Daniel Genkin}, \bibinfo{person}{Daniel Gruss}, \bibinfo{person}{Werner Haas}, \bibinfo{person}{Mike Hamburg}, \bibinfo{person}{Moritz Lipp}, \bibinfo{person}{Stefan Mangard}, \bibinfo{person}{Thomas Prescher}, \bibinfo{person}{Michael Schwarz}, {and} \bibinfo{person}{Yuval Yarom}.} \bibinfo{year}{2019}\natexlab{}.
\newblock \showarticletitle{Spectre Attacks: Exploiting Speculative Execution}. In \bibinfo{booktitle}{\emph{Proceedings of the 40th IEEE Symposium on Security and Privacy}} \emph{(\bibinfo{series}{S\&P '19})}.
\newblock


\bibitem[Kocher(1996)]%
        {kocher1996timing}
\bibfield{author}{\bibinfo{person}{Paul~C. Kocher}.} \bibinfo{year}{1996}\natexlab{}.
\newblock \showarticletitle{Timing Attacks on Implementations of Diffie-Hellman, RSA, DSS, and Other Systems}. In \bibinfo{booktitle}{\emph{Advances in Cryptology}} \emph{(\bibinfo{series}{CRYPTO '96})}, \bibfield{editor}{\bibinfo{person}{Neal Koblitz}} (Ed.). \bibinfo{publisher}{Springer Berlin Heidelberg}.
\newblock


\bibitem[Koruyeh et~al\mbox{.}(2018)]%
        {spectreRsb}
\bibfield{author}{\bibinfo{person}{Esmaeil~Mohammadian Koruyeh}, \bibinfo{person}{Khaled~N. Khasawneh}, \bibinfo{person}{Chengyu Song}, {and} \bibinfo{person}{Nael Abu-Ghazaleh}.} \bibinfo{year}{2018}\natexlab{}.
\newblock \showarticletitle{Spectre Returns! Speculation Attacks Using the Return Stack Buffer}. In \bibinfo{booktitle}{\emph{Proceedings of the 12th USENIX Workshop on Offensive Technologies}} \emph{(\bibinfo{series}{WOOT '18})}. \bibinfo{publisher}{USENIX Association}.
\newblock


\bibitem[Kruse and Patrignani(2022)]%
        {csc22}
\bibfield{author}{\bibinfo{person}{Matthis Kruse} {and} \bibinfo{person}{Marco Patrignani}.} \bibinfo{year}{2022}\natexlab{}.
\newblock \showarticletitle{Composing Secure Compilers}.
\newblock


\bibitem[Maisuradze and Rossow(2018)]%
        {ret2spec}
\bibfield{author}{\bibinfo{person}{Giorgi Maisuradze} {and} \bibinfo{person}{Christian Rossow}.} \bibinfo{year}{2018}\natexlab{}.
\newblock \showarticletitle{Ret2spec: Speculative Execution Using Return Stack Buffers}. In \bibinfo{booktitle}{\emph{Proceedings of the 25th ACM SIGSAC Conference on Computer and Communications Security}} \emph{(\bibinfo{series}{CCS '18})}. \bibinfo{publisher}{ACM}.
\newblock


\bibitem[Mosier et~al\mbox{.}(2024)]%
        {serberus}
\bibfield{author}{\bibinfo{person}{N. Mosier}, \bibinfo{person}{H. Nemati}, \bibinfo{person}{J.~C. Mitchell}, {and} \bibinfo{person}{C. Trippel}.} \bibinfo{year}{2024}\natexlab{}.
\newblock \showarticletitle{Serberus: Protecting Cryptographic Code from Spectres at Compile-Time}. In \bibinfo{booktitle}{\emph{Proceedings of the 45th IEEE Symposium on Security and Privacy}} \emph{(\bibinfo{series}{S\&P '24})}. \bibinfo{publisher}{IEEE}.
\newblock


\bibitem[Narayan et~al\mbox{.}(2021)]%
        {swivel}
\bibfield{author}{\bibinfo{person}{Shravan Narayan}, \bibinfo{person}{Craig Disselkoen}, \bibinfo{person}{Daniel Moghimi}, \bibinfo{person}{Sunjay Cauligi}, \bibinfo{person}{Evan Johnson}, \bibinfo{person}{Zhao Gang}, \bibinfo{person}{Anjo Vahldiek-Oberwagner}, \bibinfo{person}{Ravi Sahita}, \bibinfo{person}{Hovav Shacham}, \bibinfo{person}{Dean Tullsen}, {and} \bibinfo{person}{Deian Stefan}.} \bibinfo{year}{2021}\natexlab{}.
\newblock \showarticletitle{Swivel: Hardening {WebAssembly} against Spectre}. In \bibinfo{booktitle}{\emph{Proceedings of the 30th USENIX Security Symposium}} \emph{(\bibinfo{series}{USENIX Security '21})}. \bibinfo{publisher}{USENIX Association}.
\newblock


\bibitem[Oleksenko et~al\mbox{.}(2023)]%
        {revizor2}
\bibfield{author}{\bibinfo{person}{Oleksii Oleksenko}, \bibinfo{person}{Marco Guarnieri}, \bibinfo{person}{Boris K\"{o}pf}, {and} \bibinfo{person}{Mark Silberstein}.} \bibinfo{year}{2023}\natexlab{}.
\newblock \showarticletitle{{Hide and Seek with Spectres: Efficient discovery of speculative information leaks with random testing}}. In \bibinfo{booktitle}{\emph{Proceedings of the 44th IEEE Symposium on Security and Privacy}} \emph{(\bibinfo{series}{S\&P 2023})}. \bibinfo{publisher}{IEEE}.
\newblock


\bibitem[Patrignani(2020)]%
        {patrignani2020should}
\bibfield{author}{\bibinfo{person}{Marco Patrignani}.} \bibinfo{year}{2020}\natexlab{}.
\newblock \showarticletitle{Why should anyone use colours? or, syntax highlighting beyond code snippets}.
\newblock \bibinfo{journal}{\emph{arXiv preprint arXiv:2001.11334}} (\bibinfo{year}{2020}).
\newblock


\bibitem[Patrignani and Blackshear(2023)]%
        {rsmove}
\bibfield{author}{\bibinfo{person}{Marco Patrignani} {and} \bibinfo{person}{Sam Blackshear}.} \bibinfo{year}{2023}\natexlab{}.
\newblock \showarticletitle{Robust Safety for Move}. In \bibinfo{booktitle}{\emph{Proceedings of the 36th IEEE Computer Security Foundations Symposium}} \emph{(\bibinfo{series}{CSF '23})}. \bibinfo{publisher}{IEEE}.
\newblock


\bibitem[Patrignani and Guarnieri(2021)]%
        {S_sec_comp}
\bibfield{author}{\bibinfo{person}{Marco Patrignani} {and} \bibinfo{person}{Marco Guarnieri}.} \bibinfo{year}{2021}\natexlab{}.
\newblock \showarticletitle{Exorcising Spectres with Secure Compilers}. In \bibinfo{booktitle}{\emph{Proceedings of the 28th ACM Conference on Computer and Communications Security}} \emph{(\bibinfo{series}{CCS '21})}. \bibinfo{publisher}{ACM}.
\newblock


\bibitem[Ponce~de Le{\'o}n and Kinder(2022)]%
        {cats}
\bibfield{author}{\bibinfo{person}{Hern{\'a}n Ponce~de Le{\'o}n} {and} \bibinfo{person}{Johannes Kinder}.} \bibinfo{year}{2022}\natexlab{}.
\newblock \showarticletitle{Cats vs. Spectre: An Axiomatic Approach to Modeling Speculative Execution Attacks}. In \bibinfo{booktitle}{\emph{Proceedings of the 43rd IEEE Symposium on Security and Privacy}} \emph{(\bibinfo{series}{S\&P '22})}. \bibinfo{publisher}{IEEE}.
\newblock


\bibitem[Sammler et~al\mbox{.}(2020)]%
        {dg-rs}
\bibfield{author}{\bibinfo{person}{Michael Sammler}, \bibinfo{person}{Deepak Garg}, \bibinfo{person}{Derek Dreyer}, {and} \bibinfo{person}{Tadeusz Litak}.} \bibinfo{year}{2020}\natexlab{}.
\newblock \showarticletitle{The high-level benefits of low-level sandboxing}.
\newblock \bibinfo{journal}{\emph{Proc. ACM Program. Lang.}} (\bibinfo{year}{2020}).
\newblock


\bibitem[Shanbhogue et~al\mbox{.}(2019)]%
        {intel_cet}
\bibfield{author}{\bibinfo{person}{Vedvyas Shanbhogue}, \bibinfo{person}{Deepak Gupta}, {and} \bibinfo{person}{Ravi Sahita}.} \bibinfo{year}{2019}\natexlab{}.
\newblock \showarticletitle{Security Analysis of Processor Instruction Set Architecture for Enforcing Control-Flow Integrity}. In \bibinfo{booktitle}{\emph{Proceedings of the 8th International Workshop on Hardware and Architectural Support for Security and Privacy}} \emph{(\bibinfo{series}{HASP '19})}. \bibinfo{publisher}{ACM}.
\newblock
\showISBNx{9781450372268}


\bibitem[Shivakumar et~al\mbox{.}(2023b)]%
        {spec_ty_v1}
\bibfield{author}{\bibinfo{person}{B. Shivakumar}, \bibinfo{person}{G. Barthe}, \bibinfo{person}{B. Gregoire}, \bibinfo{person}{V. Laporte}, \bibinfo{person}{T. Oliveira}, \bibinfo{person}{S. Priya}, \bibinfo{person}{P. Schwabe}, {and} \bibinfo{person}{L. Tabary-Maujean}.} \bibinfo{year}{2023}\natexlab{b}.
\newblock \showarticletitle{Typing High-Speed Cryptography against Spectre v1}. In \bibinfo{booktitle}{\emph{Proceedings of the 44th IEEE Symposium on Security and Privacy}} \emph{(\bibinfo{series}{S\&P '23})}. \bibinfo{publisher}{IEEE}.
\newblock


\bibitem[Shivakumar et~al\mbox{.}(2023a)]%
        {spec_dec}
\bibfield{author}{\bibinfo{person}{Basavesh~Ammanaghatta Shivakumar}, \bibinfo{person}{Jack Barnes}, \bibinfo{person}{Gilles Barthe}, \bibinfo{person}{Sunjay Cauligi}, \bibinfo{person}{Chitchanok Chuengsatiansup}, \bibinfo{person}{Daniel Genkin}, \bibinfo{person}{Sioli O’Connell}, \bibinfo{person}{Peter Schwabe}, \bibinfo{person}{Rui~Qi Sim}, {and} \bibinfo{person}{Yuval Yarom}.} \bibinfo{year}{2023}\natexlab{a}.
\newblock \showarticletitle{Spectre Declassified: Reading from the Right Place at the Wrong Time}. In \bibinfo{booktitle}{\emph{Proceedings of the 44th IEEE Symposium on Security and Privacy}} \emph{(\bibinfo{series}{S\&P '23})}. \bibinfo{publisher}{IEEE}.
\newblock


\bibitem[Swasey et~al\mbox{.}(2017)]%
        {davidcaps}
\bibfield{author}{\bibinfo{person}{David Swasey}, \bibinfo{person}{Deepak Garg}, {and} \bibinfo{person}{Derek Dreyer}.} \bibinfo{year}{2017}\natexlab{}.
\newblock \showarticletitle{Robust and Compositional Verification of Object Capability Patterns}. In \bibinfo{booktitle}{\emph{Proceedings of the 2017 {ACM} {SIGPLAN} International Conference on Object-Oriented Programming, Systems, Languages, and Applications.}} \emph{(\bibinfo{series}{OOPSLA '17})}. \bibinfo{publisher}{ACM}.
\newblock


\bibitem[Vassena et~al\mbox{.}(2021)]%
        {ST_blade}
\bibfield{author}{\bibinfo{person}{Marco Vassena}, \bibinfo{person}{Craig Disselkoen}, \bibinfo{person}{Klaus~von Gleissenthall}, \bibinfo{person}{Sunjay Cauligi}, \bibinfo{person}{Rami~G\"{o}khan K\i{}c\i{}}, \bibinfo{person}{Ranjit Jhala}, \bibinfo{person}{Dean Tullsen}, {and} \bibinfo{person}{Deian Stefan}.} \bibinfo{year}{2021}\natexlab{}.
\newblock \showarticletitle{Automatically Eliminating Speculative Leaks from Cryptographic Code with Blade}.
\newblock \bibinfo{journal}{\emph{Proc. ACM Program. Lang.}} (\bibinfo{year}{2021}).
\newblock


\bibitem[Wikner et~al\mbox{.}(2023)]%
        {wikner_phantom_2023}
\bibfield{author}{\bibinfo{person}{Johannes Wikner}, \bibinfo{person}{Dani\"{e}l Trujillo}, {and} \bibinfo{person}{Kaveh Razavi}.} \bibinfo{year}{2023}\natexlab{}.
\newblock \showarticletitle{Phantom: Exploiting Decoder-detectable Mispredictions}. In \bibinfo{booktitle}{\emph{Proceedings of the 56th Annual IEEE/ACM International Symposium on Microarchitecture}} \emph{(\bibinfo{series}{MICRO '23})}. \bibinfo{publisher}{ACM}.
\newblock


\bibitem[Zhang et~al\mbox{.}(2020)]%
        {S_trans_troj}
\bibfield{author}{\bibinfo{person}{Tao Zhang}, \bibinfo{person}{Kenneth Koltermann}, {and} \bibinfo{person}{Dmitry Evtyushkin}.} \bibinfo{year}{2020}\natexlab{}.
\newblock \showarticletitle{Exploring Branch Predictors for Constructing Transient Execution Trojans}. In \bibinfo{booktitle}{\emph{Proceedings of the 25th International Conference on Architectural Support for Programming Languages and Operating Systems}} \emph{(\bibinfo{series}{ASPLOS '20})}. \bibinfo{publisher}{ACM}.
\newblock


\bibitem[Zhang et~al\mbox{.}(2023)]%
        {uslh}
\bibfield{author}{\bibinfo{person}{Zhiyuan Zhang}, \bibinfo{person}{Gilles Barthe}, \bibinfo{person}{Chitchanok Chuengsatiansup}, \bibinfo{person}{Peter Schwabe}, {and} \bibinfo{person}{Yuval Yarom}.} \bibinfo{year}{2023}\natexlab{}.
\newblock \showarticletitle{Ultimate SLH: taking speculative load hardening to the next level}. In \bibinfo{booktitle}{\emph{Proceedings of the 32nd USENIX Conference on Security Symposium}} \emph{(\bibinfo{series}{USENIX Security '23})}. \bibinfo{publisher}{USENIX Association}.
\newblock


\end{thebibliography}

\clearpage
\newpage
\onecolumn

\section{Disclaimer}

While we presented the semantics using a separate $\startObs{}$ observation, we do not have this observation in our semantics. Mainly, we can always reconstruct the start of a speculative transaction even without a concrete $\startObs{}$.
For example, we could annotate the observation that is made when speculation is started with a label, telling us that it started speculation.
An actual $\startObs{}$ complicates the semantics and the reasoning about it.

Furthermore, while we use $\mathghost$ to indicate the semantics, we use $\contract{}{x}$ here.

\setlist[description]{style=nextline}
\appendix

\section{Language Definition}

\begin{gather*}
\begin{aligned}
	\mi{Programs}~ \com{W}, \com{P} \bnfdef&\ \com{M , \OB{F} , \OB{I}}
	&
	\mi{Component}~\com{C} \bnfdef&\ \com{\OB{F} , \OB{I}}
	&
	\mi{Imports}~\com{I} \bnfdef&\ \com{f}
	\end{aligned}
	\\
	\begin{aligned}
	\mi{Functions}~\com{F} \bnfdef&\ \emptyset \mid \com{F}; f \mapsto (l_{start}, c_1;  c_2, f_{set})
	&
	\mi{(Code)}~\com{p} \bnfdef&\  l_n:i \mid p_1;p_2
	&
	\mi{Attackers}~\ctxc{} \bnfdef&\ \com{M, \OB{F}\hole{\cdot}} 
	\\
	\mi{Frames}~\com{B} \bnfdef
                &\ \emptyset \mid \com{\OB{n}}; B
	&
	\end{aligned}
	\\
	\begin{aligned}
	\mi{(Registers)}~ x \in&\ \Reg
	&
    \mi{(Labels)}~ l \in&\ \labelset
    &
    \mi{(Nats)}~ n \in&\ \Nat \cup \{\bot\}
    &
    \mi{(Values)}~ v \in&\ \Val = \mi{Nats} \cup  \labelset
	\end{aligned}
	\\
	\begin{aligned}
	\mi{Register File}~\com{A} \bnfdef&\ \emptyset \mid \com{A}; x \mapsto v : \taint
	&
	\com{A^{v}} \bnfdef&\ \emptyset \mid \com{A}; x \mapsto v
	&
	\com{A^{\taint}} \bnfdef&\ \emptyset \mid \com{A}; x \mapsto \taint
	\end{aligned}
	\\
	\begin{aligned}
	\mi{Memory}~\com{M} \bnfdef&\ \emptyset \mid \com{M}; n \mapsto v : \taint
	&
	\com{M^{v}} \bnfdef&\ \emptyset \mid \com{M}; n \mapsto v
	&
	\com{M^{\taint}} \bnfdef&\ \emptyset \mid \com{M}; n \mapsto \taint
	\end{aligned}
	\\
	\begin{aligned}
	\com{\sigma} \bnfdef&\ (p, M, A)
	&
	\mi{State}~ \bnfdef \src{C;\OB{B}; \sigma}
	\end{aligned}
	\\
	\begin{aligned}
	\mi{(Expressions)}~  e \bnfdef&\  n \mid x \mid \ominus e \mid e_1 \otimes e_2
	\\
	\mi{(Instructions)}~ i \bnfdef&\ \pskip \mid \passign{x}{e} \mid \pload{x}{e} \mid \pstore{x}{e} \mid \pjmp{e} \mid  \\
                                                  & \pjz{x}{l} \mid \pcondassign{x}{e}{e'} \mid \pbarrier  \mid \pcall{f} \mid \pret \\
                                                  & \ploadprv{x}{e} \mid \pstoreprv{x}{e}
	\\ 
	\mi{(Administrative\ Inst.)}~ i \bnfdef&\ \ppopret \mid \pmodret{e}			  
	\end{aligned}
\end{gather*}

We extend \muasm{} with a notion of components (a unit of compilation) i.e, partial programs (\com{P}) and attackers (\ctxc{}).
A partial program \com{P} defines a list of functions $\OB{F}$ and a list of imports \OB{I}, which are all the functions an attacker can define. 
An attacker $\ctxc{}$ just defines its functions.

Functions \com{F} are maps from function names to a pair consisting of the abstract starting label of the function, the code of that function. We define as $f_{set}$ all the labels of that function which we recover by iterating over the code of the function.

Code of a \muasm{} program p is defined as a sequence of pairs $l_n : i$ where $i$ is an instruction and $l_n$ is an abstract label drawn from a partially ordered set. We will use $p$ as a partial function from natural numbers to instructions, where $p(n)$ either returns the instruction if the label exists in p or $\bot$.
The same abstract label $l_n$ cannot be used twice to label two different instructions.

Instructions $i$ include skipping, register assignments, loads, store, indirect jumps, branching, conditional assignments, speculation barriers, calls, and returns.
Instructions can contain expressions and values.
The former come from the set $\Reg$, containing register identifiers and designated registers $\pc$ and $\spR$ modelling the program counter and stack pointer respectively, while the latter come from the set $\Val$, which includes natural numbers and $\bot$.

The Administrative instructions $\ppopret{}$ and $\pmodret{e}$ can only be inserted by the compiler and allow for return address manipulation necessary for the Retpoline countermeasure(\Cref{comp:v2-retpoline} and \Cref{comp:v5-retpoline}).

$B$ consists of a stack of stacks of natural numbers $n$. These natural numbers represent the return addresses.
A new empty stack $\OB{n}$ is created, whenever a context switch between component and attacker happens. 

The register file \com{A} and the memory \com{M} contain mappings from registers and natural numbers to values. 
The memory consists of a public and a private part. The public one is the positive one and the private one is the negative part.
We sometimes write $\com{M^{\high}}$ and $\com{M^{\low}}$ to indicate the private and public part respectively of memory $M$.
The memory is preallocated. This means all locations from 0 up to plus and minus infinity is allocated and initialised to 0.

A configuration $\sigma$ is a triple $\tup{p, m, a}$ containing the program $p$, the memory $m$ and the register file or assignments $a$.

The reason to use these abstract labels $l_i$ instead of using natural numbers has to do with compilation. Since our compilation is from \muasm{} to \muasm{} and the compiler inserts new instructions, we run into a lot of trouble with for example jumps or branches. Some jump targets now need to be changed but it is unclear when it needs to change.
Abstract labels do not have that problem. When a new instruction is inserted by the compiler, i.e. between $l_2$ and $l_3$, where $succ(l_2) = l_3$ then we draw a new label $l_2'$ and insert it into the partial order. If these were natural numbers we would have needed to change all labels of all later instructions. But in this way, we do not change old jump targets.

In the semantics we keep these labels abstract. We will explain this choice in the following section. We note that it is possible to instantiate these labels. See \Cref{sec:Linker}.

Furthermore, in the paper we used $\beh^{}_x(P)$ instead of $\behx{P}$ to indicate the behaviour of the program, where $\omega$ is the size of the speculation window.
\subsection{Trace Model}
Computation steps are labelled with events $\tau$, which is either the empty label $\varepsilon$, an action $\alpha?$ or $\alpha!$ recording the control flow interactions between attacker and code (needed for secure compilation proofs), or a microarchitectural event $\delta$.
\begin{align*}
    \mi{Labels}~ \lambda \bnfdef&\ \epsilon \mid \alpha \mid \delta \mid \ts
	\\
	\mi{Actions}~\alpha \bnfdef&\ (\clh{f}) \mid (\cbh{f}) \mid (\rth_id) \mid (\rbh)
	\\
	\mi{Heap\&Pc\ Act.s}~\delta \bnfdef&\ \storeObs{n} \mid \loadObs{n} \mid (\pcObsID{l}) \mid \rollbackObsID{} ~\text{where $\id = \{\Bv, \Jv, \Sv, \Rv, \SLSv, \empTr \}$}
	\\
	\mi{Traces}~\tra{^\taint} \bnfdef&\ \srce \mid \tra{^\taint}\cdot\alpha^\taint \mid \tras{^\taint}\cdot\delta^\taint
\end{align*}
The index will tell us later when and from which semantics speculation started. This is needed to define the projections functions on traces as well.

\subsection{Non-speculative Operational Semantics}

We will write the shorthand $a[x \mapsto y]$ with $x \in \Reg$ and $y \in \Val$ for the assignment $a' \in \Assgn$ that satisfies $a'(x) = y$ and otherwise behaves similar to a. This updates a single register.

We will do the same for memories and lift this to configurations in the obvious way with $\sigma = \tup{p, m,a}$:
\begin{align*}
    \tup{p, m, a}[x \mapsto y] :=&\ \tup{p, m, a[x \mapsto y]} \text{with}~ x \in \Reg~ y \in \Val \\
    \tup{p, m, a}(n \mapsto y) :=&\ \tup{p, m(n \mapsto y), a} \text{with}~ n \in \mathbb{Z} ~ y \in \Val 
\end{align*}

\Cref{tr:plug-us} tells how to obtain a whole program from a component and an attacker.
\Cref{tr:whole-us} tells when a program is whole. It relies on the helper function $\labels{F}$ which extracts the labels $l$ used in the function.

We change the way the list of imports is used between partial and whole programs.
For partial programs, imports are effectively imports, i.e., the functions that contexts define and that the program relies on.
So a context can define more functions.
In whole programs, we change the imports to be the list of all context-defined function (\Cref{tr:plug-us}) to keep track of what is and what is not context.

\begin{center}
    \mytoprule{Helpers}
    
    \typerule{Intfs}{
		\src{C} = \src{\OB{F} , \OB{I}}
	}{
		\src{C}.\mtt{intfs} = \src{\OB{I}}
	}{}
	\typerule{Funs}{
		\src{C} = \src{\OB{F} , \OB{I}}
	}{
		\src{C}.\mtt{funs} = \src{\OB{F}}
	}{}

	\typerule{Jump-Internal}{
		((\src{f'}\in\src{\OB{I}} \wedge \src{f}\in\src{\OB{I}}) \vee
				\\
		(\src{f'}\notin\src{\OB{I}} \wedge \src{f}\notin\src{\OB{I}}))
	}{
		\src{\OB{I}}\vdash\src{f,f'}:\src{internal}
	}{us-aux-intern}
	\typerule{Jump-IN}{
		\src{f}\in\src{\OB{I}} \wedge \src{f'}\notin\src{\OB{I}}
	}{
		\src{\OB{I}}\vdash\src{f,f'}:\src{in}
	}{us-aux-in}
	\typerule{Jump-OUT}{
		\src{f}\notin\src{\OB{I}} \wedge \src{f'}\in\src{\OB{I}}
	}{
		\src{\OB{I}}\vdash\src{f,f'}:\src{out}
	}{us-aux-out}

	\typerule{Call Label - call}{
		\src{C}.\mtt{intfs}\vdash\src{f',f} : \src{in}
	}{
		\src{C};\src{f',f} ; \src{call} \vdash \src{\pcall{f}?}
	}{us-aux-call}
	\typerule{Call Label - callback}{
		\src{C}.\mtt{intfs}\vdash\src{f',f} : \src{out}
	}{
		\src{C};\src{f',f} ; \src{call} \vdash \src{\pcall{f}!}
	}{us-aux-callback}
	\typerule{Call Label - internal}{
		\src{C}.\mtt{intfs}\vdash\src{f',f} : \src{internal}
	}{
		\src{C};\src{f',f} ; \src{call} \vdash \src{\epsilon}
	}{us-aux-call-int}
	
	\typerule{Ret Label - return}{
		\src{C}.\mtt{intfs}\vdash\src{f',f} : \src{in}
	}{
		\src{C};\src{f',f} ; \src{ret}  \vdash \src{\rth{f}{}{}}
	}{us-aux-ret}
	\typerule{Ret Label - returnback}{
		\src{C}.\mtt{intfs}\vdash\src{f',f} : \src{out}
	}{
		\src{C};\src{f',f} ; \src{ret}  \vdash \src{\rbh{f}{}{}}
	}{us-aux-retback}
	\typerule{Ret Label - internal}{
		\src{C}.\mtt{intfs}\vdash\src{f',f} : \src{internal}
	}{
		\src{C};\src{f',f} ; \src{ret}  \vdash \src{\epsilon}
	}{us-aux-ret-int}

	\typerule{Plug}{
		\src{\ctxs{}} \equiv \src{M , \OB{F}\hole{\cdot}}
		&
		\src{C}\equiv\src{M' \OB{F'} , \OB{I}} 
		&
		\vdash\src{\OB{F'};\OB{F} , \OB{I}}:\src{whole}
		&
		\src{main}\in\fun{names}{\src{\OB{F}}}
		\\
		\dom{\src{M}}\cap\dom{\src{M'}}=\emptyset
		\\
		\forall \src{n\mapsto v : \sigma} \in \src{M'}, \src{n}<\src{0} \text{ and } \src{\sigma}=\src{\unta}
	}{
		\src{\ctxs{}\hole{C}} = \src{M;M', \OB{F;F'}, \dom{\src{\OB{F}}}}
	}{plug-us}
	\typerule{Whole}{
		\fun{names}{\src{\OB{F}}}\cap\fun{names}{\src{\OB{F'}}}=\emptyset
		\\
		\fun{names}{\src{\OB{I}}}\subseteq \fun{names}{\src{\OB{F}}}\cup\fun{names}{\src{\OB{F'}}}
		&
		\labels{\OB{F'}} \cap \labels{\OB{F}} \cap \labels{\OB{I}} = \emptyset
	}{
		\vdash\src{\OB{F'};\OB{F} , \OB{I}}:\src{whole}
	}{whole-us}
    
    \typerule{Initial State}{
		\src{M_0} = \src{M''}\cup\src{M}\cup\src{M'}
		\\
		\src{M'} = \myset{ \src{n\mapsto 0 : \safeta} }{ \src{n}\in\mb{N}\setminus\dom{\src{M}} } 
		\\
		\src{M''} = \myset{ \src{-n\mapsto 0 : \unta} }{ \src{n}\in\mb{N}, \src{-n}\notin\dom{\src{M}} } \\
		\src{A_0} = \myset {x \mapsto 0 : \safeta}{x \in \Reg} \\
		\src{A_0'} = \src{A_0}[\pc \mapsto \ffun{main}]
	}{
		\SInits{(\src{M , \OB{F} , \OB{I}})} = \src{\OB{F} , \OB{I} , \emptyset, \tup{p, M_0, A_0'} }
	}{ini-us}
	
	\typerule{Terminal State}{
		\nexists \src{\Omega},\src{\lambda} . \src{\Omega \nsarrow{\lambda} \Omega'}
	}{
		\finType{\src{\Omega}}
	}{term-state}
    
	\typerule{Merge-B-base}{}{
		\srce + \srce = \srce
	}{}
	\typerule{Merge-B-ind}{
		\src{\OB{\Bva} + \OB{\Bt}} = \src{\OB{B}} & \src{\OB{n_v}+\OB{n_t}}=\src{\OB{n}}
	}{
		\src{\OB{\Bva} ; \OB{n_v} + \OB{\Bt} ; \OB{n_t}} = \src{\OB{B} ; R}
	}{}
	
	\typerule{Merge-M-base}{}{
		\srce + \srce = \srce
	}{}
	\typerule{Merge-M-ind}{
		\src{\mv+\mta}=\src{M}
	}{
		\src{\mv ; n\mapsto v+\mta ; n\mapsto \taint} = \src{M ; n\mapsto v:\taint}
	}{}
	\typerule{Merge-A-base}{}{
		\srce + \srce = \srce
	}{}
	\typerule{Merge-A-ind}{
		\src{\av+\at}=\src{A}
	}{
		\src{\av ; n\mapsto v+\at ; n\mapsto \taint} = \src{A ; n\mapsto v:\taint}
	}{}
	
	\typerule{Merge-s-}{
		\src{\mv+\mta}=\src{M}
		&
		\src{\OB{\Bva} + \OB{\Bt}} = \src{\OB{B}} & \src{\av+\at}=\src{A}
	}{
		\src{C;\OB{\Bva}; \tup{p, \mv, \av} + C;\OB{\Bt}; \tup{p, \mta, \at}} = \src{C ; \OB{B'}; \tup{p, M, A}}
	}{}
	
\end{center}

A non-speculative state consists of the codebase $C$, a stack $B$ tracking the return address and a configuration $\sigma$, which is pair of a register file $a \in \Reg$ and a memory $m \in M$.

We can split a configuration $\sigma = \tup{m, a}$ into their value and taint parts and write $\sigmav = \tup{m^{v}, a^{v}}$ and $\sigmat = \tup{m^{\taint}, a^{\taint}}$ respectively.

\mytoprule{\mi{Judgements}}
\begin{align*}
    &
    \exprEval{\sigmav}{e}{v}
    && \text{Expression e big-steps to value v}
    \\
    &
    \exprEval{\sigmat}{e}{\taint}
    && \text{Expression e is tainted \taint}
    \\
    &
    \exprEval{\sigma}{e}{v : \taint}
    && \text{Expression e big-steps to value v tagged \taint}
    \\
    &
    \sigma \nsarrow{\tau} \sigma'
    && \text{Configuration $\sigma$ small-steps to $\sigma'$ and emits observation $\tau$.}
    \\
    &
    \sigma \nsbigarrow{\tauStack} \sigma'
    && \text{Configuration $\sigma$ big-steps to $\sigma'$ and emits a list of observations $\tauStack$.}
    \\
    &
    \Trace{(p, \sigma)}{\tauStack}
    && \text{Program $p$ and initial configuration $\sigma$ produce the observations $\tauStack$ during execution.}
\end{align*}

\subsubsection{Expression Evaluation}

\begin{center}
    \mytoprule{\exprEval{\sigma}{e}{v}}

\typerule{E-val}
{}
{
\exprEval{\sigmav}{v}{v}
}{E-val}
\typerule{E-lookup}
{\sigmav(r) = v & \text{$r \in \Reg$}}
{
\exprEval{\sigmav}{r}{v}
}{E-lookup}

\typerule{E-binop}
{
\exprEval{\sigmav}{e_1}{n_1} & \exprEval{\sigmav}{e_2 }{n_2} & n = n_1 \otimes n_2'
}
{
\exprEval{\sigmav}{e_1 \otimes e_2 }{n}
}{E-binop}
\typerule{E-unop}
{
\exprEval{\sigmav}{e }{n}
}
{
\exprEval{\sigmav}{\ominus e }{\ominus n}
}{E-unop}

\end{center}

\begin{center}
    \mytoprule{\sigmat \triangleright{} e \bigreds \taint}

\typerule{T-val}
{}
{
\exprEval{\sigmat}{v}{\safeta}
}{T-val}
\typerule{T-lookup}
{\sigmat(r) = \taint & \text{$r \in \Reg$}}
{
\exprEval{\sigmat}{r}{\taint}
}{T-lookup}

\typerule{T-binop}
{
\exprEval{\sigmat}{e_1}{\taint} & \exprEval{\sigmat}{e_2 }{\taint'} & \taint'' = \taint \sqcup \taint'
}
{
\exprEval{\sigmat}{e_1 \otimes e_2 }{\taint''}
}{T-binop}
\typerule{T-unop}
{
\exprEval{\sigmat}{e}{\taint}
}
{
\exprEval{\sigmat}{\ominus e }{\taint}
}{T-unnop}

\end{center}

\begin{center}
    \mytoprule{\exprEval{\sigma}{e}{v : \taint}}

\typerule{Combine-val}
{\sigmav + \sigmat = \sigma & \exprEval{\sigmav}{e}{v} & \exprEval{\sigmat}{e}{\taint}}
{
\exprEval{\sigma}{e}{v : \taint}
}{Combine}

\end{center}

Note the rules for \Cref{tr:E-binop} and \Cref{tr:E-unop}. They do not allow the appearance of labels. This means expression like $2 + l$ are not allowed in the program. However, labels can be looked up in \Cref{tr:E-lookup}. For more information we refer to \Cref{sec:Linker}.

\subsection{Part Two}
Furthermore, we only allow jmp l where l is a label, We have an indirect jump if the location is stored in a register.

All arguments of functions are put into registers according to the calling convention. The caller puts all arguments into registers $r_{f1} \cdots r_{f_n}$.

Since each code location $l$ belongs to a function, we can define a helper function $\ffun{}$ that maps each location $l$ of a program to its corresponding function.

\begin{center}\small
\mytoprule{\sigma \nsarrow{\tau} \sigma}

\typerule{Skip}
{
\select{p}{a(\pc)} = \pskip
}
{
\src{C;\OB{\Bva}; \tup{p,m,a}} \nsarrow{} \src{C;\OB{\Bva};\tup{p, m, \av[\pc \mapsto \av(\pc)+1]}}
}{skip}

\typerule{Assign}
{
\select{p}{\av(\pc)} = \passign{x}{e} & x \neq \pc
}
{
\src{C;\OB{\Bva}; \tup{p,\mv,\av}} \nsarrow{} \src{C;\OB{\Bva}; \tup{m, \av[\pc \mapsto \av(\pc)+1,x \mapsto \exprEval{\av}{e}{v}]}}
}{assign}

\typerule{ConditionalUpdate-Sat}
{
	p(\av(\pc)) = \pcondassign{x}{e}{e'} &  \exprEval{\av}{e'}{v} \\
	x \neq \pc & v = 0
}
{
	\src{C;\OB{\Bva}; \tup{p,\mv,\av}} \nsarrow{} \src{C;\OB{\Bva}; \tup{\mv,\av[\pc \mapsto \av(\pc) + 1, x \mapsto \exprEval{\av}{e}{v}]}}
}{condup-sat}

\typerule{ConditionalUpdate-Unsat}
{
	p(\av(\pc)) = \pcondassign{x}{e}{e'} & \exprEval{\av}{e'}{v} \\
	x \neq \pc & v \neq 0
}
{
    \src{C;\OB{\Bva}; \tup{p,\mv,\av}} \nsarrow{} \src{C;\OB{\Bva};\tup{\mv,\av[\pc \mapsto \av(\pc) + 1]}}
}{condup-unsat}

\typerule{Terminate}
{
	\select{p}{\av(\pc)} = \bot
}
{
\src{C;\OB{\Bva}; \tup{p,\mv,\av}}  \nsarrow{} \src{C;\OB{\Bva}; \tup{\mv, \av[\pc \mapsto \bot]}}
}{terminate}

\typerule{Load}
{
\select{p}{\av(\pc)} = \pload{x}{e} & x \neq \pc & \exprEval{\av}{e}{n}
}
{
\src{C;\OB{\Bva}; \tup{p,\mv,\av}}  \nsarrow{\loadObs{n}} \src{C;\OB{\Bva}; \tup{\mv, \av[\pc \mapsto \av(\pc)+1, x \mapsto \mv^{\low}(n)]}}
}{load}

\typerule{Store}
{
\select{p}{\av(\pc)} = \pstore{x}{e} &  \exprEval{\av}{e}{n}
}
{
\src{C;\OB{\Bva}; \tup{p,\mv,\av}} \nsarrow{\storeObs{n}} \src{C;\OB{\Bva}; \tup{ \mv^{\low}[n \mapsto \av(x)], \av[\pc \mapsto \av(\pc)+1]}}
}{store}

\typerule{Beqz-Sat}
{
\select{p}{\av(\pc)} = \pjz{x}{\lbl} &  \av(x) = 0
}
{
\src{C;\OB{\Bva}; \tup{p,\mv,\av}}  \nsarrow{\pcObs{\lbl}} \src{C;\OB{\Bva};\tup{\mv, \av[\pc \mapsto \lbl]}}
}{beqz-sat}

\typerule{Beqz-Unsat}
{
\select{p}{\av(\pc)} = \pjz{x}{\lbl}  & \av(x) \neq 0
}
{
\src{C;\OB{\Bva}; \tup{p,\mv,\av}}  \nsarrow{\pcObs{\av(\pc)+1}} \src{C;\OB{\Bva};\tup{\mv, \av[\pc \mapsto a(\pc) +1]}}
}{beqz-unsat}

\typerule{Jmp}
{
\select{p}{\av(\pc)} = \pjmp{e} &  \exprEval{\av}{e}{\lbl} \\
\ffun{\av(\pc)} = f & \ffun{\lbl} = f }
{
\src{C;\OB{\Bva}; \tup{p,\mv,\av}}  \nsarrow{\pcObs{\lbl}} \src{C;\OB{\Bva};\tup{ m, a[\pc \mapsto \lbl]}}
}{jmp}

\typerule{Load-Prv}
{
\select{p}{a(\pc)} = \ploadprv{x}{e} & x \neq \pc & \exprEval{\av}{e}{n}
}
{
\src{C;\OB{\Bva}; \tup{p,\mv,\av}}  \nsarrow{\loadObs{n}} \src{C;\OB{\Bva}; \tup{\mv, \av[\pc \mapsto a(\pc)+1, x \mapsto \mv^{\high}(n)]}}
}{load-prv}

\typerule{Store-Prv}
{
\select{p}{a(\pc)} = \pstoreprv{x}{e} &  \exprEval{\av}{e}{n}
}
{
\src{C;\OB{\Bva}; \tup{p,\mv,\av}} \nsarrow{\storeObs{n}} \src{C;\OB{\Bva}; \tup{ \mv^{\high}[n \mapsto \av(x)], \av[\pc \mapsto \av(\pc)+1]}}
}{store-prv}

\typerule{Call}
{
\select{p}{\av(\pc)} = \pcall{f'} & \mathcal{F}(f') = n' & f' \in C.\mtt{funs}\\
\av' = \av[\pc \mapsto n'] & \av(\pc) = n & \ffun{n} = f  \\
\src{C}.\mtt{intfs}\vdash\src{f',f} : \src{in} }
{
\src{C;\OB{\Bva}; \tup{p,\mv,\av}}  \nsarrow{\pcall{f'}?} \src{C;  (\OB{\Bva \cdot \av(\pc) + 1} ;\emptyset) ;\tup{p, \mv, \av'}}
}{call}
\typerule{Call-internal}
{
\select{p}{\av(\pc)} = \pcall{f'} & \mathcal{F}(f') = n' & f' \in C.\mtt{funs}\\
\av' = \av[\pc \mapsto n'] & \av(\pc) = n & \ffun{n} = f \\
\src{C}.\mtt{intfs}\vdash\src{f',f} : \src{internal} }
{
\src{C;\OB{\Bva}; \tup{p,\mv,\av}}  \nsarrow{\varepsilon} \src{C;\OB{\Bva} \cdot\av(\pc) + 1;\tup{p, \mv, \av'}}
}{call-internal}

\typerule{Callback}
{
\select{p}{\av(\pc)} = \pcall{f'} & \mathcal{F}(f') = n' & f' \in C.\mtt{funs}\\
\av' = \av[\pc \mapsto n'] & \av(\pc) = n & \ffun{n} = f &
\src{C}.\mtt{intfs}\vdash\src{f',f} : \src{out}
}
{
\src{C;\OB{\Bva}; \tup{p,\mv,\av}}  \nsarrow{\pcall{f'}!} \src{C; (\OB{\Bva} \av(\pc) + 1 ; \emptyset);\tup{p, \mv, \av'}}
}{callback}

\typerule{Ret}
{
\select{p}{\av(\pc)} = \pret & \av(\pc) = n & \av' = \av[\pc \mapsto l]  \\
\ffun{n} = f & \ffun{l} = f' & \src{C}.\mtt{intfs}\vdash\src{f',f} : \src{in}
}
{
\src{C; (\OB{\Bva} \cdot l; \OB{n}) \tup{p,\mv,\av}}  \nsarrow{\pret!} \src{C;\OB{\Bva}; \tup{p,\mv,\av'}} 
}{ret}
\typerule{Ret-internal}
{
\select{p}{\av(\pc)} = \pret & \av(\pc) = n \av' = \av[\pc \mapsto l]  \\
\ffun{n} = f & \ffun{l} = f' & \src{C}.\mtt{intfs}\vdash\src{f',f} : \src{internal}
}
{
\src{C;\OB{\Bva} \cdot l; \tup{p,\mv,\av}}  \nsarrow{\varepsilon} \src{C;\OB{\Bva}; \tup{p,\mv,\av'}} 
}{ret-internal}

\typerule{Retback}
{
\select{p}{\av(\pc)} = \pret &  \av(\pc) = n & \av' = \av[\pc \mapsto l]  \\
\ffun{n} = f & \ffun{l} = f' & \src{C}.\mtt{intfs}\vdash\src{f',f} : \src{out}
}
{
\src{C;(\OB{\Bva} \cdot l; \OB{n}); \tup{p,\mv,\av}}  \nsarrow{\pret?} \src{C;\OB{\Bva}; \tup{p,\mv,\av'}} 
}{retback}

\mytoprule{\text{Administrative instructions}}

\typerule{Popret}
{
\select{p}{\av(\pc)} = \ppopret & \av(\pc) = n & \ffun{n} = f & f \in C.funs &  \\
}
{
\src{C; (\OB{\Bva} \cdot l; \OB{n}) \tup{p,\mv,\av}}  \nsarrow{} \src{C;\OB{\Bva}; \tup{p,\mv,\av[\pc \mapsto \av(\pc)+1]}} 
}{popret}

\typerule{Modret}
{
\select{p}{\av(\pc)} = \pmodret{e} & \av(\pc) = n &  \ffun{n} = f \\
 f \in C.funs &\exprEval{\av}{e}{l'}
}
{
\src{C; (\OB{\Bva} \cdot l; \OB{n}) \tup{p,\mv,\av}}  \nsarrow{\pret!} \src{C;\OB{\Bva \cdot l'}; \tup{p,\mv, \av[\pc \mapsto \av(\pc)+1]}} 
}{modret}

\end{center}

$\pjmp{}$ instructions are required to be internal and only jmp to the same function. Since the component does not know with which context it is linked against, it does not make sense to allow these kind of external jumps. Components and context can only interact via $\pcall{}$ and $\pret$ instructions.
Next, jumping to other functions destroys the call stack and makes it harder to reason about the program.

Return addresses are stored in the current Frame except for cross component calls. There, the return address is stored in the old Frame and a new frame is created.

Store and load operations (\Cref{tr:store}, \Cref{tr:load}) require that the memory location is a natural number using the side-condition $\exprEval{\av}{e}{n}$. This means labels cannot appear there. For an example why this restriction is needed see \Cref{sec:linker-store-ex}. Note that it is still allowed to store and load labels from memory.

Private stores and loads are required to store/load to a High part of the memory. Furthermore, these instructions cannot be executed by the context.

$\taintpc$ will indicate if we speculate or not. While speculating this taint will be \unta. Furthermore, we need to track the value of the $\pc$ even during taint tracking to know which instruction to execute next.

\mytoprule{\taintpc; \Omega \nsarrow{\taint} \Omega;}
\begin{center}

\typerule{T-Skip}
{
\select{p}{a(\pc)} = \pskip
}
{
\src{\taintpc; C;\OB{B}; \tup{p,\mta, a}} \nsarrow{\epsilon} \src{C;\OB{B};\tup{p, \mta, a[\pc \mapsto a(\pc)+1]}}
}{t-skip}

\typerule{T-Assign}
{
\select{p}{a(\pc)} = \passign{x}{e} & x \neq \pc & \exprEval{a}{e}{v : \taint}
}
{
\src{\taintpc; C;\OB{B}; \tup{p,\mta,a}} \nsarrow{\epsilon} \src{C;\OB{B}; \tup{p, m, a[\pc \mapsto a(\pc)+1, x \mapsto v : \taint ]}}
}{t-assign}

\typerule{T-ConditionalUpdate-Sat}
{
	p(a(\pc)) = \pcondassign{x}{e}{e'} &  \exprEval{a}{e'}{v' : \taint'} \\
	x \neq \pc & v = 0 & \exprEval{a}{e}{v : \taint}
}
{
	\src{\taintpc; C;\OB{B}; \tup{p,\mta, a}} \nsarrow{\epsilon} \src{ C;\OB{B}; \tup{\mta,a[\pc \mapsto a(\pc) + 1, x \mapsto v : \taint]}}
}{t-condup-sat}

\typerule{T-ConditionalUpdate-Unsat}
{
	p(a(\pc)) = \pcondassign{x}{e}{e'} & \exprEval{a}{e'}{v : \taint} \\
	x \neq \pc & v \neq 0
}
{
    \src{\taintpc; C;\OB{\Bva}; \tup{p,\mta, a}} \nsarrow{} \src{C;\OB{\Bva};\tup{\mta, a[\pc \mapsto a(\pc) + 1]}}
}{t-condup-unsat}

\typerule{T-Terminate}
{
	\select{p}{a(\pc)} = \bot
}
{
\src{\taintpc; C;\OB{B}; \tup{p,\mta,a}}  \nsarrow{\epsilon} \src{ C;\OB{B}; \tup{p, \mta, a[\pc \mapsto \bot]}}
}{t-terminate}

\typerule{T-Load}
{
\select{p}{a(\pc)} = \pload{x}{e} & x \neq \pc \\
\exprEval{a}{e}{n : \taint} & \mta^{\low}(n) = \taint'
}
{
\src{\taintpc; C;\OB{B}; \tup{p,\mta,a}}  \nsarrow{\taint \lub \taint'} \src{C;\OB{B}; \tup{p, \mta, a[\pc \mapsto a(\pc)+1, x \mapsto 0 : \taint']}}
}{t-load}

\typerule{T-Store}
{
\select{p}{a(\pc)} = \pstore{x}{e} &  \exprEval{a}{e}{n : \taint} 
}
{
\src{\taintpc; C;\OB{B}; \tup{p,\mta,a}} \nsarrow{\taint} \src{C;\OB{B}; \tup{p, \mta^{\low}[\abs{n} \mapsto \safeta], a[\pc \mapsto \av(\pc)+1]}}
}{t-store}

\typerule{T-Beqz-Sat}
{
\select{p}{a(\pc)} = \pjz{x}{\lbl} &  a(x) = 0 : \taint
}
{
\src{\taintpc; C;\OB{B}; \tup{p,\mta, a}}  \nsarrow{\taint} \src{C;\OB{B};\tup{p, \mta, a[\pc \mapsto \lbl]}}
}{t-beqz-sat}

\typerule{T-Beqz-Unsat}
{
\select{p}{a(\pc)} = \pjz{x}{\lbl}  & a(x) = n : \taint & n \neq 0
}
{
\src{\taintpc; C;\OB{B}; \tup{p,\mta,a}}  \nsarrow{\taint} \src{C;\OB{\Bva};\tup{p, \mta, a[\pc \mapsto a(\pc) +1]}}
}{t-beqz-unsat}

\typerule{T-Jmp}
{
\select{p}{a(\pc)} = \pjmp{e} &  \exprEval{a}{e}{\lbl : \taint} \\
\ffun{\av(\pc)} = f & \ffun{\lbl} = f }
{
\src{\taintpc; C;\OB{B}; \tup{p,\mta, a}}  \nsarrow{\taint} \src{C;\OB{B};\tup{p, \mta, a[\pc \mapsto \lbl]}}
}{t-jmp}

\typerule{T-Load-Prv}
{
\select{p}{a(\pc)} = \ploadprv{x}{e} & x \neq \pc \\
\exprEval{a}{e}{n : \taint} & \mta^{\high}(n) = \taint' & \taint'' = \taint \sqcup \taint'
}
{
\src{\taintpc; C;\OB{B}; \tup{p,\mta, a}}  \nsarrow{\taint''} \src{C;\OB{B}; \tup{p, \mta, a[\pc \mapsto a(\pc)+1, x \mapsto 0 : \unta]}}
}{t-load-prv}

\typerule{T-Store-Prv}
{
\select{p}{a(\pc)} = \pstoreprv{x}{e} &  \exprEval{a}{e}{n : \taint} & a(x) = v : \taint'
}
{
\src{\taintpc; C;\OB{B}; \tup{p,\mta,a}} \nsarrow{\taint} \src{C;\OB{B}; \tup{p, \mta^{\high}[n \mapsto \taint'], a[\pc \mapsto a(\pc)+1]}}
}{t-store-prv}

\typerule{T-Call}
{
\select{p}{a(\pc)} = \pcall{f'} & \mathcal{F}(f') = n' & f' \in C.\mtt{funs}\\
a' = a[\pc \mapsto n' : \safeta] & a(\pc) = n & \ffun{n} = f \\
\src{C}.\mtt{intfs}\vdash\src{f',f} : \src{in} 
}
{
\src{\taintpc; C;\OB{B}; \tup{p,\mta,a}}  \nsarrow{\safeta} \src{C;\OB{B} \cdot 0 \mapsto \av(\pc) + 1 : \taintpc ;\tup{p, \mta, a'}}
}{t-call}
\typerule{Call-internal}
{
\select{p}{a(\pc)} = \pcall{f'} & \mathcal{F}(f') = n' & f' \in C.\mtt{funs}\\
a' = a[\pc \mapsto n' : \safeta] & a(\pc) = n & \ffun{n} = f \\
\src{C}.\mtt{intfs}\vdash\src{f',f} : \src{internal} }
{
\src{C;\OB{\Bva}; \tup{p,\mta,\av}}  \nsarrow{\varepsilon} \src{C;\OB{B} \cdot 0 \mapsto \av(\pc) + 1 : \taintpc ;\tup{p, \mta, a'}}
}{t-call-internal}

\typerule{Callback}
{
\select{p}{a(\pc)} = \pcall{f'} & \mathcal{F}(f') = n' & f' \in C.\mtt{funs}\\
a' = a[\pc \mapsto n' : \safeta] & a(\pc) = n & \ffun{n} = f
\src{C}.\mtt{intfs}\vdash\src{f',f} : \src{out}
}
{
\src{C;\OB{\Bva}; \tup{p,\mta, a}}  \nsarrow{\safeta} \src{C;\OB{\Bva} \cdot 0 \mapsto \av(\pc) + 1;\tup{p, \mta, a'}}
}{t-callback}

\typerule{T-Ret}
{
\select{p}{a(\pc)} = \pret & a(\pc) = n & a' = a[\pc \mapsto l] \\
\ffun{n} = f & \ffun{l} = f' & \src{C}.\mtt{intfs}\vdash\src{f',f} : \src{in}
}
{
\src{\taintpc; C;\OB{B} \cdot 0 \mapsto l : \taint; \tup{p,\mta, a}}  \nsarrow{\safeta} \src{C;\OB{B}; \tup{p,\mta,a'}} 
}{t-ret}

\typerule{Ret-internal}
{
\select{p}{a(\pc)} = \pret & a(\pc) = n & a' = a[\pc \mapsto l]  \\
\ffun{n} = f & \ffun{l} = f' & \src{C}.\mtt{intfs}\vdash\src{f',f} : \src{internal}
}
{
\src{C;\OB{B} \cdot 0 \mapsto l; \tup{p,\mta, a}}  \nsarrow{\varepsilon} \src{C;\OB{B}; \tup{p,\mv,\av'}} 
}{t-ret-internal}

\typerule{Retback}
{
\select{p}{a(\pc)} = \pret &  a(\pc) = n & a' = a[\pc \mapsto l]  \\
\ffun{n} = f & \ffun{l} = f' & \src{C}.\mtt{intfs}\vdash\src{f',f} : \src{out}
}
{
\src{C;\OB{B} \cdot 0 \mapsto l; \tup{p,\mta,a}}  \nsarrow{\safeta} \src{C;\OB{B}; \tup{p, \mta,a'}} 
}{t-retback}

\mytoprule{Administrative instructions}

\typerule{Popret}
{
\select{p}{a(\pc)} = \ppopret & a(\pc) = n & \ffun{n} = f & f \in C.funs
}
{
\src{C; (\OB{B} \cdot l; \OB{n}) \tup{p,\mta,a}}  \nsarrow{} \src{C;\OB{B}; \tup{p,\mta,a}} 
}{popret}

\typerule{Modret}
{
\select{p}{a(\pc)} = \pmodret{e} & a(\pc) = n &  \ffun{n} = f \\
 f \in C.funs &\exprEval{a}{e}{l' : \taint}
}
{
\src{C; (\OB{B} \cdot l; \OB{n}) \tup{p,\mta,a}}  \nsarrow{} \src{C;\OB{B \cdot l' : \taint}; \tup{p,\mta,a}} 
}{modret}

\end{center}
We use a trick for the taint tracking of reading values. We do not know the value that is read, because we only have the taint part of the memory and so we read 0 by default. However, we taint correctly and when the memory is merged, we do not consider the values stored here.
Note that \Cref{tr:combine-s} merges the state after each step. This additionally means we will always have the correct values when evaluating the expression. (Since the value part of A is updated with each combined step).

\begin{center}
\mytoprule{\Omega \nsarrow{\tau^{\taint}} \Omega'}

\typerule{Combine-s}
{
\Omega = \Omega_v + \Omega_t &  \Omega_v' + \Omega_t' = \Omega' \\
\Omega_v \nsarrow{\tau} \Omega_v' & \safeta;\Omega_t \nsarrow{\taint} \Omega_t'
}
{\Omega \nsarrow{\tau^{\taint \sqcap \safeta}} \Omega'
}
{combine-s}
\end{center}

We add rules to define a run of the program:

\begin{center} 
    \mytoprule{\Omega \nsbigarrow{\tauStackT} \Omega'}
    
    \typerule{NS-Reflection}
    {
    }
    {\Omega
    \nsbigarrow{\varepsilon}
    \Omega
    }{ns-reflect}
    
    \typerule{NS-Single}
    {\Omega \nsbigarrow{\tra{^\taint}} \Omega'' & \Omega'' \nsarrow{\alpha^{\taint}} \Omega' & \Omega'' = \src{\OB{F}}; \src{\OB{I}}; \src{B}; \src{\sigma''} & \Omega' = \src{\OB{F}}; \src{\OB{I}}; \src{B'}; \src{\sigma'} \\
    \ffun{\sigma''(\pc)} = f & \ffun{\sigma'(\pc)} = f'  \\
    \text{ if } \src{f} == \src{f'} \text{and} \src{f} \in \src{I} \text{ then } \tra{^\taint} = \src{\epsilon} 
            \text{ else } \src{\tra{^\taint}} = \src{\alpha^{\safeta}} 
    }
    {\Omega
    \nsbigarrow{\tra{^\taint}}
    \Omega'
    }{ns-single}
    \typerule{NS-Single-silent}
    {
    \Omega \nsbigarrow{\tra{^\taint}} \Omega'' & \Omega'' \nsarrow{\epsilon} \Omega'
    }
    {\Omega
    \nsbigarrow{\tra{^\taint}}
    \Omega'
    }{ns-silent}

\end{center}

\Cref{tr:ns-single} tells that when you have a single action, if that action is done within the context, then no action is shown.
Otherwise, if the action is done within the component, or if the action is done between component and context, then the action is shown and it is tagged as $\src{\safeta}$.
All generated actions are $\src{\safeta}$ despite their label ($\src{\taint}$) because the only source of unsafety is speculation, which cannot happen in the source language.

\begin{center}
    \mytoprule{\Trace{P}{\tauStackT}}
    
    \typerule {NS-Trace}
    { \exists \Omega' \vdash \bot &  \Omega^0(P) \nsbigarrow{\tauStackT} \Omega'
    }
    { \Trace{P}{\tauStack}
     }{ns-trace}
    \typerule {NS-Beh}
    {}
    { \behNs{P} =  \{\tauStackT \mid \Trace{P}{\tauStack} \}
     }{ns-beh}
\end{center}
We define the behaviour of a whole program $P$, written $\behNs{P}$, as the set of traces that terminate.

 \section{Leakage Ordering}

Each speculative semantics introduced capture different "attacker models".
We follow \cite{contracts} and and introduce a partial order in terms of leakage between the different semantics.

In particular, we say that semantics $\contract{}{1}$ is weaker than another semantics $\contract{}{2}$, written $\contract{}{1} \sqsubseteq \contract{}{2}$ iff $\contract{}{2}$ leaks more than $\contract{}{1}$, i.e., if any two initial configurations that result in different traces for $\contract{}{1}$ also result in different traces for $\contract{}{2}$.
 \section{Contract Compiler and Program Security}

Here we define our compiler contract satisfaction and we first introduce the notions of contract-aware program security.

\begin{insight}
    In the following sections, we will use the names heaps and memories interchangeably.
\end{insight}

For convenience, we restate the definitions needed to define speculative safety.
\subsection{General Definitions}
We will define $\MemH{}$ for high or private memory and $\MemL{}$ for low or public memory.
Again we will use $\MemL{}$ and  $\MemH{}$ defined here to mean the public and private part,
\begin{align*}
    \MemL{} =&  \left.M\right|_{n \geq 0}\\
    \MemH{} =& \left.M\right|_{n < 0}
\end{align*}

\begin{definition}[Dot notation]
In the following, we use a dot notation to refer to components of a speculative instance $\Phix$. For example, we write $\Phix . \Omega$ to refer to $\Omega$ inside the instance $\Phix$.

We lift this notation to speculative states $\Sigmax$ in the following way:
we write $\Sigmax \cdot \sigma$ to denote the configuration $\sigma$ of the topmost speculative instance $\Phix$ of the speculative state $\Sigmax$. For example if $\Sigmax = \phiStackx \cdot \Phix$ then $\Sigmax . \Omega$ refers to $\Phix . \Omega$.
\end{definition}

A memory $M$ is valid if and only if its private part contains unsafe values.
We do not enforce this on the public part too because a program may write a private value in the public heap.
\begin{definition}[Valid Memory]\label{def:val-mem}
	\begin{align*}
		\vdash \com{M} :\com{vld} \isdef&\
			\forall \com{n}\mapsto \com{v}:\com{\taint} \in \com{\MemH{}}, \com{\sigma}=\com{\unta}
	\end{align*}
\end{definition}

An attacker has no private Memory nor instructions to manipulate it directly.
\begin{definition}[Attacker]\label{def:atk}
	\begin{align*}
		\vdash\com{\ctxc{}} :\com{atk} \isdef&\
			\com{\ctxc{}} \equiv \com{\MemL{}};\com{\OB{F}} 
			\\&\
			\text{ and }
			\forall \com{f} \in \com{\OB{F}}, \com{\ppopret{}} \lor \com{\pmodret{e}} \lor \com{\ploadprv{x}{e}} \lor \com{\pstoreprv{x}{e}} \notin \com{f}
	\end{align*}
\end{definition}

We first define what it means for memories and programs to be low equivalent.
\begin{definition}[Low-equivalence for memories and programs]\label{def:loweq}
	\begin{align*}
		\com{M \loweq M'} \isdef&\
			\vdash \com{M} : \com{vld} \wedge \vdash \com{M'} : \com{vld} \wedge \dom{\MemL{}} = \dom{\MemL{'}} \wedge \dom{\MemH{}} = \dom{\MemH{'}} \wedge 
			\\&\
			\forall \com{n\mapsto v:\taint}\in \com{\MemL{}} \text{ then } \com{n\mapsto v:\taint} \in \com{\MemL{'}}
			\\&\
			\forall \com{n\mapsto v:\taint}\in \com{\MemH{}} \text{ then } \exists v'.\ \com{n\mapsto v':\taint} \in \com{\MemH{'}}
		\\
		\com{P\loweq P'} \isdef&\
			\com{P}\equiv \com{M;\OB{F};\OB{I}} \text{ and } \com{P'}\equiv \com{M';\OB{F};\OB{I}} \text{ and } \com{M \loweq M'}
	\end{align*}
\end{definition}

This is a relation we later need to show that SS overapproximates SNI for a given contract. 

Let us first define what it means for program states to be related / safe-equivalent. 
We note that this notion is similar for almost all the semantics, since the state is the same. Except for $\Rv$. Here, we additionally need to require that the RSB $\Rsb$ is related as well.

\begin{definition}[Safe-equivalence for heaps, traces and program states]\label{def:safeeq}
    We first define what it means for program states and traces to be $\safe{}$ or $\unsafe{}$:
	\begin{align*}
		\unsafe{w, ss_0 \cdot (\Omega, m, \sigma)} \isdef & \sigma = \unta\\
		\safe{w, (\Omega, m, \sigma)} \isdef & \sigma = \safeta\\
		\safe{\epsilon} \isdef & \mathit{true} \\
		\safe{\alpha^\sigma} \isdef & \sigma = \safeta \\
		\safe{\trac{^\sigma} \cdot \acac{^\sigma}} \isdef & 
			\safe{\trac{^\sigma}} \text{ and } \safe{\acac{^\sigma}}\\
	    \intertext{And now the definition of $\approx$}
		\com{v:\sigma \approx v':\sigma'} \isdef&\ 
			\sigma = \sigma' \text{ and }
			\text{if } \sigma = \safeta \text{ then } v = v'\\
		\com{M \approx M'} \isdef&\ \dom{M} = \dom{M'} \text{ and } \\
		    & \qquad 
			\forall \com{n} \in \dom{M}.\ M(n) \approx M'(n) \\ 
		\com{A \approx A'} \isdef&\
				\forall \com{x}.\ A(x) \approx A'(x) \\
		\com{\Omega \approx \Omega'} \isdef&\
			\com{\Omega}\equiv \com{\OB{F}, \OB{I}, B, \tup{p, M, A}}  \text{ and } \com{\Omega'}\equiv \OB{F},\OB{I}, B, \tup{p, M', A'} \text{ and } \com{M \approx M'} \text{ and } \com{A \approx A'}
		\\
		\com{\emptyset \approx \emptyset}
		\\
		\com{\Sigma \approx \Sigma'} \isdef&\
			\com{\Sigma} \equiv \phiStack \cdot \tup{\Omega, n, \taintpc} \text{ and } \com{\Sigma'} \equiv \phiStack' \cdot \tup{\Omega', n, \taintpc} \text{ and } \Omega \approx \Omega'\text{ and } \phiStack\approx \phiStack'
	\end{align*}
\end{definition}

Since we do not have start observations we use a marker on our trace actions to indicate if speculation started together with a tag. For example, for $\contract{}{\Bv}$ this would mark $\pcObs{n}_{\mathghost\Bv}$ when speculation is started. We can do this for all our contracts.

\begin{definition}[Non-speculative Projection]
	\begin{align*}
		\nspecProject{\tauStack} =& \nspecProject{(\tauStack, 0)}
		\\
		\nspecProject{(\empTr, 0)} =& \empTr
		\\
		\nspecProject{(\tau_{id} \cdot \tauStack, 0)} =& \tau \cdot \nspecProject{(\tauStack, 1)} 
		\\
		\nspecProject{(\tau \cdot \tauStack, 0)} =& \tau \cdot \nspecProject{(\tauStack, 0)}
		\\
		\nspecProject{(\tau_{id} \cdot \tauStack, n + 1)} =& \nspecProject{(\tauStack, n + 2)} 
		\\
		\nspecProject{(\tau \cdot \tauStack, n + 1)} =& \nspecProject{(\tauStack, n + 1)}
		\\
		\nspecProject{(\rollbackObs{}_{id} \cdot \tauStack, n + 1)} =& \nspecProject{(\tauStack, n)}
	\end{align*}
\end{definition}

\subsection{Program level security}

\begin{definition}[Speculative Non-interference (SNI)]\label{def:sni}
	\begin{align*}
		\contract{}{x} \vdash \com{P} : \sni \isdef&\
			\forall P'. P' \loweq P,
			\forall \trac{_1}\in\contract{}{x}(\SInit{P}), \trac{_2}\in\contract{}{x}(\SInit{P'}) 
			\\&
			\text{ if } \nspecProject{\trac{_1}} = \nspecProject{\trac{_1}} \text{ then } \trac{_1} = \trac{_2}
	\end{align*}
\end{definition}

A component $P$ is robustly speculatively non-interferent for a contract $\contract{}{x}$ if it is SNI for any valid attacker it is linked to (\Cref{def:rsni}).

\begin{definition}[Robust Speculative Non-Interference \cite{S_sec_comp}]\label{def:rsni}
	\begin{align*}
		\contract{}{x} \vdash \com{P} : \rsni \isdef&\
			\forall \ctxc{} \text{ if } \vdash \ctxc{} : \com{atk} \text{ then } \vdash \com{\ctxc{}\hole{P}} : \sni
	\end{align*}
\end{definition}

\begin{definition}[Speculative Safety (SS)]\label{def:ss-contract}
\begin{align*}
	\contract{}{x} \vdash \com{P} : \ss \isdef&\
		\forall \trac{^\sigma}\in\contract{}{x}(P)\ldotp \forall \acac{^\sigma}\in\trac{^\sigma}\ldotp \com{\sigma}\equiv\com{\safeta}
\end{align*}
\end{definition}

A component is robustly speculatively safe (\rssdef{}) for a spec. semantics $\contract{}{x}$ if it is SS no matter what attacker it is linked against (\ssdef{}).
\begin{definition}[ Robust Speculative Safety \cite{S_sec_comp}]\label{def:rdss-contract}
\begin{align*}
	\contract{}{x} \vdash\com{P} : \rss \isdef&\
		\forall \ctxc{} \ldotp \text{ if } \vdash\ctxc{}:\com{atk} \text{ then } \contract{}{x} \vdash\com{\ctxc{}\hole{P}} : \ss
\end{align*}	
\end{definition}

Since there is no speculation in $\SR$ both definitions are trivially satisfied.
\begin{theorem}[\sstext{} and SNI hold for all source programs]\label{thm:ss-sni-source}
\begin{align*}
	\forall \src{P} \in \SR.
	 \contract{}{NS}\vdash \src{P} : \ss 
	\text{ and } \contract{}{NS} \vdash \src{P} : \sni
\end{align*}
\end{theorem}
\begin{proof}
This trivially holds from (1) programs in $\SR$ produce only actions labelled with $\src{\safeta}$ (for \ss{}, and (2) traces are identical to their non-speculative projection for programs $\SR$ (for $\contract{}{NS}$).
\end{proof}

\subsection{Compiler-level Security}

A compiler preserves $\rss$ for a given contract $\contract{}{x}$ if given a source component that is $\rss$, the compiled counterpart is also $\rss$ (\rdsspdef).

\begin{definition}[Robust speculative safety preservation \cite{S_sec_comp} (\rssp)]\label{def:rdsspc}
    $\contract{}{x} \vdash  \comp{\cdot} : \rssp \isdef \forall \src{P} \in \src{L}. \text{ if }
    \contract{}{NS} \vdash \src{P} : \rss \text{ then } \contract{}{x} \vdash \comp{\src{P}} : \rss$
\end{definition}

Similarly, a compiler preserves $\rsni$ for a given contract $\contract{}{x}$ if given a source component that is $\rsni$, the compiled counterpart is also $\rsni$.
\begin{definition}[Robust speculative non-interference preservation \cite{S_sec_comp} (\rsnip)]\label{def:rsnipc}
    $\contract{}{x} \vdash  \comp{\cdot} : \rsnip \isdef \forall \src{P} \in \src{L}. \text{ if }
    \contract{}{NS} \vdash \src{P} : \rsni \text{ then } \contract{}{x} \vdash \comp{\src{P}} : \rsni$
\end{definition}

Furthermore, we can also specify multiple contract security in the obvious way.

We can recover the original definitions of \citet{S_sec_comp} by specifying contracts $\Mcontract{}{NS}$ and $\Mcontract{}{\Bv}$ for their semantics, capturing non-speculative execution and speculation on branch instructions in their languages.

Given that contracts form a partial order $\sqsubset$, it follows that a compiler preserving the security for a stronger contract also implies that the same compiler preserves security even for a weaker contract (\Cref{cor:weak-contract-sat}).

\begin{corollary}[Leakage Ordering and RSNIP]\label{cor:weak-contract-sat}
    If $\contract{}{1} \vdash  \comp{\cdot} : \rsnip$ and $\contract{}{2} \sqsubset \contract{}{1}$
    then $\contract{}{2} \vdash \comp{\cdot} : \rsnip$
\end{corollary}

Similarly, if a compiler is not secure for a weaker contract, it also is not secure for the stronger contract as well.
\begin{corollary}[[Leakage ordering and \rsnip\ Negative]\label{cor:weak-contract-unsat}
    If $\contract{}{1} \nvdash  \comp{\cdot} : \rsnip$ and $\contract{}{1} \sqsubset \contract{}{2}$
    then $\contract{}{2} \nvdash \comp{\cdot} : \rsnip$
\end{corollary}

Expansion of \rssp for RSS yields:
\begin{align*}
	\forall\src{P}\ldotp 
	\text{ if }&\ \forall \ctxs{}\ldotp \forall \tras{^\sigma}\in\behavs{\ctxs{}\hole{P}}\ldotp \forall \acas{^\sigma}\in\tras{^\sigma}\ldotp \src{\sigma}\equiv\src{\safeta}
	\\
	\text{ then }&\ \forall \ctx{}\ldotp \forall \tra{^{\taint}}\in\contract{}{NS}(\ctx{}\hole{\comp{\src{P}}})\ldotp \forall \acac{^{\taint}}\in\tra{^{\taint}}\ldotp \taint\equiv\safeta
\end{align*}

We say that a source and a target trace are related ($\rels$) if the latter contains the source trace plus interleavings of only safe actions.
The trace relation relies on a relation on actions which in turn relies on a relation on values and heaps.

The last two are compiler-dependent, so they are presented later. 

\mytoprule{\text{Trace relation} \reldef }
\begin{center}
	\typerule{Trace-Relation}{
	}{
		\srce\rels\come
	}{tr-rel-empty}
	\typerule{Trace-Relation-Same-Act}{
		\src{\tras{^\sigma} } \rels \com{\trac{^{\taint}} }	
		&
		\src{\alpha^\sigma} \arel \com{ \acac{^{\taint}} }
	}{
		\src{\tras{^\sigma} \cdot \alpha^\sigma} \rels \com{\trac{^{\taint}} \cdot \acac{^{\taint}} }
	}{tr-rel-same}
	\typerule{Trace-Relation-Same-Heap}{
		\src{\tras{^\sigma} } \rels \com{\trac{^{\taint}} }	
		&
		\src{\delta^\sigma} \arel \com{ \comb{\delta}^{\taint} }
	}{
		\src{\tras{^\sigma} \cdot \delta^\sigma} \rels \com{\trac{^{\taint}} \cdot \comb{\delta}^{\taint} }
	}{tr-rel-same-h}
	\typerule{Trace-Relation-Safe-Act}{
		\src{\tras{^\sigma} } \rels \com{\trac{^{\taint}} }
		&
		\src{\epsilon} \arel \acac{^{\safeta}}
	}{
		\src{\tras{^\sigma}} \rels \com{\trac{^{\taint}} \cdot \acac{^{\safeta}} }
	}{tr-rel-safe-a}
	\typerule{Trace-Relation-Safe-Heap}{
		\src{\tras{^\sigma} } \rels \com{\trac{^{\taint}} }	
		&
		\src{\epsilon} \arel \com{\comb{\delta}^{\safeta}}
	}{
		\src{\tras{^\sigma} } \rels \com{\trac{^{\taint}} \cdot \comb{\delta}^{\safeta} }
	}{tr-rel-safe-h}
	\typerule{Trace-Relation-Rollback}{
		\src{\tras{^\sigma} } \rels \com{\trac{^{\taint}} }	
		&
		\src{\epsilon}\arel\comb{\rollbackObsID}^{\taint}
	}{
		\src{\tras{^\sigma} } \rels \com{\trac{^{\taint}} \cdot \comb{\rollbackObsID}^{\taint} }
	}{tr-rel-rollb}

\mytoprule{\text{Action relation} \areldef }

	\typerule{Action Relation - call}{
		\src{f}\equiv\com{f}
		&
		\src{\taint}\equiv\taint'
	}{
		\src{\clh{f}^\taint} \arel \com{\clh{f}^{\taint'}}
	}{ac-rel-cl}
	\typerule{Action Relation - return}{
	}{
		\src{\rth{}^\safeta} \arel \com{\rth{}^\safeta}
	}{ac-rel-rt}
	\typerule{Action Relation - callback}{
		\src{f}\equiv\com{f}
		&
		\src{\sigma}\equiv\taint'
	}{
		\src{\cbh{f}^\taint} \arel \com{\cbh{f}^{\taint'}}
	}{ac-rel-cb}
	\typerule{Action Relation - returnback}{
	}{
		\src{\rbh{}^\safeta} \arel \com{\rbh{}^\safeta}
	}{ac-rel-rb}
	\typerule{Action Relation - read}{
		\src{n} \vrel \com{n}
		&
		\src{\taint}\equiv\taint'
	}{
		\src{\loadObs{n}^\taint} \arel \com{\loadObs{n}^{\taint'}}
	}{ac-rel-rd}
	\typerule{Action Relation - write}{
		\src{n} \vrel \com{n}
		&
		\src{\taint}\equiv\taint'
	}{
		\src{\storeObs{n}^\taint} \arel \com{\storeObs{n}^{\taint'}}
	}{ac-rel-wr}
	\typerule{Action Relation - if}{
		\src{l} \vrel \com{l;}
		&
		\src{\taint}\equiv\taint'
	}{
		\src{\pcObs{l}^{\taint}} \arel \com{\pcObs{l'}^{\taint'}}
	}{ac-rel-if}
	\typerule{Action Relation - epsi alpha}{
		\taint\equiv \com{\safeta}
	}{
		\src{\epsilon} \arel \acac{^{\taint}}
	}{ac-rel-ep-al}
	\typerule{Action Relation - epsi heap}{
		\taint\equiv \com{\safeta}
	}{
		\src{\epsilon} \arel \comb{\delta}^{\taint}
	}{ac-rel-ep-hp}
	\typerule{Action Relation - rlb}{
		\taint\equiv \com{\safeta}
	}{
		\src{\epsilon}\arel\com{\rollbackObsID}^{\taint}
	}{ac-rel-rlb}
\end{center}
\botrule

And we can define the property-free characterization (\rdssdef):
\begin{definition}[Robust Speculative Safety Compilation]\label{def:rdss}
	\begin{align*}
		\contract{}{x} \vdash \comp{\cdot} : \rdss \isdef&\
			\forall\src{P}, \ctx{}, \trac{^{\taint}},
			\exists \ctxs{}, \tras{^\sigma} \ldotp 
			\\
			\text{ if }&\
			\com{\ctx{}\hole{\comp{\src{P}}}} \mathghost^{\omega}_{id} \trac{^{\taint}}
			\text{ then }
			\Trace{\src{\ctxs{}\hole{P}}}{\tras{^\sigma}} 
			\text{ and }
			\tras{^\sigma}\rels\trac{^{\taint}}
	\end{align*}
\end{definition}

\begin{theorem}[\rdss implies \rdssp]\label{thm:rdss-impl-rdsp}
With the relation \relref:
\begin{align*}
	\forall \comp{\cdot}
	\text{ if } \contract{}{x} \vdash\comp{\cdot} : \rdss
	\text{ then } \contract{}{x} \vdash\comp{\cdot} : \rdssp
\end{align*}
\end{theorem}

\begin{theorem}[\rdssp implies \rdss]\label{thm:rdssp-impl-rdss}
With the relation \relref:
\begin{align*}
	\forall \comp{\cdot}
	\text{ if } \contract{}{x} \vdash\comp{\cdot} :\rdssp
	\text{ then } \contract{}{x} \vdash\comp{\cdot} : \rdss
\end{align*}
\end{theorem}

\begin{theorem}[\rdss and \rdssp are equivalent]\label{thm:rdss-eq-rdsp}
With the relation \relref:
\begin{align*}
	\forall \comp{\cdot}
	\contract{}{x} \vdash\comp{\cdot} : \rdss
	\iff
	\contract{}{x} \vdash\comp{\cdot} : \rdssp
\end{align*}
\end{theorem}

\subsection{Context-based Backtranslation}
Since the backtranslation can be reused by all of their compilers, we define it here. If there are small additions needed, we will describe them in the specific compiler. The reason why we can reuse the backtranslation is that the syntax of our different target languages ($\Bv, \Jv, \Sv, \Rv, \SLSv$) is the same.

\begin{align*}
    \backtrfencec{ \trg{M ; \OB{F} ; \OB{I}}} &= \src{ \backtrfencec{ \trg{ M }  ; \backtrfencec{ \trg{\OB{F}} } ; \backtrfencec{ \trg{\OB{I}}} }}
	\\\\
	\backtrfencec{ \trg{\emptyset}} &= \src{\emptyset}
	\\
	\backtrfencec{\trg{\OB{I}\cdot f} } &= \src{ \backtrfencec{ \trg{\OB{I}} } \cdot \backtrfencec{ \trg{f}} }
	\\
	\backtrfencec{ \trg{name}, \trg{p1;p2} } &= \src{name, \backtrfencec{\trg{p1;p2}}}  
	\\\\
	\backtrfencec{\trg{M ; n\mapsto v : \taint}} &= \src{\backtrfencec{ \trg{M} } ; \backtrfencec{ \trg{n} } \mapsto \backtrfencec{ \trg{v}} : \backtrfencec{ \trg{\taint}} }
	\\
	\backtrfencec{\trg{\taint}} =&\ \src{\taint}
	\\\\
	\backtrfencec{ \trg{n} } =&\ \src{n}
	\\
	\backtrfencec{ \trg{e\op e'} } =&\  \src{\backtrfencec{ \trg{e} }\op\backtrfencec{ \trg{e'} }}
	\\
	\backtrfencec{ \trg{e\bop e'} } =&\ \src{\backtrfencec{ \trg{e} }\bop\backtrfencec{ \trg{e'} }}
	\\\\
	\backtrfencec{\trg{p_1;p_2}} &= \src{\backtrfencec{ \trg{p_1} } ; \backtrfencec{ \trg{p_2}} }
	\\
	\backtrfencec{ \trg{l : instr} } &= \src{\backtrfencec{ \trg{l} } : \backtrfencec{ \trg{instr} }}
	\\
	\backtrfencec{ \trg{l} } =&\ \src{l}
	\\\\
	\backtrfencec{\trg{\pskip}} &= \src{\pskip} 
	\\
	\backtrfencec{\trg{\passign{x}{e}}} &= \src{\passign{x}{\backtrfencec{\trg{e}}}}
        \\
        \backtrfencec{\trg{\pcondassign{x}{e}{e'}}} &= \src{\pcondassign{x}{\backtrfencec{\trg{e}}}{\backtrfencec{\trg{e'}}}}
	\\
	\backtrfencec{\trg{\pload{x}{e}}} &= \src{\pload{x}{\backtrfencec{\trg{e}}}}
	\\
	\backtrfencec{\trg{\pstore{x}{e}}} &= \src{ \pstore{x}{\backtrfencec{\trg{e}}}}
	\\
	\backtrfencec{\trg{\pjmp{e}}} &= \src{\pjmp{\backtrfencec{\trg{e}}}} \\
	\backtrfencec{\trg{\pjz{x}{l}}} &= \src{\pjz{x}{l}}
	\\
	\backtrfencec{\trg{\pcall{f}}} &= \src{\pcall{f} }
	\\
	\backtrfencec{\trg{\pret}} &= \src{ \pret}
	\\
	\backtrfencec{ \trg{\pbarrier} } =&\ \src{\pskip}
\end{align*}

\subsection{Additional Lemmas for the Compilers}

Here we define some Lemmas that all compilers we define share. These are mostly related to bookkeeping. All of these Lemmas follow from the inspection of the respective compiler.

First, the function maps are related between source and target. This means that compilation does not change the starting label of a source function during compilation (except for $\compuslhB{\cdot}$, however, we do not use these definitions there).
However, compilation can add new functions to the compiled program. For example, if a countermeasure needs an additional function like retpoline.

\begin{definition}
    We say that a function map $\mathcal{F}$ extends another function map $\mathcal{F'}$ written $\mathcal{F} \subset \mathcal{F'}$ iff 
    $\forall f \in \mathcal{F} \ldotp \mathcal{F}(f) = \mathcal{F'}(f)$.
\end{definition}
This means they have the same starting label and the same code.

\begin{definition}[Equivalence of Function Maps]
Two function maps $\mathcal{F}$ and $\mathcal{F'}$ are equivalent, written $\mathcal{F} \equiv \mathcal{F'}$ iff
$\mathcal{F} \subset \mathcal{F'}$ and $\mathcal{F'} \subset \mathcal{F}$
\end{definition}

This just means that labels do not change their function during compilation.
\begin{lemma}[$\ffun{}$ remains invariant]
$\forall \src{l} \in \src{p} \ldotp \ffun{l} = \src{f}$ then $\ffun{\comp{l}} = \comp{f}$ and $\src{l} = \comp{l}$ and $\src{f} = \comp{f}$.
\end{lemma}

\section{Speculative Semantics $\Bv$}

The speculative state $\trgB{\SigmaB}$ is a stack of speculative instances $\trgB{\PhiB}$ that contain the non-speculative state $\Omega$ and the speculation window $n$.
Note that we define $n + 1$ to match $\bot$ as well.
All elements from the source language exist in this language

\begin{align*}
    \mi{Instructions}~ \trgB{i} \bnfdef&\ \cdots \mid \pbarrier \\
    \mi{Speculative\ States}~ \trgB{\SigmaB} \bnfdef&\ \trgB{\phiStackB} \\
    \mi{Speculative\ Instance}~ \trgB{\PhiB} \bnfdef&\ \trgB{\tup{\Omega, n}} \\
    \mi{Speculative\ Instance Vals.}~ \trgB{\PhiBv} \bnfdef&\ \trgB{\tup{\Omegav, n}} \\
    \mi{Speculative\ Instance Taint}~ \trgB{\PhiBt} \bnfdef&\ \trgB{\tup{\taintpc, \Omegat, n}} \\ \mi{Windows}~ \trgB{w} \bnfdef&\ n \mid \bot \\
    \mi{Observations}~ \trgB{\tau} \bnfdef&\ \trgB{\rollbackObsB}
\end{align*}

\begin{center}
    \mytoprule{\trgB{\PhiBv \specarrowB{\tau} \phiStackBv'}}

    \typerule {$\Bvr$:AM-barr}
    {\trgB{\instrOv{\pbarrier}} & \trgB{\Omegav' = \Omegav[\pc \mapsto \pc + 1]}
    }
    {\trgB{\tup{\Omegav, \bot} \specarrowB{\epsilon} \tup{\Omegav', \bot}}
     }{v1-barr}
     \typerule {$\Bvr$:AM-barr-spec}
    {\trgB{\instrOv{\pbarrier}} & \trgB{\Omegav' = \Omegav[\pc \mapsto \pc + 1]}
    }
    {\trgB{\tup{\Omegav, n + 1} \specarrowB{\epsilon} \tup{\Omegav', 0}}
     }{v1-barr-spec}

    \typerule {$\Bvr$:AM-NoSpec-epsilon} {\trgB{\instrneqOv{\pjz{x}{l}, \pbarrier, Z}} &  \trgB{\Omegav \nsarrow{\epsilon} \Omegav'}
    }
    {\trgB{\tup{\Omegav, n + 1} \specarrowB{\epsilon} \tup{\Omegav', n}}
     }{v1-nospec-eps}
    \typerule {$\Bvr$:AM-NoSpec-action}
    {\trgB{\instrneqOv{\pjz{x}{l}, \pbarrier, Z}} & \trgB{\Omegav \nsarrow{\tau} \Omegav'}
    }
    {\trgB{\tup{\Omegav, n + 1} \specarrowB{\tau} \tup{\Omegav', n}}
     }{v1-nospec-act}
     
    \typerule{$\Bvr$:AM-Branch-Spec}
    {\trgB{\instrOv{\pjz{x}{l}}} &  \trgB{\Omegav \nsarrow{\tau} \Omegav'} & \trgB{\Omegav = \OB{F}; \OB{I}; \OB{B}; \sigmav}\\
    \ffun{\trgB{\Omegav(\pc)}} = \trgB{f} & \trgB{f \notin \OB{I}} & \trgB{\Omegav'' = \OB{F}; \OB{I}; \OB{B}; \sigmav''}\\
    l' = {\begin{cases} \sigmav(\pc) + 1 & \text{if $\sigmav'(\pc) = l$} \\
                        l & \text{if $\sigmav'(\pc) \neq l$}
    \end{cases}
    }\\
     \trgB{\sigmav'' = \sigmav[\pc \mapsto l']} & \trgB{j = min(\omega, n)} 
    }
    {\trgB{\tup{\Omegav, n + 1} \specarrowB{\tau} \tup{\Omegav', n} \cdot \tup{\Omegav'', j}}
    }{v1-spec}

    \typerule{$\Bvr$:AM-Branch-Spec-att}
    {\trgB{\instrOv{\pjz{x}{l}}} &  \trgB{\Omegav \nsarrow{\tau} \Omegav'} \\
    \ffun{\trgB{\Omegav(\pc)}} = \trgB{f} & \trgB{f \in \OB{I}} \\
    }
    {\trgB{\tup{\Omegav, n + 1} \specarrowB{\tau} \tup{\Omegav', n}}
    }{v1-skip-att}

    \typerule{$\Bvr$:AM-Rollback}
    { \trgB{n = 0}\ \text{or}\ \finType{\trgB{\Omegav}} }
    {\trgB{\tup{\Omegav, n} \specarrowB{\rollbackObsB} \varepsilon}
    }{v1-rollback}
    
\end{center}

\begin{center}
    \mytoprule{\trgB{\SigmaB \bigspecarrowB{\tauStack} \SigmaB'}}
    
    \typerule{$\Bvr$:AM-Reflection}
    {
    }
    {\trgB{\SigmaB \bigspecarrowB{\varepsilon} \SigmaB}
    }{v1-reflect}
    \typerule{$\Bvr$:AM-Single}
    {
    \trgB{\SigmaB \bigspecarrowB{\OB{\tau^{\taint}}} \SigmaB''} & \trgB{\SigmaB'' \specarrowB{\tau^{\taint}} \SigmaB'} \\
    \trgB{\SigmaB'' = \phiStackB \cdot  \tup{\OB{F}; \OB{I}; \OB{B}; \Omega, n}} &
    \trgB{\SigmaB' = \phiStackB \cdot \tup{\OB{F}; \OB{I}; \OB{B'}; \Omega', n' }} \\
    \ffun{\trgB{\Omega(\pc)}} = \trgB{f} & \ffun{\trgB{\Omega'(\pc)}} = \trgB{f'} \\
    \text{ if } \trgB{f == f'} and ~\trgB{f \in \OB{I}} \text{ then } \trgB{\tau^{\taint} = \epsilon} \text{ else } \trgB{\tau^{\taint} = \tau^{\taint}}
    }
    {
    \trgB{\SigmaB \bigspecarrowB{\OB{\tau^{\taint}} \cdot \tau^{\taint}} \SigmaB'}
    }{v1-single}
    
    \typerule{$\Bvr$:AM-silent}
    {
    \trgB{\SigmaB \bigspecarrowB{\tra{^\taint}} \SigmaB''} & \trgB{\SigmaB'' \specarrowB{\epsilon} \SigmaB'}
    }
    {
    \trgB{\SigmaB \bigspecarrowB{\tra{^\taint}} \SigmaB'}
    }{v1-silent}

\end{center}
Again, we ignore events that happen in the context.

\begin{center}
    \mytoprule{Helpers}
    
    \typerule{$\Bvr$:AM-Merge}
    {\trgB{\Omegav + \Omegat = \Omega }
    }
    {\trgB{\tup{\Omegav, n} + \tup{\taintpc, \Omegat, n} = \tup{\Omega, n , \taintpc}}
    }{v1-merge}
    
    \typerule{$\Bvr$:AM-Init}
    {\trgB{\SigmaB = \initFunc(M, \OB{F}, \OB{I}), \bot, \safeta}
    }
    {\trgB{\initFuncB(M ,\OB{F}, \OB{I}) = \SigmaB}
    }{v1-init}
    \typerule{$\Bvr$:AM-Fin-Ending}
    { \trgB{\SigmaB = (\Omega, \bot)} & \finType{\trgB{\Omega}}
    }
    {\finTypef{\trgB{\SigmaB}}
    }{v1-fin-end}
    
    \typerule{$\Bvr$:AM-Fin}
    { \trgB{\SigmaB = \phiStackB \cdot (\Omega, n)} & \finType{\trgB{\Omega}}
    }
    {\finType{\trgB{\SigmaB}}
    }{v1-fin}
    
    \typerule {$\Bvr$:AM-Trace}
    { \exists \finTypef{\trgB{\SigmaB'}}\ & \trgB{\initFuncB(P) \bigspecarrowB{\tauStackT}\ \SigmaB'}
    }
    { \trgB{\amTracevB{P}{\tauStackT}}
     }{v1-trace}
    \typerule {$\Bvr$:AM-Beh}
    {}
    { \trgB{\behB{P}} =  \{\trgB{\tauStackT} \mid \trgB{\amTracevB{P}{\tauStackT}} \}
     }{v1-beh}
\end{center}

 \section{Speculative Semantics $\Sv$}\label{app:semV4} 

The speculative state $\trgS{\SigmaS}$ is a stack of speculative instances $\trgS{\PhiS}$ that contain the non-speculative state $\Omega$ and the speculation window $n$.
Note that we define $n + 1$ to match $\bot$ as well.
All elements from the source language exist in this language

\begin{align*}
    \mi{Instructions}~ \trgS{i} \bnfdef&\ \cdots \mid \pbarrier \\
    \mi{Speculative\ States}~ \trgS{\SigmaS} \bnfdef&\ \trgS{\phiStackS} \\
    \mi{Speculative\ Instance}~ \trgS{\PhiS} \bnfdef&\ \trgS{\tup{\Omega, n, \taintpc}} \\
    \mi{Speculative\ Instance Vals.}~ \trgS{\PhiSv} \bnfdef&\ \trgS{\tup{\Omegav, n}} \\
    \mi{Speculative\ Instance Taint}~ \trgS{\PhiSt} \bnfdef&\ \trgS{\tup{\taintpc, \Omegat, n}} \\
    \mi{Windows}~ \trgS{w} \bnfdef&\ n \mid \bot \\
    \mi{Observations}~ \trgS{\tau} \bnfdef&\ \trgS{\rollbackObsS}
\end{align*}

We use the shorthand notation $\Omega(\pc)$ for $\sigma(\pc)$ with $\Omega = \OB{C};\OB{B}; \sigma$. The same notation works on $\Omegav$ and $\Omegat$ as well. We note that for $\Omegat$, the notation will still yield a value, since the value of the $\pc$ is tracked even when doing taint tracking

\mytoprule{Judgements}
\begin{align*}
    &
    \trgS{\phiStackSv \specarrowS{\tau} \phiStackSv'}
    && \text{Stack $\PhiSv$ small-steps to $\phiStackSv$ and emits observation $\tau$.}
    \\
    &
    \trgS{\PhiSv \specarrowS{\tau} \phiStackSv'}
    && \text{Speculative instance $\PhiSv$ small-steps to $\phiStackSv$ and emits observation $\tau$.}
    \\
    &
    \trgS{\phiStackSt \specarrowS{\taint} \phiStackSt}
    && \text{Stack $\phiStackSt$ small-steps to $\phiStackSt'$ and emits taint $\taint$.}
    \\
    &
    \trgS{\PhiSt \specarrowS{\taint} \phiStackSt'}
    && \text{Speculative instance $\PhiSt$ small-steps to $\phiStackSt'$ and emits taint $\taint$.}
    \\
    &
    \trgS{\SigmaS \specarrowS{\tau^{\taint}} \SigmaS'}
    && \text{Speculative state $\SigmaS$ small-steps to $\SigmaS'$ and emits tainted observation $\tau^{\taint}$.}
    \\
    &
    \trgS{\SigmaS \bigspecarrowS{\OB{\tau^{\taint}}} \SigmaS'}
    && \text{State $\SigmaS$ big-steps to $\SigmaS'$ and emits a list of observations $\tauStack$.}
    \\
    &
    \trgS{\amTracevS{P}{\tauStackT}}
    && \text{Program $P$ produces the observations $\tauStackT$ during execution.}
\end{align*}

\begin{center}
    \centering
    \small
    \mytoprule{\trgS{\phiStackSv \specarrowS{\tau} \phiStackSv'}}
    
    \typerule {$\Svr$:AM-Context}
    {\trgS{\PhiSv \specarrowS{\tau} \phiStackSv'}
    }
    {
    \trgS{
    \phiStackSv \cdot \PhiSv
    \specarrowS{\tau}
    \phiStackSv \cdot \phiStackSv'
    }
     }{v4-context}
     
     \mytoprule{\trgS{\PhiSv \specarrowS{\tau} \phiStackSv'}}
     
     \typerule {$\Svr$:AM-barr}
    {\trgS{\instrOv{\pbarrier}} & \trgS{\Omegav' = \Omegav[\pc \mapsto \pc + 1]}
    }
    {\trgS{\tup{\Omegav, \bot} \specarrowS{\epsilon} \tup{\Omegav', \bot}}
     }{v4-barr}
     \typerule {$\Svr$:AM-barr-spec}
    {\trgS{\instrOv{\pbarrier}} & \trgS{\Omegav' = \Omegav[\pc \mapsto \pc + 1]}
    }
    {\trgS{\tup{\Omegav, n + 1} \specarrowS{\epsilon} \tup{\Omegav', 0}}
     }{v4-barr-spec}
     
     \typerule {$\Svr$:AM-NoSpec-epsilon} {\trgS{\instrneqOv{\pstore{x}{e}, \pbarrier, Z}} &  \trgS{\Omegav \nsarrow{\epsilon} \Omegav'}
    }
    {\trgS{\tup{\Omegav, n + 1} \specarrowS{\epsilon} \tup{\Omegav', n}}
     }{v4-nospec-eps}
    \typerule {$\Svr$:AM-NoSpec-action}
    {\trgS{\instrneqOv{\pstore{x}{e}, \pbarrier, Z}} & \trgS{\Omegav \nsarrow{\tau} \Omegav'}
    }
    {\trgS{\tup{\Omegav, n + 1} \specarrowS{\tau} \tup{\Omegav', n}}
     }{v4-nospec-act}
    
    \typerule{$\Svr$:AM-Store-Spec}
    {\trgS{\instrOv{\pstore{x}{e}}} &  \trgS{\Omegav \nsarrow{\tau} \Omegav'} & \trgS{\Omegav = \OB{F}; \OB{I}; \OB{B}; \sigmav}\\
    \ffun{\trgS{\Omegav(\pc)}} = \trgS{f} & \trgS{f \notin \OB{I}} & \trgS{\Omegav'' = \OB{F}; \OB{I}; \OB{B}; \sigmav''}\\
     \trgS{\sigmav'' = \sigmav[\pc \mapsto \Omegav(\pc) + 1]} & \trgS{j = min(\omega, n)} 
    }
    {\trgS{\tup{\Omegav, n + 1} \specarrowS{\tau} \tup{\Omegav', n} \cdot \tup{\Omegav'', j}}
    }{v4-skip}
    \typerule{$\Svr$:AM-Rollback}
    { \trgS{n = 0}\ \text{or}\ \finType{\trgS{\Omegav}} }
    {\trgS{\tup{\Omegav, n} \specarrowS{\rollbackObsS} \varepsilon}
    }{v4-rollback}
    
    \typerule{$\Svr$:AM-Store-Spec-att}
    {\trgS{\instrOv{\pstore{x}{e}}} &  \trgS{\Omegav \nsarrow{\tau} \Omegav'} \\
    \ffun{\trgS{\Omegav(\pc)}} = \trgS{f} & \trgS{f \in \OB{I}} \\
    }
    {\trgS{\tup{\Omegav, n + 1} \specarrowS{\tau} \tup{\Omegav', n}}
    }{v4-skip-att}

\end{center}

We do not speculate when we are in the attacker (\Cref{tr:v4-skip-att}). Since we quantify over all contexts, we could find another attacker that does exactly what the speculation would do. In a sense, speculation does give the attacker no additional power.

\subsection{Taint Semantics}
\begin{center}
    \centering
    \small
    \mytoprule{\trgS{\phiStackSt \specarrowS{\taint} \phiStackSt}}
    
    \typerule {$\Svr$:T-Context}
    {\trgS{\PhiSt \specarrowS{\tau} \phiStackSt}
    }
    {\trgS{\phiStackSt \cdot \PhiSt \specarrowS{\tau} \phiStackSt \cdot \phiStackSt'}
     }{v4-t-context}

     \mytoprule{\trgS{\PhiSt \specarrowS{\taint} \phiStackSt'}}
     
     \typerule {$\Svr$:T-barr}
    {\trgS{\instrOt{\pbarrier}} & \trgS{\Omegat' = \Omegat[\pc \mapsto \pc + 1]}
    }
    {\trgS{\tup{\taintpc, \Omegat, \bot}
    \specarrowS{\epsilon}
    \tup{\taintpc, \Omegat', \bot}}
     }{v4-t-barr}
     \typerule {$\Svr$:T-barr-spec}
    {\trgS{\instrOt{\pbarrier}} & \trgS{\Omegat' = \Omegat[\pc \mapsto \pc + 1]}
    }
    {\trgS{\tup{\taintpc, \Omegat, n + 1}
    \specarrowS{\epsilon}
    \tup{\taintpc, \Omegat', 0}}
     }{v4-t-barr-spec}
     
      \typerule {$\Svr$:AM-NoSpec-epsilon} {\trgS{\instrneqOv{\pstore{x}{e}, \pbarrier, Z}} &  \trgS{\taintpc; \Omegat \nsarrow{\epsilon} \Omegat'}
    }
    {\trgS{\tup{\taintpc, \Omegat, n + 1} \specarrowS{\epsilon} \tup{\taintpc, \Omegat', n}}
     }{v4-t-nospec-eps}
     \typerule {$\Svr$:T-NoSpec-action}
    {\trgS{\instrneqOt{\pstore{x}{e}, \pbarrier, Z}} & \trgS{\taintpc;\Omegat \nsarrow{\taint'} \Omegat'}
    }
    {\trgS{\tup{\taintpc, \Omegat, n + 1} \specarrowS{\taintpc \glb \taint'} \tup{\taintpc, \Omegat', n}}
     }{v4-t-nospec-act}

    \typerule{$\Svr$:T-Store-Spec}
    {\trgS{\instrOt{\pstore{x}{e}}} & \trgS{\taintpc;\Omegat \nsarrow{\taint'} \Omegat'} & \trgS{\Omegat = \OB{F}; \OB{I}; \OB{B}; \sigmat} \\
    \ffun{\trgS{\Omegav(\pc)}} = \trgS{f} & \trgS{f \notin \OB{I}} & \trgS{\Omegat'' = \OB{F}; \OB{I}; \OB{B}; \sigmat''} \\
     \trgS{\sigmat'' = \sigmat[\pc \mapsto \Omegat(\pc) + 1 : \unta]} & \trgS{j = min(\omega, n)}
    }
    {\trgS{\tup{\taintpc, \Omegat, n + 1} \specarrowS{\taintpc \glb \taint'}
    \tup{\taintpc, \Omegat', n}\cdot \tup{\unta, \Omegat'', j}} }{v4-t-skip}
    \typerule{$\Svr$:T-Store-Spec-att}
    {\trgS{\instrOt{\pstore{x}{e}}} & \trgS{\taintpc;\Omegat \nsarrow{\taint'} \Omegat'}  \\
    \ffun{\trgS{\Omegav(\pc)}} = \trgS{f} & \trgS{f \notin \OB{I}} \\
    }
    {\trgS{\tup{\taintpc, \Omegat, n + 1} \specarrowS{\taintpc \glb \taint'} \tup{\taintpc, \Omegat', n}}
    }{v4-t-skip-att}
    
    \typerule{$\Svr$:T-Rollback}
    { \trgS{n = 0} \text{ or } \finType{\trgS{\Omegat}} }
    { 
    \trgS{\tup{\taintpc, \Omegat, n } \specarrowS{\safeta} \emptyset}
    }{v4-t-rollback}
    
\end{center}

\begin{center}
    \centering
    \small
    \mytoprule{\trgS{\SigmaS \specarrowS{\tau^{\taint}} \SigmaS'}}
    
    \typerule {V4:Combine}
    {\trgS{\SigmaS = \phiStackS{}} & \trgS{\SigmaS = \phiStackS'} \\
    \trgS{\phiStackSv + \phiStackSt = \phiStackS} &  \trgS{\phiStackSv' + \phiStackSt' = \phiStackS'} \\
    \trgS{\phiStackSv \specarrowS{\tau} \phiStackSv'} & \trgS{\phiStackSt \specarrowS{\taint} \phiStackSt'}
    }
    {\trgS{\SigmaS \specarrowS{\tau^{\taint}} \SigmaS'}
     }{v4-combine}
\end{center}

\begin{center}
    \mytoprule{\trgS{\SigmaS \bigspecarrowS{\tauStack} \SigmaS'}}
    
    \typerule{$\Svr$:AM-Reflection}
    {
    }
    {\trgS{\SigmaS \bigspecarrowS{\varepsilon} \SigmaS}
    }{v4-reflect}
    \typerule{$\Svr$:AM-Single}
    {
    \trgS{\SigmaS \bigspecarrowS{\OB{\tau^{\taint}}} \SigmaS''} & \trgS{\SigmaS'' \specarrowS{\tau^{\taint}} \SigmaS'} \\
    \trgS{\SigmaS'' = \phiStackS \cdot  \tup{\OB{F}; \OB{I}; \OB{B}; \Omega, n}} &
    \trgS{\SigmaS' = \phiStackS \cdot \tup{\OB{F}; \OB{I}; \OB{B'}; \Omega', n' }} \\
    \ffun{\trgS{\Omega(\pc)}} = \trgS{f} & \ffun{\trgS{\Omega'(\pc)}} = \trgS{f'} \\
    \text{ if } \trgS{f == f'} and ~\trgS{f \in \OB{I}} \text{ then } \trgS{\tau^{\taint} = \epsilon} \text{ else } \trgS{\tau^{\taint} = \tau^{\taint}}
    }
    {
    \trgS{\SigmaS \bigspecarrowS{\OB{\tau^{\taint}} \cdot \tau^{\taint}} \SigmaS'}
    }{v4-single}
    
    \typerule{$\Svr$:AM-silent}
    {
    \trgS{\SigmaS \bigspecarrowS{\tra{^\taint}} \SigmaS''} & \trgS{\SigmaS'' \specarrowS{\epsilon} \SigmaS'}
    }
    {
    \trgS{\SigmaS \bigspecarrowS{\tra{^\taint}} \SigmaS'}
    }{v4-silent}

\end{center}
Again, we ignore events that happen in the context.

\begin{center}
    \mytoprule{Helpers}
    
    \typerule{$\Svr$:AM-Merge}
    {\trgS{\Omegav + \Omegat = \Omega }
    }
    {\trgS{\tup{\Omegav, n} + \tup{\taintpc, \Omegat, n} = \tup{\Omega, n , \taintpc}}
    }{v4-merge}
    \typerule{$\Svr$:AM-Init}
    {\trgS{\SigmaS = \initFunc(M, \OB{F}, \OB{I}), \bot, \safeta}
    }
    {\trgS{\initFuncS(M ,\OB{F}, \OB{I}) = \SigmaS}
    }{v4-init}
    
    \typerule{$\Svr$:AM-Fin-Ending}
    { \trgS{\SigmaS = (\Omega, \bot)} & \finType{\trgS{\Omega}}
    }
    {\finTypef{\trgS{\SigmaS}}
    }{v4-fin-end}
    \typerule{$\Svr$:AM-Fin}
    { \trgS{\SigmaS = \phiStackS \cdot (\Omega, n)} & \finType{\trgS{\Omega}}
    }
    {\finType{\trgS{\SigmaS}}
    }{v4-fin}
    
    \typerule {$\Svr$:AM-Trace}
    { \exists \finTypef{\trgS{\SigmaS'}}\ & \trgS{\initFuncS(P) \bigspecarrowS{\tauStackT}\ \SigmaS'}
    }
    { \trgS{\amTracevS{P}{\tauStackT}}
     }{v4-trace}
    \typerule {$\Svr$:AM-Beh}
    {}
    { \trgS{\behS{P}} =  \{\trgS{\tauStackT} \mid \trgS{\amTracevS{P}{\tauStackT}} \}
     }{v4-beh}
\end{center}

The speculative semantics starts with an initial value of $\safeta$ for the $\taintpc$. When speculation happens (\Cref{tr:v4-t-skip}) it will be set to $\unta$ for the newly created speculative instance and after rollback (\Cref{tr:v4-t-rollback}) the previous value of $\taintpc$.
 \section{Speculative Semantics $\Jvr$}\label{app:semV2}
A detailed description is given in \Cref{app:semV4}.

\begin{align*}
    \mi{Instructions}~ \trgJ{i} \bnfdef&\ \cdots \mid \pbarrier \\
    \mi{Speculative\ States}~ \trgJ{\SigmaJ} \bnfdef&\ \trgJ{\phiStackJ} \\
    \mi{Speculative\ Instance}~ \trgJ{\PhiJ} \bnfdef&\ \trgJ{\tup{\Omega, n, \taintpc}} \\
    \mi{Speculative\ Instance Vals.}~ \trgJ{\PhiJv} \bnfdef&\ \trgJ{\tup{\Omegav, n}} \\
    \mi{Speculative\ Instance Taint}~ \trgJ{\PhiJt} \bnfdef&\ \trgJ{\tup{\taintpc, \Omegat, n}} \\ \mi{Windows}~ \trgJ{w} \bnfdef&\ \trgJ{n} \mid \trgJ{\bot} \\
    \mi{Observations}~ \trgJ{\tau} \bnfdef&\ \trgJ{\rollbackObsJ}
\end{align*}

\mytoprule{Judgements}
\begin{align*}
    &
    \trgJ{\SigmaJ \specarrowJ{\tau} \SigmaJ'}
    && \text{State $\trgJ{\SigmaJ}$ small-steps to $\trgJ{\SigmaJ'}$ and emits observation $\trgJ{\tau}$.}
    \\
    &
    \trgJ{\PhiJ \specarrowJ{\tau} \phiStackJ'}
    && \text{Speculative instance $\trgJ{\PhiJ}$ small-steps to $\trgJ{\phiStackJ{'}}$ and emits observation $\trgJ{\tau}$.}
    \\
    &
    \trgJ{\SigmaJ \specarrowJ{\taint} \SigmaJ'}
    && \text{State $\trgJ{\SigmaJ}$ small-steps to $\trgJ{\SigmaJ'}$ and emits taint $\trgJ{\taint}$.}
    \\
    &
    \trgJ{\PhiJ \specarrowJ{\taint} \phiStackJ'}
    && \text{Speculative instance $\trgJ{\PhiJ}$ small-steps to $\trgJ{\phiStackJ{'}}$ and emits taint $\trgJ{\taint}$.}
    \\
    &
    \trgJ{\SigmaJ \specarrowJ{\tau^{\taint}} \SigmaJ'}
    && \text{Speculative state $\trgJ{\SigmaJ}$ small-steps to $\trgJ{\SigmaJ'}$ and emits tainted observation $\trgJ{\tau^{\taint}}$.}
    \\
    &
    \trgJ{\SigmaJ \bigspecarrowJ{\OB{\tau^{\taint}}} \SigmaJ'}
    && \text{State $\trgJ{\SigmaJ}$ big-steps to $\trgJ{\SigmaJ'}$ and emits a list of tainted observations $\trgJ{\tauStackT}$.}
    \\
    &
    \trgJ{\amTracevJ{P}{\tauStackT}}
    && \text{Program $\trgJ{P}$ produces the tainted observations $\trgJ{\tauStackT}$ during execution.}
\end{align*}

\begin{center}
    \centering
    \small
    \mytoprule{\trgJ{\phiStackJv \specarrowJ{\tau} \phiStackJv'}}
    
    \typerule {$\Jvr$:AM-Context}
    {\trgJ{\PhiJv \specarrowJ{\tau} \phiStackJv'}
    }
    {\trgJ{\phiStackJv \cdot \PhiJv \specarrowJ{\tau} \phiStackJv \cdot \phiStackJv'}
     }{v2-context}
     
    \mytoprule{\trgJ{\PhiJv \specarrowJ{\tau} \phiStackJv'}}
     
     \typerule {$\Jvr$:barr}
    {\trgJ{\instrOv{\pbarrier}} & \trgJ{\Omegav' = \Omegav[\pc \mapsto \pc + 1]}
    }
    {
    \trgJ{\tup{\Omegav, \bot} \specarrowJ{\epsilon} \tup{\Omegav', \bot}}
     }{v2-barr}
     \typerule {$\Jvr$:barr-spec}
    {\trgJ{\instrOv{\pbarrier}} & \trgJ{\Omegav' = \Omegav[\pc \mapsto \pc + 1]}
    }
    {\trgJ{\tup{\Omegav, n + 1} \specarrowJ{\epsilon} \tup{\Omegav', 0}}
     }{v2-barr-spec}
     
      \typerule {$\Jvr$:AM-NoSpec-epsilon} {\trgJ{\instrneqOv{\pjmp{x}, \pbarrier, Z}} & \trgJ{\Omegav \nsarrow{\epsilon} \Omegav'}
    }
    {\trgJ{\tup{\Omegav, n + 1} \specarrowJ{\epsilon} \tup{\Omegav', n}}
     }{v2-nospec-eps}
    \typerule {$\Jvr$:AM-NoSpec-action}
    {\trgJ{\instrneqOv{\pjmp{e}, \pbarrier, Z}} & \trgJ{\Omegav \nsarrow{\tau} \Omegav'}
    }
    {\trgJ{\tup{\Omegav, n + 1} \specarrowJ{\tau} \tup{\Omegav', n}}
     }{v2-nospec-act}

    \typerule{$\Jvr$:AM-Jmp-Spec}
    {\trgJ{\instrOv{\pjmp{x}}} & \trgJ{x \in \Reg} & \trgJ{\Omegav \nsarrow{\tau} \Omegav'}  &  \trgJ{\Omegav = \OB{F}; \OB{I}; \OB{B}; \sigmav}\\
      \ffun{\trgJ{\Omegav'(\pc)}} = \trgJ{f'} & \ffun{\trgJ{\Omegav(\pc)}} = \trgJ{f} & \trgJ{f \notin \OB{I}} &  \src{\OB{I}}\vdash\src{f,f'}:\src{internal} \\
    \trgJ{j = min(\omega, n)} & \trgJ{\OB{\SigmaJ''}} = \bigcup_{l \in p} \trgJ{\tup{\Omegav, j}} \text{ where } \trgJ{\Omegav = \OB{F}; \OB{I}; \OB{B}; \sigmav[\pc \mapsto l]} \text{ and } \ffun(l) = f \text{ and } f \in \OB{F}
    }
    {\trgJ{\tup{\Omegav, n + 1} \specarrowJ{\tau} \tup{\Omegav', n} \cdot \OB{\SigmaJ''}}
    }{v2-spec}
    \typerule{$\Jvr$:AM-Rollback}
    { \trgJ{n = 0}\ \text{ or }\ \finType{\trgJ{\Omegav}} }
    {\trgJ{\tup{\Omegav, n} \specarrowJ{\rollbackObsJ} \emptyset}
    }{v2-rollback}
    
    \typerule{$\Jvr$:AM-Jmp-att}
    {\trgJ{\instrOv{\pjmp{e}}} &  \trgJ{\Omegav \nsarrow{\tau} \Omegav'} &  \trgJ{\Omegav = \OB{F}; \OB{I}; \OB{B}; \sigmav}\\
    \ffun{\trgJ{\Omegav(\pc)}} = \trgJ{f} & \trgJ{f \in \OB{I}} \\
    }
    {\trgJ{\tup{\Omegav, n + 1} \specarrowJ{\tau} \tup{\Omegav', n}}
    }{v2-spec-att}
\end{center}

Here $\bigcup$ is the concatenation operator for states.
We allow speculation to only jump to any location in the component.
\subsection{Taint semantics}

\begin{center}
    \centering
    \small
    \mytoprule{\trgJ{\phiStackJt \specarrowJ{\taint} \phiStackJt'}}
    
    \typerule {$\Jvr$:AM-Context}
    {\trgJ{\PhiJt \specarrowJ{\taint} \phiStackJt'}
    }
    {\trgJ{\phiStackJt \cdot \PhiJt \specarrowJ{\taint} \phiStackJt \cdot \phiStackJt'}
     }{v2-t-context}
     
    \mytoprule{\trgJ{\PhiJt \specarrowJ{\taint} \phiStackJt'}}
     
     \typerule {$\Jvr$:barr}
    {\trgJ{\instrOt{\pbarrier}} & \trgJ{\Omegat' = \Omegat[\pc \mapsto \pc + 1]}
    }
    {
    \trgJ{\tup{\taintpc, \Omegat, \bot} \specarrowJ{\epsilon} \tup{\taintpc; \Omegat', \bot}}
     }{v2-t-barr}
     \typerule {$\Jvr$:barr-spec}
    {\trgJ{\instrOt{\pbarrier}} & \trgJ{\Omegat' = \Omegat[\pc \mapsto \pc + 1]}
    }
    {\trgJ{\tup{\taintpc; \Omegat, n + 1} \specarrowJ{\epsilon} \tup{\taintpc; \Omegat', 0}}
     }{v2-t-barr-spec}
     
      \typerule {$\Jvr$:AM-NoSpec-epsilon} {\trgJ{\instrneqOt{\pjmp{e}, \pbarrier, Z}} & \trgJ{\taintpc;\Omegat \nsarrow{\epsilon} \Omegat'}
    }
    {\trgJ{\tup{\taintpc, \Omegat, n + 1} \specarrowJ{\epsilon} \tup{\taintpc, \Omegat', n}}
     }{v2-t-nospec-eps}
    \typerule {$\Jvr$:AM-NoSpec-action}
    {\trgJ{\instrneqOt{\pjmp{e}, \pbarrier, Z}} & \trgJ{\Omegav \nsarrow{\taint'} \Omegat'}
    }
    {\trgJ{\tup{\taintpc, \Omegat, n + 1} \specarrowJ{\taintpc \glb \taint'} \tup{\taintpc, \Omegat', n}}
     }{v2-t-nospec-act}

    \typerule{$\Jvr$:AM-Jmp-Spec}
    {\trgJ{\instrOt{\pjmp{e}}} & \trgJ{\Omegat \nsarrow{\taint'} \Omegat'}  &  \trgJ{\Omegat = \OB{F}; \OB{I}; \OB{B}; \sigmat}\\
    \trgJ{\exprEval{a}{e}{l : \taint''}} &  \trgJ{l' \in \labelset} &  \trgJ{\Omegat = \OB{F}; \OB{I}; \OB{B}; \sigmat''}\\
    \ffun{\trgJ{l'}} = \trgJ{f'} & \ffun{\trgJ{l}} = \trgJ{f} & \trgJ{f \notin \OB{I}}  & \trgJ{j = min(\omega, n)} \\
    \src{\OB{I}}\vdash\src{f,f'}:\src{internal}  & \trgJ{\OB{\SigmaJ''}} = \bigcup_{l \in p} \trgJ{\tup{\Omegat, j}} \text{ where } \trgJ{\Omegat = \OB{F}; \OB{I}; \OB{B}; \sigmav[\pc \mapsto l]} \text{ and } \ffun(l) = f \text{ and } f \in \OB{F}
    }
    {\trgJ{\tup{\taintpc, \Omegat, n + 1} \specarrowJ{\taintpc \glb \taint'} \tup{\taintpc, \Omegat', n} \cdot \tup{\unta, \Omegat'', j}} }{v2-t-spec}
    \typerule{$\Jvr$:AM-Rollback}
    { \trgJ{n = 0}\ \text{ or }\ \finType{\trgJ{\Omegat}} }
    {\trgJ{\tup{\taintpc, \Omegat, n} \specarrowJ{\safeta} \emptyset}
    }{v2-t-rollback}
    
    \typerule{$\Jvr$:AM-Jmp-att}
    {\trgJ{\instrOt{\pjmp{e}}} &  \trgJ{\Omegat \nsarrow{\taint'} \Omegat'} &  \trgJ{\Omegat = \OB{F}; \OB{I}; \OB{B}; \sigmat}\\
    \ffun{\trgJ{\Omegat(\pc)}} = \trgJ{f} & \trgJ{f \in \OB{I}} \\
    }
    {\trgJ{\tup{\taintpc, \Omegat, n + 1} \specarrowJ{\taintpc \glb \taint'} \tup{\taintpc, \Omegat', n}}
    }{v2-t-spec-att}
\end{center}

\begin{center}
    \centering
    \small
    \mytoprule{\trgJ{\SigmaJ \specarrowJ{\tau^{\taint}} \SigmaJ'}}
    
    \typerule {$\Jvr$:Combine}
    {\trgJ{\SigmaJ = \phiStackJ{}} & \trgJ{\SigmaJ = \phiStackJ'} \\
    \trgJ{\phiStackJv + \phiStackJt = \phiStackJ} &  \trgJ{\phiStackJv' + \phiStackJt' = \phiStackJ'} \\
    \trgJ{\phiStackJv \specarrowJ{\tau} \phiStackJv'} & \trgJ{\phiStackJt \specarrowJ{\taint} \phiStackJt'}
    }
    {\trgJ{\SigmaJ \specarrowJ{\tau^{\taint}} \SigmaJ'}
     }{v2-combine}
\end{center}

Note that speculation only happens on indirect jumps in \Cref{tr:v2-spec-att}. An indirect jump happens when the address to jump to has to be fetched. This can happen if the jump address/label is inside a register. All other jumps are executed non-speculatively.

\begin{center}
    \mytoprule{\trgJ{\SigmaJ \bigspecarrowJ{\tauStack} \SigmaJ'}}
    
    \typerule{$\Jvr$:AM-Reflection}
    {
    }
    {\trgJ{\SigmaJ \bigspecarrowJ{\varepsilon} \SigmaJ}
    }{v2-reflect}
    \typerule{$\Jvr$:AM-Single}
    {
    \trgJ{\SigmaJ \bigspecarrowJ{\OB{\tau^{\taint}}} \SigmaJ''} & \trgJ{\SigmaJ'' \specarrowJ{\tau^{\taint}} \SigmaJ'} \\
    \trgJ{\SigmaJ'' = \phiStackJ \cdot  \tup{\OB{F}; \OB{I}; \OB{B}; \sigma, n}} &
    \trgJ{\SigmaJ' = \phiStackJ \cdot \tup{\OB{F}; \OB{I}; \OB{B'}; \sigma', n' }} \\
    \ffun{\trgJ{\sigma(\pc)}} = \trgJ{f} & \ffun{\trgJ{\sigma'(\pc)}} = \trgJ{f'} \\
    \text{ if } \trgJ{f == f'} and ~\trgJ{f \in \OB{I}} \text{ then } \trgJ{\tau^{\taint} = \epsilon} \text{ else } \trgJ{\tau^{\taint} = \tau^{\taint}}
    }
    {
    \trgJ{\SigmaJ \bigspecarrowJ{\OB{\tau^{\taint}} \cdot \tau^{\taint}} \SigmaJ'}
    }{v2-single}
    
    \typerule{$\Jvr$:AM-silent}
    {
    \trgJ{\SigmaJ \bigspecarrowJ{\tra{^\taint}} \SigmaJ''} & \trgJ{\SigmaJ'' \specarrowJ{\epsilon} \SigmaJ'}
    }
    {
    \trgJ{\SigmaJ \bigspecarrowJ{\tra{^\taint}} \SigmaJ'}
    }{v2-silent}
\end{center}

\begin{center}
    \mytoprule{Helpers}
    
    \typerule{$\Jvr$:AM-Merge}
    {\trgJ{\Omegav + \Omegat = \Omega }
    }
    {\trgJ{\tup{\Omegav, n} + \tup{\taintpc, \Omegat, n} = \tup{\Omega, n , \taintpc}}
    }{v2-merge}
    \typerule{$\Jvr$:AM-Init}
    {\trgJ{\SigmaJ = \initFunc(M, \OB{F}, \OB{I}), \bot, \safeta}
    }
    {\trgJ{\initFuncJ(M ,\OB{F}, \OB{I}) = \SigmaJ}
    }{v2-init}
    
    \typerule{$\Jvr$:AM-Fin-Ending}
    { \trgJ{\SigmaJ = (\Omega, \bot, \taintpc)} & \finType{\trgJ{\Omega}}
    }
    {\finTypef{\trgJ{\SigmaJ}}
    }{v2-fin-end}
    \typerule{$\Jvr$:AM-Fin}
    { \trgJ{\SigmaJ = \phiStackJ \cdot (\Omega, n, \taintpc)} & \finType{\trgJ{\Omega}}
    }
    {\finType{\trgJ{\SigmaJ}}
    }{v2-fin}
    
    \typerule {$\Jvr$:AM-Trace}
    { \exists \finTypef{\trgJ{\SigmaJ'}}\ & \trgJ{\initFuncJ(P) \bigspecarrowJ{\tauStackT}\ \SigmaJ'}
    }
    { \trgJ{\amTracevJ{P}{\tauStackT}}
     }{v2-trace}
    \typerule {$\Jvr$:AM-Beh}
    {}
    { \trgJ{\behJ{P}} =  \{\trgJ{\tauStackT} \mid \trgJ{\amTracevJ{P}{\tauStackT}} \}
     }{v2-beh}
\end{center}

 \section{Speculative Semantics $\Rv$}\label{app:semV5}

\begin{align*}
    \mi{Instructions}~ \trgR{i} \bnfdef&\ \cdots \mid \pbarrier \\
    \mi{Speculative\ States}~ \trgR{\SigmaR} \bnfdef&\ \trgR{\phiStackR} \\
    \mi{Speculative\ Instance}~ \trgR{\PhiR} \bnfdef&\ \trgR{\tup{\Omega, \Rsb, n, \taintpc}} \\
    \mi{Speculative\ Instance Vals.}~ \trgR{\PhiRv} \bnfdef&\ \trgR{\tup{\Omegav, \Rsb, n}} \\
    \mi{Speculative\ Instance Taint}~ \trgR{\PhiRt} \bnfdef&\ \trgR{\tup{\taintpc, \Omegat, \Rsb, n}} \\ \mi{Windows}~ \trgR{w} \bnfdef&\ \trgR{n} \mid \trgR{\bot} \\
    \mi{Observations}~ \trgR{\tau} \bnfdef&\ \trgR{\rollbackObsR}
\end{align*}

The return stack buffer (RSB) is infinite in our model. Typically, the RSB is limited in size in actual CPUs and tracks the last $N$ calls. 
For the modified retpoline countermeasure \cite{ret2spec}, it is only important what is at the top of the RSB, which is also correctly tracked in our model. That is why we simplify the model here. The different underflow behaviours (\cite{ret2spec}) are not of interest here, because we are interested in Compiler Countermeasures against speculation and not in attacks using speculation. Again, we use the simplest behaviour here by just stopping speculation when the RSB is empty.

\mytoprule{Judgements}
\begin{align*}
    &
    \trgR{\SigmaR \specarrowR{\tau} \SigmaR'}
    && \text{State $\trgR{\SigmaR}$ small-steps to $\trgR{\SigmaR'}$ and emits observation $\trgR{\tau}$.}
    \\
    &
    \trgR{\PhiR \specarrowR{\tau} \phiStackR'}
    && \text{Speculative instance $\trgR{\PhiR}$ small-steps to $\trgR{\phiStackR{'}}$ and emits observation $\trgR{\tau}$.}
    \\
    &
    \trgR{\SigmaR \specarrowR{\taint} \SigmaR'}
    && \text{State $\trgR{\SigmaR}$ small-steps to $\trgR{\SigmaR'}$ and emits taint $\trgR{\taint}$.}
    \\
    &
    \trgR{\PhiR \specarrowR{\taint} \phiStackR'}
    && \text{Speculative instance $\trgR{\PhiR}$ small-steps to $\trgR{\phiStackR{'}}$ and emits taint $\trgR{\taint}$.}
    \\
    &
    \trgR{\SigmaR \specarrowR{\tau^{\taint}} \SigmaR'}
    && \text{Speculative state $\trgR{\SigmaR}$ small-steps to $\trgR{\SigmaR'}$ and emits tainted observation $\trgR{\tau^{\taint}}$.}
    \\
    &
    \trgR{\SigmaR \bigspecarrowR{\OB{\tau^{\taint}}} \SigmaR'}
    && \text{State $\trgR{\SigmaR}$ big-steps to $\trgR{\SigmaR'}$ and emits a list of tainted observations $\trgR{\tauStackT}$.}
    \\
    &
    \trgR{\amTracevR{P}{\tauStackT}}
    && \text{Program $\trgR{P}$ produces the tainted observations $\trgR{\tauStackT}$ during execution.}
\end{align*}

\begin{center}
    \centering
    \small
    \mytoprule{\trgR{\phiStackRv \specarrowR{\tau} \phiStackRv'}}
    
    \typerule {$\Rvr$:AM-Context}
    {\trgR{\PhiRv \specarrowR{\tau} \phiStackRv'}
    }
    {\trgR{\phiStackRv \cdot \PhiRv \specarrowR{\tau} \phiStackRv \cdot \phiStackRv'}
     }{v5-context}
     
    \mytoprule{\trgR{\PhiRv \specarrowR{\tau} \phiStackRv'}}
     
     \typerule {$\Rvr$:barr}
    {\trgR{\instrOv{\pbarrier}} & \trgR{\Omegav' = \Omegav[\pc \mapsto \pc + 1]}
    }
    {
    \trgR{\tup{\Omegav, \Rsb, \bot} \specarrowR{\epsilon} \tup{\Omegav', \Rsb, \bot}}
     }{v5-barr}
     \typerule {$\Rvr$:barr-spec}
    {\trgR{\instrOv{\pbarrier}} & \trgR{\Omegav' = \Omegav[\pc \mapsto \pc + 1]}
    }
    {\trgR{\tup{\Omegav, \Rsb, n + 1} \specarrowR{\epsilon} \tup{\Omegav', \Rsb, 0}}
     }{v5-barr-spec}
     
      \typerule {$\Rvr$:AM-NoSpec-epsilon} {\trgR{\instrneqOv{\pcall{f}, \pret, \pbarrier, Z}} & \trgR{\Omegav \nsarrow{\epsilon} \Omegav'}
    }
    {\trgR{\tup{\Omegav, \Rsb, n + 1} \specarrowR{\epsilon} \tup{\Omegav', \Rsb, n}}
     }{v5-nospec-eps}
    \typerule {$\Rvr$:AM-NoSpec-action}
    {\trgR{\instrneqOv{\pcall{f}, \pret, \pbarrier, Z}} & \trgR{\Omegav \nsarrow{\tau} \Omegav'}
    }
    {\trgR{\tup{\Omegav, \Rsb, n + 1} \specarrowR{\tau} \tup{\Omegav', \Rsb, n}}
     }{v5-nospec-act}

    \typerule{$\Rvr$:AM-Call}
    {\trgR{\instrOv{\pcall{f}}} &  \trgR{\Omegav \nsarrow{\tau} \Omegav'} \\ 
      \trgR{\Rsb'} = \trgR{\Rsb \cdot (\Omegav(\pc) + 1)}  \\
       \ffun{\trgR{\Omegav(\pc)}} = \trgR{f'} & \trgR{f, f' \notin \OB{I}} 
    }
    {
    \trgR{\tup{\Omegav, \Rsb, n + 1} \specarrowR{\tau} \tup{\Omegav', \Rsb', n}}
    }{v5-call}
    \typerule{$\Rvr$:AM-Call-att}
    {\trgR{\instrOv{\pcall{f}}} &  \trgR{\Omegav \nsarrow{\tau} \Omegav'} \\ 
     \ffun{\trgR{\Omegav(\pc)}} = \trgR{f} & \trgR{f \in \OB{I}} 
    }
    {
    \trgR{\tup{\Omegav, \Rsb, n + 1} \specarrowR{\tau} \tup{\Omegav', \Rsb, n}}
    }{v5-call-att}
     
    \typerule{$\Rvr$:AM-Ret-Same}
    {\trgR{\instrOv{\pret}} & \trgR{\Omegav \nsarrow{\tau} \Omegav'} & \trgR{\Omegav = \OB{F}; \OB{I}; \OB{B}; \sigmav''} & \trgR{\Rsb = \Rsb' ; l} & \trgR{B = B' ; l'}  \\
    \trgR{l} = \trgR{l'} & \ffun{\trgR{\Omegav(\pc)}} = \trgR{f} &
    \ffun{\trgR{\Omegav'(\pc)}} = \trgR{f'} & \trgR{f, f' \notin \OB{I}} 
    }
    {
    \trgR{\tup{\Omegav, \Rsb, n + 1} \specarrowR{\tau} \tup{\Omegav', \Rsb', n}}
     }{v5-retS}
     \typerule{$\Rvr$:AM-Ret-Empty}
    {
    \trgR{\instr{\pret}} & \trgR{\Omegav \nsarrow{\tau} \Omegav'}
    }
    {
    \trgR{\tup{\Omegav, \emptyset, n + 1} \specarrowR{\tau} \tup{\Omegav', \emptyset, n}}
    }{v5-retE}

    \typerule{$\Rvr$:AM-Ret-Spec}
    {\trgR{\instrOv{\pret}} & \trgR{\Omegav \nsarrow{\tau} \Omegav'}  &  \trgR{\Omegav = \OB{F}; \OB{I}; \OB{B}; \sigmav}\\
    \trgR{\Rsb = \Rsb' \cdot l}  &  \trgR{l} \neq \trgR{B(0)}  &  \trgR{\Omegav'' = \OB{F}; \OB{I}; \OB{B'}; \sigmav''}\\
    \ffun{\trgR{\Omegav(\pc)}} = \trgR{f} &
    \ffun{\trgR{\Omegav'(\pc)}} = \trgR{f'} & \trgR{f, f' \notin \OB{I}} \\
     \trgR{\sigmav'' = \sigmav[\pc \mapsto l]} & \trgR{j = min(\omega, n)}
    }
    {\trgR{\tup{\Omegav, \Rsb, n + 1} \specarrowR{\tau} \tup{\Omegav', \Rsb', n} \cdot \tup{\Omegav'', \Rsb', j}}
    }{v5-spec}
    \typerule{$\Rvr$:AM-Rollback}
    { \trgR{n = 0}\ \text{ or }\ \finType{\trgR{\Omegav}} }
    {\trgR{\tup{\Omegav, \Rsb, n} \specarrowR{\rollbackObsR} \emptyset}
    }{v5-rollback}
    
    \typerule{$\Rvr$:AM-Ret-att}
    {\trgR{\instrOv{\pret}} &  \trgR{\Omegav \nsarrow{\tau} \Omegav'} &  \trgR{\Omegav = \OB{F}; \OB{I}; \OB{B}; \sigmav}\\
    \ffun{\trgR{\Omegav(\pc)}} = \trgR{f} & \trgR{f \in \OB{I}} \\
    }
    {\trgR{\tup{\Omegav, \Rsb, n + 1} \specarrowR{\tau} \tup{\Omegav', \Rsb, n}}
    }{v5-spec-att}
\end{center}

We only speculate internally, not between cross-components. That means a return in the component to the context does not trigger speculation. Furthermore, call from context to attacker does not modify the RSB $\Rsb$. This means the RSB only contains addresses of the component.
Call and return observations are generally not unsafe. That is why we do not use the glb there.
\subsection{Taint semantics}

\begin{center}
    \centering
    \small
    \mytoprule{\trgR{\phiStackRt \specarrowR{\taint} \phiStackRt'}}
    
    \typerule {$\Rvr$:AM-Context}
    {\trgR{\PhiRt \specarrowR{\taint} \phiStackRt'}
    }
    {\trgR{\phiStackRt \cdot \PhiRt \specarrowR{\taint} \phiStackRt \cdot \phiStackRt'}
     }{v5-t-context}
     
    \mytoprule{\trgR{\PhiRt \specarrowR{\taint} \phiStackRt'}}
     
     \typerule {$\Rvr$:barr}
    {\trgR{\instrOt{\pbarrier}} & \trgR{\Omegat' = \Omegat[\pc \mapsto \pc + 1]}
    }
    {
    \trgR{\tup{\taintpc, \Omegat, \Rsb, \bot} \specarrowR{\epsilon} \tup{\taintpc; \Omegat', \Rsb, \bot}}
     }{v5-t-barr}
     \typerule {$\Rvr$:barr-spec}
    {\trgR{\instrOt{\pbarrier}} & \trgR{\Omegat' = \Omegat[\pc \mapsto \pc + 1]}
    }
    {\trgR{\tup{\taintpc; \Omegat, \Rsb, n + 1} \specarrowR{\epsilon} \tup{\taintpc; \Omegat', \Rsb, 0}}
     }{v5-t-barr-spec}
     
      \typerule {$\Rvr$:AM-NoSpec-epsilon} {\trgR{\instrneqOt{\pcall{f}, \pret, \pbarrier, Z}} & \trgR{\taintpc;\Omegat \nsarrow{\epsilon} \Omegat'}
    }
    {\trgR{\tup{\taintpc, \Omegat, \Rsb, n + 1} \specarrowR{\epsilon} \tup{\taintpc, \Omegat', \Rsb, n}}
     }{v5-t-nospec-eps}
    \typerule {$\Rvr$:AM-NoSpec-action}
    {\trgR{\instrneqOt{\pcall{f}, \pret, \pbarrier, Z}} & \trgR{\taintpc; \Omegat \nsarrow{\taint'} \Omegat'}
    }
    {\trgR{\tup{\taintpc, \Omegat, \Rsb, n + 1} \specarrowR{\taintpc \glb \taint'} \tup{\taintpc, \Omegat', \Rsb, n}}
     }{v5-t-nospec-act}

    \typerule{$\Rvr$:AM-Call}
    {\trgR{\instrOv{\pcall{f}}} &  \trgR{\taintpc; \Omegat \nsarrow{\taint} \Omegat'} \\ 
      \trgR{\Rsb'} = \trgR{\Rsb \cdot (\Omegat(\pc) + 1)}  \\
       \ffun{\trgR{\Omegat(\pc)}} = \trgR{f} & \ffun{\trgR{\Omegat'(\pc)}} = \trgR{f'} & \trgR{f, f' \notin \OB{I}} 
    }
    {
    \trgR{\tup{\taintpc, \Omegat, \Rsb, n + 1} \specarrowR{\taint} \tup{\taintpc, \Omegat', \Rsb', n}}
    }{v5-t-call}
    \typerule{$\Rvr$:AM-Call-att}
    {\trgR{\instrOv{\pcall{f}}} &  \trgR{\taintpc; \Omegat \nsarrow{\taint} \Omegat'} \\ 
     \ffun{\trgR{\Omegat(\pc)}} = \trgR{f} & \trgR{f \in \OB{I}} 
    }
    {
    \trgR{\tup{\taintpc, \Omegat, \Rsb, n + 1} \specarrowR{\taint} \tup{\taintpc, \Omegat', \Rsb, n}}
    }{v5-t=call-att}
     
    \typerule{$\Rvr$:AM-Ret-Same}
    {\trgR{\instrOv{\pret}} & \trgR{\taintpc; \Omegat \nsarrow{\tau} \Omegat'} & \trgR{\Rsb = \Rsb' ; l} & \trgR{B = B' \cdot l'} \\ 
    \trgR{l} = \trgR{l'} & \ffun{\trgR{\Omegat(\pc)}} = \trgR{f} & \ffun{\trgR{\Omegat'(\pc)}} = \trgR{f'} & \trgR{f, f' \notin \OB{I}} 
    }
    {
    \trgR{\tup{\taintpc, \Omegat, \Rsb, n + 1} \specarrowR{\safeta} \tup{\taintpc, \Omegat', \Rsb', n}}
     }{v5-t-retS}
     \typerule{$\Rvr$:AM-Ret-Empty}
    {
    \trgR{\instr{\pret}} & \trgR{\taintpc; \Omegat \nsarrow{\taint} \Omegat'}
    }
    {
    \trgR{\tup{\taintpc, \Omegat, \emptyset, n + 1} \specarrowR{\safeta} \tup{\taintpc, \Omegat', \emptyset, n}}
    }{v5-t=retE}
    
    \typerule{$\Rvr$:AM-Ret-Spec}
    {\trgR{\instrOt{\pret}} & \trgR{\taintpc; \Omegat \nsarrow{\taint'} \Omegat'}  &  \trgR{\Omegat = \OB{F}; \OB{I}; \OB{B}; \sigmat}\\
     \trgR{\Rsb = \Rsb' ; l} & & \trgR{B = B'' \cdot l'} & l \neq l'  &  \trgR{\Omegat' = \OB{F}; \OB{I}; \OB{B'}; \sigmat''}\\
    \ffun{\trgR{l}} = \trgR{f'} & \ffun{\trgR{\Omegat(\pc)}} = \trgR{f} & \trgR{f, f' \notin \OB{I}} \\
    \trgR{\sigmat'' = \sigmat[\pc \mapsto l : \unta]} & \trgR{j = min(\omega, n)}
    }
    {\trgR{\tup{\taintpc, \Omegat, \Rsb, n + 1} \specarrowR{\taintpc \glb \taint'} \tup{\taintpc, \Omegat', \Rsb', n} \cdot \tup{\unta, \Omegat'', \Rsb', j}} }{v5-t-spec}
    \typerule{$\Rvr$:AM-Rollback}
    { \trgR{n = 0}\ \text{ or }\ \finType{\trgR{\Omegat}} }
    {\trgR{\tup{\taintpc, \Omegat, \Rsb,  n} \specarrowR{\safeta} \emptyset}
    }{v5-t-rollback}
    
    \typerule{$\Rvr$:AM-Ret-att}
    {\trgR{\instrOt{\pret}} &  \trgR{\taintpc; \Omegat \nsarrow{\taint} \Omegat'} &
    \ffun{\trgR{\Omegat(\pc)}} = \trgR{f} & \trgR{f \in \OB{I}}
    }
    {\trgR{\tup{\taintpc, \Omegat, \Rsb, n + 1} \specarrowR{\taint} \tup{\taintpc, \Omegat', \Rsb, n}}
    }{v5-t-spec-att}
\end{center}

\begin{center}
    \centering
    \small
    \mytoprule{\trgR{\SigmaR \specarrowR{\tau^{\taint}} \SigmaR'}}
    
    \typerule {$\Rvr$:Combine}
    {\trgR{\SigmaR = \phiStackR{}} & \trgR{\SigmaR = \phiStackR'} \\
    \trgR{\phiStackRv + \phiStackRt = \phiStackR} &  \trgR{\phiStackRv' + \phiStackRt' = \phiStackR'} \\
    \trgR{\phiStackRv \specarrowR{\tau} \phiStackRv'} & \trgR{\phiStackRt \specarrowR{\taint} \phiStackRt'}
    }
    {\trgR{\SigmaR \specarrowR{\tau^{\taint}} \SigmaR'}
     }{v5-combine}
\end{center}

\begin{center}
    \mytoprule{\trgR{\SigmaR \bigspecarrowR{\tauStack} \SigmaR'}}
    
    \typerule{$\Rvr$:AM-Reflection}
    {
    }
    {\trgR{\SigmaR \bigspecarrowR{\varepsilon} \SigmaR}
    }{v5-reflect}
    \typerule{$\Rvr$:AM-Single}
    {
    \trgR{\SigmaR \bigspecarrowR{\OB{\tau^{\taint}}} \SigmaR''} & \trgR{\SigmaR'' \specarrowR{\tau^{\taint}} \SigmaR'} \\
    \trgR{\SigmaR'' = \phiStackR \cdot  \tup{\OB{F}; \OB{I}; \OB{B}; \sigma, n}} &
    \trgR{\SigmaR' = \phiStackR \cdot \tup{\OB{F}; \OB{I}; \OB{B'}; \sigma', n' }} \\
    \ffun{\trgR{\sigma(\pc)}} = \trgR{f} & \ffun{\trgR{\sigma'(\pc)}} = \trgR{f'} \\
    \text{ if } \trgR{f == f'} and ~\trgR{f \in \OB{I}} \text{ then } \trgR{\tau^{\taint} = \epsilon} \text{ else } \trgR{\tau^{\taint} = \tau^{\taint}}
    }
    {
    \trgR{\SigmaR \bigspecarrowR{\OB{\tau^{\taint}} \cdot \tau^{\taint}} \SigmaR'}
    }{v5-single}
    
    \typerule{$\Rvr$:AM-silent}
    {
    \trgR{\SigmaR \bigspecarrowR{\tra{^\taint}} \SigmaR''} & \trgR{\SigmaR'' \specarrowR{\epsilon} \SigmaR'}
    }
    {
    \trgR{\SigmaR \bigspecarrowR{\tra{^\taint}} \SigmaR'}
    }{v5-silent}

\end{center}

\begin{center}
    \mytoprule{Helpers}
    
    \typerule{$\Rvr$:AM-Merge}
    {\trgR{\Omegav + \Omegat = \Omega }
    }
    {\trgR{\tup{\Omegav, \Rsb, n} + \tup{\taintpc, \Omegat, \Rsb, n} = \tup{\Omega, n , \taintpc}}
    }{v5-merge}
    \typerule{$\Rvr$:AM-Init}
    {\trgR{\SigmaR = \initFunc(M, \OB{F}, \OB{I}),\emptyset, \bot, \safeta}
    }
    {\trgR{\initFuncR(M ,\OB{F}, \OB{I}) = \SigmaR}
    }{v5-init}
    
    \typerule{$\Rvr$:AM-Fin-Ending}
    { \trgR{\SigmaR = (\Omega, \Rsb, \bot, \taintpc)} & \finType{\trgR{\Omega}}
    }
    {\finTypef{\trgR{\SigmaR}}
    }{v5-fin-end}
    \typerule{$\Rvr$:AM-Fin}
    { \trgR{\SigmaR = \phiStackR \cdot (\Omega, \Rsb, n, \taintpc)} & \finType{\trgR{\Omega}}
    }
    {\finType{\trgR{\SigmaR}}
    }{v5-fin}
    
    \typerule {$\Rvr$:AM-Trace}
    { \exists \finTypef{\trgR{\SigmaR'}}\ & \trgR{\initFuncR(P) \bigspecarrowR{\tauStackT}\ \SigmaR'}
    }
    { \trgR{\amTracevR{P}{\tauStackT}}
     }{v5-trace}
    \typerule {$\Rvr$:AM-Beh}
    {}
    { \trgR{\behR{P}} =  \{\trgR{\tauStackT} \mid \trgR{\amTracevR{P}{\tauStackT}} \}
     }{v5-beh}
\end{center}

 \section{Speculative Semantics $\SLSv$}\label{app:semSLS}

Modern CPUs can speculatively bypass $\retC$ instructions. This kind of speculation is called straight-line speculation (SLS) \cite{sls-whitepaper, sls-whitepaper2}.

The speculative state $\trgSLS{\SigmaSLS}$ is a stack of speculative instances $\trgSLS{\PhiSLS}$ that contain the non-speculative state $\Omega$ and the speculation window $n$.
Note that we define $n + 1$ to match $\bot$ as well.
All elements from the source language exist in this language

\begin{align*}
    \mi{Instructions}~ \trgSLS{i} \bnfdef&\ \cdots \mid \pbarrier \\
    \mi{Speculative\ States}~ \trgSLS{\SigmaSLS} \bnfdef&\ \trgSLS{\phiStackSLS} \\
    \mi{Speculative\ Instance}~ \trgSLS{\PhiSLS} \bnfdef&\ \trgSLS{\tup{\Omega, n, \taintpc}} \\
    \mi{Speculative\ Instance Vals.}~ \trgSLS{\PhiSLSv} \bnfdef&\ \trgSLS{\tup{\Omegav, n}} \\
    \mi{Speculative\ Instance Taint}~ \trgSLS{\PhiSLSt} \bnfdef&\ \trgSLS{\tup{\taintpc, \Omegat, n}} \\
    \mi{Windows}~ \trgSLS{w} \bnfdef&\ n \mid \bot \\
    \mi{Observations}~ \trgSLS{\tau} \bnfdef&\ \trgSLS{\rollbackObsSLS}
\end{align*}

We use the shorthand notation $\Omega(\pc)$ for $\sigma(\pc)$ with $\Omega = \OB{C};\OB{B}; \sigma$. The same notation works on $\Omegav$ and $\Omegat$ as well. We note that for $\Omegat$, the notation will still yield a value, since the value of the $\pc$ is tracked even when doing taint tracking

\mytoprule{Judgements}
\begin{align*}
    &
    \trgSLS{\phiStackSLSv \specarrowSLS{\tau} \phiStackSLSv'}
    && \text{Stack $\PhiSLSv$ small-steps to $\phiStackSLSv$ and emits observation $\tau$.}
    \\
    &
    \trgSLS{\PhiSLSv \specarrowSLS{\tau} \phiStackSLSv'}
    && \text{Speculative instance $\PhiSLSv$ small-steps to $\phiStackSLSv$ and emits observation $\tau$.}
    \\
    &
    \trgSLS{\phiStackSLSt \specarrowSLS{\taint} \phiStackSLSt}
    && \text{Stack $\phiStackSLSt$ small-steps to $\phiStackSLSt'$ and emits taint $\taint$.}
    \\
    &
    \trgSLS{\PhiSLSt \specarrowSLS{\taint} \phiStackSLSt'}
    && \text{Speculative instance $\PhiSLSt$ small-steps to $\phiStackSLSt'$ and emits taint $\taint$.}
    \\
    &
    \trgSLS{\SigmaSLS \specarrowSLS{\tau^{\taint}} \SigmaSLS'}
    && \text{Speculative state $\SigmaSLS$ small-steps to $\SigmaSLS'$ and emits tainted observation $\tau^{\taint}$.}
    \\
    &
    \trgSLS{\SigmaSLS \bigspecarrowSLS{\OB{\tau^{\taint}}} \SigmaSLS'}
    && \text{State $\SigmaSLS$ big-steps to $\SigmaSLS'$ and emits a list of observations $\tauStack$.}
    \\
    &
    \trgSLS{\amTracevSLS{P}{\tauStackT}}
    && \text{Program $P$ produces the observations $\tauStackT$ during execution.}
\end{align*}

\begin{center}
    \centering
    \small
    \mytoprule{\trgSLS{\phiStackSLSv \specarrowSLS{\tau} \phiStackSLSv'}}
    
    \typerule {$\SLSvr$:AM-Context}
    {\trgSLS{\PhiSLSv \specarrowSLS{\tau} \phiStackSLSv'}
    }
    {
    \trgSLS{
    \phiStackSLSv \cdot \PhiSLSv
    \specarrowSLS{\tau}
    \phiStackSLSv \cdot \phiStackSLSv'
    }
     }{sls-context}
     
     \mytoprule{\trgSLS{\PhiSLSv \specarrowSLS{\tau} \phiStackSLSv'}}
     
     \typerule {$\SLSvr$:AM-barr}
    {\trgSLS{\instrOv{\pbarrier}} & \trgSLS{\Omegav' = \Omegav[\pc \mapsto \pc + 1]}
    }
    {\trgSLS{\tup{\Omegav, \bot} \specarrowSLS{\epsilon} \tup{\Omegav', \bot}}
     }{sls-barr}
     \typerule {$\SLSvr$:AM-barr-spec}
    {\trgSLS{\instrOv{\pbarrier}} & \trgSLS{\Omegav' = \Omegav[\pc \mapsto \pc + 1]}
    }
    {\trgSLS{\tup{\Omegav, n + 1} \specarrowSLS{\epsilon} \tup{\Omegav', 0}}
     }{sls-barr-spec}
     
     \typerule {$\SLSvr$:AM-NoSpec-epsilon} {\trgSLS{\instrneqOv{\pret, \pbarrier, Z}} &  \trgSLS{\Omegav \nsarrow{\epsilon} \Omegav'}
    }
    {\trgSLS{\tup{\Omegav, n + 1} \specarrowSLS{\epsilon} \tup{\Omegav', n}}
     }{sls-nospec-eps}
    \typerule {$\SLSvr$:AM-NoSpec-action}
    {\trgSLS{\instrneqOv{\pret, \pbarrier, Z}} & \trgSLS{\Omegav \nsarrow{\tau} \Omegav'}
    }
    {\trgSLS{\tup{\Omegav, n + 1} \specarrowSLS{\tau} \tup{\Omegav', n}}
     }{sls-nospec-act}
    
    \typerule{$\SLSvr$:AM-Ret-Spec}
    {\trgSLS{\instrOv{\pret}} &  \trgSLS{\Omegav \nsarrow{\tau} \Omegav'} & \trgSLS{\Omegav = \OB{F}; \OB{I}; \OB{B}; \sigmav}\\
    \ffun{\trgSLS{\Omegav(\pc)}} = \trgSLS{f} & \trgSLS{f \notin \OB{I}} & \trgSLS{\Omegav'' = \OB{F}; \OB{I}; \OB{B}; \sigmav''}\\
     \trgSLS{\sigmav'' = \sigmav[\pc \mapsto \Omegav(\pc) + 1]} & \trgSLS{j = min(\omega, n)} 
    }
    {\trgSLS{\tup{\Omegav, n + 1} \specarrowSLS{\tau} \tup{\Omegav', n} \cdot \tup{\Omegav'', j}}
    }{sls-spec}
    \typerule{$\SLSvr$:AM-Rollback}
    { \trgSLS{n = 0}\ \text{or}\ \finType{\trgSLS{\Omegav}} }
    {\trgSLS{\tup{\Omegav, n} \specarrowSLS{\rollbackObsS} \varepsilon}
    }{sls-rollback}
    
    \typerule{$\SLSvr$:AM-Ret-Spec-att}
    {\trgSLS{\instrOv{\pret}} &  \trgSLS{\Omegav \nsarrow{\tau} \Omegav'} \\
    \ffun{\trgSLS{\Omegav(\pc)}} = \trgSLS{f} & \trgSLS{f \in \OB{I}} \\
    }
    {\trgSLS{\tup{\Omegav, n + 1} \specarrowSLS{\tau} \tup{\Omegav', n}}
    }{sls-spec-att}

\end{center}

In \Cref{tr:sls-spec} we essentially ignore the return instruction by bypassing it and just executing the instruction that comes next.

\subsection{Taint Semantics}
\begin{center}
    \centering
    \small
    \mytoprule{\trgSLS{\phiStackSLSt \specarrowSLS{\taint} \phiStackSLSt}}
    
    \typerule {$\SLSvr$:T-Context}
    {\trgSLS{\PhiSLSt \specarrowSLS{\tau} \phiStackSLSt}
    }
    {\trgSLS{\phiStackSLSt \cdot \PhiSLSt \specarrowSLS{\tau} \phiStackSLSt \cdot \phiStackSLSt'}
     }{sls-t-context}

     \mytoprule{\trgSLS{\PhiSLSt \specarrowSLS{\taint} \phiStackSLSt'}}
     
     \typerule {$\SLSvr$:T-barr}
    {\trgSLS{\instrOt{\pbarrier}} & \trgSLS{\Omegat' = \Omegat[\pc \mapsto \pc + 1]}
    }
    {\trgSLS{\tup{\taintpc, \Omegat, \bot}
    \specarrowSLS{\epsilon}
    \tup{\taintpc, \Omegat', \bot}}
     }{sls-t-barr}
     \typerule {$\SLSvr$:T-barr-spec}
    {\trgSLS{\instrOt{\pbarrier}} & \trgSLS{\Omegat' = \Omegat[\pc \mapsto \pc + 1]}
    }
    {\trgSLS{\tup{\taintpc, \Omegat, n + 1}
    \specarrowSLS{\epsilon}
    \tup{\taintpc, \Omegat', 0}}
     }{sls-t-barr-spec}
     
      \typerule {$\SLSvr$:AM-NoSpec-epsilon} {\trgSLS{\instrneqOv{\pret, \pbarrier, Z}} &  \trgSLS{\taintpc; \Omegat \nsarrow{\epsilon} \Omegat'}
    }
    {\trgSLS{\tup{\taintpc, \Omegat, n + 1} \specarrowSLS{\epsilon} \tup{\taintpc, \Omegat', n}}
     }{sls-t-nospec-eps}
     \typerule {$\SLSvr$:T-NoSpec-action}
    {\trgSLS{\instrneqOt{\pret, \pbarrier, Z}} & \trgSLS{\taintpc;\Omegat \nsarrow{\taint'} \Omegat'}
    }
    {\trgSLS{\tup{\taintpc, \Omegat, n + 1} \specarrowSLS{\taintpc \glb \taint'} \tup{\taintpc, \Omegat', n}}
     }{sls-t-nospec-act}

    \typerule{$\SLSvr$:T-Ret-Spec}
    {\trgSLS{\instrOt{\pret}} & \trgSLS{\taintpc;\Omegat \nsarrow{\taint'} \Omegat'} & \trgSLS{\Omegat = \OB{F}; \OB{I}; \OB{B}; \sigmat} \\
    \ffun{\trgSLS{\Omegat(\pc)}} = \trgSLS{f} & \trgSLS{f \notin \OB{I}} & \trgSLS{\Omegat'' = \OB{F}; \OB{I}; \OB{B}; \sigmat''} \\
     \trgSLS{\sigmat'' = \sigmat[\pc \mapsto \Omegat(\pc) + 1 : \unta]} & \trgSLS{j = min(\omega, n)}
    }
    {\trgSLS{\tup{\taintpc, \Omegat, n + 1} \specarrowSLS{\taintpc \glb \taint'}
    \tup{\taintpc, \Omegat', n}\cdot \tup{\unta, \Omegat'', j}} }{sls-t-spec}
    \typerule{$\SLSvr$:T-Ret-Spec-att}
    {\trgSLS{\instrOt{\pret}} & \trgSLS{\taintpc;\Omegat \nsarrow{\taint'} \Omegat'}  \\
    \ffun{\trgSLS{\Omegat(\pc)}} = \trgSLS{f} & \trgSLS{f \notin \OB{I}} \\
    }
    {\trgSLS{\tup{\taintpc, \Omegat, n + 1} \specarrowSLS{\taintpc \glb \taint'} \tup{\taintpc, \Omegat', n}}
    }{sls-t-spec-att}
    
    \typerule{$\SLSvr$:T-Rollback}
    { \trgSLS{n = 0} \text{ or } \finType{\trgSLS{\Omegat}} }
    { 
    \trgSLS{\tup{\taintpc, \Omegat, n } \specarrowSLS{\safeta} \emptyset}
    }{sls-t-rollback}
    
\end{center}

\begin{center}
    \centering
    \small
    \mytoprule{\trgSLS{\SigmaSLS \specarrowSLS{\tau^{\taint}} \SigmaSLS'}}
    
    \typerule {$\SLSvr$:Combine}
    {\trgSLS{\SigmaSLS = \phiStackSLS{}} & \trgSLS{\SigmaSLS = \phiStackSLS'} \\
    \trgSLS{\phiStackSLSv + \phiStackSLSt = \phiStackSLS} &  \trgSLS{\phiStackSLSv' + \phiStackSLSt' = \phiStackSLS'} \\
    \trgSLS{\phiStackSLSv \specarrowSLS{\tau} \phiStackSLSv'} & \trgSLS{\phiStackSLSt \specarrowSLS{\taint} \phiStackSLSt'}
    }
    {\trgSLS{\SigmaSLS \specarrowSLS{\tau^{\taint}} \SigmaSLS'}
     }{sls-combine}
\end{center}

\begin{center}
    \mytoprule{\trgSLS{\SigmaSLS \bigspecarrowSLS{\tauStack} \SigmaSLS'}}
    
    \typerule{$\SLSvr$:AM-Reflection}
    {
    }
    {\trgSLS{\SigmaSLS \bigspecarrowSLS{\varepsilon} \SigmaSLS}
    }{sls-reflect}
    \typerule{$\SLSvr$:AM-Single}
    {
    \trgSLS{\SigmaSLS \bigspecarrowSLS{\OB{\tau^{\taint}}} \SigmaSLS''} & \trgSLS{\SigmaSLS'' \specarrowSLS{\tau^{\taint}} \SigmaSLS'} \\
    \trgSLS{\SigmaSLS'' = \phiStackSLS \cdot  \tup{\OB{F}; \OB{I}; \OB{B}; \sigma, n}} &
    \trgSLS{\SigmaSLS' = \phiStackSLS \cdot \tup{\OB{F}; \OB{I}; \OB{B'}; \sigma', n' }} \\
    \ffun{\trgSLS{\sigma(\pc)}} = \trgSLS{f} & \ffun{\trgSLS{\sigma'(\pc)}} = \trgSLS{f'} \\
    \text{ if } \trgSLS{f == f'} and ~\trgSLS{f \in \OB{I}} \text{ then } \trgSLS{\tau^{\taint} = \epsilon} \text{ else } \trgSLS{\tau^{\taint} = \tau^{\taint}}
    }
    {
    \trgSLS{\SigmaSLS \bigspecarrowSLS{\OB{\tau^{\taint}} \cdot \tau^{\taint}} \SigmaSLS'}
    }{sls-single}
    
    \typerule{$\SLSvr$:AM-silent}
    {
    \trgSLS{\SigmaSLS \bigspecarrowSLS{\tra{^\taint}} \SigmaSLS''} & \trgSLS{\SigmaSLS'' \specarrowSLS{\epsilon} \SigmaSLS'}
    }
    {
    \trgSLS{\SigmaSLS \bigspecarrowSLS{\tra{^\taint}} \SigmaSLS'}
    }{sls-silent}

\end{center}
Again, we ignore events that happen in the context.

\begin{center}
    \mytoprule{Helpers}
    
    \typerule{$\SLSvr$:AM-Merge}
    {\trgSLS{\Omegav + \Omegat = \Omega }
    }
    {\trgSLS{\tup{\Omegav, n} + \tup{\taintpc, \Omegat, n} = \tup{\Omega, n , \taintpc}}
    }{sls-merge}            
    \typerule{$\SLSvr$:AM-Init}
    {\trgSLS{\SigmaSLS = \initFunc(M, \OB{F}, \OB{I}), \bot, \safeta}
    }
    {\trgSLS{\initFuncSLS(M ,\OB{F}, \OB{I}) = \SigmaSLS}
    }{sls-init}
    
    \typerule{$\SLSvr$:AM-Fin-Ending}
    { \trgSLS{\SigmaSLS = (\Omega, \bot)} & \finType{\trgSLS{\Omega}}
    }
    {\finTypef{\trgSLS{\SigmaSLS}}
    }{sls-fin-end}
    \typerule{$\SLSvr$:AM-Fin}
    { \trgSLS{\SigmaSLS = \phiStackSLS \cdot (\Omega, n)} & \finType{\trgSLS{\Omega}}
    }
    {\finType{\trgSLS{\SigmaSLS}}
    }{sls-fin}
    
    \typerule {$\SLSvr$:AM-Trace}
    { \exists \finTypef{\trgSLS{\SigmaSLS'}}\ & \trgSLS{\initFuncSLS(P) \bigspecarrowSLS{\tauStackT}\ \SigmaSLS'}
    }
    { \trgSLS{\amTracevSLS{P}{\tauStackT}}
     }{sls-trace}
    \typerule {$\SLSvr$:AM-Beh}
    {}
    { \trgSLS{\behSLS{P}} =  \{\trgSLS{\tauStackT} \mid \trgSLS{\amTracevSLS{P}{\tauStackT}} \}
     }{sls-beh}
\end{center}

\section{NS: Safe Equivalence Preservation}
For the speculative contracts, we need to ensure that the underlying non-speculative semantics preserve the safe equivalence relation. We do this once here.

\begin{center}
\begin{forest}
    [\Thmref{lemma:non-spec-preserves-safe}
        [\Thmref{lemma:safe-equi-expression1}
            [\Thmref{lemma:safe-equi-expression2}
                [\Thmref{lemma:bounds-preserve-relation}]
            ]
        ]
    ]
\end{forest}
\end{center}
\begin{insight}
    The distinction of splitting the step into the taint and value part is important here. The speculative semantics uses either the taint or the value non-speculative semantics to delegate back. For example in \cref{tr:v4-nospec-eps}.
    For convenience, we write in the combined step that a "full" non-spec step $\src{\Omega \nsarrow{\tau^\taint} \Omega'}$ was done.
    Note however that in \Cref{tr:combine-s} we modify the taint by $\taint \sqcap \safeta$. 
    This has implications for \Thmref{lemma:non-spec-preserves-safe} as it is stated here because then the fact of $\safe{\tau}$ would not help us with the taint because the taint is always $\safeta$ independent of the action.
    However, in the proofs about the speculative semantics that uses a non-speculative step the \Cref{tr:combine-s} is never used and so the $\safeta$ does not appear and $\safe{\tau}$ actually tells us if the action was saved or not.
    Again, for convenience, we will leave \Thmref{lemma:non-spec-preserves-safe} as is where \Cref{tr:combine-s} is not used. But note that it should have been stated once for
    $\Omega_v \nsarrow{\tau} \Omega_v'$  and once for $\taintpc;\Omega_t \nsarrow{\taint} \Omega_t'$. Both proofs follow the exact same structure.
\end{insight}
\begin{lemma}[Steps of $\nsarrow{}$ with safe Observations preserve Safe-Equivalence]\label{lemma:non-spec-preserves-safe}
If
\begin{enumerate}
    \item $\src{\Omega \nsarrow{\tau} \Omega^{\dagger}}$ and
    \item $\safe{\src{\tau}}$ and 
    \item $\src{\Omega} \relsa \src{\Omega'}$
\end{enumerate}
Then
\begin{enumerate}[label=\Roman*]
    \item $\src{\Omega' \nsarrow{\tau'} \Omega^{*}}$
    \item $\src{\Omega^{\dagger}} \relsa \src{\Omega^{*}}$
    \item $\src{\tau} = \src{\tau'}$
\end{enumerate}
\end{lemma}

\begin{proof}

Let us denote what we already know. 

\begin{align*}
    \src{\Omega} &= \src{C;B; \tup{p, M, A}} \\
    \src{\Omega'} &= \src{C;B'; \tup{p, M', A'}}
\end{align*}
From $\src{\Omega} \relsa \src{\Omega'}$ we get
\begin{enumerate}
    \item $\src{\Omega(\pc)} = \src{\Omega'(\pc)}$
    \item $\src{M} \relsa \src{M'}$ and $\src{A} \relsa \src{A'}$
\end{enumerate}
We now proceed by inversion on $\src{\Omega \nsarrow{\tau} \Omega'}$:
\begin{description}
    \item[\Cref{tr:skip}]
    Then $\src{\Omega(\pc)} = \pskip$ and $\src{\tau} = \src{\varepsilon}$. 
    We thus have $\src{\Omega'(\pc)} = \pskip$ and can apply \Cref{tr:skip} to derive the step $\src{\Omega' \nsarrow{\tau'} \Omega_{*}}$. 
    
    \begin{description}
        \item[II]
         \Cref{tr:skip} only increments the program counter. 
         Since $\src{\Omega(\pc)} = \src{\Omega'(\pc)}$ we have $\src{\Omega^{\dagger}(\pc)} = \src{\Omega^{*}(\pc)}$.
         
         Thus $\src{\Omega^{\dagger}} \relsa \src{\Omega^{*}}$.
        \item[III]
        From  \Cref{tr:skip} it follows that $\src{\tau'} = \src{\varepsilon}$ and thus $\src{\tau} = \src{\tau'}$.
    \end{description}
    
    \item[\Cref{tr:assign}]
    Then $\src{\Omega(\pc)} = \passign{x}{e}$, $\src{\tau} = \src{\varepsilon}$ and $\src{\exprEval{a}{e}{v : \taint}}$. 
    
    We can apply \Cref{lemma:safe-equi-expression1} and get $\src{\exprEval{A'}{e}{v' : \taint'}}$ with $\src{v : \taint} \relsa \src{v' : \taint'}$.
    
    We thus have $\src{\Omega'(\pc)} = \passign{x}{e}$ and $\src{\exprEval{A'}{e}{v' : \taint'}}$ and can apply \Cref{tr:assign} to derive the step $\src{\Omega' \nsarrow{\tau'} \Omega_{*}}$. 
    
    \begin{description}
        \item[II]
        
        Since the rule only updates the stores $\src{A}$ and $\src{A'}$, we need to show that $\src{A^{\dagger}} \relsa \src{A^{*}}$.
        
        Because of $\src{v : \taint} \relsa \src{v' : \taint'}$ and these values are the only change to the store, we can conclude that $\src{A^{\dagger}} \relsa \src{A^{*}}$.

         Since $\src{\Omega(\pc)} = \src{\Omega'(\pc)}$ we have $\src{\Omega^{\dagger}(\pc)} = \src{\Omega^{*}(\pc)}$.
         
         Thus, $\src{\Omega^{\dagger}} \relsa \src{\Omega^{*}}$.
        \item[III]
        From \Cref{tr:assign} it follows that $\tau' = \varepsilon$ and thus $\tau = \tau'$.
    \end{description}
    
    \item[\Cref{tr:beqz-sat}]
    Then $\src{\Omega(\pc)} = \pjz{x}{\lbl}$ and $\src{A(x)} = \src{0 : \taint}$
    
    Since $\src{A} \relsa \src{A'}$ we know that $\src{A'(x)} = \src{v'} : \src{\taint'}$ with $\src{0 : \taint} \relsa \src{v' : \taint'}$.
    
    Additionally, because of $\safe{\src{\tau}}$, we know $\src{\taint} = \safeta$ and we can conclude that $\src{v'} = \src{0}$ and $\src{\taint} = \src{\taint'}$ using $\src{0 : \taint} \relsa \src{v' : \taint'}$.

    Thus, we can apply \Cref{tr:beqz-sat} to derive the step $\src{\Omega' \nsarrow{\tau'} \Omega_{*}}$. 
    
    \begin{description}
        \item[II]
        
        We need to show $\src{A^{\dagger}} \relsa \src{A^{*}}$, since that is the only component changing during the step. The rest follows from $\src{\Omega} \relsa \src{\Omega'}$.
        
        Note that, $\src{A^{\dagger}} = \src{A}[\pc \mapsto \lbl : \taint]$ and $\src{A^{*}} = \src{A'}[\pc \mapsto \lbl : \taint']$.
        
        Because $\src{\taint} = \src{\taint'}$ we get $\src{A^{\dagger}} \relsa \src{A^{*}}$.
        
        \item[III]
        We know that $\src{\tau} = \src{\pcObs{\lbl}^{\taint}}$ and $\src{\tau'} = \src{\pcObs{\lbl}^{\taint'}}$. 
        
        Since $\src{\taint} = \src{\taint'}$, we have $\src{\tau} = \src{\tau'}$.
    \end{description}

    \item[\Cref{tr:beqz-unsat}]
    
    The proof of this case is similar to the \Cref{tr:beqz-sat} case above.
    
    We need the additional fact that $\src{A(\pc)} = \src{A'(\pc)}$ since $\src{\tau} = \src{\pcObs{(A(\pc) +1)}}$ and $\src{\tau'} = \src{(\pcObs{A(\pc) + 1)}}$. That fact follows immediately from the assumption $\src{\Omega(\pc)} = \src{\Omega'(\pc)}$.

    \item[\Cref{tr:jmp}]
    
    Then $\src{\Omega(\pc)} = \pjmp{e}$. $ \src{\exprEval{A}{e}{\lbl : \taint}}$ and $\ffun{\src{A(\pc)}} = \src{f}$, $\ffun{\src{\lbl}} = \src{f'}$ and $\src{C.\mtt{intfs}}\vdash\src{f,f'}:\src{internal}$
    
    We can apply \Cref{lemma:safe-equi-expression1} and get $\src{\exprEval{A'}{e}{n' : \taint'}}$ with $\src{\lbl : \taint} \relsa \src{\lbl' : \taint'}$.

    Additionally, because of $\safe{\src{\tau}}$, we know $\src{\taint} = \src{\safeta}$ and we can conclude that $\src{\lbl} = \src{\lbl'}$ and $\src{\taint} = \src{\taint'}$ using $\src{\lbl : \taint} \relsa \src{\lbl' : \taint'}$.

    Thus, we can apply \Cref{tr:jmp} to derive the step $\src{\Omega' \nsarrow{\tau'} \Omega_{*}}$.
    
    \begin{description}
        
        \item[II]
        
        We need to show $\src{A^{\dagger}} \relsa \src{A^{*}}$, since that is the only component changing during the step. The rest follows from $\src{\Omega} \relsa \src{\Omega'}$.
        
        Note that, $A^{\dagger} = A[\pc \mapsto \lbl : \taint]$ and $A^{*} = A'[\pc \mapsto \lbl : \taint']$.
        
        Because $\src{\taint} = \src{\taint'}$ we get $\src{A^{\dagger}} \relsa \src{A^{*}}$.
        
        \item[III]
        We know that $\src{\tau} = \src{\pcObs{\lbl}^{\taint}}$ and $\src{\tau'} = \src{\pcObs{\lbl'}^{\taint'}}$. 
        
        Since $\src{\lbl} = \src{\lbl'}$ and $\src{\taint} = \src{\taint'}$, we have $\src{\tau} = \src{\tau'}$.
    \end{description}

    \item[\Cref{tr:load}]
    
    Then $\src{\Omega(\pc)} = \pload{x}{e}$, $\src{x} \neq \src{\pc}$ and $\src{\exprEval{A}{e}{n : \taint}}$.
    
    We can apply \Cref{lemma:safe-equi-expression1} and get $\src{\exprEval{A'}{e}{n' : \taint'}}$ with $\src{n : \taint} \relsa \src{n' : \taint'}$.
    
    Thus, we can apply \Cref{tr:load} to derive the step $\src{\Omega' \nsarrow{\tau'} \Omega^{*}}$.

    Additionally, because of $\safe{\src{\tau}}$, we know $\src{\taint} = \safeta$ and we can conclude that $\src{n} = \src{n'}$ and $\src{\taint} = \src{\taint'}$ using $\src{n : \taint} \relsa \src{n' : \taint'}$.
    
    \begin{description}
        \item[II]
        
        Since the rule only updates the stores $\src{A}$ and $\src{A'}$, we need to show that $\src{A^{\dagger}} \relsa \src{A^{*}}$. The rest follows from the assumptions.
        
        Let $\src{M(n)} = \src{v : \taint_1}$ and $\src{M(n')} = \src{v' : \taint_1'}$.
        We know that $\src{M} \relsa \src{M'}$ and $\src{n} = \src{n'}$. Thus, $\src{v : \taint_1} \relsa \src{v' : \taint_1'}$.
        
        Because of $\src{v : \taint_1} \relsa \src{v' : \taint_1'}$ and these values are the only change to the store, we can conclude that $\src{A^{\dagger}} \relsa \src{A^{*}}$.
        
        \item[III]
        
        We know that $\src{\tau} = \src{\loadObs{n}^{\taint \lub \taint_1}}$ and $\src{\tau' = \loadObs{n'}^{\taint' \lub \taint_2}}$. 
        
        Follows from $\src{n = n'}$, $\src{\taint = \taint'}$ and \Thmref{lemma:bounds-preserve-relation}.
    \end{description}

    \item[\Cref{tr:store}]
    
    Then $\src{\Omega(\pc)} = \pstore{x}{e}$, $\src{x} \neq \src{\pc}$ and $\src{\exprEval{A}{e}{n : \taint}}$ and $\src{A(x) = v : \taint_1}$
    
    We can apply \Cref{lemma:safe-equi-expression1} and get $\src{\exprEval{A'}{e}{n' : \taint'}}$ with $\src{n : \taint} \relsa \src{n' : \taint'}$.
    
    Furthermore, by $\src{A} \relsa \src{A'}$ we know that $\src{v : \taint_1} \relsa \src{v' : \taint_1'}$ where $\src{A'(x) = v' : \taint_1'}$.
    
    Thus, we can apply \Cref{tr:store} to derive the step $\src{\Omega' \nsarrow{\tau'} \Omega_{*}}$. 
    
    Additionally, because of $\safe{\src{\tau}}$, we know $\src{\taint} = \safeta$ and we can conclude that $\src{n} = \src{n'}$, $\src{v} = \src{v'}$, $\src{\taint} = \src{\taint'}$ and $\src{\taint_1} = \src{\taint_1'}$ using $\src{n : \taint} \relsa \src{n' : \taint'}$ and $\src{v : \taint_1} \relsa \src{v' : \taint_1'}$.
    
    \begin{description}
        
        \item[II]
        
        The rules only update the memories $\src{M}$ and $\src{M'}$ (Ignoring the increment of pc in the store since that trivially holds. See the proof of \Cref{tr:skip}).
        
        Because of $\src{n} = \src{n'}$, $\src{v} = \src{v'}$, $\src{\taint} = \src{\taint'}$ and $\src{\taint_1} = \src{\taint_1'}$ and the fact that $\src{M^{\dagger}} = \src{M[n \mapsto v : \taint_1]}$ and $\src{M^{*}} = \src{M^{\low'}[n' \mapsto v' : \taint_1']}$ we get $\src{M^{\dagger}} \relsa \src{M^{*}}$.
        
        \item[III]
        
        We know that $\src{\tau} = \src{\storeObs{n}^{\taint}}$ and $\src{\tau' = \storeObs{n'}^{\taint'}}$. 
        
        Since $\src{n} = \src{n'}$ and $\src{\taint} = \src{\taint'}$  we know that $\src{\tau} = \src{\tau'}$.
    \end{description}

    \item[\Cref{tr:load-prv}]
    
    Then $\src{\Omega(\pc)} = \ploadprv{x}{e}$, $\src{x} \neq \src{\pc}$ and $\src{\exprEval{A}{e}{n : \taint}}$.
    
    We can apply \Cref{lemma:safe-equi-expression1} and get $\src{\exprEval{A'}{e}{n' : \taint'}}$ with $\src{n : \taint} \relsa \src{n' : \taint'}$.
    
    Thus, we can apply \Cref{tr:load} to derive the step $\src{\Omega' \nsarrow{\tau'} \Omega^{*}}$.

    Additionally, because of $\safe{\src{\tau}}$, we know $\src{\taint} = \src{\safeta}$ and we can conclude that $\src{n} = \src{n'}$ and $\src{\taint} = \src{\taint'}$ using $\src{n : \taint} \relsa \src{n' : \taint'}$.
    
    \begin{description}
        \item[II]
        
        Since the rule only updates the stores $\src{A}$ and $\src{A'}$, we need to show that $\src{A^{\dagger}} \relsa \src{A^{*}}$. The rest follows from the assumptions.
        
        Let $\src{M^{\high}(n)} = \src{v : \unta}$ and $\src{M^{\high'}(n')} = \src{v' : \unta}$
        We know that $\src{M} \relsa \src{M'}$ and $\src{n} = \src{n'}$. Thus, $\src{v : \unta} \relsa \src{v' : \unta'}$.

        Here, $\src{A^{\dagger} = A[x \mapsto v : \unta]}$ and $\src{A^{*} = A'[x \mapsto v' : \unta]}$
        
        Because of $\src{v : \unta} \relsa \src{v' : \unta}$ and these values are the only change to the store, we can conclude that $\src{A^{\dagger}} \relsa \src{A^{*}}$.
        
        \item[III]
        We know that $\src{\tau} = \loadObs{n}^{\src{\taint}}$ and $\tau' = \loadObs{n'}^{\src{\taint'}}$. 
        
        Since $\src{n} = \src{n'}$ and $\src{\taint} = \src{\taint'}$, we know that $\src{\tau} = \src{\tau'}$.
    \end{description}
    
    \item[\Cref{tr:store-prv}]
    
    Then $\src{\Omega(\pc)} = \pstoreprv{x}{e}$ and $\src{\exprEval{A}{e}{n : \taint}}$ and $\src{A(x) = v : \taint_1}$
    
    We can apply \Thmref{lemma:safe-equi-expression1} and get $\src{\exprEval{A'}{e}{n' : \taint'}}$ with $\src{n : \taint} \relsa \src{n' : \taint'}$.
    
    Furthermore, by $\src{A} \relsa \src{A'}$ we know that $\src{v : \taint_1} \relsa \src{v' : \taint_1'}$ where $\src{A'(x) = v' : \taint_1'}$.
    
    Thus, we can apply \Cref{tr:store-prv} to derive the step $\src{\Omega' \nsarrow{\tau'} \Omega^{*}}$. 
    
    Additionally, because of $\safe{\src{\tau}}$, we know $\src{\taint = \safeta}$ and we can conclude that $\src{n} = \src{n'}$, $\src{v} = \src{v'}$, $\src{\taint} = \src{\taint'}$ and $\src{\taint_1} = \src{\taint_1'}$ using $\src{n : \taint} \relsa \src{n' : \taint'}$ and $\src{v : \taint_1} \relsa \src{v' : \taint_1'}$.
    
    \begin{description}
        
        \item[II]
        
        The rules only update the memories $\src{M}$ and $\src{M'}$ (Ignoring the increment of pc in the store since that trivially holds. See the proof of \Cref{tr:skip}).
        
        Because of $\src{n} = \src{n'}$, $\src{v} = \src{v'}$, $\src{\taint} = \src{\taint'}$ and $\src{\taint_1} = \src{\taint_1'}$ and the fact that $\src{M^{\high \dagger} = M^{\high}[n \mapsto v : \taint_1]}$ and $\src{M^{\high*} = M^{\high'}[n' \mapsto v' : \taint_1']}$ we get $\src{M^{\dagger}} \relsa \src{M^{*}}$.
        
        \item[III]
        We know that $\src{\tau} = \storeObs{n}^{\src{\taint}}$ and $\src{\tau'} = \storeObs{n'}^{\src{\taint'}}$. 
        
        Since $\src{n} = \src{n'}$, $\src{\taint} = \src{\taint'}$  we know that $\src{\tau} = \src{\tau'}$.
    \end{description}

    \item[\Cref{tr:call-internal}]
    
    Then $\src{\Omega(\pc)} = \pcall{f'}$, $\mathcal{F}(f') = n $, $f' \in C.\mtt{funs}$, $\src{A(\pc)} = \src{n_1 : \taint_1}$, $\ffun{\src{n_1}} = \src{f}$
    and $\src{C}.\mtt{intfs}\vdash\src{f',f} : \src{internal}$.

    Because of $\src{\Omega(\pc)} = \src{\Omega'(\pc)}$ and the fact that the $\src{\pc}$ is always tainted safe, we get $\src{A'(\pc)} = \src{n' : \taint_1'}$ with $\src{n_1} = \src{n'}$ and $\src{\taint_1} = \src{\taint_1'}$

    Thus, we can apply \Cref{tr:call-internal} to derive the step $\src{\Omega' \nsarrow{\tau'} \Omega^{*}}$. 
    
    \begin{description}
        
        \item[II]
        
        We need to show $\src{A^{\dagger}} \relsa \src{A^{*}}$ and $\src{B^{\dagger}} \relsa \src{B^{*}}$. The rest follows from $\src{\Omega} \relsa \src{\Omega'}$.
        
        Note that, $\src{A^{\dagger}} = \src{A}[\pc \mapsto  n' : \safeta]$ and $\src{A^{*}} = \src{A'}[\pc \mapsto n' : \safeta]$. Thus, $\src{A^{\dagger}} \relsa \src{A^{*}}$.

        Here $\src{B^{\dagger}} = \src{OB{B} \cdot 0 \mapsto A(\pc) + 1}$ and $\src{B^{*}} = \src{OB{B} \cdot 0 \mapsto A'(\pc) + 1}$.
        
        Since $\src{A(\pc)} = \src{A'(\pc)}$ we have $\src{B^{\dagger}} \relsa \src{B^{*}}$.

        \item[III]
        
        We know that $\src{\tau} = \src{\empTr}$ and $\src{\tau'} = \src{\empTr}$. Thus, $\src{\tau} = \src{\tau'}$ trivially holds.
        
    \end{description}

    \item[\Cref{tr:call}]
    
    Then $\src{\Omega(\pc)} = \pcall{f'}$, $\mathcal{F}(f') = n $, $f' \in C.\mtt{funs}$, $\src{A(\pc) = n_1 : \taint_1}$, $\ffun{\src{n_1}} = f$
    and $\src{C}.\mtt{intfs}\vdash\src{f',f} : \src{in}$.

    Because of $\src{\Omega(\pc)} = \src{\Omega'(\pc)}$ and the fact that the $\pc$ is always tainted safe, we get $\src{A'(\pc) = n' : \taint_1'}$ with $\src{n_1} = \src{n'}$ and $\src{\taint_1} = \src{\taint_1'}$

    Thus, we can apply \Cref{tr:call} to derive the step $\src{\Omega' \nsarrow{\tau'} \Omega^{*}}$. 
    
    \begin{description}
        
        \item[II]
        
        This case is analogous to \Cref{tr:call-internal} above.

        \item[III]
        
        We know that $\src{\tau} = \pcall{f'}?$ and $\src{\tau'} = \pcall{f}?$. Thus, $\src{\tau} = \src{\tau'}$ trivially holds.
        
    \end{description}
    
    \item[\Cref{tr:callback}]
    Analogous to \Cref{tr:call} above.

    \item[\Cref{tr:ret-internal}]
    
    Then $\src{\Omega(\pc)} = \pret$, $\src{\OB{B} = \OB{B_1} \cdot 0 \mapsto l : \taint}$, $\ffun{l} = f'$, $f' \in C.\mtt{funs}$, $\src{A(\pc) = n_1 : \taint_1}$, $\ffun{\src{n_1}} = f$and $\src{C}.\mtt{intfs}\vdash\src{f',f} : \src{internal}$.

    Because of $\src{\Omega(\pc)} = \src{\Omega'(\pc)}$ and the fact that the $\pc$ register is always tainted safe, we get $\src{A'(\pc) = n' : \taint_1'}$ with $\src{n_1} = \src{n'}$ and $\src{\taint_1} = \src{\taint_1'}$
    
    Let $\src{\OB{B'} = \OB{B_1'} \cdot 0 \mapsto l' : \taint'}$.
    Additionally, because of $\safe{\tau}$ and $\OB{B} \relsa \OB{B'}$, we know $\taint = \safeta$ and we can conclude that $l = l'$, $\taint = \taint'$.

    Thus, we can apply \Cref{tr:ret-internal} to derive the step $\src{\Omega' \nsarrow{\tau'} \Omega^{*}}$. 
    
    \begin{description}
        \item[II]
        
        We need to show $\src{A^{\dagger}} \relsa \src{A^{*}}$ and $\src{B^{\dagger}} \relsa \src{B^{*}}$. The rest follows from $\src{\Omega} \relsa \src{\Omega'}$.
        
        Note that, $\src{A^{\dagger}} = \src{A[\pc \mapsto  l : \taint]}$ and $\src{A^{*}} = \src{A'[\pc \mapsto l' : \taint']}$.
        
        Since $\src{l} = \src{l'}$, $\src{\taint} = \src{\taint'}$, we can conclude that $\src{A^{\dagger}} \relsa \src{A^{*}}$.

        Here $\src{B^{\dagger}} = \src{\OB{B_1}}$ and $\src{B^{*}} = \src{\OB{B_1'}}$.
        
        Since $\src{\OB{B}} \relsa \src{\OB{B'}}$, we now trivially have $\src{B^{\dagger}} \relsa \src{B^{*}}$.
        
        \item[III]
        
        We know that $\src{\tau} = \src{\empTr}$ and $\src{\tau'} = \src{\empTr}$. Thus, $\src{\tau} = \src{\tau'}$ trivially holds.
        
    \end{description}

    \item[\Cref{tr:ret}]

    Then $\src{\Omega(\pc)} = \pret$, $\src{\OB{B} = \OB{B_1} \cdot 0 \mapsto l : \taint}$, $\ffun{l} = f'$, $f' \in C.\mtt{funs}$, $\src{A(\pc) = n_1 : \taint_1}$, $\ffun{n_1} = f$and $\src{C}.\mtt{intfs}\vdash\src{f',f} : \src{in}$.

    Because of $\src{\Omega(\pc)} = \src{\Omega'(\pc)}$ and the fact that the $\pc$ is always tainted safe, we get $\src{A'(\pc) = n' : \taint_1'}$ with $\src{n_1} = \src{n'}$ and $\src{\taint_1} = \src{\taint_1'}$
    
    Let $\src{\OB{B'}} = \src{\OB{B_1'} \cdot 0 \mapsto l' : \taint'}$.
    Additionally, because of $\safe{\src{\tau}}$ and $\src{\OB{B}} \relsa \src{\OB{B'}}$, we know $\src{\taint} = \src{\safeta}$ and we can conclude that $\src{l} = \src{l'}$, $\src{\taint} = \src{\taint'}$.

    Thus, we can apply \Cref{tr:ret} to derive the step $\src{\Omega' \nsarrow{\tau'} \Omega^{*}}$. 
    
    \begin{description}
        \item[II]
        
        The case is analogous to \Cref{tr:ret-internal} above.
        
        \item[III]
        
        We know that $\src{\tau} = \pret^{\taint}$ and $\src{\tau'} = \pret^{\taint'}$. 
        
        Thus, $\src{\tau} = \src{\tau'}$ follows from $\src{\taint} = \src{\taint'}$.
        
    \end{description}

    \item[\Cref{tr:retback}]
    Analogous to \Cref{tr:ret} above.
    
\end{description}
\end{proof}

\begin{lemma}[Safe-Equivalent configurations produce safe-equivalent expression evaluations-1]\label{lemma:safe-equi-expression1}
If 
\begin{enumerate}
    \item $\src{\exprEval{A}{e}{v : \taint}}$ and
    \item $\src{\exprEval{A'}{e}{v' : \taint'}}$ and
    \item $\src{A} \relsa \src{A'}$
\end{enumerate}
Then
\begin{enumerate}[label=\Roman*]
    \item $\src{v} \relsa \src{v'}$ and
    \item $\src{\taint} \relsa \src{\taint'}$
\end{enumerate}

\end{lemma}
\begin{proof}

Applying \Cref{lemma:safe-equi-expression2} on $\src{\exprEval{A}{e}{v : \taint}}$ together with $\src{A} \relsa \src{A'}$ we get $\src{\exprEval{A'}{e}{v'' : \taint''}}$ and $\src{v} \relsa \src{v''}$ and $\src{\taint} \relsa \src{\taint''}$

Since $\src{\exprEval{}{}{}}$ is deterministic, we have $\src{v'} = \src{v''}$ and $\src{\taint'} = \src{\taint''}$.

Thus, $\src{v} \relsa \src{v'}$ and $\src{\taint} \relsa \src{\taint'}$
\end{proof}

\begin{lemma}[Safe-Equivalent configurations produce safe-equivalent expression evaluations-2]\label{lemma:safe-equi-expression2}
If 
\begin{enumerate}
    \item $\src{\exprEval{A}{e}{v : \taint}}$ and
    \item $\src{A} \relsa \src{A'}$
\end{enumerate}

Then there exists $v'$ and $\taint'$ such that
\begin{enumerate}[label=\Roman*]
    \item $\src{\exprEval{A'}{e}{v' : \taint'}}$ and
    \item $\src{v : \taint} \relsa \src{v' : \taint'}$ and
\end{enumerate}
\end{lemma}

\begin{proof}
We proceed by structural induction on $\src{\exprEval{A}{e}{v : \taint}}$:

\begin{description}
    \item[\Cref{tr:E-val}]
    Then by definition $\src{e} = \src{v}$ and $\src{\taint} = \src{\safeta}$. Pick $\src{v'} = \src{v}$ and $\src{\taint'} = \src{\safeta}$. 
    
    We can now derive $\src{\exprEval{A'}{e}{v' : \taint'}}$ by \Cref{tr:E-val}.
    
    Now, $\src{v : \taint} \relsa \src{v' : \taint'}$ trivially holds.
    
    \item[\Cref{tr:E-lookup}]
    Then $\src{e} = \src{r}$ where $\src{r} \in \Reg$ and $\src{A(r) = v : \taint}$.
    
    Since $\src{A} \relsa \src{A'}$ we get $\src{A'(r) = v' : \taint'}$ with $\src{v : \taint} \relsa \src{v' : \taint'}$.

    \item[\Cref{tr:E-binop}]
    
    Then $\src{e} = \src{e_1} \otimes \src{e_2}$, $\src{v} =\src{v_1} \otimes \src{v_2'}$ and $\src{\taint} = \src{\taint_1 \lub \taint_2}$. 
    
    From the IH we get $\src{\exprEval{A'}{e_1}{v_1' : \taint_1'}}$ and $\src{\exprEval{A'}{e_2}{v_2' : \taint_2'}}$ such that $\src{v_1 : \taint_1} \relsa \src{v_1' : \taint_1'}$ and $\src{v_2 : \taint_2} \relsa \src{v_2' : \taint_2'}$.
    
    Pick $\src{v'} = \src{v_1'} \otimes \src{v_2'}$ and $\src{\taint'} = \src{\taint_1'} \lub \src{\taint_2'}$.
    We can now derive $\src{\exprEval{A'}{e}{v' : \taint'}}$ by \Cref{tr:E-binop} using $\src{\exprEval{A'}{e_1}{v_1 : \taint_1}}$ and $\src{\exprEval{A'}{e_2}{v_2 : \taint_2}}$.

    Applying \Thmref{lemma:bounds-preserve-relation} we get that $\src{\taint_1 \lub \taint_2} \relsa \src{\taint_1' \lub \taint_2'}$ which means
    $\src{\taint_1 \lub \taint_2} = \src{\taint_1' \lub \taint_2'}$.
    
    We continue by case distinction on $\src{\taint_1 \lub \taint_2}$: 
    \begin{description}
        \item[$\src{\taint_1 \lub \taint_2} = \src{\safeta}$]

        Thus, $\src{\taint_1} = \src{\safeta}$, $\src{\taint_2} = \src{\safeta}$, $\src{\taint_1'} = \src{\safeta}$ and $\src{\taint_2'} = \src{\safeta}$.
        From $\src{v_1 : \taint_1} \relsa \src{v_1' : \taint_1'}$ and $\src{v_2 : \taint_2} \relsa \src{v_2' : \taint_2'}$ we get $\src{v1} = \src{v_1'}$ and $\src{v_2} = \src{v_2'}$. 
    
        Now, we can follow that $\src{v1} \otimes \src{v_2} = \src{v_1'} \otimes \src{v_2'}$.
        
        Thus, we have $\src{v1} \otimes \src{v_2} \src{: \taint_1 \lub \taint_2} \relsa  \src{v_1'} \otimes \src{v_2' : \taint_1' \lub \taint_2'}$ and after substitution arrive at $\src{v : \taint} \relsa \src{v' : \taint'}$.
    
        \item[$\src{\taint_1 \lub \taint_2} = \src{\unta}$]
        
        Then, $\src{\taint_1' \lub \taint_2'} = \src{\unta}$ as well.
        
        Thus, we can derive $\src{v : \taint_1 \lub \taint_2} \relsa \src{v' : \taint_1' \lub \taint_2'}$
        
        Substituting $\src{\taint}$ and $\src{\taint'}$ we get $\src{v : \taint} \relsa \src{v' : \taint'}$.
    \end{description}

    \item[\Cref{tr:E-unop}]
    
    Then $\src{e} = \ominus \src{e_1}$ and from the IH we get $\src{\exprEval{A'}{e_1}{v_1' : \taint_1'}}$ with $\src{v_1 : \taint_1} \relsa \src{v_1' : \taint_1'}$.
    
    We can now derive $\src{\exprEval{A'}{e}{v' : \taint'}}$ by \Cref{tr:E-unop} using $\src{\exprEval{A'}{e_1}{v_1' : \taint_1'}}$.
    
    Since, $\src{\taint} = \src{\taint_1}$ and $\src{\taint'} = \src{\taint_1'}$ we know that $\src{\taint} = \src{\taint'}$.
    
    We do a case distinction on $\taint$: 
    \begin{description}
        \item[$\src{\taint} = \src{\safeta}$]
        
        Similar to corresponding case in \Cref{tr:E-binop} above.
        \item[$\src{\taint} = \src{\unta}$]
        
        Similar to corresponding case in \Cref{tr:E-binop} above.
    \end{description}
\end{description}
\end{proof}

\begin{lemma}[$\lub$ and $\glb$ preserve Safe-Equivalence relation $\relsa$]\label{lemma:bounds-preserve-relation}
If 
\begin{enumerate}
    \item $\src{\taint_1} \relsa \src{\taint_1'}$ and 
    \item $\src{\taint_2} \relsa \src{\taint_2'}$
\end{enumerate}
Then 
\begin{enumerate}[label=\Roman*]
    \item $\src{\taint_1 \lub \taint_2} \relsa \src{\taint_1' \lub \taint_2'}$ and
    \item $\src{\taint_1 \glb \taint_2} \relsa \src{\taint_1' \glb \taint_2'}$ 
\end{enumerate}
\end{lemma}

\begin{proof}
\begin{description}
    \item[I]
    We do a case distinction on $\src{\taint_1 \lub \taint_2}$:
    \begin{description}
        \item[$\src{\taint_1 \lub \taint_2} = \src{\safeta}$]
        
        Then by definition of $\lub$ we know that $\src{\taint_1} = \src{\safeta}$ and $\src{\taint_2} = \src{\safeta}$.
        
        From our assumptions we can derive $\src{\taint_1'} = \src{\safeta}$ and $\src{\taint_2'} = \src{\safeta}$ as well.
        
        Thus, $\src{\taint_1' \lub \taint_2'} = \src{\safeta}$ and we have $\src{\taint_1 \lub \taint_2} \relsa \src{\taint_1' \lub \taint_2'}$ as needed.
        
        \item[$\src{\taint_1 \lub \taint_2} = \src{\unta}$]
        
        Then, either $\src{\taint_1} = \src{\unta}$ or $\src{\taint_2} = \src{\unta}$.
        
        W.l.o.g assume that $\src{\taint_1} = \src{\unta}$. By assumption, we get $\src{\taint_1'} = \src{\unta}$ as well. 
        
        Thus, $\src{\taint_1' \lub \taint_2'} = \src{\unta}$ independent of the value of $\src{\taint_2'}$ and we have $\src{\taint_1 \lub \taint_2} \relsa \src{\taint_1' \lub \taint_2'}$ as needed.
        
    \end{description}
    
    \item[II]
    
    The proof of this case is similar to the case for $\lub$.
\end{description}
\end{proof} 
 \section{$\Sv$: SS overapproximates SNI}

The proof structure is as follows:

\begin{center}
	\scalebox{0.85}{
    \begin{forest}
    for tree={
        align = left,
        font = \footnotesize,
        forked edge,
    }
        [\Thmref{thm:v4:ss-impl-sni}
            [\Thmref{lemma:v4-bigspecarrow-preserves-safe}
                [\Thmref{lemma:v4-specarrow-preserves-safe}
                    [\Thmref{lemma:non-spec-preserves-safe}]
                ]
            ]
            [\Thmref{lemma:v4-low-equivalent-programs-low-equivalent-states}]
        ]
    \end{forest}
    }
\end{center}

\begin{insight}
We use a small trick when doing inversion on $\trgS{\specarrowS{}}$. To compute a step in that semantics, we need to do a step in the taint semantics and in the value semantics. These steps are later combined using \Cref{tr:v4-combine}. Since this step is always necessary to derive a step, we will omit it from the inversion chain. Similarly for \Cref{tr:v4-t-context} and \Cref{tr:v4-context} that just unfold the state and are also always necessary. This just makes our life easier.
\end{insight}

\begin{theorem}[$\Sv$: \sstext{} implies SNI]\label{thm:v4:ss-impl-sni}
If
\begin{enumerate}
    \item $\contractSpec{}{\Sv} \vdash \trgS{P} : \ss $
\end{enumerate}
Then
\begin{enumerate}[label=\Roman*]
    \item $\contractSpec{}{\Sv} \vdash \trgS{P} : \sni $
\end{enumerate}

\end{theorem}

\begin{proof}
Let $\trgS{P}$ be an arbitrary program such that $\contractSpec{}{\Sv} \vdash \trgS{P} : \ss$.

The proof proceeds by contradiction. 

Assume that $\contractSpec{}{\Sv} \nvdash \trgS{P} : \sni$. That is, there is another program $\trgS{P'}$ and traces $\trgS{\tra{_1}} \in \trgS{ \contractSpec{}{\Sv}(\SInit{P}) }$ and  $\trgS{\tra{_2}} \in \trgS{ \contractSpec{}{\Sv}(\SInit{P'}) }$ such that $\trgS{P} \loweq \trgS{P'}$, $\nspecProject{\trgS{\tra{_1}}} = \nspecProject{\trgS{\tra{_2}}}$, and $\trgS{\tra{_1}} \neq \trgS{\tra{_2}}$.

Since $\trgS{\tra{_1}} \in \trgS{ \contractSpec{}{\Sv}(\SInit{P}})$ we have $\trgS{\amTracevS{P}{\tauStackT}}$ and by \Cref{tr:v4-trace} we know that there exists $\SigmaSDS{F}$ such that $\trgS{\initFuncS{\trgS{(P)}} \bigspecarrowS{\tra{_1}}\ \SigmaSDS{F}}$. Similarly, we get  $\trgS{\initFuncS{(\trgS{P'})} \bigspecarrowS{\tra{_2}}\ \SigmaSDSP{F}}$.

Combined with the fact that $\trgS{\tra{_1}} \neq \trgS{\tra{_2}}$, it follows that there exists speculative states $\SigmaSt{\dagger}, \SigmaSDS{}, \SigmaSt{'}, \SigmaSt{''}$ and sequences of observations $\trgS{\tauStack, \tauStack_{end}, \tauStack'_{end}}, \trgS{\aca{^\taint}, {\tau'}^{\taint'} } $ such that $\trgS{\aca{^{\taint}}} \neq \trgS{ {\tau'}^{\taint'}}$  and:

\begin{align*}
\trgS{\initFuncS{(\trgS{P})}} & \trgS{\bigspecarrowS{\tauStack}} \SigmaSt{\dagger} \specarrowS{\aca{^{\taint}}} \SigmaSt{\dagger\dagger} \bigspecarrowS{\tauStack_{end}} \trgS{\SigmaSDS{F}} 
\\
\trgS{\initFuncS{(\trgS{P'})}} & \bigspecarrowS{\tauStack} \SigmaSt{'} \specarrowS{{\tau'}^{\taint'}} \SigmaSt{''} \bigspecarrowS{\tauStack'_{end}} \trgS{\SigmaSDSP{F}}
\end{align*}

From $\contractSpec{}{\Sv} \vdash \trgS{P} : \ss $ we get that $\safe{\trgS{\tra{_1}}}$.

From \Thmref{lemma:v4-low-equivalent-programs-low-equivalent-states} we get ${\initFuncS{(\trgS{P})}} \relsa {\initFuncS{(\trgS{P'})}}$.

Using \Thmref{lemma:v4-bigspecarrow-preserves-safe} on $\initFuncS{(\trgS{P})} \trgS{\bigspecarrowS{\tauStack}} \SigmaSt{\dagger}$ and the fact $\initFuncS{(\trgS{P})} \relsa {\initFuncS{(\trgS{P'})}}$ we get $\trgS{\initFuncS{(\trgS{P'})}} \bigspecarrowS{\tauStack} \SigmaSt{'''} $ and  $\SigmaSt{\dagger} \relsa \SigmaSt{'''}$

We now need to show that $\SigmaSt{'''} = \SigmaSt{'}$.
For that, we use the determinism of $\specarrowS{}$ on $\trgS{\initFuncS{(\trgS{P'})}} \bigspecarrowS{\tauStack} \SigmaSt{'''} $ and $\trgS{\initFuncS{(\trgS{P'})}} \bigspecarrowS{\tauStack} \SigmaSt{'}$ which directly yields $\SigmaSt{'''} = \SigmaSt{'}$.

Let us collect everything we know:
\begin{enumerate}[label=(\alph*)]
    \item $\safe{\trgS{\tra{_1}}}$
    \item ${\initFuncS{(\trgS{P})}} \relsa {\initFuncS{(\trgS{P'})}}$
    \item $\SigmaSt{\dagger} \relsa \SigmaSt{'}$
    \item $\nspecProject{\trgS{\tra{_1}}} = \nspecProject{\trgS{\tra{_2}}}$
    \item $\trgS{\aca{^{\taint}}} \neq \trgS{ {\tau'}^{\taint'}}$
\end{enumerate}

We apply \Thmref{lemma:v4-specarrow-preserves-safe} on $\SigmaSt{\dagger} \specarrowS{\aca{^{\taint}}} \SigmaSt{\dagger\dagger}$ with $\SigmaSt{\dagger} \relsa \SigmaSt{'}$ and get $\SigmaSt{'} \specarrowS{\aca{^\taint}} \SigmaSt{''''}$ with $\SigmaSt{\dagger\dagger} \relsa \SigmaSt{''''}$.

Applying the determinism of $\specarrowS{}$  on  $\SigmaSt{'} \specarrowS{\aca{^\taint}} \SigmaSt{''''}$ and $ \SigmaSt{'} \specarrowS{{\tau'}^{\taint'}} \SigmaSt{''}$ we get $\SigmaSt{''''} = \SigmaSt{''}$ and $\trgS{\aca{^\taint}} = \trgS{{\tau'}^{\taint'}}$

Now we have $\trgS{{\tau'}^{\taint'}} = \trgS{\aca{^\taint}}$ contradicting our assumption $\trgS{\aca{^\taint}} \neq \trgS{ {\tau'}^{\taint'}}$.

\end{proof}

\begin{lemma}[$\Sv$: Low-equivalent programs, safe-equivalent initial states]\label{lemma:v4-low-equivalent-programs-low-equivalent-states}
\begin{align*}
\forall \trgS{P = M, \OB{F}, \OB{I}}, \trgS{P' = M', \OB{F'}, \OB{I'}}. &
	\text{ if } \trgS{P} \loweq \trgS{P'} 
	\text{ then } \initFuncS{(\trgS{P})} \relsa \initFuncS{(\trgS{P'})}
\end{align*}
\end{lemma}
\begin{proof}

By definition of $\initFuncS{}$ (\Cref{tr:v4-init}) we know that:
\begin{align*}
    \initFuncS{P} = \SigmaS \\ 
    \initFuncS{P'} = \SigmaS' \\ 
    \SigmaS = \initFunc(M, \OB{F}, \OB{I}), \bot, \safeta \\
    \SigmaS' = \initFunc(M', \OB{F'}, \OB{I'}), \bot, \safeta \\
    \initFunc(M, \OB{F}, \OB{I}) =& \OB{F}, \OB{I}, \emptyset, \tup{p, M_0, A_1} \\
    M_0 =& M_1 \cup M \cup M_2 \\
    M_1 =& \myset{ n\mapsto 0 : \safeta }{n\in\mb{N}\setminus\dom{M} } \\
	M_2 =& \myset{ -n\mapsto 0 : \unta }{ n\in\mb{N}, -n\notin\dom{M} } \\
	A_0 =& \myset {x \mapsto 0 : \safeta}{x \in \Reg} \\
	A_1 =& A_0[\pc \mapsto \ffun{main} : \safeta] \\ 
    \initFunc(M', \OB{F'}, \OB{I'}) =& \OB{F'}, \OB{I'}, \emptyset, \tup{p, M_0', A_1'} \\
	M_0' =& M_1' \cup M' \cup M_2' \\
    M_1' =& \myset{ n\mapsto 0 : \safeta }{n\in\mb{N}\setminus\dom{M'} } \\
	M_2' =& \myset{ -n\mapsto 0 : \unta }{ n\in\mb{N}, -n\notin\dom{M'} } \\
	A_0' =& \myset {x \mapsto 0 : \safeta}{x \in \Reg} \\
	A_1' =& A_0'[\pc \mapsto \ffun{main} : \safeta]
\end{align*}

From $\trgS{P} \loweq \trgS{P'}$ we get:
\begin{align*}
    \trgS{\OB{F}} = \trgS{\OB{F'}} \\ 
    \trgS{\OB{I}} = \trgS{\OB{I'}} \\ 
    \trgS{M} \loweq \trgS{M'}
\end{align*}

To show $\SigmaS \relsa \SigmaS'$ it remains to show that $ \initFunc(M, \OB{F}, \OB{I}) \relsa \initFunc(M', \OB{F'}, \OB{I'})$.

Because of $\trgS{\OB{F}} = \trgS{\OB{F'}}$ and $\trgS{\OB{I}} = \trgS{\OB{I'}}$ it remains to show that $A_1 \relsa A_1'$ and $M_0 \relsa M_0'$:
\begin{description}
    \item[$\trgS{A_1} \relsa \trgS{A_1'}$]
    Trivially follows by inspecting the definition of $\trgS{A_1}$ and $\trgS{A_1'}$.
    \item[$\trgS{M_0} \relsa \trgS{M_0'}$]
    Follows from the definition of $\trgS{M_0}$ and $\trgS{M_0'}$ and the fact that $\trgS{M} \loweq \trgS{M'}$. 
\end{description}

Thus, $\initFuncS{P} \relsa \initFuncS{P'}$ as required.

\end{proof}

\begin{lemma}[$\Sv$: Steps of $\bigspecarrowS{}$ with safe observations preserve Safe-equivalence]\label{lemma:v4-bigspecarrow-preserves-safe}
If
\begin{enumerate}
    \item $\trgS{\SigmaS} \trgS{\bigspecarrowS{\tauStack}} \SigmaSt{\dagger}$ and 
    \item $\trgS{\SigmaS} \relsa \trgS{\SigmaSt{'}}$
    \item $\safe{\tauStack}$ \end{enumerate}
Then there exists $\SigmaSt{''}$ such that
\begin{enumerate}
    \item $\trgS{\SigmaSt{'}} \bigspecarrowS{\tauStack} \SigmaSt{''}$ and
    \item $\SigmaSt{\dagger} \relsa \SigmaSt{''}$
\end{enumerate}
\end{lemma}

\begin{proof}
The proof proceeds by induction on $\trgS{\SigmaS} \trgS{\bigspecarrowS{\tauStack}} \SigmaSt{\dagger}$:

\begin{description}
    \item[\Cref{tr:v4-reflect}]
    
    Then we have $\SigmaS \bigspecarrowS{\tauStack} \SigmaSt{\dagger}$ where $\tauStack = \empTr$ and $\SigmaSt{\dagger} = \SigmaS$.
    
    Thus, we use \Cref{tr:v4-reflect} on $\SigmaSt{'}$ and get $\SigmaSt{'} \bigspecarrowS{\tauStack'} \SigmaSt{''}$, where where $\tauStack' = \empTr$ and $\SigmaSt{''} = \SigmaSt{'}$.
    
    Thus, $\tauStack = \tauStack'$ trivially holds and $\SigmaSt{\dagger} \relsa \SigmaSt{''}$ holds by $\trgS{\SigmaS} \relsa \trgS{\SigmaSt{'}}$.
    
    \item[\Cref{tr:v4-silent}]
    
    Then we have $\trgS{\SigmaS \bigspecarrowS{\tra{^\taint}} \SigmaSt{\dagger\dagger}}$ and $\trgS{\SigmaSt{\dagger\dagger} \specarrowS{\epsilon} \SigmaSt{\dagger}}$.
    
    Applying the IH on $\SigmaS \bigspecarrowS{\tra{^\taint}} \SigmaS''$ we get:
    \begin{align*}
    \SigmaSt{'} \bigspecarrowS{\tauStack''} \SigmaSt{'''} \\
    \SigmaSt{\dagger\dagger} \relsa \SigmaSt{'''} \\
    \safe{\tauStack''}
    \end{align*}
    
    Applying \Thmref{lemma:v4-specarrow-preserves-safe} on $\trgS{\SigmaSt{\dagger\dagger} \specarrowS{\epsilon} \SigmaSt{\dagger}}$ using $\SigmaSt{\dagger\dagger} \relsa \SigmaSt{'''}$ and $\safe{\empTr}$ we get $\SigmaSt{'''} \specarrowS{\empTr} \SigmaSt{*}$ and $\SigmaSt{\dagger} \relsa \SigmaSt{*}$.
    
    Now choose $\SigmaSt{''} = \SigmaSt{*}$.
    
    We now use \Cref{tr:v4-silent} and get $\trgS{\SigmaSt{'}} \bigspecarrowS{\tauStack} \SigmaSt{''}$
    
    $\SigmaSt{\dagger} \relsa \SigmaSt{''}$ follows from $\SigmaSt{\dagger} \relsa \SigmaSt{*}$ and $\SigmaSt{''} = \SigmaSt{*}$.
    
    \item[\Cref{tr:v4-single}]
    
    Then we have $\trgS{\SigmaS \bigspecarrowS{\tra{^\taint}} \SigmaSt{\dagger\dagger}}$ and $\trgS{\SigmaSt{\dagger\dagger} \specarrowS{\tau^{\taint}} \SigmaSt{\dagger}}$.
    
    Furthermore, we get:
    \begin{align*}
    \trgS{\SigmaS''} =& \trgS{\phiStackS \cdot  \tup{\Omega, n}} \\
    \trgS{\SigmaS'} =& \trgS{\phiStackS \cdot \tup{\Omega', n' }} \\
    \ffun{\trgS{\sigma''(\pc)}} =& \trgS{f} \\
    \ffun{\trgS{\sigma'(\pc)}} =& \trgS{f'} \\
    \text{ if } \trgS{f == f'}& and ~\trgS{f \in \OB{I}} \text{ then }& \trgS{\tau^{\taint} = \epsilon} \text{ else } \trgS{\tau^{\taint} = \tau^{\taint}}
    \end{align*}
    
    Applying the IH on $\SigmaS \bigspecarrowS{\tra{^\taint}} \SigmaS''$ we get:
    \begin{align*}
    \SigmaSt{'} \bigspecarrowS{\tauStack''} \SigmaSt{'''} \\
    \SigmaSt{\dagger\dagger} \relsa \SigmaSt{'''} \\
    \safe{\tauStack''}
    \end{align*}
    
    Applying \Thmref{lemma:v4-specarrow-preserves-safe} on $\trgS{\SigmaSt{\dagger\dagger} \specarrowS{\tau^{\taint}} \SigmaSt{\dagger}}$ using $\SigmaSt{\dagger\dagger} \relsa \SigmaSt{'''}$ and $\safe{\tau^{\taint}}$ (since $\safe{\tauStack}$) we get $\SigmaSt{'''} \specarrowS{\tau^{\taint}} \SigmaSt{*}$ and $\SigmaSt{\dagger} \relsa \SigmaSt{*}$.
    
    Now choose $\SigmaSt{''} = \SigmaSt{*}$.
    
    We can now use \Cref{tr:v4-single} and get $\trgS{\SigmaSt{'}} \bigspecarrowS{\tauStack} \SigmaSt{''}$
    
    $\SigmaSt{\dagger} \relsa \SigmaSt{''}$ follows from $\SigmaSt{\dagger} \relsa \SigmaSt{*}$ and $\SigmaSt{''} = \SigmaSt{*}$.
\end{description}
\end{proof}

\begin{lemma}[Steps of $\specarrowS{}$ with safe observations preserve Safe-equivalence]\label{lemma:v4-specarrow-preserves-safe}
If
\begin{enumerate}
    \item $\SigmaSt{\dagger} \relsa \SigmaSt{'}$
    \item $\SigmaSt{\dagger} \specarrowS{\tau} \SigmaSt{\dagger\dagger}$
    \item $\safe{\tau}$
\end{enumerate}
Then 
\begin{enumerate}
    \item $\SigmaSt{'} \specarrowS{\tau} \SigmaSt{''}$
    \item $\SigmaSt{\dagger\dagger} \relsa \SigmaSt{''}$
\end{enumerate}
\end{lemma}

\begin{proof}

We have 
\begin{align*}
    \SigmaSt{\dagger} =& \phiStackSv \cdot (\trgS{\taint, \Omega, n}) \\
    \SigmaSt{'} =& \phiStackSv' \cdot (\trgS{\taint', \Omega', n}) \\ 
    \trgS{\Omega} \relsa& \trgS{\Omega'} \\
    \trgS{\taint} =& \trgS{\taint'} \\
    \trgS{\Omega(\pc)} =& \trgS{\Omega'(\pc)}
\end{align*}

The proof proceeds by inversion on $\SigmaSt{\dagger} \specarrowS{\tau} \SigmaSt{\dagger\dagger}$:
(Implicitly unpack the context rule since it is always used)

\begin{description}

    \item[\Cref{tr:v4-barr}]
    Then we have $\trgS{(\taint, \Omega, \bot) \specarrowS{\empTr} (\taint, \Omega^{\dagger}, \bot)}$ with $\trgS{\Omega(\pc)} = \pbarrier$ and $\trgS{\Omega^{\dagger}} = \trgS{\Omega[\pc \mapsto \pc + 1]}$.
    
    Since $\trgS{\Omega(\pc)} = \pbarrier$ and $\trgS{\Omega(\pc)} = \trgS{\Omega'(\pc)}$ we can use \Cref{tr:v4-barr} 
    \\
    to derive $\trgS{(\taint', \Omega', \bot) \specarrowS{\empTr} (\taint', \Omega'', \bot)}$, where $\trgS{\Omega'' = \Omega'[\pc \mapsto \pc + 1]}$.
    
    Since $\trgS{\Omega^{\dagger}} \relsa \trgS{\Omega''}$ we get $(\trgS{\taint, \Omega^{\dagger}, \bot}) \relsa (\trgS{\taint', \Omega'', \bot})$.
    
    Combined with $\phiStackSv \relsa \phiStackSv'$ we get $\trgS{\SigmaSt{\dagger\dagger} \relsa \SigmaSt{''}}$
    
    \item[\Cref{tr:v4-barr-spec}]
    
    Then we have $\trgS{(\taint, \Omega, n) \specarrowS{\empTr} (\taint, \Omega', 0)}$ with $\trgS{\Omega(\pc)} = \pbarrier$.
    
    The case is analogous to \Cref{tr:v4-barr}.
    
    \item[\Cref{tr:v4-nospec-eps}]
    
    Then we have $\trgS{(\taint, \Omega, n + 1) \specarrowS{\empTr} (\taint, \Omega', n)}$ with $\trgS{\Omega(\pc) \neq \pbarrier \mid \pstore{x}{e}}$ and $\trgS{\Omega \nsarrow{\empTr} \Omega'}$.
    
     We use \Cref{lemma:non-spec-preserves-safe} on $\trgS{\Omega \nsarrow{\empTr} \Omega'}$ and $\trgS{\Omega \relsa \Omega'}$ and get $\trgS{\Omega' \nsarrow{\empTr} \Omega''}$ with $\trgS{\Omega^{\dagger} \relsa \Omega''}$.
    
    Thus, we can use \Cref{tr:v4-nospec-eps} to derive $\trgS{(\taint', \Omega', n + 1) \specarrowS{\empTr} (\taint', \Omega'', n)}$.
    
    Since $\trgS{\Omega^{\dagger} \relsa \Omega''}$ we get $\trgS{(\taint, \Omega^{\dagger}, n) \relsa (\taint', \Omega'', n)}$.
    
    Combined with $\trgS{\phiStackSv \relsa \phiStackSv'}$ we get $\trgS{\SigmaSt{\dagger\dagger} \relsa \SigmaSt{''}}$

    \item[\Cref{tr:v4-nospec-act}]
    
    Then we have $\trgS{(\taint, \Omega, n + 1) \specarrowS{\tau^{\taint \glb \taint_1}} (\taint, \Omega^{\dagger}, n)}$ with $\Omega(\pc) \neq \pbarrier \mid \pstore{x}{e}$ and $\trgS{\Omega \nsarrow{\tau^{\taint_1}} \Omega^{\dagger}}$.
    
    We use \Thmref{lemma:non-spec-preserves-safe} on $\trgS{\Omega \nsarrow{\tau^{\taint_1}} \Omega^{\dagger}}$ and $\trgS{\Omega \relsa \Omega^{'}}$ and get $\trgS{\Omega' \nsarrow{\tau^{\taint_1}} \Omega''}$ with $\trgS{\Omega^{\dagger} \relsa \Omega''}$.
    
    Thus, we can use \Cref{tr:v4-nospec-act} to derive $\trgS{(\taint', \Omega', n' + 1) \specarrowS{\tau^{\taint' \glb \taint_1}} (\taint', \Omega'', n')}$.
    
    Since $\trgS{\Omega^{\dagger} \relsa \Omega''}$ we get $\trgS{(\taint, \Omega^{\dagger}, n' + 1) \relsa (\taint', \Omega'', n')}$.
    
    Combined with $\trgS{\phiStackSv \relsa \phiStackSv'}$ we get $\trgS{\SigmaSt{\dagger\dagger} \relsa \SigmaSt{''}}$.
    
    Furthermore, $\trgS{\tau^{\taint' \glb \taint_1}} = \trgS{\tau^{\taint \glb \taint_1}}$ follows from $\trgS{\taint} = \trgS{\taint'}$.
    
    \item[\Cref{tr:v4-skip-att}]
    
    Then we have $\trgS{(\taint, \Omega, n + 1) \specarrowS{\tau^{\taint \glb \taint_1}} (\taint, \Omega^{\dagger}, n)}$ with $\trgS{\Omega(\pc)} =  \trgS{\pstore{x}{e}}$ and $\trgS{\Omega \nsarrow{\tau^{\taint_1}} \Omega^{\dagger}}$ and  $\ffun{\trgS{\Omega(\pc)}} = \trgS{f}$ and $\trgS{f \in \OB{I}}$.
    
    Thus, no speculation is happening.
    
    The rest of the case is analogous to \Cref{tr:v4-nospec-act}.
    
    \item[\Cref{tr:v4-skip}]
    Then we have $\trgS{(\taint, \Omega, n + 1) \specarrowS{\tau^{\taint \glb \taint_1}} \tup{\taint \Omega^{\dagger}, n} \cdot \tup{\unta, \Omega^{\dagger\dagger}, j}}$ with $\trgS{\Omega(\pc)} = \trgS{\pstore{x}{e}}$ and $\trgS{\Omega \nsarrow{\tau^{\taint_1}} \Omega^{\dagger}}$.

    Furthermore, we have $ \trgS{\Omega^{\dagger\dagger} = \OB{F}; \OB{I}; \OB{B}; \sigma^{\dagger\dagger}}$ and $\trgS{\sigma^{\dagger\dagger} = \sigma[\pc \mapsto \Omegav(\pc) + 1]}$.

    We use  \Thmref{lemma:non-spec-preserves-safe} on $\src{\Omega \nsarrow{\tau^{\taint_1}} \Omega^{\dagger}}$  and $\src{\Omega \relsa \Omega^{\dagger}}$ and get $\Omega' \nsarrow{\tau^{\taint_1}} \Omega''$ with $\trgS{\Omega^{\dagger} \relsa \Omega''}$.

    Next, we define $\trgS{\Omega'''}$ as $\trgS{\Omega''' = \OB{F'}; \OB{I'}; \OB{B'}; \sigma'''}$ and $\trgS{\sigma''' = \sigma[\pc \mapsto \sigma'(\pc) + 1]}$ where $\trgS{\Omega'} =  \trgS{\OB{F'}; \OB{I'}; \OB{B'}; \sigma'}$.

     From $\trgS{\Omega \relsa \Omega'}$ we can derive that $\trgS{\Omega^{\dagger\dagger} \relsa \Omega'''}$.

     We can now apply \Cref{tr:v4-skip} to derive $\trgS{(\taint', \Omega', n' + 1) \specarrowS{\tau^{\taint' \glb \taint_1}} \tup{\taint' \Omega'', n'} \cdot \tup{\unta, \Omega''', j'}}$.

     From $\trgS{\SigmaS \relsa \SigmaS'}$ we get $\trgS{j = j'}$, $\trgS{\taint = \taint'}$ and we can now conclude that $\trgS{\tau^{\taint \glb \taint_1}} = \trgS{\tau^{\taint' \glb \taint_1}}$.

     $\trgS{\SigmaS{\dagger\dagger} \relsa \SigmaS''}$ follows from $\trgS{\SigmaS^{\dagger} \relsa \SigmaS'}$, $\trgS{\Omega^{\dagger} \relsa \Omega''}$ and $\trgS{\Omega^{\dagger\dagger} \relsa \Omega'''}$.

    \item[\Cref{tr:v4-rollback}]
    Then we have $\trgS{(\taint, \Omega, 0) \specarrowS{\rollbackObsS{}}} $.
    
    Thus we can apply \Cref{tr:v4-rollback} and get $\trgS{(\taint', \Omega', n) \specarrowS{\rollbackObsS{}}}$.
    
    Now $\trgS{\SigmaSt{\dagger\dagger} \relsa \SigmaSt{''}}$ directly follows from $\trgS{\phiStackSv \relsa \phiStackSv'}$.
\end{description}
\end{proof}

 \section{$\SLSv$: SS overapproximates SNI}

The proof structure is as follows:

\begin{center}
	\scalebox{0.8}{
    \begin{forest}
    for tree={
        align = left,
        font = \footnotesize,
        forked edge,
    }
        [\Thmref{thm:sls-ss-impl-sni}
            [\Thmref{lemma:sls-bigspecarrow-preserves-safe}
                [\Thmref{lemma:sls-specarrow-preserves-safe}
                    [\Thmref{lemma:non-spec-preserves-safe}]
                ]
            ]
            [\Thmref{lemma:sls-low-equivalent-programs-low-equivalent-states}]
        ]
    \end{forest}
    }
\end{center}

\begin{theorem}[$\SLSv$: \sstext{} implies SNI]\label{thm:sls-ss-impl-sni}
If
\begin{enumerate}
    \item $\contractSpec{}{\SLSv} \vdash \trgSLS{P} : \ss $
\end{enumerate}
Then
\begin{enumerate}[label=\Roman*]
    \item $\contractSpec{}{\SLSv} \vdash \trgSLS{P} : \sni $
\end{enumerate}

\end{theorem}
\begin{proof}

    Let $\trgSLS{P}$ be an arbitrary program such that $\contractSpec{}{\SLSv} \vdash \trgSLS{P} : \ss$.

    The proof proceeds by contradiction. 
    
    Assume that $\contractSpec{}{\SLSv} \nvdash \trgSLS{P} : \sni$. That is, there is another program $\trgSLS{P'}$ and traces $\trgSLS{\tra{_1}} \in \trgSLS{ \contractSpec{}{\Sv}(\SInit{P}) }$ and  $\trgSLS{\tra{_2}} \in \trgSLS{ \contractSpec{}{\SLSv}(\SInit{P'}) }$ such that $\trgSLS{P} \loweq \trgSLS{P'}$, $\nspecProject{\trgSLS{\tra{_1}}} = \nspecProject{\trgSLS{\tra{_2}}}$, and $\trgSLS{\tra{_1}} \neq \trgSLS{\tra{_2}}$.

    Since $\trgSLS{\tra{_1}} \in \trgSLS{ \contractSpec{}{\SLSv}(\SInit{P}})$ we have $\trgSLS{\amTracevSLS{P}{\tauStackT}}$ and by \Cref{tr:sls-trace} we know that there exists $\SigmaSLSDS{F}$ such that $\trgSLS{\initFuncSLS{\trgSLS{(P)}} \bigspecarrowSLS{\tra{_1}}\ \SigmaSLSDS{F}}$. Similarly, we get  $\trgSLS{\initFuncSLS{(\trgSLS{P'})} \bigspecarrowSLS{\tra{_2}}\ \SigmaSLSDSP{F}}$.
    
    Combined with the fact that $\trgSLS{\tra{_1}} \neq \trgSLS{\tra{_2}}$, it follows that there exists speculative states $\SigmaSLSt{\dagger}, \SigmaSLSDS{}, \SigmaSLSt{'}, \SigmaSLSt{''}$ and sequences of observations $\trgSLS{\tauStack, \tauStack_{end}, \tauStack'_{end}}, \trgSLS{\aca{^\taint}, {\tau'}^{\taint'} } $ such that $\trgSLS{\aca{^{\taint}}} \neq \trgSLS{ {\tau'}^{\taint'}}$  and:
    
    \begin{align*}
    \trgSLS{\initFuncSLS{(\trgSLS{P})}} & \trgSLS{\bigspecarrowSLS{\tauStack}} \SigmaSLSt{\dagger} \specarrowSLS{\aca{^{\taint}}} \SigmaSLSt{\dagger\dagger} \bigspecarrowSLS{\tauStack_{end}} \trgSLS{\SigmaSLSDS{F}} 
    \\
    \trgSLS{\initFuncSLS{(\trgSLS{P'})}} & \bigspecarrowSLS{\tauStack} \SigmaSLSt{'} \specarrowSLS{{\tau'}^{\taint'}} \SigmaSLSt{''} \bigspecarrowSLS{\tauStack'_{end}} \trgSLS{\SigmaSLSDSP{F}}
    \end{align*}
    
    From $\contractSpec{}{\SLSv} \vdash \trgSLS{P} : \ss $ we get that $\safe{\trgSLS{\tra{_1}}}$.

    From \Thmref{lemma:sls-low-equivalent-programs-low-equivalent-states} we get ${\initFuncSLS{(\trgSLS{P})}} \relsa {\initFuncSLS{(\trgSLS{P'})}}$.

    Using \Thmref{lemma:sls-bigspecarrow-preserves-safe} on $\initFuncSLS{(\trgSLS{P})} \trgSLS{\bigspecarrowSLS{\tauStack \cdot \tau^{\taint}}} \SigmaSLSt{\dagger\dagger}$ and the fact that $\initFuncSLS{(\trgSLS{P})} \relsa {\initFuncSLS{(\trgSLS{P'})}}$ we get $\trgSLS{\initFuncSLS{(\trgSLS{P'})}} \bigspecarrowSLS{\tauStack \cdot \tau^{\taint}} \SigmaSLSt{'''} $ and  $\SigmaSLSt{\dagger\dagger} \relsa \SigmaSLSt{'''}$.
    
    By determinism of $\specarrowSLS{}$ we get $\SigmaSLSt{'''} = \SigmaSLSt{''}$.

    However, this means that $\trgSLS{\tau^{'\taint'}} = \trgSLS{\tau^{\taint}}$ which is a contradiction.
\end{proof}

\begin{lemma}[$\SLSv$: Low-equivalent programs have low-equivalent initial states]\label{lemma:sls-low-equivalent-programs-low-equivalent-states}
\begin{align*}
\forall \trgSLS{P = M, \OB{F}, \OB{I}}, \trgSLS{P' = M', \OB{F'}, \OB{I'}}. &
	\text{ if } \trgSLS{P} \loweq \trgSLS{P'} 
	\text{ then } \initFuncSLS{(\trgSLS{P})} \approx \initFuncSLS{(\trgSLS{P'})}
\end{align*}
\end{lemma}

\begin{proof}
    Analogous to \Thmref{lemma:v4-low-equivalent-programs-low-equivalent-states}.
\end{proof}

\begin{lemma}[$\SLSv$: Steps of $\bigspecarrowSLS{}$ with safe observations preserve Safe-equivalence]\label{lemma:sls-bigspecarrow-preserves-safe}
If
\begin{enumerate}
    \item $\trgSLS{\SigmaSLS} \trgSLS{\bigspecarrowSLS{\tauStack}} \SigmaSLSt{\dagger}$ and 
    \item $\trgSLS{\SigmaSLS} \relsa \trgSLS{\SigmaSLSt{'}}$
    \item $\safe{\tauStack}$ \end{enumerate}
Then there exists $\SigmaSLSt{''}$ such that
\begin{enumerate}
    \item $\trgSLS{\SigmaSLSt{'}} \bigspecarrowSLS{\tauStack} \SigmaSLSt{''}$ and
    \item $\SigmaSLSt{\dagger} \relsa \SigmaSLSt{''}$
\end{enumerate}
\end{lemma}

\begin{proof}
    \begin{description}
    \item[\Cref{tr:sls-reflect}]

    Analogous to the corresponding case in \Thmref{lemma:v4-bigspecarrow-preserves-safe}
    \item[\Cref{tr:sls-silent}]

    Analogous to the corresponding case in \Thmref{lemma:v4-bigspecarrow-preserves-safe} using \Thmref{lemma:sls-specarrow-preserves-safe}
    \item[\Cref{tr:sls-single}]
    Analogous to the corresponding case in \Thmref{lemma:v4-bigspecarrow-preserves-safe} using \Thmref{lemma:sls-specarrow-preserves-safe}
    
    \end{description}
\end{proof}

\begin{lemma}[$\SLSv$: Steps of $\specarrowSLS{}$ with safe observations preserve Safe-equivalence]\label{lemma:sls-specarrow-preserves-safe}
If
\begin{enumerate}
    \item $\SigmaSLSt{\dagger} \relsa \SigmaSLSt{'}$
    \item $\SigmaSLSt{\dagger} \specarrowSLS{\tau} \SigmaSLSt{\dagger\dagger}$
    \item $\safe{\tau}$
\end{enumerate}
Then 
\begin{enumerate}
    \item $\SigmaSLSt{'} \specarrowSLS{\tau} \SigmaSLSt{''}$
    \item $\SigmaSLSt{\dagger\dagger} \relsa \SigmaSLSt{''}$
\end{enumerate}
\end{lemma}

\begin{proof}
    The proof proceeds by inversion on $\SigmaSLSt{\dagger} \specarrowSLS{\tau} \SigmaSLSt{\dagger\dagger}$.
    
    Most cases are similar to \Thmref{lemma:v4-specarrow-preserves-safe}.
    The interesting case arises when speculation is triggered when a $\pret$ instruction is encountered.

    \begin{description}

    \item[\Cref{tr:sls-barr}, \Cref{tr:sls-barr-spec}, \Cref{tr:sls-nospec-eps}, \Cref{tr:sls-nospec-act}, \Cref{tr:sls-rollback}, \Cref{tr:sls-spec-att}]
    These cases are analogous to the corresponding cases in \Thmref{lemma:v4-specarrow-preserves-safe}.
    
    \item[\Cref{tr:sls-spec}]
    Then we have $\trgSLS{(\taint, \Omega, n + 1) \specarrowSLS{\tau^{\taint \glb \taint_1}} \tup{\taint, \Omega^{\dagger}, n} \cdot \tup{\unta, \Omega^{\dagger\dagger}, j}}$ with $\trgSLS{\Omega(\pc)} = \trgSLS{\pret}$ and $\trgSLS{\Omega \nsarrow{\tau^{\taint_1}} \Omega^{\dagger}}$.

    Furthermore, we have $ \trgSLS{\Omega^{\dagger\dagger} = \OB{F}; \OB{I}; \OB{B}; \sigma^{\dagger\dagger}}$ and $\trgSLS{\sigma^{\dagger\dagger} = \sigma[\pc \mapsto \Omegav(\pc) + 1]}$.
    
    We use  \Thmref{lemma:non-spec-preserves-safe} on $\src{\Omega \nsarrow{\tau^{\taint_1}} \Omega^{'}}$  and $\src{\Omega \relsa \Omega^{'}}$ and get $\src{\Omega' \nsarrow{\tau^{\taint_1}} \Omega''}$ with $\trgSLS{\Omega^{\dagger} \relsa \Omega''}$.

    Next, we define $\trgSLS{\Omega'''}$ as $\trgSLS{\Omega''' = \OB{F'}; \OB{I'}; \OB{B'}; \sigma'''}$ and $\trgSLS{\sigma''' = \sigma[\pc \mapsto \sigma'(\pc) + 1]}$ where $\trgSLS{\Omega'} =  \trgSLS{\OB{F'}; \OB{I'}; \OB{B'}; \sigma'}$.

     From $\trgSLS{\Omega \relsa \Omega'}$ we can derive that $\trgSLS{\Omega^{\dagger\dagger} \relsa \Omega'''}$.

     We can now apply \Cref{tr:sls-spec} to derive $\trgSLS{(\taint', \Omega', n' + 1) \specarrowSLS{\tau^{\taint', \glb \taint_1}} \tup{\taint' \Omega'', n'} \cdot \tup{\unta, \Omega''', j'}}$.

     From $\trgSLS{\SigmaSLS^{\dagger} \relsa \SigmaSLS'}$ we get $\trgSLS{j = j'}$, $\trgSLS{\taint = \taint'}$ and we can now conclude that $\trgSLS{\tau^{\taint \glb \taint_1}} = \trgSLS{\tau^{\taint' \glb \taint_1}}$.

     $\trgSLS{\SigmaSLS{\dagger\dagger} \relsa \SigmaSLS''}$ follows from $\trgSLS{\SigmaSLS^{\dagger} \relsa \SigmaSLS'}$, $\trgSLS{\Omega^{\dagger} \relsa \Omega''}$ and $\trgSLS{\Omega^{\dagger\dagger} \relsa \Omega'''}$.
     
    \end{description}
\end{proof}
 \section{$\Jvr$: SS overapproximates SNI}

\begin{center}
	\scalebox{0.85}{
    \begin{forest}
    for tree={
        align = left,
        font = \footnotesize,
        forked edge,
    }
        [\Thmref{thm:v2-ss-impl-sni}
            [\Thmref{lemma:v2-bigspecarrow-preserves-safe}
                [\Thmref{lemma:v2-specarrow-preserves-safe}
                    [\Thmref{lemma:non-spec-preserves-safe}]
                ]
            ]
            [\Thmref{lemma:v2-low-equivalent-programs-low-equivalent-states}]
        ]
    \end{forest}
    }
\end{center}

\begin{theorem}[$\Jv$: \sstext{} implies SNI]\label{thm:v2-ss-impl-sni}
If
\begin{enumerate}
    \item $\contract{}{\Jv} \vdash \trgJ{P} : \ss $
\end{enumerate}
Then
\begin{enumerate}[label=\Roman*]
    \item $\contract{}{\Jv}  \vdash \trgJ{P} : \sni $
\end{enumerate}

\end{theorem}
\begin{proof}

    Let $\trgJ{P}$ be an arbitrary program such that $\contractSpec{}{\Jv} \vdash \trgJ{P} : \ss$.

    The proof proceeds by contradiction. 
    
    Assume that $\contractSpec{}{\Jv} \nvdash \trgJ{P} : \sni$. That is, there is another program $\trgJ{P'}$ and traces $\trgJ{\tra{_1}} \in \trgJ{ \contractSpec{}{\Jv}(\SInit{P}) }$ and  $\trgJ{\tra{_2}} \in \trgJ{ \contractSpec{}{\Jv}(\SInit{P'}) }$ such that $\trgJ{P} \loweq \trgJ{P'}$, $\nspecProject{\trgJ{\tra{_1}}} = \nspecProject{\trgJ{\tra{_2}}}$, and $\trgJ{\tra{_1}} \neq \trgJ{\tra{_2}}$.

    Since $\trgJ{\tra{_1}} \in \trgJ{ \contractSpec{}{\Sv}(\SInit{P}})$ we have $\trgJ{\amTracevJ{P}{\tauStackT}}$ and by \Cref{tr:v2-trace} we know that there exists $\SigmaJDS{F}$ such that $\trgJ{\initFuncJ{\trgJ{(P)}} \bigspecarrowJ{\tra{_1}}\ \SigmaJDS{F}}$. Similarly, we get  $\trgJ{\initFuncJ{(\trgJ{P'})} \bigspecarrowJ{\tra{_2}}\ \SigmaJDSP{F}}$.
    
    Combined with the fact that $\trgJ{\tra{_1}} \neq \trgJ{\tra{_2}}$, it follows that there exists speculative states $\SigmaJt{\dagger}, \SigmaJDS{}, \SigmaJt{'}, \SigmaJt{''}$ and sequences of observations $\trgJ{\tauStack, \tauStack_{end}, \tauStack'_{end}}, \trgJ{\aca{^\taint}, {\tau'}^{\taint'} } $ such that $\trgJ{\aca{^{\taint}}} \neq \trgJ{ {\tau'}^{\taint'}}$  and:
    
    \begin{align*}
    \trgJ{\initFuncJ{(\trgJ{P})}} & \trgJ{\bigspecarrowJ{\tauStack}} \SigmaJt{\dagger} \specarrowJ{\aca{^{\taint}}} \SigmaJt{\dagger\dagger} \bigspecarrowJ{\tauStack_{end}} \trgJ{\SigmaJDS{F}} 
    \\
    \trgJ{\initFuncJ{(\trgJ{P'})}} & \bigspecarrowJ{\tauStack} \SigmaJt{'} \specarrowJ{{\tau'}^{\taint'}} \SigmaJt{''} \bigspecarrowJ{\tauStack'_{end}} \trgJ{\SigmaJDSP{F}}
    \end{align*}
    
    From $\contractSpec{}{\Jv} \vdash \trgJ{P} : \ss $ we get that $\safe{\trgJ{\tra{_1}}}$.

    From \Thmref{lemma:v2-low-equivalent-programs-low-equivalent-states} we get ${\initFuncJ{(\trgJ{P})}} \relsa {\initFuncJ{(\trgJ{P'})}}$.

    Using \Thmref{lemma:v2-bigspecarrow-preserves-safe} on $\initFuncJ{(\trgJ{P})} \trgJ{\bigspecarrowJ{\tauStack \cdot \tau^{\taint}}} \SigmaJt{\dagger\dagger}$ and the fact that $\initFuncJ{(\trgJ{P})} \relsa {\initFuncJ{(\trgJ{P'})}}$ we get $\trgJ{\initFuncJ{(\trgJ{P'})}} \bigspecarrowJ{\tauStack \cdot \tau^{\taint}} \SigmaJt{'''} $ and  $\SigmaJt{\dagger\dagger} \relsa \SigmaJt{'''}$.
    
    By determinism of $\specarrowJ{}$ we get $\SigmaJt{'''} = \SigmaJt{''}$.

    However, this means that $\trgJ{\tau^{'\taint'}} = \trgJ{\tau^{\taint}}$ which is a contradiction.
\end{proof}

\begin{lemma}[$\Jv$: Low-equivalent programs have low-equivalent initial states]\label{lemma:v2-low-equivalent-programs-low-equivalent-states}
\begin{align*}
\forall \trgJ{P = M, \OB{F}, \OB{I}}, \trgJ{P' = M', \OB{F'}, \OB{I'}}. &
	\text{ if } \trgJ{P} \loweq \trgJ{P'} 
	\text{ then } \initFuncJ{P} \approx \initFuncJ{P'}
\end{align*}
\end{lemma}

\begin{proof}
    Analogous to \Thmref{lemma:v4-low-equivalent-programs-low-equivalent-states}.
\end{proof}

\begin{lemma}[$\Jv$: Steps of $\bigspecarrowJ{}$ with safe observations preserve Safe-equivalence]\label{lemma:v2-bigspecarrow-preserves-safe}
If
\begin{enumerate}
    \item $\trgJ{\SigmaJ} \trgJ{\bigspecarrowJ{\tauStack}} \SigmaJt{\dagger}$ and 
    \item $\trgJ{\SigmaJ} \relsa \trgJ{\SigmaJt{'}}$
    \item $\safe{\tauStack}$ \end{enumerate}
Then there exists $\SigmaJt{''}$ such that
\begin{enumerate}
    \item $\trgJ{\SigmaJt{'}} \bigspecarrowJ{\tauStack} \SigmaJt{''}$ and
    \item $\SigmaJt{\dagger} \relsa \SigmaJt{''}$
\end{enumerate}
\end{lemma}

\begin{proof}
The proof proceeds by induction on $\trgJ{\SigmaJ} \trgJ{\bigspecarrowJ{\tauStack}} \SigmaJt{\dagger}$:

\begin{description}
    \item[\Cref{tr:v2-reflect}]
    
    Then we have $\SigmaJ \bigspecarrowJ{\tauStack} \SigmaJt{\dagger}$ where $\tauStack = \empTr$ and $\SigmaJt{\dagger} = \SigmaJ$.
    
    Thus, we use \Cref{tr:v2-reflect} on $\SigmaJt{'}$ and get $\SigmaJt{'} \bigspecarrowJ{\tauStack'} \SigmaJt{''}$, where where $\tauStack' = \empTr$ and $\SigmaJt{''} = \SigmaJt{'}$.
    
    Thus, $\tauStack = \tauStack'$ trivially holds and $\SigmaJt{\dagger} \relsa \SigmaJt{''}$ holds by $\trgJ{\SigmaJ} \relsa \trgJ{\SigmaJt{'}}$.
    
    \item[\Cref{tr:v4-silent}]
    
    Then we have $\trgJ{\SigmaJ \bigspecarrowJ{\tra{^\taint}} \SigmaJt{\dagger\dagger}}$ and $\trgJ{\SigmaJt{\dagger\dagger} \specarrowJ{\epsilon} \SigmaJt{\dagger}}$.
    
    Applying the IH on $\SigmaJ \bigspecarrowJ{\tra{^\taint}} \SigmaJ''$ we get:
    \begin{align*}
    \SigmaJt{'} \bigspecarrowJ{\tauStack''} \SigmaJt{'''} \\
    \SigmaJt{\dagger\dagger} \relsa \SigmaJt{'''} \\
    \safe{\tauStack''}
    \end{align*}
    
    The case is analogous to \Thmref{lemma:v4-bigspecarrow-preserves-safe} using \Thmref{lemma:v2-specarrow-preserves-safe}.

    \item[\Cref{tr:v2-single}]
    
    Then we have $\trgJ{\SigmaJ \bigspecarrowJ{\tra{^\taint}} \SigmaJt{\dagger\dagger}}$ and $\trgJ{\SigmaJt{\dagger\dagger} \specarrowJ{\tau^{\taint}} \SigmaJt{\dagger}}$.
    
    Applying the IH on $\SigmaJ \bigspecarrowJ{\tra{^\taint}} \SigmaJ''$ we get:
    \begin{align*}
        \SigmaJt{'} \bigspecarrowJ{\tauStack''} \SigmaJt{'''} \\
        \SigmaJt{\dagger\dagger} \relsa \SigmaJt{'''} \\
        \safe{\tauStack''}
    \end{align*}
    
    The case is analogous to \Thmref{lemma:v4-bigspecarrow-preserves-safe} using \Thmref{lemma:v2-specarrow-preserves-safe}.
\end{description}
\end{proof}

\begin{lemma}[$\Jv$: Steps of $\specarrowJ{}$ with safe observations preserve Safe-equivalence]\label{lemma:v2-specarrow-preserves-safe}
If
\begin{enumerate}
    \item $\SigmaJt{\dagger} \relsa \SigmaJt{'}$
    \item $\SigmaJt{\dagger} \specarrowJ{\tau} \SigmaJt{\dagger\dagger}$
    \item $\safe{\tau}$
\end{enumerate}
Then 
\begin{enumerate}
    \item $\SigmaJt{'} \specarrowJ{\tau} \SigmaJt{''}$
    \item $\SigmaJt{\dagger\dagger} \relsa \SigmaJt{''}$
\end{enumerate}
\end{lemma}

\begin{proof}

We have 
\begin{align*}
    \SigmaJt{\dagger} =& \phiStackSv \cdot (\taint, \Omega, n) \\
    \SigmaJt{'} =& \phiStackSv' \cdot (\taint', \Omega', n) \\ 
    \Omega \relsa& \Omega' \\
    \taint =& \taint' \\
    \Omega(\pc) =& \Omega'(\pc)
\end{align*}

The proof proceeds by inversion on $\SigmaJt{\dagger} \specarrowJ{\tau} \SigmaJt{\dagger\dagger}$:
(Implicitly unpack the context rule since it is always used)

\begin{description}

    \item[\Cref{tr:v2-barr}, \Cref{tr:v2-barr-spec}, \Cref{tr:v2-nospec-eps}, \Cref{tr:v2-nospec-act}, \Cref{tr:v2-spec-att}, \Cref{tr:v2-rollback}]
    These cases are analogous to the corresponding cases in \Thmref{lemma:v4-specarrow-preserves-safe}.

    \item[\Cref{tr:v2-spec}]

    Then we have $\trgJ{(\taint, \Omega, n + 1) \specarrowJ{\tau^{\taint \glb \taint_1}} \tup{\taint, \Omega^{\dagger}, n} \cdot \OB{\SigmaJ''}}$ with $\trgJ{\Omega(\pc)} = \trgJ{\pjmp{x}}$ and $\trgJ{\Omega \nsarrow{\tau^{\taint_1}} \Omega^{\dagger}}$.

    Here $\trgJ{\OB{\SigmaJ''}} = \bigcup_{l \in p} \trgJ{\tup{\Omegav, j}} \text{ where } \trgJ{\Omegav = \OB{F}; \OB{I}; \OB{B}; \sigmav[\pc \mapsto l]}$.

    We use  \Thmref{lemma:non-spec-preserves-safe} on $\src{\Omega \nsarrow{\tau^{\taint_1}} \Omega^{\dagger}}$  and $\src{\Omega \relsa \Omega^{'}}$ and get $\src{\Omega' \nsarrow{\tau^{\taint_1}} \Omega''}$ with $\trgJ{\Omega^{\dagger} \relsa \Omega''}$.
    
    Next, choose $\trgJ{\OB{\SigmaJ''}} = \bigcup_{l \in p} \trgJ{\tup{\Omegav, j}} \text{ where } \trgJ{\Omegav' = \OB{F}; \OB{I}; \OB{B}; \sigmav'[\pc \mapsto l]}$.

     From $\trgJ{\Omega \relsa \Omega'}$ we can derive that $\trgJ{\Omega^{\dagger\dagger} \relsa \Omega'''}$ for each $\Omega^{\dagger\dagger} \in \trgJ{\OB{\SigmaJ''}}$ and $\Omega''' \in \trgJ{\OB{\SigmaJ'''}}$.

     We can now apply \Cref{tr:v2-spec} to derive 
     $\trgJ{(\taint', \Omega', n' + 1) \specarrowJ{\tau^{\taint' \glb \taint_1}} \tup{\taint' \Omega'', n'} \cdot \OB{\SigmaJ'''}}$.

     From $\trgJ{\SigmaJ \relsa \SigmaJ'}$ we get, $\trgJ{\taint = \taint'}$ and we can now conclude that $\trgJ{\tau^{\taint \glb \taint_1}} = \trgJ{\tau^{\taint' \glb \taint_1}}$.

     $\trgJ{\SigmaJ{\dagger\dagger} \relsa \SigmaJ''}$ follows from $\trgJ{\SigmaJ^{\dagger} \relsa \SigmaJ'}$, $\trgJ{\Omega^{\dagger} \relsa \Omega''}$ and $\trgJ{\Omega^{\dagger\dagger} \relsa \Omega'''}$.

\end{description}
\end{proof}
 \section{$\Rvr$: SS overapproximates SNI}

We extend $\approx$ to cover the Rsb $\Rsb$ as well.
\begin{align*}
    \trgR{\Rsb \cdot (n : \taint)} \approx \trgR{\Rsb' \cdot n' : \taint} \isdef&\
    \trgR{n : \taint} \approx \trgR{n' : \taint'} \land \trg{\Rsb} \approx \trgR{\Rsb'}
\end{align*}

\begin{center}
	\scalebox{0.85}{

    \begin{forest}
    for tree={
        align = left,
        font = \footnotesize,
        forked edge,
    }
        [\Thmref{thm:v5-ss-impl-sni}
            [\Thmref{lemma:v5-bigspecarrow-preserves-safe}
                [\Thmref{lemma:v5-specarrow-preserves-safe}
                    [\Thmref{lemma:non-spec-preserves-safe}]
                ]
            ]
            [\Thmref{lemma:v5-low-equivalent-programs-low-equivalent-states}]
        ]
    \end{forest}
}
\end{center}

\begin{theorem}[$\Rv$: \sstext{} implies SNI]\label{thm:v5-ss-impl-sni}
If
\begin{enumerate}
    \item $\contractSpec{}{\Rv} \vdash \trgR{P} : \ss $
\end{enumerate}
Then
\begin{enumerate}[label=\Roman*]
    \item $\contractSpec{}{\Rv} \vdash \trgR{P} : \sni $
\end{enumerate}

\end{theorem}
\begin{proof}
    Let $\trgR{P}$ be an arbitrary program such that $\contractSpec{}{\Rv} \vdash \trgR{P} : \ss$.

    The proof proceeds by contradiction. 
    
    Assume that $\contractSpec{}{\Rv} \nvdash \trgR{P} : \sni$. That is, there is another program $\trgR{P'}$ and traces $\trgR{\tra{_1}} \in \trgR{ \contractSpec{}{\Rv}(\SInit{P}) }$ and  $\trgR{\tra{_2}} \in \trgR{ \contractSpec{}{\Rv}(\SInit{P'}) }$ such that $\trgR{P} \loweq \trgR{P'}$, $\nspecProject{\trgR{\tra{_1}}} = \nspecProject{\trgR{\tra{_2}}}$, and $\trgR{\tra{_1}} \neq \trgR{\tra{_2}}$.

    Since $\trgR{\tra{_1}} \in \trgR{ \contractSpec{}{\Rv}(\SInit{P}})$ we have $\trgR{\amTracevR{P}{\tauStackT}}$ and by \Cref{tr:v5-trace} we know that there exists $\SigmaRDS{F}$ such that $\trgR{\initFuncR{\trgR{(P)}} \bigspecarrowR{\tra{_1}}\ \SigmaRDS{F}}$. Similarly, we get  $\trgR{\initFuncR{(\trgR{P'})} \bigspecarrowR{\tra{_2}}\ \SigmaRDSP{F}}$.
    
    Combined with the fact that $\trgR{\tra{_1}} \neq \trgR{\tra{_2}}$, it follows that there exists speculative states $\SigmaRt{\dagger}, \SigmaRDS{}, \SigmaRt{'}, \SigmaRt{''}$ and sequences of observations $\trgR{\tauStack, \tauStack_{end}, \tauStack'_{end}}, \trgR{\aca{^\taint}, {\tau'}^{\taint'} } $ such that $\trgR{\aca{^{\taint}}} \neq \trgR{ {\tau'}^{\taint'}}$  and:
    
    \begin{align*}
    \trgR{\initFuncR{(\trgR{P})}} & \trgR{\bigspecarrowR{\tauStack}} \SigmaRt{\dagger} \specarrowR{\aca{^{\taint}}} \SigmaRt{\dagger\dagger} \bigspecarrowR{\tauStack_{end}} \trgR{\SigmaRDS{F}} 
    \\
    \trgR{\initFuncR{(\trgR{P'})}} & \bigspecarrowR{\tauStack} \SigmaRt{'} \specarrowR{{\tau'}^{\taint'}} \SigmaRt{''} \bigspecarrowR{\tauStack'_{end}} \trgR{\SigmaRDSP{F}}
    \end{align*}
    
    From $\contractSpec{}{\Rv} \vdash \trgR{P} : \ss $ we get that $\safe{\trgR{\tra{_1}}}$.

    From \Thmref{lemma:v5-low-equivalent-programs-low-equivalent-states} we get ${\initFuncR{(\trgR{P})}} \relsa {\initFuncR{(\trgR{P'})}}$.

    Using \Thmref{lemma:v5-bigspecarrow-preserves-safe} on $\initFuncR{(\trgR{P})} \trgR{\bigspecarrowR{\tauStack \cdot \tau^{\taint}}} \SigmaRt{\dagger\dagger}$ and the fact that $\initFuncR{(\trgR{P})} \relsa {\initFuncR{(\trgR{P'})}}$ we get $\trgR{\initFuncR{(\trgR{P'})}} \bigspecarrowR{\tauStack \cdot \tau^{\taint}} \SigmaRt{'''} $ and  $\SigmaRt{\dagger\dagger} \relsa \SigmaRt{'''}$.
    
    By determinism of $\specarrowR{}$ we get $\SigmaRt{'''} = \SigmaRt{''}$.

    However, this means that $\trgR{\tau^{'\taint'}} = \trgR{\tau^{\taint}}$ which is a contradiction.
\end{proof}

\begin{lemma}[$\Rv$: Low-equivalent programs have low-equivalent initial states]\label{lemma:v5-low-equivalent-programs-low-equivalent-states}
\begin{align*}
\forall \trgR{P = M, \OB{F}, \OB{I}}, \trgR{P' = M', \OB{F'}, \OB{I'}}. &
	\text{ if } \trgR{P} \loweq \trgR{P'} 
	\text{ then } \initFuncR{\trgR{(P)}} \approx \initFuncR{\trgR{(P')}}
\end{align*}
\end{lemma}
\begin{proof}
    Analogous to \Thmref{lemma:v4-low-equivalent-programs-low-equivalent-states} since the RSB is initialized to be empty.
\end{proof}

\begin{lemma}[$\Rv$: Steps of $\bigspecarrowR{}$ with safe observations preserve Safe-equivalence]\label{lemma:v5-bigspecarrow-preserves-safe}
If
\begin{enumerate}
    \item $\trgR{\SigmaR} \trgR{\bigspecarrowR{\tauStack}} \SigmaRt{\dagger}$ and 
    \item $\trgR{\SigmaR} \relsa \trgR{\SigmaRt{'}}$
    \item $\safe{\tauStack}$ \end{enumerate}
Then there exists $\SigmaRt{''}$ such that
\begin{enumerate}
    \item $\trgR{\SigmaRt{'}} \bigspecarrowR{\tauStack} \SigmaRt{''}$ and
    \item $\SigmaRt{\dagger} \relsa \SigmaRt{''}$
\end{enumerate}
\end{lemma}

\begin{proof}
    \begin{description}
    \item[\Cref{tr:v5-reflect}]

    Analogous to the corresponding case in \Thmref{lemma:v4-bigspecarrow-preserves-safe}
    \item[\Cref{tr:v5-silent}]

    Analogous to the corresponding case in \Thmref{lemma:v4-bigspecarrow-preserves-safe} using \Thmref{lemma:v5-specarrow-preserves-safe}
    \item[\Cref{tr:v5-single}]
    Analogous to the corresponding case in \Thmref{lemma:v4-bigspecarrow-preserves-safe} using \Thmref{lemma:v5-specarrow-preserves-safe}
    
    \end{description}
\end{proof}

\begin{lemma}[$\Rv$: Steps of $\specarrowR{}$ with safe observations preserve Safe-equivalence]\label{lemma:v5-specarrow-preserves-safe}
If
\begin{enumerate}
    \item $\SigmaRt{\dagger} \relsa \SigmaRt{'}$
    \item $\SigmaRt{\dagger} \specarrowR{\tau} \SigmaRt{\dagger\dagger}$
    \item $\safe{\tau}$
\end{enumerate}
Then 
\begin{enumerate}
    \item $\SigmaRt{'} \specarrowR{\tau} \SigmaRt{''}$
    \item $\SigmaRt{\dagger\dagger} \relsa \SigmaRt{''}$
\end{enumerate}
\end{lemma}

\begin{proof}

    The proof proceeds by inversion on $\SigmaRt{\dagger} \specarrowR{\tau} \SigmaRt{\dagger\dagger}$.
    Most cases are similar to \Thmref{lemma:v4-specarrow-preserves-safe}.
    The interesting case arises when speculation is triggered when a $\pret$ instruction is encountered.

    \begin{description}

    \item[\Cref{tr:v5-barr}, \Cref{tr:v5-barr-spec}, \Cref{tr:v5-nospec-eps}, \Cref{tr:v5-nospec-act}, \Cref{tr:v5-spec-att}, \Cref{tr:v5-rollback}]
    These cases are analogous to the corresponding cases in \Thmref{lemma:v4-specarrow-preserves-safe}

    \item[\Cref{tr:v5-call}]
    Then we have $\trgR{(\taint, \Omega, \Rsb, n + 1) \specarrowR{\tau^{\taint \glb \taint_1}} (\taint, \Omega^{\dagger}, \Rsb^{\dagger}, n)}$ with $\Omega(\pc) = \pcall{f}$, $\trgR{\Rsb^{\dagger}} = \trgR{\Rsb \cdot (\Omega(\pc) + 1)}$, $\trgR{f \notin \OB{I}}$ and $\trgR{\Omega \nsarrow{\tau^{\taint_1}} \Omega^{\dagger}}$.

    By using \Thmref{lemma:non-spec-preserves-safe} on $\trgR{\Omega \nsarrow{\tau^{\taint_1}} \Omega^{\dagger}}$ and $\trgR{\Omega \relsa \Omega^{'}}$ and get $\trgR{\Omega' \nsarrow{\tau^{\taint_1}} \Omega''}$ with $\trgR{\Omega^{\dagger}} \relsa \trgR{\Omega''}$.

    Since $\trgR{\Omega(\pc)} = \pcall{f}$, $\trgR{\Omega(\pc)} = \trgR{\Omega'(\pc)}$ thus $\trgR{\Rsb^{''}} = \trgR{\Rsb' \cdot (\Omega'(\pc) + 1)}$ and $\trgR{\Omega' \nsarrow{\tau^{\taint_1}} \Omega''}$,
    we can use \Cref{tr:v5-call} to derive $\trgR{(\taint', \Omega', \Rsb', n' + 1) \specarrowR{\tau^{\taint' \glb \taint_1}} (\taint', \Omega'', \Rsb'', n')}$

    Since $\trgR{\Omega^{\dagger}} \relsa \trgR{\Omega''}$ and $\trgR{\Rsb^{\dagger}} = \trgR{\Rsb^{''}}$ we get $(\trgR{\taint, \Omega^{\dagger}, \Rsb^{\dagger}, n}) \relsa (\trgR{\taint', \Omega'', \Rsb^{''}, n'})$.
    
    Combined with $\phiStackRv \relsa \phiStackRv'$ we get $\trgR{\SigmaRt{\dagger\dagger} \relsa \SigmaRt{''}}$.

    Furthermore, $\trgR{\tau^{\taint' \glb \taint_1}} = \trgR{\tau^{\taint \glb \taint_1}}$ follows from $\trgR{\taint} = \trgR{\taint'}$
    
    \item[\Cref{tr:v5-call-att}]
    Then we have $\trgR{(\taint, \Omega, \Rsb, n + 1) \specarrowR{\tau^{\taint \glb \taint_1}} (\taint, \Omega^{\dagger}, \Rsb, n)}$ with $\Omega(\pc) = \pcall{f}$, $\trgR{f \in \OB{I}}$ and $\trgR{\Omega \nsarrow{\tau^{\taint_1}} \Omega^{\dagger}}$.
    
    The case is analogous to \Cref{tr:v5-call} above except that the RSB does not change.

    \item[\Cref{tr:v5-retS}]
    Then we have $\trgR{(\taint, \Omega, \Rsb, n + 1) \specarrowR{\tau^{\taint \glb \taint_1}} (\taint, \Omega^{\dagger}, \Rsb^{\dagger}, n)}$ with $\Omega(\pc) = \pret$, $\trgR{\Rsb} = \trgR{\Rsb^{\dagger} \cdot l}$, $\trgR{f \notin \OB{I}}$, $\trgR{B} = \trgR{B; l'}$, $\trgR{l} = \trgR{l'}$ and $\trgR{\Omega \nsarrow{\tau^{\taint_1}} \Omega^{\dagger}}$.

    By using \Thmref{lemma:non-spec-preserves-safe} on $\trgR{\Omega \nsarrow{\tau^{\taint_1}} \Omega^{\dagger}}$ and $\trgR{\Omega \relsa \Omega^{'}}$ and get $\trgR{\Omega' \nsarrow{\tau^{\taint_1}} \Omega''}$ with $\trgR{\Omega^{\dagger}} \relsa \trgR{\Omega''}$.

    Since $\trgR{\Omega} \relsa \trgR{\Omega'}$ we get $\trgR{B} = \trgR{B'}$ and from $\SigmaRt{\dagger} \relsa \SigmaRt{'}$ we get $\trgR{\Rsb} = \trgR{\Rsb'}$ and $\trgR{n} = \trgR{n'}$.

    Thus, we fulfill all conditions for \Cref{tr:v5-retS} and derive \\
     $\trgR{(\taint', \Omega', \Rsb', n' + 1) \specarrowR{\tau^{\taint' \glb \taint_1}} (\taint', \Omega'', \Rsb'', n')}$.

    $\trgR{\Rsb^{\dagger}} = \trgR{\Rsb''}$ follows from $\trgR{\Rsb} = \trgR{\Rsb'}$.
    
    Since $\trgR{\Omega^{\dagger}} \relsa \trgR{\Omega''}$ and $\trgR{\Rsb^{\dagger}} = \trgR{\Rsb^{''}}$ we get $(\trgR{\taint, \Omega^{\dagger}, \Rsb^{\dagger}, n}) \relsa (\trgR{\taint', \Omega'', \Rsb^{''}, n'})$.
    
    Combined with $\phiStackRv \relsa \phiStackRv'$ we get $\trgR{\SigmaRt{\dagger\dagger} \relsa \SigmaRt{''}}$.

     Furthermore, $\trgR{\tau^{\taint' \glb \taint_1}} = \trgR{\tau^{\taint \glb \taint_1}}$ follows from $\trgR{\taint} = \trgR{\taint'}$

    \item[\Cref{tr:v5-retE}]
    Then we have $\trgR{(\taint, \Omega, \empty, n + 1) \specarrowR{\tau^{\taint \glb \taint_1}} (\taint, \Omega^{\dagger}, \empty, n)}$ with $\Omega(\pc) = \pret$ and $\trgR{\Omega \nsarrow{\tau^{\taint_1}} \Omega^{\dagger}}$.
    
    The case is analogous to \Cref{tr:v5-call} above except that the RSB does not change.
    
    \item[\Cref{tr:v5-spec}]
    Then we have $\trgR{(\taint, \Omega, \Rsb, n + 1) \specarrowR{\tau^{\taint \glb \taint_1}} \tup{\taint, \Omega^{\dagger}, \Rsb^{\dagger}, n} \cdot \tup{\unta, \Omega^{\dagger\dagger}, \Rsb^{\dagger}, j}}$ with $\trgR{\Omega(\pc)} = \trgR{\pret}$, $\trgR{\Rsb} = \trgR{\Rsb^{\dagger} \cdot l}$, $\trgR{f \notin \OB{I}}$, $\trgR{B} = \trgR{B; l'}$, $\trgR{l} \neq \trgR{l'}$ and $\trgR{\Omega \nsarrow{\tau^{\taint_1}} \Omega^{\dagger}}$.

    Furthermore, we have $ \trgR{\Omega^{\dagger\dagger} = \OB{F}; \OB{I}; \OB{B}; \sigma^{\dagger\dagger}}$ and $\trgR{\sigma^{\dagger\dagger} = \sigma[\pc \mapsto l]}$.
    
    We use  \Thmref{lemma:non-spec-preserves-safe} on $\src{\Omega \nsarrow{\tau^{\taint_1}} \Omega^{\dagger}}$  and $\src{\Omega \relsa \Omega^{'}}$ and get $\src{\Omega' \nsarrow{\tau^{\taint_1}} \Omega''}$ with $\trgR{\Omega^{\dagger} \relsa \Omega''}$.

    Next, we define $\trgR{\Omega'''}$ as $\trgR{\Omega''' = \OB{F'}; \OB{I'}; \OB{B'}; \sigma'''}$ and $\trgR{\sigma''' = \sigma[\pc \mapsto l]}$ where $\trgR{\Omega'} =  \trgR{\OB{F'}; \OB{I'}; \OB{B'}; \sigma'}$.

     From $\trgR{\Omega \relsa \Omega'}$ we can derive that $\trgR{\Omega^{\dagger\dagger} \relsa \Omega'''}$.

    From $\SigmaRt{\dagger} \relsa \SigmaRt{'}$ we get $\trgR{\Rsb} = \trgR{\Rsb'}$ and $\trgR{n} = \trgR{n'}$.

    We can now apply \Cref{tr:v5-spec} to derive $\trgR{(\taint', \Omega', \Rsb', n' + 1) \specarrowR{\tau^{\taint' \glb \taint_1}} \tup{\taint' \Omega'', \Rsb^{''}, n'} \cdot \tup{\unta, \Omega''', j'}}$.

     From $\trgR{\SigmaR \relsa \SigmaR'}$ we get $\trgR{j = j'}$, $\trgR{\taint = \taint'}$ and we can now conclude that $\trgR{\tau^{\taint \glb \taint_1}} = \trgR{\tau^{\taint' \glb \taint_1}}$.

     $\trgR{\SigmaR{\dagger\dagger} \relsa \SigmaR''}$ follows from $\trgR{\SigmaR^{\dagger} \relsa \SigmaR'}$, $\trgR{\Omega^{\dagger} \relsa \Omega''}$, $\trgR{\Rsb^{\dagger}} = \trgR{\Rsb''}$ and $\trgR{\Omega^{\dagger\dagger} \relsa \Omega'''}$.
     
    \end{description}
\end{proof}

\newpage
\section{$\Svr$: THE LFENCE COMPILER FOR \specs}\label{comp:v4-fence}

The countermeasure against \specs is to disable STL speculation on the hardware. 
A possible compiler would insert a fence after every $\storeC{}$ instruction.

\begin{align*}
    \complfenceS{ M ; \OB{F} ; \OB{I}} &= \trgS{ \complfenceS{M} ; \complfenceS{\OB{F}} ; \complfenceS{\OB{I}}}
	\\\\
	\complfenceS{ \srce} &= \trgS{\emptyset}
	\\
	\complfenceS{ \OB{I}\cdot f } &= \trgS{ \complfenceS{ \OB{I} }\cdot \complfenceS{f} }
	\\\\
	\complfenceS{ M ; -n\mapsto v : \unta} &= \trgS{\complfenceS{M} ; -\complfenceS{n} \mapsto \complfenceS{v} : \unta}
	\\\\
	\complfenceS{p_1;p_2} &= \trgS{\complfenceS{p_1} ; \complfenceS{p_2}}
	\\
	\complfenceS{l : i} &= \trgS{\complfenceS{l} : \complfenceS{i}} ~~ \text{if $i \neq \storeC$}
	\\\\
	\complfenceS{\pskip} &= \trgS{\pskip} 
	\\
	\complfenceS{\passign{x}{e}} &= \trgS{\passign{x}{\complfenceS{e}}}
        \\
        \complfenceS{\pcondassign{x}{e}{e'}} &= \trgS{\pcondassign{x}{\complfenceS{e}}{\complfenceS{e'}}}
        \\
        \complfenceS{\ploadprv{x}{e}} &= \trgS{\ploadprv{x}{\complfenceS{e}}}
	\\
	\complfenceS{\pload{x}{e}} &= \trgS{\pload{x}{\complfenceS{e}}}
	\\
        \complfenceS{\pstoreprv{x}{e}} &= \trgS{\pstoreprv{x}{\complfenceS{e}}}
        \\
	\complfenceS{\pjz{x}{l'}} &= \trgS{\pjz{x}{\complfenceS{l'}}} 
	\\
	\complfenceS{\pjmp{e}} &= \trgS{\pjmp{\complfenceS{e}}} 
	\\
	\complfenceS{\pcall{f}} &= \trgS{\pcall{f}} ~ \text{Here $f$ is a function name}
	\\
	\complfenceS{\pret} &= \trgS{\pret}
	\\\\
    \complfenceS{l : \pstore{x}{e}} &= \trgS{l} : \trgS{\pstore{x}{e}} ; l_{new} : \trgS{\pbarrier}  ~ \text{where $l \subseteq l_{new}$ and $l_{new} \leq next(l)$}
    \\\\
    \complfenceS{l} &= \trgS{l}
    \\
    \complfenceS{n} &= \trgS{n}
    \\
    \complfenceS{ \ominus e} &= \trgS{ \ominus \complfenceS{e}}
    \\
    \complfenceS{e_1 \otimes e_2} &= \trgS{\complfenceS{e_1} \otimes \complfenceS{e_2}}
\end{align*}

\subsection{Compilation of instructions}
Note that the compilation of instructions is a function from programs to programs. Since $l : i$ belongs to the syntactic category of programs. This means stuff like $\trgS{p(l)} = \complfenceS{i}$ does not make much sense because $p(l)$ is a partial function from labels to instructions.
Thus, when we state $\trgS{l : i} = \complfenceS{l} : \complfenceS{i}$ it should be $\trgS{l : i} = \complfenceS{l : i}$. However, for example for a store instruction, the compiler will insert an lfence. Thus, now we have a sequence.
$\complfenceS{l : i} = p_1; p_2$ which cannot be equal.
Thus, in this case we can state this: 
\begin{align*}
  \complfenceS{l : i} =& p_c \\
  \trgS{p_c} =& \trgS{l : i} ; \trgS{p_c'}\\
  \complfenceS{l : i} =&  \trgS{l : i} ; \trgS{p_c'}
\end{align*}
Recovering the relation to $\trgS{l : i}$. This means, that in the compiled program, we need to handle all steps introduced by the compilation in one big step. Otherwise, this relation will not hold, because we would be somewhere in $\trgS{p_c'}$.
We will see that sequence in the following proof to relate a source instruction to its compilation.
This forces the proof to handle all additional instructions added by the compiler.

\begin{theorem}[$\Sv$: The lfence compiler is \rdss]\label{thm:v4-lfence-comp-rdss}(\showproof{v4-lfence-comp-rdss})
	\begin{align*}
		\contract{}{\Sv} \vdash \complfenceS{\cdot} : \rdss
	\end{align*}
\end{theorem}

\begin{theorem}[$\Sv$: All lfence-compiled programs are \rss]\label{thm:v4-all-lfence-comp-are-rdss}(\showproof{v4-all-lfence-comp-are-rdss})
	\begin{align*}
		\forall\src{P}\ldotp \contract{}{\Sv} \vdash\complfenceS{P} : \rss
	\end{align*}
\end{theorem}

\begin{center}
	\scalebox{0.7}{
	\begin{forest}
		for tree={
			align = left,
			font = \footnotesize,
			forked edge,
			}
			[\Thmref{thm:v4-all-lfence-comp-are-rdss}
			[\Thmref{thm:v4-lfence-comp-rdss}
			[\Thmref{thm:v4-corr-bt-lfence}
			[\Thmref{thm:v4-ini-state-rel}
			[\Cref{thm:v4-heap-rel-comp-lfence}]
			[\Cref{thm:v4-heap-rel-bt-lfence}, name=heap-bt]
			]
			[\Thmref{thm:v4-bwd-sim-lfence}
			[\Thmref{thm:v4-bwd-sim-comp-steps-lfence}
			[\textbf{\Thmref{thm:v4-fwd-sim-stm-lfence}}
			[\Thmref{thm:v4-fwd-sim-exp-lfence}]
			]
			]
			[\Thmref{thm:v4-back-sim-bts-lfence}   
			[\Thmref{thm:v4-back-sim-bte-lfence}, name=exp-bt]
			]
			]
			]
			]
			]
			\node (sniSource) at (-4,3) {\Thmref{thm:ss-sni-source}};
			\node (rsscEq) [below=of sniSource] {\Thmref{thm:rdss-eq-rdsp}};
			\node [font=\footnotesize](test) at (-4,-7) {\Thmref{thm:v4-val-rel-bt-lfence}};
			\draw[->] (test.west) to [bend left=10] (heap-bt.south);
			\draw[->] (test.south) to [bend right] (exp-bt.south);
			\draw[red, thick, dotted] ($(sniSource.north west)+(-0.3,0.3)$)  rectangle ($(rsscEq.south east)+(0.9,-0.3)$);
		\end{forest}
	}
	\end{center}
		
The most interesting proof is \Thmref{thm:v4-fwd-sim-stm-lfence} since one has to reason about the speculation.

\subsection{Properties of the Context-based Backtranslation}\label{sec:ctx-bt-fence-props}

We derive the usual backward simulation result from forward simulation and determinism of the semantics.
Values are only numbers. so they are related if they are the same. 
Memories are related if they map related addresses (which are numbers) to related values.

\mytoprule{\text{Memory relation} \hreldef \text{Register Relation} \rrel \text{Value relation} \vreldef }
\begin{center}
	\typerule{Memory - base }{}{
		\srce \hrel \trgSe
	}{hrel-b}
	\typerule{Memory - ind }{
		\src{M} \hrel \trgS{M}
		\\
		\src{z} \vrel \trgS{z}
		&
		\src{v} \vrel \trgS{v}
		&
		\src{\sigma} \equiv \taintS
	}{
		\src{M ; z\mapsto v:\taint} \hrel \trgS{M; z\mapsto v:\taintS}
	}{hrel-i}
	\typerule{Memory - start }{
		\src{M} \hrel \trgS{M}
		&
		\src{M'} \hrel \trgS{M'}
		\\
		\src{v} \vrel \trgS{v}
		&
		\src{\sigma} \equiv \taintS
	}{
		\src{M ; 0\mapsto v:\sigma ; M'} \hrel \trgS{M ;0\mapsto v:\taintS ; M'}
	}{hrel-start}

    \typerule{ Register File }{
		\forall x, \src{A}(x) = \src{v} : \taint
		&
		\trgS{A}(x) = \trgS{v} : \taintS
		\\
		\src{v} \vrel \trgS{v}
		& 
		\taint \equiv \taintS
	}{
		\src{A} \rrel \trgS{A}
	}{rsrel}
 
	\typerule{Value - num }{
		\src{z} \equiv \trgS{z} &
		\com{z}\in\mb{Z}
	}{
		\src{z} \vrel \trgS{z}
	}{vr}

    \typerule{Value - label }{
		\src{l} \equiv \trgS{l} &
		\com{l}\in\mb{\labelset}
	}{
		\src{l} \vrel \trgS{l}
	}{vr}
\end{center}
\botrule

Register Files are related if the same register maps to related values.

\begin{insight}
    All compilers that we investigate fulfil the property that they do not delete source instructions and labels. If one would consider an optimizing compiler that e.g. deletes unnecessary instructions then the program relations does not hold.
    Then one would use the relation in Programs2
\end{insight}

\mytoprule{\text{Register relation} \rreldef \text{ Program relation} \creldef  \text{ State relation} \sreldef}
\begin{center}

    \typerule{Stack - base }{}{
		\srce \brel \trgSe
	}{v4-brel-b}
    \typerule{Stack - start }{
		\src{B} \brel \trgS{B}
		&
		\src{\OB{n}} \brel \trgS{\OB{n}}
	}{
		\src{B ; \OB{n}} \brel \trgS{B ; \OB{n}}
	}{v4-brel-start}

 \typerule{Stack Region - ind }{
		\src{\OB{n}} \brel \trgS{\OB{n}} 
		\\
		\src{z} \vrel \trgS{z}
		&
		\src{\sigma} \equiv \taintS
        & 
        \src{v} \vrel \trgS{v}
	}{
		\src{\OB{n} ; z \mapsto v:\sigma} \brel \trgS{\OB{n} ; z\mapsto v:\taintS}
	}{v4-brel-region}

\typerule{ Programs }{
        \forall \src{l : i} \in \src{p} \text{ and } \ffun{\src{l}} = \src{f} \text{ and } \ffun{\complfenceS{\src{l}}} = \trgS{f} 
        \\
		\text{ if } \src{f} \in \src{\OB{F}} \text{ and } \src{f} \in \src{\OB{f}}
		\text{ then } \trgS{p}(\complfenceS{l}) = \complfenceS{i}
		\\
  \forall \trgS{l : i} \in \trgS{p} \ldotp \ffun{\trgS{l}} = \trgS{f} \text{ and } \trgS{f} \notin \src{\OB{f}} \text{ and } \trgS{f} \in \trgS{\OB{F}}\text{ then }
    \src{p(\backtrfencec{\trgS{l}})} = \backtrfencec{\trgS{i}}
		\\
	}{
		\src{\OB{F}; \src{p}} \crel_{\src{\OB{f}}} \trgS{\OB{F}; \trgS{p}} 
	}{v4-comps}

 \typerule{ States }{
        \src{A} \rrel \trgS{A}
        &
		\src{\OB{F}; p} \crel_{\src{\OB{f''}}} \trgS{\OB{F}; p}
		&
		\src{M} \hrel \trgS{M}
		&
        \src{B} \brel \trgS{B}
        &
        \src{\OB{I}} \equiv \trgS{\OB{I}}
	}{
		\src{\OB{F};\OB{I},\OB{B}, \tup{p, M, A} } \srel_{\src{\OB{f''}}} \trgS{\OB{F};\OB{I},\OB{B}, \tup{p, M, A}, \bot, \safeta }
	}{v4-states}
 
\end{center}
\botrule

We define a helper function $\gprog{\src{p}}{\src{f}}$ returns the code corresponding to the function $\src{f}$.

Components are related according to a list of function names $\OB{f}$ which identify compiled code. We need this because we lose the information which functions are defined by the Program $P$ when we plug P together with the context A. Since only the functions in P will be compiled, we index the state relation by exactly those functions (See \Cref{thm:v4-ini-state-rel}).
All functions in this list have a compiled counterpart in the target component while all functions not in the list have a backtranslated counterpart in the source component.

\begin{theorem}[$\Sv$: Correctness of the Backtranslation for lfence]\label{thm:v4-corr-bt-lfence}(\showproof{v4-corr-bt-lfence})
	\begin{align*}
		\text{ if } 
			&\
			\trgS{ \amTracevS{\ctxS{}\hole{\complfenceS{P}}}{\tra{^{\taint}}} }
		\\
		\text{ then }
			&
			\src{ \Trace{\backtrfencec{\ctxS{}}\hole{P}}{ \tras{^{\sigma}}} }
		\\
		\text{ and }
			&\
			\tras{^{\sigma}} \rels \trgS{\tra{^{\taint}}}
	\end{align*}
\end{theorem}

\begin{theorem}[$\Sv$: Generalised Backward Simulation for lfence]\label{thm:v4-bwd-sim-lfence}(\showproof{v4-bwd-sim-lfence})
	\begin{align*}
		\text{ if }
			&\
			\ffun{\src{A(\pc)}}\in\src{\OB{f''}} \text{ then } \trgS{A(\pc)} : \trgS{i} = \complfenceS{l : i} \text{ or } (\complfenceS{A(\pc) : i} = \trgS{A(\pc)} : \trgS{i} ; \trgS{p_c'}) 
			\\
		\text{ else }
			&\
			\src{A(\pc)} : \src{i} = \backtrfencec{\trgS{A(\pc)}} : \backtrfencec{\trgS{i}} 
			\\
		\text{ and}
			&\
			\text{ if } \ffun{\src{A'(\pc)}}\in\src{\OB{f''}} \text{ then } \trgS{A'(\pc)} : \trgS{i'} = \complfenceS{A'(\pc) : i'} \text{ or } (\complfenceS{A'(\pc) : i'} = \trgS{A'(\pc)} : \trgS{i'} ; \trgS{p_c''})
			\\
		\text{ else }
			&\
			\src{A'(\pc)} : \src{i'} = \backtrfencec{\trgS{A'(\pc)}} : \backtrfencec{\trgS{i'}}
		\\
		\text{ and }
			&\
			\SigmaS = \trgS{(C, \OB{B}, \tup{p, M, A}, \bot,\safeta)}
		\\
		\text{ and }
			&\
			\SigmaSt{'}=\trgS{(C, \OB{B'}, \tup{p, M', A'},\bot,\safeta)}
		\\
		\text{ and }
			&\
			\trgS{ \SigmaS \bigspecarrowS{\tra{^{\taint}}} \SigmaS' }
		\\
		\text{ and }
			&\
			\src{\Omega} \srel_{\src{\OB{f''}}} \trgS{\SigmaS}
		\\
            \text{ and }
    			&\
    			\trgS{p(A(\pc))} = \trgS{i}  \text{ and } \trgS{p(A'(\pc))} = \trgS{i'}
    		\\
             \text{ and }
    			&\
    			\src{p(A(\pc))} =  \src{i} \text{ and } \src{p(A'(\pc))} = \src{i'}
            \\
		\text{ then }
			&\
			\src{\Omega}=\src{C, \OB{B}, \tup{p , M, A}  \nsbigarrow{\tras{^\sigma}} C, \OB{B'}, \tup{p , M', A'}} =\src{\Omega'}
		\\
		\text{ and }
			&\
			\src{\tras{^\sigma}} \tracerel \trgS{\tra{^{\taint}}}
		\\
		\text{ and }
			&\
			\src{\Omega'} \srelref_{\src{\OB{f''}}}\, \trgS{\SigmaS{'}}
	\end{align*}
\end{theorem}

In the general theorem we need backward simulation, which we derive from forward simulation in a standard way (This assumes determinism of the language).
Note that by ending up in a state with a compiled statement, we rule out cross-boundary calls and returns.
These cases pop up in the proofs using this theorem and we deal with them there. 
\begin{theorem}[$\Sv$: Backward Simulation for Compiled Steps in lfence]\label{thm:v4-bwd-sim-comp-steps-lfence}(\showproof{v4-bwd-sim-comp-steps-lfence})
	\begin{align*}
		\text{ if }
			&\
			\trgS{\SigmaS}=\trgS{(C, \OB{B}, \tup{p, M, A}, \bot, \safeta) }
		\\
		\text{ and }
			&\
			\trgS{\SigmaS'}=\trgS{(C, \OB{B'}, \tup{p, M', A'}, \bot, \safeta)}
		\\
		\text{ and }
			&\
			\trgS{\SigmaS \bigspecarrowS{\tra{^{\taint}}} \SigmaS'}
		\\
		\text{ and }
			&\
			\src{\Omega} \srelref_{\src{\OB{f''}}}\, \SigmaS
		\\
            \text{ and }
			&\
			\src{p(A(\pc))} =  \src{i} \text{ and } \src{p(A'(\pc))} =  \src{i'}
		\\
            \text{ and }
			&\
			\trgS{A(\pc) : i} = \complfenceS{A(\pc) : i}   \text{ or } (\complfenceS{A(\pc) : i} = \trgS{A(\pc)} : \trgS{i} ; \trgS{p_c'})
            \\
            \text{ and }
                &\
                \trgS{A'(\pc) : i'} = \complfenceS{A'(\pc) : i'}  \text{ or } (\complfenceS{A'(\pc) : i'} = \trgS{A'(\pc)} : \trgS{i'} ; \trgS{p_c''})
		\\
		\text{ then }
			&\
			\src{\Omega}=\src{C, \OB{B}, \tup{p, M, A} \nsarrow{\tau^\sigma} C, \OB{B'}, \tup{p, M', A'}}= \src{\Omega'}
        \\
		\text{ and }
			&\
			\src{\tau^\sigma} \tracerel \trgS{\tra{^{\taint}}} \qquad \text{( using the trace relation!)}
		\\
		\text{ and }
			&\
			\src{\Omega'} \srelref_{\src{\OB{f''}}}\, \SigmaS'
	\end{align*}
\end{theorem}

Starting with related components, heaps and stack frames, if a source statement takes a step emitting a label, then the compiled statement can take several steps and emit a trace that is related to the label.
Effectively, the target also only emits a single label, but we need to account for multiple steps (for the compilation of the store).

Note that by ending up in a state with a compiled statement, we rule out cross-boundary calls and returns. These cases pop up in the proofs using this theorem and we deal with them there.

\begin{lemma}[$\Sv$: Forward Simulation for Compiled Statements in lfence]\label{thm:v4-fwd-sim-stm-lfence}(\showproof{v4-fwd-sim-stm-lfence})
	\begin{align*}
		\text{ if }
			&\
			\src{\Omega}=\src{C, \OB{B}, \tup{p, M, A} \nsarrow{\tau^\sigma} C, \OB{B'}, \tup{p, M', A'}} =\src{\Omega'}
		\\
		\text{ and }
			&\
			\src{\Omega} \srelref_{\src{\OB{f''}}} \trgS{\SigmaS}
		\\
        \text{ and }
			&\
			\src{p(A(\pc))} = \src{i} \text{ and } \src{p(A'(\pc))} = \src{i'}
		\\
		\text{ and }
			&\
			\trgS{\SigmaS}= \trgS{C, \OB{B}, \tup{p, M , A}, \bot, \safeta}
		\\
		\text{ and }
			&\
			\trgS{\SigmaS'}= \trgS{C, \OB{B}, \tup{p, M' , A'}, \bot, \safeta}
		\\
            \text{ and }
                &\
                \trgS{A(\pc) : i} = \complfenceS{A(\pc) : i} \text{ or } (\complfenceS{A(\pc) : i} = \trgS{A(\pc)} : \trgS{i} ; \trgS{p_c'})
            \\
            \text{ and } 
                &\
                \trgS{A'(\pc) : i'} = \complfenceS{A'(\pc) : i'} \text{ or } (\complfenceS{A'(\pc) : i'} = \trgS{A'(\pc)} : \trgS{i'} ; \trgS{p_c''})
		\\
		\text{ then }
			&\
			\trgS{\SigmaS \bigspecarrowS{\tra{^{\taint}}} \SigmaS'}
		\\
		\text{ and }
			&\
			\src{\tau^\sigma} \tracerel \trgS{\tra{^{\taint}}} \qquad \text{( using the trace relation!)}
		\\
		\text{ and }
			&\
			\src{\Omega'} \srelref_{\src{\OB{f''}}}\ \trgS{\SigmaS'}
	\end{align*}
\end{lemma}

Starting with a related stack frame, if a source expression with a substitution star-reduces to a value, then the compiled expression with the compiled substitution star-reduces to the compiled value.
\begin{lemma}[Forward Simulation for Expressions in lfence]\label{thm:v4-fwd-sim-exp-lfence}(\showproof{v4-fwd-sim-exp-lfence})
	\begin{align*}
		\text{ if }
			&\ 
			\src{A\triangleright e \bigreds v : \sigma}
		\\
		\text{ and }
			&\
			\src{A}\rrelref\trgS{A}
		\\
		\text{ and }
			&\
			\src{\sigma}\equiv\taintS
		\\
		\text{ then }
			&\
			\trgS{A \triangleright \complfenceS{e}\bigreds \complfenceS{v} : \taintS}
	\end{align*}
\end{lemma}

\subsection{$\Sv$: Simulation and Relation for Backtranslated Code}\label{sec:bwd-sim-lfence}
We need backward simulation for backtranslated code.

As before, we need two lemmas on the backward simulation for backtranslated expressions and on the backward simulation for statements.

\begin{lemma}[Backward Simulation for Backtranslated Code]\label{thm:v4-back-sim-bts-lfence}(\showproof{v4-back-sim-bts-lfence})
	\begin{align*}
		\text{ if }
			&\
			\trgS{\SigmaS}=\trgS{(C, \OB{B}, \tup{p, M, A},\bot, \safeta) }
		\\
		\text{ and }
			&\
			\trgS{\SigmaS'}=\trgS{(C, \OB{B'}, \tup{p, M', A'}, \bot, \safeta)}
		\\
		\text{ and }
			&\
			\trgS{\SigmaS \specarrowS{\tau{^{\taint}}} \SigmaS'}
		\\
		\text{ and }
			&\
			\src{\Omega} \srel_{\src{\OB{f''}}} \trgS{\SigmaS}
		\\
        \text{ and }
			&\
			\trgS{p(A(\pc))} = \trgS{l} : \trgS{i} \text{ and } \trgS{p(A'(\pc))} = \trgS{l'} : \trgS{i'}
		\\
		\text{ then }
			&\
			\src{\Omega}=\src{C, \OB{B}, \tup{p, M, A} \nsarrow{\tau^\sigma} C, \OB{B'},\tup{p, M', A'} } = \src{\Omega'}
		\\
        \text{ and }
			&\
			\src{p(A(\pc))} = \backtrfencec{\trgS{l}} : \backtrfencec{\trgS{i}} \text{ and } \src{p(A'(\pc))} = \backtrfencec{\trgS{l'}} : \backtrfencec{\trgS{i'}}
		\\
		\text{ and }
			&\
			\src{\tau^\sigma} \arelref \trgS{\tau^\taint} \qquad \text{( using the action relation!)}
		\\
		\text{ and }
			&\
			\src{\Omega'} \srel_{\src{\OB{f''}}} \trgS{\SigmaS'}
	\end{align*}
\end{lemma}

\begin{lemma}[Backward Simulation for Backtranslated Expressions]\label{thm:v4-back-sim-bte-lfence}(\showproof{v4-back-sim-bte-lfence})
	\begin{align*}
		\text{ if }
			&\ 
			\trgS{A \triangleright e \bigreds v : \taintS }
		\\
		\text{ and }
			&\
			\src{A}\rrelref\trgS{A}
		\\
		\text{ and }
			&\
			\src{\taint}\equiv\taintS
		\\
		\text{ then }
			&\
			\src{A \triangleright \backtrfencec{\trgS{e}} \bigreds \backtrfencec{\trgS{v}} : \taint}
	\end{align*}
\end{lemma}

We need the fact that initial states are related for the general theorem (\Cref{thm:v4-corr-bt-lfence}). This means the initial states made of a compiled program and a backtranslated context are related. 
\begin{lemma}[Initial States are Related]\label{thm:v4-ini-state-rel}(\showproof{v4-ini-state-rel})
	\begin{align*}
		&
		\forall \src{P}, \forall \src{\OB{f}}=\dom{\src{P}.\src{F}}, \forall \ctxS{}
		\\
		&
		\initFunc{\backtrfencec{\ctxS{}}\src{\hole{P}}} \srelref_{\src{\OB{f}}}\, \initFuncS{\ctxS{\hole{\complfenceS{\src{P}}}}}
	\end{align*}
\end{lemma}

\begin{lemma}[A Value is Related to its Compilation for lfence]\label{thm:v4-val-rel-comp-lfence}(\showproof{v4-val-rel-comp-lfence})
	\begin{align*}
		{\src{v}}\vrel\complfenceS{v}
	\end{align*}
\end{lemma}

\begin{lemma}[A Heap is Related to its Compilation for lfence]\label{thm:v4-heap-rel-comp-lfence}(\showproof{v4-heap-rel-comp-lfence})
	\begin{align*}
		{\src{H}}\hrel\complfenceS{H}
	\end{align*}
\end{lemma}

\begin{lemma}[A Value is Related to its Backtranslation]\label{thm:v4-val-rel-bt-lfence}(\showproof{v4-val-rel-bt-lfence})
	\begin{align*}
		\backtrfencec{\trgS{v}}\vrel\trgS{v}
	\end{align*}
\end{lemma}

\begin{lemma}[A Memory is Related to its Backtranslation]\label{thm:v4-heap-rel-bt-lfence}(\showproof{v4-heap-rel-bt-lfence})
	\begin{align*}
		\backtrfencec{\trgS{M}}\hrel\trgS{M}
	\end{align*}
\end{lemma}

\begin{lemma}[A Taint is Related to its Backtranslation]\label{thm:v4-taint-rel-bt-lfence}(\showproof{v4-taint-rel-bt-lfence})
	\begin{align*}
		\backtrfencec{\trgS{\taint}}\equiv \trgS{\taint}
	\end{align*}
\end{lemma}

\begin{lemma}[The Value Relation is Functional]\label{thm:v4-val-uniqueR}(\showproof{v4-val-uniqueR})
If
\begin{enumerate}
    \item $\src{v} \vrel \trgS{v}$ and 
    \item $\src{v} \vrel \trgS{v'}$
\end{enumerate}
Then
\begin{enumerate}[label=\Roman*]
    \item $\trgS{v} = \trgS{v'}$
\end{enumerate}
    
\end{lemma}

\begin{proof}
Follows by inspection of $\vrelref$.
\end{proof}

\begin{lemma}[The Value Relation is Injective]\label{thm:v4-val-uniqueL}(\showproof{v4-val-uniqueL})
If
\begin{enumerate}
    \item $\src{v} \vrel \trgS{v}$ and 
    \item $\src{v'} \vrel \trgS{v}$
\end{enumerate}
Then
\begin{enumerate}[label=\Roman*]
    \item $\src{v} = \src{v'}$
\end{enumerate}
    
\end{lemma}

\begin{proof}
Follows by inspection of $\vrelref$.
\end{proof}

\begin{lemma}[Related Function maps after compilation]\label{thm:v4-fmaps-related}
For any program $\src{p}$ and its function map $\src{\mathcal{F}}$ we have $\mathcal{F}(f) = l \implies \complfenceS{\mathcal{F}}(f) = l$, where $\complfenceS{\mathcal{F}}$ is the function map for the compiled program $\complfenceS{p}$.
\end{lemma}

\begin{proof}
    Follows by inspection of the compiler.
\end{proof}

Since we are doing a context-based backtranslation, and since back-translated values are related to source values using $\vrelref$ as for compiled values, we do not need extra relations.

\begin{proof}[Proof of \Thmref{thm:v4-lfence-comp-rdss}]\proofref{}{v4-lfence-comp-rdss}\hfill

    Instantiate $\src{A}$ with $\backtrfencec{\trgS{A}}$.
    This holds by \Thmref{thm:v4-corr-bt-lfence}.
\end{proof}

\BREAK

\begin{proof}[Proof of \Thmref{thm:v4-all-lfence-comp-are-rdss}]\proofref{}{v4-all-lfence-comp-are-rdss}\hfill

    By \Thmref{thm:ss-sni-source} we know all source languages are SS: $\text{HPS:} \forall\src{P}\ldotp \contract{}{NS}\vdash P : \rss$.
    
    By \Thmref{thm:v4-lfence-comp-rdss} we have HPC: $\contractSpec{}{\Sv} \vdash \complfenceS{\cdot} : \rdss$.

    By \Thmref{thm:rdss-eq-rdsp} with HPC we have $\text{HPP:} \contractSpec{}{\Sv}\vdash \complfenceS{\cdot} : \rdssp$

    Unfolding \Thmref{def:rdsspc} of HPP we get 
    $\forall\src{P}\ldotp \text{ if } \contract{}{NS} \vdash\src{P} : \rss \text{ then } \contractSpec{}{\Sv}\vdash\comp{\src{P}} : \rss$. 
  
    Now, we instantiate it with HPS and can conclude that $\forall\src{P}\ldotp \contractSpec{}{\Sv}\vdash\complfenceS{P} : \rss$.
    
\end{proof}

\BREAK 

\begin{proof}[Proof of \Thmref{thm:v4-corr-bt-lfence}]\proofref{}{v4-corr-bt-lfence}\hfill

    This holds by \Thmref{thm:v4-ini-state-rel} and \Thmref{thm:v4-bwd-sim-lfence}.
\end{proof}

\BREAK

\begin{proof}[Proof of \Thmref{thm:v4-bwd-sim-lfence}]\proofref{}{v4-bwd-sim-lfence}\hfill

We proceed by induction on $\SigmaS \bigspecarrowS{\tra{^{\taint}}} \SigmaS'$:
\begin{description}
    \item[\Cref{tr:v4-reflect}]
    Trivial.
    \item[\Cref{tr:v4-single}]
    We have 
    \begin{align*}
        \trgS{\SigmaS \bigspecarrowS{\OB{\tau^{\taint}}} \SigmaS''} \\
        \trgS{\SigmaS'' \specarrowS{\tau^{\taint}} \SigmaS'} \\ 
        \trgS{\SigmaS''} = \trgS{\tup{\OB{C''}, \OB{B''}, \tup{p ,M'', A''}, \bot, \safeta}}
    \end{align*}

    Applying IH on $\trgS{\SigmaS \bigspecarrowS{\OB{\tau^{\taint''}}} \SigmaS''}$ we get: 
    \begin{enumerate}
        \item $\src{\Omega  \nsbigarrow{\tras{^{\sigma''}}} \Omega''}$ and
        \item $\src{\tras{^{\sigma''}}} \tracerel \trgS{\tra{^{\taint''}}}$ and
        \item $\src{\Omega''} \srelref_{\src{\OB{f''}}}\, \trgS{\SigmaS{''}}$
    \end{enumerate}
    Let $\ffun{\src{A''(\pc)}} = \src{f''}$ and let $\ffun{\trgS{l'}} = \trgS{f'}$ (where $\trgS{\SigmaS' . A(\pc)} = \trgS{l'}$).

    From  $\src{\Omega''} \srelref_{\src{\OB{f''}}}\, \trgS{\SigmaS{''}}$ we have $\src{\OB{F}} \equiv \trgS{\OB{F}}$. Thus, there is $\src{f'}$.
    
    We proceed by case analysis on $\src{f''} \in \src{\OB{f''}}$:
    \begin{description}
        \item[$\src{f''} \in \src{\OB{f''}}$ (\textbf{in the compiled component})]
            By case analysis on $\src{f'} \in \src{\OB{f''}}$:
            \begin{description}
                \item[$\src{f'} \in \src{\OB{f''}}$ (\textbf{in the compiled component})]
                This holds by \Thmref{thm:v4-bwd-sim-comp-steps-lfence}.
                
                \item[$\src{f'} \notin \src{\OB{f''}}$ (\textbf{in the context})]
                This switch from component to context arises in two cases:
                \begin{description}
                    \item[\textbf{call}]
                    This is a call from a compiled function to a context function.

                    In this case we know that $\trgS{\SigmaS'' \specarrowS{\tau^{\taint}} \SigmaS'}$ was derived by \Cref{tr:v4-nospec-act} combined with \trgS{\Cref{tr:call}} and we know $\trgS{\tau^{\taint}} = \trgS{\clh{f'}^{\safeta}}$.
                    
                    So $\trgS{\SigmaS'} = \trgS{\tup{\OB{C''}, \OB{B'}, \tup{p, M', A'}, \bot, \safeta}}$, where 
                    $\trgS{\OB{B'}} = \trgS{\OB{B''} \cdot \mathcal{F}(f) : \safeta ; \emptyset}$ and
                    $\trgS{M'} = \trgS{M''}$ and $\trgS{A'} = \trgS{A''[\pc \mapsto \mathcal{F}(f) : \safeta]}$.

                    By definition of $\src{\OB{F}; p} \crel_{\src{\OB{f''}}} \trgS{\OB{F}; p}$ and $\src{f'} \notin \src{\OB{f''}}$ we have $\gprog{\trgS{p}}{\trgS{f'}} = \trgS{p'}$ and
                    and $\gprog{\src{p}}{\src{f'}} = \backtrfencec{\trgS{p'}}$.

                    By \src{\Cref{tr:call}}, we get $\src{\Omega \nsarrow{\clh{f'}^{\safeta}} \Omega'}$ .
                    So $\src{\Omega'} = \src{\OB{C''}, \OB{B'}, \tup{p, M', A'}}$, where 
                    $\src{\OB{B'}} = \src{\OB{B'} \cdot \mathcal{F}(f) : \safeta ; \emptyset}$ and
                    $\src{M'} = \src{M''}$ and $\src{A'} = \src{A''[\pc \mapsto \mathcal{F}(f) : \safeta]}$.

                    We now need to prove that: 
                    \begin{description}
                        \item[$\src{\Omega'} \srelref_{\src{\OB{f''}}}\, \SigmaS'$]
                        By \Cref{tr:v4-states} it remains to show:
                        \begin{enumerate}
                            \item[$\src{A'} \rrel \trgS{A'}$]
                            
                            Follows from $\src{A''} \hrel \trgS{A''}$ from IH and $\src{\mathcal{F}(f) = l} \implies \trgS{\mathcal{F}(f) = l}$ by \Thmref{thm:v4-fmaps-related}.
                            \item[$\src{M'} \hrel \trgS{M'}$]

                            Follows from $\src{M'} = \src{M''}$ and $\trgS{M'} = \trgS{M''}$ and $\src{M''} \hrel \trgS{M''}$ from IH.
                            
                            \item[$ \src{B'} \brel \trgS{B'}$]

                            From \Thmref{thm:v4-fmaps-related} we get $\src{\mathcal{F}(f)} = \trgS{\mathcal{F}(f)}$ and are thus finished.
                            
                            \item[$\src{\OB{F}; p} \crel_{\src{\OB{f''}}} \trgS{\OB{F}; p}$]

                            Follows from $\src{\Omega''} \srelref_{\src{\OB{f''}}}\, \trgS{\SigmaS{''}}$ since programs do not change during execution.
                           
                            \item[$\src{\OB{I}} \equiv \trgS{\OB{I}}$] 
                            
                            Follows from $\src{\Omega''} \srelref_{\src{\OB{f''}}}\, \trgS{\SigmaS{''}}$ since components do not change.
                            
                        \end{enumerate}

                        \item[$\src{\clh{f'}^{\safeta}} \arelref \trgS{\clh{f'}^{\safeta}}$]
                        Holds trivially by \Cref{tr:ac-rel-cl}.
                        
                    \end{description}

                    \item[\textbf{return}]
                    This is a return from a compiled function to a context function.
                    This is the dual of the case below for return from context to component.
                \end{description}
            \end{description}
        
        \item[$\src{f''} \notin \src{\OB{f''}}$ (\textbf{in the context})]
        
            By case analysis on $\src{f'} \in \src{\OB{f''}}$:
            \begin{description}
                \item[$\src{f'} \in \src{\OB{f''}}$ (\textbf{in the compiled component})]
                
                This switch from context to component arises in two cases:
                
                \begin{description}
                    \item[\textbf{call}]
                    
                    This is a call from a context function to a compiled function.
                    This is the dual of the case for call above.

                    \item[\textbf{return}]
                    
                    This is a return from a context function to a compiled function.
                    In this case we know that $\trgS{\SigmaS'' \specarrowS{\tau^{\taint}} \SigmaS'}$ was derived by \Cref{tr:v4-nospec-act} combined with \trgS{\Cref{tr:ret}} and we know $\trgS{\tau^{\taint}} = \trgS{\rth^{\safeta}}$.
                    So $\trgS{\SigmaS'} = \trgS{(\OB{C''}, \OB{B'}, \tup{p, M', A'}, \bot, \safeta)}$, where 
                    $\trgS{\OB{B''}} = \trgS{\OB{B'} \cdot 0 \mapsto l : \safeta}$ and
                    $\trgS{M'} = \trgS{M''}$ and $\trgS{A'} = \trgS{A''[\pc \mapsto l]}$.

                    By \src{\Cref{tr:ret}}, we get $\src{\Omega \nsarrow{\retObs^{\safeta}} \Omega'}$ .
                    So $\src{\Omega'} = \src{\OB{C''}, \OB{B'}, \tup{p, M', A'}}$, where 
                    $\src{\OB{B''}} = \src{\OB{B'} \cdot 0 \mapsto l : \safeta}$ and
                    $\src{M'} = \src{M''}$ and $\src{A'} = \src{A''[\pc \mapsto l : \safeta]}$.
                    We now need to prove that: 
                    \begin{description}
                        \item[$\src{\Omega'} \srelref_{\src{\OB{f''}}}\, \SigmaS'$]
                        By \Cref{tr:v4-states} it remains to show:
                        \begin{enumerate}
                            \item[$\src{A'} \rrel \trgS{A'}$]

                            From $\src{\Omega''} \srelref_{\src{\OB{f''}}} \trgS{\SigmaS''}$ we know $\src{B''} \brel \trgS{B''}$ and thus
                            $\src{0 \mapsto l : \safeta} \vrel \trgS{0 \mapsto l : \safeta}$ and get $\src{l} \equiv \trgS{l}$.

                            Thus, together with IH $\src{A''} \rrel \trgS{A''}$, we have $\src{A'} \rrel \trgS{A'}$.
                            
                            \item[$\src{M'} \hrel \trgS{M'}$]
                            
                            Follows from $\src{M'} = \src{M''}$ and $\trgS{M'} = \trgS{M''}$ and $\src{M''} \hrel \trgS{M''}$ from IH.
                            
                            \item[$\src{B'} \brel \trgS{B'}$]
                            
                            Trivially follows from $\src{\Omega''} \srelref_{\src{\OB{f''}}} \trgS{\SigmaS''}$.
                            
                            \item[$\src{\OB{F}; p} \crel_{\src{\OB{f''}}} \trgS{\OB{F}; p}$] 
                            
                            Follows from $\src{\Omega''} \srelref_{\src{\OB{f''}}}\, \trgS{\SigmaS{''}}$ since programs do not change during execution.
                            
                            \item[$\src{\OB{F}} \equiv \trgS{\OB{F}}$ and $\src{\OB{I}} \equiv \trgS{\OB{I}}$] 
                            
                            Follows from $\src{\Omega''} \srelref_{\src{\OB{f''}}}\, \trgS{\SigmaS{''}}$ since components do not change.
                            
                        \end{enumerate}
                        \item[$\src{\rth^{\safeta}} \arelref \trgS{\rth^{\safeta}}$]
                        Holds trivially by \Cref{tr:ac-rel-rt}.
                        
                    \end{description}
                        
                \end{description}

                \item[$\src{f'} \notin \src{\OB{f''}}$ (\textbf{in the context})]
                This holds by \Thmref{thm:v4-back-sim-bts-lfence}.
                
            \end{description}
        
    \end{description}
    \item[\Cref{tr:v4-silent}]
    
    Then we have $\trgS{\SigmaS \bigspecarrowS{\tra{^\taint}} \SigmaS''}$ and $\trgS{\SigmaS'' \specarrowS{\epsilon} \SigmaS'}$.

     Applying IH on $\trgS{\SigmaS \bigspecarrowS{\OB{\tau^{\taint''}}} \SigmaS''}$ we get: 
    \begin{enumerate}
        \item $\src{\Omega  \nsbigarrow{\tras{^{\sigma''}}} \Omega''}$ and
        \item $\src{\tras{^{\sigma''}}} \tracerel \trgS{\tra{^{\taint''}}}$ and
        \item $\src{\Omega''} \srelref_{\src{\OB{f''}}}\, \trgS{\SigmaS{''}}$
    \end{enumerate}

     The step $\trgS{\SigmaS'' \specarrowS{\epsilon} \SigmaS'}$ cannot be a cross component call or return, because the matching observation is missing.
     This means execution continues in either the component or the context:
     \begin{description}
         \item[Execution continues in component] 
         This holds by \Thmref{thm:v4-bwd-sim-comp-steps-lfence}.
         
         \item[Execution continues in context] 
         This holds by \Thmref{thm:v4-back-sim-bts-lfence}.
     \end{description}
\end{description}

\end{proof}

\BREAK

\begin{proof}[Proof of \Thmref{thm:v4-bwd-sim-comp-steps-lfence}]\proofref{}{v4-bwd-sim-comp-steps-lfence}\hfill

    Proof by Contradiction
    \begin{description}

        \item[no trace relation]
        Then there exists $\src{\Omega''}$ such that $\src{\Omega \nsarrow{\tras{^{\sigma'}}} \Omega''}$ such that $\src{\tras{^{\sigma'}}} \tracerel \trgS{\tra{^{\taint}}}$ does not hold.
        By \Thmref{thm:v4-fwd-sim-stm-lfence} we have that the target ends up in a state $\trgS{\SigmaS''}$ such that $\src{\tras{^{\sigma'}}} \tracerel \trgS{\tra{^{\taint'}}}$ holds.
        Since the target semantics is deterministic, we know that $\SigmaS' = \SigmaS''$ and that $\trgS{\tra{^{\taint'}}} = \trgS{\tra{^{\taint}}}$.
        
        This contradicts our assumption that the traces are not in the trace relation.
        \item[not in state relation]

        Then there exists $\src{\Omega''}$ such that $\src{\Omega \nsarrow{\tras{^{\sigma'}}} \Omega''}$ such that $\src{\Omega''} \srelref_{\src{\OB{f''}}}\, \trgS{\SigmaS{'}}$ does not hold.
        
        By \Thmref{thm:v4-fwd-sim-stm-lfence} we have that the target ends up in a state $\trgS{\SigmaS''}$ such that $\src{\Omega''} \srelref_{\src{\OB{f''}}}\, \trgS{\SigmaS{''}}$ holds.

        Since the target semantics is deterministic, we know that $\SigmaS' = \SigmaS''$ and thus $\src{\Omega''} \srelref_{\src{\OB{f''}}}\, \trgS{\SigmaS{'}}$.

        We now have a contradiction.
    \end{description}
\end{proof}

\BREAK

\begin{proof}[Proof of \Thmref{thm:v4-fwd-sim-stm-lfence}]\proofref{}{v4-fwd-sim-stm-lfence}\hfill

    Since $\src{\Omega} \srel_{\src{\OB{f''}}}\ \trgS{\SigmaS}$ we know that $\src{A(\pc)} = \trgS{A(\pc)}$.
    We proceed by induction on $\nsarrow{}$:
    \begin{description}
        \item[\Cref{tr:skip}]
        The only change to the state is $\src{A'} = \src{A[pc \mapsto \succes(\pc)]}$ and we know $\tau^\sigma = \epsilon$.

        \begin{description}
            \item[$\trgS{\SigmaS \bigspecarrowS{\tra{^{\taint}}} \SigmaS'}$]
                Thus, we can use \Cref{tr:skip} to derive $\OmegaS \trgS{\nsarrow{\epsilon}} \OmegaS'$ with $\trgS{A'} = \trgS{A[pc \mapsto \succes(\pc)]}$.
                
                Next, we can use \Cref{tr:v4-t-nospec-eps} to derive a single step which we combine with \Cref{tr:v4-silent} to derive
                $\SigmaS \bigspecarrowS{\varepsilon} \SigmaS'$.
                
            \item[$\src{\tau^\sigma} \tracerel \trgS{\tra{^{\taint}}}$]
                $\src{\varepsilon} \tracerel \trgS{\varepsilon}$ follows trivially.
                
            \item[$\src{\Omega'} \srelref_{\src{\OB{f''}}}\, \trgS{\SigmaS'}$]
            Follows from $\src{A'} \rrel \trgS{A'}$ which follows from $\src{A} \rrel \trgS{A}$.
            
        \end{description}

        \item[\Cref{tr:assign}]
        The only change to the state is $\src{A'} = \src{A[pc \mapsto \succes(\pc), x \mapsto v : \taint]}$ where $\src{\exprEval{A}{e}{v : \taint}}$.
        Furthermore, we know $\src{\tau^\sigma} = \src{\epsilon}$. 

        \begin{description}
            \item[$\trgS{\SigmaS \bigspecarrowS{\tra{^{\taint}}} \SigmaS'}$]
                By \Thmref{thm:v4-fwd-sim-exp-lfence} we get $\trgS{\exprEval{A}{\complfenceS{e}}{\complfenceS{v} : \taintS}}$ and $\src{\taint} \equiv \taintS$.
                
                Thus, we can use \Cref{tr:assign} to derive $\OmegaS \trgS{\nsarrow{\epsilon}} \OmegaS'$ with $\trgS{A'} = \trgS{A[pc \mapsto \succes(\pc), x \mapsto \complfenceS{v} : \taintS]}$.
                
                Next, we can use \Cref{tr:v4-t-nospec-eps} to derive a single step which we combine with \Cref{tr:v4-silent} to derive
                $\SigmaS \bigspecarrowS{\varepsilon} \SigmaS'$.
                
            \item[$\src{\epsilon} \tracerel \trgS{\epsilon}$]
                $\src{\varepsilon} \tracerel \trgS{\varepsilon}$ follows trivially.
                
            \item[$\src{\Omega'} \srelref_{\src{\OB{f''}}}\, \trgS{\SigmaS'}$]
                Follows from $\src{A'} \rrel \trgS{A'}$ which follows from $\src{A} \rrel \trgS{A}$ and \Thmref{thm:v4-val-rel-comp-lfence} for 
                $\src{v} \vrel \complfenceS{v}$.
            
        \end{description}

        \item[\Cref{tr:load}]
        The only change to the state is $\src{A'} = \src{A[pc \mapsto \succes(\pc), x \mapsto M(n)]}$ where $\src{\exprEval{A}{e}{n : \taint}}$.
        Furthermore, we know $\src{\tau^\safeta} = \src{\loadObs^{\mta(n)}}$.
        \begin{description}
            \item[$\trgS{\SigmaS \bigspecarrowS{\tra{^{\taint}}} \SigmaS'}$]
            By \Thmref{thm:v4-fwd-sim-exp-lfence} we get $\trgS{\exprEval{A}{\complfenceS{e}}{\complfenceS{v} : \taintS}}$ and $\src{\taint} \equiv \taintS$.
            
            Thus, we can use \Cref{tr:load} to derive $\OmegaS \trgS{\nsarrow{\loadObs{n}^{\taint'}}} \OmegaS'$ with $\trgS{A'} = \trgS{A[pc \mapsto \succes(\pc), x \mapsto \complfenceS{v} : \taintS]}$.

            The rest follows trivially.
            \item[$\src{\loadObs{n}^{\safeta}} \tracerel \trgS{\loadObs{\complfenceS{n}}^{\taintpc \glb \taint'}}$]
            From $\srel$ we have $\trgS{\taintpc} = \trgS{\safeta}$ and thus $\trgS{\taintpc \glb \taint'} = \trgS{\safeta}$.
            
            Next, \Cref{tr:ac-rel-rd} with \Thmref{thm:v4-val-rel-comp-lfence} we follow that $\src{n} \vrel \complfenceS{n}$.

            \item[$\src{\Omega'} \srelref_{\src{\OB{f''}}}\, \trgS{\SigmaS'}$]
            Follows from $\src{A'} \rrel \trgS{A'}$ which follows from $\src{A} \rrel \trgS{A}$ and $\src{M} \hrel \trgS{M}$.
        \end{description}

        \item[\Cref{tr:load-prv}]
        This case is similar to the \textbf{load} case above and follows from \Thmref{thm:v4-fwd-sim-exp-lfence} and \Cref{tr:ac-rel-rd}.
        
        \item[\Cref{tr:store}]
        
        \begin{insight}
            This is the reason why we have $(\complfenceS{l : i} = \trgS{A(\pc)} : \trgS{i} ; \trgS{p_c'})$ in the statement of the lemma. Since the compilation introduces multiple instructions.
        \end{insight}
        This is the case where we need to account for multiple steps in the target, since the $\complfenceS{}$ adds an $\trgS{\pbarrier}$ instruction after the $\storeC{}$ instruction:
        \begin{align*}
             \complfenceS{l : \pstore{x}{e}} &= \trgS{
                                            \begin{aligned}[t]
                                                    &
                                                    l : \pstore{x}{e} 
                                                    ~\text{\Cref{tr:v4-skip}}
                                                        \\
                                                        &\
                                                       l' :  \pbarrier
                                                       ~\text{\Cref{tr:v4-barr-spec}}
                                            \end{aligned}
                                            }
        \end{align*}
        The only change to the state is $\src{A'} = \src{A[pc \mapsto \succes(\pc)}$ and $\src{M'} = \src{M[n \mapsto A(x)]}$ where $\src{\exprEval{A}{e}{n : \taint}}$.

        Furthermore, we know $\tau^\sigma = \storeObs^{\safeta}$.

        \begin{description}
            \item[$\trgS{\SigmaS \bigspecarrowS{\tra{^{\taint}}} \SigmaS'}$]
                By \Thmref{thm:v4-fwd-sim-exp-lfence} we get $\trgS{\exprEval{A}{\complfenceS{e}}{\complfenceS{n} : \taintS}}$ and $\src{\taint} \equiv \taintS$.
                
                Thus, we can use \Cref{tr:store} to derive $\OmegaS \trgS{\nsarrow{\storeObs{n}^{\taint}}} \OmegaS'$ with $\trgS{A'} = \trgS{A[pc \mapsto \succes(\pc)]}$ and $\trgS{M'} = \trgS{M[n \mapsto A(x)]}$.

                This is used in the step \Cref{tr:v4-skip}. The next instruction is the $\pbarrier$ and after \Cref{tr:v4-barr}, a rollback is triggered via \Cref{tr:v4-rollback}.

                Thus, we have the execution  $\SigmaS \bigspecarrowS{\storeObs{n}^{\taintpc \glb \taint} \cdot \rollbackObsS} \SigmaS'$, where $\SigmaS' = \trgS{\tup{\OB{C}, \OB{B}, \tup{p, M', A'}, \bot, \safeta}}$

            \item[$\src{\storeObs{n}^{\safeta}} \tracerel \trgS{\storeObs{\complfenceS{n}}^{\taintpc \glb \taint} \cdot \rollbackObsS }$]
                From $\srel$ we have $\trgS{\taintpc} = \trgS{\safeta}$ and thus $\trgS{\taintpc \glb \taint} = \trgS{\safeta}$.

                Next, from \Cref{tr:ac-rel-rlb} we get $\src{\epsilon} \arel \rollbackObsS$ and from \Cref{tr:ac-rel-wr} with \Thmref{thm:v4-val-rel-comp-lfence} for $\src{n} \vrel \complfenceS{n}$ we get $\src{\storeObs{n}^{\safeta}} \arel \trgS{\storeObs{\complfenceS{n}}^{\taintpc \glb \taint}}$.

            \item[$\src{\Omega'} \srelref_{\src{\OB{f''}}}\, \trgS{\SigmaS'}$]
                Follows from $\src{A'} \rrel \trgS{A'}$ which follows from $\src{A} \rrel \trgS{A}$ and $\src{M} \hrel \trgS{M}$ and \Thmref{thm:v4-val-rel-comp-lfence} for $\src{n} \vrel \complfenceS{n}$.
        \end{description}
        \item[\Cref{tr:store-prv}]
        This case is similar to the \textbf{load} case above. 
        
        \item[\Cref{tr:beqz-sat}]
        The only change to the state is $\src{A'} = \src{A[pc \mapsto l]}$ and $\src{A(x) = 0 : \taint}$.
        Furthermore, we know $\src{\tau^\sigma} = \src{\pcObs{l}^{\taint}}$.

        \begin{description}
            \item[$\trgS{\SigmaS \bigspecarrowS{\tra{^{\taint}}} \SigmaS'}$]
            From  $\src{A} \rrel \trgS{A}$ we know that $\trgS{A(x)} = \trgS{0} : \taintS'$ and $\src{\taint} \equiv \trgS{\taintS'}$.
            
            Thus, we can use \Cref{tr:beqz-sat} to derive $\OmegaS \trgS{\nsarrow{\pcObs{l}^{\taintS'}}} \OmegaS'$ with $\trgS{A'} = \trgS{A[pc \mapsto l}$.

            The rest follows trivially.
            \item[$\src{\pcObs{l}^{\taint}} \tracerel \trgS{\pcObs{l}^{\taintS'}}$]
            Follows from \Cref{tr:ac-rel-if} with \Thmref{thm:v4-val-rel-comp-lfence} for $\src{l} \vrel \complfenceS{l}$ and $\src{\taint} \equiv \trgS{\taintS'}$.

            \item[$\src{\Omega'} \srelref_{\src{\OB{f''}}}\, \trgS{\SigmaS'}$]
            Follows from $\src{A'} \rrel \trgS{A'}$ which follows from $\src{A} \rrel \trgS{A}$.
        \end{description}

        \item[\Cref{tr:beqz-unsat}]
        This case is analogous to case \textbf{beqz-sat} above.
        \item[\Cref{tr:call}]
        By \Thmref{thm:v4-fmaps-related} we have $\src{\mathcal{F}(f)} = \trgS{\mathcal{F}(f)}$.
        This means the register $\pc$ is updated in the same way. The rest follows from the relatedness of memories, registers and stacks.
        \item[\Cref{tr:call-internal}]
        This case is analogous to case \textbf{call} above.
        \item[\Cref{tr:callback}]
        This case cannot arise as we do not step to a compiled instructions $\complfenceS{i'}$ in the target.
        
        \item[\Cref{tr:ret}]
        Follows from the relatedness of stacks $\src{B} \brel \trgS{B}$. Thus, the $\pc$ register is updated to related values, which means we can show $\src{A'} \arel \trgS{A'}$.
        \item[\Cref{tr:ret-internal}]
        This case is analogous to case \textbf{ret} above.

        \item[\Cref{tr:retback}]
        This case cannot arise as we do not step to a compiled instructions $\complfenceS{i'}$ in the target.
    \end{description}
    
\end{proof}

\BREAK

\begin{proof}[Proof of \Thmref{thm:v4-fwd-sim-exp-lfence}]\proofref{}{v4-fwd-sim-exp-lfence}\hfill

    The proof proceeds by structural induction on $\src{A\triangleright e \bigreds v : \sigma}$:
\begin{description}
    \item[\Cref{tr:E-val}]
    Then we have $\src{e = v}$ and $\complfenceS{\src{v}} = \trgS{v}$ and $\sigma = \safeta$.
    
    Now, we fulfil all conditions for \cref{tr:E-val} for $\trgS{A \triangleright \complfenceS{v} \bigreds \complfenceS{v} : \safeta}$ and are finished.
    
    \item[\Cref{tr:E-lookup}]
    Then we have $\src{e = r}$ and $\complfenceS{r} = \trgS{r}$, where $r \in \Reg$ and $\complfenceS{v} = \trgS{v}$.
    Additionally, $\src{A(r) = v : \sigma}$.
    
    From $\src{A}\rrelref\trgS{A}$ we get $\trgS{A(r) = v' : \taint}$ with $\trgS{\taint} \equiv \taintS$ and $\src{v'} \vrel \trgS{v}$.
    
    By \Thmref{thm:v4-val-rel-comp-lfence} we have $\src{v} \vrel \complfenceS{v}$.
    
    By \Thmref{thm:v4-val-uniqueR} with $\src{v} \vrel \complfenceS{v}$ and $\src{v} \vrel \trgS{v'}$ we get $\complfenceS{v} = \trgS{v'}$.

    We can now use \Cref{tr:E-lookup} with $\trgS{A(r) = v' : \taint}$ for $\trgS{A \triangleright \trgS{r} \bigreds \trgS{v'} : \taint}$.

    \item[\Cref{tr:E-binop}]
    Then we have as IH: 
    \begin{align*}
        \src{A \triangleright \complfenceS{e_1} \bigreds \complfenceS{n_1} : \taint_1 } \\
        \src{A \triangleright \complfenceS{e_2} \bigreds \complfenceS{n_1} : \taint_2 }
    \end{align*}
    and $ \complfenceS{e_1 + e_2} = \complfenceS{e_1} + \complfenceS{e_2}$.
    
    Follows from IH and \Cref{tr:E-binop}.
    
    \item[\Cref{tr:E-unop}]
    Then by IH we have $\src{A \triangleright \complfenceS{e'} \bigreds \complfenceS{n} : \taint }$ and $ \complfenceS{\src{\ominus e'}} = \ominus \complfenceS{\src{e'}}$.
    
    Follows from IH and \Cref{tr:E-unop}.
\end{description}
\end{proof}

\BREAK

\begin{proof}[Proof of \Thmref{thm:v4-back-sim-bte-lfence}]\proofref{}{v4-back-sim-bte-lfence}\hfill

The proof proceeds by structural induction on $\trgS{A \triangleright e \bigreds v : \taintS }$:
\begin{description}
    \item[\Cref{tr:E-val}]
    Then we have $\trgS{e = v}$ and $\backtrfencec{\trgS{v}} = \src{v}$ and $\taintS = \safeta$.
    
    Now, we fulfill all conditions for \cref{tr:E-val} for $\src{A \triangleright \src{v} \bigreds v : \safeta}$ and are finished.
    
    \item[\Cref{tr:E-lookup}]
    Then we have $\trgS{e = r}$ and $\backtrfencec{\trgS{r}} = \src{r}$, where $r \in \Reg$ and $\backtrfencec{\trgS{v}} = \src{v}$.
    Additionally, $\trgS{A(r) = v : \taintS}$.
    
    From $\src{A}\rrelref\trgS{A}$ we get $\src{A(r) = v' : \taint}$ with $\src{\taint} \equiv \taintS$ and $\src{v'} \vrel \trgS{v}$.
    
    By \Thmref{thm:v4-val-rel-bt-lfence} we have $\backtrfencec{\trgS{v}} \vrel \trgS{v}$.
    
    By \Thmref{thm:v4-val-uniqueL} with $\backtrfencec{\trgS{v}} \vrel \trgS{v}$ and $\src{v'} \vrel \trgS{v}$ we get $\backtrfencec{\trgS{v}} = \src{v'}$.

    We can now use \Cref{tr:E-lookup} with $\src{A(r) = v' : \taint}$ for $\src{A \triangleright \src{r} \bigreds \src{v'} : \taint}$.

    \item[\Cref{tr:E-binop}]
    Then we have as IH: 
    \begin{align*}
        \trgS{A \triangleright \backtrfencec{\trgS{e_1}} \bigreds \backtrfencec{\trgS{n_1}} : \taint_1 } \\
        \trgS{A \triangleright \backtrfencec{\trgS{e_2}} \bigreds \backtrfencec{\trgS{n_1}} : \taint_2 }
    \end{align*}
    and $ \backtrfencec{\trgS{e_1 + e_2}} = \backtrfencec{\trgS{e_1}} + \backtrfencec{\trgS{e_2}}$.
    
    Follows from IH and \Cref{tr:E-binop}.
    
    \item[\Cref{tr:E-unop}]
    Then by IH we have $\trgS{A \triangleright \backtrfencec{\trgS{e'}} \bigreds \backtrfencec{\trgS{n}} : \taint }$ and $ \backtrfencec{\trgS{\ominus e'}} = \ominus \backtrfencec{\trgS{e'}}$.
    
    Follows from IH and \Cref{tr:E-unop}.
\end{description}
\end{proof}

\BREAK

\begin{proof}[Proof of \Thmref{thm:v4-back-sim-bts-lfence}]\proofref{}{v4-back-sim-bts-lfence}\hfill

 Since $\src{\Omega} \srel_{\src{\OB{f''}}}\ \trgS{\SigmaS}$ we know that $\src{p(A(\pc))} = \trgS{p(A(\pc))}$.
    We proceed by induction on $\trgS{i}$:
    
    \begin{description}
        \item[\textbf{skip}]
        The only change to the state is $\trgS{A'} = \trgS{A[pc \mapsto \succes(\pc)]}$ and we know $\src{\tau^\sigma} = \src{\epsilon}$.

        \begin{description}
            \item[$\src{\Omega \nsarrow{\tau^\sigma} \Omega'}$]
            
                Thus, we can use \Cref{tr:skip} to derive $\src{\Omega \nsarrow{\epsilon} \Omega'}$ with $\src{A'} = \src{A[pc \mapsto \succes(\pc)]}$.
                
            \item[$\src{\tau^\sigma} \tracerel \trgS{\tra{^{\taint}}}$]
            
                $\src{\varepsilon} \tracerel \trgS{\varepsilon}$ follows trivially.
                
            \item[$\src{\Omega'} \srelref_{\src{\OB{f''}}}\, \trgS{\SigmaS'}$]
            
                Follows from $\src{A'} \rrel \trgS{A'}$ which follows from $\src{A} \rrel \trgS{A}$.
            
        \end{description}

        \item[\textbf{spbarr}]
        
        Similar to the \textbf{skip} case above.
        
        \item[\textbf{assign}]
        The only change to the state is $\trgS{A'} = \trgS{A[pc \mapsto \succes(\pc), x \mapsto v : \taint]}$ where $\trgS{\exprEval{A}{e}{v : \taintS}}$.
        Furthermore, we know $\tau^\taint = \epsilon$. 

        \begin{description}
            \item[$\src{\Omega \nsarrow{\tau^\sigma} \Omega'}$]
            
                By \Thmref{thm:v4-back-sim-bte-lfence} we get $\src{\exprEval{A}{\backtrfencec{\trgS{e}}}{\backtrfencec{\trgS{v}} : \taint}}$ and $\src{\taint} \equiv \taintS$.
                
                Thus, we can use \Cref{tr:assign} to derive $\src{\Omega \trgS{\nsarrow{\epsilon}} \Omega'}$ with $\src{A'} = \src{A[pc \mapsto \succes(\pc), x \mapsto \backtrfencec{\trgS{v}} : \taint]}$.

            \item[$\src{\epsilon} \tracerel \trgS{\epsilon}$]
            
                $\src{\varepsilon} \tracerel \trgS{\varepsilon}$ follows trivially.
                
            \item[$\src{\Omega'} \srelref_{\src{\OB{f''}}}\, \trgS{\SigmaS'}$]
            
                Follows from $\src{A'} \rrel \trgS{A'}$ which follows from $\src{A} \rrel \trgS{A}$ and \Thmref{thm:v4-val-rel-bt-lfence} for 
                $\backtrfencec{\trgS{v}} \vrel \trgS{v}$.
            
        \end{description}

        \item[\textbf{load}]
        The only change to the state is $\trgS{A'} = \trgS{A[pc \mapsto \succes(\pc), x \mapsto M(n)]}$ where $\trgS{\exprEval{A}{e}{n : \taint}}$.
        Furthermore, we know $\trgS{\tau^\taint} = \trgS{\loadObs^{\taintpc \glb \taint'}}$.
        
        \begin{description}
            \item[$\src{\Omega \nsarrow{\tau^\sigma} \Omega'}$]
            By \Thmref{thm:v4-back-sim-bte-lfence} we get $\src{\exprEval{A}{\backtrfencec{\trgS{e}}}{\backtrfencec{\trgS{v}} : \taint}}$ and $\src{\taint} \equiv \taintS$.
            
            Thus, we can use \Cref{tr:load} to derive $\src{\Omega \nsarrow{\loadObs{n}^{\safeta}} \Omega'}$ with $\src{A'} = \src{A[pc \mapsto \succes(\pc), x \mapsto \backtrfencec{\trgS{n}} : \taint]}$.

            \item[$\src{\loadObs{\backtrfencec{\trgS{n}}}^{\taint}} \tracerel \trgS{\loadObs{n}^{\taint}}$]
            Equivalence of taints follows from $\srel$, which implies $\trgS{\taintpc} = \trgS{\safeta}$ and thus $\trgS{\taintpc} \glb \trgS{\taint'} = \safeta$.
            The remaining parts follow from \Cref{tr:ac-rel-rd} with \Thmref{thm:v4-val-rel-bt-lfence} for $\backtrfencec{\trgS{n}} \vrel \trgS{n}$.

            \item[$\src{\Omega'} \srelref_{\src{\OB{f''}}}\, \trgS{\SigmaS'}$]
            Follows from $\src{A'} \rrel \trgS{A'}$ which follows from $\src{A} \rrel \trgS{A}$ and $\src{M} \hrel \trgS{M}$.
        \end{description}

        \item[\textbf{load-prv}]
        This case cannot arise by \Thmref{def:atk}.
        
        \item[\textbf{store}]
        In this case \Cref{tr:v4-skip} is not applicable (because we only look at backtranslated statements we know that we are in the context), so we cannot speculate.

        The only change to the state is $\trgS{A'} = \trgS{A[pc \mapsto \succes(\pc)}$ and $\trgS{M'} = \trgS{M[n \mapsto A(x)]}$ where $\trgS{\exprEval{A}{e}{n : \taint}}$.

        Furthermore, we know $\trgS{\tau^\sigma} = \trgS{\storeObs^{\taintpc \glb \taint'}}$.
        
        \begin{description}
            \item[$\src{\Omega \nsarrow{\tau^\sigma} \Omega'}$]
                By \Thmref{thm:v4-back-sim-bte-lfence} we get $\src{\exprEval{A}{\backtrfencec{\trgS{e}}}{\backtrfencec{\trgS{n}} : \taint}}$ and $\src{\taint} \equiv \taintS$.
                
                Thus, we can use \Cref{tr:store} to derive $\src{\Omega \nsarrow{\storeObs{n}^{\at(x)}} \Omega'}$ with $\src{A'} = \src{A[pc \mapsto \succes(\pc)}$ and $\src{M'} = \src{M[n \mapsto A(x)]}$.

            \item[$\src{\storeObs{\backtrfencec{\trgS{n}}}^{\taint}} \tracerel \trgS{\storeObs{n}^{\taint}}$]
                Equivalence of taints follows from $\srel$, which implies $\trgS{\taintpc} = \trgS{\safeta}$ and thus $\trgS{\taintpc} \glb \trgS{\taint'} = \safeta$.
                
                The remaining parts follow from \Cref{tr:ac-rel-wr} with \Thmref{thm:v4-val-rel-bt-lfence} for $\backtrfencec{\trgS{n}} \vrel \trgS{n}$ and $\src{\at(x)} \equiv \trgS{\at(x)}$ follows from $\src{A} \rrel \trgS{A}$.

            \item[$\src{\Omega'} \srelref_{\src{\OB{f''}}}\, \trgS{\SigmaS'}$]
            
                Follows from $\src{A'} \rrel \trgS{A'}$ which follows from $\src{A} \rrel \trgS{A}$ and $\src{M} \hrel \trgS{M}$ and \Thmref{thm:v4-val-rel-bt-lfence} for $\backtrfencec{\trgS{n}} \vrel \trgS{n}$.
        \end{description}
        
        \item[\textbf{store-prv}]
        This case cannot arise by \Thmref{def:atk}.
        
        \item[\textbf{beqz-sat}]
        The only change to the state is $\trgS{A'} = \trgS{A[pc \mapsto l]}$ and $\trgS{A(x) = 0 : \taint}$.
        Furthermore, we know $\trgS{\tau^\taint} = \trgS{\pcObs{l}^{\taint}}$.

        \begin{description}
            \item[$\src{\Omega \nsarrow{\tau^\sigma} \Omega'}$]
            
            From  $\src{A} \rrel \trgS{A}$ we know that $\src{A(x)} = \src{0} : \src{\sigma}$ and $\src{\sigma} \equiv \trgS{\taint}$.
            
            Thus, we can use \Cref{tr:beqz-sat} to derive $\src{\Omega \nsarrow{\pcObs{l}^{\taint}} \Omega'}$ with $\src{A'} = \src{A[pc \mapsto l}$.

            \item[$\src{\pcObs{l}^{\taint}} \tracerel \trgS{\pcObs{l}^{\taintS'}}$]
            
            Follows from \Cref{tr:ac-rel-if} with \Thmref{thm:v4-val-rel-bt-lfence} for $\backtrfencec{\trgS{l}} \vrel \trgS{l}$ and $\src{\taint} \equiv \trgS{\taint}$.

            \item[$\src{\Omega'} \srelref_{\src{\OB{f''}}}\, \trgS{\SigmaS'}$]
            Follows from $\src{A'} \rrel \trgS{A'}$ which follows from $\src{A} \rrel \trgS{A}$.
        \end{description}

        \item[\textbf{beqz-unsat}]
        
        This case is analogous to case \textbf{beqz-sat} above.
        
        \item[\textbf{call}]
        
        Two cases arise:
        
        \begin{description}

            \item[$\trgS{f}$ is defined by the component]
                This case cannot arise as we do not step to a backtranslated instruction $\backtrfencec{\trgS{i}}$ in the target.
            
            \item[$\trgS{f}$ is defined by context]
                Compilation does not change existing functions in the function map $\mathcal{F}$. Thus, $\src{\mathcal{F}(f)} = \trgS{\mathcal{F}(f)}$.
                This means the register $\pc$ is updated in the same way. The rest follows from the relatedness of memories, registers and stacks.
        \end{description}

        \item[\textbf{ret}]
        
        Let $\ffun{\trgS{A(\pc)}} = \trgS{f}$.
        Two cases arise:

        \begin{description}

            \item[$\trgS{f}$ is defined by the component]
                This case cannot arise as we do not step to a backtranslated instruction $\backtrfencec{\trgS{i}}$ in the target.
            \item[$\trgS{f}$ is defined by context]
            Follows from the relatedness of stacks $\src{B} \brel \trgS{B}$. Thus, the $\pc$ register is updated to related values, which means we can show $\src{A'} \arel \trgS{A'}$.
        \end{description}
        
    \end{description}

\end{proof}

\BREAK

\begin{proof}[Proof of \Thmref{thm:v4-ini-state-rel}]\proofref{}{v4-ini-state-rel}\hfill

By definition of \Cref{tr:v4-init} we have $\initFuncS{\ctxS{}\hole{\complfenceS{\src{P}}}} = \trgS{\initFunc(M, \OB{F}, \OB{I}), \bot, \safeta}$
    Let us look at the requirements of $\srel$:

    \begin{description}
        \item[$\src{M} \hrel \trgS{M}$]
        By definition of \Cref{tr:plug-us} we know that the domains of the context heap $M_A$ and the program heap $M_C$ do not overlap. Thus, 
        $\src{M} = \backtrfencec{\trgS{M_A}} \cup \src{M_C}$ and $\trgS{M} = \trgS{M_A} \cup \complfenceS{M_C}$

        Next, for numbers all $z \in \mathbb{Z}$ not in the domain of $M$ we know by definition of $\initFunc{}$ (\Cref{tr:ini-us}) that the relatedness holds.

        We now tackle the remaining heap parts:
        \begin{description}
            \item[$\src{M_C} \hrel \complfenceS{M_C}$]
            This follows from \Thmref{thm:v4-heap-rel-comp-lfence}.

            \item[$\backtrfencec{\trgS{M_A}} \hrel \trgS{M_A}$]
            For the context heap we know by \Cref{def:atk} that its domain is the natural numbers. So this holds by \Thmref{thm:v4-heap-rel-bt-lfence}.
        \end{description}

        \item[$\src{A} \rrel \trgS{A}$]
        Follows from \Cref{tr:ini-us}. All register values are set to $0 : \safeta$ and the $\pc$ register is set to the label of the \textit{main} function. This follows from \Thmref{thm:v4-val-rel-comp-lfence} for the label of \textit{main}.
        We now trivially satisfy $\src{A} \rrel \trgS{A}$.

        \item[$\src{\OB{F}; p} \crel \trgS{\OB{F}; p}$]
        Follows by the definition of $\complfenceS{}$ and $\backtrfencec{}$. Note, that the compiler only adds new instructions (for the countermeasure) but does not delete any other instructions.

        \item[$\src{\OB{B}} \brel \trgS{\OB{B}}$]
        Follows by inspection of $\initFunc{}$ (since it sets the stacks to the empty set) and \Cref{tr:v4-brel-b}.
        
        \item[$\src{\OB{F}} \equiv \trgS{\OB{F}}$ and $\src{\OB{I}} \equiv \trgS{\OB{I}}$]
        Follows from the inspection of $\backtrfencec{}$ and $\complfenceS{}$.

        \item[target taint is $\trgS{\safeta}$]
        Follows from the definition of $\initFuncS$.
        
        \item[target window is $\trgS{\bot}$]
        Follows from the definition of $\initFuncS{}$.
    \end{description}  
\end{proof}

\BREAK

\begin{proof}[Proof of \Thmref{thm:v4-val-rel-comp-lfence}]\proofref{}{v4-val-rel-comp-lfence}\hfill

Trivial analysis of the compiler.
\end{proof}

\BREAK 

\begin{proof}[Proof of \Thmref{thm:v4-heap-rel-comp-lfence}]\proofref{}{v4-heap-rel-comp-lfence}\hfill

Trivial analysis of the compiler.
\end{proof}

\BREAK 

\begin{proof}[Proof of \Thmref{thm:v4-val-rel-bt-lfence}]\proofref{}{v4-val-rel-bt-lfence}\hfill

Trivial analysis of the backtranslation.
\end{proof}

\BREAK

\begin{proof}[Proof of \Thmref{thm:v4-heap-rel-bt-lfence}]\proofref{}{v4-heap-rel-bt-lfence}\hfill

Trivial analysis of the backtranslation.
\end{proof}

\BREAK

\begin{proof}[Proof of \Thmref{thm:v4-taint-rel-bt-lfence}]\proofref{}{v4-taint-rel-bt-lfence}\hfill

Trivial analysis of the backtranslation.
\end{proof}

\BREAK

\begin{proof}[Proof of \Thmref{thm:v4-val-uniqueR}]\proofref{}{v4-val-uniqueR}\hfill

Follows by inspection of $\vrelref$.
\end{proof}

\BREAK

\begin{proof}[Proof of \Thmref{thm:v4-val-uniqueL}]\proofref{}{v4-val-uniqueL}\hfill

Follows by inspection of $\vrelref$.
\end{proof}
 \section{$\Jv$: The Retpoline compiler}\label{comp:v2-retpoline}

Retpoline is a specific code sequence that replaces indirect jumps by a "return trampoline".

Conceptually retpoline looks like this:
\begin{center}
\begin{minipage}[b]{0.45\linewidth}
\begin{lstlisting}[basicstyle=\small,style=MUASMstyle,  escapechar=|, captionpos=t,]
jmp rax
\end{lstlisting}
\end{minipage}
\hfill
\begin{minipage}[b]{0.45\linewidth}
\begin{lstlisting}[basicstyle=\small,style=MUASMstyle, escapechar=|, captionpos=t,]
call retpoline_rax
|$\cdots$|
retpoline_rax:
l_0 : call LIND1   
LIND0:
l_1 : skip
l_2 : spbarr
l_3 : jmp l_1
LIND1
l_4 : modret rax
l_5 : ret
\end{lstlisting}
\end{minipage}
\captionof{lstlisting}{Code on the left and the compiled program with the retpoline countermeasure on the right.}    
\end{center}
We also assume that the labels used for this new thunk are not used in the program.
In $l_4$ we modify the return address to be the content of the $rax$ register, which was the initial jump goal.
Thus the $\pret$ instruction at the end returns to the target of the original indirect jump.

Notice, that this retpoline is constructed for an indirect jump targeting the $rax$ register. We need to create additional functions for each indirect jump. This is option $thunk$.

Thus, we assume a pass before the retpoline pass that for each indirect jump in the program, adds a new function with a fresh name \textit{retpo\textunderscore trg\textunderscore l} for the indirect jump at location $l$.
This function for $\pjmp{e}$ consists of:
\begin{align*}
	l' :& \pmodret \compretpJ{e} \\
    l'':& \pret \\
\end{align*}

where $l'$ and $l''$ are freshly chosen labels.

We are now ready to state the compiler:
\begin{align*}
    \compretpJ{ M ; \OB{F} ; \OB{I}} &= \trgJ{ \compretpJ{M} ; \compretpJ{\OB{F}} ; \compretpJ{\OB{I}}} \\\\
	\compretpJ{ \srce} &= \trgJ{\emptyset}
	\\
	\compretpJ{ \OB{I}\cdot f } &= \trgJ{ \compretpJ{ \OB{I} }\cdot \compretpJ{f} }
	\\\\
	\compretpJ{ M ; -n\mapsto v : \unta} &= \trgJ{\compretpJ{M} ; -\compretpJ{n} \mapsto \compretpJ{v} : \unta}
	\\\\
	\compretpJ{p_1;p_2} &= \trgJ{\compretpJ{p_1} ; \compretpJ{p_2}}
	\\
	\compretpJ{l : i} &= \trgJ{\compretpJ{l} : \compretpJ{i}} ~~ \text{if $i \neq \jmpC$}
	\\\\
	\compretpJ{\pskip} &= \trgJ{\pskip} 
	\\
	\compretpJ{\passign{x}{e}} &= \trgJ{\passign{x}{\compretpJ{e}}}
        \\
        \compretpJ{\pcondassign{x}{e}{e'}} &= \trgJ{\pcondassign{x}{\compretpJ{e}}{\compretpJ{e'}}}
        \\
        \compretpJ{\ploadprv{x}{e}} &= \trgJ{\ploadprv{x}{\compretpJ{e}}}
	\\
	\compretpJ{\pload{x}{e}} &= \trgJ{\pload{x}{\compretpJ{e}}}
	\\
        \compretpJ{\pstore{x}{e}} &= \trgJ{\pstore{x}{\compretpJ{e}}}
        \\
        \compretpJ{\pstoreprv{x}{e}} &= \trgJ{\pstoreprv{x}{\compretpJ{e}}}
        \\
	\compretpJ{\pjz{x}{l'}} &= \trgJ{\pjz{x}{\compretpJ{l'}}} 
	\\
	\compretpJ{\pcall{f}} &= \trgJ{\pcall{f}} ~ \text{Here $f$ is a function name}
	\\
	\compretpJ{\pret} &= \trgJ{\pret}
        \\\\
	\compretpJ{l : \pjmp{e}} &= \trgJ{l : \pjmp{\compretpJ{e}}} ~~ \text{if $\pjmp{e}$ not indirect}
        \\
	\compretpJ{l : \pjmp{e}} &= \trgJ{
                                         \begin{aligned}[t]
                                                    &
                                                    l': \pcall{\mathit{retpo\textunderscore trg\textunderscore l}}
                                                        \\
                                                        &\
                                                        l'_1: \pskip
                                                            \\
                                                            &\
                                                            l'_2 : \barrierKywd
                                                            \\
                                                            &\
                                                            l'_3 : \pjmp{l'_1}
                                            \end{aligned}
    } \text{ if $r \in \Reg$ and $r \in e$. This means its indirect. All labels are chosen fresh.}
\end{align*}

The Memory relation $\hrel$, register relation $\rrel$ and value relation $\vrel$ are the exact same as in \Cref{comp:v4-fence}.

\mytoprule{\text{Register relation} \rreldef  \text{ State relation} \sreldef}
\begin{center}

    \typerule{Stack - base }{}{
		\srce \brel \trgSe
	}{brel-b}
	\typerule{Stack - start }{
		\src{B} \brel \trgJ{B}
		&
		\src{\OB{n}} \brel \trgJ{\OB{n}}
	}{
		\src{B ; \OB{n}} \brel \trgJ{B ; \OB{n}}
	}{v2-brel-start}

 \typerule{Stack Region - ind }{
		\src{\OB{n}} \brel \trgJ{\OB{n}} 
		\\
		\src{z} \vrel \trgJ{z}
		&
		\src{\sigma} \equiv \trgJ{\taint}
        & 
        \src{v} \vrel \trgJ{v}
	}{
		\src{\OB{n} ; z \mapsto v:\sigma} \brel \trgJ{\OB{n} ; z\mapsto v:\taint}
	}{v2-brel-region}

	\typerule{ States }{
        \src{A} \rrel \trgJ{A}
        &
		\src{M} \hrel \trgJ{M}
		&
		\src{\OB{f}} \equiv \trgJ{\OB{f}}
        &
		\src{C} \crel_{\src{\OB{f''}}} \trgJ{C}
	}{
		\src{C,\OB{B}, \tup{p, M, A} \proc{s}{\OB{f}} } \srel_{\src{\OB{f''}}} \trgJ{C,\OB{B}, \tup{p, M, A}, \bot, \safeta }
	}{stats}

\typerule{ Programs }{
        \forall \src{l : i} \in \src{p} \text{ and } \ffun{\src{l}} = \src{f} \text{ and } \ffun{\compretpJ{\src{l}}} = \trgJ{f} 
        \\
		\text{ if } \src{f} \in \src{\OB{F}} \text{ and } \src{f} \in \src{\OB{f}}
		\text{ then } \trgJ{p}(\compretpJ{l}) = \compretpJ{i}
		\\
  \forall \trgJ{l : i} \in \trgJ{p} \ldotp \ffun{\trgJ{l}} = \trgJ{f} \text{ and } \trgJ{f} \notin \src{\OB{f}} \text{ and } \trgJ{f} \in \trgJ{\OB{F}}\text{ then }
    \src{p(l)} = \backtrfencec{\trgJ{i}}
		\\
	}{
		\src{\OB{F}; \src{p}} \crel_{\src{\OB{f}}} \trgJ{\OB{F}; \trgJ{p}} 
	}{v2-comps}

 \typerule{ States }{
        \src{A} \rrel \trgJ{A}
        &
		\src{\OB{F}; p} \crel \trgJ{\OB{F}; p}
		&
		\src{M} \hrel \trgJ{M}
		&
        \src{B} \brel \trgJ{B}
        &
        \src{\OB{I}} \equiv \trgJ{\OB{I}}
	}{
		\src{\OB{F};\OB{I},\OB{B}, \tup{p, M, A}} \srel_{\src{\OB{f''}}} \trgJ{\OB{F};\OB{I},\OB{B}, \tup{p, M, A}, \bot, \safeta }
	}{v2-states}
\end{center}
\botrule

\subsection{The Proofs}

We need to be especially careful with the trace relation $\tracerel$. Look at the possible trace generated before and after compilation:
\begin{tikzpicture}
    \node[circle, draw] (start) {jmp rax};
    \node[draw, above=0.5cm of start]{Original Program};

    \node[circle, draw, right= of start] (0) {$\cdots$};
    \draw[->] (start) edge [bend left] node[above] {} (0) ;

    \node[red, circle, draw, below=2cm of start] (1m) {call};
    \node[draw, above=0.8cm of 1m]{Compiled Program};

    \node[circle, draw, right= of 1m] (3m) {$\cdots$};
    \node[circle, draw, right= of 3m] (4m) {modret rax};
    \node[circle, draw, right= of 4m] (5m) {ret};
    \node[circle, draw, right= of 5m] (6m) {$\cdots$};

    \draw[->] (1m) edge [bend left] node[above] {call retpo} (4m);
    \draw[->] (4m) edge [bend left] node[above] {} (5m);
    \draw[->] (5m) edge [bend left] node[above] {ret!} (6m);
    
\end{tikzpicture}

Note that the retpoline compiler replaces all indirect $\jmpC$ instructions in the program. Thus, the compiled program cannot speculate anymore.
This simplifies the proofs.
\begin{theorem}[$\Jv$: The retpoline compiler is \rdss]\label{thm:v2-retpo-comp-rdss}
	\begin{align*}
		\contract{}{\Jv} \vdash \compretpJ{\cdot} : \rdss
	\end{align*}
\end{theorem}
\begin{proof}

This means there is an action $\trgJ{\tau^{\taint}}$ that is not related to a source action that can possibly be empty.
This means that this particular action is not in the action relation with the source.
Inspecting the action relation we can conclude that $\trgJ{\taint} = \unta$. As otherwise the action would be related to some source action (even the empty one).

However, this means that speculation needed to happen.
But we know that is not the case
\end{proof}

\begin{theorem}[$\Jv$: All retpo-compiled programs are \rss]\label{thm:v2-all-retpo-comp-are-rdss}
	\begin{align*}
		\forall\src{P}\ldotp \contract{}{\Jv} \vdash\compretpJ{P} : \rss
	\end{align*}
\end{theorem}
\begin{proof}
	By \Thmref{thm:ss-sni-source} we know all source languages are SS: $\text{HPS:} \forall\src{P}\ldotp \contract{}{NS}\vdash P : \rss$.
    
    By \Thmref{thm:v2-retpo-comp-rdss} we have HPC: $\contractSpec{}{\Jv} \vdash \compretpJ{\cdot} : \rdss$.

    By \Thmref{thm:rdss-eq-rdsp} with HPC we have $\text{HPP:} \contractSpec{}{\Jv}\vdash \compretpJ{}{\cdot} : \rdssp$

    Unfolding \Thmref{def:rdsspc} of HPP we get 
    $\forall\src{P}\ldotp \text{ if } \contract{}{NS} \vdash\src{P} : \rss \text{ then } \contractSpec{}{\Jv}\vdash\comp{\src{P}} : \rss$. 
  
    Now, we instantiate it with HPS and can conclude that $\forall\src{P}\ldotp \contractSpec{}{\Jv}\vdash\compretpJ{P} : \rss$.
\end{proof}

 \section{$\Jv$: The Retpoline compiler with a Fence at the end}\label{comp:v2-retpolineF}

The original retpoline compiler (\Cref{comp:v2-retpoline}) adds return instructions to the code that can be misused by straightline-speculation.
We propose an extended retpoline that adds another fence instruction to the end of the retpoline sequence. Thus, stopping possible mispredictions of $\semsls$.

\begin{align*}
    \compretpJF{ M ; \OB{F} ; \OB{I}} &= \trgJ{ \compretpJF{M} ; \compretpJF{\OB{F}} ; \compretpJF{\OB{I}}} \\\\
	\compretpJF{ \srce} &= \trgJ{\emptyset}
	\\
	\compretpJF{ \OB{I}\cdot f } &= \trgJ{ \compretpJF{ \OB{I} }\cdot \compretpJF{f} }
	\\\\
	\compretpJF{ M ; -n\mapsto v : \unta} &= \trgJ{\compretpJF{M} ; -\compretpJF{n} \mapsto \compretpJF{v} : \unta}
	\\\\
	\compretpJF{p_1;p_2} &= \trgJ{\compretpJF{p_1} ; \compretpJF{p_2}}
	\\
	\compretpJF{l : i} &= \trgJ{\compretpJF{l} : \compretpJF{i}} ~~ \text{if $i \neq \jmpC$}
	\\\\
	\compretpJF{\pskip} &= \trgJ{\pskip} 
	\\
	\compretpJF{\passign{x}{e}} &= \trgJ{\passign{x}{\compretpJF{e}}}
        \\
        \compretpJF{\pcondassign{x}{e}{e'}} &= \trgJ{\pcondassign{x}{\compretpJF{e}}{\compretpJF{e'}}}
        \\
        \compretpJF{\ploadprv{x}{e}} &= \trgJ{\ploadprv{x}{\compretpJF{e}}}
	\\
	\compretpJF{\pload{x}{e}} &= \trgJ{\pload{x}{\compretpJF{e}}}
	\\
        \compretpJF{\pstore{x}{e}} &= \trgJ{\pstore{x}{\compretpJF{e}}}
        \\
        \compretpJF{\pstoreprv{x}{e}} &= \trgJ{\pstoreprv{x}{\compretpJF{e}}}
        \\
	\compretpJF{\pjz{x}{l'}} &= \trgJ{\pjz{x}{\compretpJF{l'}}} 
	\\
	\compretpJF{\pcall{f}} &= \trgJ{\pcall{f}} ~ \text{Here $f$ is a function name}
	\\
	\compretpJF{\pret} &= \trgJ{\pret}
        \\\\
	\compretpJF{l : \pjmp{e}} &= \trgJ{l : \pjmp{\compretpJF{e}}} ~~ \text{if $\pjmp{e}$ not indirect}
        \\
	\compretpJF{l : \pjmp{e}} &= \trgJ{
                                         \begin{aligned}[t]
                                                    &
                                                    l': \pcall{}
                                                        \\
                                                        &\
                                                        l'_1: \pskip
                                                            \\
                                                            &\
                                                            l'_2 : \barrierKywd
                                                            \\
                                                            &\
                                                            l'_3 : \pjmp{l'_1}
                                                            \\
                                                            &\
                                                            l'_4 : modret \compretpJF{e}
                                                            \\
                                                            &\
                                                            l'_5: \pret
                                                            \\
                                                            &\
                                                            l'_6 : \barrierKywd
                                            \end{aligned}
    } \text{ if $r \in \Reg$ and $r \in e$. This means its indirect.} 
    \\
    \text{ All labels are chosen fresh.}
\end{align*}

\subsection{The Proofs}
\begin{theorem}[$\Jv$: The retpoline fence compiler is \rdss]\label{thm:v2-retpof-comp-rdss}
	\begin{align*}
		\contract{}{\Jv} \vdash \compretpJF{\cdot} : \rdss
	\end{align*}
\end{theorem}
\begin{proof}

The proof is analogous to \Thmref{thm:v2-retpo-comp-rdss} since the additional fence is never executed and all indirect jumps are removed from the code.
\end{proof}

\begin{theorem}[$\Jv$: All retpo-fence-compiled programs are \rss]\label{thm:v2-all-retpof-comp-are-rdss}
	\begin{align*}
		\forall\src{P}\ldotp \contract{}{\Jv} \vdash\compretpJF{P} : \rss
	\end{align*}
\end{theorem}
\begin{proof}
The proof is analogous to \Thmref{thm:v2-all-retpo-comp-are-rdss} since the additional fence is never executed.
\end{proof}

And we can now show that this compiler is independent w.r.t $\contract{}{\SLSv}$. The additional $\retC$ instructions inserted by the compiler are now guarded by a barrier instruction afterwards.

\section{$\Rvr$: The Lfence Compiler for \specr}\label{comp:v5-fence}

Here, fences are added after every call instruction. Thus, all misprediction sites caused by a misprediction of the RSB are guarded by a fence.
Since this speculation is related to the return stack buffer (RSB) which uses the return address of every call for speculation, this barrier will stop speculation. Note that our semantics does not speculate when the RSB is empty.

\begin{align*}
    \complfenceR{ M ; \OB{F} ; \OB{I}}{} &= \trgR{ \complfenceR{M}{} ; \complfenceR{\OB{F}}{} ; \complfenceR{\OB{I}}{}}
	\\\\
	\complfenceR{ \srce}{} &= \trgR{\emptyset}
	\\
	\complfenceR{ \OB{I}\cdot f }{} &= \trgR{ \complfenceR{ \OB{I} }\cdot \complfenceR{f}{} }
	\\\\
	\complfenceR{ M ; -n\mapsto v : \unta}{} &= \trgR{\complfenceR{M}{} ; -\complfenceR{n}{} \mapsto \complfenceR{v}{} : \unta}
	\\\\
	\complfenceR{p_1;p_2}{} &= \trgR{\complfenceR{p_1}{} ; \complfenceR{p_2}{}}
	\\
        \complfenceR{l : i} &= \trgR{\complfenceR{l} : \complfenceR{i}} ~~ \text{if $i \neq \callC$}
	\\\\
	\complfenceR{\pskip} &= \trgR{\pskip} 
	\\
	\complfenceR{\passign{x}{e}} &= \trgR{\passign{x}{\complfenceR{e}}}
        \\
        \complfenceR{\pcondassign{x}{e}{e'}} &= \trgR{\pcondassign{x}{\complfenceR{e}}{\complfenceR{e'}}}
        \\
        \complfenceR{\ploadprv{x}{e}} &= \trgR{\ploadprv{x}{\complfenceR{e}}}
	\\
	\complfenceR{\pload{x}{e}} &= \trgR{\pload{x}{\complfenceR{e}}}
	\\
        \complfenceR{\pstore{x}{e}} &= \trgR{\pstore{x}{\complfenceR{e}}}
        \\
        \complfenceR{\pstoreprv{x}{e}} &= \trgR{\pstoreprv{x}{\complfenceR{e}}}
        \\
	\complfenceR{\pjz{x}{l'}} &= \trgR{\pjz{x}{\complfenceR{l'}}} 
	\\
	\complfenceR{\pjmp{e}} &= \trgR{\pjmp{\complfenceR{e}}} 
	\\
	\complfenceR{\pret} &= \trgR{\pret}
	\\
	\complfenceR{l : \pcall{f}}{} &= \trgR{\complfenceR{l}{} : \pcall{f} ; \complfenceR{l'}{} : \pbarrier } ~ 
	\\ \end{align*}

This means that every entry in the RSB $\trgR{\Rsb}$ leads to a barrier instruction.

\subsection{Proofs}

Inserting a fence stops speculation. That is why these proofs are very similar to the ones in \Cref{comp:v4-fence}.

\begin{center}
	
\scalebox{0.7}{
\begin{forest}
for tree={
    align = left,
    font = \footnotesize,
    forked edge,
}
[\Thmref{thm:v5-all-lfence-comp-are-rdss}
        [\Thmref{thm:v5-lfence-comp-rdss}
                [\Thmref{thm:v5-corr-bt-lfence}
                    [\Thmref{thm:v5-ini-state-rel}
                        [\Cref{thm:v4-heap-rel-comp-lfence}]
                        [\Cref{thm:v4-heap-rel-bt-lfence}, name=heap-bt]
                    ]
                    [\Thmref{thm:v5-bwd-sim-lfence}
                            [\Thmref{thm:v5-bwd-sim-comp-steps-lfence}
                                    [\Thmref{thm:v5-fwd-sim-stm-lfence}
                                        [\Thmref{thm:v4-fwd-sim-exp-lfence}]
                                    ]
                            ]
                            [\Thmref{thm:v4-back-sim-bts-lfence}   
                                [\Thmref{thm:v4-back-sim-bte-lfence}, name=exp-bt]
                            ]
                    ]
                ]
        ]
]
\node (sniSource) at (-4,3) {\Thmref{thm:ss-sni-source}};
\node (rsscEq) [below=of sniSource] {\Thmref{thm:rdss-eq-rdsp}};
\draw[red, thick, dotted] ($(sniSource.north west)+(-0.3,0.3)$)  rectangle ($(rsscEq.south east)+(0.9,-0.6)$);
\end{forest}
	}
\end{center}

\begin{theorem}[$\Rv$: The lfence compiler is \rdss]\label{thm:v5-lfence-comp-rdss}(\showproof{v5-lfence-comp-rdss})
	\begin{align*}
		\contractSpec{}{\Rv} \vdash \complfenceR{\cdot} : \rdss
	\end{align*}
\end{theorem}

\begin{theorem}[$\Rv$: All lfence-compiled programs are \rss)]\label{thm:v5-all-lfence-comp-are-rdss}(\showproof{v5-all-lfence-comp-are-rdss})
	\begin{align*}
		\forall\src{P}\ldotp \contractSpec{}{\Rv} \vdash\complfenceR{P} : \rss
	\end{align*}
\end{theorem}

The relation is the same as the one in \Cref{comp:v4-fence}.

\begin{theorem}[$\Rv$: Correctness of the Backtranslation]\label{thm:v5-corr-bt-lfence}(\showproof{v5-corr-bt-lfence})
	\begin{align*}
		\text{ if } 
			&\
			\trgR{ \amTracevR{\ctxR{}\hole{\complfenceR{P}}}{\tra{^{\taint}}} }
		\\
		\text{ then }
			&
			\src{ \Trace{\backtrfencec{\ctxR{}}\hole{P}}{ \tras{^{\sigma}}} }
		\\
		\text{ and }
			&\
			\tras{^{\sigma}} \rels \trgR{\tra{^{\taint}}}
	\end{align*}
\end{theorem}

\begin{theorem}[$\Rv$: Generalised Backward Simulation]\label{thm:v5-bwd-sim-lfence}(\showproof{v5-bwd-sim-lfence})
	\begin{align*}
		\text{ if }
			&\
			\ffun{\src{A(\pc)}}\in\src{\OB{f''}} \text{ then } \trgR{A(\pc)} : \trgR{i} = \complfenceR{l : i} \text{ or } (\complfenceR{A(\pc) : i} = \trgR{A(\pc)} : \trgR{i} ; \trgR{p_c'}) 
			\\
		\text{ else } 
			&\
			\src{A(\pc)} : \src{i} = \backtrfencec{\trgR{A(\pc)}} : \backtrfencec{\trgR{i}} 
			\\
		\text{ and}
			&\
			\text{ if } \ffun{\src{A'(\pc)}}\in\src{\OB{f''}} \text{ then } \trgR{A'(\pc)} : \trgR{i'} = \complfenceR{A'(\pc) : i'} \text{ or } (\complfenceR{A'(\pc) : i'} = \trgR{A'(\pc)} : \trgR{i'} ; \trgR{p_c''}) 
			\\
		\text{ else }
			&\ 
			\src{A'(\pc)} : \src{i'} = \backtrfencec{\trgR{A'(\pc)}} : \backtrfencec{\trgR{i'}}
			\\
		\text{ and }
			&\
			\SigmaR = \trgR{(C, \OB{B}, \tup{p, M, A}, \bot,\safeta)}
		\\
		\text{ and }
			&\
			\SigmaRt{'}=\trgR{(C, \OB{B'}, \tup{p, M', A'},\bot,\safeta)}
		\\
		\text{ and }
			&\
			\trgR{ \SigmaR \bigspecarrowR{\tra{^{\taint}}} \SigmaR' }
		\\
		\text{ and }
			&\
			\src{\Omega} \srel_{\src{\OB{f''}}} \trgR{\SigmaR}
		\\
        \text{ and }
			&\
			\trgR{p(A(\pc))} = \trgR{i}  \text{ and } \trgR{p(A'(\pc))} = \trgR{i'}
		\\
        \text{ and }
			&\
			\src{p(A(\pc))} =  \src{i} \text{ and } \src{p(A'(\pc))} = \src{i'}
        \\
		\text{ then }
			&\
			\src{\Omega}=\src{C, \OB{B}, \tup{p , M, A}  \nsbigarrow{\tras{^\sigma}} C, \OB{B'}, \tup{p , M', A'} }=\src{\Omega'}
		\\
		\text{ and }
			&\
			\src{\tras{^\sigma}} \tracerel \trgR{\tra{^{\taint}}}
		\\
		\text{ and }
			&\
			\src{\Omega'} \srelref_{\src{\OB{f''}}}\, \trgR{\SigmaR{'}}
	\end{align*}
\end{theorem}

\begin{theorem}[$\Rv$: Backward Simulation for Compiled Steps]\label{thm:v5-bwd-sim-comp-steps-lfence}(\showproof{v5-bwd-sim-comp-steps-lfence})
	\begin{align*}
		\text{ if }
			&\
			\trgR{\SigmaR}=\trgR{(C, \OB{B}, \tup{p, M, A}, \bot, \safeta) }
		\\
		\text{ and }
			&\
			\trgR{\SigmaR'}=\trgR{(C, \OB{B'}, \tup{p, M', A'}, \bot, \safeta)}
		\\
		\text{ and }
			&\
			\trgR{\SigmaR \bigspecarrowR{\tra{^{\taint}}} \SigmaR'}
		\\
		\text{ and }
			&\
			\src{\Omega} \srelref_{\src{\OB{f''}}}\, \SigmaR
		\\
        \text{ and }
			&\
			\trgR{p(A(\pc))} = \complfenceR{l} : \complfenceR{i} \text{ and } \trgR{p(A'(\pc))} = \complfenceR{l'} : \complfenceR{i'}
		\\
         \text{ and }
            &\
            \src{p(A(\pc))} =  \src{i} \text{ and } \src{p(A'(\pc))} =  \src{i'}
            \\
        \text{ and }
			&\
			\trgR{A(\pc) : i} = \complfenceR{A(\pc) : i}   \text{ or } (\complfenceR{A(\pc) : i} = \trgR{A(\pc)} : \trgR{i} ; \trgR{p_c'})
            \\
            \text{ and }
                &\
                \trgR{A'(\pc) : i'} = \complfenceR{A'(\pc) : i'}  \text{ or } (\complfenceR{A'(\pc) : i'} = \trgR{A'(\pc)} : \trgR{i'} ; \trgR{p_c''})
		\\
		\text{ then }
			&\
			\src{\Omega}=\src{C, \OB{B}, \tup{p, M, A} \nsarrow{\tau^\sigma} C, \OB{B'}, \tup{p, M', A'}}= \src{\Omega'}
        \\
		\text{ and }
			&\
			\src{\tau^\sigma} \tracerel \trgR{\tra{^{\taint}}} \qquad \text{( using the trace relation!)}
		\\
		\text{ and }
			&\
			\src{\Omega'} \srelref_{\src{\OB{f''}}}\, \SigmaR'
	\end{align*}
\end{theorem}

\begin{lemma}[$\Rv$: Forward Simulation for Compiled Statements]\label{thm:v5-fwd-sim-stm-lfence}(\showproof{v5-fwd-sim-stm-lfence})
	\begin{align*}
		\text{ if }
			&\
			\src{\Omega}=\src{C, \OB{B}, \tup{p, M, A} \nsarrow{\tau^\sigma} C, \OB{B'}, \tup{p, M', A'} } =\src{\Omega'}
		\\
		\text{ and }
			&\
			\src{\Omega} \srelref_{\src{\OB{f''}}} \trgR{\SigmaR}
		\\
        \text{ and }
			&\
			\src{p(A(\pc))} = \src{i} \text{ and } \src{p(A'(\pc))} = \src{i'}
		\\
		\text{ and }
			&\
			\trgR{\SigmaR}= \trgR{C, \OB{B}, \tup{p, M , A}, \bot, \safeta}, 
		\\
		\text{ and }
			&\
			\trgR{\SigmaR'}= \trgR{C, \OB{B}, \tup{p, M' , A'}, \bot, \safeta}
		\\
        \text{ and }
			&\
			\trgR{A(\pc) : i} = \complfenceR{A(\pc) : i}   \text{ or } (\complfenceR{A(\pc) : i} = \trgR{A(\pc)} : \trgR{i} ; \trgR{p_c'})
            \\
            \text{ and }
                &\
                \trgR{A'(\pc) : i'} = \complfenceR{A'(\pc) : i'}  \text{ or } (\complfenceR{A'(\pc) : i'} = \trgR{A'(\pc)} : \trgR{i'} ; \trgR{p_c''})
		\\
		\text{ then }
			&\
			\trgR{\SigmaR \bigspecarrowR{\tra{^{\taint}}} \SigmaR'}
		\\
		\text{ and }
			&\
			\src{\tau^\sigma} \tracerel \trgR{\tra{^{\taint}}} \qquad \text{( using the trace relation!)}
		\\
		\text{ and }
			&\
			\src{\Omega'} \srelref_{\src{\OB{f''}}}\ \trgR{\SigmaR'}
	\end{align*}
\end{lemma}

\begin{lemma}[Initial States are Related]\label{thm:v5-ini-state-rel}(\showproof{v5-ini-state-rel})
	\begin{align*}
		&
		\forall \src{P}, \forall \src{\OB{f}}=\dom{\src{P}.\src{F}}, \forall \ctxR{}
		\\
		&
		\initFunc{\backtrfencec{\ctxR{}}\src{\hole{P}}} \srelref_{\src{\OB{f}}}\, \initFuncR{\ctxR{\hole{\complfenceR{\src{P}}}}}
	\end{align*}
\end{lemma}

\begin{proof}[Proof of \Thmref{thm:v5-lfence-comp-rdss}]\proofref{}{v5-lfence-comp-rdss}\hfill

    Instantiate $\src{A}$ with $\backtrfencec{\trgR{A}}$.
    This holds by \Thmref{thm:v5-corr-bt-lfence}.
\end{proof}

\BREAK

\begin{proof}[Proof of \Thmref{thm:v5-all-lfence-comp-are-rdss}]\proofref{}{v5-all-lfence-comp-are-rdss}\hfill

    By \Thmref{thm:ss-sni-source} we know all source languages are SS: $\text{HPS:} \forall\src{P}\contract{}{NS}\ldotp\vdash P : \rss$.
    
    Since \Thmref{thm:v5-lfence-comp-rdss} we have HPC: $\contractSpec{}{\Rv}\vdash \complfenceR{\cdot} : \rdss$.

    By \Thmref{thm:rdss-eq-rdsp}  with HPC we have $\text{HPP} \contractSpec{}{\Rv}\vdash \complfenceR{\cdot} : \rdssp$

    Unfolding \Thmref{def:rdsspc} of HPP we get 
    $\forall\src{P}\ldotp \text{ if } \contract{}{NS}\vdash\src{P} : \rss \text{ then } \contractSpec{}{\Rv} \vdash\comp{\src{P}} : \rss$. 
  
    Now, we instantiate it with HPS and can conclude that $\forall\src{P}\ldotp \contractSpec{}{\Rv}\vdash\complfenceR{P} : \rss$.
    
\end{proof}

\BREAK

\begin{proof}[Proof of \Thmref{thm:v5-ini-state-rel}]\proofref{}{v5-ini-state-rel}\hfill

Analogous to \Thmref{thm:v4-ini-state-rel},
\end{proof}

\BREAK

\begin{proof}[Proof of \Thmref{thm:v5-corr-bt-lfence}]\proofref{}{v5-corr-bt-lfence}\hfill

This holds by \Thmref{thm:v5-ini-state-rel} and \Thmref{thm:v5-bwd-sim-lfence}.
\end{proof}

\BREAK

\begin{proof}[Proof of \Thmref{thm:v5-bwd-sim-lfence}]\proofref{}{v5-bwd-sim-lfence}\hfill

We proceed by induction on $\SigmaR \bigspecarrowR{\tra{^{\taint}}} \SigmaR '$:
\begin{description}
    \item[\Cref{tr:v5-reflect}]
    Trivial.
    \item[\Cref{tr:v5-single}]
    We have 
    \begin{align*}
        \trgR{\SigmaR \bigspecarrowR{\OB{\tau^{\taint}}} \SigmaR ''} \\
        \trgR{\SigmaR '' \specarrowR{\tau^{\taint}} \SigmaR '} \\ 
        \trgR{\SigmaR ''} = \trgR{\OB{C''}, \OB{B''}, \tup{p ,M'', A''}, \bot, \safeta}
    \end{align*}

    Applying IH on $\trgR{\SigmaR  \bigspecarrowR{\OB{\tau^{\taint''}}} \SigmaR ''}$ we get: 
    \begin{enumerate}
        \item $\src{\Omega  \nsbigarrow{\tras{^{\sigma''}}} \Omega''}$ and
        \item $\src{\tras{^{\sigma''}}} \tracerel \trgR{\tra{^{\taint''}}}$ and
        \item $\src{\Omega''} \srelref_{\src{\OB{f''}}}\, \trgR{\SigmaR {''}}$
    \end{enumerate}
    Let $\ffun{\trgR{l''}} = \trgR{f''}$ and let $\ffun{\trgR{l'}} = \trgR{f'}$ (where $\trgR{\SigmaR ' . A(\pc)} = \trgR{l'}$).
    We proceed by case analysis on $\src{f''} \in \src{\OB{f''}}$:
    \begin{description}
        \item[$\src{f''} \in \src{\OB{f''}}$ (\textbf{in the compiled component})]
            By case analysis on $\src{f'} \in \src{\OB{f''}}$:
            \begin{description}
                \item[$\src{f'} \in \src{\OB{f''}}$ (\textbf{in the compiled component})]

                This holds by \Thmref{thm:v5-bwd-sim-comp-steps-lfence}.
                
                \item[$\src{f'} \notin \src{\OB{f''}}$ (\textbf{in the context})]
                This switch from component to context arises in two cases:
                \begin{description}
                    \item[\textbf{call}]
                    This is a call from a compiled function to a context function.
                    
                    In this case we know that $\trgR{\SigmaR '' \specarrowR{\tau^{\taint}} \SigmaR '}$ was derived by \Cref{tr:v5-nospec-act} combined with \trgR{\Cref{tr:call}} and we know $\trgR{\tau^{\taint}} = \trgR{\clh{}^{\safeta}}$.
                    
                    This case is analogous to the corresponding case in  \Thmref{thm:v4-bwd-sim-lfence}.

                    \item[\textbf{return}]
                    This is a return from a compiled function to a context function.
                    This is the dual of the case below for return from context to component.

                    Note that there is no speculation here since it is a cross-component return.
                \end{description}
            \end{description}
        
        \item[$\src{f''} \notin \src{\OB{f''}}$ (\textbf{in the context})]

            By case analysis on $\src{f'} \in \src{\OB{f''}}$:
            \begin{description}
                \item[$\src{f'} \in \src{\OB{f''}}$ (\textbf{in the compiled component})]
                
                This switch from context to component arises in two cases:
            
                \begin{description}
                    \item[\textbf{call}]
                    
                    This is a call from a context function to a compiled function.
                    This is the dual of the case for call above.

                    \item[\textbf{return}]
                    
                    This is a return from a context function to a compiled function.
                    In this case we know that $\trgR{\SigmaR '' \specarrowR{\tau^{\taint}} \SigmaR '}$ was derived by \Cref{tr:v5-nospec-act} combined with \trgR{\Cref{tr:ret}} and we know $\trgR{\tau^{\taint}} = \trgR{\retObs{}!^{\safeta}}$.
                    So $\trgR{\SigmaR '} = \trgR{(\OB{C''}, \OB{B'}, \tup{p, M', A'}, \bot, \safeta)}$, where 
                    $\trgR{\OB{B''}} = \trgR{\OB{B'} \cdot 0 \mapsto l : \safeta}$ and
                    $\trgR{M'} = \trgR{M''}$ and $\trgR{A'} = \trgR{A''[\pc \mapsto l]}$.

                    This case is analogous to the corresponding case in  \Thmref{thm:v4-bwd-sim-lfence}.
                    Note that here is no speculation since the $\pret$ instruction is part of the context.
                \end{description}

            \item[$\src{f'} \notin \src{\OB{f''}}$ (\textbf{in the context})]
            
            This holds by \Thmref{thm:v4-back-sim-bts-lfence}.
            
        \end{description}
        
    \end{description}
    \item[\Cref{tr:v5-silent}]
    Analogous to the corresponding case in \Thmref{thm:v4-bwd-sim-lfence}.
\end{description}

\end{proof}

\BREAK

\begin{proof}[Proof of \Thmref{thm:v5-fwd-sim-stm-lfence}]\proofref{}{v5-fwd-sim-stm-lfence}\hfill

    Since $\src{\Omega} \srel_{\src{\OB{f''}}}\ \trgR{\SigmaR}$ we know that $\src{p(A(\pc))} = \trgR{p(A(\pc))}$.
    We proceed by induction on $\nsarrow{}$:
    \begin{description}
        \item[\Cref{tr:skip}, \Cref{tr:assign}, \Cref{tr:load}, \Cref{tr:load-prv}, \Cref{tr:beqz-sat}, \Cref{tr:beqz-unsat}, \Cref{tr:callback}]

        Analogous to the corresponding cases in \Thmref{thm:v4-fwd-sim-stm-lfence}.

        \item[\Cref{tr:store}, \Cref{tr:store-prv}] 
    
        The only change to the state is $\src{A'} = \src{A[pc \mapsto \succes(\pc)}$ and $\src{M'} = \src{M[n \mapsto A(x)]}$ where $\src{\exprEval{A}{e}{n : \taint}}$.
    
        Furthermore, we know $\src{\tau^\sigma} = \src{\storeObs^{\safeta}}$.
        This case is analogous to the corresponding case in \Thmref{thm:sls-fwd-sim-stm-lfence}.
       
        \item[\Cref{tr:call-internal}]

            The only change to the state is $\src{A'} = \src{A[pc \mapsto \mathcal{F}(f)]}$ and $\src{B'} = \src{B \cdot (\Omega(\pc) + 1)}$.
       
            By \Thmref{thm:v4-fmaps-related} we have $\src{\mathcal{F}(f)} = \trgR{\mathcal{F}(f)}$.
            Furthermore, we know $\src{\tau^\safeta} = \varepsilon$.
            \begin{description}
                \item[$\trgR{\SigmaR \bigspecarrowR{\tra{^{\taint}}} \SigmaR'}$]
                
                Thus, we can use \Cref{tr:call-internal} to derive $\OmegaR \trgR{\nsarrow{\varepsilon}} \OmegaR'$ with $\trgR{A'} = \trgR{A[pc \mapsto \mathcal{F}{f}]}$ and $\trgR{B'} = \trgR{B \cdot (\OmegaR(\pc) +1)}$.
    
                Let $\Rsb' = \trgR{\Rsb \cdot (\OmegaR(\pc) +1)}$. We can now apply \Cref{tr:v5-call} on $\trgR{\SigmaR}$ and are finished.
                \item[$\src{\varepsilon} \tracerel \trgR{\varepsilon}$]
                Trivial.
    
                \item[$\src{\Omega'} \srelref_{\src{\OB{f''}}}\, \trgR{\SigmaR'}$]
                Follows from $\src{A'} \rrel \trgR{A'}$ which follows from $\src{A} \rrel \trgR{A}$.
                Furthermore, $\src{B'} \brel \trgR{B'}$ follows from $\src{B} \brel \trgR{B}$ and $\src{\Omega}  \srelref_{\src{\OB{f''}}} \trgR{\SigmaR}$.
                
            \end{description}   
        \end{description}

        \item[\Cref{tr:call}]
    
        Then we can use \Cref{tr:v5-call-att} since this is a cross-component call.

        The rest is analogous to the corresponding case in \Thmref{thm:v4-fwd-sim-stm-lfence}.

        \item[\Cref{tr:callback}]
        
        This case cannot arise as we do not step to a compiled instructions $\complfenceR{i'}$ in the target.
        
        \item[\Cref{tr:ret-internal}]
        
        The only change to the state is $\src{A'} = \src{A[pc \mapsto l]}$ and the removal of the return location $\src{l}$ from $\src{B}$.

        Furthermore, we know $\src{\tau^\sigma} = \src{\retObs^{\safeta}}$.
        
        We first do a case distinction if the $\trgR{\Rsb}$ is empty or not, followed by a case distinction based on the top value of the $\Rsb$ and the stored return address $B = B' \cdot l'$:
        \begin{description}
            \item[$\trgR{\Rsb} = \trgR{\emptyset}$]
           
                 \begin{description}
                    \item[$\trgR{\SigmaR \bigspecarrowR{\tra{^{\taint}}} \SigmaR'}$]
                        
                        Thus, we can use \Cref{tr:ret} to derive $\OmegaR \trgR{\nsarrow{\pret^{\safeta}}} \OmegaR'$ with $\trgR{A'} = \trgR{A[pc \mapsto l]}$ where $\trgR{B} = \trgR{B' \cdot l}$.
    
                        Since $\trgR{\Rsb} = \trgR{\emptyset}$, we can now apply \Cref{tr:v5-retE} on $\SigmaR$ and are finished.

                    \item[$\src{\retObs^{\safeta}} \tracerel \trgR{\retObs^{\taintpc \glb \safeta}}$]
                        From $\srel$ we have $\trgR{\taintpc} = \trgR{\safeta}$ and thus $\trgR{\taintpc \glb \safeta} = \trgR{\safeta}$.

                    \item[$\src{\Omega'} \srelref_{\src{\OB{f''}}}\, \trgR{\SigmaR'}$]
                        Follows from $\src{A'} \rrel \trgR{A'}$ which follows from $\src{A} \rrel \trgR{A}$ and $\src{B} \brel \trgR{B}$.
        
                        Furthermore, $\src{B'} \brel \trgR{B'}$ which follows from $\src{B} \brel \trgR{B}$.
                \end{description}

            \item[$\trgR{\Rsb} = \trgR{\Rsb' \cdot l}$ and $\trgR{l} = \trgR{l'}$]
                 \begin{description}
                    \item[$\trgR{\SigmaR \bigspecarrowR{\tra{^{\taint}}} \SigmaR'}$]
                        
                        Thus, we can use \Cref{tr:ret} to derive $\OmegaR \trgR{\nsarrow{\pret^{\safeta}}} \OmegaR'$ with $\trgR{A'} = \trgR{A[pc \mapsto l]}$ where $\trgR{B} = \trgR{B' \cdot l}$.
    
                        Since $\trgR{\Rsb} = \trgR{\Rsb' \cdot l}$ and $l = l'$ we can now apply \Cref{tr:v5-retS} on $\SigmaR$ and are finished.

                    \item[$\src{\retObs^{\safeta}} \tracerel \trgR{\retObs^{\taintpc \glb \safeta}}$]
                        From $\srel$ we have $\trgR{\taintpc} = \trgR{\safeta}$ and thus $\trgR{\taintpc \glb \safeta} = \trgR{\safeta}$.

                    \item[$\src{\Omega'} \srelref_{\src{\OB{f''}}}\, \trgR{\SigmaR'}$]
                        Follows from $\src{A'} \rrel \trgR{A'}$ which follows from $\src{A} \rrel \trgR{A}$ and $\src{B} \brel \trgR{B}$.
        
                        Furthermore, $\src{B'} \brel \trgR{B'}$ which follows from $\src{B} \brel \trgR{B}$.
                \end{description}

            \item[$\trgR{\Rsb} = \trgR{\Rsb' \cdot l}$ and $\trgR{l} \neq \trgR{l'}$] 

                This is the case of speculation and where we need to account for multiple steps in the target.
                
                Importantly, all targets in the $\Rsb$ are the instruction numbers after a $\callC$ instruction. Note that the $\complfenceR{\cdot}$ compiler adds an $\trgR{\pbarrier}$ instruction after every $\callC$ instruction and actually does not modify the $\retC$ instruction:
                    \begin{align*}
                         \complfenceR{l : \pcall{f}} &= \trgR{
                                                        \begin{aligned}[t]
                                                                &
                                                                l : \pcall{f} 
                                                                    \\
                                                                    &\
                                                                   l' :  \pbarrier
                                                                   ~\text{\Cref{tr:v5-barr-spec}}
                                                        \end{aligned}
                                                        }
                    \end{align*}
           Thus, every entry in the RSB $\trgR{\Rsb}$ points to a barrier instruction.

             \begin{description}
                \item[$\trgR{\SigmaR \bigspecarrowR{\tra{^{\taint}}} \SigmaR'}$]
                    
                    Thus, we can use \Cref{tr:ret} to derive $\OmegaR \trgR{\nsarrow{\pret^{\safeta}}} \OmegaR'$ with $\trgR{A'} = \trgR{A[pc \mapsto l]}$ where $\trgR{B} = \trgR{B' \cdot l}$.
    
                    This is used in the step \Cref{tr:v5-spec}. 

                    Since we mispredict and return to the address stored in the $\RSB$ and we know that this points to a barrier instruction.
                    
                    Thus, the next instruction is the $\pbarrier$ and after \Cref{tr:v5-barr}, a rollback is triggered via \Cref{tr:v5-rollback}.
    
                    Thus, we have the execution  $\SigmaR \bigspecarrowR{\retObs^{\taintpc \glb \safeta} \cdot \rollbackObsR} \SigmaR'$, where $\SigmaR' = \trgR{\tup{\OB{C}, \OB{B'}, \tup{p, M, A'}, \bot, \safeta}}$

                \item[$\src{\retObs^{\safeta}} \tracerel \trgR{\retObs^{\taintpc \glb \safeta} \cdot \rollbackObsR }$]
                    From $\srel$ we have $\trgR{\taintpc} = \trgR{\safeta}$ and thus $\trgR{\taintpc \glb \safeta} = \trgR{\safeta}$.

                    Next, from \Cref{tr:ac-rel-rlb} we get $\src{\epsilon} \arel \rollbackObsR$  and are finished.
    
                \item[$\src{\Omega'} \srelref_{\src{\OB{f''}}}\, \trgR{\SigmaR'}$]
                    Follows from $\src{A'} \rrel \trgR{A'}$ which follows from $\src{A} \rrel \trgR{A}$ and $\src{B} \brel \trgR{B}$.
    
                    Furthermore, $\src{B'} \brel \trgR{B'}$ which follows from $\src{B} \brel \trgR{B}$.
            \end{description}

        \item[\Cref{tr:ret}]
        Thus, there is a context switch on which we do not speculate.
        This case is analogous to the corresponding case in \Thmref{thm:v4-fwd-sim-stm-lfence}.

        \item[\Cref{tr:retback}]
        This case cannot arise as we do not step to a compiled instructions $\complfenceR{i'}$ in the target.
    \end{description}
    
\end{proof}

\BREAK

\begin{proof}[Proof of \Thmref{thm:v5-bwd-sim-comp-steps-lfence}]\proofref{}{v5-bwd-sim-comp-steps-lfence}\hfill

Analogous to \Thmref{thm:v4-bwd-sim-comp-steps-lfence} using \Thmref{thm:v5-fwd-sim-stm-lfence}.
\end{proof}
 \section{$\Rv$: THE RETPOLINE COMPILER FOR \specr}\label{comp:v5-retpoline}

The modified retpoline of \citet{ret2spec} works in a similar way to the original retpoline \Cref{comp:v2-retpoline}.
Instead of targeting indirect jumps as in \Cref{comp:v2-retpoline}, here the $\pret$ instructions are the target:

\begin{center}
\begin{minipage}[b]{0.45\linewidth}
\begin{lstlisting}[basicstyle=\small,style=MUASMstyle,  escapechar=|, captionpos=t,]
ret
\end{lstlisting}
\end{minipage}
\hfill
\begin{minipage}[b]{0.45\linewidth}
\begin{lstlisting}[basicstyle=\small,style=MUASMstyle, escapechar=|, captionpos=t,]
l_0 : call Retpo   
l_1 : skip
l_2 : spbarr
l_3 : jmp l_1
Retpo:
l_4 : popret
l_5 : ret
\end{lstlisting}
\end{minipage}
\captionof{lstlisting}{Code on the left and the compiled program with the modified retpoline countermeasure on the right.}    
\end{center}
Here $\popretC$ is used to recover the original return address, by popping the return address of the call $Retpo$ from the stack. Thus, the last return will use the original return address.
Here, we only have one retpoline function for all returns, since there are no different register that could be used for indirect jumps as for \Cref{comp:v2-retpoline}.

\begin{align*}
    \compretpR{ M ; \OB{F} ; \OB{I}} &= \trgR{ \compretpR{M} ; \compretpR{\OB{F}} \cdot (Retpo, l'_4) ; \compretpR{\OB{I}}}
	\\\\
	\compretpR{ \srce} &= \trgR{\emptyset}
	\\
	\compretpR{ \OB{I}\cdot f } &= \trgR{ \compretpR{ \OB{I} }\cdot \compretpR{f} }
	\\\\
	\compretpR{ M ; -n\mapsto v : \unta} &= \trgR{\compretpR{M} ; -\compretpR{n} \mapsto \compretpR{v} : \unta}
	\\\\
	\compretpR{p_1;p_2} &= \trgR{\compretpR{p_1} ; \compretpR{p_2}}
	\\
        \compretpR{l : i} &= \trgR{\compretpR{l} : \compretpR{i}} ~~ \text{if $i \neq \retC$}
	\\\\
	\compretpR{\pskip} &= \trgR{\pskip} 
	\\
	\compretpR{\passign{x}{e}} &= \trgR{\passign{x}{e}}
	\\
        \compretpR{\pcondassign{x}{e}{e'}} &= \trgR{\pcondassign{x}{\compretpR{e}}{\compretpR{e'}}}
        \\
        \compretpR{\ploadprv{x}{e}} &= \trgR{\ploadprv{x}{\compretpR{e}}}
	\\
	\compretpR{\pload{x}{e}} &= \trgR{\pload{x}{\compretpR{e}}}
	\\
        \compretpR{\pstore{x}{e}} &= \trgR{\pstore{x}{\compretpR{e}}}
        \\
        \compretpR{\pstoreprv{x}{e}} &= \trgR{\pstoreprv{x}{\compretpR{e}}}
	\\
	\compretpR{\pjmp{e}} &= \trgR{\pjmp{\compretpR{e}}}
	\\
	\compretpR{\pjz{x}{l}} &= \trgR{\pjz{x}{l}}
	\\
	\compretpR{\pcall{f}} &= \trgR{\pcall{f}}
	\\
	\compretpR{l : \pret} &= \trgR{
                                         \begin{aligned}[t]
                                                    &
                                                    l': \pcall{Retpo}
                                                        \\
                                                        &\
                                                        l'_1: \pskip
                                                            \\
                                                            &\
                                                            l'_2 : \barrierKywd
                                                            \\
                                                            &\
                                                            l'_3 : \pjmp{l'_1}
                                                            \\
                                                            &\
                                                            l'_4 : \ppopret
                                                            \\
                                                            &\
                                                            l'_5: \pret
                                            \end{aligned}
    }
\end{align*}

We model the thunk-inline version of the countermeasure, which inlines the biggest part of the retpoline and adds one call.
We do not model the thunk version because that uses a jump instruction to jump to the retpoline construct. However, our jumps are not allowed to jump out of the function.

Notice that we add the Retpo function to $\OB{F}$ during compilation.

\subsection{The Proofs}

We need to be especially careful with the trace relation $\tracerel$. Look at the possible trace generated before and after compilation for a return instruction:

\begin{tikzpicture}
    \node[circle, draw] (start) {call};
    \node[draw, above=0.5cm of start]{Original Program};

    \node[circle, draw, right= of start] (0) {$\cdots$};
    \node[circle, draw, right= of 0] (1) {ret};
    \node[circle, draw, right= of 1] (2) {$\cdots$};

    \draw[->] (start) edge [bend left] node[above] {call! origfunc} (0) ;
    \draw[->] (0) edge [bend left] node[above] {$\tau$} (1) ;
    \draw[->] (1) edge [bend left] node[above] {ret} (2) ;

    \node[circle, draw, below=2cm of start](mod){call};
    \node[draw, above=0.5cm of mod]{Compiled Program};

    \node[circle, draw, right= of mod] (0m) {$\cdots$};
    \node[red, circle, draw, right= of 0m] (1m) {call};
    \node[circle, draw, right= of 1m] (3m) {$\cdots$};
    \node[circle, draw, right= of 3m] (4m) {popret};
    \node[circle, draw, right= of 4m] (5m) {ret};
    \node[circle, draw, right= of 5m] (6m) {$\cdots$};

    \draw[->] (mod) edge [bend left] node[above] {call! origfunc} (0m) ;
    \draw[->] (0m) edge [bend left] node[above] {$\tau$} (1m) ;
    \draw[->] (1m) edge [bend left] node[above] {call retpo} (3m);
    \draw[->] (3m) edge [bend left] node[above] {} (4m);
    \draw[->] (4m) edge [bend left] node[above] {} (5m);
    \draw[->] (5m) edge [bend left] node[above] {ret} (6m);
    
\end{tikzpicture}

Note that only components are compiled. This means that \textit{call retpo} is an internal function calls which are not visible on the trace. Thus, we can match the traces again.

\begin{center}
	
	\scalebox{0.7}{
\begin{forest}
for tree={
    align = left,
    font = \footnotesize,
    forked edge,
}
[\Thmref{thm:v5-all-retpo-comp-are-rdss}
        [\Thmref{thm:v5-retpo-comp-rdss}
                [\Thmref{thm:v5-corr-bt-retpo}
                    [\Thmref{thm:v5-ini-state-rel}
                        [\Cref{thm:v4-heap-rel-comp-lfence}]
                        [\Cref{thm:v4-heap-rel-bt-lfence}, name=heap-bt]
                    ]
                    [\Thmref{thm:v5-bwd-sim-retpo}
                            [\Thmref{thm:v5-bwd-sim-comp-steps-retpo}
                                    [\Thmref{thm:v5-fwd-sim-stm-retpo}
                                        [\Thmref{thm:v4-fwd-sim-exp-lfence}]
                                    ]
                            ]
                            [\Thmref{thm:v4-back-sim-bts-lfence}   
                                [\Thmref{thm:v4-back-sim-bte-lfence}, name=exp-bt]
                            ]
                    ]
                ]
        ]
]
\node (sniSource) at (-4,3) {\Thmref{thm:ss-sni-source}};
\node (rsscEq) [below=of sniSource] {\Thmref{thm:rdss-eq-rdsp}};
\draw[red, thick, dotted] ($(sniSource.north west)+(-0.3,0.3)$)  rectangle ($(rsscEq.south east)+(0.9,-0.6)$);
\end{forest}
	}
\end{center}

\begin{theorem}[$\Rv$: The retpoline compiler is \rdss]\label{thm:v5-retpo-comp-rdss}(\showproof{v5-retpo-comp-rdss})
	\begin{align*}
		\contractSpec{}{\Rv} \vdash \compretpR{\cdot} : \rssp
	\end{align*}
\end{theorem}

\begin{theorem}[$\Rv$: All retpoline-compiled programs are \rss)]\label{thm:v5-all-retpo-comp-are-rdss}(\showproof{v5-all-retpo-comp-are-rdss})
	\begin{align*}
		\forall\src{P}\ldotp \contractSpec{}{\Rv} \vdash\compretpR{P} : \rss
	\end{align*}
\end{theorem}

The relation is the same as the one in \Cref{comp:v5-fence}.

\begin{theorem}[$\Rv$: Correctness of the Backtranslation]\label{thm:v5-corr-bt-retpo}(\showproof{v5-corr-bt-retpo})
	\begin{align*}
		\text{ if } 
			&\
			\trgR{ \amTracevR{\ctxR{}\hole{\compretpR{P}}}{\tra{^{\taint}}} }
		\\
		\text{ then }
			&
			\src{ \Trace{\backtrfencec{\ctxR{}}\hole{P}}{ \tras{^{\sigma}}} }
		\\
		\text{ and }
			&\
			\tras{^{\sigma}} \rels \trgR{\tra{^{\taint}}}
	\end{align*}
\end{theorem}

\begin{theorem}[$\Rv$: Generalised Backward Simulation]\label{thm:v5-bwd-sim-retpo}(\showproof{v5-bwd-sim-retpo})
	\begin{align*}
		\text{ if }
			&\
			\ffun{\src{A(\pc)}}\in\src{\OB{f''}} \text{ then } \trgR{A(\pc)} : \trgR{i} = \compretpR{l : i} \text{ or } (\compretpR{A(\pc) : i} = \trgR{A(\pc)} : \trgR{i} ; \trgR{p_c'}) 
			\\
			\text{ else } 
			&\ 
			\src{A(\pc)} : \src{i} = \backtrfencec{\trgR{A(\pc)}} : \backtrfencec{\trgR{i}} 
			\\
		\text{ and}
			&\
			\text{ if } \ffun{\src{A'(\pc)}}\in\src{\OB{f''}} \text{ then } \trgR{A'(\pc)} : \trgR{i'} = \compretpR{A'(\pc) : i'}
			\\
			\text{ or }
			&\  
			(\compretpR{A'(\pc) : i'} = \trgR{A'(\pc)} : \trgR{i'} ; \trgR{p_c''})
			\text{ else } \src{A'(\pc)} : \src{i'} = \backtrfencec{\trgR{A'(\pc)}} : \backtrfencec{\trgR{i'}}
			\\
		\text{ and }
			&\
			\SigmaR = \trgR{(C, \OB{B}, \tup{p, M, A}, \bot,\safeta)}
		\\
		\text{ and }
			&\
			\SigmaRt{'}=\trgR{(C, \OB{B'}, \tup{p, M', A'},\bot,\safeta)}
		\\
		\text{ and }
			&\
			\trgR{ \SigmaR \bigspecarrowR{\tra{^{\taint}}} \SigmaR' }
		\\
		\text{ and }
			&\
			\src{\Omega} \srel_{\src{\OB{f''}}} \trgR{\SigmaR}
		\\
        \text{ and }
			&\
			\trgR{p(A(\pc))} = \trgR{i}  \text{ and } \trgR{p(A'(\pc))} = \trgR{i'}
		\\
        \text{ and }
			&\
			\src{p(A(\pc))} =  \src{i} \text{ and } \src{p(A'(\pc))} = \src{i'}
        \\
		\text{ then }
			&\
			\src{\Omega}=\src{C, \OB{B}, \tup{p , M, A}  \nsbigarrow{\tras{^\sigma}} C, \OB{B'}, \tup{p , M', A'} }=\src{\Omega'}
		\\
		\text{ and }
			&\
			\src{\tras{^\sigma}} \tracerel \trgR{\tra{^{\taint}}}
		\\
		\text{ and }
			&\
			\src{\Omega'} \srelref_{\src{\OB{f''}}}\, \trgR{\SigmaR{'}}
	\end{align*}
\end{theorem}

\begin{theorem}[$\Rv$: Backward Simulation for Compiled Steps]\label{thm:v5-bwd-sim-comp-steps-retpo}(\showproof{v5-bwd-sim-comp-steps-retpo})
	\begin{align*}
		\text{ if }
			&\
			\trgR{\SigmaR}=\trgR{(C, \OB{B}, \tup{p, M, A}, \bot, \safeta) }
		\\
		\text{ and }
			&\
			\trgR{\SigmaR'}=\trgR{(C, \OB{B'}, \tup{p, M', A'}, \bot, \safeta)}
		\\
		\text{ and }
			&\
			\trgR{\SigmaR \bigspecarrowR{\tra{^{\taint}}} \SigmaR'}
		\\
		\text{ and }
			&\
			\src{\Omega} \srelref_{\src{\OB{f''}}}\, \SigmaR
		\\
        \text{ and }
			&\
			\trgR{p(A(\pc))} = \compretpR{l} : \compretpR{i} \text{ and } \trgR{p(A'(\pc))} = \compretpR{l'} : \compretpR{i'}
		\\
         \text{ and }
            &\
            \src{p(A(\pc))} =  \src{i} \text{ and } \src{p(A'(\pc))} =  \src{i'}
            \\
        \text{ and }
			&\
			\trgR{A(\pc) : i} = \compretpR{A(\pc) : i}   \text{ or } (\compretpR{A(\pc) : i} = \trgR{A(\pc)} : \trgR{i} ; \trgR{p_c'})
            \\
            \text{ and }
                &\
                \trgR{A'(\pc) : i'} = \compretpR{A'(\pc) : i'}  \text{ or } (\compretpR{A'(\pc) : i'} = \trgR{A'(\pc)} : \trgR{i'} ; \trgR{p_c''})
		\\
		\text{ then }
			&\
			\src{\Omega}=\src{C, \OB{B}, \tup{p, M, A} \nsarrow{\tau^\sigma} C, \OB{B'}, \tup{p, M', A'}}= \src{\Omega'}
        \\
		\text{ and }
			&\
			\src{\tau^\sigma} \tracerel \trgR{\tra{^{\taint}}} \qquad \text{( using the trace relation!)}
		\\
		\text{ and }
			&\
			\src{\Omega'} \srelref_{\src{\OB{f''}}}\, \SigmaR'
	\end{align*}
\end{theorem}

\begin{lemma}[$\Rv$: Forward Simulation for Compiled Statements]\label{thm:v5-fwd-sim-stm-retpo}(\showproof{v5-fwd-sim-stm-retpo})
	\begin{align*}
		\text{ if }
			&\
			\src{\Omega}=\src{C, \OB{B}, \tup{p, M, A} \nsarrow{\tau^\sigma} C, \OB{B'}, \tup{p, M', A'} } =\src{\Omega'}
		\\
		\text{ and }
			&\
			\src{\Omega} \srelref_{\src{\OB{f''}}} \trgR{\SigmaR}
		\\
        \text{ and }
			&\
			\src{p(A(\pc))} = \src{i} \text{ and } \src{p(A'(\pc))} = \src{i'}
		\\
		\text{ and }
			&\
			\trgR{\SigmaR}= \trgR{C, \OB{B}, \tup{p, M , A}, \bot, \safeta}, 
		\\
		\text{ and }
			&\
			\trgR{\SigmaR'}= \trgR{C, \OB{B}, \tup{p, M' , A'}, \bot, \safeta}
		\\
        \text{ and }
			&\
			\trgR{A(\pc) : i} = \compretpR{A(\pc) : i}   \text{ or } (\compretpR{A(\pc) : i} = \trgR{A(\pc)} : \trgR{i} ; \trgR{p_c'})
            \\
            \text{ and }
                &\
                \trgR{A'(\pc) : i'} = \compretpR{A'(\pc) : i'}  \text{ or } (\compretpR{A'(\pc) : i'} = \trgR{A'(\pc)} : \trgR{i'} ; \trgR{p_c''})
		\\
		\text{ then }
			&\
			\trgR{\SigmaR \bigspecarrowR{\tra{^{\taint}}} \SigmaR'}
		\\
		\text{ and }
			&\
			\src{\tau^\sigma} \tracerel \trgR{\tra{^{\taint}}} \qquad \text{( using the trace relation!)}
		\\
		\text{ and }
			&\
			\src{\Omega'} \srelref_{\src{\OB{f''}}}\ \trgR{\SigmaR'}
	\end{align*}
\end{lemma}

 \begin{proof}[Proof of \Thmref{thm:v5-retpo-comp-rdss}]\proofref{}{v5-retpo-comp-rdss}\hfill

    Instantiate $\src{A}$ with $\backtrfencec{\trgR{A}}$.
    This holds by \Thmref{thm:v5-corr-bt-retpo}.
\end{proof}

\BREAK

\begin{proof}[Proof of \Thmref{thm:v5-all-retpo-comp-are-rdss}]\proofref{}{v5-all-retpo-comp-are-rdss}\hfill

    By \Thmref{thm:ss-sni-source} we know all source languages are SS: $\text{HPS:} \forall\src{P}\contract{}{NS}\ldotp\vdash P : \rss(\SR)$.
    
    Since \Thmref{thm:v5-retpo-comp-rdss} we have HPC: $\contractSpec{}{\Rv}\vdash \compretpR{\cdot} : \rdss$.

    By \Thmref{thm:rdss-eq-rdsp}  with HPC we have $\text{HPP} \contractSpec{}{\Rv}\vdash \compretpR{\cdot} : \rdssp$

    Unfolding \Thmref{def:rdsspc} of HPP we get 
    $\forall\src{P}\ldotp \text{ if } \contract{}{NS}\vdash\src{P} : \rss \text{ then } \contractSpec{}{\Rv} \vdash\comp{\src{P}} : \rss$. 
  
    Now, we instantiate it with HPS and can conclude that $\forall\src{P}\ldotp \contractSpec{}{\Rv}\vdash\compretpR{P} : \rss$.
    
\end{proof}

\BREAK 

\begin{proof}[Proof of \Thmref{thm:v5-corr-bt-retpo}]\proofref{}{v5-corr-bt-retpo}\hfill

This holds by \Thmref{thm:v5-ini-state-rel} and \Thmref{thm:v5-bwd-sim-retpo}.
\end{proof}

\BREAK

\begin{proof}[Proof of \Thmref{thm:v5-bwd-sim-retpo}]\proofref{}{v5-bwd-sim-retpo}\hfill

We proceed by induction on $\SigmaR \bigspecarrowR{\tra{^{\taint}}} \SigmaR '$:
\begin{description}
    \item[\Cref{tr:v5-reflect}]
    Trivial.
    \item[\Cref{tr:v5-single}]
    We have 
    \begin{align*}
        \trgR{\SigmaR \bigspecarrowR{\OB{\tau^{\taint}}} \SigmaR ''} \\
        \trgR{\SigmaR '' \specarrowR{\tau^{\taint}} \SigmaR '} \\ 
        \trgR{\SigmaR ''} = \trgR{\OB{C''}, \OB{B''}, \tup{p ,M'', A''}, \bot, \safeta}
    \end{align*}

    Applying IH on $\trgR{\SigmaR  \bigspecarrowR{\OB{\tau^{\taint''}}} \SigmaR ''}$ we get: 
    \begin{enumerate}
        \item $\src{\Omega  \nsbigarrow{\tras{^{\sigma''}}} \Omega''}$ and
        \item $\src{\tras{^{\sigma''}}} \tracerel \trgR{\tra{^{\taint''}}}$ and
        \item $\src{\Omega''} \srelref_{\src{\OB{f''}}}\, \trgR{\SigmaR {''}}$
    \end{enumerate}
    Let $\ffun{\trgR{l''}} = \trgR{f''}$ and let $\ffun{\trgR{l'}} = \trgR{f'}$ (where $\trgR{\SigmaR ' . A(\pc)} = \trgR{l'}$).
    We proceed by case analysis on $\src{f''} \in \src{\OB{f''}}$:
    \begin{description}
        \item[$\src{f''} \in \src{\OB{f''}}$ (\textbf{in the compiled component})]
            By case analysis on $\src{f'} \in \src{\OB{f''}}$:
            \begin{description}
                \item[$\src{f'} \in \src{\OB{f''}}$ (\textbf{in the compiled component})]

                This holds by \Thmref{thm:v5-bwd-sim-comp-steps-lfence}.
                
                \item[$\src{f'} \notin \src{\OB{f''}}$ (\textbf{in the context})]
                This switch from component to context arises in two cases:
                \begin{description}
                    \item[\textbf{call}]
                    Analogous to the corresponding case in \Thmref{thm:v5-bwd-sim-lfence}.

                    \item[\textbf{return}]
                    This is a return from a compiled function to a context function.
                    This is the dual of the case below for return from context to component.
                    
                    Note that there is no speculation here since it is a cross-component return.
                \end{description}
            \end{description}
        
        \item[$\src{f''} \notin \src{\OB{f''}}$ (\textbf{in the context})]

            By case analysis on $\src{f'} \in \src{\OB{f''}}$:
            \begin{description}
                \item[$\src{f'} \in \src{\OB{f''}}$ (\textbf{in the compiled component})]
                
                This switch from context to component arises in two cases:
            
                \begin{description}
                    \item[\textbf{call}]
                    
                    This is a call from a context function to a compiled function.
                    This is the dual of the case for call above.

                    \item[\textbf{return}]
                    
                    Note that there is no speculation since the $\pret$ instruction is part of the context.

                    Analogous to the corresponding case in \Thmref{thm:v5-bwd-sim-lfence}.
                \end{description}

            \item[$\src{f'} \notin \src{\OB{f''}}$ (\textbf{in the context})]
            
            This holds by \Thmref{thm:v4-back-sim-bts-lfence}.
            
        \end{description}
        
    \end{description}
    \item[\Cref{tr:v5-silent}]
    Analogous to the corresponding case in \Thmref{thm:v4-bwd-sim-lfence}.
\end{description}

\end{proof}

\BREAK

\begin{proof}[Proof of \Thmref{thm:v5-fwd-sim-stm-retpo}]\proofref{}{v5-fwd-sim-stm-retpo}\hfill

    Since $\src{\Omega} \srel_{\src{\OB{f''}}}\ \trgR{\SigmaR}$ we know that $\src{p(A(\pc))} = \trgR{p(A(\pc))}$.
    We proceed by induction on $\nsarrow{}$:
    \begin{description}
        \item[\Cref{tr:skip}, \Cref{tr:assign}, \Cref{tr:load}, \Cref{tr:store}, \Cref{tr:store-prv}, \Cref{tr:load-prv}, \Cref{tr:beqz-sat}, \Cref{tr:beqz-unsat}, \Cref{tr:callback}, \Cref{tr:call}, \Cref{tr:call-internal}]
        Analogous to the corresponding cases in \Thmref{thm:v5-fwd-sim-stm-lfence}.
        
        \item[\Cref{tr:ret}] 

        The only change to the state is $\src{A'} = \src{A[pc \mapsto l]}$ and the removal of the return location $\src{l}$ from $\src{B}$.

        Furthermore, we know $\src{\tau^\sigma} = \src{\retObs^{\safeta}}$.
        
        We first do a case distinction if the $\trgR{\Rsb}$ is empty or not, followed by a case distinction based on the top value of the $\Rsb$ and the stored return address $B = B' \cdot l'$:
        \begin{description}
            \item[$\trgR{\Rsb} = \trgR{\emptyset}$]
                Analogous to the corresponding case in \Thmref{thm:v5-fwd-sim-stm-lfence}

            \item[$\trgR{\Rsb} = \trgR{\Rsb' \cdot l}$ and $\trgR{l} = \trgR{l'}$]
                 Analogous to the corresponding case in \Thmref{thm:v5-fwd-sim-stm-lfence}

            \item[$\trgR{\Rsb} = \trgR{\Rsb' \cdot l}$ and $\trgR{l} \neq \trgR{l'}$] 
                This is the case of speculation and where we need to account for multiple steps in the target.
                
                Let's look at the code sequence that is generated by our retpoline compiler $\compretpR{\cdot}$:
                \begin{align*}
                    \trgR{
                        \begin{aligned}[t]
                            &
                            l': \pcall{Retpo}    ~~\text{\Cref{tr:v5-call}}
                            \\
                            &\
                            l'_1: \pskip ~~\text{\Cref{tr:v5-barr-spec}}
                            \\
                            &\
                            l'_2 : \barrierKywd ~~\text{\Cref{tr:v5-rollback}}
                            \\
                            &\
                            l'_3 : \pjmp{l'_1} ~~\text{\Cref{tr:v5-rollback}}
                            \\
                            &\
                            Retpo: 
                            \\
                            &\
                            l'_4 : \ppopret ~~\text{\Cref{tr:v5-nospec-eps}}
                            \\
                            &\
                            l'_5: \pret ~~\text{\Cref{tr:v5-spec}}
                        \end{aligned}
                        }
                \end{align*}
            We annotated each line with the rule that is used. Not the non-linear control flow here.
            Instructions executed at the labels are in order: $l'$, $l'_4$, $l'_5$, $l'_1$, $l'_2$ and $l'_3$.

             \begin{description}
                \item[$\trgR{\SigmaR \bigspecarrowR{\tra{^{\taint}}} \SigmaR'}$]
                    
                    Thus, we can use \Cref{tr:ret} to derive $\OmegaR \trgR{\nsarrow{\pret^{\safeta}}} \OmegaR'$ with $\trgR{A'} = \trgR{A[pc \mapsto l]}$ where $\trgR{B} = \trgR{B' \cdot l}$.
    
                    This is used in the step \Cref{tr:v5-spec} as well. 

                    And at $l'_3$ a rollback is triggered because of the previous barrier.
    
                    Thus, we have the execution  $\SigmaR \bigspecarrowR{\retObs^{\taintpc \glb \safeta} \cdot \rollbackObsR} \SigmaR'$, where $\SigmaR' = \trgR{\tup{\OB{C}, \OB{B'}, \tup{p, M, A'}, \bot, \safeta}}$

                \item[$\src{\retObs^{\safeta}} \tracerel \trgR{\retObs^{\taintpc \glb \safeta} \cdot \rollbackObsR }$]
                    From $\srel$ we have $\trgR{\taintpc} = \trgR{\safeta}$ and thus $\trgR{\taintpc \glb \safeta} = \trgR{\safeta}$.

                    Next, from \Cref{tr:ac-rel-rlb} we get $\src{\epsilon} \arel \rollbackObsR$  and are finished.
    
                \item[$\src{\Omega'} \srelref_{\src{\OB{f''}}}\, \trgR{\SigmaR'}$]
                    Follows from $\src{A'} \rrel \trgR{A'}$ which follows from $\src{A} \rrel \trgR{A}$ and $\src{B} \brel \trgR{B}$.
    
                    Furthermore, $\src{B'} \brel \trgR{B'}$ which follows from $\src{B} \brel \trgR{B}$.
            \end{description}

        \end{description}
        
        \item[\Cref{tr:ret-internal}]
        This case is analogous to case \textbf{ret} above.

        \item[\Cref{tr:retback}]
        This case cannot arise as we do not step to a compiled instructions $\complfenceR{i'}$ in the target.

        \end{description}
\end{proof}

\BREAK

\begin{proof}[Proof of \Thmref{thm:v5-bwd-sim-comp-steps-retpo}]\proofref{}{v5-bwd-sim-comp-steps-retpo}\hfill

Analogous to \Thmref{thm:v4-bwd-sim-comp-steps-lfence} using \Thmref{thm:v5-fwd-sim-stm-retpo}.
\end{proof}
 \section{$\SLSv$: THE SLS Compiler for Straighline-Speculation}

A compiler against Straightline speculation inserts an int3 instruction, which is an interrupt and stops speculation. We model this by fence instruction. The interrupt instruction is used because it is only 1 byte and not 3 like for the fence.
These instructions are architecturally never executed, because the control flow diverted before, so they have no additional cost except for an increase in binary size.

Most interestingly, this kind of speculation depends on the program's layout because it executes the instructions following a return.

Again, this is a vulnerability only made visible by our assembly-like semantics. In higher-level language models, there is nothing after a return. So it is hard to model the vulnerability and the countermeasure.

\begin{align*}
    \complfenceSLS{ M ; \OB{F} ; \OB{I}} &= \trgSLS{ \complfenceSLS{M} ; \complfenceSLS{\OB{F}} ; \complfenceSLS{\OB{I}}}
	\\\\
	\complfenceSLS{ \srce} &= \trgSLS{\emptyset}
	\\
	\complfenceSLS{ \OB{I}\cdot f } &= \trgSLS{ \complfenceSLS{ \OB{I} }\cdot \complfenceSLS{f} }
	\\\\
	\complfenceSLS{ M ; -n\mapsto v : \unta} &= \trgSLS{\complfenceSLS{M} ; -\complfenceSLS{n} \mapsto \complfenceSLS{v} : \unta}
	\\\\
	\complfenceSLS{p_1;p_2} &= \trgSLS{\complfenceSLS{p_1} ; \complfenceSLS{p_2}}
	\\
	\complfenceSLS{l : i} &= \trgSLS{\complfenceSLS{l} : \complfenceSLS{i}} ~~ \text{if $i \neq \retC$}
	\\\\
	\complfenceSLS{\pskip} &= \trgSLS{\pskip} 
	\\
	\complfenceSLS{\passign{x}{e}} &= \trgSLS{\passign{x}{\complfenceSLS{e}}}
        \\
        \complfenceSLS{\pcondassign{x}{e}{e'}} &= \trgSLS{\pcondassign{x}{\complfenceSLS{e}}{\complfenceSLS{e'}}}
        \\
        \complfenceSLS{\ploadprv{x}{e}} &= \trgSLS{\ploadprv{x}{\complfenceSLS{e}}}
	\\
	\complfenceSLS{\pload{x}{e}} &= \trgSLS{\pload{x}{\complfenceSLS{e}}}
	\\
        \complfenceSLS{\pstore{x}{e}} &= \trgSLS{\pstore{x}{\complfenceSLS{e}}}
        \\
        \complfenceSLS{\pstoreprv{x}{e}} &= \trgSLS{\pstoreprv{x}{\complfenceSLS{e}}}
        \\
	\complfenceSLS{\pjz{x}{l'}} &= \trgSLS{\pjz{x}{\complfenceSLS{l'}}} 
	\\
	\complfenceSLS{\pjmp{e}} &= \trgSLS{\pjmp{\complfenceSLS{e}}} 
	\\
	\complfenceSLS{\pcall{f}} &= \trgSLS{\pcall{f}} ~ \text{Here $f$ is a function name}
	\\\\
	\complfenceSLS{l : \pret} &= \trgSLS{l : \pret}  ; l_{new} : \trgSLS{\pbarrier} ~ \text{ where $\l \subseteq l_{new}$ and $l_{new} \leq next(l)$} 
\end{align*}

\subsection{The proofs}

These proofs are very similar to \Cref{comp:v4-fence}.

\begin{center}
	
	\scalebox{0.7}{
\begin{forest}
for tree={
    align = left,
    font = \footnotesize,
    forked edge,
}
[\Thmref{thm:sls-all-lfence-comp-are-rdss}
        [\Thmref{thm:sls-lfence-comp-rdss}
                [\Thmref{thm:sls-corr-bt-lfence}
                    [\Thmref{thm:sls-ini-state-rel}
                        [\Cref{thm:v4-heap-rel-comp-lfence}]
                        [\Cref{thm:v4-heap-rel-bt-lfence}, name=heap-bt]
                    ]
                    [\Thmref{thm:sls-bwd-sim-lfence}
                            [\Thmref{thm:sls-bwd-sim-comp-steps-lfence}
                                    [\Thmref{thm:sls-fwd-sim-stm-lfence}
                                        [\Thmref{thm:v4-fwd-sim-exp-lfence}]
                                    ]
                            ]
                            [\Thmref{thm:v4-back-sim-bts-lfence}   
                                [\Thmref{thm:v4-back-sim-bte-lfence}, name=exp-bt]
                            ]
                    ]
                ]
        ]
]
\node (sniSource) at (-4,3) {\Thmref{thm:ss-sni-source}};
\node (rsscEq) [below=of sniSource] {\Thmref{thm:rdss-eq-rdsp}};
\draw[red, thick, dotted] ($(sniSource.north west)+(-0.3,0.3)$)  rectangle ($(rsscEq.south east)+(0.9,-0.6)$);
\end{forest}
	}
\end{center}

\begin{theorem}[$\SLSv$: The Fence compiler is \rdss]\label{thm:sls-lfence-comp-rdss}(\showproof{sls-lfence-comp-rdss})
	\begin{align*}
		\contractSpec{}{\SLSv} \vdash \complfenceSLS{\cdot} : \rdss
	\end{align*}
\end{theorem}

\begin{theorem}[$\SLSv$: All SLS-compiled programs are \rss]\label{thm:sls-all-lfence-comp-are-rdss}(\showproof{sls-all-lfence-comp-are-rdss})
	\begin{align*}
		\forall\src{P}\ldotp \contractSpec{}{\SLSv}\vdash\complfenceSLS{P} : \rss
	\end{align*}
\end{theorem}

The relation is the same as the one in \Cref{comp:v4-fence}.

\begin{theorem}[$\SLSv$: Correctness of the Backtranslation]\label{thm:sls-corr-bt-lfence}(\showproof{sls-corr-bt-lfence})
	\begin{align*}
		\text{ if } 
			&\
			\trgSLS{ \amTracevSLS{\ctxSLS{}\hole{\complfenceSLS{P}}}{\tra{^{\taint}}} }
		\\
		\text{ then }
			&
			\src{ \Trace{\backtrfencec{\ctxSLS{}}\hole{P}}{ \tras{^{\sigma}}} }
		\\
		\text{ and }
			&\
			\tras{^{\sigma}} \rels \trgSLS{\tra{^{\taint}}}
	\end{align*}
\end{theorem}

\begin{theorem}[$\SLSv$:Generalised Backward Simulation]\label{thm:sls-bwd-sim-lfence}(\showproof{sls-bwd-sim-lfence})
	\begin{align*}
		\text{ if }
			&\
			\ffun{\src{A(\pc)}}\in\src{\OB{f''}} \text{ then } \trgSLS{A(\pc)} : \trgSLS{i} = \complfenceSLS{l : i} \text{ or } (\complfenceSLS{A(\pc) : i} = \trgSLS{A(\pc)} : \trgSLS{i} ; \trgSLS{p_c'}) 
			\\
		\text{ else }
			&\ 
			\src{A(\pc)} : \src{i} = \backtrfencec{\trgSLS{A(\pc)}} : \backtrfencec{\trgSLS{i}} 
			\\
		\text{ and}
			&\
			\text{ if } \ffun{\src{A'(\pc)}}\in\src{\OB{f''}} \text{ then } \trgSLS{A'(\pc)} : \trgSLS{i'} = \complfenceSLS{A'(\pc) : i'} \text{ or } (\complfenceSLS{A'(\pc) : i'} = \trgSLS{A'(\pc)} : \trgSLS{i'} ; \trgSLS{p_c''}) 
			\\
		\text{ else }
			&\
			\src{A'(\pc)} : \src{i'} = \backtrfencec{\trgSLS{A'(\pc)}} : \backtrfencec{\trgSLS{i'}}
		\\
		\text{ and }
			&\
			\trgSLS{\SigmaSLS} = \trgSLS{(C, \OB{B}, \tup{p, M, A}, \bot,\safeta)}
		\\
		\text{ and }
			&\
			\SigmaSLSt{'}=\trgSLS{(C, \OB{B'}, \tup{p, M', A'},\bot,\safeta)}
		\\
		\text{ and }
			&\
			\trg{ \SigmaSLS \bigspecarrowSLS{\tra{^{\taint}}} \SigmaSLS' }
		\\
		\text{ and }
			&\
			\src{\Omega} \srel_{\src{\OB{f''}}} \trgSLS{\SigmaSLS}
		\\
        \text{ and }
			&\
			\trgSLS{p(A(\pc))} = \trgSLS{i}  \text{ and } \trgSLS{p(A'(\pc))} = \trgSLS{i'}
		\\
     \text{ and }
            &\
            \src{p(A(\pc))} =  \src{i} \text{ and } \src{p(A'(\pc))} = \src{i'}
        \\
		\text{ then }
			&\
			\src{\Omega}=\src{C, \OB{B}, \tup{p , M, A} \nsbigarrow{\tras{^\sigma}} C, \OB{B'}, \tup{p , M', A'} }=\src{\Omega'}
		\\
		\text{ and }
			&\
			\src{\tras{^\sigma}} \tracerel \trgSLS{\tra{^{\taint}}}
		\\
		\text{ and }
			&\
			\src{\Omega'} \srelref_{\src{\OB{f''}}}\, \trgSLS{\SigmaSLS{'}}
	\end{align*}
\end{theorem}

\begin{theorem}[$\SLSv$: Backward Simulation for Compiled Steps]\label{thm:sls-bwd-sim-comp-steps-lfence}(\showproof{sls-bwd-sim-comp-steps-lfence})
	\begin{align*}
		\text{ if }
			&\
			\trgSLS{\SigmaSLS}=\trgSLS{(C, \OB{B}, \tup{p, M, A}, \bot, \safeta) }
		\\
		\text{ and }
			&\
			\trgSLS{\SigmaSLS'}=\trgSLS{(C, \OB{B'}, \tup{p, M', A'}, \bot, \safeta)}
		\\
		\text{ and }
			&\
			\trgSLS{\SigmaSLS \bigspecarrowSLS{\tra{^{\taint}}} \SigmaSLS'}
		\\
		\text{ and }
			&\
			\src{\Omega} \srelref_{\src{\OB{f''}}}\, \trgSLS{\SigmaSLS}
		\\
             \text{ and }
			&\
			\src{p(A(\pc))} =  \src{i} \text{ and } \src{p(A'(\pc))} =  \src{i'}
		\\
            \text{ and }
			&\
			\trgSLS{A(\pc) : i} = \complfenceSLS{A(\pc) : i}   \text{ or } (\complfenceSLS{A(\pc) : i} = \trgSLS{A(\pc)} : \trgSLS{i} ; \trgSLS{p_c'})
            \\
            \text{ and }
                &\
                \trgSLS{A'(\pc) : i'} = \complfenceSLS{A'(\pc) : i'}  \text{ or } (\complfenceSLS{A'(\pc) : i'} = \trgSLS{A'(\pc)} : \trgSLS{i'} ; \trgSLS{p_c''})
		\\
		\text{ then }
			&\
			\src{\Omega}=\src{C, \OB{B}, \tup{p, M, A} \nsarrow{\tau^\sigma} C, \OB{B'}, \tup{p, M', A'}}= \src{\Omega'}
            \\
		\text{ and }
			&\
			\src{\tau^\sigma} \tracerel \trgSLS{\tra{^{\taint}}} \qquad \text{( using the trace relation!)}
		\\
		\text{ and }
			&\
			\src{\Omega'} \srelref_{\src{\OB{f''}}}\, \SigmaSLS'
	\end{align*}
\end{theorem}

\begin{lemma}[Forward Simulation for Compiled Statements in SLS]\label{thm:sls-fwd-sim-stm-lfence}(\showproof{sls-fwd-sim-stm-lfence})
	\begin{align*}
		\text{ if }
			&\
			\src{\Omega}=\src{C, \OB{B}, \tup{p, M, A} \nsarrow{\tau^\sigma} C, \OB{B'}, \tup{p, M', A'}} =\src{\Omega'}
		\\
		\text{ and }
			&\
			\src{\Omega} \srelref_{\src{\OB{f''}}} \trgSLS{\SigmaSLS}
		\\
        \text{ and }
			&\
			\src{p(A(\pc))} = \src{i} \text{ and } \src{p(A'(\pc))} =  \src{i'}
		\\
		\text{ and }
			&\
			\trgSLS{\SigmaSLS}= \trgSLS{C, \OB{B}, \tup{p, M , A}, \bot, \safeta \OB{f}}, 
		\\
		\text{ and }
			&\
			\trgSLS{\SigmaSLS'}= \trgSLS{C, \OB{B}, \tup{p, M' , A'}, \bot, \safeta} \OB{f'}
		\\
            \text{ and }
                &\
                \trgSLS{A(\pc) : i} = \complfenceSLS{A(\pc) : i} \text{ or } (\complfenceSLS{A(\pc) : i} = \trgSLS{A(\pc)} : \trgSLS{i} ; \trgSLS{p_c'})
            \\
            \text{ and } 
                &\
                \trgSLS{A'(\pc) : i'} = \complfenceSLS{A'(\pc) : i'} \text{ or } (\complfenceSLS{A'(\pc) : i'} = \trgSLS{A'(\pc)} : \trgSLS{i'} ; \trgSLS{p_c''})
		\\
		\text{ then }
			&\
			\trgSLS{\SigmaSLS \bigspecarrowSLS{\tra{^{\taint}}} \SigmaSLS'}
		\\
		\text{ and }
			&\
			\src{\tau^\sigma} \tracerel \trgSLS{\tra{^{\taint}}} \qquad \text{( using the trace relation!)}
		\\
		\text{ and }
			&\
			\src{\Omega'} \srelref_{\src{\OB{f''}}}\ \trgSLS{\SigmaSLS'}
	\end{align*}
\end{lemma}

\begin{lemma}[Initial States are Related]\label{thm:sls-ini-state-rel}(\showproof{sls-ini-state-rel})
	\begin{align*}
		&
		\forall \src{P}, \forall \src{\OB{f}}=\dom{\src{P}.\src{F}}, \forall \ctxSLS{}
		\\
		&
		\initFunc{\backtrfencec{\ctxSLS{}}\hole{P}} \srelref_{\src{\OB{f}}}\, \initFuncSLS{\ctxSLS{}\hole{\complfenceSLS{\src{P}}}}
	\end{align*}
\end{lemma}
 
\begin{proof}[Proof of \Thmref{thm:sls-lfence-comp-rdss}]\proofref{}{sls-lfence-comp-rdss}\hfill

    Instantiate $\src{A}$ with $\backtrfencec{\trgSLS{A}}$.
    This holds by \Thmref{thm:sls-corr-bt-lfence}.
\end{proof}

\BREAK

\begin{proof}[Proof of \Thmref{thm:sls-all-lfence-comp-are-rdss}]\proofref{}{sls-all-lfence-comp-are-rdss}\hfill

    By \Thmref{thm:ss-sni-source} we know all source languages are SS: $\text{HPS:} \forall\src{P}\ldotp \contract{}{NS}\vdash P : \rss$.
    
    Since \Thmref{thm:sls-lfence-comp-rdss} we have HPC: $\contractSpec{}{\SLSv}\vdash \complfenceSLS{\cdot} : \rdss$.

    By \Thmref{thm:rdss-eq-rdsp} we have $\text{HPP:} \contractSpec{}{\SLSv}\vdash \complfenceSLS{\cdot} :\rdssp$

    Unfolding \Thmref{def:rdsspc} of HPP we get 
    $\forall\src{P}\ldotp \text{ if } \contract{}{NS}\vdash\src{P} : \rss \text{ then } \contractSpec{}{\SLSv}\vdash\comp{\src{P}} : \rss$. 
  
    Now, we instantiate it with HPS and can conclude that $\forall\src{P}\ldotp \contractSpec{}{\SLSv}\vdash\complfenceSLS{P} : \rss$.
    
\end{proof}

\BREAK

\begin{proof}[Proof of \Thmref{thm:sls-ini-state-rel}]\proofref{}{sls-ini-state-rel}\hfill

Analogous to \Thmref{thm:v4-ini-state-rel},
\end{proof}

\BREAK

\begin{proof}[Proof of \Thmref{thm:sls-corr-bt-lfence}]\proofref{}{sls-corr-bt-lfence}\hfill

    This holds by \Thmref{thm:sls-ini-state-rel} and \Thmref{thm:sls-bwd-sim-lfence}.
\end{proof}

\BREAK

\begin{proof}[Proof of \Thmref{thm:sls-bwd-sim-lfence}]\proofref{}{sls-bwd-sim-lfence}\hfill

We proceed by induction on $\trgSLS{\SigmaSLS \bigspecarrowSLS{\tra{^{\taint}}} \SigmaSLS'}$:
\begin{description}
    \item[\Cref{tr:sls-reflect}]
    Trivial.
    \item[\Cref{tr:sls-single}]
    We have 
    \begin{align*}
        \trgSLS{\SigmaSLS \bigspecarrowSLS{\OB{\tau^{\taint}}} \SigmaSLS ''} \\
        \trgSLS{\SigmaSLS '' \specarrowSLS{\tau^{\taint}} \SigmaSLS '} \\ 
        \trgSLS{\SigmaSLS ''} = \trgSLS{\OB{C''}, \OB{B''}, \tup{p ,M'', A''}, \bot, \safeta}
    \end{align*}

    Applying IH on $\trgSLS{\SigmaSLS  \bigspecarrowSLS{\OB{\tau^{\taint''}}} \SigmaSLS ''}$ we get: 
    \begin{enumerate}
        \item $\src{\Omega  \nsbigarrow{\tras{^{\sigma''}}} \Omega''}$ and
        \item $\src{\tras{^{\sigma''}}} \tracerel \trgSLS{\tra{^{\taint''}}}$ and
        \item $\src{\Omega''} \srelref_{\src{\OB{f''}}}\, \trgSLS{\SigmaSLS {''}}$
        \item $\src{p(A''(\pc))} = \src{l''} : \src{i''}$ and $\trgSLS{p(A''(\pc))} = \trgSLS{l''} : \trgSLS{i''}$
    \end{enumerate}

     Let $\ffun{\src{A''(\pc)}} = \src{f''}$ and let $\ffun{\trgSLS{l'}} = \trgSLS{f'}$ (where $\trgSLS{\SigmaSLS' . A(\pc)} = \trgSLS{l'}$).

    From  $\src{\Omega''} \srelref_{\src{\OB{f''}}}\, \trgSLS{\SigmaSLS{''}}$ we have $\src{\OB{F}} \equiv \trgSLS{\OB{F}}$. Thus, there is $\src{f'}$.
    
    We proceed by case analysis on $\src{f''} \in \src{\OB{f''}}$:
    \begin{description}
        \item[$\src{f''} \in \src{\OB{f''}}$ (\textbf{in the compiled component})]
            By case analysis on $\src{f'} \in \src{\OB{f''}}$:
            \begin{description}
                \item[$\src{f'} \in \src{\OB{f''}}$ (\textbf{in the compiled component})]
                This holds by \Thmref{thm:v4-bwd-sim-comp-steps-lfence} together with .
                
                \item[$\src{f'} \notin \src{\OB{f''}}$ (\textbf{in the context})]
                This switch from component to context arises in two cases:
                \begin{description}
                    \item[\textbf{call}]
                    This is a call from a compiled function to a context function.
                    In this case we know that $\trgSLS{\SigmaSLS '' \specarrowSLS{\tau^{\taint}} \SigmaSLS '}$ was derived by \Cref{tr:sls-nospec-act} combined with \trgSLS{\Cref{tr:call}} and we know $\trgSLS{\tau^{\taint}} = \trgSLS{\clh{}^{\safeta}}$.
                    
                    This is analogous to the corresponding case in \Thmref{thm:v4-bwd-sim-lfence}.

                    \item[\textbf{return}]
                    This is a return from a compiled function to a context function.
                    Since speculation happens here and we are in the component, we can reuse the reasoning of \Thmref{thm:sls-fwd-sim-stm-lfence} in the ret case.

                    The rest of the case is analogous to the corresponding case in \Thmref{thm:v4-bwd-sim-lfence}.
                \end{description}
            \end{description}
        
        \item[$\src{f''} \notin \src{\OB{f''}}$ (\textbf{in the context})]

            By case analysis on $\src{f'} \in \src{\OB{f''}}$:
            \begin{description}
                \item[$\src{f'} \in \src{\OB{f''}}$ (\textbf{in the compiled component})]
                    
                    This switch from context to component arises in two cases:
                    
                    \begin{description}
                        \item[\textbf{call}]
                        
                        This is a call from a context function to a compiled function.
                        This is the dual of the case for call above.

                        \item[\textbf{return}]
                        This is a return from a context function to a compiled function.
                        In this case we know that $\trgSLS{\SigmaSLS '' \specarrowSLS{\tau^{\taint}} \SigmaSLS '}$ was derived by \Cref{tr:sls-nospec-act} combined with \trgSLS{\Cref{tr:ret}} and we know $\trgSLS{\tau^{\taint}} = \trgSLS{\retObs{}!^{\safeta}}$.

                        This is analogous to the corresponding case in \Thmref{thm:v4-bwd-sim-lfence}.
                        Note that no speculation can happen here because the $\pret$ instruction is part of the context.
                            
                \end{description}

                \item[$\src{f'} \notin \src{\OB{f''}}$ (\textbf{in the context})]
                This holds by \Thmref{thm:v4-back-sim-bts-lfence}.
                
            \end{description}
        
    \end{description}
    
    \item[\Cref{tr:sls-silent}]
    
    Analogous to the corresponding case in \Thmref{thm:v4-bwd-sim-lfence}.
\end{description}
\end{proof}

\BREAK

\begin{proof}[Proof of \Thmref{thm:sls-bwd-sim-comp-steps-lfence}]\proofref{}{sls-bwd-sim-comp-steps-lfence}\hfill

Analogous to \Thmref{thm:v4-bwd-sim-comp-steps-lfence} using  \Thmref{thm:sls-fwd-sim-stm-lfence}.
\end{proof}

\BREAK

\begin{proof}[Proof of \Thmref{thm:sls-fwd-sim-stm-lfence}]\proofref{}{sls-fwd-sim-stm-lfence}\hfill

    Since $\src{\Omega} \srel_{\src{\OB{f''}}}\ \trgSLS{\SigmaSLS}$ we know that $\src{p(A(\pc))} = \trgSLS{p(A(\pc))}$.
    We proceed by induction on $\nsarrow{}$:
    \begin{description}
        \item[\Cref{tr:skip}, \Cref{tr:assign}, \Cref{tr:load}, \Cref{tr:load-prv}, \Cref{tr:beqz-sat}, \Cref{tr:beqz-unsat}, \Cref{tr:call}, \Cref{tr:call-internal}, \Cref{tr:callback}]

        Analogous to the corresponding cases in \Thmref{thm:v4-fwd-sim-stm-lfence}.

        \item[\Cref{tr:store}, \Cref{tr:store-prv}] 

        The only change to the state is $\src{A'} = \src{A[pc \mapsto \succes(\pc)}$ and $\src{M'} = \src{M[n \mapsto A(x)]}$ where $\src{\exprEval{A}{e}{n : \taint}}$.

        Furthermore, we know $\src{\tau^\sigma} = \src{\storeObs^{\safeta}}$.
         
        \begin{description}
            \item[$\trgSLS{\SigmaSLS \bigspecarrowSLS{\tra{^{\taint}}} \SigmaSLS'}$]
            By \Thmref{thm:v4-fwd-sim-exp-lfence} we get $\trgSLS{\exprEval{A}{\complfenceSLS{e}}{\complfenceSLS{n} : \taintSLS}}$ and $\src{\taint} \equiv \taintSLS$.
            
            Thus, we can use \Cref{tr:store} to derive $\OmegaSLS \trgSLS{\nsarrow{\storeObs{\complfenceSLS{n}}^{\taint'}}} \OmegaSLS'$ with  $\trgSLS{A'} = \trgSLS{A[pc \mapsto \succes(\pc)}$ and $\trgSLS{M'} = \trgSLS{M[\complfenceSLS{n} \mapsto A(x)]}$.

            The rest follows trivially.
            
            \item[$\src{\storeObs{n}^{\safeta}} \tracerel \trgSLS{\storeObs{\complfenceSLS{n}}^{\taintpc \glb \taint'}}$]
            From $\srel$ we have $\trgSLS{\taintpc} = \trgSLS{\safeta}$ and thus $\trgSLS{\taintpc \glb \taint'} = \trgSLS{\safeta}$.
            
            The rest follows from \Cref{tr:ac-rel-rd} with \Thmref{thm:v4-val-rel-comp-lfence} for $\src{n} \vrel \complfenceSLS{n}$.

            \item[$\src{\Omega'} \srelref_{\src{\OB{f''}}}\, \trgSLS{\SigmaSLS'}$]
            Follows from $\src{A'} \rrel \trgSLS{A'}$ and $\src{M'} \rrel \trgSLS{M'}$ which follow from $\src{A} \rrel \trgSLS{A}$ and $\src{M} \hrel \trgSLS{M}$.
        \end{description}
        
        \item[\Cref{tr:ret}]
        This is the case where we need to account for multiple steps in the target, since the $\complfenceSLS{\cdot}$ adds an $\trgSLS{\pbarrier}$ instruction after the $\pret$ instruction:
        \begin{align*}
             \complfenceSLS{l : \pret} &= \trgSLS{
                                            \begin{aligned}[t]
                                                    &
                                                    l : \pret 
                                                    ~\text{\Cref{tr:sls-spec}}
                                                        \\
                                                        &\
                                                       l' :  \pbarrier
                                                       ~\text{\Cref{tr:sls-barr-spec}}
                                            \end{aligned}
                                            }
        \end{align*}
        The only change to the state is $\src{A'} = \src{A[pc \mapsto l]}$ and the removal of the return location $\src{l}$ from $\src{B}$.

        Furthermore, we know $\tau^\sigma = \retObs^{\safeta}$.
         \begin{description}
            \item[$\trgSLS{\SigmaSLS \bigspecarrowSLS{\tra{^{\taint}}} \SigmaSLS'}$]
                
                Thus, we can use \Cref{tr:ret} to derive $\OmegaSLS \trgSLS{\nsarrow{\pret^{\safeta}}} \OmegaSLS'$ with $\trgSLS{A'} = \trgSLS{A[pc \mapsto l]}$ where $\trgSLS{\OB{B}} = \trgSLS{\OB{B'} \cdot l}$.

                This is used in the step \Cref{tr:sls-spec}. The next instruction is the $\pbarrier$ and after \Cref{tr:sls-barr}, a rollback is triggered via \Cref{tr:sls-rollback}.

                Thus, we have the execution  $\SigmaSLS \bigspecarrowSLS{\retObs^{\taintpc \glb \safeta} \cdot \rollbackObsSLS} \SigmaSLS'$, where $\SigmaSLS' = \trgSLS{\tup{\OB{C}, \OB{B'}, \tup{p, M, A'}, \bot, \safeta}}$

            \item[$\src{\retObs^{\safeta}} \tracerel \trgSLS{\retObs^{\taintpc \glb \safeta} \cdot \rollbackObsSLS }$]
                From $\srel$ we have $\trgSLS{\taintpc} = \trgSLS{\safeta}$ and thus $\trgSLS{\taintpc \glb \safeta} = \trgSLS{\safeta}$.

                Next, from \Cref{tr:ac-rel-rlb} we get $\src{\epsilon} \arel \rollbackObsSLS$  and are finished.

            \item[$\src{\Omega'} \srelref_{\src{\OB{f''}}}\, \trgSLS{\SigmaSLS'}$]
                Follows from $\src{A'} \rrel \trgSLS{A'}$ which follows from $\src{A} \rrel \trgSLS{A}$ and $\src{B} \brel \trgSLS{B}$.

                Furthermore, $\src{B'} \brel \trgSLS{B'}$ which follows from $\src{B} \brel \trgSLS{B}$.
        \end{description}

        \item[\Cref{tr:ret-internal}]
        This case is analogous to case \textbf{ret} above.

        \item[\Cref{tr:retback}]
        This case cannot arise as we do not step to a compiled instructions $\complfenceSLS{i'}$ in the target.
    \end{description}
    
\end{proof}

\section{The USLH Compiler}

Here we describe ultimate speculative-load-hardening (USLH) by \citet{uslh}.
The authors show that variable time instructions can be used to leak a secret.

This is not captured in the model of \citet{S_sec_comp}, because variable time latency is not captured in the attacker model that \citet*{S_sec_comp} use.

While the authors of \cite{uslh} have a proof that their uslh is secure (w.r.t. a notion of speculative constant time), their proof is not robust in our sense.

To model variable time latency instructions, we extend our model of \muasm{} with a second assignment operation for variable time instructions. This new assignment operations leaks the operands of the operation.

We note that the base for this compiler is the strong-slh compiler $\compsslhB{\cdot}$ of \cite{S_sec_comp} without the additional rule 
for variable time latency instructions. 
Thus, the compiler $\compuslhB{\cdot}$ and $\compsslhB{\cdot}$ only differ for variable time latency instructions and we do not
present $\compsslhB{\cdot}$ separately.

\subsection{Extensions to \muasm}

We extend \muasm{} by adding a new assignment instruction 
\begin{figure}[!h]
\hfsetfillcolor{red!10}
\hfsetbordercolor{red} \centering
    \begin{alignat*}{3}
    \textit{(Instructions)} \quad  &i  &\coloneqq \quad & \cdots \mid \tikzmarkin{last} \vassign{x}{y}{z}
    \tikzmarkend{last}
    \end{alignat*}
\end{figure}

Furthermore, we need to extend the trace model to include this new form of leakage:

\begin{align*}
	\mi{Heap\&Pc\ Act.s}~\src{\delta} \bnfdef&\ \cdots \mid \opObs{x}{y}
\end{align*}

\begin{center}\small
\mytoprule{\sigma \nsarrow{\tau} \sigma}

\typerule{VAssign}
{
\select{p}{\av(\pc)} = \vassign{x}{y}{z} & x \neq \pc & \exprEval{\av}{y}{n1} & \exprEval{a}{y}{n2}
}
{
\tup{p,\mv, \av} \nsarrow{\opObs{n1}{n2}} \tup{p, \mv, \av[\pc \mapsto \av(\pc)+1,x \mapsto n1 \binop n2]}
}{vassign}

\typerule{T-VAssign}
{
\select{p}{a(\pc)} = \vassign{x}{y}{z} & x \neq \pc & \exprEval{a}{y}{n1 : \taint} & \exprEval{a}{y}{\taint'}
}
{
\taintpc; C;\OB{B};\tup{p,\mta,a} \nsarrow{\taint \lub \taint'} \taintpc; C;\OB{B}; \tup{p, \mta, a[\pc \mapsto a(\pc)+1,x \mapsto n1 \binop n2 : \taint \lub \taint']}
}{t-vassign}
\end{center}

And we extend the Action relation $\arel$ as well:
\begin{center}
    \mytoprule{\text{Action relation} \areldef }

    \typerule{Action Relation - vassign}{
		\src{v} \vrel \trgB{v}
		&
            \src{v'} \vrel \trgB{v'}
            &
		\src{\taint}\equiv\taint'
	}{
		\src{\opObs{v}{v'}^\taint} \arel \trgB{\opObs{v}{v'}^{\taint'}}
	}{ac-rel-vassign}
 
\end{center}

\subsection{The Problem}

We now explain why the model of \citet{S_sec_comp} does not capture this vulnerability.
Here is the program taken from \cite[Listing 5]{uslh}  translated to \muasm{}:
\begin{lstlisting}[style=MUASMstyle, caption={Example Program of \cite{uslh} translated (roughly) to \muasm{}}, label={lst:uslh-example}, escapeinside=!!]
Main:
    x !$\leftarrow{}$! isPublic
    !$\pjz{x}{\bot}$!
    !$\vassign{v}{v}{v}$!
    !$\vdots$!
\end{lstlisting}
Depending on the value, which can be slow or fast. The memory access is cached or not.

The Compiler combines SSLH of \citet{S_sec_comp} and the countermeasure proposed by \cite{uslh} into one compiler. We follow the approach by \cite{uslh} in Figure 5.

To protect across function calls, the high bits of the stack pointer $rsp$ are used.
We however use another register $\rslhC$ that we use.
Before every function call and before every return we store the speculation flag in the register $\rslh$.
Next, we change the prologue of every function to contain an assignment that recovers the speculation from $\rslhC$.
This more closely models the current implementation of SLH in compilers.

{\allowdisplaybreaks
\begin{align*}
    \compuslhB{ M ; \OB{F} ; \OB{I}} &= \trgB{ \compuslhB{M} ; \compuslhB{\OB{F}} ; \compuslhB{\OB{I}}}
	\\\\
	\compuslhB{ \srce} &= \trgB{\emptyset}
	\\
	\compuslhB{ \OB{I}\cdot f } &= \trgB{ \compuslhB{ \OB{I} }\cdot \compuslhB{f} }
	\\\\
	\compuslhB{ M ; -n\mapsto v : \unta} &= \trgB{\compuslhB{M} ; -\compuslhB{n} \mapsto \compuslhB{v} : \unta}
	\\\\
	\compuslhB{fname, p_1;p_2} &= \trgB{fname, l_{new} : \passign{\rslh}{\rslhC}; \compuslhB{p_1;p_2}}
	\\
	\compuslhB{p_1;p_2} &= \trgB{\compuslhB{p_1} ; \compuslhB{p_2}}
	\\\\
	\compuslhB{l : \pskip} &= \trgB{l : \pskip} 
	\\
        \compuslhB{\pcondassign{x}{e}{e'}} &= \trgB{
                                            \begin{aligned}[t]
                                                    &
                                                    l :\passign{\rscr}{\compuslhB{e}}
                                                        \\
                                                        &\
                                                        l_1 :\passign{\rscr'}{\compuslhB{e'}}
                                                            \\
                                                            &\
                                                            l_2 :\pcondassign{\rscr}{0}{\rslh}
                                                                \\
                                                                &\
                                                                l_3 :\pcondassign{\rscr'}{0}{\rslh}
                                                                    \\
                                                                    &\
                                                                    l_4 :\pcondassign{x}{\rscr}{\rscr'}
                                            \end{aligned}
                                            }
        \\
        \compuslhB{l : \pload{x}{e}} &= \trgB{
                                            \begin{aligned}[t]
                                                    &
                                                    l :\passign{\rscr}{\compuslhB{e}}
                                                        \\
                                                        &\
                                                        l' :\pcondassign{\rscr}{0}{\rslh}
                                                            \\
                                                            &\
                                                            l'' :\pload{x}{\rscr}
                                            \end{aligned}
                                        }
        \\
        \compuslhB{l : \ploadprv{x}{e}} &= \trgB{
                                            \begin{aligned}[t]
                                                    &
                                                    l :\passign{\rscr}{\compuslhB{e}}
                                                        \\
                                                        &\
                                                        l' :\pcondassign{\rscr}{0}{\rslh}
                                                            \\
                                                            &\
                                                            l'' :\ploadprv{x}{\rscr}
                                            \end{aligned}
                                        }
        \\
        \compuslhB{l : \pstore{x}{e}} &= \trgB{
                                            \begin{aligned}[t]
                                                    &
                                                    l :\passign{\rscr}{\compuslhB{e}}
                                                        \\
                                                        &\
                                                        l' :\pcondassign{\rscr}{0}{\rslh}
                                                            \\
                                                            &\
                                                            l'' :\pstore{x}{\rscr}
                                            \end{aligned}
                                        }
        \\
        \compuslhB{l : \pstoreprv{x}{e}} &= \trgB{
                                            \begin{aligned}[t]
                                                    &
                                                    l :\passign{\rscr}{\compuslhB{e}}
                                                        \\
                                                        &\
                                                        l' :\pcondassign{\rscr}{0}{\rslh}
                                                            \\
                                                            &\
                                                            l'' :\pstoreprv{x}{\rscr}
                                            \end{aligned}
                                        }
        \\
	\compuslhB{l : \passign{x}{e}} &= \trgB{
                                            \begin{aligned}[t]
                                                    &
                                                    l : \passign{\rscr}{\compuslhB{e}}
                                                        \\
                                                        &\
                                                        l' :\pcondassign{\rscr}{0}{\rslh}
                                                            \\
                                                            &\
                                                            l'' :\passign{x}{\rscr}
                                            \end{aligned}
                                        }
        \\
        \compuslhB{l : \vassign{x}{e}{e'}} &= \trgB{
                                            \begin{aligned}[t]
                                                    &
                                                    l :\passign{\rscr}{\compuslhB{e}}
                                                        \\
                                                        &\
                                                        l_1 :\passign{\rscr'}{\compuslhB{e'}}
                                                            \\
                                                            &\
                                                            l_2 :\pcondassign{\rscr}{0}{\rslh}
                                                                \\
                                                                &\
                                                                l_3 :\pcondassign{\rscr'}{0}{\rslh}
                                                                    \\
                                                                    &\
                                                                    l_4 :\vassign{x}{\rscr}{\rscr'}
                                            \end{aligned}
                                        }
        \\
        \compuslhB{l : \pjmp{e}} &= \trgB{
                                            \begin{aligned}[t]
                                                    &
                                                    l : \passign{\rscr}{\compuslhB{e}}
                                                        \\
                                                        &\
                                                        l' : \pcondassign{\rscr}{\bot}{\rslh}
                                                            \\
                                                            &\
                                                            l'' :\pjmp{\rscr}
                                            \end{aligned}
                                        }
	\\
        \\
	\compuslhB{l : \pcall{f}} &= \trgB{
										\begin{aligned}[t]
												&
												l :\passign{\rslhC}{\rslh} \\
													&\
													l' :\pcall{f}
														\\
														&\ l'' : \passign{\rslh}{\rslhC}
										\end{aligned}
									} ~ \text{Here $f$ is a function name}
	\\
	\compuslhB{l : \pret} &= \trgB{
									\begin{aligned}[t]
											&
											l :\passign{\rslhC}{\rslh} \\
												&\
												l' :\pret
									\end{aligned}
								}
	\\
	\compuslhB{l : \pjz{x}{l'}} &= \trgB{
                                            \begin{aligned}[t]
                                                    &
                                                    l : \passign{\rscr}{x}
                                                        \\
                                                        &\
                                                        l_1' : \pcondassign{\rscr}{0}{\rslh}
                                                            \\
                                                            &\
                                                           l_1'' :  \pjz{\rscr}{l_{new}} \\
                                                                &\
                                                                l_2 : \passign{\rslh}{\rslh \lor \rscr} \\
                                                                    &\
                                                                    l_2' : \pjmp{next(l)}
                                                                        \\
                                                                        &\ 
                                                                        l_{new} : \passign{\rslh}{\rslh \lor \neg \rscr}
                                                                            \\ 
                                                                            &\ 
                                                                            l_{new}' : \pjmp(l')
                                            \end{aligned}
                                        }
	\\\\
\end{align*}
}
The compiler uses a scratch registers to save the immediate values and we reserve a register $\rslh$ as a speculation flag. This register is initially set to $\mi{false}$ and is set to $\mi{true}$ when speculation starts. Note, that we assume that this register is not used by the program and only inserted by the compiler.

We need to be careful when compiling the $\jzC{}$ instruction. Because we operate on the \muasm{} level we actually do not have access to both branches when just looking at the instruction. We know that the instructions after the branch instruction are the ones of the branch is not taken. However, we do not know where the label $l$ in the program is for the instruction $\pjz{x}{l}$. However, we need the code at this position as well to correctly update the predicate bit $\rslh$ (Note that the same would have to be done for a fence compiler since you need to protect both branches).
In general, this means that the compilation of a $\jzC{}$ instruction is a non-local transformation.
The idea is to \textit{redirect} the original branch target. That is achieved by the label $l_{new}$ which first updates the register $\rslh$ and then jumps to the original target $l'$ when the branch is taken.
If the branch is not taken then we go to $l_2$ and update $\rslh$. A jump to $\mi{next}(\src{l})$ allows us to access the original "else" branch of the conditional. All in all, this way we achieved to update $\rslh$ in both parts of the branches without needing to modify other parts of the program. This is a local transformation for the branch.

\subsection{Relations and Proofs}

\mytoprule{\text{Memory relation} \hreldef \text{Value relation} \vreldef }
\begin{center}
	\typerule{Memory - base }{}{
		\srce \hrel \trgBe
	}{hrel-b}  
	\typerule{Memory - ind }{
		\src{M} \hrel \trgB{M}
		\\
		\src{z} \vrel \trgB{z}
		&
		\src{v} \vrel \trgB{v}
		&
		\src{\sigma} \equiv \taintB
	}{
		\src{M ; z\mapsto v:\taint} \hrel \trgB{M; z\mapsto v:\taint}
	}{hrel-i}
	\typerule{Memory - start }{
		\src{M} \hrel \trgB{M}
		&
		\src{M'} \hrel \trgB{M'}
		\\
		\src{v} \vrel \trgB{v}
		&
		\src{\sigma} \equiv \taint
	}{
		\src{M ; 0\mapsto v:\sigma ; M'} \hrel \trgB{M ;0\mapsto v:\taint ; M'}
	}{hrel-start}

    \typerule{ Register File }{
		\forall x \in \Reg \setminus \{\rscr, \rscr', \rslh, \rslhC \}, \src{A}(x) = \src{v} : \taint
		&
		\trgB{A}(x) = \trgB{v} : \taintB
		\\
		\src{v} \vrel \trgB{v}
		& 
		\taint \equiv \taintB
            & 
            \trgB{A(\rslh)} = \trgB{false}
	}{
		\src{A} \rrel \trgB{A}
	}{rsrel}
 
\end{center}
\botrule

Since we use scratch registers $\rscr$ and $\rscr'$ in our compiler to compute the values of expression, we need to exclude them in the relation $\rrel$. These scratch registers are only used in the compiled target program.
Normally, we would parameterize the relation $\rrel$ by a set containing registers used as scratch registers. However, since there are only two we decided to just 'hard code' them in the relation.

The Stack relation $\brel$ is indexed by the current function. 
Since we added an instruction ot the beginning of each compiled function we changed teh return address
between compiled and source functions in the component. They are offsetted by 1.
That is why we track the current function $f$ and the component functions $\OB{F}$, now on $\brel_f$. This way we can relate the correct 
return addresses in component and context. 
\begin{center}
	\mytoprule{\text{Register relation} \rreldef \text{ Program relation} \creldef \text{Stack Relation} \brel \text{ State relation} \ssreldef}

    \typerule{Stack - base }{}{
		\OB{F}, \srce \brel_{f} \OB{F}, \trgBe
	}{v1-brel-b}
 \typerule{Stack Region - start }{
		\src{\OB{n}} \brel \trgB{\OB{n}} & f \notin \OB{F}
		\\
		\src{\sigma} \equiv \taintB
        & 
        \src{n} \vrel \trgB{n}
	}{
		\src{\OB{F}, \OB{n} ; l :\sigma} \brel_{f} \trgB{\OB{F}, \OB{n} ; l:\taintB}
	}{v1-brel-region-cont}
 
	\typerule{Stack Region - start-comp }{
		\src{\OB{n}} \brel \trgB{\OB{n}} & f \in \OB{F}
		\\
		\src{\sigma} \equiv \taintB
        & 
        \src{l} \vrel \trgB{l + 1}
	}{
		\src{\OB{F}, \OB{n} ; l:\sigma} \brel_f \trgB{\OB{F}, \OB{n} ;l :\taintB}
	}{v1-brel-region-comp}

\typerule{ Programs }{
        \forall \src{l : i} \in \src{p} \text{ and } \ffun{\src{l}} = \src{f} \text{ and } \ffun{\compuslhB{\src{l}}} = \trgB{f} 
        \\
		\text{ if } \src{f} \in \src{\OB{F}} \text{ and } \src{f} \in \src{\OB{f}}
		\text{ then } \trgB{p}(\compuslhB{l}) = \compuslhB{i}
		\\
  \forall \trgB{l : i} \in \trgB{p} \ldotp \ffun{\trgB{l}} = \trgB{f} \text{ and } \trgB{f} \notin \src{\OB{f}} \text{ and } \trgB{f} \in \trgB{\OB{F}}\text{ then }
    \src{p(l)} = \backtrfencec{\trgB{i}}
		\\
	}{
		\src{\OB{F}; \src{p}} \crel_{\src{\OB{f}}} \trgB{\OB{F}; \trgB{p}} 
	}{v1-uslh-comps}

 \typerule{ States }{
       \src{\Omega} \srel_{\src{\OB{f''}}}  \trgB{\Omega, \bot, \safeta}
        &
        \forall \trgB{\Omega_S} \in \trgB{\OB{\Omega}}, \src{\Omega} \ssrel \trgB{\Omega_S}
	}{
		\src{\Omega } \ssrel_{\src{\OB{f''}}} \trgB{ (\Omega, \bot, \safeta) \cdot \OB{(\Omega, n, \unta)} }
	}{v1-states}

 \typerule{ Base States }{
        \src{A} \rrel \trgB{A}
        &
		\src{\OB{F}; p} \crel_{\src{\OB{f''}}} \trgB{\OB{F}; p}
		&
		\src{M} \hrel \trgB{M}
		&
        \src{B} \brel \trgB{B}
        &
        \src{\OB{I}} \equiv \trgB{\OB{I}}
	}{
		\src{\OB{F};\OB{I},\OB{B}, \tup{p, M, A}} \srel_{\src{\OB{f''}}} \trgB{\OB{F};\OB{I},\OB{B}, \tup{p, M, A}, \bot, \safeta }
	}{v1-base-states}

 \typerule{ Single States }{
        \src{\OB{F}; p} \crel_{\src{\OB{f''}}} \trgB{\OB{F}; p}
        &
        \src{\OB{I}} \equiv \trgB{\OB{I}}
        &
        \vdash \trgB{M} : \shrel & \vdash \trgB{A}: \shrel
	}{
		\src{\OB{F};\OB{I},\OB{B}, \tup{p, M, A}} \ssrel_{\src{\OB{f''}}} \trgB{\OB{F};\OB{I},\OB{B}, \tup{p, M, A}, \bot, \safeta }
	}{v1-single-states}

 \typerule{Memory - relation same }{
            \forall \trgB{n \mapsto v : \taint} \in \trgB{M} 
            \text{ if } \trgB{n} \geq \trgB{0} \text{ then } \trgB{\taint} = \trgB{\safeta}
	}{
		\vdash \trgB{M}: \shrel
	}{shrel-same}
 \typerule{Register - relation same }{
            \forall \trgB{x} \in \trgB{\Reg}  &
            \trgB{A(x)} = \_ : \trgB{\safeta}
            &
            \trgB{A(\rslh)} = \trgB{true} : \trgB{\safeta}
	}{
		\vdash \trgB{A}: \shrel
	}{shrel-same-register}
 
\end{center}
\botrule

\begin{center}
	
	\scalebox{0.7}{
\begin{forest}
	for tree={
		align = left,
		font = \footnotesize,
		forked edge,
	}
	[\Thmref{thm:v1-uslh-comp-rdss}
			[\Thmref{thm:v1-corr-bt-uslh}
				[\Thmref{thm:v1-ini-state-rel}
					[\Cref{thm:v4-heap-rel-comp-lfence}]
					[\Cref{thm:v4-heap-rel-bt-lfence}, name=heap-bt]
				]
				[\Thmref{thm:v1-bwd-sim-uslh}
						[\Thmref{thm:v1-bwd-sim-comp-steps-uslh}
								[\textbf{\Thmref{thm:v1-fwd-sim-stm-uslh}}
									[
										[\Thmref{thm:v1-fwd-sim-exp-uslh}]
										[\Thmref{thm:spec-rel-satte-safe}]
									]
								]
						]
						[\Thmref{thm:v1-back-sim-bts-uslh}   
							[\Thmref{thm:v4-back-sim-bte-lfence}, name=exp-bt]
						]
				]
			]
	]
	\node (sniSource) at (-4,3) {\Thmref{thm:ss-sni-source}};
	\node (rsscEq) [below=of sniSource] {\Thmref{thm:rdss-eq-rdsp}};
	\draw[red, thick, dotted] ($(sniSource.north west)+(-0.3,0.3)$)  rectangle ($(rsscEq.south east)+(0.9,-0.3)$);
\end{forest}
	}
\end{center}

\begin{theorem}[The uslh compiler is \rdss]\label{thm:v1-uslh-comp-rdss}(\showproof{v1-uslh-comp-rdss})
	\begin{align*}
		\contract{ct + vl}{\Bv}\vdash \compuslhB{\cdot} : \rdss
	\end{align*}
\end{theorem}

\begin{theorem}[Correctness of the Backtranslation for uslh]\label{thm:v1-corr-bt-uslh}(\showproof{v1-corr-bt-uslh})
	\begin{align*}
		\text{ if } 
			&\
			\trgB{ \amTracevB{\ctxB{}\hole{\compuslhB{P}}}{\tra{^{\taint}}} }
		\\
		\text{ then }
			&
			\src{ \Trace{\backtrfencec{\ctxB{}}\hole{P}}{ \tras{^{\sigma}}} }
		\\
		\text{ and }
			&\
			\tras{^{\sigma}} \rels \trgB{\tra{^{\taint}}}
	\end{align*}
\end{theorem}

\begin{theorem}[Generalised Backward Simulation for uslh]\label{thm:v1-bwd-sim-uslh}(\showproof{v1-bwd-sim-uslh})
	\begin{align*}
		\text{ if }
			&\
			\ffun{\src{l}}\in\src{\OB{f''}} \text{ then } \trgB{l} : \trgB{i} = \compuslhB{l} : \compuslhB{i} \text{ else } \src{l} : \src{i} = \backtrfencec{\trgB{l}} : \backtrfencec{\trgB{i}} 
			\\
		\text{ and}
			&\
			\text{ if } \ffun{\src{l'}}\in\src{\OB{f''}} \text{ then } \trgB{l'} : \trgB{i'} = \compuslhB{l'} : \compuslhB{i'} \text{ else } \src{l'} : \src{i'} = \backtrfencec{\trgB{l'}} : \backtrfencec{\trgB{i'}}
		\\
		\text{ and }
			&\
			\SigmaB = \trgB{(C, \OB{B}, \tup{p, M, A}, \bot,\safeta)}
		\\
		\text{ and }
			&\
			\SigmaB' =\trgB{(C, \OB{B'}, \tup{p, M', A'}, \bot,\safeta)}
		\\
		\text{ and }
			&\
			\trgB{ \SigmaB \bigspecarrowB{\tra{^{\taint}}} \SigmaB' }
		\\
		\text{ and }
			&\
			\src{\Omega} \ssrel_{\src{\OB{f''}}} \trgB{\SigmaB}
		\\
        \text{ and }
			&\
			\trgB{p(A(\pc))} = \trgB{l} : \trgB{i}  \text{ and } \trgB{p(A'(\pc))} = \trgB{l'} : \trgB{i'}
		\\
		\text{ then }
			&\
			\src{\Omega}=\src{C, \OB{B}, \tup{p , M, A} \nsbigarrow{\tras{^\sigma}} C, \OB{B'}, \tup{p , M', A'}}=\src{\Omega'}
		\\
        \text{ and }
			&\
			\src{p(A(\pc))} = \src{l} : \src{i} \text{ and } \src{p(A'(\pc))} = \src{l'} : \src{i'}
        \\
		\text{ and }
			&\
			\src{\tras{^\sigma}} \tracerel \trgB{\tra{^{\taint}}}
		\\
		\text{ and }
			&\
			\src{\Omega'} \ssrelref_{\src{\OB{f''}}}\, \trgB{\SigmaB{'}}
	\end{align*}
\end{theorem}

\begin{theorem}[Backward Simulation for Compiled Steps in uslh]\label{thm:v1-bwd-sim-comp-steps-uslh}(\showproof{v1-bwd-sim-comp-steps-uslh})
	\begin{align*}
		\text{ if }
			&\
			\trgB{\SigmaB}=\trgB{(C, \OB{B}, \tup{p, M, A}, \bot, \safeta) }
		\\
		\text{ and }
			&\
			\trgB{\SigmaB'}=\trgB{(C, \OB{B'}, \tup{p, M', A'}, \bot, \safeta)}
		\\
		\text{ and }
			&\
			\trgB{\SigmaB \bigspecarrowB{\tra{^{\taint}}} \SigmaB'}
		\\
		\text{ and }
			&\
			\src{\Omega} \ssrelref_{\src{\OB{f''}}}\, \SigmaB
		\\
        \text{ and }
			&\
			\trgB{p(A(\pc))} = \compuslhB{l} : \compuslhB{i} \text{ and } \trgB{p(A'(\pc))} = \compuslhB{l'} : \compuslhB{i'}
		\\
		\text{ then }
			&\
			\src{\Omega}=\src{C, \OB{B}, \tup{p, M, A} \OB{f} \nsarrow{\tau^\sigma} C, \OB{B'}, \tup{p, M', A'} \OB{f'}}= \src{\Omega'}
        \\
        \text{ and }
			&\
			\src{p(A(\pc))} = \src{l} : \src{i} \text{ and } \src{p(A'(\pc))} = \src{l'} : \src{i'}
		\\
		\text{ and }
			&\
			\src{\tau^\sigma} \tracerel \trgB{\tra{^{\taint}}} \qquad \text{( using the trace relation!)}
		\\
		\text{ and }
			&\
			\src{\Omega'} \ssrelref_{\src{\OB{f''}}}\, \SigmaB'
	\end{align*}
\end{theorem}

\begin{lemma}[Forward Simulation for Compiled Statements in uslh]\label{thm:v1-fwd-sim-stm-uslh}(\showproof{v1-fwd-sim-stm-uslh})
	\begin{align*}
		\text{ if }
			&\
			\src{\Omega}=\src{C, \OB{B}, \tup{p, M, A} \OB{f} \nsarrow{\tau^\sigma} C, \OB{B'}, \tup{p, M', A'} \OB{f'}} =\src{\Omega'}
		\\
		\text{ and }
			&\
			\src{\Omega} \ssrelref_{\src{\OB{f''}}} \trgB{\SigmaB}
		\\
        \text{ and }
			&\
			\src{p(A(\pc))} = \src{l} : \src{i} \text{ and } \src{p(A'(\pc))} = \src{l'} : \src{i'}
		\\
		\text{ and }
			&\
			\trgB{\SigmaB}= \trgB{C, \OB{B''}, \tup{p, M , A}, \bot, \safeta}
		\\
		\text{ and }
			&\
			\trgB{\SigmaB'}= \trgB{C, \OB{B'''}, \tup{p, M' , A'}, \bot, \safeta}
		\\
        \text{ and }
			&\
			\trgB{p(A(\pc))} = \compuslhB{l} : \compuslhB{i} \text{ and } \trgB{p(A'(\pc))} = \compuslhB{l'} : \compuslhB{i'}
		\\
		\text{ then }
			&\
			\trgB{\SigmaB \bigspecarrowB{\tra{^{\taint}}} \SigmaB'}
		\\
		\text{ and }
			&\
			\src{\tau^\sigma} \tracerel \trgB{\tra{^{\taint}}} \qquad \text{( using the trace relation!)}
		\\
		\text{ and }
			&\
			\src{\Omega'} \ssrelref_{\src{\OB{f''}}}\ \trgB{\SigmaB'}
	\end{align*}
\end{lemma}

\begin{lemma}[Forward Simulation for Expressions in uslh]\label{thm:v1-fwd-sim-exp-uslh}(\showproof{v1-fwd-sim-exp-uslh})
	\begin{align*}
		\text{ if }
			&\ 
			\src{A\triangleright e \bigreds v : \sigma}
		\\
		\text{ and }
			&\
			\src{A}\rrelref\trgB{A}
		\\
		\text{ and }
			&\
			\src{\sigma}\equiv\taintB
		\\
		\text{ then }
			&\
			\trgB{A \triangleright \compuslhB{e}\bigreds \compuslhB{v} : \taintB}
	\end{align*}
\end{lemma}

\begin{lemma}[Backward Simulation for Backtranslated Statements]\label{thm:v1-back-sim-bts-uslh}(\showproof{v1-back-sim-bts-uslh})
	\begin{align*}
		\text{ if }
			&\
			\trgB{\SigmaB}=\trgB{(C, \OB{B}, \tup{p, M, A},\bot, \safeta) }
		\\
		\text{ and }
			&\
			\trgB{\SigmaB'}=\trgB{(C, \OB{B'}, \tup{p, M', A'}, \bot, \safeta)}
		\\
		\text{ and }
			&\
			\trgB{\SigmaB \specarrowB{\tau{^{\taint}}} \SigmaB'}
		\\
		\text{ and }
			&\
			\src{\Omega} \ssrel_{\src{\OB{f''}}} \trgB{\SigmaB}
		\\
        \text{ and }
			&\
			\trgB{p(A(\pc))} = \trgB{l} : \trgB{i} \text{ and } \trgB{p(A'(\pc))} = \trgB{l'} : \trgB{i'}
		\\
		\text{ then }
			&\
			\src{\Omega}=\src{C, \OB{B}, \tup{p, M, A} \nsarrow{\tau^\sigma} C, \OB{B'},\tup{p, M', A'}} = \src{\Omega'}
		\\
        \text{ and }
			&\
			\src{p(A(\pc))} = \backtrfencec{\trgB{l}} : \backtrfencec{\trgB{i}} \text{ and } \src{p(A'(\pc))} = \backtrfencec{\trgB{l'}} : \backtrfencec{\trgB{i'}}
		\\
		\text{ and }
			&\
			\src{\tau^\sigma} \arelref \trgB{\tau^\taint} \qquad \text{( using the action relation!)}
		\\
		\text{ and }
			&\
			\src{\Omega'} \ssrel_{\src{\OB{f''}}} \trgB{\SigmaB'}
	\end{align*}
\end{lemma}

\begin{lemma}[Initial States are Related]\label{thm:v1-ini-state-rel}(\showproof{v1-ini-state-rel})
	\begin{align*}
		&
		\forall \src{P}, \forall \src{\OB{f}}=\dom{\src{P}.\src{F}}, \forall \ctxB{}
		\\
		&
		\initFunc{\backtrfencec{\ctxB{}}\hole{P}} \ssrelref_{\src{\OB{f}}}\, \initFuncB{\ctxB{}\hole{\compuslhB{\src{P}}}}
	\end{align*}
\end{lemma}

\begin{lemma}[A Value is Related to its Compilation for uslh]\label{thm:v1-val-rel-comp-uslh}(\showproof{v1-val-rel-comp-uslh})
	\begin{align*}
		{\src{v}}\vrel\compuslhB{v}
	\end{align*}
\end{lemma}

\begin{lemma}[A Heap is Related to its Compilation for uslh]\label{thm:v1-heap-rel-comp-uslh}(\showproof{v1-heap-rel-comp-uslh})
	\begin{align*}
		{\src{H}}\hrel\compuslhB{H}
	\end{align*}
\end{lemma}

\begin{lemma}[A Value is Related to its Backtranslation]\label{thm:v1-val-rel-bt-uslh}(\showproof{v1-val-rel-bt-uslh})
	\begin{align*}
		\backtrfencec{\trgB{v}}\vrel\trgB{v}
	\end{align*}
\end{lemma}

\begin{lemma}[A Memory is Related to its Backtranslation]\label{thm:v1-heap-rel-bt-uslh}(\showproof{v1-heap-rel-bt-uslh})
	\begin{align*}
		\backtrfencec{\trgB{M}}\hrel\trgB{M}
	\end{align*}
\end{lemma}

\subsection{Speculation proceeds safely}
Since SLH does not stop speculation but masks all vulnerable accesses we need to make sure that speculation is safe during the execution. The n

First, we state that speculation terminates, which follows form the fact that the $\minWindow{}$ of the state is decreased with each step.

Furthermore, by construction there is an upper bound of the number of speculative steps. Even if in each step a new speculative transaction is started, we still have an upper bound n, 
While the safeness comes from the countermeasure uslh.

\begin{forest}
	for tree={
		align = left,
		font = \footnotesize,
		forked edge,
	}
	[\Thmref{thm:spec-rel-satte-safe}
		[
			\Thmref{thm:spec-most-omega}
				[
					[\Thmref{thm:ctx-spec-most-omega}
						[
							[\Thmref{thm:ctx-spec-sing-safe}]
						]
					]
					[\Thmref{thm:comp-spec-most-omega}
						[
							[\Thmref{thm:comp-spec-sing-safe}]
						]
					]
				]
		]
	]
\end{forest}
\begin{lemma}[Any Speculation from related states is safe]\label{thm:spec-rel-satte-safe}(\showproof{spec-rel-satte-safe})
	\begin{align*}
		\text{ if }
		&\
			\src{\Sigma} \ssrel_{\src{\OB{f_c}}} \SigmaB
		\\
		\text{ and }
		&\
			\SigmaB = \trgB{ (\Omega_b, \bot, \safeta) \cdot \OB{(\Omega', n', \unta)} \cdot (\Omega, n, \unta)}
		\\
		\text{ and }
		&\ 
			\SigmaB = \trgB{ (\Omega_b, \bot, \safeta)}
		\\
		\text{ then }
		&\ 
			\SigmaB \bigspecarrowB{\tra{^{\taint}}} \SigmaB'
		\\ 
		\text{ and }
		&\
			\src{\empTr} \tracerel \trgB{\tra{^{\taint}}}
		\\ 
		\text{ and }
		&\
		\src{\Sigma} \ssrel_{\src{\OB{f_c}}} \SigmaB'
	\end{align*}
\end{lemma}

\begin{lemma}[Speculation is safe and lasts at most omega]\label{thm:spec-most-omega}(\showproof{spec-most-omega})
	\begin{align*}
		\text{ if }
		&\
			\src{\Omega} \ssrel_{\src{\OB{f_c}}} \SigmaB
		\\
		\text{ and }
		&\
			\SigmaB = \trgB{ (\Omega_b, \bot, \safeta) \cdot \OB{(\Omega', n', \unta)} \cdot (\Omega, n, \unta)}
		\\
		\text{ and }
		&\ 
			\SigmaB' = \trgB{ (\Omega_b, \bot, \safeta)  \cdot \OB{(\Omega', n', \unta)} \cdot (\Omega'', 0, \unta)}
		\\
		\text{ then }
		&\ 
			\SigmaB \bigspecarrowB{\tra{^{\taint}}} \SigmaB'
		\\ 
		\text{ and }
		&\
			\src{\empTr} \tracerel \trgB{\tra{^{\taint}}}
		\\ 
		\text{ and }
		&\
		\src{\Omega} \ssrel_{\src{\OB{f_c}}} \SigmaB'
	\end{align*}
\end{lemma}

\begin{lemma}[Context Speculation is safe and lasts at most n]\label{thm:ctx-spec-most-omega}(\showproof{ctx-spec-most-omega})
	\begin{align*}
		\text{ if }
		&\
			\src{\Sigma} \ssrel_{\src{\OB{f_c}}} \SigmaB
		\\
		\text{ and }
		&\
			\SigmaB = \trgB{ (\Omega_b, \bot, \safeta) \cdot \OB{(\Omega', n', \unta)} \cdot (\Omega, n, \unta)}
		\\
		\text{ and }
		&\ 
			\SigmaB = \trgB{ (\Omega_b, \bot, \safeta) \cdot \OB{(\Omega', n', \unta)} \cdot (\Omega'', 0, \unta)}
		\\
		\text{ and }
		&\
			\trgB{p(\Omega(\pc))} = \trgB{l} : \trgB{i} \text{ and } \trgB{p(\Omega''(\pc))} = \trgB{l''} : \src{i''}
		\\
		\text{ and }
		&\
			\ffun{\trgB{l}}\notin\src{\OB{f_c}} \text{ and } \ffun{\trgB{l''}} \notin\src{\OB{f_c}}
		\\
		\text{ then }
		&\ 
			\SigmaB \bigspecarrowB{\tra{^{\taint}}} \SigmaB'
		\\ 
		\text{ and }
		&\
			\src{\empTr} \tracerel \trgB{\tra{^{\taint}}}
		\\ 
		\text{ and }
		&\
			\src{\Sigma} \ssrel_{\src{\OB{f_c}}} \SigmaB'
	\end{align*}
	
\end{lemma}

\begin{lemma}[Single Context Speculation is safe]\label{thm:ctx-spec-sing-safe}(\showproof{ctx-spec-sing-safe})
    \begin{align*}
        \text{ if }
        &\
            \src{\Omega} \ssrel_{\src{\OB{f_c}}} \SigmaB
        \\
        \text{ and }
        &\
            \SigmaB = \trgB{ (\Omega_b, \bot, \safeta) \cdot \OB{(\Omega', n', \unta)} \cdot (\Omega, n, \unta)}
        \\
        \text{ and }
        &\ 
            \SigmaB = \trgB{ (\Omega_b, \bot, \safeta) \cdot \OB{(\Omega', n', \unta)} \cdot (\Omega'', n - 1, \unta)}
        \\
        \text{ and }
        &\
            \trgB{p(\Omega(\pc))} = \trgB{l} : \trgB{i} \text{ and } \trgB{p(\Omega''(\pc))} = \trgB{l''} : \trgB{i''}
        \\
        \text{ and }
        &\
            \ffun{\trgB{l}}\notin\src{\OB{f_c}} \text{ and } \ffun{\trgB{l''}} \notin\src{\OB{f_c}}
        \\
        \text{ then }
        &\ 
            \SigmaB \specarrowB{\tau^{\taint}} \SigmaB'
        \\ 
        \text{ and }
        &\
            \src{\empTr} \tracerel \trgB{\tau^{\taint}}
        \\ 
        \text{ and }
        &\
            \src{\Sigma} \ssrel_{\src{\OB{f_c}}} \SigmaB'
    \end{align*}
\end{lemma}

\begin{lemma}[Expression reductions with safe registers is safe]\label{thm:exp-red-safe}(\showproof{exp-red-safe})
    \begin{align*}
        \text{ if }
            &\ 
            \vdash \trgB{A}: \shrel
        \\
        \text{ then }
            &\
            \trgB{A \triangleright e \bigreds v : \safeta}
    \end{align*}
\end{lemma}

\begin{lemma}[Compiled Speculation is safe and lasts at most n]\label{thm:comp-spec-most-omega}(\showproof{comp-spec-most-omega})
	\begin{align*}
		\text{ if }
		&\
			\src{\Sigma} \ssrel_{\src{\OB{f_c}}} \SigmaB
		\\
		\text{ and }
		&\
			\SigmaB = \trgB{ (\Omega_b, \bot, \safeta) \cdot \OB{(\Omega', n', \unta)} \cdot (\Omega, n, \unta)}
		\\
		\text{ and }
		&\ 
			\SigmaB = \trgB{ (\Omega_b, \bot, \safeta) \cdot \OB{(\Omega', n', \unta)} \cdot (\Omega'', 0, \unta)}
		\\
		\text{ and }
		&\
			\trgB{p(\Omega(\pc))} = \compuslhB{l} : \compuslhB{i} \text{ and } \trgB{p(\Omega''(\pc))} = \compuslhB{l''} : \compuslhB{i''}
		\\
		\text{ and }
		&\
			\ffun{\trgB{l}}\in\src{\OB{f_c}} \text{ and } \ffun{\trgB{l''}} \in\src{\OB{f_c}}
		\\
		\text{ then }
		&\ 
			\SigmaB \bigspecarrowB{\tra{^{\taint}}} \SigmaB'
		\\ 
		\text{ and }
		&\
			\src{\empTr} \tracerel \trgB{\tra{^{\taint}}}
		\\ 
		\text{ and }
		&\
			\src{\Sigma} \ssrel_{\src{\OB{f_c}}} \SigmaB'
	\end{align*}
	\end{lemma}

\begin{lemma}[Single Compiled Speculation is safe]\label{thm:comp-spec-sing-safe}(\showproof{comp-spec-sing-safe})
		\begin{align*}
			\text{ if }
			&\
				\src{\Sigma} \ssrel_{\src{\OB{f_c}}} \SigmaB
			\\
			\text{ and }
			&\
				\SigmaB = \trgB{ (\Omega_b, \bot, \safeta) \cdot \OB{(\Omega', n', \unta)} \cdot (\Omega, n, \unta)}
			\\
			\text{ and }
			&\ 
				\SigmaB''' = \trgB{ (\Omega_b, \bot, \safeta) \cdot \OB{(\Omega', n', \unta)} \cdot (\Omega''', n''', \unta)}
			\\
			\text{ and }
			&\
				\trgB{p(\Omega(\pc))} = \compuslhB{l} : \compuslhB{i} \text{ and } \trgB{p(\Omega''(\pc))} = \compuslhB{l''} : \compuslhB{i''}
			\\
			\text{ and }
			&\
				\ffun{\trgB{l}}\in\src{\OB{f_c}} \text{ and } \ffun{\trgB{l''}} \in\src{\OB{f_c}}
			\\
			\text{ then }
			&\ 
				\SigmaB \bigspecarrowB{\tra{^\taint}} \SigmaB'''
			\\ 
			\text{ and }
			&\
				\src{\empTr} \tracerel \trgB{\tra{^\taint}}
			\\ 
			\text{ and }
			&\
				\trgB{n''} \leq \trgB{n} 
			\\ 
			\text{ and }
			&\
				\src{\Sigma} \ssrel_{\src{\OB{f_c}}} \SigmaB'''
		\end{align*}
\end{lemma}

\subsection{Another (faulty) USLH Compiler}
We leave it here because we think it is insightful.
Previously, we used this version of the compiler 

\begin{align*}
        \compuslhB{l : i} &= \trgB{
                                    \begin{aligned}[t]
                                    &
                                        l : \passign{\rslh}{\rslh \lor \neg \rscr}
                                            \\
                                            &\ 
                                            \compuslhB{l' : i}
                                    \end{aligned}
                            } ~~ \text{if $l \in \branchlabels$}
        \\
        \compuslhB{\passign{x}{e}} &= \trgB{
                                            \begin{aligned}[t]
                                                    &
                                                    \passign{\rscr}{\compuslhB{e}}
                                                        \\
                                                        &\
                                                        \pcondassign{\rscr}{0}{\rslh}
                                                            \\
                                                            &\
                                                            \passign{x}{\rscr}
                                            \end{aligned}
                                        }
        \\
	\compuslhB{\pjz{x}{l'}} &= \trgB{
                                            \begin{aligned}[t]
                                                    &
                                                    \passign{\rscr}{x}
                                                        \\
                                                        &\
                                                        \pcondassign{\rscr}{0}{\rslh}
                                                            \\
                                                            &\
                                                            \pjz{\rscr}{l'} \\
                                                            &\
                                                            l_1 : \passign{\rslh}{\rslh \lor \rscr} \end{aligned}
                                        }
\end{align*}

We define the branch label set as $\branchlabels \bnfdef \{\ l \mid \forall i \in p, i = \pjz{x}{l} \}$.

However, this has certain problems. Consider a do while loop.

\begin{minipage}{.45\textwidth}
\begin{lstlisting}[style=MUASMstyle, caption={}, escapeinside=!!]
    !$\passign{z}{10}$!
    !$l_1 : \passign{x}{5}$!
    !$l_2 : \passign{y}{2}$!
    !$l_3 :\passign{z}{z - 1}$!
    !$l_4 :\pjz{z}{l_3}$!
    \end{lstlisting}
\end{minipage}\hfill 
\begin{minipage}{.45\textwidth}
    \begin{lstlisting}[style=MUASMstyle, escapeinside=!!]
    !$\vdots$!
    !$l_2' : \passign{\rscr}{2}$!
    !$l_2'' : \pcondassign{\rscr}{0}{\rslh}$!
    !$l_2''' : \passign{y}{\rscr}$!
    !$l_3 : \passign{\rslh}{\rslh \lor \neg \rscr}$!
    !$l_3' : \passign{\rscr}{z - 1}$!
    !$l_3'' : \pcondassign{\rscr}{0}{\rslh}$!
    !$l_3''' : \passign{z}{\rscr}$!
    !$\vdots$!
    \end{lstlisting}
\end{minipage}

Since $\rslh = \mi{false}$ but $\rscr = 5 = \mi{false}$ we get that $\rslh = \mi{false} \lor \neg \mi{false} = \mi{true}$ when executing $\l_3$ which is the countermeasure introduced because of the branch later on.
Even though no speculation is happening!
The problem is that the use of $\rscr$ is not guarded which is normally not a problem because every instruction overwrites the value in the compiler. However, in the case for the branch target, we do not know if the branch instruction is executed right before and as we can see, it is not the only way to reach the branch target in a program.

So this branch is reachable from two independent positions in the code, i.e. the start of the do while loop and then when the branch is executed.
That is why we need to reset the value of the scratch register to 0 (which equals true) so that we do not accidentally pull up the value of $\rslh$. The fix looks like this:

\begin{minipage}{.45\textwidth}
    \begin{lstlisting}[style=MUASMstyle, escapeinside=!!]
    !$\vdots$!
    !$l_2' : \passign{\rscr}{2}$!
    !$l_2'' : \pcondassign{\rscr}{0}{\rslh}$!
    !$l_2''' : \passign{y}{\rscr}$!
    !$l_2'''' : \passign{\rscr}{0}$!
    !$l_3 : \passign{\rslh}{\rslh \lor \neg \rscr}$!
    !$l_3' : \passign{\rscr}{z - 1}$!
    !$l_3'' : \pcondassign{\rscr}{0}{\rslh}$!
    !$l_3''' : \passign{z}{\rscr}$!
    !$l_3'''' : \passign{\rscr}{0}$!
    !$\vdots$!
    \end{lstlisting}
\end{minipage}

Again the problem is that one of the branches is not related to the branch instruction itself and can be reached before the branch is executed. The branch construct is stronger than a normal if, which makes sense.
 
\begin{proof}[Proof of \Thmref{thm:v1-uslh-comp-rdss}]\proofref{}{v1-uslh-comp-rdss}\hfill

    Instantiate $\src{A}$ with $\backtrfencec{\trgB{A}}$.
    This holds by \Thmref{thm:v1-corr-bt-uslh}.
\end{proof}

\BREAK

\begin{proof}[Proof of \Thmref{thm:v1-corr-bt-uslh}]\proofref{}{v1-corr-bt-uslh}\hfill

    This holds by \Thmref{thm:v1-ini-state-rel} and \Thmref{thm:v1-bwd-sim-uslh}.
\end{proof}

\BREAK

\begin{proof}[Proof of \Thmref{thm:v1-bwd-sim-uslh}]\proofref{}{v1-bwd-sim-uslh}\hfill

    There are three kinds of reductions:
    \begin{description}
        \item[compiled-to-compiled code]
        Follows from \Thmref{thm:v1-bwd-sim-comp-steps-uslh}
        \item[backtranslated-to-backtranslated]
        Follows from \Thmref{thm:v1-back-sim-bts-uslh}
        \item[compiled-to-backtranslated or backtranslated-to-compiled code]
        Analogous to the corresponding cases in \Thmref{thm:v4-bwd-sim-lfence}.
        
    \end{description}

\end{proof}

\BREAK

\begin{proof}[Proof of \Thmref{thm:v1-bwd-sim-comp-steps-uslh}]\proofref{}{v1-bwd-sim-comp-steps-uslh}\hfill

Analogous to \Thmref{thm:v4-bwd-sim-comp-steps-lfence} using \Thmref{thm:v1-fwd-sim-stm-uslh}.
\end{proof}

\BREAK

\begin{proof}[Proof of \Thmref{thm:v1-fwd-sim-stm-uslh}]\proofref{}{v1-fwd-sim-stm-uslh}\hfill

The proof proceeds by induction on $\src{i}$:
\begin{description}
        \item[skip]
        Analogous to \Thmref{thm:v4-fwd-sim-stm-lfence} together with \Thmref{thm:v1-fwd-sim-exp-uslh}.

        \item[assign-op (vassign)] 
        Then we have this source reduction: 
        $\src{\Omega \nsarrow{\opObs{v}{v'}}^\sigma } \src{\Omega'}$ by \Cref{tr:vassign}, where $\src{\exprEval{A}{e}{v : \taint}}$ and $\src{\exprEval{A}{e'}{v' : \taint'}}$ and $\sigma = \taint \lub \taint'$.
        \begin{insight}
            In this proof, we start in a state that has no speculation. The interesting case for these instructions is when speculation is ongoing and is handled in \Thmref{thm:comp-spec-sing-safe}. That is because then the predicate bit $\rslh$ is true.
        \end{insight}
        Because of $\src{\Omega} \ssrelref_{\src{\OB{f''}}} \SigmaB$, we know that $\rslh = false$ and $\src{A} \rrelref \trgB{A}$.

        By \Thmref{thm:v1-fwd-sim-exp-uslh} we get $\trgB{\exprEval{A}{\compuslhB{e}}{\compuslhB{v}: \taint}}$ and
        $\trgB{\exprEval{A}{\compuslhB{e'}}{\compuslhB{v'} : \taint'}}$ with $\src{\taint} \equiv \trgB{\taint}$ and 
        $\src{\taint'} \equiv \trgB{\taint'}$

        Let us look at the reduction generated for the compiled instruction. We annotate each instruction with the rule and show the trace that is generated at the end.
        \begin{align*}
             \compuslhB{\vassign{x}{e}{e'}} &= \trgB{
                                            \begin{aligned}[t]
                                                    &
                                                    \passign{\rscr}{\compuslhB{e}} 
                                                    ~\text{\Cref{tr:assign} uses $\trgB{\exprEval{A}{\compuslhB{e}}{\compuslhB{v}: \taint}}$}
                                                        \\
                                                        &\
                                                        \passign{\rscr'}{\compuslhB{e'}} 
                                                        ~\text{\Cref{tr:assign} uses $\trgB{\exprEval{A}{\compuslhB{e'}}{\compuslhB{v'} : \taint'}}$}
                                                            \\
                                                            &\
                                                            \pcondassign{\rscr}{0}{\rslh} 
                                                            ~\text{\Cref{tr:condup-unsat}}
                                                                \\
                                                                &\
                                                                \pcondassign{\rscr'}{0}{\rslh} 
                                                                ~\text{\Cref{tr:condup-unsat}}
                                                                    \\
                                                                    &\
                                                                    \vassign{x}{\rscr}{\rscr'} 
                                                                    ~\text{\Cref{tr:vassign}}
                                            \end{aligned}
                                        }
        \end{align*}
        
        Call the state after executing these instructions $\trgB{\SigmaB'}$.
        The generated trace is: $\tauStack = \opObs{\compuslhB{v}}{\compuslhB{v'}}^{\trgB{\taint \lub \taint'}}$, 

        We now need to show that the traces and the states are related:
        \begin{description}
        
            \item[$\src{{\opObs{v}{v'}}^\sigma} \tracerel \trgB{\opObs{\compuslhB{v}}{\compuslhB{v'}}^{\trgB{\taint \lub \taint'}}}$]

            This follows from \Cref{tr:ac-rel-vassign} together with \Thmref{thm:v1-val-rel-comp-uslh} for $\src{v} \vrel \compuslhB{v}$ and $\src{v'} \vrel \compuslhB{v'}$, while $\src{\taint \lub \taint'} \equiv \trgB{\taint \lub \taint'}$ follows from \Thmref{lemma:bounds-preserve-relation}.
            
            \item[$\src{\Omega'} \ssrelref_{\src{\OB{f''}}}\ \trgB{\SigmaB'}$] 

            Since the only changes to the state are the registers it remains to show that $\src{A'} \rrelref \trgB{A'}$:
            
            We know $\src{A'} = A[x \mapsto v \binop v' : \sigma]$ and $\trgB{A'} = \trgB{A[x \mapsto \compuslhB{v} \binop \compuslhB{v'} : \sigma, \rscr \mapsto \compuslhB{v} : \taint, \rscr' \mapsto \compuslhB{v'} : \taint']}$ are the changes to the state.
            
            Using \Thmref{thm:v1-val-rel-comp-uslh} we can derive $\src{A'} \rrel \trgB{A'}$, since we can ignore $\rscr$ and  $\rscr'$.
        \end{description}

        \item[assign, load, load-prv, store, store-prv]
        
        Analogous to the vassign case above

        \item[call, call-internal, call-callback] 

        Let us look at the compilation: 
        $\compuslhB{l : \pcall{f}} = \trgB{
										\begin{aligned}[t]
												&
												l :\passign{\rslhC}{\rslh} \\
													&\
													l' :\pcall{f}
														\\
														&\ l'' : \passign{\rslh}{\rslhC}
										\end{aligned}
									} ~ \text{Here $f$ is a function name}$
        Since, the compilation changed the prologue of every function we additionally have
        $l_{new} : \passign{\rslh}{\rslhC}$ at the beginning of the called function $f$.
        This only changes $\rslh$ and the program counter is increased.
        
        By definition of the compiler the $\pc$ now points to original beginning of the function.

        Furthermore, we have $\trgB{\OB{B'''} = \OB{B''}; l''}$ and $\src{\OB{B'} = \OB{B}; next(l)}$
        Note that $next(l'') = next(l)$ because of how the compiler works.
        Thus, we can relate $\trgB{\OB{B'''}} \brel_f \src{\OB{B'}}$ via \Cref{tr:v1-brel-region-comp}, since $f \in \OB{F}$ (we are only in compiled code). 
        We thus have $\src{\Omega'} \ssrelref_{\src{\OB{f''}}} \SigmaB'$.

        \item[ret, ret-internal]
        Because of $\src{\Omega} \ssrelref_{\src{\OB{f''}}} \SigmaB$, we know that $\rslh = false$ and $\src{A} \rrelref \trgB{A}$.

        Let us look at the compiled return:
        \trgB{
            \begin{aligned}[t]
                    &
                    l :\passign{\rslhC}{\rslh} \\
                        &\
                        l' :\pret
            \end{aligned}
        }
        Next, we know that according to the compilation of call instruction (See above for details) that the speculation flag is assigned after the call instruction with $l'' : \passign{\rslh}{\rslhC}$. That is where the $\pc$ is pointing to.

        Thus, the $\pc$ register is updated in the same way which means we can show $A' \arel \trgB{A'}$. Since only $\rslhC$ and $\rslh$ were updated apart from $\pc$.

        Now by definition of the compiler  we have that the $\pc$ points to the same instruction after the return as for the initial source program.

        Next, we have $\trgB{\OB{B''} = \OB{B'''}; n}$ and $\src{\OB{B} = \OB{B'}; n}$.

        Thus, we have $\trgB{\OB{B'''}} \brel_f \src{\OB{B'}}$ from $\src{\Omega} \ssrelref_{\src{\OB{f''}}} \SigmaB$.

        Since these are the only changes to the state we have.
        $\src{\Omega'} \ssrelref_{\src{\OB{f''}}} \SigmaB'$

        \item[beqz]

            W.l.o.g assume that $\src{A(x) = 0 : \sigma}$.
            Then we have this source reduction: 
            $\src{\Omega \nsarrow{\pcObs{l'}^\sigma }} \src{\Omega'}$ by \Cref{tr:beqz-sat}.
    
            By $\src{A} \rrel \trgB{A}$ from $\src{\Omega} \ssrelref_{\src{\OB{f''}}}\ \SigmaB$ we have $\trgB{A(x) = 0 : \sigma}$ with $\src{\sigma} \equiv \trgB{\sigma}$ and $\rslh = \mi{false}$.

            The compiled statement looks like this:
            \begin{align*}
               \compuslhB{l : \pjz{x}{l'}} &= \trgB{
                                            \begin{aligned}[t]
                                                    &
                                                    l_1 : \passign{\rscr}{x}
                                                     ~\text{\Cref{tr:assign} uses $\trgB{\exprEval{A}{\compuslhB{x}}{\compuslhB{0}: \sigma}}$}
                                                        \\
                                                        &\
                                                        l_1' : \pcondassign{\rscr}{0}{\rslh}
                                                        ~\text{\Cref{tr:condup-unsat}}
                                                            \\
                                                            &\
                                                           l_1'' :  \pjz{\rscr}{l_{new}} ~\text{\Cref{tr:v1-spec}}
                                                                \\
                                                                &\
                                                                l_2 : \passign{\rslh}{\rslh \lor \rscr} ~\text{\Cref{tr:assign}}
                                                                    \\
                                                                    &\
                                                                    l_2' : \pjmp{next(l)}
                                                                    ~\text{\Cref{tr:jmp}}
                                                                        \\
                                                                        &\ 
                                                                        l_{new} : \passign{\rslh}{\rslh \lor \neg \rscr}
                                                                         ~\text{\Cref{tr:assign}}
                                                                            \\ 
                                                                            &\ 
                                                                            l_{new}' : \pjmp(l')
                                                                             ~\text{\Cref{tr:jmp}}
                                            \end{aligned}
                                        }
            \end{align*}
            Because $\rscr = 0$ since $\rslh = \mi{false}$ and $\trgB{A(x) = 0 : \sigma}$ the speculation continues with $l_2$.

            Furthermore, we know $\trgB{\exprEval{A}{\compuslhB{\rslh \lor \rscr}}{\compuslhB{\mi{true}}: \taint}}$.
            Next the $\pjmp{\mi{next}(l)}$ is executed.

            Thus, after execution of $l_2 : \passign{\rslh}{\rslh \lor \rscr}$ and $l_2' : \pjmp{\mi{next}(l)}$we have the state 
            $\SigmaB'$ with $\trgB{A'} = A[\pc \mapsto \mi{next}(l) : \safeta, \rslh = \mi{true} : \safeta, \rscr \mapsto 0 : \sigma]$ and the trace $\trgB{\tauStack} = \trgB{\pcObs{l_{new}}^{\sigma} \cdot \pcObs{\src{\mi{next}(l)}}^{\safeta} }$.

            It is easy to show that $\src{\Omega'} \ssrelref_{\src{\OB{f''}}}\ \trgB{\SigmaB'}$ since the first states are related at $\ssrelref_{\src{\OB{f''}}}$ by assumption and the speculating states in $\trgB{\SigmaB'}$ only had safe register updates.

            By \Thmref{thm:spec-rel-satte-safe} we know:
            \begin{enumerate}
                \item $\trgB{\SigmaB' \bigspecarrowB{\tra{^\taint}} \SigmaB''}$ and 
                \item $\src{\empTr} \tracerel \trgB{\tra{^{\taint}}}$ and
                \item $\src{\Sigma} \ssrel_{\src{\OB{f_c}}} \SigmaB''$
            \end{enumerate}
            $\trgB{\SigmaB''}$ is the state after the speculation was finished.
            
            This means that it is the state where the branch is taken so $\trgB{A''(\pc)} = \trgB{l_{new}}$
            and execution continues with these two instructions.
            
            \begin{align*}
                &
                l_{new} : \passign{\rslh}{\rslh \lor \neg \rscr} ~\text{\Cref{tr:assign}}
                \\ 
                &\ 
                l_{new}' : \pjmp(l') ~\text{\Cref{tr:jmp}}
            \end{align*}
            Remember that $\rscr = 0$ and $\rslh = \mi{false}$.

            Thus, we know $\trgB{\exprEval{A}{\compuslhB{\rslh \lor \neg \rscr}}{\compuslhB{\mi{false}}: \taint}}$.
            Next, the $\pjmp{l'}$ is executed.
            
            Call this state $\SigmaB'''$ which has $\trgB{A'''} = \trgB{A''[\pc \mapsto l' : \safeta, \rslh \mapsto \mi{false} : \safeta]}$ and the trace generated $\trgB{\tauStack'} = \trgB{\pcObs{l'}^{\safeta}}$.
             We now need to show that the traces and the states are related:
             \begin{description}
        
                \item[$\src{{\pcObs{l'}^\sigma}} \tracerel \trgB{\tauStack \cdot \tra{^\taint} \cdot \tauStack'}$]

                We can relate the different parts of the trace:
                \begin{enumerate}
                    
                    \item[$\src{\empTr} \tracerel \trgB{\tauStack}$]  
                    Since $\trgB{\tauStack} = \trgB{\pcObs{l_{new}}^{\sigma} \cdot \pcObs{\src{l + 1}}^{\safeta} }$ and $\src{\sigma} \equiv \trgB{\sigma}$ and $\src{\sigma} = \safeta$ we can relate by \Cref{tr:tr-rel-safe-a}.
                    Since $\trgB{\tauStack} = \trgB{\pcObs{l_{new}}^{\sigma} \cdot \pcObs{\src{l + 1}}^{\safeta} }$.
                    \item[$\src{\empTr} \tracerel \trgB{\tra{^\taint}}$]
                    from $\src{\empTr} \tracerel \trgB{\tra{^{\taint}}}$ above
                    \item[$\src{\empTr} \tracerel \trgB{\tauStack'}$]
                    Since $\trgB{\tauStack'} = \trgB{\pcObs{l'}^{\safeta}}$ we can relate by \Cref{tr:tr-rel-safe-a}.
                \end{enumerate}
                
                \item[$\src{\Omega'} \ssrelref_{\src{\OB{f''}}}\ \trgB{\SigmaB'''}$] 
                From above we have $\src{\Omega} \ssrel_{\src{\OB{f_c}}} \SigmaB''$. The only update to $\trgB{\SigmaB'''}$ is $\trgB{A'''} = \trgB{A''[\pc \mapsto l' : \safeta, \rslh \mapsto \mi{false} : \safeta]}$.
                The update to $\rslh$ is ignored by $\rrel$ and the change to $\src{\Omega'}$ is $\src{A'} = \src{A[\pc \mapsto l']}$ by \Cref{tr:beqz-sat}. 

                Thus, we have $\src{A'} \rrel \trgB{A'''}$ and thus $\src{\Omega'} \ssrelref_{\src{\OB{f''}}}\ \trgB{\SigmaB'''}$.
                
            \end{description}

\end{description}
\end{proof}

\begin{proof}[Proof of \Thmref{thm:v1-fwd-sim-exp-uslh}]\proofref{}{v1-fwd-sim-exp-uslh}\hfill

Analogous to \Thmref{thm:v4-fwd-sim-exp-lfence}.

\end{proof}

\BREAK

\begin{proof}[Proof of \Thmref{thm:v1-back-sim-bts-uslh}]\proofref{}{v1-back-sim-bts-uslh}\hfill

Analogous to \Thmref{thm:v4-back-sim-bts-lfence} using \Thmref{thm:v4-back-sim-bte-lfence}.
\end{proof}

\BREAK

\begin{proof}[Proof of \Thmref{thm:v1-ini-state-rel}]\proofref{}{v1-ini-state-rel}\hfill

Analogous to \Thmref{thm:v4-ini-state-rel} but with \Thmref{thm:v1-heap-rel-comp-uslh}, \Thmref{thm:v1-heap-rel-bt-uslh} and \Thmref{thm:v1-val-rel-comp-uslh}
\end{proof}

\BREAK

\begin{proof}[Proof of \Thmref{thm:v1-val-rel-comp-uslh}]\proofref{}{v1-val-rel-comp-uslh}\hfill

Trivial analysis of the compiler.
\end{proof}

\BREAK 

\begin{proof}[Proof of \Thmref{thm:v1-heap-rel-comp-uslh}]\proofref{}{v1-heap-rel-comp-uslh}\hfill

Trivial analysis of the compiler.
\end{proof}

\BREAK 

\begin{proof}[Proof of \Thmref{thm:v1-val-rel-bt-uslh}]\proofref{}{v1-val-rel-bt-uslh}\hfill

Trivial analysis of the backtranslation.
\end{proof}

\BREAK

\begin{proof}[Proof of \Thmref{thm:v1-heap-rel-bt-uslh}]\proofref{}{v1-heap-rel-bt-uslh}\hfill

Trivial analysis of the backtranslation.
\end{proof}

\BREAK

\begin{proof}[Proof of \Thmref{thm:spec-rel-satte-safe}]\proofref{}{spec-rel-satte-safe}\hfill

    The proof proceeds by induction on the stack of configurations:
    \begin{description}
        \item[Empty Stack]
        By \Thmref{thm:spec-most-omega} we get:
        \begin{enumerate}
            \item $\SigmaB \bigspecarrowB{\tra{^{\taint}}} \SigmaB''$ and
            \item $\src{\empTr} \tracerel \trgB{\tra{^{\taint}}}$ and
            \item $ \src{\Sigma} \ssrel_{\src{\OB{f_c}}} \SigmaB''$
        \end{enumerate}
        Furthermore, we know that $\SigmaB'' =  \trgB{ (\Omega_b, \bot, \safeta) \cdot \OB{(\Omega', n', \unta)} \cdot (\Omega'', 0, \unta)}$ and thus a rollback reduction applies.
        This reduction can then be related by \Cref{tr:ac-rel-rlb}
        \item[Non-Empty Stack] 
        This holds by IH and the same reasoning as in the base case.
    \end{description}
\end{proof}

\BREAK

\begin{proof}[Proof of \Thmref{thm:spec-most-omega}]\proofref{}{spec-most-omega}\hfill

     We proceed by case analysis on $\trgB{f}$ and $\trgB{f''}$::
    \begin{description}
        \item[$\trgB{f}, \trgB{f''} \in \src{\OB{f''}}$ (\textbf{in the compiled component})]
        This holds by \Thmref{thm:comp-spec-most-omega}.

        \item[$\trgB{f}, \trgB{f''} \notin \src{\OB{f''}}$ (\textbf{in the context})]
         This holds by \Thmref{thm:ctx-spec-most-omega}.

        \item[$\trgB{f} \in \src{\OB{f''}}$ and $\trgB{f''} \notin \src{\OB{f''}}$ compiled to context]

            We proceed by induction on $\trgB{n}$:
            \begin{description}
                \item[$\trgB{n} = 0$]
                Trivial by \Cref{tr:v1-init}.
                \item[$\trgB{n} = \trgB{n_1 + 1}$] 

                    This switch from component to context arises in two cases:
                    \begin{description}
                        
                        \item[\textbf{call}]
            
                            We derive $\trgB{\SigmaB \bigspecarrowB{\tra{^{\taint}}} \SigmaB'}$ in combination by \Cref{tr:v1-single} together with \Cref{tr:v1-nospec-act} and \Cref{tr:call} for $\src{\Omega \nsarrow{\clh{}^{\safeta}} \Omega'}$.
                            \begin{description}
                                \item[$\src{\empTr} \tracerel \trgB{\clh{}^{\safeta}}$] 
                                    Holds trivially by \Cref{tr:ac-rel-cl}. Because call actions are always safe (there is no argument to be evaluated).
                                \item[$\src{\Omega} \ssrel_{\src{\OB{f_c}}} \SigmaB'$] 
                                A call only updates the program counter with a safe value. Thus, $\vdash \trgB{A'}: \shrel$ trivially holds and we have $\src{\Omega} \ssrel_{\src{\OB{f_c}}} \SigmaB'$.
                            \end{description}
                                        
                        \item[\textbf{return}]
                            Analogous to the call case above together with \Cref{tr:ac-rel-rt}. Because return actions are always safe.
                    \end{description}
            \end{description}
                 
        \item[$\trgB{f} \notin \src{\OB{f''}}$ and $\trgB{f''} \in \src{\OB{f''}}$ context to compiled]

            We proceed by induction on $\trgB{n}$:
            \begin{description}
                \item[$\trgB{n} = 0$]
                Trivial by \Cref{tr:v1-init}.

                \item[$\trgB{n} = \trgB{n_1 + 1}$] 
                This switch from context to component arises in two cases:
                        
                    \begin{description}
                        \item[\textbf{call}]
                        
                        This is a call from a context function to a compiled function.
                        Analogous to the call case above.

                        \item[\textbf{return}]
                        Analogous to the call case above
                    \end{description}
            \end{description}

    \end{description}
\end{proof}

\BREAK

\begin{proof}[Proof of \Thmref{thm:ctx-spec-most-omega}]\proofref{}{ctx-spec-most-omega}\hfill

By induction on the speculation window $n$.
\begin{description}
    \item[$n = 0$]
    Trivial by using \Cref{tr:v1-reflect}, since $\SigmaB' = \SigmaB$.
    
    \item[$n = n' + 1$] 
    Using \Thmref{thm:ctx-spec-sing-safe} we get: 
    \begin{enumerate}
        \item $\SigmaB \specarrowB{\tau^{\taint}} \SigmaB''$ and 
        \item $\src{\empTr} \tracerel \trgB{\tau^{\taint}}$ and
        \item $\src{\Sigma} \ssrel_{\src{\OB{f_c}}} \SigmaB''$
    \end{enumerate}
    We can now apply the IH on $\Sigma''$ and get:
    \begin{enumerate}
        \item $\SigmaB'' \bigspecarrowB{\tra{^{\taint}}} \SigmaB'$ and
        \item $\src{\empTr} \tracerel \trgB{\tra{^{\taint'}}}$ and
        \item $\src{\Sigma} \ssrel_{\src{\OB{f_c}}} \SigmaB''$
    \end{enumerate}
    
    We now use \Cref{tr:v1-single} on$\SigmaB \specarrowB{\tau^{\taint}} \SigmaB''$ and $\SigmaB'' \bigspecarrowB{\tra{^{\taint}}} \SigmaB'$ and are finished. 
\end{description}
\end{proof}

\BREAK

\begin{proof}[Proof of \Thmref{thm:ctx-spec-sing-safe}]\proofref{}{ctx-spec-sing-safe}\hfill

    By induction on the current instruction $\trgB{i}$:
    All cases are trivial.
    All expressions evaluate to $\safeta$ due to \Thmref{thm:exp-red-safe}.
    No reduction can load $\unta$ values because the context cannot use $\pstoreprv{x}{e}$ or $\ploadprv{x}{e}$.
    All actions are tagged $\safeta$ and thus they can be related to $\src{\empTr}$ by \Cref{tr:tr-rel-safe-a}
\end{proof}

\BREAK

\begin{proof}[Proof of \Thmref{thm:exp-red-safe}]\proofref{}{exp-red-safe}\hfill

    By induction on $\trgB{e}$:
    \begin{description}
        \item[$\trgB{e} = \trgB{v}$]
        Trivial since values are safe.
        \item[$\trgB{e} =\trgB{e_1} \otimes \trgB{e_2}$, $\trgB{e} = \ominus \trgB{e_2}$ ]
        Trivial application of the IH.
        \item[$\trgB{e} = \trgB{r}$]
        Follows from $\vdash \trgB{A}: \shrel$.
    \end{description}
\end{proof}

\BREAK

\begin{proof}[Proof of \Thmref{thm:comp-spec-most-omega}]\proofref{}{comp-spec-most-omega}\hfill

By induction on the speculation window $n$.
\begin{description}
    \item[$n = 0$]
    Trivial by using \Cref{tr:v1-reflect}, since $\SigmaB' = \SigmaB$.
    
    \item[$n = n' + 1$] 
    Using \Thmref{thm:comp-spec-sing-safe} on $\SigmaB$ we get: 
    \begin{enumerate}
        \item $\SigmaB \specarrowB{\tau^{\taint}} \SigmaB''$ and 
        \item $\src{\empTr} \tracerel \trgB{\tau^{\taint}}$ and
        \item $\src{\Sigma} \ssrel_{\src{\OB{f_c}}} \SigmaB''$
    \end{enumerate}
    We can now apply the IH on $\Sigma''$ and get:
    \begin{enumerate}
        \item $\SigmaB'' \bigspecarrowB{\tra{^{\taint}}} \SigmaB'$ and
        \item $\src{\empTr} \tracerel \trgB{\tra{^{\taint'}}}$ and
        \item $\src{\Sigma} \ssrel_{\src{\OB{f_c}}} \SigmaB''$
    \end{enumerate}
    
    We now use \Cref{tr:v1-single} on $\SigmaB \specarrowB{\tau^{\taint}} \SigmaB''$ and $\SigmaB'' \bigspecarrowB{\tra{^{\taint}}} \SigmaB'$ and are finished. 
\end{description}

\end{proof}

\BREAK

\begin{proof}[Proof of \Thmref{thm:comp-spec-sing-safe}]\proofref{}{comp-spec-sing-safe}\hfill

    We do an induction on the speculation window $\trgB{n}$:

    \begin{description}
        \item[$n = 0$]
        Trivial by \Cref{tr:v1-reflect}.

        \item[$n = n' + 1$]  
        By case analysis on the current instruction $\trgB{\compuslhB{i}}$:
        \begin{description}
            \item[skip]
            Trivial
    
            \item[assign-op (vassign)] 
    
            Because of $\src{\Omega} \ssrelref_{\src{\OB{f''}}} \SigmaB$, we know that $\rslh = \mi{true}$ and $\vdash \trgB{A}: \shrel$.
            
            Let us look at the reduction generated for the compiled instruction. We annotate each instruction with the rule and show the trace generated at the end. We only show the non-speculative rule used that is then used in either \Cref{tr:v1-nospec-eps} or \Cref{tr:v1-nospec-act}.
            \begin{align*}
                 \compuslhB{\vassign{x}{e}{e'}} &= \trgB{
                                                \begin{aligned}[t]
                                                        &
                                                        \passign{\rscr}{\compuslhB{e}} 
                                                        ~\text{\Cref{tr:assign} uses $\trgB{\exprEval{A}{\compuslhB{e}}{ v: \taint}}$}
                                                            \\
                                                            &\
                                                            \passign{\rscr'}{\compuslhB{e'}} 
                                                            ~\text{\Cref{tr:assign} uses $\trgB{\exprEval{A}{\compuslhB{e'}}{v' : \taint'}}$}
                                                                \\
                                                                &\
                                                                \pcondassign{\rscr}{0}{\rslh} 
                                                                ~\text{\Cref{tr:condup-sat} since $\rslh = \mi{true}$}
                                                                    \\
                                                                    &\
                                                                    \pcondassign{\rscr'}{0}{\rslh} 
                                                                    ~\text{\Cref{tr:condup-sat} since $\rslh = \mi{true}$}
                                                                        \\
                                                                        &\
                                                                        \vassign{x}{\rscr}{\rscr'} 
                                                                        ~\text{\Cref{tr:vassign}}
                                                \end{aligned}
                                            }
            \end{align*}
            
            Call the state after executing these instructions $\trgB{\SigmaB'}$.
            \begin{insight}
                Note that the interesting case is if the speculation window is big enough to account for these 5 steps. This happens if the $n'$ is chosen big enough. If it is not, then an early rollback will be triggered and it is still fine. We only show the interesting case when the speculation window $n$ is big enough. It is easy to see that if it works when all reductions go through that it works when it the reductions are aborted earlier since then we have a subtrace of the original safe trace.
                That is why slh protects. The potentially dangerous values are replaced by safe values before they can be used.
            \end{insight}
            The generated trace is: $\tauStack = \opObs{0}{0}^{\safeta}$, 
    
            We now need to show that the traces and the states are related:
            \begin{description}
            
                \item[$\src{\empTr} \tracerel \trgB{\opObs{0}{0}^{\trgB{\safeta}}}$]
    
                This follows from \Cref{tr:tr-rel-safe-a} since the action is safe because only the constants $0$ are evaluated, which is safe.
                
                \item[$\src{\Omega'} \ssrelref_{\src{\OB{f''}}}\ \trgB{\SigmaB'}$] 
                Since the only changes to the state are the registers it remains to show that $\vdash \trgB{A'}: \shrel$:
                
                We know  $\trgB{A'} = \trgB{A[x \mapsto 0 \binop 0 : \safeta, \rscr \mapsto 0 : \safeta, \rscr' \mapsto 0 : \safeta]}$ are the changes to the state.
    
                This means we trivially satisfy $\vdash \trgB{A'}: \shrel$ and are done.
    
                \item[$\trgB{n''} \leq \trgB{n} $]
                Each step reduces the speculation window by 1. This means that $\trgB{n''} = n - 5$ and we are finished.

            \end{description}
            
            \item[assign, load, load-prv, store, store-prv, jmp]
            
            Analogous to the vassign case above

            \item[call, call-internal, call-callback] 
            
            Because of $\src{\Omega} \ssrelref_{\src{\OB{f''}}} \SigmaB$, we know that $\rslh = \mi{true}$ and $\vdash \trgB{A}: \shrel$.
            
            The compiled statement looks like this:
            $\compuslhB{l : \pcall{f}} = \trgB{
                \begin{aligned}[t]
                        &
                        l :\passign{\rslhC}{\rslh} ~\text{\Cref{tr:assign}}\\
                            &\
                            l' :\pcall{f}
                                \\
                                &\ l'' : \passign{\rslh}{\rslhC}
                \end{aligned}
            } ~ \text{Here $f$ is a function name}$
            
            Since the compiler changes the function prologue of every function an additional assign: $l''' :\passign{\rslh}{\rslhC}$ is executed at the start of $f$.

            Lets call the state after executing these instructions $\trgB{\SigmaB'}$.
            We now need to show that the traces and the states are related:
            \begin{description}
            
                \item[$\src{\empTr} \tracerel \trgB{\callObs{f}^{\trgB{\safeta}}}$]
    
                This follows from \Cref{tr:tr-rel-safe-a} since the action is safe.
                
                \item[$\src{\Omega'} \ssrelref_{\src{\OB{f''}}}\ \trgB{\SigmaB'}$] 
                Since the only changes to the state are the registers it remains to show that $\vdash \trgB{A'}: \shrel$:

                We know  that the only changes are for the registers $\rslh$, $\rslhC$ which we can ignore and the $\pc$ register which is always safe.
                
                Since $\rslh = \mi{true}$ we have that $\rslhC = \mi{true}$.
                Thus, when updating $\rslh$ at $l'''$, we have $\rslh = \mi{true}$.

                This means we satisfy $\vdash \trgB{A'}: \shrel$ and are done.
    
                \item[$\trgB{n''} \leq \trgB{n} $]
                Each step reduces the speculation window by 1. This means that $\trgB{n''} = n - 3$ and we are finished.

            \end{description}

            \item[ret, ret-internal, ret-retback]
            
            Because of $\src{\Omega} \ssrelref_{\src{\OB{f''}}} \SigmaB$, we know that $\rslh = \mi{true}$ and $\vdash \trgB{A}: \shrel$.
            
            The compiled statement looks like this:
            $\compuslhB{l : \pret} = \trgB{
                \begin{aligned}[t]
                        &
                        l :\passign{\rslhC}{\rslh} ~\text{\Cref{tr:assign}}  \\
                            &\
                            l' :\pret
                \end{aligned}
            }$
            Furthermore, we execute the additional instruction after the call (See above) $ l'' : \passign{\rslh}{\rslhC}$.

            Lets call the state after executing these instructions $\trgB{\SigmaB'}$.
            We now need to show that the traces and the states are related:
            \begin{description}
            
                \item[$\src{\empTr} \tracerel \trgB{\retObs{}^{\trgB{\safeta}}}$]
    
                This follows from \Cref{tr:tr-rel-safe-a} since the action is safe.
                
                \item[$\src{\Omega'} \ssrelref_{\src{\OB{f''}}}\ \trgB{\SigmaB'}$] 
                Since the only changes to the state are the registers it remains to show that $\vdash \trgB{A'}: \shrel$:

                We know  that the only changes are for the registers $\rslh$, $\rslhC$ which we can ignore and the $\pc$ register which is always safe.

                Since $\rslh = \mi{true}$ we have that $\rslhC = \mi{true}$.
                Thus, when updating $\rslh$ at $l''$, we have $\rslh = \mi{true}$.
    
                This means we satisfy $\vdash \trgB{A'}: \shrel$ and are done.
    
                \item[$\trgB{n''} \leq \trgB{n} $]
                Each step reduces the speculation window by 1. This means that $\trgB{n''} = n - 3$ and we are finished.

            \end{description}
                    
        \item[beqz]

            Because of $\src{\Omega} \ssrelref_{\src{\OB{f''}}} \SigmaB$, we know that $\rslh = \mi{true}$ and $\vdash \trgB{A}: \shrel$.
            
            The compiled statement looks like this:
            \begin{align*}
               \compuslhB{l : \pjz{x}{l'}} &= \trgB{
                                            \begin{aligned}[t]
                                                    &
                                                    l_1 : \passign{\rscr}{x}
                                                     ~\text{\Cref{tr:assign} uses $\trgB{\exprEval{A}{\compuslhB{x}}{v: \taint}}$}
                                                        \\
                                                        &\
                                                        l_1' : \pcondassign{\rscr}{0}{\rslh}
                                                        ~\text{\Cref{tr:condup-sat}}
                                                            \\
                                                            &\
                                                           l_1'' :  \pjz{\rscr}{l_{new}} ~\text{\Cref{tr:v1-spec}}
                                                                \\
                                                                &\
                                                                l_2 : \passign{\rslh}{\rslh \lor \rscr} ~\text{\Cref{tr:assign}}
                                                                    \\
                                                                    &\
                                                                    l_2' : \pjmp{next(l)}
                                                                    ~\text{\Cref{tr:jmp}}
                                                                        \\
                                                                        &\ 
                                                                        l_{new} : \passign{\rslh}{\rslh \lor \neg \rscr}
                                                                         ~\text{\Cref{tr:assign}}
                                                                            \\ 
                                                                            &\ 
                                                                            l_{new}' : \pjmp(l')
                                                                             ~\text{\Cref{tr:jmp}}
                                            \end{aligned}
                                        }
            \end{align*}
            Because $\rscr = 0$ since $\rslh = \mi{true}$ the speculation continues with $l_2$.

            Furthermore, we know $\trgB{\exprEval{A}{\compuslhB{\rslh \lor \rscr}}{\compuslhB{\mi{true}}: \taint}}$ since $\rslh = \mi{true}$.
            
            Next, the $\l_2' : \pjmp{\mi{next}(l)}$ is executed.

            Thus, after execution of $l_2 : \passign{\rslh}{\rslh \lor \rscr}$ and $l_2' : \pjmp{\mi{next}(l)}$ we have the state $\SigmaB'$ with:
            \begin{align*}
                \trgB{A'} =& A[\pc \mapsto l + 1 : \safeta, \rslh = \mi{true} : \safeta, \rscr \mapsto 0 : \safeta]
                \\
                \text{and the trace}
                \\
                \trgB{\tauStack} =& \trgB{\pcObs{l_{new}}^{\safeta} \cdot \pcObs{\src{l + 1}}^{\safeta} }
            \end{align*}

            Note that $\rscr$ is tainted safe because of \Cref{tr:t-vassign} and the fact that constants are safe by \Cref{tr:T-val}.

            Each of the steps reduced the speculation window of the speculative instance. However, \Cref{tr:v1-spec} created a new speculative instance that we now need to handle.

            First, we want to apply the IH (since the speculation window did get smaller). To apply it we need two things:
            \begin{enumerate}
                \item $\src{\Omega} \ssrelref_{\src{\OB{f''}}}\ \trgB{\SigmaB'}$ and
                \item $ \SigmaB'' = \trgB{ (\Omega_b, \bot, \safeta) \cdot \OB{(\Omega', n', \unta)} \cdot (\Omega'', n'' - 3, \unta) \cdot (\Omega''', 0, \unta)}$
            \end{enumerate}

            Since the only changes to the state are the registers it remains to show that $\vdash \trgB{A'}: \shrel$:
                
            We know  $\trgB{A'} = \trgB{A[\pc \mapsto l + 1 : \safeta, \rslh \mapsto \mi{true} : \safeta, \rscr \mapsto 0 : \safeta]}$ are the changes to the state.
    
            This means we trivially satisfy $\vdash \trgB{A'}: \shrel$ and have $\src{\Omega} \ssrelref_{\src{\OB{f''}}}\ \trgB{\SigmaB'}$.
            
            By \Thmref{thm:spec-ends} there is a state $\trgB{\SigmaB''}$ such that $\trgB{\SigmaB' \bigspecarrowB{\tra{^\taint}} \SigmaB''}$ and we know that the speculation window of the speculative instance created by the $\jzC{}$ instruction is $0$.

            Now we can apply the IH and get:
            \begin{enumerate}
                \item $\trgB{\SigmaB' \bigspecarrowB{\tra{^\taint}''} \SigmaB''}$ and
                \item $\src{\empTr} \tracerel \trgB{\tra{^\taint}}$ and 
                \item $\src{\Sigma} \ssrel_{\src{\OB{f_c}}} \SigmaB''$
                \item $0 \leq \trgB{n}$
            \end{enumerate}
            Since the language is deterministic, we have $\tra{^\taint} = \tra{^\taint}''$.

            We can apply \Cref{tr:v1-rollback} on $\trgB{\SigmaB''}$ and finish the speculation started by the $\jzC$ instruction. Call this state $\trgB{\SigmaB'''}$.

            Thus, the full trace is $\trgB{\tauStack} = \trgB{\pcObs{l_{new}}^{\safeta} \cdot \pcObs{\src{l + 1}}^{\safeta}  \cdot \tra{^\taint} \cdot \rollbackObsB}$

            \begin{description}
            
                \item[$\src{\empTr} \tracerel \trgB{\pcObs{l_{new}}^{\safeta} \cdot \pcObs{\src{l + 1}}^{\safeta}  \cdot \tra{^\taint} \cdot \rollbackObsB}$]
    
                $\src{\empTr} \tracerel \trgB{\pcObs{l_{new}}^{\safeta}}$ follows from \Cref{tr:tr-rel-safe-a} since the action is safe because only the scratch register $\rscr$ was evaluated during the execution of the $\jzC$ instruction. Since this register always contains the value 0 here, it is safe.

                $\src{\empTr} \tracerel \trgB{\pcObs{\src{l + 1}^{\safeta}}}$ follows from \Cref{tr:tr-rel-safe-a}.
                
                $\src{\empTr} \tracerel \trgB{\tra{^\taint}}$ follows from IH.
                
                $\src{\empTr} \tracerel \trgB{\rollbackObsB}$ follows from \Cref{tr:ac-rel-rlb}.

                \item[$\src{\Omega} \ssrelref_{\src{\OB{f''}}}\ \trgB{\SigmaB'''}$] 
                Since a rollback just deletes a speculative instance we have
                \\
                $\trgB{\SigmaB'''} = \trgB{ (\Omega_b, \bot, \safeta) \cdot \OB{(\Omega', n', \unta)} \cdot (\Omega'', n'' - 3, \unta)}$ which is the same as $\SigmaB''$.

                Since $\src{\Omega'} \ssrelref_{\src{\OB{f''}}}\ \trgB{\SigmaB'''}$ as explained above we are finished.
    
                \item[$\trgB{n''' - 3} \leq \trgB{n} $]
                Trivial.
                
            \end{description}
    
    \end{description}
\end{description}
\end{proof}

We note there is a measure similar to $\minWindow{}$ that reduces with each step (modulo the general steps in the semantics). 

Here we state $\textit{max spec}(\SigmaB) = \SigmaB . n!$.
We use $!$ because each nested speculative transaction can create new nested transactions.

Furthermore, note that we look at terminating executions here. So speculation necessarily has to end.
\begin{lemma}[Speculation lasts at most n]\label{thm:spec-ends-comp}
If
\begin{enumerate}
    \item $\textit{max spec}(\SigmaB) = n$
\end{enumerate}
Then $\exists \SigmaB'$ such that
\begin{enumerate}
    \item $\trgB{\SigmaB \bigspecarrowB{\tra{^{\taint}}} \SigmaB'}$ \text{ and }
    \item $\trgB{\SigmaB'} = \trgB{(\Omega, \bot, \safeta)}$ \end{enumerate}
\end{lemma}
\begin{proof}
    We proceed by induction on n:
    \begin{description}
        \item[n = 0]
        This means there are no more speculative steps in the execution.
        Choose $\trgB{\Sigma'} = \trgB{\Sigma}$ and we are trivially done.

        \item[n = n' + 1] 
        We proceed by case distinction if $\SigmaB$ is currently speculating or not (Technically we do a case distinction on the depth of the stack)
        \begin{description}
            \item[No speculation]
            Choose $\SigmaB' = \SigmaB$ and we are trivially done.
            
            \item[Speculation]
            Then $\SigmaB$ can do some step to $\SigmaB \specarrowB{} \SigmaB''$ and $\textit{max spec}(\SigmaB'') = \trgB{n''}$ where $n'' < n$.
            We now apply the IH for $\SigmaB''$ and get:
            \begin{enumerate}
                \item $\SigmaB'' \bigspecarrowB{} \SigmaB'$ \text{ and }
                \item $\SigmaB' = \trgB{(\Omega, \bot, \safeta)}$ \end{enumerate}
            and we are done.
        \end{description}
    \end{description}
\end{proof}

As a corollary we get the following:
\begin{lemma}[The current Speculation lasts at most n]\label{thm:spec-ends}
If
\begin{enumerate}
    \item $\SigmaB =  \trgB{ (\Omega_b, \bot, \safeta) \cdot \OB{(\Omega', n', \unta)} \cdot (\Omega, n, \unta)}$
    \item $\textit{max spec}(\Sigma) = n$
\end{enumerate}
Then $\exists \SigmaB'$ such that
\begin{enumerate}
    \item $\trgB{\SigmaB \bigspecarrowB{\tra{^{\taint}}} \SigmaB'}$ \text{ and }
    \item $\trgB{\SigmaB'} = \trgB{ (\Omega_b, \bot, \safeta) \cdot \OB{(\Omega', n', \unta)} \cdot (\Omega'', 0, \unta)}$ 
\end{enumerate}
\end{lemma}
\begin{proof}
Follows from \Thmref{thm:spec-ends-comp}.
\end{proof}
 
We first restate the definition of Safe Nesting:

\begin{definition}[Safe nested Speculation]\label{def:safe-nest-copy}
    Given a program $p$, we say that the program has safe nesting, written $$\contract{}{x+y} \vdash p : \safeN{}$$,
    iff
    \begin{align*}
        \contract{}{x+y} \vdash P : \safeN{} \isdef&\
            \forall \tauStack \in \behxy{(P)}, 
            \text{ if }
            \\
            &\
            \text{(1) }  \startl{a} \cdot \tauStack' \cdot \rollbl{a} \text{ is a subtrace of $\tauStack$}
            \\
            &\
            \text{(2) } \startl{b} \cdot \tauStack'' \cdot \rollbl{b} \text{ is a subtrace of $\tauStack'$} 
            \\
            &\
            \text{ then }
            \safe{\tauStack''} 
            \\
            &\ \quad \text{ where }
            a \in \{x,y\}
            \text{ and } 
            b \in \{x,y\} \setminus \{a\}
        \end{align*}
\end{definition}

We introduce some auxilary Notation because we need to deal with nested transactions:

\begin{definition}[Trace Prefix Notation]
We write $\tauStack \vdash \SigmaBS'$ if there exists an execution $\SigmaBS \bigspecarrowBS{\tauStack} \SigmaBS'$    
\end{definition}

Next, we need a relation similar to the slh-relation that ensured that
the speculation flag is correctly during speculation.
Instead of requiring that the flag is correctly set when speculation starts, we weaken it to only require it when branch speculation starts and then for all newly created speculative transactions afterwards, even if its not branch speculation.
We phrase this here in general fashion, since we look at different combinations of SLH.

Further, we do not need to necessarily relate this to a source program as we did in previous proofs.

\begin{center}
    \typerule{States SLS}{
         \forall \trgB{\Omega_S} \in \trgB{(\Omega', n', \unta)_{\Bv} \cdot \OB{(\Omega'', n'', \unta)_x}}, \vdash \trgB{\Omega_S} : \hrel
     }{
         \vdash  \trgB{ (\Omega, \bot, \safeta) \cdot \OB{(\Omega, n, \unta)_x}  \cdot (\Omega', n', \unta)_{\Bv} \cdot \OB{(\Omega'', n'', \unta)_x}:\hrel }
     }{v1-states-sls-comb}

     \typerule{Single states SLS}{
        \vdash \trgB{M} : \shrel & \vdash \trgB{A}: \shrel
        }{
            \vdash \trgB{\Omega_S} : \shrel
        }{single-states-sls}

    \typerule{Memory - relation same (copied)}{
            \forall \trgB{n \mapsto v : \taint} \in \trgB{M} 
            \text{ if } \trgB{n} \geq \trgB{0} \text{ then } \trgB{\taint} = \trgB{\safeta}
	}{
		\vdash \trgB{M}: \shrel
	}{shrel-same-copy}

    \typerule{Register - relation SLH (copied)}{
        \forall \trgB{x} \in \trgB{\Reg}  &
        \trgB{A(x)} = \_ : \trgB{\safeta}
        &
        \trgB{A(\rslh)} = \trgB{true} : \trgB{\safeta}
        }{
            \vdash \trgB{A}: \shrel
        }{shrel-same-register-copy}
        \typerule{Program to Compiled program}{
            \Omega_n \in (\Omega, \bot, \safeta) \cdot \OB{(\Omega'', n, \unta)} &
            \Omega = (p, M, A) \\
            p = \comp{\src{p}} & A(pc) = l : i & l \in p
        }{
            \src{p} \vdash (\Omega, \bot, \safeta) \cdot \OB{(\Omega'', n, \unta)}
        }{p-cp-relation}
\end{center}

The Invariant we want to carry through is that whenever speculation of V1 occurs, we know that $\rslh=\mi{true}$.
That is something that we have proven for only $\Bv$. So here we need to ensure it holds for the combinations as well.

\begin{definition}[Invariant]
    We write $\vdash \SigmaBS : \slhInv$ iff 
    either $\vdash \SigmaBS : \hrel$ or there is no 
    ongoing speculation of $\Bv$.
\end{definition}

\begin{definition}[Nesting]
    We write $\vdash \SigmaBS : \nesting$ iff there are nested speculative transactions of the different sources in $\SigmaBS$ (here $\Bv$ and $\Sv$).
\end{definition}

\begin{lemma}[Relation Invariant and Nesting]\label{lem:sn-inv-nesting}
    If 
    \begin{enumerate}
        \item $\vdash \SigmaBS : \slhInv$ and 
        \item $\vdash \SigmaBS : \nesting$
    \end{enumerate}
    Then 
    \begin{enumerate}
        \item $\vdash \SigmaBS : \hrel$
    \end{enumerate}
\end{lemma}
\begin{proof}
    Unfolding $\vdash \SigmaBS : \slhInv$ we get: 
    \begin{description}
        \item[$\vdash \SigmaBS : \hrel$] 
            This means we are finished.
        \item[there is no ongoing speculation of $\Bv$ in $\SigmaBS$] 
        Because of $\vdash \SigmaBS : \nesting$ we know there is a transaction of $\Bv$ ongoing.
        This means we have a contradiction.
    \end{description}
\end{proof}

\begin{theorem}[$\Bv + \Sv$:Safe Nesting]\label{thm:v14-sn}
    $\contract{}{\Bv+\Sv} \vdash \compuslhB{\cdot} : \safeN{}$
\end{theorem}
\begin{proof}
    Let $p$ be any source program. We need to show:
    $\contract{}{\Bv+\Sv} \vdash \comp{p} : \safeN{}$.

    Thus, for the trace we need to show that for $tauStack \in \behxy{\compuslhB{P}}$ if 
    $\startl{a} \cdot \tauStack' \cdot \rollbl{a} \in \tauStack$
    and $\startl{b} \cdot \tauStack'' \cdot \rollbl{b} \in \tauStack'$ then $\safe{\tauStack''}$ were $a \in \{x,y\}
    \text{ and } 
    b \in \{x,y\} \setminus \{a\}$.

    W.l.o.g. assume that there is nesting (Otherwise $\tauStack''$ is trivially safe).

    Since $tauStack \in \behxy{\compuslhB{P}}$ and we know that there is nesting, we know there exists a states $\SigmaBS, \SigmaBS'$ and $\SigmaBS \specarrowBS{\startl{b}} \SigmaBS'$ 
    .

    Furthermore, we know there exists $\SigmaBS''$ such that 
    $\SigmaBS' \bigspecarrowBS{\tauStack'' \cdot rollbl{b}} \SigmaBS''$.

    Thus, $\SigmaBS'$ is the state at the start of nesting.

    First, this trivially means that $\vdash \SigmaBS : \nesting$ holds.

    We can now apply \Thmref{thm:v14-comp-nest-spec-multi-safe} on 
    $\SigmaBS' \bigspecarrowBS{\tauStack'' \cdot rollbl{b}} \SigmaBS''$, which gives us that $\safe{\tauStack''}$.

    It remains to show how the invariants hold for $\SigmaBS'$:

\end{proof}

\begin{lemma}[$\Bv+\Sv:$ Steps in Nesting are safe]\label{thm:v14-comp-nest-spec-multi-safe}
If 
\begin{enumerate}
    \item $\SigmaBS \bigspecarrowBS{\tauStack} \SigmaBS'$
    \item $\vdash \SigmaBS : \slhInv $
    \item $\vdash \SigmaBS : \nesting$
    \item $\src{p} \vdash \SigmaBS$
    \item $\src{p} \vdash \SigmaBS'$ \end{enumerate}
Then 
\begin{enumerate}
    \item $\vdash \SigmaBS' : \slhInv $ and 
    \item $\safe{\tauStack}$
\end{enumerate}
\end{lemma}
\begin{proof}
    We proceed by Induction on the derivation $\SigmaBS \bigspecarrowBS{\tauStack} \SigmaBS'$: 
    \begin{description}
        \item[\Cref{tr:xy-reflect}] 
        Trivial, since $\SigmaBS' = \SigmaBS$ and $\tauStack = \empTr$.

        Thus, we trivially fulfill all conclusions.

        \item[\Cref{tr:xy-single}]

        Then we have $\SigmaBS \bigspecarrowBS{\tauStack'} \SigmaBS''$ and $\SigmaBS'' \specarrowBS{\tau^{\taint}} \SigmaBS'$.

Because of $\vdash \SigmaBS : \slhInv $ and $\vdash \SigmaBS : \nesting$
        we can apply \Thmref{thm:v14-comp-nest-spec-sing-safe} on $\SigmaBS$ and get 
        $\SigmaBS \bigspecarrowBS{\tauStack'''} \SigmaBS'''$.
        If there is still nesting in $\SigmaBS'''$, we can use the theorem \Thmref{thm:v14-comp-nest-spec-sing-safe} again and continue until we reach $\SigmaBS''$. 
        We will reach $\SigmaBS'$ because of determinism and alignment

        Then, we are done.

        \item[\Cref{tr:xy-silent}]
        Analogous to the case above
    \end{description}
\end{proof}

\begin{lemma}[$\Bv + \Sv$: A step in Nesting is safe]\label{thm:v14-comp-nest-spec-sing-safe}
    If 
    \begin{enumerate}
        \item $\SigmaBS$
        \item $\vdash \SigmaBS : \slhInv $
        \item $\vdash \SigmaBS : \nesting$
        \item $\src{p} \vdash \SigmaBS$
    \end{enumerate}
    Then 
    \begin{enumerate}
        \item $\SigmaBS \bigspecarrowBS{\tau} \SigmaBS'$
        \item $\vdash \SigmaBS' : \slhInv$
        \item $\src{p} \vdash \SigmaBS'$
        \item $\safe{\tau}$
    \end{enumerate}
\end{lemma}
\begin{proof}
    Because we have $\vdash \SigmaBS : \slhInv$ and $\vdash \SigmaBS : \nesting$ we can apply \Thmref{lem:sn-inv-nesting} and get $\vdash \SigmaBS : \hrel$.

    Next, we look at the next instruction that will be executed.

    \begin{description}
        \item[branch instruction]
        Let us look ath the reduction that the rules generate: 
        \begin{align*}
            \compuslhB{l : \pjz{x}{l'}} &= \trgB{
                                            \begin{aligned}[t]
                                                    &
                                                    l : \passign{\rscr}{x}
                                                        \\
                                                        &\
                                                        l_1' : \pcondassign{\rscr}{0}{\rslh}
                                                        ~\text{\Cref{tr:condup-sat} since $\rslh = \mi{true}$}
                                                            \\
                                                            &\
                                                           l_1'' :  \pjz{\rscr}{l_{new}} ~\text{\Cref{tr:v1-spec}}
                                                                \\
                                                                &\
                                                                l_2 : \passign{\rslh}{\rslh \lor \rscr} ~\text{\Cref{tr:assign}}
                                                                    \\
                                                                    &\
                                                                    l_2' : \pjmp{next(l)}
                                                                        \\
                                                                        &\ 
                                                                        l_{new} : \passign{\rslh}{\rslh \lor \neg \rscr}
                                                                            \\ 
                                                                            &\ 
                                                                            l_{new}' : \pjmp(l')       
                                            \end{aligned}
                                        }
        \end{align*}
        Let us call the resulting execution afterwards $\SigmaBS \bigspecarrowBS{\tauStack} \SigmaBS'$. 

        \begin{description}
 
\item [$\src{p} \vdash \SigmaBS'$]
                From how we evaluate this program we have $\src{p} \vdash \SigmaBS'$ already.
            \item [$\safe{\tauStack}$]
                From $\vdash \SigmaBS : \hrel $ we get that $\SigmaBS(\rslh) = \mi{true}$.
                Thus, $\rscr = 0$ and after execution $l_1'' :\pjz{\rscr}{l_{new}}$ by \Cref{tr:v1-spec},
                we continue with $l_2$ because of mispeculation. The observation $\tau = \pcObs{{l_2}}^\safeta_{\Bv}$ is generated.

                we get the trace $\tauStack = \storeObs{0}^\safeta_{\Sv}$.

                Thus, we have $\safe{\tauStack}$.
            \item [$\vdash \SigmaBS' : \slhInv$]
                The register was changed in $l_2 : \passign{\rslh}{\rslh \lor \rscr}$. Since $\rslh = \mi{true}$ before the instruction
                is executed, we have $\rslh = \mi{true}$ afterwards as well.

                We thus have $\vdash : \SigmaBS' : \hrel$.
        \end{description}

        \item[store instruction]
        Let us look at the reduction that the rules generate: 
        \begin{align*}
            \compuslhB{l : \pstore{x}{e}} &= \trgB{
                                            \begin{aligned}[t]
                                                    &
                                                    l :\passign{\rscr}{\compuslhB{e}}
                                                        ~\text{\Cref{tr:assign} uses $\trgB{\exprEval{A}{\compuslhB{e}}{v: \taint}}$}
                                                        \\
                                                        &\
                                                        l' :\pcondassign{\rscr}{0}{\rslh} 
                                                            ~\text{\Cref{tr:condup-sat} since $\rslh = \mi{true}$}
                                                            \\
                                                            &\
                                                            l'' :\pstore{x}{\rscr} 
                                                            ~\text{\Cref{tr:v4-skip}}
                                            \end{aligned}
                                        }
        \end{align*}

        Let us call the resulting execution afterwards $\SigmaBS \bigspecarrowBS{\tauStack} \SigmaBS'$. 
        \begin{insight}
            As before, here we describe the case when the speculation window is big enough. 
            If the speculation window reaches 0 before, a rollback is triggered, which is described below. However, a rollback never leaks information, so the most interesting case is when all instructions go through.
        \end{insight}

        \begin{description}
            \item [$\vdash \SigmaBS' : \slhInv$]
                Since these instructions do not change the register $\rslh$, we also have $\vdash \SigmaBS' : \hrel$.
            \item [$\src{p} \vdash \SigmaBS'$]
                From how we evaluate this program we have $\src{p} \vdash \SigmaBS'$ already.
            \item [$\safe{\tauStack}$]
                From $\vdash \SigmaBS : \hrel $ we get that $\SigmaBS(\rslh) = \mi{true}$.
                Thus, $\rscr = 0$ and after execution of $\pstore{x}{0}$, 
                we get the trace $\tauStack = \storeObs{0}^\safeta_{\Sv}$.

                Thus, we have $\safe{\tauStack}$.
        \end{description}
        
        \item[otherwise]
        This is similar to the proof for \Thmref{thm:comp-spec-sing-safe}, which yields $\SigmaBS'$ as desired.
        
        \item[Rollback]
        This happens when the speculation windows was not big enough. 
        This is similar to the rollback case in \Thmref{thm:comp-spec-sing-safe}.
        
        We show why $\vdash \SigmaBS' : \slhInv$ holds. 
        \begin{description}
            \item[The last ongoing V1 transaction was rolled back]
            Then we trivially satisfy $\vdash \SigmaBS' : \slhInv$
            \item[otherwise]
            Then there is at least one ongoing $\Bv$ speculation. 
            From the fact that $\vdash \SigmaBS : \slhInv$ and we had nesting going on, we have $\forall \trgB{\Omega_S} \in 
            \trgB{(\Omega', n', \unta)_{\Bv} \cdot \OB{(\Omega'', n'', \unta)_x}}, \vdash \trgB{\Omega_S} : \hrel$.
            Since a rollback only deletes the topmost state and not change the states below, we have
            $\vdash \SigmaBS' : \slhInv$.
        \end{description}
    \end{description} 
\end{proof}

\begin{lemma}[Initial state fulfills Invariant]
    If 
    \begin{enumerate}
        \item $\Sigma_{\Bv + y} = \initFuncxy{P}$
    \end{enumerate}
    Then 
    \begin{enumerate}
        \item $\vdash \Sigma_{\Bv + y} : \slhInv$
    \end{enumerate}
\end{lemma}
\begin{proof}
    Initially, there is no ongoing speculation. Thus, there is no $\Bv$ speculation 
    and we trivially satisfy $\vdash \Sigma_{\Bv + y} : \slhInv$
\end{proof}

\begin{lemma}[Steps of compiled programs preserve relation]
   If 
   \begin{enumerate}
    \item $\vdash \SigmaBS : \slhInv$
    \item $\src{p} \vdash \SigmaBS$
   \end{enumerate}
   Then 
   \begin{enumerate}
    \item $\SigmaBS \bigspecarrowBS{\tauStack} \SigmaBS'$
    \item $\vdash \SigmaBS' : \slhInv$
    \item $\src{p} \vdash \SigmaBS$
   \end{enumerate}
\end{lemma}
\begin{proof}
    Unfolding $\vdash \SigmaBS : \slhInv$ we have two cases: 
    \begin{description}
        \item[No ongoing speculation of $\Bv$]
        Let us look at the current instruction 
        \begin{description}
            \item[$i = \pjz{x}{l}$]
            Then, a speculative transaction of $\Bv$ is started.
            Since there was no before, we now trivially satisfy $\vdash \SigmaBS : \hrel$.
            \item[otherwise]
            Then, no speculative transaction of $\Bv$ is started and we 
            thus, trivially satisfy $\vdash \SigmaBS' : \slhInv$.

        \end{description}

        \item[$\vdash \SigmaBS : \hrel$]
        Thus, there is ongoing speculation of $\Bv$.
        \begin{description}
            \item[$i = \pjz{x}{l}$] 
            Then the flag is correctly updated and we have $\vdash \SigmaBS : \slhInv$.
            \item[otherwise]
            Note that these instructions cannot change the flag.
            We thus trivially have $\vdash \SigmaBS : \slhInv$. 
        \end{description}
    \end{description}
\end{proof}

Instead of writing 5 different versions of \Thmref{thm:v14-comp-nest-spec-multi-safe} that just use a different step lemma, we just use \Thmref{thm:v14-comp-nest-spec-multi-safe}
\begin{theorem}[$\Bv + \Rv$: Safe Nesting]\label{thm:v15-sn}
    $\contract{}{\Bv+\Rv} \vdash \compuslhB{\cdot} : \safeN{}$
\end{theorem}
\begin{proof}
    Analogous to \Thmref{thm:v14-sn} using \Thmref{thm:v14-comp-nest-spec-multi-safe}.
\end{proof}

\begin{theorem}[$\Bv + \SLSv$: Safe Nesting]\label{thm:v1sls-sn}
    $\contract{}{\Bv+\SLSv} \vdash \compuslhB{\cdot} : \safeN{}$
\end{theorem}
\begin{proof}
    Analogous to \Thmref{thm:v14-sn} using \Thmref{thm:v14-comp-nest-spec-multi-safe} with \Thmref{thm:v1sls-comp-nest-spec-sing-safe}.
\end{proof}

\begin{theorem}[$\Bv + \Sv + \SLSv$: Safe Nesting]\label{thm:v14sls-sn}
    $\contract{}{\Bv+\Sv +\SLSv} \vdash \compuslhB{\cdot} : \safeN{}$
\end{theorem}
\begin{proof}
    Analogous to \Thmref{thm:v14-sn} using \Thmref{thm:v14-comp-nest-spec-multi-safe} with \Thmref{thm:v14sls-comp-nest-spec-sing-safe}.
\end{proof}

\begin{theorem}[$\Bv + \Sv + \Rv$: Safe Nesting]\label{thm:v145-sn}
    $\contract{}{\Bv+ \Sv+\Rv} \vdash \compuslhB{\cdot} : \safeN{}$
\end{theorem}
\begin{proof}
    Analogous to \Thmref{thm:v14-sn} using \Thmref{thm:v14-comp-nest-spec-multi-safe} with \Thmref{thm:v145-comp-nest-spec-sing-safe}.
\end{proof}

\begin{lemma}[$\Bv + \SLSv$: A step in Nesting is safe]\label{thm:v1sls-comp-nest-spec-sing-safe}
    If 
    \begin{enumerate}
        \item $\SigmaBSLS$
        \item $\vdash \SigmaBSLS : \slhInv $
        \item $\vdash \SigmaBSLS : \nesting$
        \item $\src{p} \vdash \SigmaBSLS$
    \end{enumerate}
    Then 
    \begin{enumerate}
        \item $\SigmaBSLS \bigspecarrowBSLS{\tau} \SigmaBSLS'$
        \item $\vdash \SigmaBSLS' : \slhInv$
        \item $\src{p} \vdash \SigmaBSLS'$
        \item $\safe{\tau}$
    \end{enumerate}
\end{lemma}
\begin{proof}
    Because we have $\vdash \SigmaBSLS : \slhInv$ and $\vdash \SigmaBSLS : \nesting$ we can apply \Thmref{lem:sn-inv-nesting} and get $\vdash \SigmaBSLS : \hrel$.

    Next, we look at the next instruction that will be executed.

    \begin{description}

        \item[return instruction]
        Then we have 
        \begin{align*}
            \compuslhB{l : \pret} &= \trgB{
									\begin{aligned}[t]
											&
											l :\passign{\rslhC}{\rslh} \\
												&\
												l' :\pret ~\text{\Cref{tr:sls-spec}}
									\end{aligned}
								}
        \end{align*}
        
Let us call the resulting execution afterwards $\SigmaBSLS \bigspecarrowBSLS{\tauStack} \SigmaBSLS'$.
        
        \begin{description}
            \item [$\vdash \SigmaBSLS' : \slhInv$]
                Since these instructions do not change the register $\rslh$, we also have $\vdash \SigmaBSLS' : \hrel$.
            \item [$\src{p} \vdash \SigmaBSLS'$]
                From how we evaluate this program we have $\src{p} \vdash \SigmaBSLS'$ already.
            \item [$\safe{\tauStack}$]
                From $\vdash \SigmaBSLS : \hrel $ we get that $\SigmaBSLS(\rslh) = \mi{true}$.
                After execution of $\pret$, we get the trace $\tauStack = \retObs{}^\safeta_{\SLSv}$.

                Thus, we have $\safe{\tauStack}$.
        \end{description}

        \item[branch instruction]
        Analogous to \Thmref{thm:v14-comp-nest-spec-sing-safe}.

        \item[Rollback]
        Analogous to \Thmref{thm:v14-comp-nest-spec-sing-safe}
        
        \item[otherwise]
        Analogous to \Thmref{thm:v14-comp-nest-spec-sing-safe}.

    \end{description} 
\end{proof}

\begin{lemma}[$\Bv + \Rv$: A step in Nesting is safe]\label{thm:v15-comp-nest-spec-sing-safe}
    If 
    \begin{enumerate}
        \item $\SigmaBR$
        \item $\vdash \SigmaBR : \slhInv $
        \item $\vdash \SigmaBR : \nesting$
        \item $\src{p} \vdash \SigmaBR$
    \end{enumerate}
    Then 
    \begin{enumerate}
        \item $\SigmaBR \bigspecarrowBR{\tau} \SigmaBR'$
        \item $\vdash \SigmaBR' : \slhInv$
        \item $\src{p} \vdash \SigmaBR'$
        \item $\safe{\tau}$
    \end{enumerate}
\end{lemma}
\begin{proof}
    Because we have $\vdash \SigmaBR : \slhInv$ and $\vdash \SigmaBR : \nesting$ we can apply \Thmref{lem:sn-inv-nesting} and get $\vdash \SigmaBR : \hrel$.

    Next, we look at the next instruction that will be executed.

    \begin{description}

        \item[return instruction]
        Then we have 
        \begin{align*}
            \compuslhB{l : \pret} &= \trgB{
									\begin{aligned}[t]
											&
											l :\passign{\rslhC}{\rslh} \\
												&\
												l' :\pret ~\text{\Cref{tr:v5-spec}}
									\end{aligned}
								}
        \end{align*}
        Note that speculation is started here and we could jump to any previous call site.
        Furthermore, we execute the additional instruction after the call $l'' : \passign{\rslh}{\rslhC}$, because that is how call instructions are compiled for the state below.
        Let us call the resulting execution afterwards $\SigmaBR \bigspecarrowBR{\tauStack} \SigmaBR'$.
        
        \begin{description}
            \item [$\vdash \SigmaBR' : \slhInv$]
                Because of $l :\passign{\rslhC}{\rslh}$ and 
                $l'' : \passign{\rslh}{\rslhC}$, we recovered the slh flag after the call.
                We thus have $\vdash \SigmaBR' : \hrel$.
            \item [$\src{p} \vdash \SigmaBR'$]
                From how we evaluate this program we have $\src{p} \vdash \SigmaBR'$ already.
            \item [$\safe{\tauStack}$]
                From $\vdash \SigmaBR : \hrel $ we get that $\SigmaBR(\rslh) = \mi{true}$.
                After execution of $\pret$, we get the trace $\tauStack = \retObs{}^\safeta_{\Rv}$.

                Thus, we have $\safe{\tauStack}$.
        \end{description}

        \item[branch instruction]
        Analogous to \Thmref{thm:v14-comp-nest-spec-sing-safe}.

        \item[Rollback]
        Analogous to \Thmref{thm:v14-comp-nest-spec-sing-safe}
        
        \item[otherwise]
        Analogous to \Thmref{thm:v14-comp-nest-spec-sing-safe}.

    \end{description} 
\end{proof}

\begin{lemma}[$\Bv + \Sv + \SLSv$: A step in Nesting is safe]\label{thm:v14sls-comp-nest-spec-sing-safe}
    If 
    \begin{enumerate}
        \item $\SigmaBSSLS$
        \item $\vdash \SigmaBSSLS : \slhInv $
        \item $\vdash \SigmaBSSLS : \nesting$
        \item $\src{p} \vdash \SigmaBSSLS$
    \end{enumerate}
    Then 
    \begin{enumerate}
        \item $\SigmaBSSLS \bigspecarrowBSSLS{\tau} \SigmaBSSLS'$
        \item $\vdash \SigmaBSSLS' : \slhInv$
        \item $\src{p} \vdash \SigmaBSSLS'$
        \item $\safe{\tau}$
    \end{enumerate}
\end{lemma}
\begin{proof}
    Because we have $\vdash \SigmaBSSLS : \slhInv$ and $\vdash \SigmaBSSLS : \nesting$ we can apply \Thmref{lem:sn-inv-nesting} and get $\vdash \SigmaBSSLS : \hrel$.

    Next, we look at the next instruction that will be executed.

    \begin{description}

        \item[return instruction]
        Analogous to the corresponding case in \Thmref{thm:v1sls-comp-nest-spec-sing-safe}.
        
        \item[branch instruction]
        Analogous to the corresponding case in \Thmref{thm:v14-comp-nest-spec-sing-safe}.

        \item[store instruction]
        Analogous to the corresponding case in \Thmref{thm:v14-comp-nest-spec-sing-safe}.
        
        \item[Rollback]
        Analogous to the corresponding case in \Thmref{thm:v14-comp-nest-spec-sing-safe}
        
        \item[otherwise]
        Analogous to the corresponding case in \Thmref{thm:v14-comp-nest-spec-sing-safe}.

    \end{description} 
\end{proof}

\begin{lemma}[$\Bv + \Sv + \Rv$: A step in Nesting is safe]\label{thm:v145-comp-nest-spec-sing-safe}
    If 
    \begin{enumerate}
        \item $\SigmaBSR$
        \item $\vdash \SigmaBSR : \slhInv $
        \item $\vdash \SigmaBSR : \nesting$
        \item $\src{p} \vdash \SigmaBSR$
    \end{enumerate}
    Then 
    \begin{enumerate}
        \item $\SigmaBSR \bigspecarrowBSR{\tau} \SigmaBSR'$
        \item $\vdash \SigmaBSR' : \slhInv$
        \item $\src{p} \vdash \SigmaBSR'$
        \item $\safe{\tau}$
    \end{enumerate}
\end{lemma}
\begin{proof}
    Because we have $\vdash \SigmaBSR : \slhInv$ and $\vdash \SigmaBSR : \nesting$ we can apply \Thmref{lem:sn-inv-nesting} and get $\vdash \SigmaBSSLS : \hrel$.

    Next, we look at the next instruction that will be executed.

    \begin{description}

        \item[return instruction]
        Analogous to the corresponding case in \Thmref{thm:v15-comp-nest-spec-sing-safe}.
        
        \item[branch instruction]
        Analogous to the corresponding case in \Thmref{thm:v14-comp-nest-spec-sing-safe}.

        \item[store instruction]
        Analogous to the corresponding case in \Thmref{thm:v14-comp-nest-spec-sing-safe}.

        \item[Rollback]
        Analogous to the corresponding case in \Thmref{thm:v14-comp-nest-spec-sing-safe}
        
        \item[otherwise]
        Analogous to the corresponding case in \Thmref{thm:v14-comp-nest-spec-sing-safe}.

    \end{description} 
\end{proof}
 \section{Independence Results for the SLH compilers}

Independence w.r.t. $\sems$, $\semr$ and $\semsls$.

Let us review the definition of Independence.
 $\IND{\contract{}{y}}{\comp{\cdot}} \isdef \forall p.\ \contract{}{y} \vdash p : \rss \implies \contract{}{y} \vdash \comp{p} : \rss$

Since the compilers $\compuslhB{\cdot}$ and $\compsslhB{\cdot}$ only differ for the extended observer $ct + vl$, we only show the results w.r.t. $\compsslhB{\cdot}$. 
The same arguments apply to $\compuslhB{\cdot}$.

\begin{lemma}[SLH Independence w.r.t. $\Svr$]\label{lem:slh-ind-s}
    $\IND{\contract{}{\Sv}}{\compsslhB{\cdot}}$
\end{lemma}
\begin{proof}
    
        By definition of Independence, we know that for all programs $p$ that $\contract{}{\Sv} \vdash p : \rss$.

        Proof by Contradiction. Suppose that $\contract{}{\Sv} \nvdash \compsslhB{p} : \rss$.

        This means there is a leak caused by speculation on store instructions.
        However, note that the compiler does not add new store instructions, but rather masks the already existing ones using the slh-flag.

        Note that the mask does not have an effect here, since there is no speculation of $\semb$ here. This means that the slh-flag is never true and masking does nothing.

        Thus, we are left with the original store instructions of $p$.
        Now we have a contradiction since $\contract{}{\Sv} \vdash p : \rss$ by assumption.

\end{proof}

\begin{lemma}[SLH Independence w.r.t. $\Rvr$]\label{lem:slh-ind-r}
    $\IND{\contract{}{\Rv}}{\compsslhB{\cdot}}$
\end{lemma}
\begin{proof}
    By definition of Independence, we know that for all programs $p$ that $\contract{}{\Rv} \vdash p : \rss$.

    Note that the mask of slh does not have an effect here, since there is no speculation of $\semb$ here. This means that the slh-flag is never true and masking does nothing.

    Since there were no leaks in the initial program and the fact that the mask will do nothing, we have $\contract{}{\Rv} \vdash \compsslhB{p} : \rss$.
\end{proof}

\begin{lemma}[SLH Independence w.r.t. $\SLSvr$]\label{lem:slh-ind-sls}
    $\IND{\contract{}{\SLSv}}{\compsslhB{\cdot}}$
\end{lemma}
\begin{proof}
    By definition of Independence, we know that for all programs $p$ that $\contract{}{\SLSv} \vdash p : \rss$.

    Note that the mask does not have an effect here, since there is no speculation of $\semb$ here. This means that the slh-flag is never true and masking does nothing.

    Since there were no leaks in the initial program and the fact that the mask will do nothing, we have $\contract{}{\SLSv} \vdash \compsslhB{p} : \rss$.
\end{proof}

\begin{lemma}[SLH Independence w.r.t. $\Svr$ + $\Rvr$]\label{lem:slh-ind-sr}
    $\IND{\contract{}{\Sv + \Rv}}{\compsslhB{\cdot}}$
\end{lemma}
\begin{proof}
    Follows from \Thmref{lem:slh-ind-s} and \Thmref{lem:slh-ind-r}.
\end{proof}

\begin{lemma}[SLH Independence w.r.t. $\Svr$ + $\SLSvr$]\label{lem:slh-ind-ssls}
    $\IND{\contract{}{\Sv + \SLSv}}{\compsslhB{\cdot}}$
\end{lemma}
\begin{proof}
    Follows from \Thmref{lem:slh-ind-s} and \Thmref{lem:slh-ind-sls}.
\end{proof}
 \section{General Framework Combinations}

Most of these definitions are analogous to \cite{spec_comb} and restated here for clarity.

To account for the taint semantics, we split the combination into a value and a taint part as well.

\subsection{Union of States}
The states $\Sigmaxy$ are defined as the union of the state of its parts
\begin{align*}
    \Sigmaxy \bnfdef&~ \Sigmax \cup \Sigmay \\
    \Phixy \bnfdef&~ \Phix \cup \Phiy 
\end{align*}

\subsection{Unified Trace Model}
Similar to the states $\Sigmaxy$, the observations are defined as the union of the trace models
\begin{align*}
    \Obsxy \bnfdef~ \Obsx \cup \Obsy 
\end{align*}

\subsection{Join and Projection on instances / states}

We define the a join operator $\joinxy$ on instances:
\begin{align*}
    \joinxy \colon (\Phix, \Phiy) \mapsto \Phixy
\end{align*}
Note, that the join operations is only defined iff all the components $c$ in $\Phix \cap \Phiy$ are equal, i.e. $\forall c \in \Phix \cap \Phiy. ~ \Phix. c = \Phiy . c$.

Furthermore, we define projection functions $\specProjectxy{}$ as the inverse of the join function:

\begin{align*}
    \specProjectxy{} \colon \Phixy \mapsto (\Phix, \Phiy)
\end{align*}

We require that 
\begin{align*}
    \joinxy(\specProjectxy{\Phixy}) =& \Phixy \\
    \specProjectxy{(\joinxy(\Phix, \Phiy))} =& (\Phix, \Phiy)
\end{align*}

Next, we will define two more specific projection $\specProjectxyx{}$ and $\specProjectxyy{}$:
\begin{align*}
    \specProjectxyx{} \colon \Phixy \mapsto \Phix \\
    \specProjectxyy{} \colon \Phixy \mapsto \Phiy 
\end{align*}

They are defined as the first and second projection of the generated pair from $\specProjectxy{}$.

\mytoprule{Judgements}
\begin{align*}
    &
    \trgxy{\phiStackxyv \specarrowxy{\tau} \phiStackxyv'}
    && \text{State $\trgxy{\phiStackxyv}$ small-steps to $\trgxy{\phiStackxyv'}$ and emits observation $\trgxy{\tau}$.}
    \\
    &
    \trgxy{\Phixyv \specarrowxy{\tau} \phiStackxy'}
    && \text{Speculative instance $\trgxy{\Phixyv}$ small-steps to $\trgxy{\phiStackxyv{'}}$ and emits observation $\trgxy{\tau}$.}
    \\
    &
    \trgxy{\phiStackxyt' \specarrowxy{\taint} \phiStackxyt'}
    && \text{State $\trgxy{\phiStackxyt}$ small-steps to $\trgxy{\phiStackxyt}$ and emits taint $\trgxy{\taint}$.}
    \\
    &
    \trgxy{\Phixyt \specarrowxy{\taint} \phiStackxyt'}
    && \text{Speculative instance $\trgxy{\Phixyt}$ small-steps to $\trgxy{\phiStackxyt{'}}$ and emits taint $\trgxy{\taint}$.}
    \\
    &
    \trgxy{\Sigmaxy \specarrowxy{\tau^{\taint}} \Sigmaxy'}
    && \text{Speculative state $\trgxy{\Sigmaxy}$ small-steps to $\trgxy{\Sigmaxy'}$ and emits tainted observation $\trgxy{\tau^{\taint}}$.}
    \\
    &
    \trgxy{\Sigmaxy \bigspecarrowxy{\OB{\tau^{\taint}}} \Sigmaxy'}
    && \text{State $\trgxy{\Sigmaxy}$ big-steps to $\trgxy{\Sigmaxy'}$ and emits a list of tainted observations $\trgxy{\tauStackT}$.}
    \\
    &
    \trgxy{\amTracevxy{P}{\tauStackT}}
    && \text{Program $\trgxy{P}$ produces the tainted observations $\trgxy{\tauStackT}$ during execution.}
\end{align*}

\begin{center}
    \centering
    \small
    \mytoprule{\trgxy{\phiStackxyv \specarrowxy{\tau} \phiStackxyv'}}
    
    \typerule {Context-xy}
    {\trgxy{\Phixyv \specarrowxy{\tau} \phiStackxyv'}
    }
    {\trgxy{\phiStackxyv \cdot \Phixyv \specarrowxy{\tau} \phiStackxyv \cdot \phiStackxyv'}
     }{xy-context}
     
    \mytoprule{\trgxy{\Phixyv \specarrowxy{\tau} \phiStackxyv}}
     
    \typerule {x-step}
    {\specProjectxyx{\Phixyv} \specarrowxZ{\tau} \specProjectxyx{\phiStackxyv}
    }
    {\Phixyv
    \specarrowxy{\tau}
    \phiStackxyv
     }{xy-x-step}
     \typerule {y-step}
    {\specProjectxyy{\Phixyv} \specarrowyZ{\tau} \specProjectxyy{\phiStackxyv}
    }
    {\Phixyv
    \specarrowxy{\tau}
    \phiStackxyv
     }{xy-y-step}
     
     \typerule{x-Rollback}
    { \trgxy{\Phixyv . n = 0}\ \text{ or }\ \finType{\trgxy{\Phixyv}} \\ 
    \specProjectxyx{\Phixyv} \specarrowx{\rollbackObsx} \emptyset 
    }
    {\trgxy{\Phixy \specarrowxy{\rollbackObsx} \emptyset}
    }{xy-x-rollback}
    \typerule{y-Rollback}
    { \trgxy{\Phixyv . n = 0}\ \text{ or }\ \finType{\trgxy{\Phixyv}} \\ 
    \specProjectxyy{\Phixyv} \specarrowy{\rollbackObsy} \emptyset 
    }
    {\trgxy{\Phixyv \specarrowxy{\rollbackObsy} \emptyset}
    }{xy-y-rollback}
     
\end{center}

\begin{center}
    \centering
    \small
    \mytoprule{\trgxy{\phiStackxyt \specarrowxy{\taint} \phiStackxyt'}}
    
    \typerule {T-Context-xy}
    {\trgxy{\Phixyt \specarrowxy{\taint} \phiStackxyt'}
    }
    {\trgxy{\phiStackxyt \cdot \Phixyt \specarrowxy{\taint} \phiStackxyt \cdot \phiStackxyt'}
     }{xy-t-context}
     
    \mytoprule{\trgxy{\Phixyt \specarrowxy{\taint} \phiStackxyt}}
     
    \typerule {T-x-step}
    {\specProjectxyx{\Phixyt} \specarrowxZ{\taint} \specProjectxyx{\phiStackxyt}
    }
    {\Phixyt
    \specarrowxy{\taint}
    \phiStackxyt
     }{xy-t-x-step}
     \typerule {T-y-step}
    {\specProjectxyy{\Phixyt} \specarrowyZ{\taint} \specProjectxyy{\phiStackxyt}
    }
    {\Phixyt
    \specarrowxy{\taint}
    \phiStackxyt
     }{xy-t-y-step}
     
     \typerule{T-x-Rollback}
    { \trgxy{\Phixyt . n = 0}\ \text{ or }\ \finType{\trgxy{\Phixyt}} \\ 
    \specProjectxyx{\Phixyt} \specarrowx{\safeta} \emptyset 
    }
    {\trgxy{\Phixyt \specarrowxy{\safeta} \emptyset}
    }{xy-t-x-rollback}
    \typerule{T-y-Rollback}
    { \trgxy{\Phixyt . n = 0}\ \text{ or }\ \finType{\trgxy{\Phixyt}} \\ 
    \specProjectxyy{\Phixyt} \specarrowy{\safeta} \emptyset 
    }
    {\trgxy{\Phixyt \specarrowxy{\safeta} \emptyset}
    }{xy-t-y-rollback}
     
\end{center}

\begin{center}
    \centering
    \small
    \mytoprule{\trgxy{\Sigmaxy \specarrowxy{\tau^{\taint}} \Sigmaxy'}}
    
    \typerule {xy:Combine}
    {\trgxy{\Sigmaxy = \phiStackxy{}} & \trgxy{\Sigmaxy' = \phiStackxy'} \\
    \trgxy{\phiStackxyv + \phiStackxyt = \phiStackxy} &  \trgxy{\phiStackxyv' + \phiStackxyt' = \phiStackxy'} \\
    \trgxy{\phiStackxyv \specarrowxy{\tau} \phiStackxyv'} & \trgxy{\phiStackxyt \specarrowxy{\taint} \phiStackxyt'}
    }
    {\trgxy{\Sigmaxy \specarrowxy{\tau^{\taint}} \Sigmaxy'}
     }{xy-combine}
\end{center}

\begin{center}
    \mytoprule{\trgxy{\Sigmaxy \bigspecarrowxy{\tauStack} \Sigmaxy'}}
    
    \typerule{xy:AM-Reflection}
    {
    }
    {\trgxy{\Sigmaxy \bigspecarrowxy{\varepsilon} \Sigmaxy}
    }{xy-reflect}
    \typerule{xy:AM-Single}
    {
    \trgxy{\Sigmaxy \bigspecarrowxy{\OB{\tau^{\taint}}} \Sigmaxy''} & \trgxy{\Sigmaxy'' \specarrowxy{\tau^{\taint}} \Sigmaxy'} \\
    \trgxy{\Sigmaxy'' = \phiStackxy \cdot  \tup{\OB{F}; \OB{I}; \OB{B}; \sigma, n}} &
    \trgxy{\Sigmaxy' = \phiStackxy \cdot \tup{\OB{F}; \OB{I}; \OB{B'}; \sigma', n' }} \\
    \ffun{\trgxy{\sigma(\pc)}} = \trgxy{f} & \ffun{\trgxy{\sigma'(\pc)}} = \trgxy{f'} \\
    \text{ if } \trgxy{f == f'} and ~\trgxy{f \in \OB{I}} \text{ then } \trgxy{\tau^{\taint} = \epsilon} \text{ else } \trgxy{\tau^{\taint} = \tau^{\taint}}
    }
    {
    \trgxy{\Sigmaxy \bigspecarrowxy{\OB{\tau^{\taint}} \cdot \tau^{\taint}} \Sigmaxy'}
    }{xy-single}
    
    \typerule{xy:AM-silent}
    {
    \trgxy{\Sigmaxy \bigspecarrowxy{\tra{^\taint}} \Sigmaxy''} & \trgxy{\Sigmaxy'' \specarrowxy{\epsilon} \Sigmaxy'}
    }
    {
    \trgxy{\Sigmaxy \bigspecarrowxy{\tra{^\taint}} \Sigmaxy'}
    }{xy-silent}

\end{center}

\begin{center}
    \mytoprule{Helpers}
    
    \typerule{xy:AM-Merge}
    {
    }
    {\Phixyv + \Phixyt = \Phixy
    }{xy-merge}
    \typerule{xy:AM-Init}
    {
    }
    {\SInit(M ,\OB{F}, \OB{I}) = \coloneqq \joinxy(\initFuncx{M ,\OB{F}, \OB{I}}, \initFuncy{M ,\OB{F}, \OB{I})}
    }{xy-init}
    
    \typerule{xy:AM-Fin-Ending}
    { \Sigmaxy = \Phixy & \finType{\Phixy}
    }
    {\finTypef{\Sigmaxy}
    }{xy-fin-end}
    \typerule{xy:AM-Fin}
    { \Sigmaxy = \phiStackxy \cdot \Phixy & \finType{\Phixy}
    }
    {\finType{\Sigmaxy}
    }{xy-fin}
    
    \typerule {xy:AM-Trace}
    { \exists \finTypef{\Sigmaxy'}\ & \initFuncxy(P) \bigspecarrowxy{\tauStackT}\ \Sigmaxy'
    }
    { \amTracevxy{P}{\tauStackT}
     }{xy-trace}
    \typerule {xy:AM-Beh}
    {}
    { \behxy{P} =  \{\tauStackT \mid \amTracevxy{P}{\tauStackT} \}
     }{xy-beh}
\end{center}

And, we define when we call composition well-formed
\begin{definition}[Well-formed composition]\label{def:wellformed}
A composition $\contract{}{xy}$ of two speculative contracts $\contract{}{x}$ and $\contract{}{y}$ is \textit{well-formed}, denoted with $\wfc{\contract{}{xy}}$ if:
\begin{compactenum}
    \item \textit{(Confluence)} Whenever $\Sigmaxy \specarrowxy{\tau'} \Sigmaxy'$ and $\Sigmaxy \specarrowxy{\tau} \Sigmaxy''$, then  $\Sigmaxy' = \Sigmaxy''$ and $\tau = \tau'$.
    \item \textit{(Projection preservation)} For all programs $p$, $\behx{p} = \specProjectxyx{\behxy{(p)}}$ and $\behy{p} = \specProjectxyy{\behxy{(p)}}$.
    \item (Relation Preservation)
    If $\Sigmax \specarrowx{\tau} \Sigmax'$ and $\Sigmax \relsaxy \Sigmaxy^{\dagger}$ then there exists $\Sigmaxy''$ such that $\Sigmaxy^{\dagger} \specarrowxy{\tau} \Sigmaxy''$ and $\Sigmaxy' \relsaxy \Sigmaxy''$.
\end{compactenum}
\end{definition}
 \section{Lifting Security Guarantees to stronger attackers}

Here we present results that lift guarantees of a compiler for a specific semantics to guarantees about the combination of different spec. semantics.  

The running example that we will use is a combination of $\specj$ and $\specr$. The retpoline countermeasure $\compretpJ{\cdot}$ introduces new $\retC$ instructions into the compiled code. These $\retC$ instructions are a source of speculation for $\specr$. Thus, the natural question to ask is, if the compiled program introduces new $\specr$ vulnerabilities into the code. 

Afterwards, we define two properties Syn, Independence and Trapped Speculation which are simpler to prove and imply the stronger conditions.

\subsection{Interplay of semantics}

\begin{figure}
\tikzset{
    between/.style args={#1 and #2}{
         at = ($(#1)!0.5!(#2)$)
    }
}
\begin{tikzpicture}
    \node (h1) {$\scriptstyle\startl{x}$};
    \node[right = 0.6 cm of h1] (m1) {$\scriptstyle\tauStack$};
    \node[right = 0.6 cm of m1] (h2) {$\scriptstyle\startl{y}$};
    \node[right = 0.6 cm of h2] (m5) {$\scriptstyle\tauStack'$};
    
    \node[below = 1 cm of h1] (b1){$\scriptstyle\startl{y}$};
    \node[right = 0.6 cm of b1] (m2) {$\scriptstyle\tauStack$};
    \node[below = 1 cm of h2] (b2){$\scriptstyle\startl{x}$};
    \node[below = 1 cm of m5] (m6) {$\scriptstyle\tauStack'$};

    \node[right = 0.6 cm of m5] (h3) {$\scriptstyle\rollbl{y}$};
    \node[right = 0.6 cm of h3] (m3) {$\scriptstyle\tauStack''$};
    \node[right = 0.6 cm of m3] (h4) {$\scriptstyle\rollbl{x}$};

    \node[below = 1 cm of h3] (b3){$\scriptstyle\rollbl{x}$};
    \node[right = 0.6 cm of b3] (m4) {$\scriptstyle\tauStack''$};
    \node[below = 1 cm of h4] (b4){$\scriptstyle\rollbl{y}$};

    \node[between = h1 and m1](t){$\cdot$};
    \node[between = b1 and m2](t){$\cdot$};
    \node[between = m1 and h2](t1){$\cdot$};
    \node[between = m2 and b2](t2){$\cdot$};
    
    \node[between = h2 and m5](t){$\cdot$};
    \node[between = b2 and m6](t){$\cdot$};
    \node[between = m5 and h3](t){$\cdot$};
    \node[between = m6 and b3](t){$\cdot$};

    \node[between = h3 and m3](t3){$\cdot$};
    \node[between = b3 and m4](t4){$\cdot$};
    \node[between = m3 and h4](t){$\cdot$};
    \node[between = m4 and b4](t){$\cdot$};

    \draw[decorate,decoration={brace,mirror, amplitude=5pt},below=10pt]
        (h2.south west) -- (h3.south east) node[font=\footnotesize, draw, midway, below=7pt] {\region{3}};
        
    \draw[decorate,decoration={brace, amplitude=5pt},above=10pt]
        (b2.north west) -- (b3.north east)  {};
    
    \node[fit=(h1) (m1) (t1), inner sep=5pt, fill=red, opacity=0.1, align=center, rounded corners] (background) {};
    \node[fit=(h2) (m5) (h3), inner sep=5pt, fill=green, opacity=0.1, align=center, rounded corners] (background) {};
    \node[fit=(t3)  (h4), inner sep=5pt, fill=red, opacity=0.1, align=center, rounded corners] (background) {};

    \node[fit=(b1) (m2) (t2), inner sep=5pt, fill=blue, opacity=0.1, align=center, rounded corners] (background) {};
    \node[fit=(b2) (m6) (b3), inner sep=5pt, fill=green, opacity=0.1, align=center, rounded corners] (background) {};
    \node[fit=(t4) (b4), inner sep=5pt, fill=blue, opacity=0.1, align=center, rounded corners] (background) {};

    \draw[draw = red, dotted] ([yshift=.3em]m1.north) to ([yshift=0.6cm] m1.north) to node[draw = black, above= 4pt,font=\footnotesize, solid](sd1){\region{1}}  ([yshift=0.6cm]m3.north) to ([yshift=.3em]m3.north);

     \draw[draw = blue, dotted ] ([yshift=-.3em]m2.south) to ([yshift=-0.6cm] m2.south) to node[draw = black, opacity = 1, below = 4pt,font=\footnotesize, solid](sd2){\region{2}}  ([yshift=-0.6cm]m4.south) to ([yshift=-.3em]m4.south);
\end{tikzpicture}
\vspace{-13pt}
\caption[]{Possible interplay of semantics $\contract{}{x}$ and $\contract{}{y}$ when executing a program under the combined semantics $\contract{}{x+y}$.}
\label{figure:lifting-explanation-tr}
\end{figure}

\Cref{figure:lifting-explanation-tr} depicts two portions of traces produced when executing a program under the composed semantics $\contract{}{x+y}$.
The first trace (top) starts with a speculative transaction from semantics $\contract{}{x}$ (highlighted in red and starting with action $\startl{x}$).
Inside the red speculative transaction, there is a \emph{nested} speculative transaction from semantics $\contract{}{y}$ (highlighted in green, starting with action $\startl{y}$ and ending with action  $\rollbl{y}$).
After the termination of the nested transaction, the outer red transaction continues until the end of its speculative window (action $\rollbl{x}$).
Dually, the second trace (bottom) starts with a speculative transaction from $\contract{}{y}$ followed by a nested transaction from $\contract{}{x}$.
 
Thus, there are three different regions that might results in leaks:
\begin{asparaitem}
\item \textbf{\region{1}}: the speculative transaction started by $\contract{}{x}$ (highlighted in red),
\item \textbf{\region{2}}: the speculative transaction started by $\contract{}{y}$ (highlighted in blue), and 
\item \textbf{\region{3}}: the nested transactions (highlighted in green).
\end{asparaitem}

While leaks in region 1 are fixed by proving the security of a compiler for spec. semantics $\contract{}{x}$, i.e., $\contract{}{x} \vdash \comp{\cdot} : \rssp$, we need additional conditions to ensure the absence of leaks in Regions 2 and 3.

Thus, we introduce the notions of Independence w.r.t. a spec. semantics $\contract{}{y}$ and Safe Nesting which ensure the absence of leaks in Regions 2 and 3 respectively.

\subsection{Independence}

A compiler $\comp{\cdot}$ for the origin semantics $\contract{}{x}$ is called independent for the extension semantics $\contract{}{y}$ iff
the compiler does not introduce further leaks under the extension semantics $\contract{}{y}$.

\begin{definition}[Independence in extension]\label{def:ind}
    $\IND{\contract{}{y}}{\comp{\cdot}} \isdef \forall p.\ \contract{}{y} \vdash p : \rss \implies \contract{}{y} \vdash \comp{p} : \rss$
\end{definition}

\begin{corollary}[Self Independence]\label{lem:self-ind-tr}
    If $\contract{}{x} \vdash \comp{\cdot} : \rssp$ then 
    
    \noindent $\IND{\contract{}{x}}{\comp{\cdot}}$
\end{corollary}
\begin{proof}

    Assume we have $\contract{}{x} \vdash P : \rss$. We now need to show that $\contract{}{x} \vdash \comp{P} : \rss$.

    Unfolding $\contract{}{x} \vdash \comp{\cdot} : \rssp$ we get $\contract{}{NS} \vdash \src{P} : \rss \text{ then } \contract{}{x} \vdash \comp{\src{P}} : \rss$

    Because of \Thmref{thm:ss-sni-source} we have $\contract{}{NS} \vdash \src{P} : \rss$ and thus can use $\contract{}{x} \vdash \comp{\cdot} : \rssp$ to get $\contract{}{x} \vdash \comp{\src{P}} : \rss$.

    With this, we are done.
\end{proof}

\subsection{Safe Nesting}
The last we need to talk about nesting. For that we define a predicate safe nesting 
\begin{definition}[Safe nested Speculation]\label{def:safe-nest}
    Given a program $p$, we say that the program has safe nesting, written $$\contract{}{x+y} \vdash p : \safeN{}$$,
    iff
    \begin{align*}
        \contract{}{x+y} \vdash P : \safeN{} \isdef&\
            \forall \tauStack \in \behxy{(P)}, 
            \text{ if }
            \\
            &\
            \text{(1) }  \startl{a} \cdot \tauStack' \cdot \rollbl{a} \text{ is a subtrace of $\tauStack$}
            \\
            &\
            \text{(2) } \startl{b} \cdot \tauStack'' \cdot \rollbl{b} \text{ is a subtrace of $\tauStack'$} 
            \\
            &\
            \text{ then }
            \safe{\tauStack''} 
            \\
            &\ \quad \text{ where }
            a \in \{x,y\}
            \text{ and } 
            b \in \{x,y\} \setminus \{a\}
        \end{align*}
\end{definition}

Finally, we say that a compiler satisfies the Safe Nested Speculation property, written $\contract{}{x+y} \vdash \comp{\cdot} : \safeN{}$, iff all its compiled programs satisfy \Cref{def:safe-nest}, i.e., $\forall p.\ \contract{}{x+y} \vdash \comp{p} : \safeN{}$.

\subsection{Conditional Robust Speculative Safety}
Often compilers implementing Spectre countermeasures are developed to prevent leaks introduced by a \emph{specific} speculation mechanism.

\rssp{} is too strict as a criterion here, since it cannot distinguish between the different speculation mechanisms that might induce the leak.
To account for this, we propose a new secure compilation criterion called Conditional Robust Speculative Safety Preservation ($\crssp$, \Cref{def:crssp-tr}).
As the name indicates, $\crssp$ is a variant of $\rssp$ that restricts $\rss$ preservation \emph{only} to those programs that do not contain leaks caused only by the additional mechanisms. 

\begin{definition}[Conditional Robust Speculative Safety Preservation (\crssp)]\label{def:crssp-tr}
    $\contract{}{x}, \contract{}{y} \vdash  \comp{\cdot} : \crssp \isdef \forall \src{P} \in \src{L}. \text{ if }
      \contract{}{NS} \vdash \src{P} : \rss \text{ and } \contract{}{y} \vdash \src{P} : \rss \text{ then } \contract{}{x+y} \vdash \comp{\src{P}} : \rss$
  \end{definition}

\subsection{Lifted Compiler Preservation}

Using Independence and Safe Nesting, we are now ready to state our main theorem:

\begin{theorem}[Lifted Compiler Preservation]\label{cor:lift-comp-pres}
    If 
    $\contract{}{x} \vdash  \comp{\cdot} : \rssp$
    and
    $\IND{\contract{}{y}}{\comp{\cdot}}$
    and
    $\contract{}{x+y} \vdash \comp{\cdot} : \safeN{}$
    and
    $\wfc{\contract{}{x+y}}$ 
    then
    $\contract{}{x}, \contract{}{y} \vdash  \comp{\cdot} : \crssp$
\end{theorem}
\begin{proof}
    Assume a program $P$ and  $\contract{}{NS} \vdash \src{P} : \rss \text{ and } \contract{}{y} \vdash \src{P} : \rss$.
    We need to show $\contract{}{x+y} \vdash \comp{\src{P}} : \rss$.

    Let us first collect we know:
    
    From $\IND{\contract{}{y}}{\comp{\cdot}}$ and $\contract{}{y} \vdash \src{P} : \rss$ we get $\contract{}{y} \vdash \comp{\src{P}} : \rss$. 

    From $\contract{}{NS} \vdash \src{P} : \rss$ and $\contract{}{x} \vdash  \comp{\cdot} : \rssp$ we get $\contract{}{x} \vdash \comp{\src{P}} : \rss$.

    By $\contract{}{x+y} \vdash \comp{\cdot} : \safeN{}$ we know that nested speculative actions are safe.

    $\tauStack$ that is generated for $\contract{}{xy}$ for $\comp{\src{P}}$.

    The non-speculative parts are safe by definition. They cannot produce unsafe actions.

    Lets look at the speculative parts of $\tauStack$.

    \begin{description}
        \item[Inside a speculation of x]
            Now lets look at the trace actions that can be generated
            Now lets look at the subtraces $\tauStack'$:
            \begin{description}
                \item[$\tauStack' = \startObsy{\ctr} \cdot \tauStack'' \cdot \rollbackObsy{\ctr}$]
                    Thus, we have nesting of speculative transactions.
                    By $\contract{}{x+y} \vdash \comp{p} : \safeN{}$  we know that nesting is safe.
                \item[otherwise]
                    From above, we have $\contract{}{x} \vdash \comp{\src{P}} : \rss$. Together with Projection Preservation of $\wfc{\contract{}{x+y}}$ 
                    we know that these actions are safe
            \end{description}
        \item[Inside a speculation of y]
            Now lets look at the trace actions that can be generated
            \begin{description}
                \item[$\tauStack' = \startObsx{\ctr} \cdot \tauStack'' \cdot \rollbackObsx{\ctr}$]
                By $\contract{}{x+y} \vdash \comp{p} : \safeN{}$  we know that nesting is safe.
                \item[otherwise]
                From above, we have $\contract{}{y} \vdash \comp{\src{P}} : \rss$ (Independence).
            \end{description}
    \end{description}

    We are now done.
\end{proof}

This means our compilers can be embedded under even stronger attackers without needing to worry if there are bad interactions with the inserted instructions of the compiler and other speculative variants.
Essentially, our well-formed compilers are secure even under a stronger attacker that allows for more speculation.

\subsection{The nesting condition and how we fulfil it}
We now explain why we have this nesting condition here.

Consider the different kinds of nesting of speculative transactions that are allowed.

\begin{tikzpicture}

\node (node1) {$st_{x}$};
\node[right=of node1] (node2) {$\dots$};
\node[right=of node2] (node3) {$st_y$};
\node[right=of node3] (node4) {$\dots$};
\node[right=of node4] (node5) {$rb_y$};
\node[right=of node5] (node6) {$\dots$};
\node[right=of node6] (node7) {$rb_x$};
\draw [decorate, decoration={brace, amplitude=5pt, raise=-5pt}] (node1.south west) -- (node2.south east |- node1.south east) node[midway, below=7pt] {RSSP};

\draw [decorate, decoration={brace, amplitude=5pt, raise=-5pt}] (node3.south west) -- (node5.south east) node[midway, below=7pt] {??};

\draw [decorate, decoration={brace, amplitude=2pt, raise=-5pt}] (node6.south west) -- (node7.south |- node6.south) node[midway, below=7pt] {RSSP};

\end{tikzpicture}

\begin{tikzpicture}

\node (node1) {$st_{y}$};
\node[right=of node1] (node2) {$\dots$};
\node[right=of node2] (node3) {$st_x$};
\node[right=of node3] (node4) {$\dots$};
\node[right=of node4] (node5) {$rb_x$};
\node[right=of node5] (node6) {$\dots$};
\node[right=of node6] (node7) {$rb_y$};
\draw [decorate, decoration={brace, amplitude=5pt, raise=-5pt}] (node1.south west) -- (node2.south east |- node1.south east) node[midway, below=7pt] {Independence};

\draw [decorate, decoration={brace, amplitude=5pt, raise=-5pt}] (node3.south west) -- (node5.south east) node[midway, below=7pt] {??};

\draw [decorate, decoration={brace, amplitude=2pt, raise=-5pt}] (node6.south west) -- (node7.south |- node6.south) node[midway, below=7pt] {Independence};

\end{tikzpicture}

This of course is a very strong property. We define a second definition that almost all our compilers fulfill.

We call speculation trivial iff the only observations made between the start of speculation is the rollback observation.

\begin{definition}[Trapped Speculation]
    $\forall \tau \in \specProject{\tauStack}, \exists \ctr, \tau = \rollbackObsx{\ctr}$
\end{definition}

\begin{definition}[Trapped Speculation of Compiler]\label{def:trapped-spec}
We write $\contract{}{x} \vdash \comp{\cdot} : \trappedC{}$ 
iff $\forall p \tauStack \in \behx{\comp{p}}$,
then $\forall \tau \in \specProject{\tauStack}, \exists \ctr, \tau = \rollbackObsx{\ctr}$ 
\end{definition}

Most of our compilers stop or trap speculation outright (the only exception is $\compuslhB{\cdot}$) and thus fulfill the condition.

Our image then looks like this: 

\begin{tikzpicture}

\node (node1) {$st_{y}$};
\node[right=of node1] (node2) {$\dots$};
\node[right=of node2] (node3) {$st_x$};
\node[right=of node3] (node5) {$rb_x$};
\node[right=of node5] (node6) {$\dots$};
\node[right=of node6] (node7) {$rb_y$};
\draw [decorate, decoration={brace, amplitude=5pt, raise=-5pt}] (node1.south west) -- (node2.south east |- node1.south east) node[midway, below=7pt] {Independence};

\draw [decorate, decoration={brace, amplitude=5pt, raise=-5pt}] (node3.south west) -- (node4.south east |- node3.south east) node[midway, below=7pt] {Trapped Speculation};

\draw [decorate, decoration={brace, amplitude=2pt, raise=-5pt}] (node6.south west) -- (node7.south |- node6.south) node[midway, below=7pt] {Independence};

\end{tikzpicture}

And from this new notion we can show that we fulfill the safe nesting condition:
\begin{lemma}[Trapped speculation implies safe nesting]\label{lem:trapped-imp-safe-nest}
    If $\contract{}{x} \vdash \comp{\cdot} : \trappedC{}$ then $ \contract{}{x+y} \vdash \comp{\cdot} : \safeN{}$

\end{lemma}
\begin{proof}
    For nesting we need to look at the specific nestings:
    \begin{description}
        \item[A y speculation nested in a x speculation]
            Cannot happen because $\comp{\cdot}$ has trivial speculation. Thus, a speculation of y
            cannot happen.
        \item[A x speculation nested in y speculation]
            Because $\comp{\cdot}$ is trivial, the speculation that is started is trivial. 
            Thus, the nesting is safe.
    \end{description}
\end{proof}

\subsection{Syntactic Independence}

Using our running example of $\specj$ and $\specr$ and their combination, we can see that the newly introduced $\retC$ instructions could interact with $\specr$ in the combination. That is why we have the second condition of well-formedness. However, what is with the other direction? The $\specr$ compiler $\complfenceR{\cdot}$ inserts only fence instructions into the program.

This is important since many of the compilers we defined do not have interactions with other speculative mechanisms. Thus, we do not need to do a full proof that they do not interact but instead can use this easier way to fulfil the condition

Before formalizing SI, we introduce some notation.
Given a compiler $\comp{\cdot}$, we denote by $\insertedInstrs{\comp{\cdot}}$ the set of instructions that the compiler inserts during compilation.
For instance, for a simple compiler  $\complfenceB{\cdot}$ that inserts $\kywd{spbarr}$ instructions after branch instruction to prevent speculation~\cite{S_sec_comp}, $\insertedInstrs{\complfenceB{\cdot}}$ is the set  $\{ \kywd{spbarr} \}$.
Given a semantics $\contract{}{}$, we denote by $\speculationInstrs{\contract{}{}}$ the set of instructions that trigger speculation in $\contract{}{}$.
For example, $\pjz{x}{l}$ for $\contract{}{\Bv}$.

\begin{definition}[Syntactic Independence]
    A compiler $\comp{\cdot}$  is  {syntactically independent} for a contract $\contract{}{}$, written $\SI{\contract{}{}}{\comp{\cdot}}{}$, iff $\insertedInstrs{\comp{\cdot}} \cap \speculationInstrs{\contract{}{}} = \emptyset \text{ and } \insertedInstrs{\comp{\cdot}} \cap \{\jzC{}, \jmpC, \storeC, \loadC, \leftarrow  \} = \emptyset$.
\end{definition}

\begin{corollary}\label{cor:syn-ind-ind}
    If
    \begin{enumerate}
        \item $\SI{\contract{}{x}}{\comp{\cdot}}{\contract{}{y}}$
    \end{enumerate}
    Then
    \begin{enumerate}
        \item $\IND{\contract{}{y}}{\comp{\cdot}}$
    \end{enumerate}
\end{corollary}
\begin{proof}
    Assume we have a program $p$ with $\contract{}{y} \vdash p : \rss$.
    
    Because of $\SI{\comp{\cdot}}{\contract{}{y}}{}$, we know that the compiler does not insert any instructions related to speculation of ${\contract{}{y}}$. Thus, no additional speculation related to $\contract{}{y}$ is introduced by the compiler.
    Since the initial program was secure for $\contract{}{y}$ and no additional speculation is added during compilation, we have $\contract{}{y} \vdash \comp{p} : \rss$
\end{proof}

\subsection{Non-Trivial Independence Results}

Here, we provide the results for the Independence results $\ITab$ in the table above.

We show the Independence result only for the "smallest" contract where the proof is necessary. In all combined contracts the same proof can be used. For example, if we have proven $\IND{\contract{}{\Rv}}{\compretpJ{\cdot}}{}$ then we have the same proof for $\IND{\contract{}{\Bv + \Rv}}{\compretpJ{\cdot}}{}$ as well.

\begin{lemma}
$\IND{\contract{}{\Rv}}{\compretpJ{\cdot}}{}$
\end{lemma}
\begin{proof}
    By definition of Independence, we know that forall programs $p$ that $\contract{}{\Rv} \vdash p : \rss$.
    Thus, we are only interested in newly added $\retC$ instructions from the compiler $\compretpJ{\cdot}$.
    These are only added when an indirect jump is encountered and is replaced with the following sequence.
    \begin{align*}
    \compretpJ{l : \pjmp{e}} &= \trgJ{
                                         \begin{aligned}[t]
                                                    &
                                                    l': \pcall{retpo\textunderscore trg\textunderscore l}  ~~ \text{\Cref{tr:v5-call}}
                                                        \\
                                                        &\
                                                        l'_1: \pskip  ~~ \text{\Cref{tr:v5-nospec-eps}}
                                                            \\
                                                            &\
                                                            l'_2 : \barrierKywd
                                                            \\
                                                            &\
                                                            l'_3 : \pjmp{l'_1}
                                                            \\
                                                            &\
                                                            l'_4 : modret \compretpJ{e}  ~~ \text{\Cref{tr:v5-nospec-eps}}
                                                            \\
                                                            &\
                                                            l'_5: \pret  ~~ \text{\Cref{tr:v5-spec}}
                                            \end{aligned}
                                }
    \end{align*}

    Here we annotated, which rule is used from the $\semr$ semantics when this code sequence is executed in the interesting case when speculation is stated. Otherwise, this sequence is trivially secure.

    When speculation is started, the topmost entry of the $\Rsb$ is taken, which is $l'_1$, because the last $\callC$ instruction was the one in $l'$.
    The code sequence at $l'_1, l'_2$ and $l'_3$ is an infinite loop containing a barrier instruction. Thus, speculation is trapped and stopped at $l_2'$. This causes a rollback to occur. Most importantly, during speculation no memory is accessed and no observations apart from the rollback were made. Thus, the speculation is secure.

    Since this is the only additional source of speculation, we have that $\contract{}{\Rv} \vdash \compretpJ{p} : \rss$.
\end{proof}

Next, we provide the Independence result for the extended retpoline compiler $\compretpJF{\cdot}$.

\begin{lemma}
$\IND{\contract{}{\SLSv}}{\compretpJF{\cdot}}{}$
\end{lemma}
\begin{proof}
    By definition of Independence, we know that for all programs $p$ that $\contract{}{\SLSv} \vdash p : \rss$.
    Thus, we are only interested in newly added $\retC$ instructions from the compiler $\compretpJF{\cdot}$.
    These are only added when an indirect jump is encountered and is replaced with the following sequence.

    \begin{align*}
    \compretpJ{l : \pjmp{e}} &= \trgJ{
                                         \begin{aligned}[t]
                                                    &
                                                    l': \pcall{retpo\textunderscore trg\textunderscore l}  ~~ \text{\Cref{tr:sls-nospec-act}}
                                                        \\
                                                        &\
                                                        l'_1: \pskip  ~~ \text{\Cref{tr:sls-nospec-eps}}
                                                            \\
                                                            &\
                                                            l'_2 : \barrierKywd
                                                            \\
                                                            &\
                                                            l'_3 : \pjmp{l'_1}
                                                            \\
                                                            &\
                                                            l'_4 : modret \compretpJ{e}  ~~ \text{\Cref{tr:sls-nospec-eps}}
                                                            \\
                                                            &\
                                                            l'_5: \pret  ~~ \text{\Cref{tr:sls-spec}}
                                                            \\
                                                            &\
                                                            l'_6 : \barrierKywd
                                            \end{aligned}
                                }
    \end{align*}

When speculation is started with \Cref{tr:sls-spec}, it is immediately stopped by the additional barrier in $l'_6$.
Thus, the additional speculation is secure.
\end{proof}

\subsection{Trivial Speculation results}
Here, we describe the compilers that satisfy \Thmref{def:trapped-spec} and together with \Thmref{lem:trapped-imp-safe-nest} we get one of the preconditions for our lifting theorem \Thmref{cor:lift-comp-pres}.

\begin{table}[h]
    \centering
    \begin{tabular}{c|c}
    \toprule
        Compiler &  Fulfills \Thmref{def:trapped-spec} \\
        \midrule
        $\complfenceS{\cdot}$ & \greencheck \\
        $\complfenceR{\cdot}$ & \greencheck \\
        $\complfenceSLS{\cdot}$ & \greencheck\\
        $\compretpJ{\cdot}$ & \greencheck \\
        $\compretpJF{\cdot}$ & \greencheck \\
        $\compretpR{\cdot}$ & \greencheck \\
        $\complfenceB{\cdot}$ & \greencheck \\
        $\compuslhB{\cdot}$ & \redxmark \\
        $\compsslhB{\cdot}$ & \redxmark \\
        \bottomrule
    \end{tabular}
\end{table}

These follow from the definition of the compilers.
 \section{Speculative Safety Combined}

We will write $\wfcn{\contract{}{xy}}{n}$ to refer to the n-th fact of \Thmref{def:wellformed}

From these conditions, we can generally derive that SS overapproximates SNI for a general well-formed combination (\Thmref{thm:Comb-ss-impl-sni}).

\begin{theorem}[\sstext{} implies SNI]\label{thm:Comb-ss-impl-sni}
If
\begin{enumerate}
    \item $\contract{}{xy}\vdash P : \ss $
\end{enumerate}
Then
\begin{enumerate}[label=\Roman*]
    \item $\contract{}{xy}\vdash P : \sni $
\end{enumerate}

\end{theorem}
\begin{proof}
Let $P$ be an arbitrary program such that $\contract{}{xy}\vdash P : \ss$.
Proof by contradiction. 

Assume that $\contract{}{xy} \vdash P : \sni$ does not hold. That is, there is another program $P'$ and traces $\tra{_1} \in \ \behav{\SInit{P}} $ and  $\tra{_2} \in \behav{\SInit{P'} }$ such that $P \loweq P'$, $\nspecProject{\tra{_1}} = \nspecProject{\tra{_2}}$, and $\tra{_1} \neq \tra{_2}$.

Since $\trgxy{\tra{_1}} \in \trgxy{ \behav{\SInit{P}} }$ we have $\trgxy{\amTracevxy{P}{\tauStackT}}$ and by \Cref{tr:xy-trace} we know that there exists $\Sigmaxy^{F}$ such that $\trgxy{\Sigmaxy^0(P) \bigspecarrowxy{\tra{_1}}\ \Sigmaxy^{F}}$. Similarly, we get  $\trgxy{\Sigmaxy^0(P') \bigspecarrowxy{\tra{_2}}\ \Sigmaxy^{F'}}$.

Combined with the fact that $\tra{_1} \neq \tra{_2}$, it follows that there exists speculative states $\Sigmaxyt{\dagger}, \Sigmaxyt{\dagger\dagger}, \Sigmaxyt{'}, \Sigmaxyt{''}$ and sequences of observations $\trgxy{\tauStack, \tauStack_{end}, \tauStack'_{end}}, \trgxy{\aca{^\taint}, {\alpha'}^{\taint'} } $ such that $\trgxy{\aca{^{\taint}}} \neq \trgxy{ {\alpha'}^{\taint'}}$  and:

\begin{align*}
\trgxy{\Sigmaxy^0(P)} & \trgxy{\bigspecarrowxy{\tauStack}} \Sigmaxyt{\dagger} \specarrowxy{\aca{^{\taint}}} \Sigmaxyt{\dagger\dagger} \bigspecarrowxy{\tauStack_{end}} \trgxy{\Sigmaxy^{F}} 
\\
\trgxy{\Sigmaxy^0(P')} & \bigspecarrowxy{\tauStack} \Sigmaxyt{'} \specarrowxy{{\alpha'}^{\taint'}} \Sigmaxyt{''} \bigspecarrowxy{\tauStack'_{end}} \trgxy{\Sigmaxy^{F'}}
\end{align*}

From $\contract{}{xy} \vdash P : \ss $ we get that $\safe{\tra{_1}}$.

From \Thmref{lemma:comb-low-equivalent-programs-low-equivalent-states} we get ${\Sigmaxy^0(P)} \relsa {\Sigmaxy^0(P')}$.

Using \Thmref{lemma:xy-bigspecarrow-preserves-safe} on $\Sigmaxy^0(P) \bigspecarrowxy{\tauStack} \Sigmaxyt{\dagger}$ we get  $\Sigmaxy^0(P') \bigspecarrowxy{\tauStack} \Sigmaxyt{'''}$ and $\Sigmaxyt{\dagger} \relsa \Sigmaxyt{'''}$.

By $\wfcn{\contract{}{xy}}{1}$ (Confluence) we know that $\Sigmaxyt{'} = \Sigmaxyt{'''}$ and thus we have $\Sigmaxyt{\dagger} \relsa \Sigmaxyt{''}$.

We know apply $\wfcn{\contract{}{xy}}{3}$ (Relation Preservation) on $\Sigmaxyt{\dagger} \specarrowxy{\aca{^{\taint}}} \Sigmaxyt{\dagger\dagger}$ and $\Sigmaxyt{\dagger} \relsa \Sigmaxyt{'}$ and get $\Sigmaxyt{'} \specarrowxy{\aca{^{\taint}}}  \Sigmaxyt{''''}$.

Again, by Confluence we know that $\Sigmaxyt{''} = \Sigmaxyt{''''}$.

However, now we have that the same observation $\aca{^{\taint}}$ is made in both executions contradicting our initial assumption

\end{proof}

From this we can follow that RSS implies RSNI as well.

\begin{theorem}[RSS implies RSNI]\label{thm:Comb-rss-impl-rsni}
If
    \begin{enumerate}
        \item $\contract{}{xy}\vdash P : \rss$
    \end{enumerate}
Then
    \begin{enumerate}[label=\Roman*]
        \item $\contract{}{xy}\vdash P : \rsni $
    \end{enumerate}

\end{theorem}

\begin{proof}
    Unfolding the definition of $\contract{}{xy} \vdash \com{P} : \rsni$, we get an attacker $A$ such that  $ \ctxc{} : \com{atk}$ and we need to prove $\vdash \com{\ctxc{}\hole{P}} : \sni$.
    Using that attacker $\ctxc{}$ on $\contract{}{xy}\vdash P : \rss$ we get $\contract{}{xy} \vdash\com{\ctxc{}\hole{P}} : \ss$.

    Now we can apply \Thmref{thm:Comb-ss-impl-sni} on $\contract{}{xy} \vdash\com{\ctxc{}\hole{P}} : \ss$ and get $\contract{}{xy} \vdash \com{P} : \sni$.

    Thus, we are finished.

\end{proof}

And we can derive the following Corollary.
\begin{corollary}
If 
\begin{enumerate}
    \item $\wfc{\contract{}{xy}}$
    \item $\contract{}{xy}\vdash P : \ss$
\end{enumerate}
Then
\begin{enumerate}[label=\Roman*]
    \item $\contract{}{x} \vdash P : \ss$ and 
    \item $\contract{}{y}\vdash P : \ss$
\end{enumerate}
\end{corollary}

\begin{proof}

Unfolding $\contract{}{xy} \vdash P : \ss $ we have:
$\forall \trac{^\sigma}\in\behavc{P}\ldotp \forall \acac{^\sigma}\in\trac{^\sigma}\ldotp \com{\sigma}\equiv\com{\safeta}$.
Choose an arbitrary $\trac{^\sigma}\in\behavc{P}$. We know it is $\safe{\trac{^\sigma}}$.

By $\wfcn{\contract{}{xy}}{2}$ (Projection Preservation), we know that $\behx{p} = \specProjectxyx{\behxy{(p)}}$ and $\behy{p} = \specProjectxyy{\behxy{(p)}}$.

Furthermore, by definition of the projections $\specProjectxyx{}$ and $\specProjectxyy{}$ they only delete elements of a trace but never add observation to it.

Since $\safe{\trac{^\sigma}}$, we know by definition that $\safe{\specProjectxyx{\trac{^\sigma}}}$ and $\safe{\specProjectxyy{\trac{^\sigma}}}$.

By Projection Preservation we know that $\specProjectxyx{\trac{^\sigma}} \in \behx{P}$ and  $\specProjectxyy{\trac{^\sigma}} \in \behy{P}$.

Since the trace was arbitrarily chosen, we are finished.
\end{proof}

Furthermore, if a compiler satisfies a combined contract $\contract{}{xy}$, than it also satisfies the two component contracts $\contract{}{x}$ and $\contract{}{y}$, which follows from \Thmref{cor:weak-contract-sat} and Projection Preservation.

\begin{corollary}[Combined Compiler Contract Security (RSNIP)]
    \hfill

    \begin{enumerate}
        \item $\wfc{\contract{}{xy}}$ and 
        \item $\contract{}{xy} \vdash  \comp{\cdot} : \rsnip$
    \end{enumerate}
    then 
    \begin{enumerate}[label=\Roman*]
        \item $\contract{}{x} \vdash  \comp{\cdot} : \rsnip$ and
        \item $\contract{}{y} \vdash  \comp{\cdot} : \rsnip$
    \end{enumerate}
\end{corollary}
\begin{proof}
    
    From $\wfc{\contract{}{xy}}$ (Projection Preservation) we get $\contract{}{x} \sqsubseteq \contract{}{xy}$ and $\contract{}{y} \sqsubseteq \contract{}{xy}$.

    We can now apply \Thmref{cor:weak-contract-sat} on $\contract{}{xy} \vdash  \comp{\cdot} : \rssp$ and $\contract{}{x} \sqsubseteq \contract{}{xy}$ and $\contract{}{y} \sqsubseteq \contract{}{xy}$ respectively.

    Thus, we are done.
\end{proof}

We can state it for RSSP as well:
\begin{corollary}[Combined Compiler Contract Security (RSSP)]
    \hfill

    \begin{enumerate}
        \item $\wfc{\contract{}{xy}}$ and 
        \item  $\contract{}{xy} \vdash  \comp{\cdot} : \rssp$
    \end{enumerate}
    then 
    \begin{enumerate}[label=\Roman*]
        \item $\contract{}{x} \vdash  \comp{\cdot} : \rssp$ and
        \item $\contract{}{y} \vdash  \comp{\cdot} : \rssp$
    \end{enumerate}
\end{corollary}
\begin{proof}

    Assume a program p.
    By $\contract{}{xy} \vdash  \comp{\cdot} : \rssp$ we know that there are only safe actions in $\behxy{\comp{p}}$.

    From $\wfc{\contract{}{xy}}$ (Projection Preservation) we get $\behx{\comp{p}} = \specProjectxyx{\behxy{\comp{p}}}$ and similar $\behy{\comp{p}} = \specProjectxyx{\behxy{\comp{p}}}$.

    Thus, there are only safe actions in $\behx{\comp{p}}$ and $\behy{\comp{p}}$.

    Since p was chosen arbitrarily, we are done.
\end{proof}

Dually, if a program is insecure w.r.t. one of the components contracts, then it is insecure w.r.t. the composition contract as well.

\begin{corollary}[Combined Compiler Contract Security Negative (RSNIP)]
    \hfill 

    \begin{enumerate}
        \item $\wfc{\contract{}{xy}}$ and 
        \item  $\contract{}{x} \nvdash  \comp{\cdot} : \rsnip$ or $\contract{}{y} \nvdash  \comp{\cdot} : \rsnip$
    \end{enumerate}
    Then 
    \begin{enumerate}[label=\Roman*]
        \item $\contract{}{xy} \nvdash  \comp{\cdot} : \rsnip$
    \end{enumerate}
\end{corollary}
\begin{proof}
    We do the proof for $\contract{}{x} \nvdash  \comp{\cdot} : \rsnip$.

    The proof for $\contract{}{y} \nvdash  \comp{\cdot} : \rsnip$ is analogous.

    Because of $\contract{}{x} \nvdash  \comp{\cdot} : \rsnip$ and $\contract{}{x} \sqsubset \contract{}{xy}$ from Projection Preservation, we can apply \Thmref{cor:weak-contract-unsat} and immediately get $\contract{}{xy} \nvdash  \comp{\cdot} : \rsnip$
    
\end{proof}

\begin{corollary}[Combined Compiler Contract Security Negative (RSSP)]
    \hfill 

    \begin{enumerate}
        \item $\wfc{\contract{}{xy}}$ and 
        \item  $\contract{}{x} \nvdash  \comp{\cdot} : \rssp$ or $\contract{}{y} \nvdash  \comp{\cdot} : \rssp$
    \end{enumerate}
    Then 
    \begin{enumerate}[label=\Roman*]
        \item $\contract{}{xy} \nvdash  \comp{\cdot} : \rssp$
    \end{enumerate}
\end{corollary}
\begin{proof}
    We do the proof for $\contract{}{x} \nvdash  \comp{\cdot} : \rssp$.

    The proof for $\contract{}{y} \nvdash  \comp{\cdot} : \rssp$ is analogous.

    Because of $\contract{}{x} \nvdash  \comp{\cdot} : \rssp$, we know there exists a program $p$ such that $\contract{}{NS} \vdash p : \rss$ and $\contract{}{x} \nvdash \comp{p} : \rss$.

    That is, there is an unsafe action $\tau^{\unta}$ in the trace of $\comp{p}$.

    From $\wfc{\contract{}{xy}}$ (Projection Preservation) we get $\behx{\comp{p}} = \specProjectxyx{\behxy{\comp{p}}}$.

    This means the unsafe action $\tau^{\unta}$ is part of the combined trace of $\comp{p}$ as well!
    Thus, $\contract{}{xy} \nvdash \comp{p} : \rss$.

    Since we have $\contract{}{NS} \vdash p : \rss$ and  $\contract{}{xy} \nvdash \comp{p} : \rss$ we can conclude that $\contract{}{xy} \nvdash  \comp{\cdot} : \rssp$.
\end{proof}

\subsection{Using Relation Preservation on the Combination}
Next, we show how we can lift Relation Preservation of the combination to the reflexive transitive closure: 

\begin{lemma}[Comb: Steps of $\bigspecarrowxy{}$ with safe observations preserve Safe-equivalence]\label{lemma:xy-bigspecarrow-preserves-safe}
If
\begin{enumerate}
    \item $\Sigmaxy \bigspecarrowxy{\tauStack} \Sigmaxyt{\dagger}$ and 
    \item $\Sigmaxy \relsa \Sigmaxyt{'}$
    \item $\safe{\tauStack}$ \item $\wfc{\contract{}{xy}}{}$
\end{enumerate}
Then there exists $\Sigmaxyt{''}$ such that
\begin{enumerate}
    \item $\Sigmaxyt{'} \bigspecarrowxy{\tauStack} \Sigmaxyt{''}$ and
    \item $\Sigmaxyt{\dagger} \relsa \Sigmaxyt{''}$
\end{enumerate}
\end{lemma}

\begin{proof}
The proof proceeds by induction on $\Sigmaxy \bigspecarrowxy{\tauStack} \Sigmaxyt{\dagger}$:

\begin{description}
    \item[\Cref{tr:xy-reflect}]
    
    Then we have $\Sigmaxy \bigspecarrowxy{\tauStack} \Sigmaxyt{\dagger}$ where $\tauStack = \empTr$ and $\Sigmaxyt{\dagger} = \Sigmaxy$.
    
    Thus, we use \Cref{tr:xy-reflect} on $\Sigmaxyt{'}$ and get $\Sigmaxyt{'} \bigspecarrowxy{\tauStack'} \Sigmaxyt{''}$, where where $\tauStack' = \empTr$ and $\Sigmaxyt{''} = \Sigmaxyt{'}$.
    
    Thus, $\tauStack = \tauStack'$ trivially holds and $\Sigmaxyt{\dagger} \relsa \Sigmaxyt{''}$ holds by $\Sigmaxy \relsa \Sigmaxyt{'}$.
    
    \item[\Cref{tr:xy-silent}]
    
    Then we have $\Sigmaxy \bigspecarrowxy{\tra{^\taint}} \Sigmaxyt{\dagger\dagger}$ and $\Sigmaxyt{\dagger\dagger} \specarrowxy{\epsilon} \Sigmaxyt{\dagger}$.
    
    Applying the IH on $\Sigmaxy \bigspecarrowxy{\tra{^\taint}} \Sigmaxy''$ we get:
    \begin{align*}
    \Sigmaxyt{'} \bigspecarrowxy{\tauStack''} \Sigmaxyt{'''} \\
    \Sigmaxyt{\dagger\dagger} \relsa \Sigmaxyt{'''} \\
    \safe{\tauStack''}
    \end{align*}

    Applying $\wfcn{\contract{}{xy}}{3}$ (Relation Preservation) on $\Sigmaxyt{\dagger\dagger} \specarrowxy{\epsilon} \Sigmaxyt{\dagger}$ using $\Sigmaxyt{\dagger\dagger} \relsa \Sigmaxyt{'''}$ and $\safe{\empTr}$ we get $\Sigmaxyt{'''} \specarrowxy{\empTr} \Sigmaxyt{*}$ and $\Sigmaxyt{\dagger} \relsa \Sigmaxyt{*}$.
    
    Now choose $\Sigmaxyt{''} = \Sigmaxyt{*}$.
    
    We now use \Cref{tr:xy-silent} and get $\Sigmaxyt{'} \bigspecarrowxy{\tauStack} \Sigmaxyt{''}$
    
    $\Sigmaxyt{\dagger} \relsa \Sigmaxyt{''}$ follows from $\Sigmaxyt{\dagger} \relsa \Sigmaxyt{*}$ and $\Sigmaxyt{''} = \Sigmaxyt{*}$.
    
    \item[\Cref{tr:xy-single}]
    
    Then we have $\Sigmaxy \bigspecarrowxy{\tra{^\taint}} \Sigmaxyt{\dagger\dagger}$ and $\Sigmaxyt{\dagger\dagger} \specarrowxy{\tau^{\taint}} \Sigmaxyt{\dagger}$.
    
    Furthermore, we get:
    \begin{align*}
    \Sigmaxy'' =& \phiStackxy \cdot  \tup{\Omega, n} \\
    \Sigmaxy' =& \phiStackxy \cdot \tup{\Omega', n' } \\
    \ffun{\sigma''(\pc)} =& f \\
    \ffun{\sigma'(\pc)} =& f' \\
    \text{ if } f == f'& \text{ and } ~f \in \OB{I} \text{ then } \tau^{\taint} = \epsilon \text{ else } \tau^{\taint} = \tau^{\taint}
    \end{align*}
    
    Applying the IH on $\Sigmaxy \bigspecarrowxy{\tra{^\taint}} \Sigmaxy''$ we get:
    \begin{align*}
    \Sigmaxyt{'} \bigspecarrowxy{\tauStack''} \Sigmaxyt{'''} \\
    \Sigmaxyt{\dagger\dagger} \relsa \Sigmaxyt{'''} \\
    \safe{\tauStack''}
    \end{align*}
    
    Applying $\wfcn{\contract{}{xy}}{3}$ (Relation Preservation) on $\Sigmaxyt{\dagger\dagger} \specarrowxy{\tau^{\taint}} \Sigmaxyt{\dagger}$ using $\Sigmaxyt{\dagger\dagger} \relsa \Sigmaxyt{'''}$ and $\safe{\tau^{\taint}}$ (since $\safe{\tauStack}$) we get $\Sigmaxyt{'''} \specarrowxy{\tau^{\taint}} \Sigmaxyt{*}$ and $\Sigmaxyt{\dagger} \relsa \Sigmaxyt{*}$.
    
    Now choose $\Sigmaxyt{''} = \Sigmaxyt{*}$.
    
    We can now use \Cref{tr:xy-single} and get $\Sigmaxyt{'} \bigspecarrowxy{\tauStack} \Sigmaxyt{''}$
    
    $\Sigmaxyt{\dagger} \relsa \Sigmaxyt{''}$ follows from $\Sigmaxyt{\dagger} \relsa \Sigmaxyt{*}$ and $\Sigmaxyt{''} = \Sigmaxyt{*}$.
\end{description}
\end{proof}

\begin{lemma}[Comb : Low-equivalent programs have low-equivalent initial states]\label{lemma:comb-low-equivalent-programs-low-equivalent-states}
\begin{align*}
\forall {P}, {P'}. &
	\text{ if } {P} \loweq {P'} 
	\text{ then } \SInit{P} \approx \SInit{P'}
\end{align*}
\end{lemma}

\begin{proof}

Follows from the definition of initial states.
\end{proof}

\subsection{Deriving Relation Preservation generally}\label{sec:comb-rel-pres-general}

Here we show how to derive Relation Preservation, a premise of $\wfc{}$, in a general fashion from additional assumptions about the source contracts

\begin{definition}[$\relsaxy$]
    $\Sigmaxy \relsaxy \Sigmaxy'$ iff
    $\specProjectxyx{\Sigmaxy} \relsax \specProjectxyx{\Sigmaxy'}$ and $\specProjectxyy{\Sigmaxy} \relsay \specProjectxyy{\Sigmaxy'}$
\end{definition}

\begin{assumption}[Relation Preservation of x]\label{ass:rel-preserve-x}

If
\begin{enumerate}
    \item $\Sigmaxt{\dagger} \relsax \Sigmaxt{'}$
    \item $\Sigmaxt{\dagger} \specarrowx{\tau} \Sigmaxt{\dagger\dagger}$
    \item $\safe{\tau}$
\end{enumerate}

Then there exists $\Sigmaxt{''}$
\begin{enumerate}[label=\Roman*]
    \item $\Sigmaxt{'} \specarrowx{\tau} \Sigmaxt{''}$
    \item $\Sigmaxt{\dagger\dagger} \relsax \Sigmaxt{''}$
\end{enumerate}

\end{assumption}

\begin{assumption}[Relation Preservation of y]\label{ass:rel-preserve-y}

If
\begin{enumerate}
    \item $\Sigmayt{\dagger} \relsay \Sigmayt{'}$
    \item $\Sigmayt{\dagger} \specarrowy{\tau} \Sigmayt{\dagger\dagger}$
    \item $\safe{\tau}$
\end{enumerate}

Then there exists $\Sigmayt{''}$
\begin{enumerate}[label=\Roman*]
    \item $\Sigmayt{'} \specarrowy{\tau} \Sigmayt{''}$
    \item $\Sigmayt{\dagger\dagger} \relsay \Sigmayt{''}$
\end{enumerate}

\end{assumption}

Here, we have two additional assumptions about the compatibility of the relations $\relsax$ and $\relsay$.
These assumptions state that we can recompose the relation $\relsaxy$ from $\relsax$ when doing a relation preserving step in the x semantics.
\begin{assumption}[Compatibility of relation $\relsax$]\label{ass:compat-x-rel}
    Let $\Sigmaxy = \phiStackxy \cdot \Phixy$ and $\Sigmaxy' = \phiStackxy' \cdot \Phixy'$.
    If
    \begin{enumerate}
        \item $\Sigmaxy \relsaxy \Sigmaxy'$ and
        \item $\specProjectxyx{\Phixy} \specarrowx{\tau} \phiStackx^\dagger$
        \item $\specProjectxyx{\Phixy'} \specarrowx{\tau} \phiStackx''$
        \item $\phiStackx^\dagger \relsax \phiStackx''$ and
        \item $\specProjectxyx{\phiStackxy^\dagger} = \phiStackx^\dagger$ and $\specProjectxyx{\phiStackxy''} = \phiStackx''$
        \item $\Sigmaxy^\dagger = \phiStackxy \cdot \phiStackxy^\dagger$ and $\Sigmaxy'' = \phiStackxy' \cdot \phiStackxy''$
    \end{enumerate}
    Then
    \begin{enumerate}
        \item $\Sigmaxy^\dagger \relsaxy \Sigmaxy''$
    \end{enumerate}
\end{assumption}

\begin{assumption}[Compatibility of relation $\relsay$]\label{ass:compat-y-rel}
    Let $\Sigmaxy = \phiStackxy \cdot \Phixy$ and $\Sigmaxy' = \phiStackxy' \cdot \Phixy'$.
    If
    \begin{enumerate}
        \item $\Sigmaxy \relsaxy \Sigmaxy'$ and
        \item $\specProjectxyy{\Phixy} \specarrowy{\tau} \phiStacky^\dagger$
        \item $\specProjectxyy{\Phixy'} \specarrowy{\tau} \phiStacky''$
        \item $\phiStacky^\dagger \relsay \phiStacky''$ and
        \item $\specProjectxyy{\phiStackxy^\dagger} = \phiStacky^\dagger$ and $\specProjectxyy{\phiStackxy''} = \phiStacky''$
        \item $\Sigmaxy^\dagger = \phiStackxy \cdot \phiStackxy^\dagger$ and $\Sigmaxy'' = \phiStackxy' \cdot \phiStackxy''$
    \end{enumerate}
    Then
    \begin{enumerate}
        \item $\Sigmaxy^\dagger \relsaxy \Sigmaxy''$
    \end{enumerate}
\end{assumption}

\begin{insight}
The safe condition is needed to get the equality of the traces. That is something that easily follows from low equivalence, since when a value is safe it can only depend on low equivalent values. Since they are low equivalent we are finished.
This proof relies on Relation Preservation (+ Obs equal) and Confluence 
\end{insight}

\begin{lemma}[Relation Preservation of $\specarrowxy{}$]\label{lemma:comb-relation-preservation-small-step}
If
\begin{enumerate}
    \item $\Sigmaxy \specarrowxy{\tau} \Sigmaxyt{\dagger}$ and
    \item $\Sigmaxy \relsaxy \Sigmaxyt'$ and
    \item \Thmref{ass:rel-preserve-x} and \Thmref{ass:rel-preserve-y} and
    \item $\safe{\tau}$
\end{enumerate}
Then
\begin{enumerate}
    \item $\Sigmaxy'\specarrowxy{\tau'} \Sigmaxyt{''}$ and
    \item $\Sigmaxyt{\dagger} \relsaxy \Sigmaxyt{''}$ and
    \item $\tau = \tau'$ 
\end{enumerate}
\end{lemma}

\begin{proof}
By $\Sigmaxy \relsaxy \Sigmaxyt'$, we know that $\Sigmaxy . \pc = \Sigmaxyt' . \pc$.

We start by inversion on $\Sigmaxy \specarrowxy{\tau} \Sigmaxyt{\dagger}$ and $\Sigmaxy'\specarrowxy{\tau'} \Sigmaxyt{''}$:
\begin{description}
    
    \item[\Cref{tr:xy-x-rollback}]
    Then we have $\Phixy \specarrowxy{\rollbackObsx} \emptyset$ with $\Phixyv . n = 0 \text{ or }\ \finType{\trgxy{\Phixyv}}$ and 
    $\specProjectxyx{\Phixy} \specarrowx{\rollbackObsx^{\safeta}} \emptyset$.
    
    By $\Sigmaxy \relsaxy \Sigmaxy'$, we have ${\Phixyv}' . n = \Phixyv . n = 0$ or $\finType{\trgxy{{\Phixyv}'}}$ and $\notfinType{\Sigmaxy}$ (since there exists a step) implies $\notfinType{\Sigmaxy'}$ because of $\Sigmaxy \relsaxy \Sigmaxy'$.
    
    Furthermore, $\notfinType{\Sigmaxy'}$ and $\Sigmaxy \relsaxy \Sigmaxy'$ together with $\specProjectxyx{\Phixy} \specarrowx{\rollbackObsx^\safeta} \emptyset$ implies $\notfinType{\specProjectxyx{\Sigmaxy'}}$ and thus $\notfinType{\specProjectxyx{\Phixy'}}$.
    
    This means there necessarily exists a $\phiStackx$ and $\tau'$ such that $\specProjectxyx{\Phixy'} \specarrowx{\tau'} \phiStackx$.
    By \Thmref{ass:rel-preserve-x} we can then conclude that $\emptyset \relsax \emptyset$ and $\tau \relsax \rollbackObsx^{\safeta}$ and together with \Thmref{ass:compat-x-rel} we get \item $\Sigmaxyt{\dagger} \relsaxy \Sigmaxyt{''}$ and
    \item $\tau \relsaxy \tau'$ as required.

    \item[\Cref{tr:xy-y-rollback}]
    Analogous to above.
    
    \item[\Cref{tr:xy-context}]
    Then we have $\Phixy \specarrowxy{\tau} \phiStackxy^\dagger$. 

    We proceed by inversion $\Phixy \specarrowxy{\tau} \phiStackxy^\dagger$:
        \begin{description}
            \item[\Cref{tr:xy-x-step}]
                Then we have $\specProjectxyx{\Phixy} \specarrowx{\tau} \specProjectxyx{\phiStackxy^\dagger}$.
            
                Furthermore, we have $\Phixy . \sigma(\pc) = \Phixy' . \sigma(\pc)$ because of $\Sigmaxy \relsaxy \Sigmaxyt'$. That means the same instruction is executed.

                Because of the definition of $\relsaxy$ we have $\specProjectxyx{\Phixy} \relsax \specProjectxyx{\Phixy'}$ as well.
                We can now use \Thmref{ass:rel-preserve-x} on  $\specProjectxyx{\Phixy} \relsax \specProjectxyx{\Phixy'}$ and $\specProjectxyx{\Phixy} \specarrowx{\tau} \specProjectxyx{\phiStackxy'}$ and $\safe{\tau}$ and get $\specProjectxyx{\Phixy'} \specarrowx \phiStackx''$ and $\specProjectxyx{\phiStackxy^\dagger} \relsax \phiStackx''$. 

                Choose $\phiStackxy''$ as $\specProjectxyx{\phiStackxy''} = \phiStackx''$

                We now use \Cref{tr:xy-x-step} with $\specProjectxyx{\Phixy'} \specarrowx \phiStackx''$ and get
                $\specProjectxyx{\Phixy'} \specarrowxy{\tau} \specProjectxyx{\phiStackxy''}$.

                It remains to show that $\Sigmaxyt{\dagger} \relsaxy \Sigmaxyt{''}$.

                Since only the topmost state changed, we need to show that $\phiStackxy^\dagger \relsaxy \phiStackxy''$.

                We use \Thmref{ass:compat-x-rel} on $\specProjectxyx{\phiStackxy^\dagger} \relsax \specProjectxyx{\phiStackxy''}$ since we fulfill all conditions and get
                $\phiStackxy^\dagger \relsaxy \phiStackxy''$ and are finished.

            \item[\Cref{tr:xy-y-step}]
                Analogous to above using \Thmref{ass:rel-preserve-y} and \Thmref{ass:compat-y-rel}.
        \end{description}

\end{description}
\end{proof}

 \section{WFC for all combinations}

Here, we argue why all our combined contracts presented are WFC.

In general, Confluence follows from the correct instantiation of the metaparameter $Z$. Our AM semantics delegate back to the non-speculative semantics in all cases not related to speculation. Since the non-speculative semantics is deterministic, all the rules delegating back are confluent. The only difference in the execution can appear if speculation is started. However, then the metaparamter restricts the no-branching rule of the other semantics part of the combination and it cannot be used.
Thus, speculation is always started and only one rule is possible, contradicting the assumption of Confluence.

For Projection Preservation, we refer to \cite{spec_comb} since the proofs are virtually the same.
Intuitively this follows from the structure of our semantics. We never forget a speculative transaction because of 
the metaparameter Z restriction. Thus, no actions are forgotten.

For Relation Preservation, we on our general construction in \Cref{sec:comb-rel-pres-general}, where we need to fulfil the assumptions made, which we will provide now:
\begin{description}
    \item[$\contract{}{\Bv}$]
        The lemmas needed are \cite[Lemma L.6]{S_sec_comp}
    \item[$\contract{}{\Jv}$] 
        The lemmas needed are \Thmref{lemma:v2-specarrow-preserves-safe}.
    \item[$\contract{}{\Sv}$]
        The lemmas needed are \Thmref{lemma:v4-specarrow-preserves-safe}.
    \item[$\contract{}{\Rv}$]
        The lemmas needed are \Thmref{lemma:v5-specarrow-preserves-safe}.
    \item[$\contract{}{\SLSv}$]
        The lemmas needed are \Thmref{lemma:sls-specarrow-preserves-safe}.
\end{description}

Now for the compatibility lemmas between $\relsax$ and $\relsay$ \Thmref{ass:compat-x-rel} and \Thmref{ass:compat-y-rel}.

First, let us quickly note that the relation $\relsa$ is the same for $\Bv$, $\Jv$, $\Sv$ and $\SLSv$.

Only for $\Rv$, we have the additional constraint that the Rsb $\Rsb$ is related.

Thus, for all combinations not involving $\contract{}{\Rv}$, the compatibility lemma trivial holds since there is no difference between $\relsax$ and $\relsaxy$.

For combinations involving $\contract{}{\Rv}$, we know that when a step of the other source semantics is happening, the Rsb $\Rsb$ is not changed. Thus, it is also easy to reinstate $\relsaxy$ in those cases.

Thus, we fulfil all assumptions to derive Relation Preservation for all our source contracts
 \section{Independence Results}\label[appendix]{app:ind}

Here, we describe our Independence results.

We first list $\insertedInstrs{\comp{\cdot}}$ and $\speculationInstrs{\contract{}{}}$ of our compilers and and spec. semantics (\Cref{tab:ins-instructions}).

\begin{table}[h]
    \begin{subtable}{0.55\linewidth}
    \centering
    \begin{tabular}{l|c}
    \toprule
        Compiler &  Inserted Instructions ($\insertedInstrs{\comp{\cdot}}$) \\
        \midrule
        $\complfenceS{\cdot}$ & $\{\barrierKywd \}$ \\
        $\complfenceR{\cdot}$ & $\{\barrierKywd \}$ \\
        $\complfenceSLS{\cdot}$ & $\{\barrierKywd \}$ \\
        $\compretpJ{\cdot}$ & $\{\callC{}, \pskip, \barrierKywd, \retC, \modretC, (static) \jmpC  \}$ \\
        $\compretpJF{\cdot}$ & $\{\callC{}, \pskip, \barrierKywd, \retC, \modretC, (static) \jmpC  \}$ \\
        $\compretpR{\cdot}$ & $\{\pskip, \barrierKywd, \popretC, \retC{}, \callC{}, (static) \jmpC{} \}$ \\
        $\compuslhB{\cdot}$ & $\{\pcondassign{r}{e'}{e}, \passign{r}{e} \}$ \\
        $\compsslhB{\cdot}$ \cite{S_sec_comp} & $\{\pcondassign{r}{e'}{e}, \passign{r}{e} \}$ \\
        $\complfenceB{\cdot}$ \cite{S_sec_comp} & $\{\barrierKywd \}$ \\
        \bottomrule
    \end{tabular}
    \end{subtable}    

    \vspace{2em}
    \setlength\tabcolsep{1.5pt}
    \begin{subtable}{0.55\linewidth}
    \centering
    \begin{tabular}{c|c}
    \toprule
        Semantics &  Speculation Instructions ($\speculationInstrs{\contract{}{}}$) \\
        \midrule
        $\contract{}{\Bv}$ & $\{ \jzC{} \}$ \\
        $\contract{}{\Jv}$ & $\{\jmpC \}$ \\
        $\contract{}{\Sv}$ & $\{ \storeC \}$ \\
        $\contract{}{\Rv}$ & $\{ \retC, \callC \}$ \\
        $\contract{}{\SLSv}$ & $\{\retC \}$ \\
        \bottomrule
    \end{tabular}
    \end{subtable}
    \caption{The type of instructions inserted by our compilers and the speculative instructions for the semantics.}
    \label{tab:ins-instructions}
\end{table}
A (static) $\jmpC$ instruction is local and produces no observation.
For combined semantics, we use the union of the speculation instructions.

\Cref{t:table-ind} describes in which cases (Syn.) Independence applies to the different spec. semantics we define.

If there is an entry for compiler $\comp{\cdot}$ in column $\contract{}{y}$ then we can show the following (Syn.) Independence statement: $\IND{\contract{}{y}}{\comp{\cdot}}$.

\begin{table*}[ht]

    \begin{subtable}{\textwidth}
    \centering
    \begin{tabular}{l*{13}{c}}
    \toprule
        \diaghead{\theadfont Contract longer text}{Compiler}{Semantics}  & \thead{$\contractSpec{}{\Bv}$}& \thead{$\contractSpec{}{\Jv}$} & \thead{$\contractSpec{}{\Sv}$} & \thead{$\contractSpec{}{\Rv}$ }& \thead{$\contractSpec{}{\SLSv}$} & \thead{$\contractSpec{}{\Bv + \Jv}$} & \thead{$\contractSpec{}{\Bv + \Sv}$} & \thead{$\contractSpec{}{\Bv + \Rv}$} & \thead{$\contractSpec{}{\Bv + \SLSv}$} & \thead{$\contractSpec{}{\Jv + \Sv}$} & \thead{$\contractSpec{}{\Jv + \Rv}$} & \thead{$\contractSpec{}{\Jv + \SLSv}$} & \thead{$\contractSpec{}{\Sv + \Rv}$}  \\
        \midrule

        $\complfenceS{\cdot}$  & \SITab & \SITab & \SITab & \SITab & \SITab & \SITab & \SITab & \SITab & \SITab &  \SITab & \SITab & \SITab & \SITab  \\
        $\complfenceR{\cdot}$  & \SITab & \SITab & \SITab  & \SITab &  \SITab & \SITab & \SITab & \SITab & \SITab & \SITab & \SITab &\SITab & \SITab  \\
        $\complfenceSLS{\cdot}$  & \SITab & \SITab & \SITab  & \SITab & \SITab & \SITab & \SITab & \SITab & \SITab & \SITab & \SITab & \SITab & \SITab \\
        $\compretpJ{\cdot}$  &  \SITab & \SITab & \SITab &  \ITab & \N & \SITab & \SITab & \ITab & \N &  \SITab & \ITab & \N & \ITab  \\
        $\compretpJF{\cdot}$ & \SITab & \SITab & \SITab &  \ITab & \ITab & \SITab & \SITab & \ITab & \ITab &  \SITab & \ITab & \ITab & \ITab  \\
        $\compretpR{\cdot}$  &  \SITab & \SITab & \SITab &  \SITab & \N  & \SITab & \SITab & \SITab & \N &  \SITab & \SITab & \N & \SITab  \\
        $\compuslhB{\cdot}$  & \ITab & \ITab &  \ITab & \ITab & \ITab & \ITab & \ITab & \ITab & \ITab & \ITab & \ITab & \ITab & \ITab  \\
        $\compsslhB{\cdot}$  & \ITab & \ITab &  \ITab & \ITab & \ITab & \ITab & \ITab & \ITab & \ITab & \SITab & \ITab & \ITab & \ITab  \\
        $\complfenceB{\cdot}$  & \SITab & \SITab & \SITab & \SITab & \SITab & \SITab & \SITab & \SITab & \SITab &  \SITab & \SITab & \SITab & \SITab  \\
     \bottomrule
    \end{tabular}
    \end{subtable}
    
\vspace{2em}

    \setlength\tabcolsep{1.5pt}
    \begin{subtable}{\textwidth}
    \centering
    \begin{tabular}{l*{10}{c}}
    \toprule
    \diaghead{\theadfont Contract longer text}{Compiler}{Semantics} & \thead{$\contractSpec{}{\Sv + \SLSv}$} & \thead{$\contractSpec{}{\Bv + \Jv + \Sv}$} & \thead{$\contractSpec{}{\Bv + \Jv + \Rv}$} &\thead{$\contractSpec{}{\Bv + \Jv + \SLSv}$} & \thead{$\contractSpec{}{\Bv + \Sv + \Rv}$} & \thead{$\contractSpec{}{\Bv + \Sv + \SLSv}$} & \thead{$\contractSpec{}{\Jv + \Sv + \Rv}$} & \thead{$\contractSpec{}{\Jv + \Sv + \SLSv}$} & \thead{$\contractSpec{}{\Bv + \Jv + \Sv + \Rv}$} &  \thead{$\contractSpec{}{\Bv + \Jv + \Sv + \SLSv}$}\\
    \midrule
    $\complfenceS{\cdot}$  & \SITab & \SITab & \SITab & \SITab & \SITab & \SITab & \SITab & \SITab & \SITab & \SITab  \\
    $\complfenceR{\cdot}$  & \SITab & \SITab & \SITab & \SITab & \SITab & \SITab & \SITab & \SITab & \SITab & \SITab \\
    $\complfenceSLS{\cdot}$ & \SITab  & \SITab &  \SITab  & \SITab  &  \SITab & \SITab & \SITab  & \SITab & \SITab & \SITab \\
    $\compretpJ{\cdot}$  & \N &  \SITab & \ITab & \N &  \ITab & \N & \ITab & \N & \ITab & \N \\
    $\compretpJF{\cdot}$ & \ITab &  \SITab & \ITab & \ITab &  \ITab & \ITab & \ITab & \ITab & \ITab & \ITab \\
    $\compretpR{\cdot}$  & \N &  \SITab & \SITab &  \N &  \SITab & \N & \SITab & \N  & \SITab &  \N \\
    $\compuslhB{\cdot}$  & \ITab & \ITab & \ITab & \ITab  & \ITab & \ITab & \ITab & \ITab & \ITab & \ITab \\
    $\compsslhB{\cdot}$  & \ITab & \ITab & \ITab & \ITab  & \ITab & \ITab & \ITab & \ITab & \ITab & \ITab \\
    $\complfenceB{\cdot}$  & \SITab & \SITab & \SITab & \SITab & \SITab & \SITab & \SITab & \SITab & \SITab & \SITab  \\
    \bottomrule
    \end{tabular}
    \end{subtable}
    \caption{Lists of Secure compilers and how they attain Independence for our semantics $\SI{\contract{}{}}{\comp{\cdot}}{}$. \SITab~ means by Syntactic Independence, \ITab~ means by proving Independence and \N ~ means that Independence is not possible. For $\compuslhB{\cdot}$  we use a different observer named the $ct + vl$ observer.
}\label{t:table-ind}
\end{table*}

Adding a fence to $\compretpJ{\cdot}$ to get $\compretpJF{\cdot}$ allows us to achieve Independence w.r.t to $\contract{}{\SLSv}$.
We could do the same for $\compretpR{\cdot}$, however, we cannot combine $\contract{}{\Rv}$ with $\contract{}{\SLSv}$ due to limitations in the combination framework.

We want to note that for $\contract{}{\Jv}$ speculation happens only on indirect jumps, while the retpoline compilers only insert direct jumps. If a jump is direct or indirect can be seen statically by inspection of the code. 
That is why we have for example $\SI{}{\compretpJ{\cdot}}{\Bv + \Jv}$.

Technically, we have that Syntactic Independence does not hold for $\compretpR{\cdot}$ w.r.t. $\contract{}{\Rv}$, since $\retC$ instructions are added. However, the Independence proof is a consequence of \Thmref{lem:self-ind-tr} and $\compretpR{\cdot} \vdash \rssp$ and this extends to all the combined semantics using $\contract{}{\Rv}$ as well.
 \section{The problem with SLH and $\Jvr$}

In the paper we claimed that lifting the security guarantees of $\compsslhB{\cdot}$ and $\compuslhB{\cdot}$ to a combination including $\semj$ is not possible.

Now, let us take an example program $p$ written in C that is secure for $\contract{}{\Jv}$ : $\contract{}{\Jv} \vdash p : \rss$:

\begin{lstlisting}[basicstyle=\small,style=Cstyle,
    label=lst:uslh-c,escapechar=|, captionpos=t]
int array1[160];
int array2[256 * 512];
unsigned int array1_size = 16;

void leak(int idx) {
    int x = array1[idx];
    int y = array2[x * 4096];
}

void main() {
    int x = 0;
    if (x != 0){
        int x2 = array1[idx];
        //int y2 = array2[x2 * 4096];
        static const void* table[] = {&&winter, &&spring, &&summer, &&fall};
        goto *table[x2];
        winter:
            return;
        spring:
            return;
        summer:
            return;
        fall:
            return;
        return;
    }
}
\end{lstlisting}
This program is secure because the indirect jump is never executed because it is guarded by the branch.

After SLH the program looks like this using clang with option -mspeculative-load-hardening.

\begin{lstlisting}[language={[x86masm]Assembler}]
    leak:                                   # @leak
        push    rbp
        mov     rbp, rsp
        mov     rax, -1
        mov     rcx, rsp
        sar     rcx, 63
        mov     dword ptr [rbp - 8], edi
        movsxd  rax, dword ptr [rbp - 8]
        mov     edx, dword ptr [4*rax + array1]
        mov     eax, ecx
        or      eax, edx
        mov     dword ptr [rbp - 12], eax
        mov     eax, dword ptr [rbp - 12]
        shl     eax, 12
        cdqe
        mov     edx, dword ptr [4*rax + array2]
        mov     eax, ecx
        or      eax, edx
        mov     dword ptr [rbp - 16], eax
        mov     eax, dword ptr [rbp - 4]
        shl     rcx, 47
        or      rsp, rcx
        pop     rbp
        ret
main:                                   # @main
        push    rbp
        mov     rbp, rsp
        mov     rax, -1
        mov     qword ptr [rbp - 32], rax       # 8-byte Spill
        mov     rax, rsp
        sar     rax, 63
        mov     qword ptr [rbp - 24], rax       # 8-byte Spill
        mov     dword ptr [rbp - 4], edi
        mov     dword ptr [rbp - 8], 0
        cmp     dword ptr [rbp - 8], 0
        je      .LBB3_8
        jmp     .LBB3_1
.LBB3_8:
        mov     rcx, qword ptr [rbp - 32]       # 8-byte Reload
        mov     rax, qword ptr [rbp - 24]       # 8-byte Reload
        cmovne  rax, rcx
        mov     qword ptr [rbp - 40], rax       # 8-byte Spill
        jmp     .LBB3_6
.LBB3_1:
        mov     rcx, qword ptr [rbp - 32]       # 8-byte Reload
        mov     rax, qword ptr [rbp - 24]       # 8-byte Reload
        cmove   rax, rcx
        mov     qword ptr [rbp - 56], rax       # 8-byte Spill
        movsxd  rcx, dword ptr [rbp - 4]
        mov     edx, dword ptr [4*rcx + array1]
        mov     ecx, eax
        or      ecx, edx
        mov     dword ptr [rbp - 12], ecx
        movsxd  rcx, dword ptr [rbp - 12]
        mov     rcx, qword ptr [8*rcx + main.table]
        or      rax, rcx
        mov     qword ptr [rbp - 48], rax       # 8-byte Spill
        jmp     .LBB3_7
.Ltmp8:                                 # Block address taken
.LBB3_2:
        mov     rcx, qword ptr [rbp - 32]       # 8-byte Reload
        mov     rax, qword ptr [rbp - 56]       # 8-byte Reload
        mov     rdx, qword ptr [rbp - 64]       # 8-byte Reload
        cmp     rdx, offset .LBB3_2
        cmovne  rax, rcx
        mov     qword ptr [rbp - 40], rax       # 8-byte Spill
        jmp     .LBB3_6
.Ltmp10:                                # Block address taken
.LBB3_3:
        mov     rcx, qword ptr [rbp - 32]       # 8-byte Reload
        mov     rax, qword ptr [rbp - 56]       # 8-byte Reload
        mov     rdx, qword ptr [rbp - 64]       # 8-byte Reload
        cmp     rdx, offset .LBB3_3
        cmovne  rax, rcx
        mov     qword ptr [rbp - 40], rax       # 8-byte Spill
        jmp     .LBB3_6
.Ltmp12:                                # Block address taken
.LBB3_4:
        mov     rcx, qword ptr [rbp - 32]       # 8-byte Reload
        mov     rax, qword ptr [rbp - 56]       # 8-byte Reload
        mov     rdx, qword ptr [rbp - 64]       # 8-byte Reload
        cmp     rdx, offset .LBB3_4
        cmovne  rax, rcx
        mov     qword ptr [rbp - 40], rax       # 8-byte Spill
        jmp     .LBB3_6
.Ltmp14:                                # Block address taken
.LBB3_5:
        mov     rcx, qword ptr [rbp - 32]       # 8-byte Reload
        mov     rax, qword ptr [rbp - 56]       # 8-byte Reload
        mov     rdx, qword ptr [rbp - 64]       # 8-byte Reload
        cmp     rdx, offset .LBB3_5
        cmovne  rax, rcx
        mov     qword ptr [rbp - 40], rax       # 8-byte Spill
        jmp     .LBB3_6
.LBB3_6:
        mov     rax, qword ptr [rbp - 40]       # 8-byte Reload
        shl     rax, 47
        or      rsp, rax
        pop     rbp
        ret
.LBB3_7:
        mov     rax, qword ptr [rbp - 56]       # 8-byte Reload
        mov     rcx, qword ptr [rbp - 48]       # 8-byte Reload
        mov     qword ptr [rbp - 64], rcx       # 8-byte Spill
        or      rax, rcx
        jmp     rax
__llvm_retpoline_r11:                   # @__llvm_retpoline_r11
        call    .Ltmp18
.LBB1_1:                                # Block address taken
        pause
        lfence
        jmp     .LBB1_1
.Ltmp18:
        mov     qword ptr [rsp], r11
        ret
\end{lstlisting}

While the target of the indirect jump is masked during possible branch speculation via SLH (code sequence starting in line 104), it is still an indirect jump.
Thus, it can mispredict, anywhere into the code. For example into the array access in the leak function.

Even with forward-edge security with Intel-CET or anything else, this is still problematic. WIth CFI, we can only indirect jump to the beginning of a function. While the function prologue uses the $rsp$ register to set up the mask for interprocedural slh, it is not correctly tracked here.

Since the indirect jump is not a call, the high bits of $rsp$ do not correctly track the speculation flag and the mask is not correctly initialized. So even if the indirect jump can only jump to the beginning of leak, it does not matter because the mask is not correctly initialized.

Thus, slh is not secure in this case and the program is not secure for the combined contract $\contract{}{\Bv + \Jv}$.

This is inline with the discussion of \citet{slh} who advocate the use of retpoline to protect indirect jumps.

 \section{How instantiate Labels: Or how to link the program to run it}\label{sec:Linker}

In the syntax of the program we have abstract labels, that annotate each instructions and are used for $\jmpC$ instructions. For example $\pjmp{l}$ (where $l$ is an abstract label of the program) should jump to that location in the program.
This is done because when we compile the source language and insert a countermeasure we would need to shift all instruction numbers by one (i.e. if we insert a fence). That is why we have these abstract labels.

A second reason is that attacker and component define their functions. But they do not know which instruction number the other uses. We cannot allow clashes because otherwise it is not defined what $\pjmp{2}$ means, when there are multiple instructions it can belong to.

While our semantics works on those abstract labels, we could instantiate these labels to concrete values, e.g., numbers.
This is essentially modelling the work of the linker and there are some pitfalls that we need to consider.
For example, since we want to prove something about the security of our compilers, we need a cross language trace relation $\tracerel$ that relates source and target traces. Thus, we need to somehow relate the concrete locations between source and target program.

The labels $l$ are drawn from a countable infinite set $\labelset{}$.

For the remainder of this technical report, we will not use the exact details of these labels.

\subsubsection{An example}
First, we need to plug together all the different functions of the attacker and the component. We assume that each instruction has a unique label, which is reasonable. Otherwise, we would rename certain labels.

Let us look at an image of two functions $f$ and $g$ with their instructions and their respective order. 
\begin{center}
    
    \begin{tikzpicture}[trim left]
        \node (top) at (0,0) {$f$};
        \node [below of=top] (i1)  {$i_1$};
        \node [below of=i1] (i2)  {$i_2$};
        \node [below of=i2] (in)  {$\cdots$};
    \node [below of=in] (in')  {$i_{n'}$};

    \node (top1) at ($ (top) + (80:2) $) {$g$};
    \node [below of=top1] (j1)  {$j_1$};
    \node [below of=j1] (j2)  {$j_2$};
    \node [below of=j2] (jn)  {$\cdots$};
    \node [below of=jn] (jn')  {$j_{n'}$};
    \draw [black,  thick, shorten <=-2pt, shorten >=-2pt] (i1) -- (i2) -- (in) -- (in');
    \draw [black, thick, shorten <=-2pt, shorten >=-2pt] (j1) -- (j2) -- (jn) -- (jn');
\end{tikzpicture}
\end{center}

Note, that this is not the control flow graph of the program. Even if function $g$ would contains a loop or an if statement the resulting diagram would be the same. We are only interested in the program order of the instructions itself.

Just looking at the functions in isolation, we can see that the labels form a total order. However, looking at the whole program together, it is course not a total order. The instructions between the two functions are not related, e.g. , $i1 < j1$ or $j1 < i1$ does not hold.
This means these functions form chains in the partial order of the whole program.

\subsubsection{The algorithm}
What we want is a function mapping the abstract labels of the program to natural numbers.
This function should be isomorphic.

First, we know that the set of labels of a program $\labelset_{P}$ is finite because we do not have infinite programs.
Every finite order has a least element and is in fact a well order

This means we can just take a totally ordered set of natural numbers ordered by $<$ that has the same cardinality as the set of labels used in the program.

This algorithm computes a bijection between the program labels ordered by an initial segment of the natural numbers.
\begin{algorithm}
    \caption{Computing the label map}
    \label{alg:labelmap}
    \begin{algorithmic}[1] \Input
      \Desc{(P, R) }{The partial order defined of the program where P includes all labels used in the program}
      \Desc{n}{A natural number n used as initial seed to generate the successor set of n.}
      \EndInput
      \Output
      \Desc{$\labelmap$}{A map from labels to natural numbers}
      \Desc{$N$}{The finite successor set of N }
      \EndOutput
        \Procedure{Euclid}{$a,b$}
            \State $N = \emptyset$
            \State $\labelmap = \emptyset$
            \While{$! empty(P)$}
                \State $x\gets rand(P)$
                \State $l \gets least\textunderscore element(x)$ \While{$cover(l, R) \neq l$} \State $P \gets P \setminus l$
                    \State $\labelmap \gets \labelmap[l \mapsto n]$
                    \State $N \gets N \cup {n}$
                    \State $n \gets n + 1$ \label{alg:line:nadd}
                \EndWhile
            \EndWhile
            \State \textbf{return} $(\labelmap, N)$
        \EndProcedure
    \end{algorithmic}
\end{algorithm}

We take a random label in the program. Each label belongs to one function in the program. Since the functions itself are total ordered and are finite, we know that there is a least element in that total order. We essentially 'climb down' the chain as the starting point. This ensures that we do not end up in situations like this for some function $f$:
\begin{lstlisting}[basicstyle=\small,style=MUASMstyle, escapechar=|, captionpos=t,]
5 : jmp l_1      
3 : ...
4 : load eax, 2
\end{lstlisting}
The algorithm ensures that the full chain / function is consistent in the numbering.

Then we 'climb up the chain by using the definition of cover (i.e. the smallest element that is bigger in the total order) and assign natural numbers in increasing order to those labels. After we are done with this chain, the algorithm will select the next one.

The map $\labelmap$ is a partial function from abstract labels to natural numbers. However, we know that the function will be total for all labels used in the program by construction.

\subsubsection{Why we use the successor set of N}

Note that instead of adding 1 to n in \Cref{alg:line:nadd} we could have increased $n$ arbitrarily. 
Furthermore, we need to use this algorithm for both the source program and the compiled program. Because the compiler can add instructions into the program.

\subsubsection{Computing the initial partial order of the whole program}
To use our algorithm, we need to give it as input the partial order of the program P.

W.l.o.g we assume that the labels used per function are disjoined.
We can define the partial order of the program P by the sum of all total orders defined by the functions appearing in the program.
\begin{definition}
$\sum_{(\mathcal{P}_i, \mathcal{T}_i) \in \{P.\OB{F}, P.\OB{I} \}}{}{(\mathcal{P}_i, \mathcal{T}_i) }$
\end{definition}

Now it remains to show, how we can compute the total orders for each function.

We give an example how to do it for function.
First, we map all entry points of functions into a set S.
First, lookup the function $f$ in the function map $F(f) = l$. This $l$ is the first instruction of the function.
A program is a sequence of labelled instructions. We iterate this sequence until the function is finished. This means until we find a label that is in the set of function entry points.
so $l : i; l_i : i' ; p2$ we add the edge $(l, l_i) \cup R$ and unfold the sequence further.

When then take the transitive closure and obtain $R^{+}$.

\subsection{The abstract label semantics}\label{sec:abstract-semantics}
Instead of using the instantiation, we just use the abstract labels.

\begin{gather*}
    \begin{aligned}
    	\mi{(Registers)}~ x \in&\ \Reg
    	&
    	\mi{(Labels)}~ l \in&\ \labelset
    	&
    	\mi{(Nats)}~ n \in&\ \Nat \cup \{\bot\}
    	&
        \mi{(Values)}~ v \in&\ \Val = \Nat \cup \{\bot\} \cup{} \labelset
    \end{aligned}
\end{gather*}

\subsubsection{Expression Evaluation}
\begin{center}
    \mytoprule{\exprEval{\sigma}{e}{v}}

\typerule{E-val}
{}
{
\exprEval{\sigmav}{v}{v}
}{E-val-inst}
\typerule{E-lookup}
{\sigmav(r) = v & \text{$r \in \Reg$}}
{
\exprEval{\sigmav}{r}{v}
}{E-lookup-inst}

\typerule{E-binop}
{
\exprEval{\sigmav}{e_1}{n_1} & \exprEval{\sigmav}{e_2 }{n_2} & v = n_1 \otimes n_2'
}
{
\exprEval{\sigmav}{e_1 \otimes e_2 }{v}
}{E-binop-inst}
\typerule{E-unop}
{
\exprEval{\sigmav}{e }{n}
}
{
\exprEval{\sigmav}{\ominus e }{\ominus n}
}{E-unop-inst}

\end{center}

BinOp and unop now take a natural number(extended with $\bot$).
Lookup can return a label

\begin{center}\small
\mytoprule{\sigma \nsarrow{\tau} \sigma}

\typerule{Assign}
{
\select{p}{\av(\pc)} = \passign{x}{e} & x \neq \pc
}
{
\src{C;\OB{\Bva}; \tup{p,\mv,\av}} \nsarrow{} \src{C;\OB{\Bva}; \tup{m, \av[\pc \mapsto \inc(\pc),x \mapsto \exprEval{\av}{e}{v}]}}
}{assign-inst}

\end{center}

Essentially we could decide to have a static semantics using labels and a run time semantics, where we instantiate all the labels.
We define the semantics on the abstract labels. However, since we defined how we can instantiate these labels we try to keep this abstract.
The difference is really just notational.
That is why we have the function $\inc$ defined:
\begin{align*}
    \inc(\pc) = 
     \begin{cases}
       \pc + 1 &\quad\text{if $\pc \in \Nat$}\\
       \succes(\pc) &\quad\text{if $\pc \in \labelset$}
     \end{cases}
\end{align*}

where $\succes(l)$ computes the cover of the label $l$. Since the cover of a label $l$ can be the label itself, i.e. at the end of a function when there is no return (which is strange in itself.). We need to annotate this with a special symbol (similar to $\bot$ for nats), recording that the program is now in a stuck state.

\begin{center}
    \typerule{Load}
{
\select{p}{\av(\pc)} = \pload{x}{e} & x \neq \pc & \exprEval{\av}{e}{n}
}
{
\src{C;\OB{\Bva}; \tup{p,\mv,\av}}  \nsarrow{\loadObs{n}} \src{C;\OB{\Bva}; \tup{\mv, \av[\pc \mapsto \av(\pc)+1, x \mapsto \mv^{\low}(n)]}}
}{load-ext}

\typerule{Store}
{
\select{p}{\av(\pc)} = \pstore{x}{e} &  \exprEval{\av}{e}{n}
}
{
\src{C;\OB{\Bva}; \tup{p,\mv,\av}} \nsarrow{\storeObs{n}} \src{C;\OB{\Bva}; \tup{ \mv^{\low}[n \mapsto \av(x)], \av[\pc \mapsto \av(\pc)+1]}}
}{store-ext}
\end{center}

Note that in both rules the expression $e$ needs to evaluate to a natural number $n$. This implicitly means that there are no labels $l$ in the expression $e$.

Let us explain why we forbid labels in the expression $e$ of store and loads. First of all, labels are not instantiated in our semantics so a store location like $l$ is not in the domain of the memory.

Of course, this argument does not hold if one simply instantiates the semantics. 
We will explain this detail in \Cref{sec:inst-semantics} where we argue that we still should forbid labels from appearing as store and load locations.

\begin{center}
    \typerule{Jmp}
{
\select{p}{\av(\pc)} = \pjmp{e} &  \exprEval{\av}{e}{\lbl} \\
\ffun{\av(\pc)} = f & \ffun{\lbl} = f' & \src{C.\mtt{intfs}}\vdash\src{f,f'}:\src{internal}
}
{
\src{C;\OB{\Bva}; \tup{p,\mv,\av}}  \nsarrow{\pcObs{\lbl}} \src{C;\OB{\Bva};\tup{ m, a[\pc \mapsto \lbl]}}
}{jmp-ext}
\end{center}

Here the expression $e$ needs to evaluate to a label. So eiher $e$ is just a label $l$ or is a register $r$ containing a label $l$.
The second case is actually called an indirect jump (relevant for SPectre V2) because the label is indirectly stored in the register.
Some CPUs allow $\pjmp{[n]}$ where $[n]$ means the value stored at memory address n. We do not have that but can simulate it by 
$\pload{r}{n}$ followed by $\pjmp{r}$. Furthermore, using the gcc flag \textit{-mindirect-branch-register} prohibits the use indirect jumps using memory and essentially does a rewrite like we do.

In \Cref{sec:inst-semantics} we argue why one should not allow natural numbers as targets for jumps in the instantiated semantics.

\subsection{Using the Instantiation : Semantics}\label{sec:inst-semantics}

We have shown the changes to the semantics when one uses the abstract labels. However, we can define the semantics using the instantiation of the labels using our computed $\labelmap$.
This means that all labels are resolved before/during the execution of the program.

For example, resolving the labels at \textbf{run-time} means that we add an additional rule the expression evaluation.
\begin{center}
    \mytoprule{\exprEval{\sigma, \labelmap}{e}{v}}

\typerule{E-label}
{\labelmap(l) = n}
{
\exprEval{\sigmav, \labelmap}{l}{n}
}{E-label}
\end{center}

This is similar to a virtual memory address lookup.

\subsubsection{Store instructions}\label{sec:linker-store-ex}
In \Cref{sec:abstract-semantics} we explained why the expression $e$ of store and loads cannot contain any labels. In the instantiation however, labels are resolved to natural numbers, so they could be used as store and load locations.

However, we argue that they still should not be used.

Consider the small program and two to possible instantiations according to our algorithm :
\begin{center}
    
\begin{minipage}[b]{0.33\linewidth}
\begin{lstlisting}[basicstyle=\small,style=MUASMstyle,  label=lst:label-store-numbers1,escapechar=|, captionpos=t,]
l_1 : load r1, l_x  |\label{label_load-1}| 
l_2 : store x, l_3 |\label{label_store-2}|
l_3 : store x, 100 |\label{label_store-3}|
\end{lstlisting}
\end{minipage}
\hfill
\begin{minipage}[b]{0.33\linewidth}
\begin{lstlisting}[basicstyle=\small,style=MUASMstyle,  label=lst:label-store-numbers2,escapechar=|, captionpos=t,]
98 : load r1, 5     
99 : store x, 100
100 : store x, 100 
\end{lstlisting}
\end{minipage}
\hfill
\begin{minipage}[b]{0.33\linewidth}
\begin{lstlisting}[basicstyle=\small,style=MUASMstyle,  label=lst:label-store-numbers3,escapechar=|, captionpos=t,]
1 : load r1, 8     
2 : store x, 3
3 : store x, 100 
\end{lstlisting}
\end{minipage}
\captionof{lstlisting}{Code on the left and the instantiated code on the right for different instantiations.}
\end{center}

In the program, we have a $\loadC$ instruction in \Cref{label_load-1} loading from the location of label $l_x$ and a $\storeC$ instruction in \Cref{label_store-2} that stores to location $l_3$.
See that the $\loadC$ instruction now non-deterministically chooses the location from memory dependent on the layout the linker created.
Second, for the $\storeC$ instruction in \Cref{label_store-2}, the store location can overlap with the store in \Cref{label_store-3}. Again, this behaviour depends on the layout of the program which the linker chooses. Most importantly, both of these behaviours are probably not intended by the programmer. Thereby, we do not allow them.

That is why we will not allow abstract labels as destination and sources for store and load instructions.

\subsubsection{Jmp instructions}

Similar for $\jmpC$ instructions, allowing natural numbers to be used as jump targets is a bad idea as well. Consider this program:

\begin{center}
    
\begin{minipage}[b]{0.33\linewidth}
\begin{lstlisting}[basicstyle=\small,style=MUASMstyle,  label=lst:label-jmp-numbers1,escapechar=|, captionpos=t,]
l_1 : jmp 5 
l_2 : store x, 5 
l_3 : store x, 100 
\end{lstlisting}
\end{minipage}
\hfill
\begin{minipage}[b]{0.33\linewidth}
\begin{lstlisting}[basicstyle=\small,style=MUASMstyle,  label={lst:label-jmp-numbers2},escapechar=|, captionpos=t,]
3 : jmp 5
4 : store x, 5 
5 : store x, 100 
\end{lstlisting}
\end{minipage}
\hfill 
\begin{minipage}[b]{0.33\linewidth}
\begin{lstlisting}[basicstyle=\small,style=MUASMstyle,  label={lst:label-jmp-numbers3},escapechar=|, captionpos=t,]
98 : jmp 5
99 : store x, 5 
100 : store x, 100 
\end{lstlisting}
\end{minipage}
\end{center}
In \Cref{lst:label-jmp-numbers2} the jump actually resolves. However, in \Cref{lst:label-jmp-numbers3} the jump does not resolve because there exists no instruction at location 5!

Again, the programmer relies on the layout of the program, which the linker decides.

\subsubsection{Using the instantiation: Changes to the Cross-Language relation}

Even though we use the abstract labels in the semantics. It is possible to use the instantiation as well.
At run-time the $\labelmap$ would be used to resolve all labels in the program.

However, one needs to be careful with the cross-language relation.
In the instantiation of the labels we have observations like $\pcObs{n}$ for $\jmpC$ instructions. Since we need to relate source and target traces, we need to very careful here.

Essentially what we do is a reverse lookup

\begin{center}

\begin{minipage}[b]{0.45\linewidth}
\begin{lstlisting}[basicstyle=\small,style=MUASMstyle, escapechar=|, captionpos=t,]
l_1 : jmp l_3       
l_2 : store x, e
l_3 : load eax, 2
\end{lstlisting}
\end{minipage}
\hfill
\begin{minipage}[b]{0.45\linewidth}
\begin{lstlisting}[basicstyle=\small,style=MUASMstyle,  escapechar=|, captionpos=t,]
l_1 : jmp l_3     
l_2 : store x, e
l_2' : spbarr
l_3 : load eax, 2
\end{lstlisting}
\end{minipage}
\captionof{lstlisting}{Code on the left and the compiled program with a countermeasure on the right.}
\end{center}
Lets assume we have an instantiation of the abstract labels according to our algorithm starting at 1.
Here, we have the observations $\src{\pcObs{3}}$ and $\trg{\pcObs{4}}$ when executing the $\jmpC$ instruction in source and target program.
However, these observations should be related, because they actually jump to the same instruction. The mismatch happens, because of the added countermeasure by the compiler.

So to relate source and target traces, we use the $\src{\labelmap^{-1}}$ and $\trg{\labelmap^{-1}}$ to do a reverse lookup and check if they map to the same abstract label
\begin{center}

\typerule{Action Relation - pcObs}{
		\src{\labelmap^{-1}(n)} = \trg{\labelmap^{-1}(n')}
		&
		\src{\taint}\equiv \trg{\taint'}
	}{
		\src{\pcObs{n}^{\taint}} \arel \trg{\pcObs{n'}^{\taint'}}
	}{ac-rel-cl-inst}
\end{center}

So even though we have $\src{\pcObs{3}}$ and $\trg{\pcObs{4}}$ they are  related by the trace relation because 
$\src{\labelmap^{-1}(3)} = l_3$ and  $\trg{\labelmap^{-1}(4)} = l_3$!

\subsection{Lifting the restriction on the jumps}
A jump instruction like $\pjmp{l_1 + 2}$ is not allowed right now. Here we explain, how one could lift this restriction 

Use the partial order of the source $\src{\labelmap}$ to resolve jump targets in the target language. See this example

\begin{center}

\begin{minipage}[b]{0.45\linewidth}
\begin{lstlisting}[basicstyle=\small,style=MUASMstyle, escapechar=|, captionpos=t,]
l_1 : jmp l_1 + 2       
l_2 : store x, e
l_3 : load eax, 2
\end{lstlisting}
\end{minipage}
\hfill
\begin{minipage}[b]{0.45\linewidth}
\begin{lstlisting}[basicstyle=\small,style=MUASMstyle,  escapechar=|, captionpos=t,]
l_1 : jmp l_1 + 2     
l_2 : store x, e
l_2' : spbarr
l_3 : load eax, 2
\end{lstlisting}
\end{minipage}
\captionof{lstlisting}{Code on the left and the compiled program with a countermeasure on the right.}
\end{center}

We expect the program in the source to jump to $l_3$. However in the target language, the compiler added a barrier instruction. Thus the target code could jump to $l_2'$. To remediate the situation, we need to think how we can resolve $l_1 + 2$.

Naively, we could use our label bijection to get $\src{\labelmap(l_1) + 2}$. Assuming our order starts at 1, this would resolve to 3.
Computing in the target, we get $\trg{\labelmap(l_1) + 2} = 3$. In this case $\src{\labelmap(l_1)} = \trg{\labelmap(l_1)}$. This is generally not the case. Think about it if the jmp was after the barrier instruction. This would shift the target level jump by 1 for example.

They both jump to instruction. This then points to the load instruction in the src and the barrier instruction in the target which is a mismatch that we do not want.

Instead, we transform the computation into the partial order domain: $l_1 + 2 = \succes(\succes(l_1))$. This computation returns a label instead of a natural number. In this case $l_3$. On which we now can apply the labelmap of the source. $\src{\labelmap(l_1) = 3}$

Note that we applied the success function on the \textbf{partial order of the source program}.
If we use the partial order of the target program we would get
$\trg{\succes(\succes(l_1))} = l_{2'}$ and we gained nothing because that still not what we want.
However, we could use the partial order of the source program to compute $\src{\succes(\succes(l_1))} = l_{3}$. On which we now apply the labelmap of the target program: $\trg{\labelmap(l_3) = 4}$.

So we compute the successor label using the partial order of the source  program and then apply the label map of either source or target language.
The trick was to interpret $+$ in the abstract domain of the partial order and to use the partial order of the source in the target as well!
$\src{\labelmap}(\trg{l + n})$, where $+$ means the successor of the label.
We think that this is a model of relocations that the linker applies when linking code together.

\end{document}